\documentclass[acmsmall,screen,nonacm]{acmart}
\newif\iffull\fulltrue
\iffull
\else
\fi

\setcopyright{rightsretained}
\acmPrice{}
\acmDOI{10.1145/3290331}
\acmYear{2019}
\copyrightyear{2019}
\acmJournal{PACMPL}
\acmVolume{3}
\acmNumber{POPL}
\acmArticle{18}
\acmMonth{1}

\usepackage[utf8]{inputenc}
\usepackage{booktabs}   
\usepackage{subcaption} 

\usepackage[]{ifdraft}

\usepackage{amsmath}
\usepackage{amssymb}
\usepackage{amsthm}
\usepackage{centernot}
\usepackage{mathtools}
\usepackage{stmaryrd}
\usepackage{thmtools}
\declaretheorem[style=acmdefinition]{definition}
\declaretheorem[style=acmplain]{lemma}
\declaretheorem[style=acmplain]{theorem}
\declaretheorem[style=acmplain]{corollary}

\declaretheorem[style=acmdefinition,numberwithin=section,name=Definition]{definitionA}
\declaretheorem[style=acmplain,numberwithin=section,name=Lemma]{lemmaA}
\declaretheorem[style=acmplain,numberwithin=section,name=Theorem]{theoremA}
\declaretheorem[style=acmplain,numberwithin=section,name=Corollary]{corollaryA}

\usepackage{color}
\usepackage{enumitem}  
\usepackage{multirow}

\newcommand{\ottdrule}[4][]{{\displaystyle\frac{\begin{array}{l}#2\end{array}}{#3}\quad\ottdrulename{#4}}}

\newcommand{\ottpremise}[1]{ #1 \\}
\newenvironment{ottdefnblock}[3][]{ \framebox{\mbox{#2}} \quad #3 \\[0pt]}{}

\newcommand{\ottnt}[1]{\mathit{#1}}
\newcommand{\ottmv}[1]{\mathit{#1}}

\newcommand{\ottsym}[1]{#1}

\newcommand{\ottdrulename}[1]{\textsc{#1}}

\renewcommand{\ottdrule}[4][]
  { {\displaystyle\frac{\begin{array}{c}#2\end{array} }{#3}\quad\ottdrulename{#4} } }

\newcommand{\ottdruleCXXBase}[1]{\ottdrule[#1]{%
}{
\iota  \sim  \iota}{%
{\ottdrulename{C\_Base}}{}%
}}

\newcommand{\ottdruleCXXTyVar}[1]{\ottdrule[#1]{%
}{
\ottmv{X}  \sim  \ottmv{X}}{%
{\ottdrulename{C\_TyVar}}{}%
}}

\newcommand{\ottdruleCXXDynR}[1]{\ottdrule[#1]{%
}{
\ottnt{U}  \sim  \star}{%
{\ottdrulename{C\_DynR}}{}%
}}

\newcommand{\ottdruleCXXDynL}[1]{\ottdrule[#1]{%
}{
\star  \sim  \ottnt{U}}{%
{\ottdrulename{C\_DynL}}{}%
}}

\newcommand{\ottdruleCXXArrow}[1]{\ottdrule[#1]{%
\ottpremise{\ottnt{U_{{\mathrm{11}}}}  \sim  \ottnt{U_{{\mathrm{21}}}}  \quad  \ottnt{U_{{\mathrm{12}}}}  \sim  \ottnt{U_{{\mathrm{22}}}}}%
}{
\ottnt{U_{{\mathrm{11}}}}  \!\rightarrow\!  \ottnt{U_{{\mathrm{12}}}}  \sim  \ottnt{U_{{\mathrm{21}}}}  \!\rightarrow\!  \ottnt{U_{{\mathrm{22}}}}}{%
{\ottdrulename{C\_Arrow}}{}%
}}

\newcommand{\ottdruleTXXVar}[1]{\ottdrule[#1]{%
\ottpremise{\ottmv{x}  \ottsym{:}  \ottnt{U} \, \in \, \Gamma}%
}{
\Gamma  \vdash  \ottmv{x}  \ottsym{:}  \ottnt{U}}{%
{\ottdrulename{T\_Var}}{}%
}}

\newcommand{\ottdruleTXXConst}[1]{\ottdrule[#1]{%
}{
\Gamma  \vdash  \ottnt{c}  \ottsym{:}   \mathit{ty} ( \ottnt{c} ) }{%
{\ottdrulename{T\_Const}}{}%
}}

\newcommand{\ottdruleTXXOp}[1]{\ottdrule[#1]{%
\ottpremise{ \mathit{ty} ( \mathit{op} )   \ottsym{=}  \iota_{{\mathrm{1}}}  \!\rightarrow\!  \iota_{{\mathrm{2}}}  \!\rightarrow\!  \iota  \quad  \Gamma  \vdash  \ottnt{f_{{\mathrm{1}}}}  \ottsym{:}  \iota_{{\mathrm{1}}}  \quad  \Gamma  \vdash  \ottnt{f_{{\mathrm{2}}}}  \ottsym{:}  \iota_{{\mathrm{2}}}}%
}{
\Gamma  \vdash  \mathit{op} \, \ottsym{(}  \ottnt{f_{{\mathrm{1}}}}  \ottsym{,}  \ottnt{f_{{\mathrm{2}}}}  \ottsym{)}  \ottsym{:}  \iota}{%
{\ottdrulename{T\_Op}}{}%
}}

\newcommand{\ottdruleTXXAbs}[1]{\ottdrule[#1]{%
\ottpremise{ \Gamma ,   \ottmv{x}  :  \ottnt{U_{{\mathrm{1}}}}    \vdash  \ottnt{f}  \ottsym{:}  \ottnt{U_{{\mathrm{2}}}}}%
}{
\Gamma  \vdash   \lambda  \ottmv{x} \!:\!  \ottnt{U_{{\mathrm{1}}}}  .\,  \ottnt{f}   \ottsym{:}  \ottnt{U_{{\mathrm{1}}}}  \!\rightarrow\!  \ottnt{U_{{\mathrm{2}}}}}{%
{\ottdrulename{T\_Abs}}{}%
}}

\newcommand{\ottdruleTXXApp}[1]{\ottdrule[#1]{%
\ottpremise{\Gamma  \vdash  \ottnt{f_{{\mathrm{1}}}}  \ottsym{:}  \ottnt{U_{{\mathrm{1}}}}  \!\rightarrow\!  \ottnt{U_{{\mathrm{2}}}}  \quad  \Gamma  \vdash  \ottnt{f_{{\mathrm{2}}}}  \ottsym{:}  \ottnt{U_{{\mathrm{1}}}}}%
}{
\Gamma  \vdash  \ottnt{f_{{\mathrm{1}}}} \, \ottnt{f_{{\mathrm{2}}}}  \ottsym{:}  \ottnt{U_{{\mathrm{2}}}}}{%
{\ottdrulename{T\_App}}{}%
}}

\newcommand{\ottdruleTXXCast}[1]{\ottdrule[#1]{%
\ottpremise{\Gamma  \vdash  \ottnt{f}  \ottsym{:}  \ottnt{U}  \quad  \ottnt{U}  \sim  \ottnt{U'}}%
}{
\Gamma  \vdash  \ottsym{(}  \ottnt{f}  \ottsym{:}   \ottnt{U} \Rightarrow  \unskip ^ { \ell }  \! \ottnt{U'}   \ottsym{)}  \ottsym{:}  \ottnt{U'}}{%
{\ottdrulename{T\_Cast}}{}%
}}

\newcommand{\ottdruleTXXBlame}[1]{\ottdrule[#1]{%
}{
\Gamma  \vdash  \textsf{\textup{blame}\relax} \, \ell  \ottsym{:}  \ottnt{U}}{%
{\ottdrulename{T\_Blame}}{}%
}}

\newcommand{\ottdruleTXXVarP}[1]{\ottdrule[#1]{%
\ottpremise{\ottmv{x}  \ottsym{:}  \forall \,  \overrightarrow{ \ottmv{X_{\ottmv{i}}} }   \ottsym{.}  \ottnt{U} \, \in \, \Gamma  \\   \text{For any }  \ottmv{X_{\ottmv{j}}} \, \in \,  \overrightarrow{ \ottmv{X_{\ottmv{i}}} }  ,   \mathbbsl{T}_{\ottmv{j}}  =   \nu   \, \text{ iff } \, \ottmv{X_{\ottmv{j}}} \, \not\in \, \textit{ftv} \, \ottsym{(}  \ottnt{U}  \ottsym{)} }%
}{
\Gamma  \vdash  \ottmv{x}  [   \overrightarrow{ \mathbbsl{T}_{\ottmv{i}} }   ]  \ottsym{:}  \ottnt{U}  [   \overrightarrow{ \ottmv{X_{\ottmv{i}}} }   \ottsym{:=}   \overrightarrow{ \mathbbsl{T}_{\ottmv{i}} }   ]}{%
{\ottdrulename{T\_VarP}}{}%
}}

\newcommand{\ottdruleTXXLetP}[1]{\ottdrule[#1]{%
\ottpremise{\Gamma  \vdash  w_{{\mathrm{1}}}  \ottsym{:}  \ottnt{U_{{\mathrm{1}}}}  \quad   \Gamma ,   \ottmv{x}  :  \forall \,  \overrightarrow{ \ottmv{X_{\ottmv{i}}} }   \ottsym{.}  \ottnt{U_{{\mathrm{1}}}}    \vdash  \ottnt{f_{{\mathrm{2}}}}  \ottsym{:}  \ottnt{U_{{\mathrm{2}}}}  \\   \overrightarrow{ \ottmv{X_{\ottmv{i}}} }   \cap  \textit{ftv} \, \ottsym{(}  \Gamma  \ottsym{)}  \ottsym{=}   \emptyset }%
}{
\Gamma  \vdash   \textsf{\textup{let}\relax} \,  \ottmv{x}  =   \Lambda    \overrightarrow{ \ottmv{X_{\ottmv{i}}} }  .\,  w_{{\mathrm{1}}}   \textsf{\textup{ in }\relax}  \ottnt{f_{{\mathrm{2}}}}   \ottsym{:}  \ottnt{U_{{\mathrm{2}}}}}{%
{\ottdrulename{T\_LetP}}{}%
}}

\newcommand{\ottdruleCIXXEmpty}[1]{\ottdrule[#1]{%
}{
 \emptyset   \rightsquigarrow   \emptyset }{%
{\ottdrulename{CI\_Empty}}{}%
}}

\newcommand{\ottdruleCIXXExtendVar}[1]{\ottdrule[#1]{%
\ottpremise{\Gamma  \rightsquigarrow  \Gamma'  \quad   \overrightarrow{ \ottmv{Y_{\ottmv{j}}} }   \cap  \textit{ftv} \, \ottsym{(}  \ottnt{U}  \ottsym{)}  \ottsym{=}   \emptyset }%
}{
 \Gamma ,   \ottmv{x}  :  \forall \,  \overrightarrow{ \ottmv{X_{\ottmv{i}}} }   \ottsym{.}  \ottnt{U}    \rightsquigarrow   \Gamma' ,   \ottmv{x}  :  \forall \,  \overrightarrow{ \ottmv{X_{\ottmv{i}}} }  \,  \overrightarrow{ \ottmv{Y_{\ottmv{j}}} }   \ottsym{.}  \ottnt{U}  }{%
{\ottdrulename{CI\_ExtendVar}}{}%
}}

\newcommand{\ottdruleCIXXVarP}[1]{\ottdrule[#1]{%
\ottpremise{\ottmv{x}  \ottsym{:}  \forall \,  \overrightarrow{ \ottmv{X_{\ottmv{i}}} }   \ottsym{.}  \ottnt{U} \, \in \, \Gamma  \quad   \text{For any }  \ottmv{X_{\ottmv{j}}} \, \in \,  \overrightarrow{ \ottmv{X_{\ottmv{i}}} }  ,   \mathbbsl{T}_{\ottmv{j}}  =   \nu   \, \text{ iff } \, \ottmv{X_{\ottmv{j}}} \, \not\in \, \textit{ftv} \, \ottsym{(}  \ottnt{U}  \ottsym{)} }%
}{
\Gamma  \vdash  \ottmv{x}  \rightsquigarrow  \ottmv{x}  [   \overrightarrow{ \mathbbsl{T}_{\ottmv{i}} }   ]  \ottsym{:}  \ottnt{U}  [   \overrightarrow{ \ottmv{X_{\ottmv{i}}} }   \ottsym{:=}   \overrightarrow{ \mathbbsl{T}_{\ottmv{i}} }   ]}{%
{\ottdrulename{CI\_VarP}}{}%
}}

\newcommand{\ottdruleCIXXConst}[1]{\ottdrule[#1]{%
}{
\Gamma  \vdash  \ottnt{c}  \rightsquigarrow  \ottnt{c}  \ottsym{:}   \mathit{ty} ( \ottnt{c} ) }{%
{\ottdrulename{CI\_Const}}{}%
}}

\newcommand{\ottdruleCIXXOp}[1]{\ottdrule[#1]{%
\ottpremise{\Gamma  \vdash  e_{{\mathrm{1}}}  \rightsquigarrow  f_{{\mathrm{1}}}  \ottsym{:}  \ottnt{U_{{\mathrm{1}}}}  \quad  \Gamma  \vdash  e_{{\mathrm{2}}}  \rightsquigarrow  f_{{\mathrm{2}}}  \ottsym{:}  \ottnt{U_{{\mathrm{2}}}}}%
\ottpremise{ \mathit{ty} ( \mathit{op} )   \ottsym{=}  \iota_{{\mathrm{1}}}  \!\rightarrow\!  \iota_{{\mathrm{2}}}  \!\rightarrow\!  \iota  \quad  \ottnt{U_{{\mathrm{1}}}}  \sim  \iota_{{\mathrm{1}}}  \quad  \ottnt{U_{{\mathrm{2}}}}  \sim  \iota_{{\mathrm{2}}}}%
}{
\Gamma  \vdash  \mathit{op} \, \ottsym{(}  e_{{\mathrm{1}}}  \ottsym{,}  e_{{\mathrm{2}}}  \ottsym{)}  \rightsquigarrow  \mathit{op} \, \ottsym{(}  f_{{\mathrm{1}}}  \ottsym{:}   \ottnt{U_{{\mathrm{1}}}} \Rightarrow  \unskip ^ { \ell_{{\mathrm{1}}} }  \! \iota_{{\mathrm{1}}}   \ottsym{,}  f_{{\mathrm{2}}}  \ottsym{:}   \ottnt{U_{{\mathrm{2}}}} \Rightarrow  \unskip ^ { \ell_{{\mathrm{2}}} }  \! \iota_{{\mathrm{2}}}   \ottsym{)}  \ottsym{:}  \iota}{%
{\ottdrulename{CI\_Op}}{}%
}}

\newcommand{\ottdruleCIXXAbsE}[1]{\ottdrule[#1]{%
\ottpremise{ \Gamma ,   \ottmv{x}  :  \ottnt{U_{{\mathrm{1}}}}    \vdash  e  \rightsquigarrow  f  \ottsym{:}  \ottnt{U_{{\mathrm{2}}}}}%
}{
\Gamma  \vdash   \lambda  \ottmv{x} \!:\!  \ottnt{U_{{\mathrm{1}}}}  .\,  e   \rightsquigarrow   \lambda  \ottmv{x} \!:\!  \ottnt{U_{{\mathrm{1}}}}  .\,  f   \ottsym{:}  \ottnt{U_{{\mathrm{1}}}}  \!\rightarrow\!  \ottnt{U_{{\mathrm{2}}}}}{%
{\ottdrulename{CI\_AbsE}}{}%
}}

\newcommand{\ottdruleCIXXAbsI}[1]{\ottdrule[#1]{%
\ottpremise{ \Gamma ,   \ottmv{x}  :  \ottnt{T}    \vdash  e  \rightsquigarrow  f  \ottsym{:}  \ottnt{U}}%
}{
\Gamma  \vdash   \lambda  \ottmv{x} .\,  e   \rightsquigarrow   \lambda  \ottmv{x} \!:\!  \ottnt{T}  .\,  f   \ottsym{:}  \ottnt{T}  \!\rightarrow\!  \ottnt{U}}{%
{\ottdrulename{CI\_AbsI}}{}%
}}

\newcommand{\ottdruleCIXXApp}[1]{\ottdrule[#1]{%
\ottpremise{\Gamma  \vdash  e_{{\mathrm{1}}}  \rightsquigarrow  f_{{\mathrm{1}}}  \ottsym{:}  \ottnt{U_{{\mathrm{1}}}}  \quad  \Gamma  \vdash  e_{{\mathrm{2}}}  \rightsquigarrow  f_{{\mathrm{2}}}  \ottsym{:}  \ottnt{U_{{\mathrm{2}}}}  \quad  \ottnt{U_{{\mathrm{1}}}}  \triangleright  \ottnt{U_{{\mathrm{11}}}}  \!\rightarrow\!  \ottnt{U_{{\mathrm{12}}}}  \quad  \ottnt{U_{{\mathrm{2}}}}  \sim  \ottnt{U_{{\mathrm{11}}}}}%
}{
\Gamma  \vdash  e_{{\mathrm{1}}} \, e_{{\mathrm{2}}}  \rightsquigarrow  \ottsym{(}  f_{{\mathrm{1}}}  \ottsym{:}   \ottnt{U_{{\mathrm{1}}}} \Rightarrow  \unskip ^ { \ell_{{\mathrm{1}}} }  \! \ottnt{U_{{\mathrm{11}}}}  \!\rightarrow\!  \ottnt{U_{{\mathrm{12}}}}   \ottsym{)} \, \ottsym{(}  f_{{\mathrm{2}}}  \ottsym{:}   \ottnt{U_{{\mathrm{2}}}} \Rightarrow  \unskip ^ { \ell_{{\mathrm{2}}} }  \! \ottnt{U_{{\mathrm{11}}}}   \ottsym{)}  \ottsym{:}  \ottnt{U_{{\mathrm{12}}}}}{%
{\ottdrulename{CI\_App}}{}%
}}

\newcommand{\ottdruleCIXXLetP}[1]{\ottdrule[#1]{%
\ottpremise{\Gamma  \vdash  v_{{\mathrm{1}}}  \rightsquigarrow  w_{{\mathrm{1}}}  \ottsym{:}  \ottnt{U_{{\mathrm{1}}}}  \quad   \Gamma ,   \ottmv{x}  :  \forall \,  \overrightarrow{ \ottmv{X_{\ottmv{i}}} }  \,  \overrightarrow{ \ottmv{Y_{\ottmv{j}}} }   \ottsym{.}  \ottnt{U_{{\mathrm{1}}}}    \vdash  e_{{\mathrm{2}}}  \rightsquigarrow  f_{{\mathrm{2}}}  \ottsym{:}  \ottnt{U_{{\mathrm{2}}}}}%
\ottpremise{ \overrightarrow{ \ottmv{X_{\ottmv{i}}} }   \ottsym{=}  \textit{ftv} \, \ottsym{(}  \ottnt{U_{{\mathrm{1}}}}  \ottsym{)}  \setminus  \ottsym{(}  \textit{ftv} \, \ottsym{(}  \Gamma  \ottsym{)}  \cup  \textit{ftv} \, \ottsym{(}  v_{{\mathrm{1}}}  \ottsym{)}  \ottsym{)}  \quad   \overrightarrow{ \ottmv{Y_{\ottmv{j}}} }   \ottsym{=}  \textit{ftv} \, \ottsym{(}  w_{{\mathrm{1}}}  \ottsym{)}  \setminus  \ottsym{(}  \textit{ftv} \, \ottsym{(}  \Gamma  \ottsym{)}  \cup  \textit{ftv} \, \ottsym{(}  \ottnt{U_{{\mathrm{1}}}}  \ottsym{)}  \cup  \textit{ftv} \, \ottsym{(}  v_{{\mathrm{1}}}  \ottsym{)}  \ottsym{)}}%
}{
\Gamma  \vdash   \textsf{\textup{let}\relax} \,  \ottmv{x}  =  v_{{\mathrm{1}}}  \textsf{\textup{ in }\relax}  e_{{\mathrm{2}}}   \rightsquigarrow   \textsf{\textup{let}\relax} \,  \ottmv{x}  =   \Lambda    \overrightarrow{ \ottmv{X_{\ottmv{i}}} }     \overrightarrow{ \ottmv{Y_{\ottmv{j}}} }  .\,  w_{{\mathrm{1}}}   \textsf{\textup{ in }\relax}  f_{{\mathrm{2}}}   \ottsym{:}  \ottnt{U_{{\mathrm{2}}}}}{%
{\ottdrulename{CI\_LetP}}{}%
}}

\newcommand{\ottdrulePXXIdBase}[1]{\ottdrule[#1]{%
}{
 \iota   \sqsubseteq _{ S }  \iota }{%
{\ottdrulename{P\_IdBase}}{}%
}}

\newcommand{\ottdrulePXXTyVar}[1]{\ottdrule[#1]{%
\ottpremise{\ottmv{X} \, \in \, \textit{dom} \, \ottsym{(}  S  \ottsym{)}}%
}{
 S  \ottsym{(}  \ottmv{X}  \ottsym{)}   \sqsubseteq _{ S }  \ottmv{X} }{%
{\ottdrulename{P\_TyVar}}{}%
}}

\newcommand{\ottdrulePXXDyn}[1]{\ottdrule[#1]{%
}{
 \ottnt{U}   \sqsubseteq _{ S }  \star }{%
{\ottdrulename{P\_Dyn}}{}%
}}

\newcommand{\ottdrulePXXArrow}[1]{\ottdrule[#1]{%
\ottpremise{ \ottnt{U_{{\mathrm{1}}}}   \sqsubseteq _{ S }  \ottnt{U'_{{\mathrm{1}}}}   \quad   \ottnt{U_{{\mathrm{2}}}}   \sqsubseteq _{ S }  \ottnt{U'_{{\mathrm{2}}}} }%
}{
 \ottnt{U_{{\mathrm{1}}}}  \!\rightarrow\!  \ottnt{U_{{\mathrm{2}}}}   \sqsubseteq _{ S }  \ottnt{U'_{{\mathrm{1}}}}  \!\rightarrow\!  \ottnt{U'_{{\mathrm{2}}}} }{%
{\ottdrulename{P\_Arrow}}{}%
}}

\newcommand{\ottdrulePXXConst}[1]{\ottdrule[#1]{%
}{
 \langle  \Gamma   \vdash   \ottnt{c}  :   \mathit{ty} ( \ottnt{c} )    \sqsubseteq _{ S }   \mathit{ty} ( \ottnt{c} )   :  \ottnt{c}  \dashv  \Gamma'  \rangle }{%
{\ottdrulename{P\_Const}}{}%
}}

\newcommand{\ottdrulePXXOp}[1]{\ottdrule[#1]{%
\ottpremise{ \mathit{ty} ( \mathit{op} )   \ottsym{=}  \iota_{{\mathrm{1}}}  \!\rightarrow\!  \iota_{{\mathrm{2}}}  \!\rightarrow\!  \iota  \quad   \langle  \Gamma   \vdash   f_{{\mathrm{1}}}  :  \iota_{{\mathrm{1}}}   \sqsubseteq _{ S }  \iota_{{\mathrm{1}}}  :  f'_{{\mathrm{1}}}  \dashv  \Gamma'  \rangle   \quad   \langle  \Gamma   \vdash   f_{{\mathrm{2}}}  :  \iota_{{\mathrm{2}}}   \sqsubseteq _{ S }  \iota_{{\mathrm{2}}}  :  f'_{{\mathrm{2}}}  \dashv  \Gamma'  \rangle }%
}{
 \langle  \Gamma   \vdash   \mathit{op} \, \ottsym{(}  f_{{\mathrm{1}}}  \ottsym{,}  f_{{\mathrm{2}}}  \ottsym{)}  :  \iota   \sqsubseteq _{ S }  \iota  :  \mathit{op} \, \ottsym{(}  f'_{{\mathrm{1}}}  \ottsym{,}  f'_{{\mathrm{2}}}  \ottsym{)}  \dashv  \Gamma'  \rangle }{%
{\ottdrulename{P\_Op}}{}%
}}

\newcommand{\ottdrulePXXAbs}[1]{\ottdrule[#1]{%
\ottpremise{ \ottnt{U_{{\mathrm{1}}}}   \sqsubseteq _{ S }  \ottnt{U'_{{\mathrm{1}}}}   \quad   \langle   \Gamma ,   \ottmv{x}  :  \ottnt{U_{{\mathrm{1}}}}     \vdash   f  :  \ottnt{U_{{\mathrm{2}}}}   \sqsubseteq _{ S }  \ottnt{U'_{{\mathrm{2}}}}  :  f'  \dashv   \Gamma' ,   \ottmv{x}  :  \ottnt{U'_{{\mathrm{1}}}}    \rangle }%
}{
 \langle  \Gamma   \vdash    \lambda  \ottmv{x} \!:\!  \ottnt{U_{{\mathrm{1}}}}  .\,  f   :  \ottnt{U_{{\mathrm{1}}}}  \!\rightarrow\!  \ottnt{U_{{\mathrm{2}}}}   \sqsubseteq _{ S }  \ottnt{U'_{{\mathrm{1}}}}  \!\rightarrow\!  \ottnt{U'_{{\mathrm{2}}}}  :   \lambda  \ottmv{x} \!:\!  \ottnt{U'_{{\mathrm{1}}}}  .\,  f'   \dashv  \Gamma'  \rangle }{%
{\ottdrulename{P\_Abs}}{}%
}}

\newcommand{\ottdrulePXXApp}[1]{\ottdrule[#1]{%
\ottpremise{ \langle  \Gamma   \vdash   f_{{\mathrm{1}}}  :  \ottnt{U_{{\mathrm{1}}}}  \!\rightarrow\!  \ottnt{U_{{\mathrm{2}}}}   \sqsubseteq _{ S }  \ottnt{U'_{{\mathrm{1}}}}  \!\rightarrow\!  \ottnt{U'_{{\mathrm{2}}}}  :  f'_{{\mathrm{1}}}  \dashv  \Gamma'  \rangle   \quad   \langle  \Gamma   \vdash   f_{{\mathrm{2}}}  :  \ottnt{U_{{\mathrm{1}}}}   \sqsubseteq _{ S }  \ottnt{U'_{{\mathrm{1}}}}  :  f'_{{\mathrm{2}}}  \dashv  \Gamma'  \rangle }%
}{
 \langle  \Gamma   \vdash   f_{{\mathrm{1}}} \, f_{{\mathrm{2}}}  :  \ottnt{U_{{\mathrm{2}}}}   \sqsubseteq _{ S }  \ottnt{U'_{{\mathrm{2}}}}  :  f'_{{\mathrm{1}}} \, f'_{{\mathrm{2}}}  \dashv  \Gamma'  \rangle }{%
{\ottdrulename{P\_App}}{}%
}}

\newcommand{\ottdrulePXXVarP}[1]{\ottdrule[#1]{%
\ottpremise{\ottmv{x}  \ottsym{:}  \forall \,  \overrightarrow{ \ottmv{X_{\ottmv{i}}} }   \ottsym{.}  \ottnt{U} \, \in \, \Gamma  \quad   \text{For any }  \ottmv{X_{\ottmv{k}}} \, \in \,  \overrightarrow{ \ottmv{X_{\ottmv{i}}} }  ,   \mathbbsl{T}_{\ottmv{k}}  =   \nu   \, \text{ iff } \, \ottmv{X_{\ottmv{k}}} \, \not\in \, \textit{ftv} \, \ottsym{(}  \ottnt{U}  \ottsym{)} }%
\ottpremise{\ottmv{x}  \ottsym{:}  \forall \,  \overrightarrow{ \ottmv{X'_{\ottmv{j}}} }   \ottsym{.}  \ottnt{U'} \, \in \, \Gamma'  \quad   \text{For any }  \ottmv{X'_{\ottmv{k}}} \, \in \,  \overrightarrow{ \ottmv{X'_{\ottmv{j}}} }  ,    { \mathbbsl{T}_{\ottmv{k}} }'   =   \nu   \, \text{ iff } \, \ottmv{X'_{\ottmv{k}}} \, \not\in \, \textit{ftv} \, \ottsym{(}  \ottnt{U'}  \ottsym{)} }%
\ottpremise{ \ottnt{U}  [   \overrightarrow{ \ottmv{X_{\ottmv{i}}} }   \ottsym{:=}   \overrightarrow{ \mathbbsl{T}_{\ottmv{i}} }   ]   \sqsubseteq _{ S }  \ottnt{U'}  [   \overrightarrow{ \ottmv{X'_{\ottmv{j}}} }   \ottsym{:=}   \overrightarrow{  { \mathbbsl{T}_{\ottmv{j}} }'  }   ] }%
}{
 \langle  \Gamma   \vdash   \ottmv{x}  [   \overrightarrow{ \mathbbsl{T}_{\ottmv{i}} }   ]  :  \ottnt{U}  [   \overrightarrow{ \ottmv{X_{\ottmv{i}}} }   \ottsym{:=}   \overrightarrow{ \mathbbsl{T}_{\ottmv{i}} }   ]   \sqsubseteq _{ S }  \ottnt{U'}  [   \overrightarrow{ \ottmv{X'_{\ottmv{j}}} }   \ottsym{:=}   \overrightarrow{  { \mathbbsl{T}_{\ottmv{j}} }'  }   ]  :  \ottmv{x}  [   \overrightarrow{  { \mathbbsl{T}_{\ottmv{j}} }'  }   ]  \dashv  \Gamma'  \rangle }{%
{\ottdrulename{P\_VarP}}{}%
}}

\newcommand{\ottdrulePXXLetP}[1]{\ottdrule[#1]{%
\ottpremise{ \langle  \Gamma   \vdash   w_{{\mathrm{1}}}  :  \ottnt{U_{{\mathrm{1}}}}   \sqsubseteq _{  [   \overrightarrow{ \ottmv{X'_{\ottmv{j}}} }   :=   \overrightarrow{ \ottnt{T'_{\ottmv{j}}} }   ]  \uplus  S  }  \ottnt{U'_{{\mathrm{1}}}}  :  w'_{{\mathrm{1}}}  \dashv  \Gamma'  \rangle   \quad   \text{For any }  \ottmv{X} \, \in \, \textit{dom} \, \ottsym{(}  S  \ottsym{)} ,   \overrightarrow{ \ottmv{X_{\ottmv{i}}} }   \cap  \textit{ftv} \, \ottsym{(}  S  \ottsym{(}  \ottmv{X}  \ottsym{)}  \ottsym{)}  \ottsym{=}   \emptyset  }%
\ottpremise{ \langle   \Gamma ,   \ottmv{x}  :  \forall \,  \overrightarrow{ \ottmv{X_{\ottmv{i}}} }   \ottsym{.}  \ottnt{U_{{\mathrm{1}}}}     \vdash   f_{{\mathrm{2}}}  :  \ottnt{U_{{\mathrm{2}}}}   \sqsubseteq _{ S }  \ottnt{U'_{{\mathrm{2}}}}  :  f'_{{\mathrm{2}}}  \dashv   \Gamma' ,   \ottmv{x}  :  \forall \,  \overrightarrow{ \ottmv{X'_{\ottmv{j}}} }   \ottsym{.}  \ottnt{U'_{{\mathrm{1}}}}    \rangle   \\   \overrightarrow{ \ottmv{X_{\ottmv{i}}} }   \cap  \textit{ftv} \, \ottsym{(}  \Gamma  \ottsym{)}  \ottsym{=}   \emptyset   \quad   \overrightarrow{ \ottmv{X'_{\ottmv{j}}} }   \cap  \textit{ftv} \, \ottsym{(}  \Gamma'  \ottsym{)}  \ottsym{=}   \emptyset }%
}{
 \langle  \Gamma   \vdash    \textsf{\textup{let}\relax} \,  \ottmv{x}  =   \Lambda    \overrightarrow{ \ottmv{X_{\ottmv{i}}} }  .\,  w_{{\mathrm{1}}}   \textsf{\textup{ in }\relax}  \ottnt{f_{{\mathrm{2}}}}   :  \ottnt{U_{{\mathrm{2}}}}   \sqsubseteq _{ S }  \ottnt{U'_{{\mathrm{2}}}}  :   \textsf{\textup{let}\relax} \,  \ottmv{x}  =   \Lambda    \overrightarrow{ \ottmv{X'_{\ottmv{j}}} }  .\,  w'_{{\mathrm{1}}}   \textsf{\textup{ in }\relax}  \ottnt{f'_{{\mathrm{2}}}}   \dashv  \Gamma'  \rangle }{%
{\ottdrulename{P\_LetP}}{}%
}}

\newcommand{\ottdrulePXXCast}[1]{\ottdrule[#1]{%
\ottpremise{ \langle  \Gamma   \vdash   f  :  \ottnt{U_{{\mathrm{1}}}}   \sqsubseteq _{ S }  \ottnt{U'_{{\mathrm{1}}}}  :  f'  \dashv  \Gamma'  \rangle   \quad  \ottnt{U_{{\mathrm{1}}}}  \sim  \ottnt{U_{{\mathrm{2}}}}  \quad  \ottnt{U'_{{\mathrm{1}}}}  \sim  \ottnt{U'_{{\mathrm{2}}}}  \quad   \ottnt{U_{{\mathrm{2}}}}   \sqsubseteq _{ S }  \ottnt{U'_{{\mathrm{2}}}} }%
}{
 \langle  \Gamma   \vdash   \ottsym{(}  f  \ottsym{:}   \ottnt{U_{{\mathrm{1}}}} \Rightarrow  \unskip ^ { \ell }  \! \ottnt{U_{{\mathrm{2}}}}   \ottsym{)}  :  \ottnt{U_{{\mathrm{2}}}}   \sqsubseteq _{ S }  \ottnt{U'_{{\mathrm{2}}}}  :  \ottsym{(}  f'  \ottsym{:}   \ottnt{U'_{{\mathrm{1}}}} \Rightarrow  \unskip ^ { \ell' }  \! \ottnt{U'_{{\mathrm{2}}}}   \ottsym{)}  \dashv  \Gamma'  \rangle }{%
{\ottdrulename{P\_Cast}}{}%
}}

\newcommand{\ottdrulePXXCastL}[1]{\ottdrule[#1]{%
\ottpremise{ \langle  \Gamma   \vdash   f  :  \ottnt{U_{{\mathrm{1}}}}   \sqsubseteq _{ S }  \ottnt{U'}  :  f'  \dashv  \Gamma'  \rangle   \quad  \ottnt{U_{{\mathrm{1}}}}  \sim  \ottnt{U}  \quad   \ottnt{U}   \sqsubseteq _{ S }  \ottnt{U'} }%
}{
 \langle  \Gamma   \vdash   \ottsym{(}  f  \ottsym{:}   \ottnt{U_{{\mathrm{1}}}} \Rightarrow  \unskip ^ { \ell }  \! \ottnt{U}   \ottsym{)}  :  \ottnt{U}   \sqsubseteq _{ S }  \ottnt{U'}  :  f'  \dashv  \Gamma'  \rangle }{%
{\ottdrulename{P\_CastL}}{}%
}}

\newcommand{\ottdrulePXXCastR}[1]{\ottdrule[#1]{%
\ottpremise{ \langle  \Gamma   \vdash   f  :  \ottnt{U}   \sqsubseteq _{ S }  \ottnt{U'_{{\mathrm{1}}}}  :  f'  \dashv  \Gamma'  \rangle   \quad  \ottnt{U'_{{\mathrm{1}}}}  \sim  \ottnt{U'}  \quad   \ottnt{U}   \sqsubseteq _{ S }  \ottnt{U'} }%
}{
 \langle  \Gamma   \vdash   f  :  \ottnt{U}   \sqsubseteq _{ S }  \ottnt{U'}  :  \ottsym{(}  f'  \ottsym{:}   \ottnt{U'_{{\mathrm{1}}}} \Rightarrow  \unskip ^ { \ell' }  \! \ottnt{U'}   \ottsym{)}  \dashv  \Gamma'  \rangle }{%
{\ottdrulename{P\_CastR}}{}%
}}

\newcommand{\ottdrulePXXBlame}[1]{\ottdrule[#1]{%
\ottpremise{\Gamma'  \vdash  f'  \ottsym{:}  \ottnt{U'}  \quad   \ottnt{U}   \sqsubseteq _{ S }  \ottnt{U'} }%
}{
 \langle  \Gamma   \vdash   \textsf{\textup{blame}\relax} \, \ell  :  \ottnt{U}   \sqsubseteq _{ S }  \ottnt{U'}  :  f'  \dashv  \Gamma'  \rangle }{%
{\ottdrulename{P\_Blame}}{}%
}}

\newcommand{\ottdruleITXXVarP}[1]{\ottdrule[#1]{%
\ottpremise{\ottmv{x}  \ottsym{:}  \forall \,  \overrightarrow{ \ottmv{X_{\ottmv{i}}} }   \ottsym{.}  \ottnt{U} \, \in \, \Gamma}%
}{
\Gamma  \vdash  \ottmv{x}  \ottsym{:}  \ottnt{U}  [   \overrightarrow{ \ottmv{X_{\ottmv{i}}} }   \ottsym{:=}   \overrightarrow{ \ottnt{T_{\ottmv{i}}} }   ]}{%
{\ottdrulename{IT\_VarP}}{}%
}}

\newcommand{\ottdruleITXXConst}[1]{\ottdrule[#1]{%
}{
\Gamma  \vdash  \ottnt{c}  \ottsym{:}   \mathit{ty} ( \ottnt{c} ) }{%
{\ottdrulename{IT\_Const}}{}%
}}

\newcommand{\ottdruleITXXOp}[1]{\ottdrule[#1]{%
\ottpremise{ \mathit{ty} ( \mathit{op} )   \ottsym{=}  \iota_{{\mathrm{1}}}  \!\rightarrow\!  \iota_{{\mathrm{2}}}  \!\rightarrow\!  \iota  \quad  \Gamma  \vdash  e_{{\mathrm{1}}}  \ottsym{:}  \ottnt{U_{{\mathrm{1}}}}  \quad  \Gamma  \vdash  e_{{\mathrm{2}}}  \ottsym{:}  \ottnt{U_{{\mathrm{2}}}}  \quad  \ottnt{U_{{\mathrm{1}}}}  \sim  \iota_{{\mathrm{1}}}  \quad  \ottnt{U_{{\mathrm{2}}}}  \sim  \iota_{{\mathrm{2}}}}%
}{
\Gamma  \vdash  \mathit{op} \, \ottsym{(}  e_{{\mathrm{1}}}  \ottsym{,}  e_{{\mathrm{2}}}  \ottsym{)}  \ottsym{:}  \iota}{%
{\ottdrulename{IT\_Op}}{}%
}}

\newcommand{\ottdruleITXXAbsE}[1]{\ottdrule[#1]{%
\ottpremise{ \Gamma ,   \ottmv{x}  :  \ottnt{U_{{\mathrm{1}}}}    \vdash  e  \ottsym{:}  \ottnt{U_{{\mathrm{2}}}}}%
}{
\Gamma  \vdash   \lambda  \ottmv{x} \!:\!  \ottnt{U_{{\mathrm{1}}}}  .\,  e   \ottsym{:}  \ottnt{U_{{\mathrm{1}}}}  \!\rightarrow\!  \ottnt{U_{{\mathrm{2}}}}}{%
{\ottdrulename{IT\_AbsE}}{}%
}}

\newcommand{\ottdruleITXXAbsI}[1]{\ottdrule[#1]{%
\ottpremise{ \Gamma ,   \ottmv{x}  :  \ottnt{T}    \vdash  e  \ottsym{:}  \ottnt{U}}%
}{
\Gamma  \vdash   \lambda  \ottmv{x} .\,  e   \ottsym{:}  \ottnt{T}  \!\rightarrow\!  \ottnt{U}}{%
{\ottdrulename{IT\_AbsI}}{}%
}}

\newcommand{\ottdruleITXXApp}[1]{\ottdrule[#1]{%
\ottpremise{\Gamma  \vdash  e_{{\mathrm{1}}}  \ottsym{:}  \ottnt{U_{{\mathrm{1}}}}  \quad  \Gamma  \vdash  e_{{\mathrm{2}}}  \ottsym{:}  \ottnt{U_{{\mathrm{2}}}}  \quad  \ottnt{U_{{\mathrm{1}}}}  \triangleright  \ottnt{U_{{\mathrm{11}}}}  \!\rightarrow\!  \ottnt{U_{{\mathrm{12}}}}  \quad  \ottnt{U_{{\mathrm{2}}}}  \sim  \ottnt{U_{{\mathrm{11}}}}}%
}{
\Gamma  \vdash  e_{{\mathrm{1}}} \, e_{{\mathrm{2}}}  \ottsym{:}  \ottnt{U_{{\mathrm{12}}}}}{%
{\ottdrulename{IT\_App}}{}%
}}

\newcommand{\ottdruleITXXLetP}[1]{\ottdrule[#1]{%
\ottpremise{\Gamma  \vdash  v_{{\mathrm{1}}}  \ottsym{:}  \ottnt{U_{{\mathrm{1}}}}  \quad   \Gamma ,   \ottmv{x}  :  \forall \,  \overrightarrow{ \ottmv{X_{\ottmv{i}}} }   \ottsym{.}  \ottnt{U_{{\mathrm{1}}}}    \vdash  e_{{\mathrm{2}}}  \ottsym{:}  \ottnt{U_{{\mathrm{2}}}}  \quad   \overrightarrow{ \ottmv{X_{\ottmv{i}}} }   \ottsym{=}  \textit{ftv} \, \ottsym{(}  \ottnt{U_{{\mathrm{1}}}}  \ottsym{)}  \setminus  \ottsym{(}  \textit{ftv} \, \ottsym{(}  \Gamma  \ottsym{)}  \cup  \textit{ftv} \, \ottsym{(}  v_{{\mathrm{1}}}  \ottsym{)}  \ottsym{)}}%
}{
\Gamma  \vdash   \textsf{\textup{let}\relax} \,  \ottmv{x}  =  v_{{\mathrm{1}}}  \textsf{\textup{ in }\relax}  e_{{\mathrm{2}}}   \ottsym{:}  \ottnt{U_{{\mathrm{2}}}}}{%
{\ottdrulename{IT\_LetP}}{}%
}}

\newcommand{\ottdruleIPXXVar}[1]{\ottdrule[#1]{%
}{
 \ottmv{x}   \sqsubseteq _{ S }  \ottmv{x} }{%
{\ottdrulename{IP\_Var}}{}%
}}

\newcommand{\ottdruleIPXXConst}[1]{\ottdrule[#1]{%
}{
 \ottnt{c}   \sqsubseteq _{ S }  \ottnt{c} }{%
{\ottdrulename{IP\_Const}}{}%
}}

\newcommand{\ottdruleIPXXOp}[1]{\ottdrule[#1]{%
\ottpremise{ e_{{\mathrm{1}}}   \sqsubseteq _{ S }  e_{{\mathrm{2}}}   \quad   e'_{{\mathrm{1}}}   \sqsubseteq _{ S }  e'_{{\mathrm{2}}} }%
}{
 \mathit{op} \, \ottsym{(}  e_{{\mathrm{1}}}  \ottsym{,}  e_{{\mathrm{2}}}  \ottsym{)}   \sqsubseteq _{ S }  \mathit{op} \, \ottsym{(}  e'_{{\mathrm{1}}}  \ottsym{,}  e'_{{\mathrm{2}}}  \ottsym{)} }{%
{\ottdrulename{IP\_Op}}{}%
}}

\newcommand{\ottdruleIPXXAbsI}[1]{\ottdrule[#1]{%
\ottpremise{ e   \sqsubseteq _{ S }  e' }%
}{
  \lambda  \ottmv{x} .\,  e    \sqsubseteq _{ S }   \lambda  \ottmv{x} .\,  e'  }{%
{\ottdrulename{IP\_AbsI}}{}%
}}

\newcommand{\ottdruleIPXXAbsIE}[1]{\ottdrule[#1]{%
\ottpremise{ e   \sqsubseteq _{ S }  e' }%
}{
  \lambda  \ottmv{x} .\,  e    \sqsubseteq _{ S }   \lambda  \ottmv{x} \!:\!  \star  .\,  e'  }{%
{\ottdrulename{IP\_AbsIE}}{}%
}}

\newcommand{\ottdruleIPXXAbsEI}[1]{\ottdrule[#1]{%
\ottpremise{ e   \sqsubseteq _{ S }  e' }%
}{
  \lambda  \ottmv{x} \!:\!  \ottnt{T_{{\mathrm{1}}}}  .\,  e    \sqsubseteq _{ S }   \lambda  \ottmv{x} .\,  e'  }{%
{\ottdrulename{IP\_AbsEI}}{}%
}}

\newcommand{\ottdruleIPXXAbsE}[1]{\ottdrule[#1]{%
\ottpremise{ \ottnt{U}   \sqsubseteq _{ S }  \ottnt{U'}   \quad   e   \sqsubseteq _{ S }  e' }%
}{
  \lambda  \ottmv{x} \!:\!  \ottnt{U}  .\,  e    \sqsubseteq _{ S }   \lambda  \ottmv{x} \!:\!  \ottnt{U'}  .\,  e'  }{%
{\ottdrulename{IP\_AbsE}}{}%
}}

\newcommand{\ottdruleIPXXApp}[1]{\ottdrule[#1]{%
\ottpremise{ e_{{\mathrm{1}}}   \sqsubseteq _{ S }  e'_{{\mathrm{1}}}   \quad   e_{{\mathrm{2}}}   \sqsubseteq _{ S }  e'_{{\mathrm{2}}} }%
}{
 e_{{\mathrm{1}}} \, e_{{\mathrm{2}}}   \sqsubseteq _{ S }  e'_{{\mathrm{1}}} \, e'_{{\mathrm{2}}} }{%
{\ottdrulename{IP\_App}}{}%
}}

\newcommand{\ottdruleIPXXLetP}[1]{\ottdrule[#1]{%
\ottpremise{ v_{{\mathrm{1}}}   \sqsubseteq _{ S }  v'_{{\mathrm{1}}}   \quad   e_{{\mathrm{2}}}   \sqsubseteq _{ S }  e'_{{\mathrm{2}}} }%
}{
  \textsf{\textup{let}\relax} \,  \ottmv{x}  =  v_{{\mathrm{1}}}  \textsf{\textup{ in }\relax}  e_{{\mathrm{2}}}    \sqsubseteq _{ S }   \textsf{\textup{let}\relax} \,  \ottmv{x}  =  v'_{{\mathrm{1}}}  \textsf{\textup{ in }\relax}  e'_{{\mathrm{2}}}  }{%
{\ottdrulename{IP\_LetP}}{}%
}}

\usepackage{./macros}

\begin{document}

\title{Dynamic Type Inference for Gradual Hindley--Milner Typing}         


\author{Yusuke Miyazaki}
\authornote{Current affiliation: Recruit Co., Ltd.}          
\orcid{0000-0003-3884-2636}             
\affiliation{
  \department{Graduate School of Informatics}              
  \institution{Kyoto University}            
  \city{Kyoto}
  \country{Japan}                    
}
\email{miyazaki@fos.kuis.kyoto-u.ac.jp}          

\author{Taro Sekiyama}
\orcid{0000-0001-9286-230X}             
\affiliation{
  \institution{National Institute of Informatics}           
  \city{Tokyo}
  \country{Japan}                   
}
\email{tsekiyama@acm.org}         

\author{Atsushi Igarashi}
\orcid{0000-0002-5143-9764}             
\affiliation{
  \department{Graduate School of Informatics}             
  \institution{Kyoto University}           
  \city{Kyoto}
  \country{Japan}                   
}
\email{igarashi@kuis.kyoto-u.ac.jp}         

\begin{abstract}
    Garcia and Cimini study a type inference problem for the ITGL, an
implicitly and gradually typed language with let-polymorphism, and
develop a sound and complete inference algorithm for it.  Soundness
and completeness mean that, if the algorithm succeeds, the input term
can be translated to a well-typed term of an explicitly typed blame
calculus by cast insertion and vice versa.  However, in general, there
are many possible translations depending on how type variables that
were left undecided by static type inference are instantiated with
concrete static types.  Worse, the translated terms may behave
differently---some evaluate to values but others raise blame.

In this paper, we propose and formalize a new blame calculus
\(\lambdaRTI\) that avoids such divergence as an intermediate language
for the ITGL.  A main idea is to allow a term to contain type
variables (that have not been instantiated during static type
inference) and defer instantiation of these type variables to run
time.  We introduce \emph{dynamic type inference} (DTI) into the
semantics of \(\lambdaRTI\) so that type variables are instantiated
along reduction.  The DTI-based semantics not only avoids the
divergence described above but also is \emph{sound and complete} with
respect to the semantics of fully instantiated terms in the following
sense: if the evaluation of a term succeeds (i.e., terminates with a
value) in the DTI-based semantics, then there is a fully instantiated
version of the term that also succeeds in the explicitly typed blame
calculus and vice versa.

Finally, we prove the gradual guarantee, which is an important
correctness criterion of a gradually typed language, for the ITGL.

\end{abstract}



\maketitle

\section{Introduction} \label{sec:introduction}
\subsection{Gradual Typing}
Statically and dynamically typed languages have complementary strengths.
On the one hand, static typing provides early detection of bugs, but
the enforced programming style can be too constrained, especially when
the type system is not very expressive.  On the other hand, dynamic
typing is better suited for rapid prototyping and fast adaption to
changing requirements, but error detection is deferred to run time.

There has been much work to integrate static and dynamic typing in a single
programming
language~\cite{DBLP:journals/toplas/AbadiCPP91,DBLP:conf/pldi/CartwrightF91,DBLP:conf/popl/Thatte90,DBLP:conf/oopsla/BrachaG93,DBLP:journals/toplas/FlanaganF99,conf/Scheme/SiekT06,DBLP:conf/popl/Tobin-HochstadtF08}.
\emph{Gradual typing}~\cite{conf/Scheme/SiekT06} is a framework to enable seamless code evolution from a fully dynamically typed program to a fully statically typed one within a single language.
The notion of gradual typing has been applied to various language features such
as objects~\cite{DBLP:conf/ecoop/SiekT07},
generics~\cite{DBLP:conf/oopsla/InaI11},
effects~\cite{DBLP:conf/icfp/SchwerterGT14,DBLP:journals/jfp/SchwerterGT16},
ownership~\cite{DBLP:conf/esop/SergeyC12},
parametric polymorphism~\cite{DBLP:conf/popl/AhmedFSW11,DBLP:journals/pacmpl/IgarashiSI17,DBLP:conf/esop/XieBO18}
and so on.
More recently, even methodologies to ``gradualize'' existing statically typed
languages systematically, i.e., to generate gradually typed languages, are also
studied~\cite{DBLP:conf/popl/GarciaCT16,DBLP:conf/popl/CiminiS16,DBLP:conf/popl/CiminiS17}.

The key notions in gradual typing are \emph{the dynamic type} and
\emph{type consistency}.  The dynamic type, denoted by $ \star $, is the
type for the dynamically typed code.  For instance, a function that
accepts an argument of type $ \star $ can use it in any way and the
function can be applied to any value.  So, both $\ottsym{(}   \lambda  \ottmv{x} \!:\!  \star  .\,  \ottmv{x}   \ottsym{+}   2   \ottsym{)} \,  3 $
and $\ottsym{(}   \lambda  \ottmv{x} \!:\!  \star  .\,  \ottmv{x}   \ottsym{+}   2   \ottsym{)} \,  \textsf{\textup{true}\relax} $ are well-typed programs in a gradually
typed language.  To formalize such loose static typechecking, a type
consistency relation, denoted by $ \sim $, on types replaces some use
of type equality.  In the typing rule for applications, the function
argument type and the type of an actual argument are required not to
be equal but to be consistent; also, both $\star  \sim   \textsf{\textup{int}\relax} $ and
$\star  \sim   \textsf{\textup{bool}\relax} $ hold, making the two terms above well typed.

The semantics of a gradually typed language is usually defined by a
``cast-inserting'' translation into an intermediate language with
explicit casts, which perform run-time typechecking.
For example, the two examples above can be translated into terms of the blame calculus~\cite{DBLP:conf/esop/WadlerF09} as follows.
\begin{eqnarray*}
   \ottsym{(}   \lambda  \ottmv{x} \!:\!  \star  .\,  \ottmv{x}   \ottsym{+}   2   \ottsym{)} \,  3 
    & \rightsquigarrow & \ottsym{(}   \lambda  \ottmv{x} \!:\!  \star  .\,  \ottsym{(}  \ottmv{x}  \ottsym{:}   \star \Rightarrow  \unskip ^ { \ell_{{\mathrm{1}}} }  \!  \textsf{\textup{int}\relax}    \ottsym{)}   \ottsym{+}   2   \ottsym{)} \, \ottsym{(}   3   \ottsym{:}    \textsf{\textup{int}\relax}  \Rightarrow  \unskip ^ { \ell_{{\mathrm{2}}} }  \! \star   \ottsym{)} \\
   \ottsym{(}   \lambda  \ottmv{x} \!:\!  \star  .\,  \ottmv{x}   \ottsym{+}   2   \ottsym{)} \,  \textsf{\textup{true}\relax} 
    & \rightsquigarrow & \ottsym{(}   \lambda  \ottmv{x} \!:\!  \star  .\,  \ottsym{(}  \ottmv{x}  \ottsym{:}   \star \Rightarrow  \unskip ^ { \ell_{{\mathrm{1}}} }  \!  \textsf{\textup{int}\relax}    \ottsym{)}   \ottsym{+}   2   \ottsym{)} \, \ottsym{(}   \textsf{\textup{true}\relax}   \ottsym{:}    \textsf{\textup{bool}\relax}  \Rightarrow  \unskip ^ { \ell_{{\mathrm{2}}} }  \! \star   \ottsym{)}
\end{eqnarray*}
Here, $\rightsquigarrow$ denotes cast-inserting translation.
The term of the form $\ottnt{f}  \ottsym{:}   \ottnt{U_{{\mathrm{1}}}} \Rightarrow  \unskip ^ { \ell }  \! \ottnt{U_{{\mathrm{2}}}} $ is a cast of $\ottnt{f}$ from
type $\ottnt{U_{{\mathrm{1}}}}$ to $\ottnt{U_{{\mathrm{2}}}}$ and appears where typechecking was loosened due to type
consistency.\footnote{%
  The symbol $\ell$ is called a blame label and is used to identify a
  cast.  Following \citet{DBLP:conf/snapl/SiekVCB15}, we use $\ottnt{f}$ for the intermediate language and save $e$ for the surface language ITGL.} 
In these examples, actual arguments $ 3 $ and
$ \textsf{\textup{true}\relax} $ are cast to $ \star $ and $\ottmv{x}$ of type $ \star $ is cast to
$ \textsf{\textup{int}\relax} $ before being passed to $+$, which expects an integer.
In what follows, a sequence of casts $\ottsym{(}  f  \ottsym{:}   \ottnt{U_{{\mathrm{1}}}} \Rightarrow  \unskip ^ { \ell_{{\mathrm{1}}} }  \! \ottnt{U_{{\mathrm{2}}}}   \ottsym{)}  \ottsym{:}   \ottnt{U_{{\mathrm{2}}}} \Rightarrow  \unskip ^ { \ell_{{\mathrm{2}}} }  \! \ottnt{U_{{\mathrm{3}}}} $ is often abbreviated to
$f  \ottsym{:}   \ottnt{U_{{\mathrm{1}}}} \Rightarrow  \unskip ^ { \ell_{{\mathrm{1}}} }  \!  \ottnt{U_{{\mathrm{2}}}} \Rightarrow  \unskip ^ { \ell_{{\mathrm{2}}} }  \! \ottnt{U_{{\mathrm{3}}}}  $.

The former term evaluates to $ 5 $ whereas the latter to an uncatchable
exception, called blame~\cite{DBLP:conf/icfp/FindlerF02}, $\textsf{\textup{blame}\relax} \, \ell_{{\mathrm{1}}}$,
indicating that the cast on $\ottmv{x}$ fails:
\[
  \begin{array}{lcl}
    \ottsym{(}   \lambda  \ottmv{x} \!:\!  \star  .\,  \ottsym{(}  \ottmv{x}  \ottsym{:}   \star \Rightarrow  \unskip ^ { \ell_{{\mathrm{1}}} }  \!  \textsf{\textup{int}\relax}    \ottsym{)}   \ottsym{+}   2   \ottsym{)} \, \ottsym{(}   3   \ottsym{:}    \textsf{\textup{int}\relax}  \Rightarrow  \unskip ^ { \ell_{{\mathrm{2}}} }  \! \star   \ottsym{)} &
     \longmapsto  & \ottsym{(}   3   \ottsym{:}    \textsf{\textup{int}\relax}  \Rightarrow  \unskip ^ { \ell_{{\mathrm{2}}} }  \!  \star \Rightarrow  \unskip ^ { \ell_{{\mathrm{1}}} }  \!  \textsf{\textup{int}\relax}     \ottsym{)}  \ottsym{+}   2  \\ &
     \longmapsto  &  3   \ottsym{+}   2  \\ &
     \longmapsto  &  5  \\[0.3em]
    \ottsym{(}   \lambda  \ottmv{x} \!:\!  \star  .\,  \ottsym{(}  \ottmv{x}  \ottsym{:}   \star \Rightarrow  \unskip ^ { \ell_{{\mathrm{1}}} }  \!  \textsf{\textup{int}\relax}    \ottsym{)}   \ottsym{+}   2   \ottsym{)} \, \ottsym{(}   \textsf{\textup{true}\relax}   \ottsym{:}    \textsf{\textup{bool}\relax}  \Rightarrow  \unskip ^ { \ell_{{\mathrm{2}}} }  \! \star   \ottsym{)} &
     \longmapsto  & \ottsym{(}   \textsf{\textup{true}\relax}   \ottsym{:}    \textsf{\textup{bool}\relax}  \Rightarrow  \unskip ^ { \ell_{{\mathrm{2}}} }  \!  \star \Rightarrow  \unskip ^ { \ell_{{\mathrm{1}}} }  \!  \textsf{\textup{int}\relax}     \ottsym{)}  \ottsym{+}   2  \\ &
     \longmapsto  & \textsf{\textup{blame}\relax} \, \ell_{{\mathrm{1}}}  \ottsym{+}   2  \\ &
     \longmapsto  & \textsf{\textup{blame}\relax} \, \ell_{{\mathrm{1}}}
  \end{array}
\]
The terms $ 3   \ottsym{:}    \textsf{\textup{int}\relax}  \Rightarrow  \unskip ^ { \ell_{{\mathrm{2}}} }  \! \star $ and $ \textsf{\textup{true}\relax}   \ottsym{:}    \textsf{\textup{bool}\relax}  \Rightarrow  \unskip ^ { \ell_{{\mathrm{2}}} }  \! \star $ are values
in the blame calculus and they can be understood as an integer and
Boolean, respectively, tagged with its type.  Being values, they are
passed to the function as they are.  The cast from $ \star $ to
$ \textsf{\textup{int}\relax} $ investigates if the tag of the target is $ \textsf{\textup{int}\relax} $; if it
is, the tag is removed and the untagged integer is passed to $+$;
otherwise, blame is raised with information on which cast has failed.

\subsection{Type Inference for Gradual Typing}

Type inference (a.k.a.\ type reconstruction) for languages with the
dynamic type has been studied.  \citet{DBLP:conf/dls/SiekV08} proposed
a unification-based type inference algorithm for a gradually typed
language with simple types.  \citet{DBLP:conf/popl/GarciaC15} later
proposed a type inference algorithm with a principal type property for
the Implicitly Typed Gradual Language (ITGL) with and without
let-polymorphism.  More recently, \citet{DBLP:conf/esop/XieBO18}
studied an extension of the Odersky--L{\"a}ufer type
system for higher-rank polymorphism~\cite{DBLP:conf/popl/OderskyL96}
with the dynamic type and bidirectional algorithmic typing for it.
Also, \citet{HengleinRehof95FPCA} studied a very close
problem of translation from an untyped functional language to an
ML-like language with coercions~\cite{DBLP:journals/scp/Henglein94}, using a constraint-based
type inference.

The key idea in Garcia and Cimini's work (inherited by Xie et al.) is to infer only
\emph{static} types---that is, types not containing $ \star $---for
where type annotations are omitted.  For example, for
\[
  \ottsym{(}   \lambda  \ottmv{x} .\,  \ottmv{x}  \,  2   \ottsym{)} \, \ottsym{(}   \lambda  \ottmv{y} \!:\!   \textsf{\textup{int}\relax}   .\,  \ottmv{y}   \ottsym{)},
\]
the type inference algorithm outputs the following fully annotated term:
\[
  e = \ottsym{(}   \lambda  \ottmv{x} \!:\!   \textsf{\textup{int}\relax}   \!\rightarrow\!   \textsf{\textup{int}\relax}   .\,  \ottmv{x}  \,  2   \ottsym{)} \, \ottsym{(}   \lambda  \ottmv{y} \!:\!   \textsf{\textup{int}\relax}   .\,  \ottmv{y}   \ottsym{)}.
\]
The idea of inferring only static types is significant for the
principal type property because it excludes terms like
$\ottsym{(}   \lambda  \ottmv{x} \!:\!  \star  \!\rightarrow\!   \textsf{\textup{int}\relax}   .\,  \ottmv{x}  \,  2   \ottsym{)} \, \ottsym{(}   \lambda  \ottmv{y} \!:\!   \textsf{\textup{int}\relax}   .\,  \ottmv{y}   \ottsym{)}$, $\ottsym{(}   \lambda  \ottmv{x} \!:\!   \textsf{\textup{int}\relax}   \!\rightarrow\!  \star  .\,  \ottmv{x}  \,  2   \ottsym{)} \, \ottsym{(}   \lambda  \ottmv{y} \!:\!   \textsf{\textup{int}\relax}   .\,  \ottmv{y}   \ottsym{)}$, and
$\ottsym{(}   \lambda  \ottmv{x} \!:\!  \star  .\,  \ottmv{x}  \,  2   \ottsym{)} \, \ottsym{(}   \lambda  \ottmv{y} \!:\!   \textsf{\textup{int}\relax}   .\,  \ottmv{y}   \ottsym{)}$, which are all well typed in the gradual
type system but not obtained by applying type substitution to
$e$.
Based on this idea, they showed that the ITGL enjoys the principal type property, which means that if there are type
annotations
to make a given term well typed, the type
inference succeeds and its output subsumes all other type annotations that make it well typed---in the sense that
they are obtained by applying some type substitution.

\subsection{Incoherence Problem}
\label{sec:intro-problem}

Unlike ordinary typed $\lambda$-calculi, however, the behavior of a
term depends on concrete types chosen for missing type annotations.
For example, for the following term
\[
  \ottsym{(}   \lambda  \ottmv{x} \!:\!  \star  .\,  \ottmv{x}  \,  2   \ottsym{)} \, \ottsym{(}   \lambda  \ottmv{y} .\,  \ottmv{y}   \ottsym{)}
\]
it is appropriate to recover any static type $\ottnt{T}$ for $\ottmv{y}$ if we
are interested only in obtaining a well-typed term because
$\star  \sim  \ottnt{T}  \!\rightarrow\!  \ottnt{T}$ but the evaluation of the resulting term
significantly differs depending on the choice for $\ottnt{T}$.  To see
this, let's translate $\ottsym{(}   \lambda  \ottmv{x} \!:\!  \star  .\,  \ottmv{x}  \,  2   \ottsym{)} \, \ottsym{(}   \lambda  \ottmv{y} \!:\!  \ottnt{T}  .\,  \ottmv{y}   \ottsym{)}$ (of type $ \star $).
It is translated to
$$\ottsym{(}   \lambda  \ottmv{x} \!:\!  \star  .\,  \ottsym{(}  \ottmv{x}  \ottsym{:}   \star \Rightarrow  \unskip ^ { \ell_{{\mathrm{1}}} }  \! \star  \!\rightarrow\!  \star   \ottsym{)}  \, \ottsym{(}   2   \ottsym{:}    \textsf{\textup{int}\relax}  \Rightarrow  \unskip ^ { \ell_{{\mathrm{2}}} }  \! \star   \ottsym{)}  \ottsym{)} \, \ottsym{(}  \ottsym{(}   \lambda  \ottmv{y} \!:\!  \ottnt{T}  .\,  \ottmv{y}   \ottsym{)}  \ottsym{:}   \ottnt{T}  \!\rightarrow\!  \ottnt{T} \Rightarrow  \unskip ^ { \ell_{{\mathrm{3}}} }  \! \star   \ottsym{)}$$
regardless of $\ottnt{T}$.  If $\ottnt{T} =  \textsf{\textup{int}\relax} $, then
it reduces to value $ 2   \ottsym{:}    \textsf{\textup{int}\relax}  \Rightarrow  \unskip ^ { \ell_{{\mathrm{3}}} }  \! \star $ as follows:
\[
  \begin{array}{ll}
  & \ottsym{(}   \lambda  \ottmv{x} \!:\!  \star  .\,  \ottsym{(}  \ottmv{x}  \ottsym{:}   \star \Rightarrow  \unskip ^ { \ell_{{\mathrm{1}}} }  \! \star  \!\rightarrow\!  \star   \ottsym{)}  \, \ottsym{(}   2   \ottsym{:}    \textsf{\textup{int}\relax}  \Rightarrow  \unskip ^ { \ell_{{\mathrm{2}}} }  \! \star   \ottsym{)}  \ottsym{)} \, \ottsym{(}  \ottsym{(}   \lambda  \ottmv{y} \!:\!   \textsf{\textup{int}\relax}   .\,  \ottmv{y}   \ottsym{)}  \ottsym{:}    \textsf{\textup{int}\relax}   \!\rightarrow\!   \textsf{\textup{int}\relax}  \Rightarrow  \unskip ^ { \ell_{{\mathrm{3}}} }  \! \star   \ottsym{)}\\
   \longmapsto^\ast  & \ottsym{(}  \ottsym{(}   \lambda  \ottmv{y} \!:\!   \textsf{\textup{int}\relax}   .\,  \ottmv{y}   \ottsym{)} \, \ottsym{(}   2   \ottsym{:}    \textsf{\textup{int}\relax}  \Rightarrow  \unskip ^ { \ell_{{\mathrm{2}}} }  \!  \star \Rightarrow  \unskip ^ {  \bar{ \ell_{{\mathrm{3}}} }  }  \!  \textsf{\textup{int}\relax}     \ottsym{)}  \ottsym{)}  \ottsym{:}    \textsf{\textup{int}\relax}  \Rightarrow  \unskip ^ { \ell_{{\mathrm{3}}} }  \! \star  \\
   \longmapsto   & \ottsym{(}  \ottsym{(}   \lambda  \ottmv{y} \!:\!   \textsf{\textup{int}\relax}   .\,  \ottmv{y}   \ottsym{)} \,  2   \ottsym{)}  \ottsym{:}    \textsf{\textup{int}\relax}  \Rightarrow  \unskip ^ { \ell_{{\mathrm{3}}} }  \! \star  \\
   \longmapsto   &  2   \ottsym{:}    \textsf{\textup{int}\relax}  \Rightarrow  \unskip ^ { \ell_{{\mathrm{3}}} }  \! \star 
  \end{array}
\]
but, if $\ottnt{T} =  \textsf{\textup{bool}\relax} $, it reduces to $\textsf{\textup{blame}\relax} \,  \bar{ \ell_{{\mathrm{3}}} } $ as follows:
\[
  \begin{array}{ll}
  & \ottsym{(}   \lambda  \ottmv{x} \!:\!  \star  .\,  \ottsym{(}  \ottmv{x}  \ottsym{:}   \star \Rightarrow  \unskip ^ { \ell_{{\mathrm{1}}} }  \! \star  \!\rightarrow\!  \star   \ottsym{)}  \, \ottsym{(}   2   \ottsym{:}    \textsf{\textup{int}\relax}  \Rightarrow  \unskip ^ { \ell_{{\mathrm{2}}} }  \! \star   \ottsym{)}  \ottsym{)} \, \ottsym{(}  \ottsym{(}   \lambda  \ottmv{y} \!:\!   \textsf{\textup{bool}\relax}   .\,  \ottmv{y}   \ottsym{)}  \ottsym{:}    \textsf{\textup{bool}\relax}   \!\rightarrow\!   \textsf{\textup{bool}\relax}  \Rightarrow  \unskip ^ { \ell_{{\mathrm{3}}} }  \! \star   \ottsym{)}\\
   \longmapsto^\ast  
  &  \ottsym{(}  \ottsym{(}   \lambda  \ottmv{y} \!:\!   \textsf{\textup{bool}\relax}   .\,  \ottmv{y}   \ottsym{)} \,  \graybox{  \ottsym{(}   2   \ottsym{:}    \textsf{\textup{int}\relax}  \Rightarrow  \unskip ^ { \ell_{{\mathrm{2}}} }  \!  \star \Rightarrow  \unskip ^ {  \bar{ \ell_{{\mathrm{3}}} }  }  \!  \textsf{\textup{bool}\relax}     \ottsym{)}  }   \ottsym{)}  \ottsym{:}    \textsf{\textup{bool}\relax}  \Rightarrow  \unskip ^ { \ell_{{\mathrm{3}}} }  \! \star  \\
   \longmapsto  & \textsf{\textup{blame}\relax} \,  \bar{ \ell_{{\mathrm{3}}} } 
  \end{array}
\]
(the shaded subterm is the source of blame and the blame label
$ \bar{ \ell_{{\mathrm{3}}} } $, which is the negation of $\ell_{{\mathrm{3}}}$, means that a functional
cast labeled $\ell_{{\mathrm{3}}}$ has failed due to type mismatch on the argument,
not on the return value).  \citet{DBLP:conf/esop/XieBO18} face the
same problem in a slightly different setting of a higher-rank
polymorphic type system with the dynamic type and point out that their
type system is not
\emph{coherent}~\cite{DBLP:journals/iandc/Breazu-TannenCGS91} in the
sense that the run-time behavior of the same source program depends on
the particular choice of types.
 
\citet{DBLP:conf/popl/GarciaC15} do not clearly discuss how to deal
with this problem.  Given term $\ottsym{(}   \lambda  \ottmv{x} \!:\!  \star  .\,  \ottmv{x}  \,  2   \ottsym{)} \, \ottsym{(}   \lambda  \ottmv{y} .\,  \ottmv{y}   \ottsym{)}$, their type
reconstruction algorithm outputs $\ottsym{(}   \lambda  \ottmv{x} \!:\!  \star  .\,  \ottmv{x}  \,  2   \ottsym{)} \, \ottsym{(}   \lambda  \ottmv{y} \!:\!  \ottmv{Y}  .\,  \ottmv{y}   \ottsym{)}$, where
$\ottmv{Y}$ is a type variable that is left undecided and they suggest
those undecided variables be replaced by type parameters but their
semantics is left unclear.  One possibility would be to understand
type parameters as distinguished base types (without constants) but it
would also make the execution fail because $ \textsf{\textup{int}\relax}  \neq \ottmv{Y}$.
The problem here is that the only choice for $\ottnt{T}$ that makes execution
successful is $ \textsf{\textup{int}\relax} $ but it is hard to see statically.

An alternative, which is close to what \citet{HengleinRehof95FPCA} do and
also what \citet{DBLP:conf/esop/XieBO18} suggest,
is to substitute the dynamic type $ \star $ for these undecided
type variables.  If we replace $\ottmv{Y}$ with $ \star $ in the example
above, we will get
\[
  \begin{array}{ll}
  & \ottsym{(}   \lambda  \ottmv{x} \!:\!  \star  .\,  \ottsym{(}  \ottmv{x}  \ottsym{:}   \star \Rightarrow  \unskip ^ { \ell_{{\mathrm{1}}} }  \! \star  \!\rightarrow\!  \star   \ottsym{)}  \, \ottsym{(}   2   \ottsym{:}    \textsf{\textup{int}\relax}  \Rightarrow  \unskip ^ { \ell_{{\mathrm{2}}} }  \! \star   \ottsym{)}  \ottsym{)} \, \ottsym{(}  \ottsym{(}   \lambda  \ottmv{y} \!:\!  \star  .\,  \ottmv{y}   \ottsym{)}  \ottsym{:}   \star  \!\rightarrow\!  \star \Rightarrow  \unskip ^ { \ell_{{\mathrm{3}}} }  \! \star   \ottsym{)}\\
   \longmapsto^\ast  & \ottsym{(}   \lambda  \ottmv{y} \!:\!  \star  .\,  \ottmv{y}   \ottsym{)} \, \ottsym{(}   2   \ottsym{:}    \textsf{\textup{int}\relax}  \Rightarrow  \unskip ^ { \ell_{{\mathrm{2}}} }  \! \star   \ottsym{)} \\
   \longmapsto   &  2   \ottsym{:}    \textsf{\textup{int}\relax}  \Rightarrow  \unskip ^ { \ell_{{\mathrm{2}}} }  \! \star \ .
  \end{array}
\]
As far as this example is concerned, substitution of $ \star $ sounds
like a good idea as, if there is static-type substitution that
makes execution successful (i.e., terminate at a value), substitution
of $ \star $ is expected to make execution also successful---this is
part of the gradual guarantee
property~\cite{DBLP:conf/snapl/SiekVCB15}.

However, substitution of $ \star $ can make execution successful, even
when there is \emph{no static type} that makes execution successful.
For example, let's consider the ITGL term
\[
 \ottsym{(}   \lambda  \ottmv{x} \!:\!  \star  \!\rightarrow\!  \star  \!\rightarrow\!  \star  .\,  \ottmv{x}  \,  2  \,  \textsf{\textup{true}\relax}   \ottsym{)} \, \ottsym{(}   \lambda  \ottmv{y_{{\mathrm{1}}}} .\,   \lambda  \ottmv{y_{{\mathrm{2}}}} .\,   \textsf{\textup{ if }\relax}  \ottnt{b}  \textsf{\textup{ then }\relax}  \ottmv{y_{{\mathrm{1}}}}  \textsf{\textup{ else }\relax}  \ottmv{y_{{\mathrm{2}}}}     \ottsym{)}
\]
where $ \textsf{\textup{ if }\relax}  e_{{\mathrm{1}}}  \textsf{\textup{ then }\relax}  e_{{\mathrm{2}}}  \textsf{\textup{ else }\relax}  e_{{\mathrm{3}}} $
requires $e_{{\mathrm{2}}}$ and $e_{{\mathrm{3}}}$ to have the same type and $\ottnt{b}$ is a
(statically typed) Boolean term that refers to neither $\ottmv{y_{{\mathrm{1}}}}$ nor $\ottmv{y_{{\mathrm{2}}}}$.
For this term, the type inference algorithm outputs
$
 \ottsym{(}   \lambda  \ottmv{x} \!:\!  \star  \!\rightarrow\!  \star  \!\rightarrow\!  \star  .\,  \ottmv{x}  \,  2  \,  \textsf{\textup{true}\relax}   \ottsym{)} \, \ottsym{(}   \lambda  \ottmv{y_{{\mathrm{1}}}} \!:\!  \ottmv{Y}  .\,   \lambda  \ottmv{y_{{\mathrm{2}}}} \!:\!  \ottmv{Y}  .\,   \textsf{\textup{ if }\relax}  \ottnt{b}  \textsf{\textup{ then }\relax}  \ottmv{y_{{\mathrm{1}}}}  \textsf{\textup{ else }\relax}  \ottmv{y_{{\mathrm{2}}}}     \ottsym{)}
$
and, if we substitute $ \star $ for $\ottmv{Y}$, then it executes as
follows:
\[
  \begin{array}{cl}
 & \ottsym{(}   \lambda  \ottmv{x} \!:\!  \star  \!\rightarrow\!  \star  \!\rightarrow\!  \star  .\,  \ottmv{x}  \,  2  \,  \textsf{\textup{true}\relax}   \ottsym{)} \, \ottsym{(}   \lambda  \ottmv{y_{{\mathrm{1}}}} \!:\!  \star  .\,   \lambda  \ottmv{y_{{\mathrm{2}}}} \!:\!  \star  .\,   \textsf{\textup{ if }\relax}  \ottnt{b}  \textsf{\textup{ then }\relax}  \ottmv{y_{{\mathrm{1}}}}  \textsf{\textup{ else }\relax}  \ottmv{y_{{\mathrm{2}}}}     \ottsym{)} \\  \rightsquigarrow  &
   \ottsym{(}   \lambda  \ottmv{x} \!:\!  \star  \!\rightarrow\!  \star  \!\rightarrow\!  \star  .\,  \ottmv{x}  \, \ottsym{(}   2   \ottsym{:}    \textsf{\textup{int}\relax}  \Rightarrow  \unskip ^ { \ell_{{\mathrm{1}}} }  \! \star   \ottsym{)} \, \ottsym{(}   \textsf{\textup{true}\relax}   \ottsym{:}    \textsf{\textup{bool}\relax}  \Rightarrow  \unskip ^ { \ell_{{\mathrm{2}}} }  \! \star   \ottsym{)}  \ottsym{)} \, \ottsym{(}   \lambda  \ottmv{y_{{\mathrm{1}}}} \!:\!  \star  .\,   \lambda  \ottmv{y_{{\mathrm{2}}}} \!:\!  \star  .\,   \textsf{\textup{ if }\relax}  \ottnt{b}  \textsf{\textup{ then }\relax}  \ottmv{y_{{\mathrm{1}}}}  \textsf{\textup{ else }\relax}  \ottmv{y_{{\mathrm{2}}}}     \ottsym{)} \\  \longmapsto^\ast  &
    \textsf{\textup{ if }\relax}  \ottnt{b}  \textsf{\textup{ then }\relax}  \ottsym{(}   2   \ottsym{:}    \textsf{\textup{int}\relax}  \Rightarrow  \unskip ^ { \ell_{{\mathrm{1}}} }  \! \star   \ottsym{)}  \textsf{\textup{ else }\relax}  \ottsym{(}   \textsf{\textup{true}\relax}   \ottsym{:}    \textsf{\textup{bool}\relax}  \Rightarrow  \unskip ^ { \ell_{{\mathrm{2}}} }  \! \star   \ottsym{)}  \\
     \longmapsto^\ast  &
   \left\{
    \begin{array}{l@{\qquad}l}
       2   \ottsym{:}    \textsf{\textup{int}\relax}  \Rightarrow  \unskip ^ { \ell_{{\mathrm{1}}} }  \! \star  & \text{(if $\ottnt{b}  \longmapsto^\ast   \textsf{\textup{true}\relax} $)} \\
       \textsf{\textup{true}\relax}   \ottsym{:}    \textsf{\textup{bool}\relax}  \Rightarrow  \unskip ^ { \ell_{{\mathrm{2}}} }  \! \star  & \text{(if $\ottnt{b}  \longmapsto^\ast   \textsf{\textup{false}\relax} $)} \\
    \end{array}
    \right.
  \end{array}
\]
but, if we substitute a static type $\ottnt{T}$ for $\ottmv{Y}$, then it results in blame
due to a failure of either of the shaded casts:
\[
  \begin{array}{cl}
 & \ottsym{(}   \lambda  \ottmv{x} \!:\!  \star  \!\rightarrow\!  \star  \!\rightarrow\!  \star  .\,  \ottmv{x}  \,  2  \,  \textsf{\textup{true}\relax}   \ottsym{)} \, \ottsym{(}   \lambda  \ottmv{y_{{\mathrm{1}}}} \!:\!  \ottnt{T}  .\,   \lambda  \ottmv{y_{{\mathrm{2}}}} \!:\!  \ottnt{T}  .\,   \textsf{\textup{ if }\relax}  \ottnt{b}  \textsf{\textup{ then }\relax}  \ottmv{y_{{\mathrm{1}}}}  \textsf{\textup{ else }\relax}  \ottmv{y_{{\mathrm{2}}}}     \ottsym{)} \\  \rightsquigarrow  &
   ( \lambda  \ottmv{x} \!:\!  \star  \!\rightarrow\!  \star  \!\rightarrow\!  \star  .\,  \ottmv{x} \, \ottsym{(}   2   \ottsym{:}    \textsf{\textup{int}\relax}  \Rightarrow  \unskip ^ { \ell_{{\mathrm{1}}} }  \! \star   \ottsym{)} \, \ottsym{(}   \textsf{\textup{true}\relax}   \ottsym{:}    \textsf{\textup{bool}\relax}  \Rightarrow  \unskip ^ { \ell_{{\mathrm{2}}} }  \! \star   \ottsym{)} )
    \\ & \qquad \ottsym{(}  \ottsym{(}   \lambda  \ottmv{y_{{\mathrm{1}}}} \!:\!  \ottnt{T}  .\,   \lambda  \ottmv{y_{{\mathrm{2}}}} \!:\!  \ottnt{T}  .\,   \textsf{\textup{ if }\relax}  \ottnt{b}  \textsf{\textup{ then }\relax}  \ottmv{y_{{\mathrm{1}}}}  \textsf{\textup{ else }\relax}  \ottmv{y_{{\mathrm{2}}}}     \ottsym{)}  \ottsym{:}   \ottnt{T}  \!\rightarrow\!  \ottnt{T}  \!\rightarrow\!  \ottnt{T} \Rightarrow  \unskip ^ { \ell_{{\mathrm{3}}} }  \! \star  \!\rightarrow\!  \star  \!\rightarrow\!  \star   \ottsym{)} \\
     \longmapsto^\ast  & \ottsym{(}  \ottsym{(}  \ottsym{(}   \lambda  \ottmv{y_{{\mathrm{1}}}} \!:\!  \ottnt{T}  .\,   \lambda  \ottmv{y_{{\mathrm{2}}}} \!:\!  \ottnt{T}  .\,   \textsf{\textup{ if }\relax}  \ottnt{b}  \textsf{\textup{ then }\relax}  \ottmv{y_{{\mathrm{1}}}}  \textsf{\textup{ else }\relax}  \ottmv{y_{{\mathrm{2}}}}     \ottsym{)} \,  \graybox{  \ottsym{(}   2   \ottsym{:}    \textsf{\textup{int}\relax}  \Rightarrow  \unskip ^ { \ell_{{\mathrm{1}}} }  \!  \star \Rightarrow  \unskip ^ {  \bar{ \ell_{{\mathrm{3}}} }  }  \! \ottnt{T}    \ottsym{)}  }   \ottsym{)}  \ottsym{:}   \ottnt{T}  \!\rightarrow\!  \ottnt{T} \Rightarrow  \unskip ^ { \ell_{{\mathrm{3}}} }  \! \star  \!\rightarrow\!  \star   \ottsym{)}
    \\ & \qquad \ottsym{(}   \textsf{\textup{true}\relax}   \ottsym{:}    \textsf{\textup{bool}\relax}  \Rightarrow  \unskip ^ { \ell_{{\mathrm{2}}} }  \! \star   \ottsym{)} \\
     \longmapsto^\ast  &
    \left\{\begin{array}{ll}
             \ottsym{(}  \ottsym{(}  \ottsym{(}   \lambda  \ottmv{y_{{\mathrm{2}}}} \!:\!  \ottnt{T}  .\,   \textsf{\textup{ if }\relax}  \ottnt{b}  \textsf{\textup{ then }\relax}   2   \textsf{\textup{ else }\relax}  \ottmv{y_{{\mathrm{2}}}}    \ottsym{)} \,  \graybox{  \ottsym{(}   \textsf{\textup{true}\relax}   \ottsym{:}    \textsf{\textup{bool}\relax}  \Rightarrow  \unskip ^ { \ell_{{\mathrm{2}}} }  \!  \star \Rightarrow  \unskip ^ {  \bar{ \ell_{{\mathrm{3}}} }  }  \! \ottnt{T}    \ottsym{)}  }   \ottsym{)}  \ottsym{:}   \ottnt{T} \Rightarrow  \unskip ^ { \ell_{{\mathrm{3}}} }  \! \star   \ottsym{)} & \multirow{2}{*}{\text{(if $\ottnt{T} =  \textsf{\textup{int}\relax} $)}} \\
              \hfill  \longmapsto  \textsf{\textup{blame}\relax} \,  \bar{ \ell_{{\mathrm{3}}} }  \\
             \textsf{\textup{blame}\relax} \,  \bar{ \ell_{{\mathrm{3}}} }  & \text{(otherwise)}
           \end{array}\right.
  \end{array}
\]
Substitution of $ \star $ is not only against the main idea that only
static types are inferred for missing type annotations but also not
very desirable from the viewpoint of program evolution and early bug
detection: Substitution of $ \star $ inadvertently hides the fact that
\emph{there is no longer a hope for a given term (in this case
  $\ottsym{(}   \lambda  \ottmv{x} \!:\!  \star  \!\rightarrow\!  \star  \!\rightarrow\!  \star  .\,  \ottmv{x}  \,  2  \,  \textsf{\textup{true}\relax}   \ottsym{)} \, \ottsym{(}   \lambda  \ottmv{y_{{\mathrm{1}}}} .\,   \lambda  \ottmv{y_{{\mathrm{2}}}} .\,   \textsf{\textup{ if }\relax}  \ottnt{b}  \textsf{\textup{ then }\relax}  \ottmv{y_{{\mathrm{1}}}}  \textsf{\textup{ else }\relax}  \ottmv{y_{{\mathrm{2}}}}     \ottsym{)}$) to be
  able to evolve to a well-typed term by replacing occurrences of
  $ \star $ with static types}.
This concealment of potential errors would be disappointing for
languages where it is hoped to detect as early as possible programs
that the underlying static type system does not accept after
evolution.

\subsection{Our Work: Dynamic Type Inference}

In this work, we propose and formalize a new blame calculus
\(\lambdaRTI\) that avoids both blame caused by wrong choice of static
types and the problem caused by substituting $ \star $.  A main idea is
to allow a $\lambdaRTI$ term to contain type variables (that represent
undecided ones during static---namely, usual compile-time---type inference) and defer instantiation of
these type variables to run time.  Specifically, we introduce
\emph{dynamic type inference} (DTI) into the semantics of
\(\lambdaRTI\) so that type variables are instantiated along
reduction.   For example, the term
\[
  \ottsym{(}   \lambda  \ottmv{x} \!:\!  \star  .\,  \ottsym{(}  \ottmv{x}  \ottsym{:}   \star \Rightarrow  \unskip ^ { \ell_{{\mathrm{1}}} }  \! \star  \!\rightarrow\!  \star   \ottsym{)}  \, \ottsym{(}   2   \ottsym{:}    \textsf{\textup{int}\relax}  \Rightarrow  \unskip ^ { \ell_{{\mathrm{2}}} }  \! \star   \ottsym{)}  \ottsym{)} \, \ottsym{(}  \ottsym{(}   \lambda  \ottmv{y} \!:\!  \ottmv{Y}  .\,  \ottmv{y}   \ottsym{)}  \ottsym{:}   \ottmv{Y}  \!\rightarrow\!  \ottmv{Y} \Rightarrow  \unskip ^ { \ell_{{\mathrm{3}}} }  \! \star   \ottsym{)}
\]
(obtained from $\ottsym{(}   \lambda  \ottmv{x} \!:\!  \star  .\,  \ottmv{x}  \,  2   \ottsym{)} \, \ottsym{(}   \lambda  \ottmv{y} .\,  \ottmv{y}   \ottsym{)}$) reduces to
\[
  \ottsym{(}  \ottsym{(}   \lambda  \ottmv{y} \!:\!  \ottmv{Y}  .\,  \ottmv{y}   \ottsym{)} \, \ottsym{(}   2   \ottsym{:}    \textsf{\textup{int}\relax}  \Rightarrow  \unskip ^ { \ell_{{\mathrm{2}}} }  \!  \star \Rightarrow  \unskip ^ {  \bar{ \ell_{{\mathrm{3}}} }  }  \! \ottmv{Y}    \ottsym{)}  \ottsym{)}  \ottsym{:}   \ottmv{Y} \Rightarrow  \unskip ^ { \ell_{{\mathrm{3}}} }  \! \star 
\]
and then, instead of raising blame, it reduces to
\[
  \ottsym{(}  \ottsym{(}   \lambda  \ottmv{y} \!:\!   \graybox{   \textsf{\textup{int}\relax}   }   .\,  \ottmv{y}   \ottsym{)} \,  2   \ottsym{)}  \ottsym{:}    \graybox{   \textsf{\textup{int}\relax}   }  \Rightarrow  \unskip ^ { \ell_{{\mathrm{3}}} }  \! \star 
\]
by instantiating $\ottmv{Y}$ with $ \textsf{\textup{int}\relax} $.  (Shaded parts denote where
instantiation took place.)  In general, when a tagged value
$w  \ottsym{:}   \iota \Rightarrow  \unskip ^ { \ell }  \! \star $ (where $\iota$ is a base type) meets a cast to
a type variable $\ottmv{Y}$, the cast succeeds and $\ottmv{Y}$ is instantiated
to $\iota$.  Similarly, if a tagged function value is cast to
$\ottmv{Y}$, two fresh type variables $\ottmv{Y_{{\mathrm{1}}}}$ and $\ottmv{Y_{{\mathrm{2}}}}$ are generated
and $\ottmv{Y}$ is instantiated to $\ottmv{Y_{{\mathrm{1}}}}  \!\rightarrow\!  \ottmv{Y_{{\mathrm{2}}}}$, expecting further
instantiation later.

Unlike the semantics based on substitution of $ \star $, the DTI-based
semantics raises blame if a type variable appears in conflicting
contexts.  For example, the term
\[
  \begin{array}{l}
   ( \lambda  \ottmv{x} \!:\!  \star  \!\rightarrow\!  \star  \!\rightarrow\!  \star  .\,  \ottmv{x} \, \ottsym{(}   2   \ottsym{:}    \textsf{\textup{int}\relax}  \Rightarrow  \unskip ^ { \ell_{{\mathrm{1}}} }  \! \star   \ottsym{)} \, \ottsym{(}   \textsf{\textup{true}\relax}   \ottsym{:}    \textsf{\textup{bool}\relax}  \Rightarrow  \unskip ^ { \ell_{{\mathrm{2}}} }  \! \star   \ottsym{)} )
    \\ \quad \ottsym{(}  \ottsym{(}   \lambda  \ottmv{y_{{\mathrm{1}}}} \!:\!  \ottmv{Y}  .\,   \lambda  \ottmv{y_{{\mathrm{2}}}} \!:\!  \ottmv{Y}  .\,   \textsf{\textup{ if }\relax}  \ottnt{b}  \textsf{\textup{ then }\relax}  \ottmv{y_{{\mathrm{1}}}}  \textsf{\textup{ else }\relax}  \ottmv{y_{{\mathrm{2}}}}     \ottsym{)}  \ottsym{:}   \ottmv{Y}  \!\rightarrow\!  \ottmv{Y}  \!\rightarrow\!  \ottmv{Y} \Rightarrow  \unskip ^ { \ell_{{\mathrm{3}}} }  \! \star  \!\rightarrow\!  \star  \!\rightarrow\!  \star   \ottsym{)}
  \end{array}
\]
(obtained from
$\ottsym{(}   \lambda  \ottmv{x} \!:\!  \star  \!\rightarrow\!  \star  \!\rightarrow\!  \star  .\,  \ottmv{x}  \,  2  \,  \textsf{\textup{true}\relax}   \ottsym{)} \, \ottsym{(}   \lambda  \ottmv{y_{{\mathrm{1}}}} \!:\!  \ottmv{Y}  .\,   \lambda  \ottmv{y_{{\mathrm{2}}}} \!:\!  \ottmv{Y}  .\,   \textsf{\textup{ if }\relax}  \ottnt{b}  \textsf{\textup{ then }\relax}  \ottmv{y_{{\mathrm{1}}}}  \textsf{\textup{ else }\relax}  \ottmv{y_{{\mathrm{2}}}}     \ottsym{)}$)
reduces in a few steps to
\[
  \ottsym{(}  \ottsym{(}  \ottsym{(}   \lambda  \ottmv{y_{{\mathrm{1}}}} \!:\!  \ottmv{Y}  .\,   \lambda  \ottmv{y_{{\mathrm{2}}}} \!:\!  \ottmv{Y}  .\,   \textsf{\textup{ if }\relax}  \ottnt{b}  \textsf{\textup{ then }\relax}  \ottmv{y_{{\mathrm{1}}}}  \textsf{\textup{ else }\relax}  \ottmv{y_{{\mathrm{2}}}}     \ottsym{)} \,  \graybox{  \ottsym{(}   2   \ottsym{:}    \textsf{\textup{int}\relax}  \Rightarrow  \unskip ^ { \ell_{{\mathrm{1}}} }  \!  \star \Rightarrow  \unskip ^ {  \bar{ \ell_{{\mathrm{3}}} }  }  \! \ottmv{Y}    \ottsym{)}  }   \ottsym{)}  \ottsym{:}   \ottmv{Y}  \!\rightarrow\!  \ottmv{Y} \Rightarrow  \unskip ^ { \ell_{{\mathrm{3}}} }  \! \star  \!\rightarrow\!  \star   \ottsym{)} \cdots
\]
which corresponds to $\ottmv{x} \,  2 $ in the source term.  The cast on
$ 2 $ succeeds and this term reduces to
$\ottsym{(}  \ottsym{(}   \lambda  \ottmv{y_{{\mathrm{2}}}} \!:\!   \graybox{   \textsf{\textup{int}\relax}   }   .\,   \textsf{\textup{ if }\relax}  \ottnt{b}  \textsf{\textup{ then }\relax}   2   \textsf{\textup{ else }\relax}  \ottmv{y_{{\mathrm{2}}}}    \ottsym{)}  \ottsym{:}    \graybox{   \textsf{\textup{int}\relax}   \!\rightarrow\!   \textsf{\textup{int}\relax}   }  \Rightarrow  \unskip ^ { \ell_{{\mathrm{3}}} }  \! \star  \!\rightarrow\!  \star   \ottsym{)} \, \ottsym{(}   \textsf{\textup{true}\relax}   \ottsym{:}    \textsf{\textup{bool}\relax}  \Rightarrow  \unskip ^ { \ell_{{\mathrm{2}}} }  \! \star   \ottsym{)}$ by
instantiating $\ottmv{Y}$ with $ \textsf{\textup{int}\relax} $ (where the shaded parts are the results of instantiation).  Then, in the next step, reduction
reaches the application of $\ottmv{x} \,  2 $ to $ \textsf{\textup{true}\relax} $ (in the source term):
\[
  \ottsym{(}  \ottsym{(}  \ottsym{(}   \lambda  \ottmv{y_{{\mathrm{2}}}} \!:\!   \textsf{\textup{int}\relax}   .\,   \textsf{\textup{ if }\relax}  \ottnt{b}  \textsf{\textup{ then }\relax}   2   \textsf{\textup{ else }\relax}  \ottmv{y_{{\mathrm{2}}}}    \ottsym{)} \,  \graybox{  \ottsym{(}   \textsf{\textup{true}\relax}   \ottsym{:}    \textsf{\textup{bool}\relax}  \Rightarrow  \unskip ^ { \ell_{{\mathrm{2}}} }  \!  \star \Rightarrow  \unskip ^ {  \bar{ \ell_{{\mathrm{3}}} }  }  \!  \textsf{\textup{int}\relax}     \ottsym{)}  }   \ottsym{)}  \ottsym{:}    \textsf{\textup{int}\relax}  \Rightarrow  \unskip ^ { \ell_{{\mathrm{3}}} }  \! \star   \ottsym{)}.
\]
However, the shaded cast on $ \textsf{\textup{true}\relax} $ fails.  In short, $\ottmv{Y}$ is
required to be both $ \textsf{\textup{int}\relax} $ and $ \textsf{\textup{bool}\relax} $ at the same time, which
is impossible.  As this example shows, DTI is not as permissive as the
semantics based on substitution of $ \star $ and detects a type error
early.

DTI is \emph{sound} and \emph{complete}.  Intuitively, soundness means
that, if a program evaluates to a value under the DTI semantics, then
the program obtained by applying---in advance---the type instantiation
that DTI found results in the same value and if a program results in
blame under the DTI semantics, then all type substitutions make the
program result also in blame.  Completeness means that, if some type
substitution makes the program evaluate to a value, then execution
with DTI also results in a related value.  Soundness also means that
the semantics is not too permissive: it is not the case that a program
evaluates to a value under the DTI semantics but no type substitution
makes the program evaluate to a value.  The semantics based on
substituting $ \star $ is complete but not sound; the semantics based on
``undecided type variables as base types'' as in
\citet{DBLP:conf/popl/GarciaC15} is neither sound nor complete
(because it just raises blame too often).

We equip \(\lambdaRTI\) with ML-style let-polymorphism~\cite{Milner78JCSS}.
Actually, Garcia and Cimini have already proposed two ways to implement the
ITGL with let-polymorphism: by translating it to the Polymorphic Blame
Calculus~\cite{DBLP:conf/popl/AhmedFSW11,DBLP:journals/pacmpl/AhmedJSW17}
or by expanding $ \textsf{\textup{let}\relax} $ before translating it to the (monomorphic) blame
calculus.  However, they have left a detailed comparison of the two to
future work.  Our semantics is very close to the latter, although we
do not statically expand definitions by $ \textsf{\textup{let}\relax} $.  Perhaps surprisingly,
the semantics is not quite parametric; we argue that translation to
the Polymorphic Blame Calculus, which dynamically enforces
parametricity~\cite{DBLP:conf/ifip/Reynolds83,DBLP:journals/pacmpl/AhmedJSW17},
has an undesirable consequence and our semantics (which is close to the one based on expanding $ \textsf{\textup{let}\relax} $)
is better suited for languages in which type abstraction
and application are implicit.

Other than soundness and completeness of DTI, we also study the gradual guarantee
property~\cite{DBLP:conf/snapl/SiekVCB15} for the ITGL.  The gradual
guarantee formalizes an informal expectation for gradual typing
systems that adding more static types to a program only exposes type
errors---whether they are static or dynamic---but should not change
the behavior of the program otherwise.
To deal with the ITGL, where bound variables come with optional type
annotations, we extend the notion of ``more static types'' (formalized
as \emph{precision} relations $ \sqsubseteq $ over types and terms) so
that an omitted type annotation is more precise than the annotation
with $ \star $ but less precise than an annotation with a static type.
So, for example,
\[
  \ottsym{(}   \lambda  \ottmv{x} \!:\!   \textsf{\textup{int}\relax}   \!\rightarrow\!   \textsf{\textup{int}\relax}   .\,  \ottmv{x}  \,  2   \ottsym{)} \, \ottsym{(}   \lambda  \ottmv{y} \!:\!   \textsf{\textup{int}\relax}   .\,  \ottmv{y}   \ottsym{)}  \sqsubseteq  
  \ottsym{(}   \lambda  \ottmv{x} .\,   2    \ottsym{)} \, \ottsym{(}   \lambda  \ottmv{y} \!:\!   \textsf{\textup{int}\relax}   .\,  \ottmv{y}   \ottsym{)}  \sqsubseteq 
  \ottsym{(}   \lambda  \ottmv{x} \!:\!  \star  .\,   2    \ottsym{)} \, \ottsym{(}   \lambda  \ottmv{y} \!:\!  \star  .\,  \ottmv{y}   \ottsym{)}.
\]
Intuitively, omitted type annotations are considered (fresh) type
variables, which are less specific than concrete types but they are
more precise than $ \star $ because they range only over static types.
We prove the gradual guarantee for the ITGL.  To our knowledge, the
gradual guarantee is proved for a language with let-polymorphism for
the first time.

Finally, we have implemented an interpreter of the ITGL, including an
implementation of Garcia and Cimini's type inference algorithm, a
translator to and an evaluator of \(\lambdaRTI\), in OCaml.  It
supports integer, Boolean, and unit types as base types, standard
arithmetic, comparison, and Boolean operators, conditional
expressions, and recursive definitions.  The source code is available
at \url{https://github.com/ymyzk/lambda-dti/}.

\paragraph{Contributions.}

Our contributions are summarized as follows:
\begin{itemize}
\item We propose DTI as a basis for new semantics of (an intermediate language for) the ITGL, an implicitly typed language with
   a gradual type system and Hindley--Milner polymorphism;
  \item We define a blame calculus \(\lambdaRTI\) with its syntax, type system, and operational semantics with DTI;
  \item We prove properties of \(\lambdaRTI\), including type safety,
    soundness and completeness of DTI, and the gradual guarantee;
  \item We also prove the gradual guarantee for the ITGL; and
  \item We have implemented an interpreter of the ITGL.
\end{itemize}

\paragraph{The organization of the paper.}
We define \(\lambdaRTI\) in Section~\ref{sec:language} and state its
basic properties, including type safety and conservative extension, in
Section~\ref{sec:properties}.  Then, we show soundness and
completeness of DTI in Section~\ref{sec:soundness-completeness} and
the gradual guarantee in Section~\ref{sec:gradual-guarantee}.
Finally, we discuss related work in Section~\ref{sec:related_work} and
conclude in Section~\ref{sec:conclusion}.  %
\iffull
Proofs of the stated properties are given in Appendix.
\else
All the stated properties have been proved but proofs are mostly omitted; see
the full version at \url{http://arxiv.org/abs/1810.12619}.
\fi

\section{$\lambdaRTI$: A Blame Calculus with Dynamic Type Inference}
\label{sec:language}

In this section, we develop a new blame calculus $\lambdaRTI$ with
dynamic type inference (DTI).  We start with a simply typed fragment and add
let-polymorphism in Section~\ref{sec:let_polymorphism}.  The core of the $\lambdaRTI$ calculus is
based on the calculus by \citet{DBLP:conf/snapl/SiekVCB15},
which is a simplified version of the blame calculus by
\citet{DBLP:conf/esop/WadlerF09} without refinement types.  We
augment its type system with type variables in a fairly
straightforward manner and operational semantics as described in the
last section.

Our blame calculus $\lambdaRTI$ is designed to be used as an
intermediate language to give semantics of the ITGL by
\citet{DBLP:conf/popl/GarciaC15}.  We will discuss the ITGL and
how ITGL programs are translated to $\lambdaRTI$ programs in
Section~\ref{sec:gradual-guarantee} in more detail, but as far as
the simply typed fragment is concerned it is very similar
to previous work~\cite{DBLP:conf/snapl/SiekVCB15,DBLP:conf/popl/GarciaC15}.

\subsection{Static Semantics}
We show the syntax of $\lambdaRTI$ in Figure~\ref{fig:def_syntax}.  The syntax
of the calculus extends that of the simply typed lambda calculus with the
dynamic type, casts, and type variables.

\begin{figure}[tb]
  {\small
    \input{figures/def_syntax}}
  \caption{Syntax of $\lambdaRTI$.}
  \label{fig:def_syntax}
\end{figure}

Gradual types, denoted by $\ottnt{U}$, consist of base types $\iota$
(such as $ \textsf{\textup{int}\relax} $ and $ \textsf{\textup{bool}\relax} $), type variables $\ottmv{X}$, the
dynamic type $ \star $, and function types $\ottnt{U}  \!\rightarrow\!  \ottnt{U}$.  Static types,
denoted by $\ottnt{T}$, are the subset of gradual types without the
dynamic type.  Ground types, which are used as tags to inject values
into the dynamic type, contain base types $\iota$ and the function
type $\star  \!\rightarrow\!  \star$.  We emphasize that type variables are \emph{not} in
ground types, because no value inhabits them (as we show in the canonical forms lemma (Lemma~\ref{lem:canonical_forms})) and we
do not need to use a type variable as a tag to inject values.

Terms, denoted by $\ottnt{f}$, consist of variables $\ottmv{x}$, constants $\ottnt{c}$,
primitive operations $\mathit{op} \, \ottsym{(}  \ottnt{f}  \ottsym{,}  \ottnt{f}  \ottsym{)}$, lambda abstractions $ \lambda  \ottmv{x} \!:\!  \ottnt{U}  .\,  \ottnt{f} $ (which bind $\ottmv{x}$ in $\ottnt{f}$),
applications $\ottnt{f} \, \ottnt{f}$, casts $\ottnt{f}  \ottsym{:}   \ottnt{U} \Rightarrow  \unskip ^ { \ell }  \! \ottnt{U} $, and blame $\textsf{\textup{blame}\relax} \, \ell$.  Since
this is an intermediate language, variables in abstractions are explicitly typed.  Casts
$\ottnt{f}  \ottsym{:}   \ottnt{U_{{\mathrm{1}}}} \Rightarrow  \unskip ^ { \ell }  \! \ottnt{U_{{\mathrm{2}}}} $ from $\ottnt{U_{{\mathrm{1}}}}$ to $\ottnt{U_{{\mathrm{2}}}}$ are inserted when
translating a term in the ITGL, and used for checking whether $\ottnt{f}$
of type $\ottnt{U_{{\mathrm{1}}}}$ can behave as type $\ottnt{U_{{\mathrm{2}}}}$ at run time.  Casts are annotated
also with blame labels, denoted by $\ell$, to indicate which cast has
failed; blame $\textsf{\textup{blame}\relax} \, \ell$ is used to denote a run-time failure of a
cast with $\ell$.  They have
\emph{polarity} to indicate which side of a cast to be
blamed~\cite{DBLP:conf/icfp/FindlerF02}.  For each blame label, there
is a negated blame label $ \bar{ \ell } $, which denotes the opposite side to $\ell$, and $ \bar{  \bar{ \ell }  }  = \ell$.  As we did in the introduction, we
often abbreviate a sequence of casts $\ottsym{(}  \ottnt{f}  \ottsym{:}   \ottnt{U_{{\mathrm{1}}}} \Rightarrow  \unskip ^ { \ell_{{\mathrm{1}}} }  \! \ottnt{U_{{\mathrm{2}}}}   \ottsym{)}  \ottsym{:}   \ottnt{U_{{\mathrm{2}}}} \Rightarrow  \unskip ^ { \ell_{{\mathrm{2}}} }  \! \ottnt{U_{{\mathrm{3}}}} $
to $\ottnt{f}  \ottsym{:}   \ottnt{U_{{\mathrm{1}}}} \Rightarrow  \unskip ^ { \ell_{{\mathrm{1}}} }  \!  \ottnt{U_{{\mathrm{2}}}} \Rightarrow  \unskip ^ { \ell_{{\mathrm{2}}} }  \! \ottnt{U_{{\mathrm{3}}}}  $.

Values, denoted by $w$, consist of constants, lambda abstractions, wrapped
functions, and injections.  A wrapped function is a function value
enclosed in the cast between function types.  Results, denoted by
$\ottnt{r}$, are values and blame.  Evaluation contexts, denoted by
$\ottnt{E}$, are standard and they mean that a term is evaluated from left
to right, under call-by-value.

We show the static semantics of $\lambdaRTI$ in Figure~\ref{fig:def_static}.
It consists of type consistency and typing.

Type consistency rules define the type consistency relation $\ottnt{U}  \sim  \ottnt{U'}$, a key
notion in gradual typing, over gradual types.  Intuitively, $\ottnt{U}  \sim  \ottnt{U'}$
means that it is possible for a cast from $\ottnt{U}$ to $\ottnt{U'}$ to succeed.
The rules
\rnp{C\_Base} and \rnp{C\_TyVar} mean that a base type and a type
variable, respectively, is consistent with itself.  The rules
\rnp{C\_DynL} and \rnp{C\_DynR} mean that all gradual types are
consistent with the dynamic type.  The rule \rnp{C\_Arrow} means that
two function types are consistent if their domain types are consistent
and so are their range types.  The type consistency relation is
reflexive and symmetric but not transitive.


Typing rules of the $\lambdaRTI$ extend those of the simply typed
lambda calculus.  The rules \rnp{T\_Var}, \rnp{T\_Const}, \rnp{T\_Op},
\rnp{T\_Abs}, and \rnp{T\_App} are standard.  $ \mathit{ty} ( \ottnt{c} ) $ used in \rnp{T\_Const}
assigns a base type to each constant $\ottnt{c}$, and $ \mathit{ty} ( \mathit{op} ) $ used in \rnp{T\_Op}
assigns a first-order static type without type variables to each operator
$ \mathit{op} $.  The rule \rnp{T\_Cast} allows a term to be cast to a consistent
type.

\begin{figure}[tb]
  {\small
  \input{figures/def_static}}
  \caption{Static semantics of $\lambdaRTI$.}
  \label{fig:def_static}
\end{figure}

A type substitution, denoted by $S$, is a finite mapping from type
variables to static types.  The empty mapping is denoted by $[  ]$; and
the composition of two type substitutions $S_{{\mathrm{1}}}$ and $S_{{\mathrm{2}}}$ is by $ S_{{\mathrm{1}}}  \circ  S_{{\mathrm{2}}} $.
Application $S  \ottsym{(}  \ottnt{f}  \ottsym{)}$ and $S  \ottsym{(}  \ottnt{U}  \ottsym{)}$ of type substitution $S$ to term
$\ottnt{f}$ and type $\ottnt{U}$ are defined in the usual way respectively.
We write $[   \overrightarrow{ \ottmv{X} }   :=   \overrightarrow{ \ottnt{T} }   ]$ for a type substitution that maps type variables
$ \overrightarrow{ \ottmv{X} } $ to types $ \overrightarrow{ \ottnt{T} } $ respectively and $\ottnt{f}  [   \overrightarrow{ \ottmv{X} }   \ottsym{:=}   \overrightarrow{ \ottnt{T} }   ]$ and $\ottnt{U}  [   \overrightarrow{ \ottmv{X} }   \ottsym{:=}   \overrightarrow{ \ottnt{T} }   ]$
for $[   \overrightarrow{ \ottmv{X} }   :=   \overrightarrow{ \ottnt{T} }   ]  \ottsym{(}  \ottnt{f}  \ottsym{)}$ and $[   \overrightarrow{ \ottmv{X} }   :=   \overrightarrow{ \ottnt{T} }   ]  \ottsym{(}  \ottnt{U}  \ottsym{)}$.
Note that the codomain of a type substitution is static
types, following \citet{DBLP:conf/popl/GarciaC15},
in which a type variable represents a placeholder for a static type.

As expected, type substitution preserves consistency and typing:
\begin{lemma}[name=Type Substitution Preserves Consistency and Typing,restate=lemTySubstConsistency] \leavevmode
  \begin{enumerate}
    \item If $\ottnt{U}  \sim  \ottnt{U'}$, then $S  \ottsym{(}  \ottnt{U}  \ottsym{)}  \sim  S  \ottsym{(}  \ottnt{U'}  \ottsym{)}$ for any $S$.
    \item If $\Gamma  \vdash  \ottnt{f}  \ottsym{:}  \ottnt{U}$, then $S  \ottsym{(}  \Gamma  \ottsym{)}  \vdash  S  \ottsym{(}  \ottnt{f}  \ottsym{)}  \ottsym{:}  S  \ottsym{(}  \ottnt{U}  \ottsym{)}$ for any $S$.
  \end{enumerate}
\end{lemma}

\subsection{Dynamic Semantics}

We show the dynamic semantics of $\lambdaRTI$ in
Figure~\ref{fig:def_dynamic}.  It is given in a small-step style by using two relations over
terms.  One is the reduction relation $\ottnt{f} \,  \xrightarrow{ \mathmakebox[0.4em]{} S \mathmakebox[0.3em]{} }  \, \ottnt{f'}$, which
represents a basic computation step, including dynamic checking by
casts and DTI.  The other is the evaluation relation $\ottnt{f} \,  \xmapsto{ \mathmakebox[0.4em]{} S \mathmakebox[0.3em]{} }  \, \ottnt{f'}$,
which represents top-level execution.  Both relations are annotated
with a type substitution $S$, which is generated by DTI.

\begin{figure}
  {\small
    \input{figures/def_dynamic}}
  \caption{Dynamic semantics of $\lambdaRTI$.}
  \label{fig:def_dynamic}
\end{figure}

\subsubsection{Basic Reduction Rules}

We first explain rules from the basic blame calculus, where type
substitutions are empty.  The rule \rnp{R\_Op} is for primitive
operations; the meta-function $ \llbracket\mathit{op}\rrbracket $ gives a meaning to the primitive
operation $ \mathit{op} $ and we assume that $ \llbracket\mathit{op}\rrbracket ( w_{{\mathrm{1}}} ,  w_{{\mathrm{2}}} ) $ returns a
value of the right type, i.e., the return type of $ \mathit{ty} ( \mathit{op} ) $.
The rule \rnp{R\_Beta} performs the standard \(\beta\)-reduction.  We write
$\ottnt{f}  [  \ottmv{x}  \ottsym{:=}  w  ]$ for the term obtained by substituting $w$ for $\ottmv{x}$ in $\ottnt{f}$;
term substitution is defined in a capture-avoiding manner as usual.
The rules
\rnp{R\_IdBase} and \rnp{R\_IdStar} discard identity casts on a base
type and the dynamic type, respectively.  The rules \rnp{R\_Succeed}
and \rnp{R\_Fail} check two casts where an injection (a cast from a
ground type to the dynamic type) meets a projection (a cast from the
dynamic type to a ground type).  If both ground types are equal, then
the projection succeeds and these casts are discarded.  Otherwise,
these casts fail and reduce to blame to abort execution of the program
with blame label $\ell_{{\mathrm{2}}}$ (the one on the projection).  The rule
\rnp{R\_AppCast} reduces an application of a wrapped function by
breaking the cast into two.  One is a cast on the argument, and the
other is a cast on the return value.  We negate the blame label for
the cast on the argument because the direction of the cast is swapped
from that of the function cast~\cite{DBLP:conf/icfp/FindlerF02}.  The
rules \rnp{R\_Ground} and \rnp{R\_Expand} decompose a cast between a
non-ground type and the dynamic type into two casts.  These rules
cannot be used if $\ottnt{U}$ is a type variable because type variables
are never consistent with any ground type.  The ground type in the
middle of the resulting two casts is uniquely determined:

\begin{lemma}[name=Ground Types,restate=lemGroundTypes] \label{lem:ground_types}
  \leavevmode
  \begin{enumerate}
    \item If $\ottnt{U}$ is neither a type variable nor the dynamic type, then there exists a unique $\ottnt{G}$ such that $\ottnt{U}  \sim  \ottnt{G}$.
    \item $\ottnt{G}  \sim  \ottnt{G'}$ if and only if $\ottnt{G}  \ottsym{=}  \ottnt{G'}$.
  \end{enumerate}
\end{lemma}


\subsubsection{Reduction Rules for Dynamic Type Inference}



As discussed in the introduction, our idea is to infer the
``value'' of a type variable when it is projected from the dynamic
type, as in $w  \ottsym{:}   \star \Rightarrow  \unskip ^ { \ell }  \! \ottmv{X} $.  The value typed at the dynamic type
is always tagged with a ground type and is of the form
$w'  \ottsym{:}   \ottnt{G} \Rightarrow  \unskip ^ { \ell' }  \! \star $.  If the ground type $\ottnt{G}$ is a base type, then
the type variable will be instantiated with it.
We revisit the example used in Section~\ref{sec:introduction} and show
evaluation steps below.
\[
  \begin{array}{ll}
    & \ottsym{(}   \lambda  \ottmv{x} \!:\!  \star  .\,  \ottsym{(}  \ottmv{x}  \ottsym{:}   \star \Rightarrow  \unskip ^ { \ell_{{\mathrm{1}}} }  \! \star  \!\rightarrow\!  \star   \ottsym{)}  \, \ottsym{(}   2   \ottsym{:}    \textsf{\textup{int}\relax}  \Rightarrow  \unskip ^ { \ell_{{\mathrm{2}}} }  \! \star   \ottsym{)}  \ottsym{)} \, \ottsym{(}  \ottsym{(}   \lambda  \ottmv{y} \!:\!  \ottmv{Y}  .\,  \ottmv{y}   \ottsym{)}  \ottsym{:}   \ottmv{Y}  \!\rightarrow\!  \ottmv{Y} \Rightarrow  \unskip ^ { \ell_{{\mathrm{3}}} }  \! \star   \ottsym{)} \\
     \xmapsto{ \mathmakebox[0.4em]{} [  ] \mathmakebox[0.3em]{} }\hspace{-0.4em}{}^\ast \hspace{0.2em}  & \ottsym{(}  \ottsym{(}   \lambda  \ottmv{y} \!:\!  \ottmv{Y}  .\,  \ottmv{y}   \ottsym{)}  \ottsym{:}   \ottmv{Y}  \!\rightarrow\!  \ottmv{Y} \Rightarrow  \unskip ^ { \ell_{{\mathrm{3}}} }  \! \star  \!\rightarrow\!  \star   \ottsym{)} \, \ottsym{(}   2   \ottsym{:}    \textsf{\textup{int}\relax}  \Rightarrow  \unskip ^ { \ell_{{\mathrm{2}}} }  \! \star   \ottsym{)} \\
     \xmapsto{ \mathmakebox[0.4em]{} [  ] \mathmakebox[0.3em]{} }  & \ottsym{(}  \ottsym{(}   \lambda  \ottmv{y} \!:\!  \ottmv{Y}  .\,  \ottmv{y}   \ottsym{)} \, \ottsym{(}   2   \ottsym{:}    \textsf{\textup{int}\relax}  \Rightarrow  \unskip ^ { \ell_{{\mathrm{2}}} }  \!  \star \Rightarrow  \unskip ^ {  \bar{ \ell_{{\mathrm{3}}} }  }  \! \ottmv{Y}    \ottsym{)}  \ottsym{)}  \ottsym{:}   \ottmv{Y} \Rightarrow  \unskip ^ { \ell_{{\mathrm{3}}} }  \! \star  \\
     \xmapsto{ \mathmakebox[0.4em]{} [  \ottmv{Y}  :=   \textsf{\textup{int}\relax}   ] \mathmakebox[0.3em]{} }  & \ottsym{(}  \ottsym{(}   \lambda  \ottmv{y} \!:\!   \graybox{   \textsf{\textup{int}\relax}   }   .\,  \ottmv{y}   \ottsym{)} \,  2   \ottsym{)}  \ottsym{:}    \graybox{   \textsf{\textup{int}\relax}   }  \Rightarrow  \unskip ^ { \ell_{{\mathrm{3}}} }  \! \star  \\
     \xmapsto{ \mathmakebox[0.4em]{} [  ] \mathmakebox[0.3em]{} }  &  2   \ottsym{:}    \textsf{\textup{int}\relax}  \Rightarrow  \unskip ^ { \ell_{{\mathrm{3}}} }  \! \star 
  \end{array}
\]
The subterm $ 2   \ottsym{:}    \textsf{\textup{int}\relax}  \Rightarrow  \unskip ^ { \ell_{{\mathrm{2}}} }  \!  \star \Rightarrow  \unskip ^ {  \bar{ \ell_{{\mathrm{3}}} }  }  \! \ottmv{Y}  $ on the third line reduces to
$ 2 $---this is where DTI is performed; this reduction step is
annotated with a type substitution $[  \ottmv{Y}  :=   \textsf{\textup{int}\relax}   ]$, which roughly
means ``$\ottmv{Y}$ must be $ \textsf{\textup{int}\relax} $ for the execution of the
program to proceed further without blame.''  As we will explain soon,
the type substitution is applied to the whole term, since the type
variable $\ottmv{Y}$ may appear elsewhere in the term.  As a result, the
occurrences of $\ottmv{Y}$ in the type annotation and the last cast are
also instantiated (as the shade on $ \textsf{\textup{int}\relax} $ indicates) and so
$\ottsym{(}  \ottsym{(}   \lambda  \ottmv{y} \!:\!  \ottmv{Y}  .\,  \ottmv{y}   \ottsym{)} \, \ottsym{(}   2   \ottsym{:}    \textsf{\textup{int}\relax}  \Rightarrow  \unskip ^ { \ell_{{\mathrm{2}}} }  \!  \star \Rightarrow  \unskip ^ {  \bar{ \ell_{{\mathrm{3}}} }  }  \! \ottmv{Y}    \ottsym{)}  \ottsym{)}  \ottsym{:}   \ottmv{Y} \Rightarrow  \unskip ^ { \ell_{{\mathrm{3}}} }  \! \star $ evaluates to
$\ottsym{(}  \ottsym{(}   \lambda  \ottmv{y} \!:\!   \textsf{\textup{int}\relax}   .\,  \ottmv{y}   \ottsym{)} \,  2   \ottsym{)}  \ottsym{:}    \textsf{\textup{int}\relax}  \Rightarrow  \unskip ^ { \ell_{{\mathrm{3}}} }  \! \star $ in one step.

%
Now, we explain formal reduction rules for DTI.  The rule
\rnp{R\_InstBase} instantiates a type variable $\ottmv{X}$ with the base
type $\iota$ and generates a type substitution $[  \ottmv{X}  :=  \iota  ]$.
The rule \rnp{R\_InstArrow} instantiates a type variable $\ottmv{X}$ with a
function type $\ottmv{X_{{\mathrm{1}}}}  \!\rightarrow\!  \ottmv{X_{{\mathrm{2}}}}$ for fresh type variables $\ottmv{X_{{\mathrm{1}}}}$ and $\ottmv{X_{{\mathrm{2}}}}$.%
\footnote{We use the term ``fresh'' here to mean that
$\ottmv{X_{{\mathrm{1}}}}$ and $\ottmv{X_{{\mathrm{2}}}}$ occur nowhere in the whole program before reduction.}
At this
point, we know that $w$ is a (possibly wrapped) function, but the
domain and range types are still unknown.  We defer the decision about these types by
generating fresh type variables, which will be instantiated in the future evaluation.

Finally, we explain the evaluation rules.  The rule \rnp{E\_Step} reduces
the subterm in an evaluation context, then apply the generated type
substitution.  The substitution is applied to the whole term
$\ottnt{E}  [  \ottnt{f'}  ]$ so that the other occurrences of the same variables are
replaced at once.  The rule \rnp{E\_Abort} aborts execution of a program
if it raises blame.

We write $\ottnt{f_{{\mathrm{0}}}} \,  \xmapsto{ \mathmakebox[0.4em]{} S \mathmakebox[0.3em]{} }\hspace{-0.4em}{}^\ast \hspace{0.2em}  \, \ottnt{f_{\ottmv{n}}}$ if $\ottnt{f_{{\mathrm{0}}}} \,  \xmapsto{ \mathmakebox[0.4em]{} S_{{\mathrm{1}}} \mathmakebox[0.3em]{} }  \, \ottnt{f_{{\mathrm{1}}}}$,
$\ottnt{f_{{\mathrm{1}}}} \,  \xmapsto{ \mathmakebox[0.4em]{} S_{{\mathrm{2}}} \mathmakebox[0.3em]{} }  \, \ottnt{f_{{\mathrm{2}}}}$, \dots, and $\ottnt{f_{{\ottmv{n}-1}}} \,  \xmapsto{ \mathmakebox[0.4em]{} S_{\ottmv{n}} \mathmakebox[0.3em]{} }  \, \ottnt{f_{\ottmv{n}}}$ and
$S =    S_{\ottmv{n}}  \circ  S_{{\ottmv{n}-1}}   \circ   \cdots    \circ  S_{{\mathrm{1}}} $ (where $n \geq 0$) and similarly for
$\ottnt{f_{{\mathrm{0}}}} \,  \xmapsto{ \mathmakebox[0.4em]{} S \mathmakebox[0.3em]{} }\hspace{-0.4em}{}^+ \hspace{0.2em}  \, \ottnt{f_{\ottmv{n}}}$ (where $n \geq 1$).

One may wonder that the rule \rnp{R\_InstArrow} is redundant
because a term $w  \ottsym{:}   \star  \!\rightarrow\!  \star \Rightarrow  \unskip ^ { \ell_{{\mathrm{1}}} }  \!  \star \Rightarrow  \unskip ^ { \ell_{{\mathrm{2}}} }  \!  \star  \!\rightarrow\!  \star \Rightarrow  \unskip ^ { \ell_{{\mathrm{2}}} }  \! X_{{\mathrm{1}}}  \!\rightarrow\!  X_{{\mathrm{2}}}   $ always reduces to
$w  \ottsym{:}   \star  \!\rightarrow\!  \star \Rightarrow  \unskip ^ { \ell_{{\mathrm{2}}} }  \! X_{{\mathrm{1}}}  \!\rightarrow\!  X_{{\mathrm{2}}} $ in the next step.  Actually, we could define this
reduction rule in the following two different manners:
\begin{itemize}
  \item $w  \ottsym{:}   \star  \!\rightarrow\!  \star \Rightarrow  \unskip ^ { \ell_{{\mathrm{1}}} }  \!  \star \Rightarrow  \unskip ^ { \ell_{{\mathrm{2}}} }  \! \ottmv{X}   \,  \xrightarrow{ \mathmakebox[0.4em]{} [  \ottmv{X}  :=  \ottmv{X_{{\mathrm{1}}}}  \!\rightarrow\!  \ottmv{X_{{\mathrm{2}}}}  ] \mathmakebox[0.3em]{} }  \, w  \ottsym{:}   \star  \!\rightarrow\!  \star \Rightarrow  \unskip ^ { \ell_{{\mathrm{2}}} }  \! \ottmv{X_{{\mathrm{1}}}}  \!\rightarrow\!  \ottmv{X_{{\mathrm{2}}}} $; or
  \item $w  \ottsym{:}   \star  \!\rightarrow\!  \star \Rightarrow  \unskip ^ { \ell_{{\mathrm{1}}} }  \!  \star \Rightarrow  \unskip ^ { \ell_{{\mathrm{2}}} }  \! \ottmv{X}   \,  \xrightarrow{ \mathmakebox[0.4em]{} [  \ottmv{X}  :=  \ottmv{X_{{\mathrm{1}}}}  \!\rightarrow\!  \ottmv{X_{{\mathrm{2}}}}  ] \mathmakebox[0.3em]{} }  \, w  \ottsym{:}   \star  \!\rightarrow\!  \star \Rightarrow  \unskip ^ { \ell_{{\mathrm{1}}} }  \!  \star \Rightarrow  \unskip ^ { \ell_{{\mathrm{2}}} }  \! \ottmv{X_{{\mathrm{1}}}}  \!\rightarrow\!  \ottmv{X_{{\mathrm{2}}}}  $.
\end{itemize}
Using these rules does not change the semantics of the language, but
we choose \rnp{R\_InstArrow} for ease of proofs.

We show how the rule \rnp{R\_InstArrow} works using the following
(somewhat contrived\footnote{In fact, this term is not in the image of
  the cast-inserting translation.  An example in the image would be much
  more complicated.}) example:
\[
  \begin{array}{ll}
    & \ottsym{(}  \ottsym{(}   \lambda  \ottmv{y} \!:\!   \textsf{\textup{int}\relax}   .\,  \ottmv{y}   \ottsym{+}   1   \ottsym{)}  \ottsym{:}    \textsf{\textup{int}\relax}   \!\rightarrow\!   \textsf{\textup{int}\relax}  \Rightarrow  \unskip ^ { \ell_{{\mathrm{1}}} }  \!  \star \Rightarrow  \unskip ^ { \ell_{{\mathrm{2}}} }  \!  \ottmv{X} \Rightarrow  \unskip ^ { \ell_{{\mathrm{3}}} }  \!  \star \Rightarrow  \unskip ^ { \ell_{{\mathrm{4}}} }  \! \star  \!\rightarrow\!  \star      \ottsym{)} \, \ottsym{(}   3   \ottsym{:}    \textsf{\textup{int}\relax}  \Rightarrow  \unskip ^ { \ell_{{\mathrm{5}}} }  \! \star   \ottsym{)} \\
     \xmapsto{ \mathmakebox[0.4em]{} [  ] \mathmakebox[0.3em]{} }  & \ottsym{(}  w  \ottsym{:}   \star  \!\rightarrow\!  \star \Rightarrow  \unskip ^ { \ell_{{\mathrm{1}}} }  \!  \star \Rightarrow  \unskip ^ { \ell_{{\mathrm{2}}} }  \!  \ottmv{X} \Rightarrow  \unskip ^ { \ell_{{\mathrm{3}}} }  \!  \star \Rightarrow  \unskip ^ { \ell_{{\mathrm{4}}} }  \! \star  \!\rightarrow\!  \star      \ottsym{)} \, \ottsym{(}   3   \ottsym{:}    \textsf{\textup{int}\relax}  \Rightarrow  \unskip ^ { \ell_{{\mathrm{5}}} }  \! \star   \ottsym{)} \\
                & \multicolumn{1}{r}{\textrm{where }w = \ottsym{(}   \lambda  \ottmv{y} \!:\!   \textsf{\textup{int}\relax}   .\,  \ottmv{y}   \ottsym{+}   1   \ottsym{)}  \ottsym{:}    \textsf{\textup{int}\relax}   \!\rightarrow\!   \textsf{\textup{int}\relax}  \Rightarrow  \unskip ^ { \ell_{{\mathrm{1}}} }  \! \star  \!\rightarrow\!  \star } \\
     \xmapsto{ \mathmakebox[0.4em]{} S \mathmakebox[0.3em]{} }  & \ottsym{(}  w  \ottsym{:}   \star  \!\rightarrow\!  \star \Rightarrow  \unskip ^ { \ell_{{\mathrm{1}}} }  \!  \star \Rightarrow  \unskip ^ { \ell_{{\mathrm{2}}} }  \!  \star  \!\rightarrow\!  \star \Rightarrow  \unskip ^ { \ell_{{\mathrm{2}}} }  \!   \graybox{  \ottmv{X_{{\mathrm{1}}}}  \!\rightarrow\!  \ottmv{X_{{\mathrm{2}}}}  }  \Rightarrow  \unskip ^ { \ell_{{\mathrm{3}}} }  \!  \star \Rightarrow  \unskip ^ { \ell_{{\mathrm{4}}} }  \! \star  \!\rightarrow\!  \star       \ottsym{)} \, \ottsym{(}   3   \ottsym{:}    \textsf{\textup{int}\relax}  \Rightarrow  \unskip ^ { \ell_{{\mathrm{5}}} }  \! \star   \ottsym{)} \\
                & \multicolumn{1}{r}{\textrm{where }S  \ottsym{=}  [  \ottmv{X}  :=  \ottmv{X_{{\mathrm{1}}}}  \!\rightarrow\!  \ottmv{X_{{\mathrm{2}}}}  ]} \\
     \xmapsto{ \mathmakebox[0.4em]{} S' \mathmakebox[0.3em]{} }\hspace{-0.4em}{}^\ast \hspace{0.2em}  &  4   \ottsym{:}    \textsf{\textup{int}\relax}  \Rightarrow  \unskip ^ { \ell_{{\mathrm{4}}} }  \! \star  \qquad \textrm{where }S'  \ottsym{=}  [  \ottmv{X_{{\mathrm{1}}}}  :=   \textsf{\textup{int}\relax}   \ottsym{,}  \ottmv{X_{{\mathrm{2}}}}  :=   \textsf{\textup{int}\relax}   ].
  \end{array}
\]
The type variable $\ottmv{X}$ is instantiated with $\ottmv{X_{{\mathrm{1}}}}  \!\rightarrow\!  \ottmv{X_{{\mathrm{2}}}}$ when the
value $w$ tagged with $\star  \!\rightarrow\!  \star$ is projected to $\ottmv{X}$.
Then, $\ottmv{X_{{\mathrm{1}}}}$ and $\ottmv{X_{{\mathrm{2}}}}$ are instantiated with $ \textsf{\textup{int}\relax} $ by
\rnp{R\_InstBase}, as we have already explained, by the time the
term evaluates to the final value $ 4   \ottsym{:}    \textsf{\textup{int}\relax}  \Rightarrow  \unskip ^ { \ell_{{\mathrm{4}}} }  \! \star $.

Perhaps surprisingly, our reduction rules to instantiate type
variables are not symmetric: There are no rules to reduce terms such
as $w  \ottsym{:}   \ottmv{X} \Rightarrow  \unskip ^ { \ell_{{\mathrm{1}}} }  \!  \star \Rightarrow  \unskip ^ { \ell_{{\mathrm{2}}} }  \! \iota  $, $w  \ottsym{:}   \ottmv{X} \Rightarrow  \unskip ^ { \ell_{{\mathrm{1}}} }  \!  \star \Rightarrow  \unskip ^ { \ell_{{\mathrm{2}}} }  \! \star  \!\rightarrow\!  \star  $,
$w  \ottsym{:}   \ottmv{X} \Rightarrow  \unskip ^ { \ell_{{\mathrm{1}}} }  \!  \star \Rightarrow  \unskip ^ { \ell_{{\mathrm{2}}} }  \! \ottmv{X'}  $, and $w  \ottsym{:}   \ottmv{X} \Rightarrow  \unskip ^ { \ell_{{\mathrm{1}}} }  \! \ottmv{X} $, even though a cast
expression such as $\ottnt{f}  \ottsym{:}   \ottmv{X} \Rightarrow  \unskip ^ { \ell }  \! \star $ \emph{does} appear during
reduction.  This is because a value is not typed at a type variable as we will show in the canonical forms lemma (Lemma~\ref{lem:canonical_forms})
and we do not need the rules to reduce these
terms.  The $\ottmv{X}$ in $\ottnt{f}  \ottsym{:}   \ottmv{X} \Rightarrow  \unskip ^ { \ell }  \! \star $ will be instantiated
during evaluation of $\ottnt{f}$.

Before closing this subsection, we revisit an example from the
introduction.  It raises blame because one type variable is used in
contexts that expect different types:
\[
  \begin{array}{l@{\;}l}
& \ottsym{(}   \lambda  \ottmv{x} \!:\!  \star  \!\rightarrow\!  \star  \!\rightarrow\!  \star  .\,  \ottmv{x}  \, \ottsym{(}   2   \ottsym{:}    \textsf{\textup{int}\relax}  \Rightarrow  \unskip ^ { \ell_{{\mathrm{1}}} }  \! \star   \ottsym{)} \, \ottsym{(}   \textsf{\textup{true}\relax}   \ottsym{:}    \textsf{\textup{bool}\relax}  \Rightarrow  \unskip ^ { \ell_{{\mathrm{2}}} }  \! \star   \ottsym{)}  \ottsym{)} \, w
    \\ & \multicolumn{1}{r}{\text{where } w  \ottsym{=}  \ottsym{(}   \lambda  \ottmv{y_{{\mathrm{1}}}} \!:\!  \ottmv{Y}  .\,   \lambda  \ottmv{y_{{\mathrm{2}}}} \!:\!  \ottmv{Y}  .\,   \textsf{\textup{ if }\relax}  \ottnt{b}  \textsf{\textup{ then }\relax}  \ottmv{y_{{\mathrm{1}}}}  \textsf{\textup{ else }\relax}  \ottmv{y_{{\mathrm{2}}}}     \ottsym{)}  \ottsym{:}   \ottmv{Y}  \!\rightarrow\!  \ottmv{Y}  \!\rightarrow\!  \ottmv{Y} \Rightarrow  \unskip ^ { \ell_{{\mathrm{3}}} }  \! \star  \!\rightarrow\!  \star  \!\rightarrow\!  \star } \\
     \xmapsto{ \mathmakebox[0.4em]{} [  ] \mathmakebox[0.3em]{} }  & w \, \ottsym{(}   2   \ottsym{:}    \textsf{\textup{int}\relax}  \Rightarrow  \unskip ^ { \ell_{{\mathrm{1}}} }  \! \star   \ottsym{)} \, \ottsym{(}   \textsf{\textup{true}\relax}   \ottsym{:}    \textsf{\textup{bool}\relax}  \Rightarrow  \unskip ^ { \ell_{{\mathrm{2}}} }  \! \star   \ottsym{)} \\
     \xmapsto{ \mathmakebox[0.4em]{} [  ] \mathmakebox[0.3em]{} }  & (\ottsym{(}  \ottsym{(}   \lambda  \ottmv{y_{{\mathrm{1}}}} \!:\!  \ottmv{Y}  .\,   \lambda  \ottmv{y_{{\mathrm{2}}}} \!:\!  \ottmv{Y}  .\,   \textsf{\textup{ if }\relax}  \ottnt{b}  \textsf{\textup{ then }\relax}  \ottmv{y_{{\mathrm{1}}}}  \textsf{\textup{ else }\relax}  \ottmv{y_{{\mathrm{2}}}}     \ottsym{)} \, \ottsym{(}   2   \ottsym{:}    \textsf{\textup{int}\relax}  \Rightarrow  \unskip ^ { \ell_{{\mathrm{1}}} }  \!  \star \Rightarrow  \unskip ^ {  \bar{ \ell_{{\mathrm{3}}} }  }  \! \ottmv{Y}    \ottsym{)}  \ottsym{)}  \ottsym{:}   \ottmv{Y}  \!\rightarrow\!  \ottmv{Y} \Rightarrow  \unskip ^ { \ell_{{\mathrm{3}}} }  \! \star  \!\rightarrow\!  \star ) \\ & \qquad \ottsym{(}   \textsf{\textup{true}\relax}   \ottsym{:}    \textsf{\textup{bool}\relax}  \Rightarrow  \unskip ^ { \ell_{{\mathrm{2}}} }  \! \star   \ottsym{)} \\
     \xmapsto{ \mathmakebox[0.4em]{} [  \ottmv{Y}  :=   \textsf{\textup{int}\relax}   ] \mathmakebox[0.3em]{} }  & \ottsym{(}  \ottsym{(}  \ottsym{(}   \lambda  \ottmv{y_{{\mathrm{1}}}} \!:\!   \graybox{   \textsf{\textup{int}\relax}   }   .\,   \lambda  \ottmv{y_{{\mathrm{2}}}} \!:\!   \graybox{   \textsf{\textup{int}\relax}   }   .\,   \textsf{\textup{ if }\relax}  \ottnt{b}  \textsf{\textup{ then }\relax}  \ottmv{y_{{\mathrm{1}}}}  \textsf{\textup{ else }\relax}  \ottmv{y_{{\mathrm{2}}}}     \ottsym{)} \,  2   \ottsym{)}  \ottsym{:}    \graybox{   \textsf{\textup{int}\relax}   }   \!\rightarrow\!   \graybox{   \textsf{\textup{int}\relax}   }  \Rightarrow  \unskip ^ { \ell_{{\mathrm{3}}} }  \! \star  \!\rightarrow\!  \star   \ottsym{)} \, \ottsym{(}   \textsf{\textup{true}\relax}   \ottsym{:}    \textsf{\textup{bool}\relax}  \Rightarrow  \unskip ^ { \ell_{{\mathrm{2}}} }  \! \star   \ottsym{)} \\
     \longmapsto^\ast  & \ottsym{(}  \ottsym{(}   \lambda  \ottmv{y_{{\mathrm{2}}}} \!:\!   \textsf{\textup{int}\relax}   .\,   \textsf{\textup{ if }\relax}  \ottnt{b}  \textsf{\textup{ then }\relax}   2   \textsf{\textup{ else }\relax}  \ottmv{y_{{\mathrm{2}}}}    \ottsym{)} \, \ottsym{(}   \textsf{\textup{true}\relax}   \ottsym{:}    \textsf{\textup{bool}\relax}  \Rightarrow  \unskip ^ { \ell_{{\mathrm{2}}} }  \!  \star \Rightarrow  \unskip ^ {  \bar{ \ell_{{\mathrm{3}}} }  }  \!  \textsf{\textup{int}\relax}     \ottsym{)}  \ottsym{)}  \ottsym{:}    \textsf{\textup{int}\relax}  \Rightarrow  \unskip ^ { \ell_{{\mathrm{3}}} }  \! \star  \\
     \longmapsto  & \textsf{\textup{blame}\relax} \,  \bar{ \ell_{{\mathrm{3}}} } 
  \end{array}
\]
As for the first example, at the third step, the subterm
$( 2   \ottsym{:}    \textsf{\textup{int}\relax}  \Rightarrow  \unskip ^ { \ell_{{\mathrm{1}}} }  \!  \star \Rightarrow  \unskip ^ {  \bar{ \ell_{{\mathrm{3}}} }  }  \! \ottmv{Y}  )$ reduces to $ 2 $ by \rnp{R\_InstBase} and
a substitution $[  \ottmv{Y}  :=   \textsf{\textup{int}\relax}   ]$ is yielded.  Then, by application of
\rnp{E\_Step}, $\ottmv{Y}$ in the evaluation context also gets replaced
with $ \textsf{\textup{int}\relax} $.  (Again, the shaded types indicate which parts are
affected.)  This term eventually evaluates to blame because the cast
$( \textsf{\textup{true}\relax}   \ottsym{:}    \textsf{\textup{bool}\relax}  \Rightarrow  \unskip ^ { \ell_{{\mathrm{2}}} }  \!  \star \Rightarrow  \unskip ^ {  \bar{ \ell_{{\mathrm{3}}} }  }  \!  \textsf{\textup{int}\relax}   )$ fails.  After all, the function
\( \lambda  \ottmv{y_{{\mathrm{1}}}} .\,   \lambda  \ottmv{y_{{\mathrm{2}}}} .\,   \textsf{\textup{ if }\relax}  \ottnt{b}  \textsf{\textup{ then }\relax}  \ottmv{y_{{\mathrm{1}}}}  \textsf{\textup{ else }\relax}  \ottmv{y_{{\mathrm{2}}}}   \) cannot be used as both $ \textsf{\textup{int}\relax}   \!\rightarrow\!   \textsf{\textup{int}\relax}   \!\rightarrow\!   \textsf{\textup{int}\relax} $ and
$ \textsf{\textup{bool}\relax}   \!\rightarrow\!   \textsf{\textup{bool}\relax}   \!\rightarrow\!   \textsf{\textup{bool}\relax} $ at the same time.  It is important to ensure
that $\ottmv{Y}$ is instantiated \emph{at most once}.



\subsection{Let-Polymorphism} \label{sec:let_polymorphism}

In this section, we extend $\lambdaRTI$ to
let-polymorphism~\cite{Milner78JCSS} by introducing type schemes and
explicit type abstraction $ \Lambda    \overrightarrow{ \ottmv{X} }  .\,  w $ and application $\ottmv{x}  [   \overrightarrow{ \ottnt{T} }   ]$,
as \textit{Core}-XML~\cite{HarperMitchell93TOPLAS}.
(Here, we abbreviate sequences by using vector notations.)
Explicit type abstraction/application is needed because we need type
information at run time.

Our formulation is fairly standard but we will find a few twists that
are motivated by our working hypothesis that $ \textsf{\textup{let}\relax} \,  \ottmv{x}  =  w  \textsf{\textup{ in }\relax}  \ottnt{f} $ (in
the surface language) should behave, both statically and dynamically,
the same as $\ottnt{f}$ where $w$ is substituted for
$\ottmv{x}$---especially in languages where type abstractions and
applications are implicit.  Consider the following expression in the
surface language (extended with pairs):
\[
   \textsf{\textup{let}\relax} \,  g  =   \lambda  \ottmv{x} .\,  \ottsym{(}   \ottsym{(}   \lambda  \ottmv{y} .\,  \ottmv{y}   \ottsym{)}  ::  \star  \!\rightarrow\!  \star   \ottsym{)}  \, \ottmv{x}  \textsf{\textup{ in }\relax}  \ottsym{(}  g \,  2   \ottsym{,}  g \,  \textsf{\textup{true}\relax}   \ottsym{)} 
\]
where $ \ottsym{(}   \lambda  \ottmv{y} .\,  \ottmv{y}   \ottsym{)}  ::  \star  \!\rightarrow\!  \star $ stands for ascription, which is translated
to a cast.  By making casts, type abstraction, and type application
explicit, one would obtain something like
\[
   \textsf{\textup{let}\relax} \,  g  =   \Lambda   \ottmv{X}   \ottmv{Y} .\,   \lambda  \ottmv{x} \!:\!  \ottmv{X}  .\,  \ottsym{(}  \ottsym{(}   \lambda  \ottmv{y} \!:\!  \ottmv{Y}  .\,  \ottmv{y}   \ottsym{)}  \ottsym{:}   \ottmv{Y}  \!\rightarrow\!  \ottmv{Y} \Rightarrow  \unskip ^ { \ell }  \! \star  \!\rightarrow\!  \star   \ottsym{)}   \, \ottmv{x}  \textsf{\textup{ in }\relax}  \ottsym{(}  g  [  \ottnt{T_{{\mathrm{1}}}}  \ottsym{,}  \ottnt{T_{{\mathrm{2}}}}  ] \,  2   \ottsym{,}  g  [  \ottnt{T_{{\mathrm{3}}}}  \ottsym{,}  \ottnt{T_{{\mathrm{4}}}}  ] \,  \textsf{\textup{true}\relax}   \ottsym{)} 
\]
in $\lambdaRTI$.
Note that, due to the use of $\star  \!\rightarrow\!  \star$, different type variables are
assigned to $\ottmv{x}$ and $\ottmv{y}$ and so $g$ is bound to a
two-argument type abstraction.  Now, what should type arguments be at
the two uses of $g$?  It should be obvious that $\ottnt{T_{{\mathrm{1}}}}$ and
$\ottnt{T_{{\mathrm{3}}}}$ have to be $ \textsf{\textup{int}\relax} $ and $ \textsf{\textup{bool}\relax} $, respectively, from
arguments to $g$.  At first, it may seem that $\ottnt{T_{{\mathrm{2}}}}$ and
$\ottnt{T_{{\mathrm{4}}}}$ should be $ \textsf{\textup{int}\relax} $ and $ \textsf{\textup{bool}\relax} $, respectively, in order to
avoid blame.  However, it is hard for a type system to see it---in
fact, $g$ is given type scheme $\forall \, \ottmv{X} \, \ottmv{Y}  \ottsym{.}  \ottmv{X}  \!\rightarrow\!  \star$ and, since
$\ottmv{Y}$ is bound but not referenced, there is no clue.  We assign a
special symbol $ \nu $ to $\ottnt{T_{{\mathrm{2}}}}$ and $\ottnt{T_{{\mathrm{4}}}}$; each occurrence
of $ \nu $ is replaced with a \emph{fresh type variable} when a
(polymorphic) value is substituted for $g$ during reduction.
These fresh type variables are expected to be instantiated by DTI---to
$ \textsf{\textup{int}\relax} $ and $ \textsf{\textup{bool}\relax} $ in this example---as reduction proceeds.
\[
  \begin{array}{ll}
    &  \textsf{\textup{let}\relax} \,  g  =   \Lambda   \ottmv{X}   \ottmv{Y} .\,   \lambda  \ottmv{x} \!:\!  \ottmv{X}  .\,  \ottsym{(}  \ottsym{(}   \lambda  \ottmv{y} \!:\!  \ottmv{Y}  .\,  \ottmv{y}   \ottsym{)}  \ottsym{:}   \ottmv{Y}  \!\rightarrow\!  \ottmv{Y} \Rightarrow  \unskip ^ { \ell }  \! \star  \!\rightarrow\!  \star   \ottsym{)}   \, \ottmv{x}  \textsf{\textup{ in }\relax}  \ottsym{(}  g  [   \textsf{\textup{int}\relax}   \ottsym{,}   \nu   ] \,  2   \ottsym{,}  g  [   \textsf{\textup{bool}\relax}   \ottsym{,}   \nu   ] \,  \textsf{\textup{true}\relax}   \ottsym{)}  \\
     \longmapsto  & (\ottsym{(}   \lambda  \ottmv{x} \!:\!   \textsf{\textup{int}\relax}   .\,  \ottsym{(}  \ottsym{(}   \lambda  \ottmv{y} \!:\!  \ottmv{Y_{{\mathrm{1}}}}  .\,  \ottmv{y}   \ottsym{)}  \ottsym{:}   \ottmv{Y_{{\mathrm{1}}}}  \!\rightarrow\!  \ottmv{Y_{{\mathrm{1}}}} \Rightarrow  \unskip ^ { \ell }  \! \star  \!\rightarrow\!  \star   \ottsym{)}  \, \ottmv{x}  \ottsym{)} \,  2 , \\ & \qquad \ottsym{(}   \lambda  \ottmv{x} \!:\!   \textsf{\textup{bool}\relax}   .\,  \ottsym{(}  \ottsym{(}   \lambda  \ottmv{y} \!:\!  \ottmv{Y_{{\mathrm{2}}}}  .\,  \ottmv{y}   \ottsym{)}  \ottsym{:}   \ottmv{Y_{{\mathrm{2}}}}  \!\rightarrow\!  \ottmv{Y_{{\mathrm{2}}}} \Rightarrow  \unskip ^ { \ell }  \! \star  \!\rightarrow\!  \star   \ottsym{)}  \, \ottmv{x}  \ottsym{)} \,  \textsf{\textup{true}\relax} ) \\
     \xmapsto{ \mathmakebox[0.4em]{} [  \ottmv{Y_{{\mathrm{1}}}}  :=   \textsf{\textup{int}\relax}   ] \mathmakebox[0.3em]{} }\hspace{-0.4em}{}^\ast \hspace{0.2em}  & \ottsym{(}  \ottsym{(}   2   \ottsym{:}    \textsf{\textup{int}\relax}  \Rightarrow  \unskip ^ { \ell }  \! \star   \ottsym{)}  \ottsym{,}  \ottsym{(}   \lambda  \ottmv{x} \!:\!   \textsf{\textup{bool}\relax}   .\,  \ottsym{(}  \ottsym{(}   \lambda  \ottmv{y} \!:\!  \ottmv{Y_{{\mathrm{2}}}}  .\,  \ottmv{y}   \ottsym{)}  \ottsym{:}   \ottmv{Y_{{\mathrm{2}}}}  \!\rightarrow\!  \ottmv{Y_{{\mathrm{2}}}} \Rightarrow  \unskip ^ { \ell }  \! \star  \!\rightarrow\!  \star   \ottsym{)}  \, \ottmv{x}  \ottsym{)} \,  \textsf{\textup{true}\relax}   \ottsym{)} \\
     \xmapsto{ \mathmakebox[0.4em]{} [  \ottmv{Y_{{\mathrm{2}}}}  :=   \textsf{\textup{bool}\relax}   ] \mathmakebox[0.3em]{} }\hspace{-0.4em}{}^\ast \hspace{0.2em}  & \ottsym{(}  \ottsym{(}   2   \ottsym{:}    \textsf{\textup{int}\relax}  \Rightarrow  \unskip ^ { \ell }  \! \star   \ottsym{)}  \ottsym{,}  \ottsym{(}   \textsf{\textup{true}\relax}   \ottsym{:}    \textsf{\textup{bool}\relax}  \Rightarrow  \unskip ^ { \ell }  \! \star   \ottsym{)}  \ottsym{)}
  \end{array}
\]
We do not assign fresh type variables when type
abstractions/applications are made explicit during cast insertion.  It is because an
expression, such as $g  [  \ottnt{T_{{\mathrm{1}}}}  \ottsym{,}  \ottnt{T_{{\mathrm{3}}}}  ]$, including type applications may be
duplicated during reduction and sharing a type variable among
duplicated expressions might cause unwanted blame.


\begin{figure}[tbp]
  {\small
    \input{figures/def_let_poly}}
  \caption{$\lambdaRTI$ with polymorphic let.}
  \label{fig:def_let_poly}
\end{figure}

Figure~\ref{fig:def_let_poly} shows the definition of the extension.
A type scheme, denoted by $\sigma$, is a gradual type abstracted over a
(possibly empty) finite sequence of type variables, denoted by
$ \overrightarrow{ \ottmv{X_{\ottmv{i}}} } $.  A type environment is changed so that it maps variables
to type schemes, instead of gradual types.  Terms are extended with
$ \textsf{\textup{let}\relax} $-expressions of the form $ \textsf{\textup{let}\relax} \,  \ottmv{x}  =   \Lambda    \overrightarrow{ \ottmv{X_{\ottmv{i}}} }  .\,  w   \textsf{\textup{ in }\relax}  \ottnt{f} $,
which bind $ \overrightarrow{ \ottmv{X_{\ottmv{i}}} } $ in value $w$,\footnote{%
A monomorphic definition $ \textsf{\textup{let}\relax} \,  \ottmv{x}  =  \ottnt{f_{{\mathrm{1}}}}  \textsf{\textup{ in }\relax}  \ottnt{f_{{\mathrm{2}}}} $ (where $\ottnt{f_{{\mathrm{1}}}}$ is not necessarily a value) can be expressed as
$\ottsym{(}   \lambda  \ottmv{x} .\,  \ottnt{f_{{\mathrm{2}}}}   \ottsym{)} \, \ottnt{f_{{\mathrm{1}}}}$ as usual.}
and variables $\ottmv{x}  [   \overrightarrow{ \mathbbsl{T}_{\ottmv{i}} }   ]$, which represent type application if
$\ottmv{x}$ is $ \textsf{\textup{let}\relax} $-bound, are now annotated with a sequence of type
arguments.\footnote{A variable introduced by a lambda abstraction is
  always monomorphic, so $ \overrightarrow{ \mathbbsl{T}_{\ottmv{i}} } $ is empty for such a variable.}
A type argument is either a static type $\ottnt{T}$ or the special symbol
$ \nu $.
%
Type substitution is defined in a capture-avoiding manner.

We adopt value restriction~\cite{DBLP:journals/lisp/Wright95}, i.e., the body
of a type abstraction has to be a syntactic value, for avoiding a subtle issue
in the dynamic semantics.
If we did not adopt value restriction, we would have to deal with
cast applications where a target type is a \emph{bound} type variable.
For example, consider the term
$ \textsf{\textup{let}\relax} \,  \ottmv{x}  =   \Lambda   \ottmv{X} .\,  w   \ottsym{:}    \textsf{\textup{int}\relax}  \Rightarrow  \unskip ^ { \ell_{{\mathrm{1}}} }  \!  \star \Rightarrow  \unskip ^ { \ell_{{\mathrm{2}}} }  \! \ottmv{X}    \textsf{\textup{ in }\relax}  \ottnt{f} $, which is actually ill
formed in our language because value restriction is violated.  It has a cast
with bound $\ottmv{X}$ as its target type.
The question is how the cast $w  \ottsym{:}    \textsf{\textup{int}\relax}  \Rightarrow  \unskip ^ { \ell_{{\mathrm{1}}} }  \!  \star \Rightarrow  \unskip ^ { \ell_{{\mathrm{2}}} }  \! \ottmv{X}  $ is evaluated.
We should not apply \rnp{R\_InstBase} here because the type of $\ottmv{x}$
would change from $\forall \, \ottmv{X}  \ottsym{.}  \ottmv{X}$ to $\forall \, \ottmv{X}  \ottsym{.}   \textsf{\textup{int}\relax} $.
It appears that there is no reasonable way to reduce this cast
further---after all, the semantics of the cast depends on what type is
substituted for $\ottmv{X}$, which may be instantiated in many ways
in $\ottnt{f}$.
Value restriction resolves this issue:\footnote{Making $ \textsf{\textup{let}\relax} $
  call-by-name~\cite{DBLP:conf/popl/Leroy93} may be another option.} a
cast with its target type being a bound type variable is executed only
after the type variable is instantiated.

The typing rules of $\lambdaRTI$ are also updated.  We replace the rule for variables with
the rule \rnp{T\_VarP} and add the rule \rnp{T\_LetP}.  The rule
\rnp{T\_LetP} is standard; it allows generalization by type variables
that do not appear free in $\Gamma$.  Note that it allows abstraction
by a type variable $\ottmv{X}$ even when $\ottmv{X}$ does not appear in
$\ottnt{U_{{\mathrm{1}}}}$.  As we have already seen such abstraction can be significant
in $\lambdaRTI$.  The second premise of the rule \rnp{T\_VarP}, which
represents type application, means that extra type variables that do
not appear in $\ottnt{U}$ have to be instantiated by $ \nu $ (and
other type variables by static types).  The type expression
$\ottnt{U}  [   \overrightarrow{ \ottmv{X_{\ottmv{i}}} }   \ottsym{:=}   \overrightarrow{ \mathbbsl{T}_{\ottmv{i}} }   ]$ is notational abuse but the result will not
contain $ \nu $ because the corresponding type variables do not
appear in $\ottnt{U}$.

The rule \rnp{R\_LetP} is an additional rule to reduce
$ \textsf{\textup{let}\relax} $-expressions.  Roughly speaking, $ \textsf{\textup{let}\relax} \,  \ottmv{x}  =   \Lambda    \overrightarrow{ \ottmv{X_{\ottmv{i}}} }  .\,  w   \textsf{\textup{ in }\relax}  \ottnt{f} $
reduces to $\ottnt{f}$ in which $w$ is substituted for $\ottmv{x}$ as
usual but, due to explicit type abstraction/application, the
definition of substitution is slightly peculiar: When a variable is
replaced with a type-abstracted value, type arguments $ \overrightarrow{ \mathbbsl{T}_{\ottmv{i}} } $
are also substituted, after fresh type variable generation, for
$ \overrightarrow{ \ottmv{X_{\ottmv{i}}} } $ in the value.  We formalize this idea as substitution of
the form $\ottnt{f}  [  \ottmv{x}  \ottsym{:=}   \Lambda    \overrightarrow{ \ottmv{X_{\ottmv{i}}} }  .\,  w   ]$, shown in the lower half of
Figure~\ref{fig:def_let_poly}; in the case for variables,
the length of a sequence $ \overrightarrow{ \ottmv{X_{\ottmv{i}}} } $ or $ \overrightarrow{ \ottnt{T_{\ottmv{i}}} } $ is
denoted by $|\cdot|$.
This nonstandard substitution makes correspondence to usual reduction
for $ \textsf{\textup{let}\relax} $ (that is, $ \textsf{\textup{let}\relax} \,  \ottmv{x}  =  w  \textsf{\textup{ in }\relax}  \ottnt{f}  \longrightarrow \ottnt{f}  [  \ottmv{x}  \ottsym{:=}  w  ]$)
easier to see.  Other reduction and evaluation rules, including
\rnp{R\_InstBase} and \rnp{R\_InstArrow}, remain unchanged.

\subsection{Discussion about the semantics of let-polymorphism}

Before proceeding further, we give a brief comparison with
\citet{DBLP:conf/popl/GarciaC15}, who have suggested that the
Polymorphic Blame Calculus
(PBC)~\cite{DBLP:conf/popl/AhmedFSW11,DBLP:journals/pacmpl/AhmedJSW17}
can be used to give the semantics of $ \textsf{\textup{let}\relax} $ in the ITGL.  Using the
PBC-style semantics for type abstraction/application would, however,
raise more of blame---because the PBC enforces parametricity at run
time---even when we just give subterms names by $ \textsf{\textup{let}\relax} $.

The difference between our approach and the PBC-based one
is exemplified by the following program in the ITGL:
\[
  \ottsym{(}   \lambda  \ottmv{x} .\,   1    \ottsym{+}  \ottsym{(}  \ottsym{(}   \lambda  \ottmv{y} \!:\!  \star  .\,  \ottmv{y}   \ottsym{)} \, \ottmv{x}  \ottsym{)}  \ottsym{)} \,  2 
\]
This program is translated into $\lambdaRTI$ and evaluated as follows:
\[
  \begin{array}{lll}
     \ottsym{(}   \lambda  \ottmv{x} .\,   1    \ottsym{+}  \ottsym{(}  \ottsym{(}   \lambda  \ottmv{y} \!:\!  \star  .\,  \ottmv{y}   \ottsym{)} \, \ottmv{x}  \ottsym{)}  \ottsym{)} \,  2  &
    \rightsquigarrow & \ottsym{(}   \lambda  \ottmv{x} \!:\!   \textsf{\textup{int}\relax}   .\,   1    \ottsym{+}  \ottsym{(}  \ottsym{(}  \ottsym{(}   \lambda  \ottmv{y} \!:\!  \star  .\,  \ottmv{y}   \ottsym{)} \, \ottsym{(}  \ottmv{x}  \ottsym{:}    \textsf{\textup{int}\relax}  \Rightarrow  \unskip ^ { \ell_{{\mathrm{1}}} }  \! \star   \ottsym{)}  \ottsym{)}  \ottsym{:}   \star \Rightarrow  \unskip ^ { \ell_{{\mathrm{2}}} }  \!  \textsf{\textup{int}\relax}    \ottsym{)}  \ottsym{)} \,  2  \\
    &  \xmapsto{ \mathmakebox[0.4em]{} [  ] \mathmakebox[0.3em]{} }\hspace{-0.4em}{}^\ast \hspace{0.2em}  & 3 \\
  \end{array}
\]
This term evaluates to a constant $ 3 $, as expected.
Next, let us rewrite this program so that it uses $ \textsf{\textup{let}\relax} $ to give
a name to the function $ \lambda  \ottmv{x} .\,   1    \ottsym{+}  \ottsym{(}  \ottsym{(}   \lambda  \ottmv{y} \!:\!  \star  .\,  \ottmv{y}   \ottsym{)} \, \ottmv{x}  \ottsym{)}$.  Then, this
program is translated into $\lambdaRTI$ and evaluated as follows:
\[
  \begin{array}{ll}
    &  \textsf{\textup{let}\relax} \,  g  =   \lambda  \ottmv{x} .\,   1    \ottsym{+}  \ottsym{(}  \ottsym{(}   \lambda  \ottmv{y} \!:\!  \star  .\,  \ottmv{y}   \ottsym{)} \, \ottmv{x}  \ottsym{)}  \textsf{\textup{ in }\relax}  g  \,  2  \\
    \rightsquigarrow &  \textsf{\textup{let}\relax} \,  g  =   \Lambda   \ottmv{X} .\,   \lambda  \ottmv{x} \!:\!  \ottmv{X}  .\,   1     \ottsym{+}  \ottsym{(}  \ottsym{(}  \ottsym{(}   \lambda  \ottmv{y} \!:\!  \star  .\,  \ottmv{y}   \ottsym{)} \, \ottsym{(}  \ottmv{x}  \ottsym{:}   \ottmv{X} \Rightarrow  \unskip ^ { \ell_{{\mathrm{1}}} }  \! \star   \ottsym{)}  \ottsym{)}  \ottsym{:}   \star \Rightarrow  \unskip ^ { \ell_{{\mathrm{2}}} }  \!  \textsf{\textup{int}\relax}    \ottsym{)}  \textsf{\textup{ in }\relax}  g   [   \textsf{\textup{int}\relax}   ] \,  2  \\
     \xmapsto{ \mathmakebox[0.4em]{} [  ] \mathmakebox[0.3em]{} }  & \ottsym{(}   \lambda  \ottmv{x} \!:\!   \textsf{\textup{int}\relax}   .\,   1    \ottsym{+}  \ottsym{(}  \ottsym{(}  \ottsym{(}   \lambda  \ottmv{y} \!:\!  \star  .\,  \ottmv{y}   \ottsym{)} \, \ottsym{(}  \ottmv{x}  \ottsym{:}    \textsf{\textup{int}\relax}  \Rightarrow  \unskip ^ { \ell_{{\mathrm{1}}} }  \! \star   \ottsym{)}  \ottsym{)}  \ottsym{:}   \star \Rightarrow  \unskip ^ { \ell_{{\mathrm{2}}} }  \!  \textsf{\textup{int}\relax}    \ottsym{)}  \ottsym{)} \,  2  \\
     \xmapsto{ \mathmakebox[0.4em]{} [  ] \mathmakebox[0.3em]{} }\hspace{-0.4em}{}^\ast \hspace{0.2em}  &  3  \\
\end{array}
\]
This program also evaluates to $ 3 $.  Notice that $g$ is given
a polymorphic type scheme $\forall \, \ottmv{X}  \ottsym{.}  \ottmv{X}  \!\rightarrow\!   \textsf{\textup{int}\relax} $ because $\ottmv{x}$ is
passed to a function that expects an argument of the dynamic type and
there is no constraint on its type.  The use of $g$ comes with the
type argument $ \textsf{\textup{int}\relax} $ so that it can be applied to an integer.

We show the evaluation sequence of the same term using the
PBC-style semantics of type abstraction and application below:
\[
  \begin{array}{ll}
    &  \textsf{\textup{let}\relax} \,  g  =   \Lambda   \ottmv{X} .\,   \lambda  \ottmv{x} \!:\!  \ottmv{X}  .\,   1     \ottsym{+}  \ottsym{(}  \ottsym{(}  \ottsym{(}   \lambda  \ottmv{y} \!:\!  \star  .\,  \ottmv{y}   \ottsym{)} \, \ottsym{(}  \ottmv{x}  \ottsym{:}   \ottmv{X} \Rightarrow  \unskip ^ { \ell_{{\mathrm{1}}} }  \! \star   \ottsym{)}  \ottsym{)}  \ottsym{:}   \star \Rightarrow  \unskip ^ { \ell_{{\mathrm{2}}} }  \!  \textsf{\textup{int}\relax}    \ottsym{)}  \textsf{\textup{ in }\relax}  g   [   \textsf{\textup{int}\relax}   ] \,  2  \\
     \xmapsto{ \mathmakebox[0.4em]{} [  ] \mathmakebox[0.3em]{} }  & \ottsym{(}   \Lambda   \ottmv{X} .\,   \lambda  \ottmv{x} \!:\!  \ottmv{X}  .\,   1     \ottsym{+}  \ottsym{(}  \ottsym{(}  \ottsym{(}   \lambda  \ottmv{y} \!:\!  \star  .\,  \ottmv{y}   \ottsym{)} \, \ottsym{(}  \ottmv{x}  \ottsym{:}   \ottmv{X} \Rightarrow  \unskip ^ { \ell_{{\mathrm{1}}} }  \! \star   \ottsym{)}  \ottsym{)}  \ottsym{:}   \star \Rightarrow  \unskip ^ { \ell_{{\mathrm{2}}} }  \!  \textsf{\textup{int}\relax}    \ottsym{)}  \ottsym{)}  [   \textsf{\textup{int}\relax}   ] \,  2  \\
     \xmapsto{ \mathmakebox[0.4em]{} [  ] \mathmakebox[0.3em]{} }  & \graybox{\nu\ottmv{X}:= \textsf{\textup{int}\relax} }. \ottsym{(}   \lambda  \ottmv{x} \!:\!  \ottmv{X}  .\,   1    \ottsym{+}  \ottsym{(}  \ottsym{(}  \ottsym{(}   \lambda  \ottmv{y} \!:\!  \star  .\,  \ottmv{y}   \ottsym{)} \, \ottsym{(}  \ottmv{x}  \ottsym{:}   \ottmv{X} \Rightarrow  \unskip ^ { \ell_{{\mathrm{1}}} }  \! \star   \ottsym{)}  \ottsym{)}  \ottsym{:}   \star \Rightarrow  \unskip ^ { \ell_{{\mathrm{2}}} }  \!  \textsf{\textup{int}\relax}    \ottsym{)}  \ottsym{)} \,  2 
  \end{array}
\]
When a polymorphic function is applied to a type argument, a type
binding $\nu\ottmv{X}:= \textsf{\textup{int}\relax} $ is generated, instead of substituting
$ \textsf{\textup{int}\relax} $ for $\ottmv{X}$.  As discussed by \citet{DBLP:conf/popl/AhmedFSW11,DBLP:journals/pacmpl/AhmedJSW17},
the generation of type bindings is crucial to ensure parametricity
dynamically: $\ottmv{X}$ behaves as if it is a fresh base type in the body
of type abstraction, so evaluation proceeds as follows:
\[
  \begin{array}{rll}
\cdots &  \xmapsto{ \mathmakebox[0.4em]{} [  ] \mathmakebox[0.3em]{} }  & \nu\ottmv{X}:= \textsf{\textup{int}\relax} .  1   \ottsym{+}  \ottsym{(}  \ottsym{(}  \ottsym{(}   \lambda  \ottmv{y} \!:\!  \star  .\,  \ottmv{y}   \ottsym{)} \, \ottsym{(}   2   \ottsym{:}   \ottmv{X} \Rightarrow  \unskip ^ { \ell_{{\mathrm{1}}} }  \! \star   \ottsym{)}  \ottsym{)}  \ottsym{:}   \star \Rightarrow  \unskip ^ { \ell_{{\mathrm{2}}} }  \!  \textsf{\textup{int}\relax}    \ottsym{)} \\
    &  \xmapsto{ \mathmakebox[0.4em]{} [  ] \mathmakebox[0.3em]{} }\hspace{-0.4em}{}^\ast \hspace{0.2em}  & \nu\ottmv{X}:= \textsf{\textup{int}\relax} .  1   \ottsym{+}   \graybox{  \ottsym{(}   2   \ottsym{:}   \ottmv{X} \Rightarrow  \unskip ^ { \ell_{{\mathrm{1}}} }  \!  \star \Rightarrow  \unskip ^ { \ell_{{\mathrm{2}}} }  \!  \textsf{\textup{int}\relax}     \ottsym{)}  }  \\
    &  \xmapsto{ \mathmakebox[0.4em]{} [  ] \mathmakebox[0.3em]{} }\hspace{-0.4em}{}^\ast \hspace{0.2em}  & \textsf{\textup{blame}\relax} \, \ell_{{\mathrm{2}}} \\
  \end{array}
\]
The cast $( 2   \ottsym{:}   \ottmv{X} \Rightarrow  \unskip ^ { \ell_{{\mathrm{1}}} }  \!  \star \Rightarrow  \unskip ^ { \ell_{{\mathrm{2}}} }  \!  \textsf{\textup{int}\relax}   )$ fails and this program evaluates
to $\textsf{\textup{blame}\relax} \, \ell_{{\mathrm{2}}}$, instead of $3$!  In fact, according to
parametricity, the polymorphic type $\forall \, \ottmv{X}  \ottsym{.}  \ottmv{X}  \!\rightarrow\!   \textsf{\textup{int}\relax} $ should
behave the same regardless of $\ottmv{X}$ and so it should uniformly
return an integer constant or uniformly fail.  From the viewpoint of
parametricity, this is completely reasonable behavior.

In general, there is a conflict between polymorphism and blame in the
PBC-based semantics.  If $\ottmv{X}$ were not generalized at $ \textsf{\textup{let}\relax} $ in
the example above, $\ottmv{X}$ would be unified with $ \textsf{\textup{int}\relax} $ and the
resulting term
\[
 \textsf{\textup{let}\relax} \,  g  =   \lambda  \ottmv{x} \!:\!   \textsf{\textup{int}\relax}   .\,   1    \ottsym{+}  \ottsym{(}  \ottsym{(}  \ottsym{(}   \lambda  \ottmv{y} \!:\!  \star  .\,  \ottmv{y}   \ottsym{)} \, \ottsym{(}  \ottmv{x}  \ottsym{:}    \textsf{\textup{int}\relax}  \Rightarrow  \unskip ^ { \ell_{{\mathrm{1}}} }  \! \star   \ottsym{)}  \ottsym{)}  \ottsym{:}   \star \Rightarrow  \unskip ^ { \ell_{{\mathrm{2}}} }  \!  \textsf{\textup{int}\relax}    \ottsym{)}  \textsf{\textup{ in }\relax}  g  \,  2  
\]
would evaluate to $ 3 $ without blame.  (This translation may be
obtained by adapting the type inference algorithm to minimize the
introduction of polymorphism~\cite{Bjorner94ML} to gradual typing.)
The conflict can be considered another instance of incoherence in the
sense that translations that differ in how $ \textsf{\textup{let}\relax} $ is generalized
may result in different behavior.

We expect that our semantics is coherent and that operational
equivalence between $ \textsf{\textup{let}\relax} \,  \ottmv{x}  =  w  \textsf{\textup{ in }\relax}  \ottnt{f} $ and $\ottnt{f}  [  \ottmv{x}  \ottsym{:=}  w  ]$ holds
(whereas it does not always hold in the PBC-based semantics as we saw
above).  However, they are achieved by sacrificing (dynamically
enforced) parametricity.

\section{Basic Properties of $\lambdaRTI$} \label{sec:properties} In
this section, we show basic properties of $\lambdaRTI$: type safety and 
conservativity over the standard blame calculus extended to
let-polymorphism.

We first define 
the predicate $ \ottnt{f} \!  \Uparrow  $ to mean that $\ottnt{f}$ diverges.



\begin{definition}[Divergence]
  We write $ \ottnt{f} \!  \Uparrow  $ if there is an infinite evaluation sequence from $\ottnt{f}$.
\end{definition}

We write $\textit{ftv} \, \ottsym{(}  \ottnt{f}  \ottsym{)}$, $\textit{ftv} \, \ottsym{(}  \ottnt{U}  \ottsym{)}$, and $\textit{ftv} \, \ottsym{(}  \Gamma  \ottsym{)}$ for the set of
free type variables in $\ottnt{f}$, $\ottnt{U}$, and $\Gamma$, respectively.
We also write $\textit{dom} \, \ottsym{(}  S  \ottsym{)}$ for the domain of the type substitution $S$.

\subsection{Type Safety}
We show the type safety property
using the progress and preservation lemmas~\cite{WrightFelleisenIC94}.  First,
we state the canonical forms property below.
As we have already mentioned, this lemma also means that no values are
typed at type variables and this is why type variables are not 
ground types.

\begin{lemma}[name=Canonical Forms,restate=lemCanonicalForms] \label{lem:canonical_forms}
  If $ \emptyset   \vdash  w  \ottsym{:}  \ottnt{U}$, then one of the following holds:
  \begin{itemize}
    \item $\ottnt{U}  \ottsym{=}  \iota$ and $w  \ottsym{=}  \ottnt{c}$ for some $\iota$ and $\ottnt{c}$;
    \item $\ottnt{U}  \ottsym{=}  \ottnt{U_{{\mathrm{1}}}}  \!\rightarrow\!  \ottnt{U_{{\mathrm{2}}}}$ and $w  \ottsym{=}   \lambda  \ottmv{x} \!:\!  \ottnt{U_{{\mathrm{1}}}}  .\,  \ottnt{f} $ for some $\ottmv{x}$, $\ottnt{f}$, $\ottnt{U_{{\mathrm{1}}}}$, and $\ottnt{U_{{\mathrm{2}}}}$;
    \item $\ottnt{U}  \ottsym{=}  \ottnt{U_{{\mathrm{1}}}}  \!\rightarrow\!  \ottnt{U_{{\mathrm{2}}}}$ and $w  \ottsym{=}  w'  \ottsym{:}   \ottnt{U'_{{\mathrm{1}}}}  \!\rightarrow\!  \ottnt{U'_{{\mathrm{2}}}} \Rightarrow  \unskip ^ { \ell }  \! \ottnt{U_{{\mathrm{1}}}}  \!\rightarrow\!  \ottnt{U_{{\mathrm{2}}}} $
      for some $w'$, $\ottnt{U_{{\mathrm{1}}}}, \ottnt{U_{{\mathrm{2}}}}, \ottnt{U'_{{\mathrm{1}}}}$, $\ottnt{U'_{{\mathrm{2}}}}$, and $\ell$; or
    \item $\ottnt{U}  \ottsym{=}  \star$ and $w  \ottsym{=}  w'  \ottsym{:}   \ottnt{G} \Rightarrow  \unskip ^ { \ell }  \! \star $
      for some $w'$, $\ottnt{G}$, and $\ell$.
  \end{itemize}
\end{lemma}

\iffull\else
\begin{proof}
  By case analysis on the last typing rule applied to derive
$ \emptyset   \vdash  w  \ottsym{:}  \ottnt{U}$.
\end{proof}
\fi

The progress lemma is standard: A well-typed term can be evaluated one step further, or is a value or blame.

\begin{lemma}[name=Progress,restate=lemProgress] \label{lem:progress}
  If $ \emptyset   \vdash  \ottnt{f}  \ottsym{:}  \ottnt{U}$, then one of the following holds:
\iffull
  \begin{itemize}
    \item $\ottnt{f} \,  \xmapsto{ \mathmakebox[0.4em]{} S \mathmakebox[0.3em]{} }  \, \ottnt{f'}$ for some $S$ and $\ottnt{f'}$;
    \item $\ottnt{f}$ is a value; or
    \item $\ottnt{f}  \ottsym{=}  \textsf{\textup{blame}\relax} \, \ell$ for some $\ell$.
    \end{itemize}
    \else
    (1) $\ottnt{f} \,  \xmapsto{ \mathmakebox[0.4em]{} S \mathmakebox[0.3em]{} }  \, \ottnt{f'}$ for some $S$ and $\ottnt{f'}$;
    (2) $\ottnt{f}$ is a value; or
    (3) $\ottnt{f}  \ottsym{=}  \textsf{\textup{blame}\relax} \, \ell$ for some $\ell$.
    \fi
\end{lemma}

\iffull\else
\begin{proof}
  By induction on the typing derivation using Lemma \ref{lem:canonical_forms}.
\end{proof}
\fi

The statement of the preservation lemma is slightly different from a
standard one, because our semantics, which is equipped with DTI, may
substitute type variables with some static type during reduction and
evaluation.

\begin{lemma}[name=Preservation,restate=lemPreservation] \label{lem:preservation}
  Suppose that $ \emptyset   \vdash  \ottnt{f}  \ottsym{:}  \ottnt{U}$.
\iffull
  \begin{enumerate}
    \item If $\ottnt{f} \,  \xrightarrow{ \mathmakebox[0.4em]{} S \mathmakebox[0.3em]{} }  \, \ottnt{f'}$, then $ \emptyset   \vdash  S  \ottsym{(}  \ottnt{f'}  \ottsym{)}  \ottsym{:}  S  \ottsym{(}  \ottnt{U}  \ottsym{)}$.
    \item If $\ottnt{f} \,  \xmapsto{ \mathmakebox[0.4em]{} S \mathmakebox[0.3em]{} }  \, \ottnt{f'}$, then $ \emptyset   \vdash  \ottnt{f'}  \ottsym{:}  S  \ottsym{(}  \ottnt{U}  \ottsym{)}$.
    \end{enumerate}
    \else
    (1) If $\ottnt{f} \,  \xrightarrow{ \mathmakebox[0.4em]{} S \mathmakebox[0.3em]{} }  \, \ottnt{f'}$, then $ \emptyset   \vdash  S  \ottsym{(}  \ottnt{f'}  \ottsym{)}  \ottsym{:}  S  \ottsym{(}  \ottnt{U}  \ottsym{)}$ and (2) if $\ottnt{f} \,  \xmapsto{ \mathmakebox[0.4em]{} S \mathmakebox[0.3em]{} }  \, \ottnt{f'}$, then $ \emptyset   \vdash  \ottnt{f'}  \ottsym{:}  S  \ottsym{(}  \ottnt{U}  \ottsym{)}$.
    \fi
\end{lemma}

\iffull\else
\begin{proof}
  \leavevmode
  \begin{enumerate}
    \item By case analysis on the typing rule applied to $\ottnt{f}$.
    \item By case analysis on the evaluation rule applied to $\ottnt{f}$.
      In the case \rnp{E\_Step}, use induction on the structure of $\ottnt{E}$.
      \qedhere
  \end{enumerate}
\end{proof}
\fi

The second item (for the evaluation relation) does not require $S$
to be applied to $\ottnt{f'}$ in the type judgment (whereas the first item
does) because it is already applied---see \rnp{E\_Step}.

Finally, type safety holds in our blame calculus using the above
lemmas.  The following statement is also slightly different from a
standard one in the literature for the same reason as the preservation.

\begin{theorem}[name=Type Safety,restate=thmTypeSafety] \label{thm:type_safety}
  If $ \emptyset   \vdash  \ottnt{f}  \ottsym{:}  \ottnt{U}$,
  then one of the following holds:
  \iffull
  \begin{itemize}
    \item $\ottnt{f} \,  \xmapsto{ \mathmakebox[0.4em]{} S \mathmakebox[0.3em]{} }\hspace{-0.4em}{}^\ast \hspace{0.2em}  \, \ottnt{r}$ for some $S$ and $\ottnt{r}$ such that $ \emptyset   \vdash  \ottnt{r}  \ottsym{:}  S  \ottsym{(}  \ottnt{U}  \ottsym{)}$; or
    \item $ \ottnt{f} \!  \Uparrow  $.
    \end{itemize}
    \else
    (1) $\ottnt{f} \,  \xmapsto{ \mathmakebox[0.4em]{} S \mathmakebox[0.3em]{} }\hspace{-0.4em}{}^\ast \hspace{0.2em}  \, \ottnt{r}$ for some $S$ and $\ottnt{r}$ such that $ \emptyset   \vdash  \ottnt{r}  \ottsym{:}  S  \ottsym{(}  \ottnt{U}  \ottsym{)}$; or
    (2) $ \ottnt{f} \!  \Uparrow  $.
    \fi
\end{theorem}

\iffull\else
\begin{proof}
  By Lemmas \ref{lem:progress} and \ref{lem:preservation}.
\end{proof}
\fi

\subsection{Conservative Extension}
Our blame calculus is a conservative extension of the standard simply typed
blame calculus~\cite{DBLP:conf/snapl/SiekVCB15} extended to let-polymorphism
(we call it $\lambdaBC$).
We can obtain
$\lambdaBC$ by disallowing free type variables and removing reduction
rules \rnp{R\_InstBase} and \rnp{R\_InstArrow}.  We denote judgments
and relations for this sublanguage by subscripting symbols with \textsf{B};
we also omit type substitutions on the reduction/evaluation relations of $\lambdaBC$
because they are always empty.  So, we write $\ottnt{f} \, \longrightarrow_{\textsf{\textup{B}\relax}\relax} \, \ottnt{f'}$ for reduction and $\ottnt{f} \, \longmapsto_{\textsf{\textup{B}\relax}\relax} \, \ottnt{f'}$ for
evaluation.  Then, it is easy to show that $\lambdaRTI$ is a
conservative extension of $\lambdaBC$ in the following sense.

\begin{theorem}[name=Conservative Extension,restate=thmConservativeExtension] \label{thm:conservative_extension}
  Suppose that $\textit{ftv} \, \ottsym{(}  \ottnt{f}  \ottsym{)}  \ottsym{=}   \emptyset $ and $\ottnt{f}$ does not contain $ \nu $.

  \begin{enumerate}
    \item $\ottnt{f} \,  \xmapsto{ \mathmakebox[0.4em]{} [  ] \mathmakebox[0.3em]{} }\hspace{-0.4em}{}^\ast \hspace{0.2em}  \, \ottnt{r}$ if and only if $\ottnt{f} \, \longmapsto_{\textsf{\textup{B}\relax}\relax}^\ast \, \ottnt{r}$.
    \item $ \ottnt{f} \!  \Uparrow  $ if and only if $ \ottnt{f} \!  \Uparrow _{\textsf{\textup{B}\relax}\relax}  $.
  \end{enumerate}
\end{theorem}

\section{Soundness and Completeness of Dynamic Type Inference}
\label{sec:soundness-completeness}

In this section, we show soundness and completeness of DTI; we
formalize these properties by comparing evaluation of a given term
under $\lambdaRTI$ and its instances of type substitution in
$\lambdaBC$.

Soundness of \emph{static} type inference means that, if the
type inference algorithm succeeds, the program is well typed (under
the reconstructed type annotation).  In DTI, the success of type inference
and the reconstructed type annotation roughly correspond to normally
terminating evaluation and type substitution obtained through
evaluation, respectively.  So, soundness of DTI means that, if
evaluation of a program normally terminates at value $w$ with a
type substitution (yielded by DTI), then applying the type
substitution \emph{before} evaluation also makes evaluation normally terminate
at the same value.

We state the soundness property below.  In addition to normally terminating
programs, we can show similar results for programs that abort by blame
(the second item) or diverge (the third item).  The second item means
that, if a well-typed term evaluates to blame with DTI, then \emph{no}
substitution $S'$ can help $S'  \ottsym{(}  \ottnt{f}  \ottsym{)}$ avoid blame.  In other
words, DTI avoids blame as much as possible.

\begin{theorem}[name=Soundness of Dynamic Type Inference,restate=thmSoundness] \label{thm:soundness}
  Suppose $ \emptyset   \vdash  \ottnt{f}  \ottsym{:}  \ottnt{U}$.
  \begin{enumerate}
    \item If $\ottnt{f} \,  \xmapsto{ \mathmakebox[0.4em]{} S \mathmakebox[0.3em]{} }\hspace{-0.4em}{}^\ast \hspace{0.2em}  \, \ottnt{r}$,
          then, for any $S'$ such that $\textit{ftv} \, \ottsym{(}  S'  \ottsym{(}  S  \ottsym{(}  \ottnt{f}  \ottsym{)}  \ottsym{)}  \ottsym{)}  \ottsym{=}   \emptyset $,
          $S'  \ottsym{(}  S  \ottsym{(}  \ottnt{f}  \ottsym{)}  \ottsym{)} \,  \xmapsto{ \mathmakebox[0.4em]{} S'' \mathmakebox[0.3em]{} }\hspace{-0.4em}{}^\ast \hspace{0.2em}  \, S'  \ottsym{(}  \ottnt{r}  \ottsym{)}$
          for some $S''$.
    \item If $\ottnt{f} \,  \xmapsto{ \mathmakebox[0.4em]{} S \mathmakebox[0.3em]{} }\hspace{-0.4em}{}^\ast \hspace{0.2em}  \, \textsf{\textup{blame}\relax} \, \ell$,
      then, for any $S'$, $S'  \ottsym{(}  \ottnt{f}  \ottsym{)} \,  \xmapsto{ \mathmakebox[0.4em]{} S'' \mathmakebox[0.3em]{} }\hspace{-0.4em}{}^\ast \hspace{0.2em}  \, \textsf{\textup{blame}\relax} \, \ell'$ for some
      $S''$ and $\ell'$.

    \item If $ \ottnt{f} \!  \Uparrow  $, then,
      for any $S$ such that $\textit{ftv} \, \ottsym{(}  S  \ottsym{(}  \ottnt{f}  \ottsym{)}  \ottsym{)}  \ottsym{=}   \emptyset $, either
      $ S  \ottsym{(}  \ottnt{f}  \ottsym{)} \!  \Uparrow  $ or
      $S  \ottsym{(}  \ottnt{f}  \ottsym{)} \,  \xmapsto{ \mathmakebox[0.4em]{} S' \mathmakebox[0.3em]{} }\hspace{-0.4em}{}^\ast \hspace{0.2em}  \, \textsf{\textup{blame}\relax} \, \ell$ for some $\ell$ and $S'$.
  \end{enumerate}
\end{theorem}

We state a main lemma to prove Theorem~\ref{thm:soundness}(1) below.
\begin{lemma}[name=,restate=lemSoundnessResult] \label{lem:soundness_result}
 If
 $\ottnt{f} \,  \xmapsto{ \mathmakebox[0.4em]{}  S_{{\mathrm{1}}}  \uplus  S_{{\mathrm{2}}}  \mathmakebox[0.3em]{} }\hspace{-0.4em}{}^\ast \hspace{0.2em}  \, \ottnt{r}$ and
 $\textit{dom} \, \ottsym{(}  S_{{\mathrm{1}}}  \ottsym{)}  \subseteq  \textit{ftv} \, \ottsym{(}  \ottnt{f}  \ottsym{)}$,
 then
 $S_{{\mathrm{1}}}  \ottsym{(}  \ottnt{f}  \ottsym{)} \,  \xmapsto{ \mathmakebox[0.4em]{} S'_{{\mathrm{2}}} \mathmakebox[0.3em]{} }\hspace{-0.4em}{}^\ast \hspace{0.2em}  \, \ottnt{r}$ for some $S'_{{\mathrm{2}}}$ such that
 $\textit{dom} \, \ottsym{(}  S'_{{\mathrm{2}}}  \ottsym{)}  \subseteq  \textit{dom} \, \ottsym{(}  S_{{\mathrm{2}}}  \ottsym{)}$ and $\textit{ftv} \, \ottsym{(}  S_{{\mathrm{1}}}  \ottsym{(}  \ottmv{X}  \ottsym{)}  \ottsym{)}  \cap  \textit{dom} \, \ottsym{(}  S'_{{\mathrm{2}}}  \ottsym{)}  \ottsym{=}   \emptyset $
 for any $\ottmv{X} \, \in \, \textit{dom} \, \ottsym{(}  S_{{\mathrm{1}}}  \ottsym{)}$.
\end{lemma}
Here, $ S_{{\mathrm{1}}}  \uplus  S_{{\mathrm{2}}} $ is a type substitution generated by DTI.\footnote{%
  The concatenation $ S_{{\mathrm{1}}}  \uplus  S_{{\mathrm{2}}} $ of $S_{{\mathrm{1}}}$ and $S_{{\mathrm{2}}}$ is defined only if
$\textit{dom} \, \ottsym{(}  S_{{\mathrm{1}}}  \ottsym{)}$ and $\textit{dom} \, \ottsym{(}  S_{{\mathrm{2}}}  \ottsym{)}$ are disjoint and it maps $\ottmv{X}$ to
$S_{{\mathrm{1}}}  \ottsym{(}  \ottmv{X}  \ottsym{)}$ if $\ottmv{X} \, \in \, \textit{dom} \, \ottsym{(}  S_{{\mathrm{1}}}  \ottsym{)}$ and to $S_{{\mathrm{2}}}  \ottsym{(}  \ottmv{X}  \ottsym{)}$ if $\ottmv{X} \, \in \, \textit{dom} \, \ottsym{(}  S_{{\mathrm{2}}}  \ottsym{)}$.}
Since fresh
type variables are generated during reduction, we split the type
substitution into two parts: $S_{{\mathrm{1}}}$ for type variables that appear
in $\ottnt{f}$ and $S_{{\mathrm{2}}}$ for generated type variables.  The conditions
on the type substitution $S'_{{\mathrm{2}}}$ mean that $S_{{\mathrm{1}}}  \ottsym{(}  \ottnt{f}  \ottsym{)}$ may generate
fewer type variables than $\ottnt{f}$ and the domain of $S'_{{\mathrm{2}}}$ is fresh
(with respect to $S_{{\mathrm{1}}}  \ottsym{(}  \ottnt{f}  \ottsym{)}$).  The statement is similar to
Theorem~\ref{thm:soundness}(1) but additional conditions make proof by
induction on the number of evaluation steps work.


We cannot say much about diverging programs---Theorem~\ref{thm:soundness}(3)
means that, if a well-typed program diverges, then no type
substitution $S$ makes $S  \ottsym{(}  \ottnt{f}  \ottsym{)}$ normally terminating.  One may
expect a stronger property that, if evaluation with DTI diverges, then
without DTI it also diverges (after applying some type substitution that
instantiates all type variables); but actually it is not the case.

\begin{theorem}[name=,restate=thmDivergence]
  There exists $\ottnt{f}$ such that (1) $ \emptyset   \vdash  \ottnt{f}  \ottsym{:}  \ottnt{U}$ and (2)
  $ \ottnt{f} \!  \Uparrow  $ and (3) for any $S$ such that
  $\textit{ftv} \, \ottsym{(}  S  \ottsym{(}  \ottnt{f}  \ottsym{)}  \ottsym{)}  \ottsym{=}   \emptyset $, it holds that $S  \ottsym{(}  \ottnt{f}  \ottsym{)} \, \longmapsto_{\textsf{\textup{B}\relax}\relax}^\ast \, \textsf{\textup{blame}\relax} \, \ell$ for
  some $\ell$.
\end{theorem}

\iffull
We can show that the term $\ottnt{f}  \ottsym{=}  \ottsym{(}  \ottsym{(}   \lambda  \ottmv{x} \!:\!  \ottmv{X}  .\,  \ottsym{(}  \ottmv{x}  \ottsym{:}   \ottmv{X} \Rightarrow  \unskip ^ { \ell }  \!  \star \Rightarrow  \unskip ^ { \ell }  \! \star  \!\rightarrow\!  \star    \ottsym{)}  \, \ottsym{(}  \ottmv{x}  \ottsym{:}   \ottmv{X} \Rightarrow  \unskip ^ { \ell }  \! \star   \ottsym{)}  \ottsym{)}  \ottsym{:}   \ottmv{X}  \!\rightarrow\!  \star \Rightarrow  \unskip ^ { \ell }  \! \star  \!\rightarrow\!  \star   \ottsym{)} \, \ottsym{(}  \ottsym{(}   \lambda  \ottmv{x} \!:\!  \star  .\,  \ottsym{(}  \ottmv{x}  \ottsym{:}   \star \Rightarrow  \unskip ^ { \ell }  \! \star  \!\rightarrow\!  \star   \ottsym{)}  \, \ottmv{x}  \ottsym{)}  \ottsym{:}   \star  \!\rightarrow\!  \star \Rightarrow  \unskip ^ { \ell }  \! \star   \ottsym{)}$ witnesses this theorem.
\else
\begin{proof}
  Let $\ottnt{f}  \ottsym{=}  \ottsym{(}  \ottsym{(}   \lambda  \ottmv{x} \!:\!  \ottmv{X}  .\,  \ottsym{(}  \ottmv{x}  \ottsym{:}   \ottmv{X} \Rightarrow  \unskip ^ { \ell }  \!  \star \Rightarrow  \unskip ^ { \ell }  \! \star  \!\rightarrow\!  \star    \ottsym{)}  \, \ottsym{(}  \ottmv{x}  \ottsym{:}   \ottmv{X} \Rightarrow  \unskip ^ { \ell }  \! \star   \ottsym{)}  \ottsym{)}  \ottsym{:}   \ottmv{X}  \!\rightarrow\!  \star \Rightarrow  \unskip ^ { \ell }  \! \star  \!\rightarrow\!  \star   \ottsym{)} \, \ottsym{(}  \ottsym{(}   \lambda  \ottmv{x} \!:\!  \star  .\,  \ottsym{(}  \ottmv{x}  \ottsym{:}   \star \Rightarrow  \unskip ^ { \ell }  \! \star  \!\rightarrow\!  \star   \ottsym{)}  \, \ottmv{x}  \ottsym{)}  \ottsym{:}   \star  \!\rightarrow\!  \star \Rightarrow  \unskip ^ { \ell }  \! \star   \ottsym{)}$.
  We can show that $ \emptyset   \vdash  \ottnt{f}  \ottsym{:}  \star$ and $ \ottnt{f} \!  \Uparrow  $ and $S  \ottsym{(}  \ottnt{f}  \ottsym{)} \, \longmapsto_{\textsf{\textup{B}\relax}\relax} \, \textsf{\textup{blame}\relax} \, \ell'$ for any $S$ where $\textit{ftv} \, \ottsym{(}  S  \ottsym{(}  \ottnt{f}  \ottsym{)}  \ottsym{)}  \ottsym{=}   \emptyset $.
\end{proof}
\fi
\iffull This term \else The term $\ottnt{f}$ in the proof \fi is essentially the combinator
$\Omega = \ottsym{(}   \lambda  \ottmv{x} .\,  \ottmv{x}  \, \ottmv{x}  \ottsym{)} \, \ottsym{(}   \lambda  \ottmv{x} .\,  \ottmv{x}  \, \ottmv{x}  \ottsym{)}$, for which we would need recursive
types to give a type (if we did not use the dynamic type).
Correspondingly, evaluation of $\ottnt{f}$ diverges while generating type
substitutions and fresh type variables \emph{infinitely often}:
$\ottmv{X}$ is instantiated to $\ottmv{X_{{\mathrm{1}}}}  \!\rightarrow\!  \ottmv{X_{{\mathrm{2}}}}$ by \rnp{R\_InstArrow}, $\ottmv{X_{{\mathrm{1}}}}$
is instantiated to $\ottmv{X_{{\mathrm{3}}}}  \!\rightarrow\!  \ottmv{X_{{\mathrm{4}}}}$, $\ottmv{X_{{\mathrm{3}}}}$ is instantiated to
$\ottmv{X_{{\mathrm{5}}}}  \!\rightarrow\!  \ottmv{X_{{\mathrm{6}}}}$, and so on.  On the other hand, if we instantiate the
type variable $\ottmv{X}$ in $\ottnt{f}$ with a finite static type without
type variables, evaluation will eventually reach a cast
$w  \ottsym{:}   \star  \!\rightarrow\!  \star \Rightarrow  \unskip ^ { \ell_{{\mathrm{1}}} }  \!  \star \Rightarrow  \unskip ^ { \ell_{{\mathrm{2}}} }  \! \iota  $ and fail.

Completeness of ordinary static type inference means that, if
there is some type substitution that makes a given program well typed,
the type inference algorithm succeeds and finds a more general type
substitution.  A similar property holds for DTI: if there is some type
substitution that makes a given program normally terminating,
evaluation with DTI also results in a related value and the obtained
type substitution is more general.  We can also prove that, if there is
some type substitution that makes a program diverge, then evaluation
with DTI also diverges.

The main lemma to prove completeness is below.  It intuitively means
completeness for one step: if a term $S  \ottsym{(}  f  \ottsym{)}$ evaluates in one step
and it does not result in blame, then the original, uninstantiated
term $f$ also evaluates with DTI.
\AI{Removed: ``and the obtained type substitution is more
  general than the first type substitution $S$.''}
The assumption ``it does not result in blame'', which means that the
first type substitution is a good one, is crucial.

\begin{lemma}[name=,restate=lemCompletenessValueEvalStep] \label{lem:completeness_value_eval_step}
  If $ \emptyset   \vdash  \ottnt{f}  \ottsym{:}  \ottnt{U}$ and $S  \ottsym{(}  \ottnt{f}  \ottsym{)} \,  \xmapsto{ \mathmakebox[0.4em]{} S'_{{\mathrm{1}}} \mathmakebox[0.3em]{} }  \, \ottnt{f'}$ and $\ottnt{f'} \,  \centernot{\xmapsto{ \mathmakebox[0.4em]{} S_{{\mathrm{0}}} \mathmakebox[0.3em]{} } }\hspace{-0.4em}{}^\ast \hspace{0.2em}  \, \textsf{\textup{blame}\relax} \, \ell$ for any $S_{{\mathrm{0}}}$ and $\ell$,
  then $\ottnt{f} \,  \xmapsto{ \mathmakebox[0.4em]{} S' \mathmakebox[0.3em]{} }  \, \ottnt{f''}$ and $S''  \ottsym{(}  \ottnt{f''}  \ottsym{)}  \ottsym{=}  \ottnt{f'}$
  for some $S'$, $S''$, and $\ottnt{f''}$.
\end{lemma}

\iffull\else
\begin{proof}
  By case analysis on the evaluation rule applied to $S  \ottsym{(}  \ottnt{f}  \ottsym{)}$.
\end{proof}
\fi

We state the completeness property of DTI as follows:

\begin{theorem}[name=Completeness of Dynamic Type Inference,restate=thmCompleteness] \label{thm:completeness}
 Suppose $ \emptyset   \vdash  \ottnt{f}  \ottsym{:}  \ottnt{U}$.
  \begin{enumerate}
    \item If $S  \ottsym{(}  \ottnt{f}  \ottsym{)} \,  \xmapsto{ \mathmakebox[0.4em]{} S' \mathmakebox[0.3em]{} }\hspace{-0.4em}{}^\ast \hspace{0.2em}  \, w$,
      then $\ottnt{f} \,  \xmapsto{ \mathmakebox[0.4em]{} S'' \mathmakebox[0.3em]{} }\hspace{-0.4em}{}^\ast \hspace{0.2em}  \, w'$ and $S'''  \ottsym{(}  w'  \ottsym{)}  \ottsym{=}  w$
      for some $w'$, $S''$, and $S'''$.
    \item If $ S  \ottsym{(}  \ottnt{f}  \ottsym{)} \!  \Uparrow  $, then $ \ottnt{f} \!  \Uparrow  $.
  \end{enumerate}
\end{theorem}

\iffull\else
\begin{proof}
  By Lemma \ref{lem:completeness_value_eval_step}.
\end{proof}
\fi

In the first item, $S''$ is required to relate the two values
because DTI may still leave some type variables uninstantiated.  In
fact, $S''$ witnesses that $S'$ is more general than $S$.


\paragraph{Remark:} As we discussed in
Section~\ref{sec:intro-problem}, an alternative to DTI, which is to
substitute $ \star $ for type variables that type inference left
uninstantiated, breaks the following property, which could be viewed
as soundness:
\begin{quotation}
  If $ \emptyset   \vdash  \ottnt{f}  \ottsym{:}  \ottnt{U}$, $S$ substitutes $ \star $ for all type
  variables in $\ottnt{f}$, and $S  \ottsym{(}  \ottnt{f}  \ottsym{)} \, \longmapsto_{\textsf{\textup{B}\relax}\relax}^\ast \, w$, then there is some
  \emph{static} type substitution $S'$ such that $S'  \ottsym{(}  \ottnt{f}  \ottsym{)} \, \longmapsto_{\textsf{\textup{B}\relax}\relax}^\ast \, w'$
  for some $w'$.
\end{quotation}
The term
\[
  \begin{array}{l}
   ( \lambda  \ottmv{x} \!:\!  \star  .\,  \ottsym{(}  \ottsym{(}  \ottsym{(}  \ottmv{x}  \ottsym{:}   \star \Rightarrow  \unskip ^ { \ell_{{\mathrm{1}}} }  \! \star  \!\rightarrow\!  \star   \ottsym{)} \, \ottsym{(}   2   \ottsym{:}    \textsf{\textup{int}\relax}  \Rightarrow  \unskip ^ { \ell_{{\mathrm{2}}} }  \! \star   \ottsym{)}  \ottsym{)}  \ottsym{:}   \star \Rightarrow  \unskip ^ { \ell_{{\mathrm{3}}} }  \! \star  \!\rightarrow\!  \star   \ottsym{)} \, \ottsym{(}   \textsf{\textup{true}\relax}   \ottsym{:}    \textsf{\textup{bool}\relax}  \Rightarrow  \unskip ^ { \ell_{{\mathrm{4}}} }  \! \star   \ottsym{)} )
    \\ \quad \ottsym{(}  \ottsym{(}   \lambda  \ottmv{y_{{\mathrm{1}}}} \!:\!  \ottmv{Y}  .\,   \lambda  \ottmv{y_{{\mathrm{2}}}} \!:\!  \ottmv{Y}  .\,   \textsf{\textup{ if }\relax}  \ottnt{b}  \textsf{\textup{ then }\relax}  \ottmv{y_{{\mathrm{1}}}}  \textsf{\textup{ else }\relax}  \ottmv{y_{{\mathrm{2}}}}     \ottsym{)}  \ottsym{:}   \star  \!\rightarrow\!  \star  \!\rightarrow\!  \star \Rightarrow  \unskip ^ { \ell_{{\mathrm{5}}} }  \! \star   \ottsym{)}
  \end{array}
\]
is a counterexample.  However, completeness would be satisfied.  The
semantics where uninstantiated type variables are regarded as
distinguished base types, as in \citet{DBLP:conf/popl/GarciaC15},
would satisfy neither the second item of soundness nor completeness.

\section{The Gradual Guarantee}
\label{sec:gradual-guarantee}

The gradual guarantee~\cite{DBLP:conf/snapl/SiekVCB15}, a property capturing the
essence of gradual code evolution, is one of the important criteria for
gradually typed languages.
%
%
%
%
It formalizes the intuition that, in a gradual type
system, giving more ``precise'' type annotations can find more typing
errors either statically or dynamically but otherwise does not change the behavior
of programs.  In a simple setting, a type is more precise than
another, if the former is obtained by replacing some occurrences of
$ \star $ by other types.  For example, $ \textsf{\textup{int}\relax}   \!\rightarrow\!  \star  \!\rightarrow\!   \textsf{\textup{bool}\relax} $ is more precise
than $\star  \!\rightarrow\!  \star$.  The statement consists of two parts: one about
typeability (the \emph{static} gradual guarantee) and the other about
evaluation (the \emph{dynamic} gradual guarantee).

The static gradual guarantee states that, if term $e$ is well
typed, then so is a term that is the same as $e$ except that it is
less precisely annotated.

The dynamic gradual guarantee ensures that a more precisely annotated program
and a less precisely annotated one behave in the same way as far as the more
precisely annotated program does not raise blame.
For example, let us consider the following gradually typed program (which is
similar to the example used in Section~\ref{sec:introduction}):
\[
  \ottsym{(}   \lambda  \ottmv{x} \!:\!  \star  \!\rightarrow\!  \star  .\,  \ottmv{x}  \,  2   \ottsym{)} \, \ottsym{(}   \lambda  \ottmv{y} \!:\!  \star  .\,  \ottmv{y}   \ottsym{)}
\]
This program is translated into a term in the blame calculus and evaluates to
value $ 2   \ottsym{:}    \textsf{\textup{int}\relax}  \Rightarrow  \unskip ^ { \ell }  \! \star $.
On one hand, we can obtain a less precisely annotated program by replacing the
type annotation $\star  \!\rightarrow\!  \star$ for variable $\ottmv{x}$ with $ \star $, which is the least
precise type.
\[
  \ottsym{(}   \lambda  \ottmv{x} \!:\!  \star  .\,  \ottmv{x}  \,  2   \ottsym{)} \, \ottsym{(}   \lambda  \ottmv{y} \!:\!  \star  .\,  \ottmv{y}   \ottsym{)}
\]
This program evaluates to the same value $ 2   \ottsym{:}    \textsf{\textup{int}\relax}  \Rightarrow  \unskip ^ { \ell }  \! \star $ as the gradual
guarantees expects.
On the other hand, we can get a more precisely annotated program by replacing
the type annotation $ \star $ for variable $\ottmv{y}$ with $ \textsf{\textup{int}\relax} $.
\[
  \ottsym{(}   \lambda  \ottmv{x} \!:\!  \star  \!\rightarrow\!  \star  .\,  \ottmv{x}  \,  2   \ottsym{)} \, \ottsym{(}   \lambda  \ottmv{y} \!:\!   \textsf{\textup{int}\relax}   .\,  \ottmv{y}   \ottsym{)}
\]
Then, it also evaluates to the same value $ 2   \ottsym{:}    \textsf{\textup{int}\relax}  \Rightarrow  \unskip ^ { \ell }  \! \star $.
However, more precisely annotated programs may raise blame even if less
precisely annotated ones do not.
For example, let us consider the case that $\ottmv{y}$ in the above program is given
a wrong type annotation, e.g., $ \textsf{\textup{bool}\relax} $.
\[
  \ottsym{(}   \lambda  \ottmv{x} \!:\!  \star  \!\rightarrow\!  \star  .\,  \ottmv{x}  \,  2   \ottsym{)} \, \ottsym{(}   \lambda  \ottmv{y} \!:\!   \textsf{\textup{bool}\relax}   .\,  \ottmv{y}   \ottsym{)}
\]
Then, it raises blame because the evaluation triggers a sequence of casts
$ 2   \ottsym{:}    \textsf{\textup{int}\relax}  \Rightarrow  \unskip ^ { \ell_{{\mathrm{1}}} }  \!  \star \Rightarrow  \unskip ^ { \ell_{{\mathrm{2}}} }  \!  \textsf{\textup{bool}\relax}   $, which fails.

In the rest of this section,
we show that the ITGL satisfies both the static and dynamic gradual guarantee.
We extend the statements to take static/dynamic type inference into
account: For example, the statement of the static gradual guarantee
becomes that, given $e$, if static
type inference succeeds, then for any less precisely annotated
term $e'$ static type inference also succeeds.
To state formally, we first introduce the ITGL and translation of ITGL terms to
$\lambdaRTI$ terms.
Then, the notion of ``more (or less) precisely annotated'' ITGL terms
is formalized by a \emph{precision} relation over the ITGL terms.
Using the precision, we show the static gradual guarantee.
The dynamic gradual guarantee is proved by giving a precision relation also for
$\lambdaRTI$ and showing the two properties: (1) the translation results of
precision-related ITGL terms are precision-related; and (2) precision-related
$\lambdaRTI$ terms behave equivalently if the more precisely annotated term does
not raise blame.

\subsection{ITGL: the Implicitly Typed Gradual Language}
\subsubsection{Definition}
This section reviews the definition of the ITGL~\cite{DBLP:conf/popl/GarciaC15}.
\begin{figure}
  {\small
    \input{figures/def_ITGL}}
  \caption{The ITGL.}
  \label{fig:def_ITGL}
\end{figure}
Figure~\ref{fig:def_ITGL} shows the syntax and the typing rules of the ITGL and
the cast insertion rules to translate the ITGL to $\lambdaRTI$.
The syntax is from the standard lambda calculus, except that 
it provides two kinds of abstractions: one is $ \lambda  \ottmv{x} \!:\!  \ottnt{U}  .\,  e $, where the type of
argument $\ottmv{x}$ is given as $\ottnt{U}$ explicitly, and the other is $ \lambda  \ottmv{x} .\,  e $,
where the type of $\ottmv{x}$ is implicit.
A $ \textsf{\textup{let}\relax} $-expression $ \textsf{\textup{let}\relax} \,  \ottmv{x}  =  v  \textsf{\textup{ in }\relax}  e $ allows $v$ to be polymorphic in
$e$; type variables to be generalized are implicit in terms of
the ITGL, while they are explicit in $\lambdaRTI$.

The typing rules are essentially the same as what is called schematic
typing in \citet{DBLP:conf/popl/GarciaC15} modulo a few minor
adaptations.
The rules \rnp{IT\_VarP} and \rnp{IT\_AbsI} require implicit types to be static.
The rule \rnp{IT\_Op} says that types of arguments have to be consistent with the
argument types of $ \mathit{op} $.
The rule \rnp{IT\_App} for application $e_{{\mathrm{1}}} \, e_{{\mathrm{2}}}$ is standard in gradual
typing~\cite{DBLP:conf/snapl/SiekVCB15}.
The type $\ottnt{U_{{\mathrm{1}}}}$ of $e_{{\mathrm{1}}}$ is consistent with function type $\ottnt{U_{{\mathrm{11}}}}  \!\rightarrow\!  \ottnt{U_{{\mathrm{12}}}}$
obtained by operation $ \triangleright $, which is defined as
\[\begin{array}{ccc@{\qquad\qquad}ccc}
  \star       & \triangleright & \star  \!\rightarrow\!  \star &
 \ottnt{U_{{\mathrm{1}}}}  \!\rightarrow\!  \ottnt{U_{{\mathrm{2}}}} & \triangleright & \ottnt{U_{{\mathrm{1}}}}  \!\rightarrow\!  \ottnt{U_{{\mathrm{2}}}},
  \end{array}\]
and the type $\ottnt{U_{{\mathrm{2}}}}$ of $e_{{\mathrm{2}}}$ has to be consistent with the argument type of
$\ottnt{U_{{\mathrm{11}}}}  \!\rightarrow\!  \ottnt{U_{{\mathrm{12}}}}$.
The rule \rnp{IT\_Let} is standard in languages with
let-polymorphism~\cite{Milner78JCSS} except the condition on type variables to
be generalized.
We do not allow type variables that appear in type annotations of $v$ to be
generalized because allowing it invalidates the property that
$ \textsf{\textup{let}\relax} \,  \ottmv{x}  =  v  \textsf{\textup{ in }\relax}  e $ is well typed if and only if so is
$e  [  \ottmv{x}  \ottsym{:=}  v  ]$ (if $e$ refers to $\ottmv{x}$).
For example, let us consider ITGL program
\[
  \textsf{\textup{let}\relax} \,  \ottmv{x}  =   \lambda  \ottmv{y} \!:\!  \ottmv{X}  .\,  \ottmv{y}   \textsf{\textup{ in }\relax}  \ottsym{(}  \ottmv{x} \,  2   \ottsym{,}  \ottmv{x} \,  \textsf{\textup{true}\relax}   \ottsym{)} .
\]
This program would be well typed if $\ottmv{X}$ could be generalized.
However, the result of expanding let
\[
 \ottsym{(}  \ottmv{x} \,  2   \ottsym{,}  \ottmv{x} \,  \textsf{\textup{true}\relax}   \ottsym{)}  [  \ottmv{x}  \ottsym{:=}   \lambda  \ottmv{y} \!:\!  \ottmv{X}  .\,  \ottmv{y}   ] = \ottsym{(}  \ottsym{(}   \lambda  \ottmv{y} \!:\!  \ottmv{X}  .\,  \ottmv{y}   \ottsym{)} \,  2   \ottsym{,}  \ottsym{(}   \lambda  \ottmv{y} \!:\!  \ottmv{X}  .\,  \ottmv{y}   \ottsym{)} \,  \textsf{\textup{true}\relax}   \ottsym{)}
\]
is not well typed because $\ottmv{X}$ could not be instantiated in two ways:
$ \textsf{\textup{int}\relax} $ and $ \textsf{\textup{bool}\relax} $.
Our restriction to generalizable type variables rejects not only the latter but
also the former; OCaml seems to adopt the same strategy (such type variables
are generalized only at the top level).

Translation $\Gamma  \vdash  e  \rightsquigarrow  \ottnt{f}  \ottsym{:}  \ottnt{U}$ of ITGL term $e$ of type $\ottnt{U}$ to
$\lambdaRTI$ term $\ottnt{f}$ is achieved by inserting casts and making implicit
type information explicit.
The cast insertion rules, which are shown in the lower half of
Figure~\ref{fig:def_ITGL}, are standard~\cite{DBLP:conf/snapl/SiekVCB15} except
\rnp{CI\_VarP} and \rnp{CI\_LetP}.
The rule \rnp{CI\_VarP} makes type instantiations explicit in terms.
The side condition is the same as that of \rnp{T\_VarP}.
The rule \rnp{CI\_LetP} is similar to \rnp{T\_LetP}, but it allows not only type
variables $ \overrightarrow{ \ottmv{X_{\ottmv{i}}} } $ that appear in the type $\ottnt{U_{{\mathrm{1}}}}$ of polymorphic value
$w_{{\mathrm{1}}}$ but also $ \overrightarrow{ \ottmv{Y_{\ottmv{i}}} } $ that appear only in $w_{{\mathrm{1}}}$---they are implicit in ITGL
terms---to be generalized for identifying $ \textsf{\textup{let}\relax} \,  \ottmv{x}  =  v  \textsf{\textup{ in }\relax}  e $ with $e  [  \ottmv{x}  \ottsym{:=}  v  ]$
semantically.
For example, let us consider the following ITGL term:
\[
  \textsf{\textup{let}\relax} \,  \ottmv{x}  =   \lambda  \ottmv{y} .\,  \ottsym{(}   \ottsym{(}   \lambda  \ottmv{z} .\,  \ottmv{z}   \ottsym{)}  ::  \star  \!\rightarrow\!  \star   \ottsym{)}  \, \ottmv{y}  \textsf{\textup{ in }\relax}  \ottsym{(}  \ottmv{x} \,  2   \ottsym{,}  \ottmv{x} \,  \textsf{\textup{true}\relax}   \ottsym{)} 
\]
where $ e  ::  \ottnt{U} $ means $\ottsym{(}   \lambda  \ottmv{x} \!:\!  \ottnt{U}  .\,  \ottmv{x}   \ottsym{)} \, e$.
If type variables that appear in only $\ottnt{U_{{\mathrm{1}}}}$ could be generalized,
it would be translated to
\[
  \textsf{\textup{let}\relax} \,  \ottmv{x}  =   \Lambda   \ottmv{X} .\,   \lambda  \ottmv{y} \!:\!  \ottmv{X}  .\,  \ottsym{(}  \ottsym{(}   \lambda  \ottmv{z} \!:\!  \ottmv{Y}  .\,  \ottmv{z}   \ottsym{)}  \ottsym{:}   \ottmv{Y}  \!\rightarrow\!  \ottmv{Y} \Rightarrow  \unskip ^ { \ell_{{\mathrm{1}}} }  \! \star  \!\rightarrow\!  \star   \ottsym{)}   \, \ottsym{(}  \ottmv{y}  \ottsym{:}   \ottmv{X} \Rightarrow  \unskip ^ { \ell_{{\mathrm{2}}} }  \! \star   \ottsym{)}  \textsf{\textup{ in }\relax}  \ottsym{(}  \ottmv{x}  [   \textsf{\textup{int}\relax}   ] \,  2   \ottsym{,}  \ottmv{x}  [   \textsf{\textup{bool}\relax}   ] \,  \textsf{\textup{true}\relax}   \ottsym{)} 
\]
where $\ottmv{Y}$ is not generalized.
The evaluation of this term results in blame because the monomorphic type
variable $\ottmv{Y}$ will be instantiated with different base types $ \textsf{\textup{int}\relax} $ and
$ \textsf{\textup{bool}\relax} $ by DTI.
However, the evaluation based on let-expansion does not trigger blame:
\[\begin{array}{l@{\;}l}
 & \ottsym{(}  \ottmv{x} \,  2   \ottsym{,}  \ottmv{x} \,  \textsf{\textup{true}\relax}   \ottsym{)}  [  \ottmv{x}  \ottsym{:=}   \lambda  \ottmv{y} .\,  \ottsym{(}   \ottsym{(}   \lambda  \ottmv{z} .\,  \ottmv{z}   \ottsym{)}  ::  \star  \!\rightarrow\!  \star   \ottsym{)}  \, \ottmv{y}  ] \\
 =                        & \ottsym{(}  \ottsym{(}   \lambda  \ottmv{y} .\,  \ottsym{(}   \ottsym{(}   \lambda  \ottmv{z} .\,  \ottmv{z}   \ottsym{)}  ::  \star  \!\rightarrow\!  \star   \ottsym{)}  \, \ottmv{y}  \ottsym{)} \,  2   \ottsym{,}  \ottsym{(}   \lambda  \ottmv{y} .\,  \ottsym{(}   \ottsym{(}   \lambda  \ottmv{z} .\,  \ottmv{z}   \ottsym{)}  ::  \star  \!\rightarrow\!  \star   \ottsym{)}  \, \ottmv{y}  \ottsym{)} \,  \textsf{\textup{true}\relax}   \ottsym{)} \\
 \rightsquigarrow         & (\ottsym{(}   \lambda  \ottmv{y} \!:\!   \textsf{\textup{int}\relax}   .\,  \ottsym{(}  \ottsym{(}   \lambda  \ottmv{z} \!:\!  \ottmv{Y_{{\mathrm{1}}}}  .\,  \ottmv{z}   \ottsym{)}  \ottsym{:}   \ottmv{Y_{{\mathrm{1}}}}  \!\rightarrow\!  \ottmv{Y_{{\mathrm{1}}}} \Rightarrow  \unskip ^ { \ell_{{\mathrm{1}}} }  \! \star  \!\rightarrow\!  \star   \ottsym{)}  \, \ottsym{(}  \ottmv{y}  \ottsym{:}    \textsf{\textup{int}\relax}  \Rightarrow  \unskip ^ { \ell_{{\mathrm{2}}} }  \! \star   \ottsym{)}  \ottsym{)} \,  2 , \\
                          & \quad \ottsym{(}   \lambda  \ottmv{y} \!:\!   \textsf{\textup{bool}\relax}   .\,  \ottsym{(}  \ottsym{(}   \lambda  \ottmv{z} \!:\!  \ottmv{Y_{{\mathrm{2}}}}  .\,  \ottmv{z}   \ottsym{)}  \ottsym{:}   \ottmv{Y_{{\mathrm{2}}}}  \!\rightarrow\!  \ottmv{Y_{{\mathrm{2}}}} \Rightarrow  \unskip ^ { \ell_{{\mathrm{3}}} }  \! \star  \!\rightarrow\!  \star   \ottsym{)}  \, \ottsym{(}  \ottmv{y}  \ottsym{:}    \textsf{\textup{bool}\relax}  \Rightarrow  \unskip ^ { \ell_{{\mathrm{4}}} }  \! \star   \ottsym{)}  \ottsym{)} \,  \textsf{\textup{true}\relax} )\footnotemark \\
  \xmapsto{ \mathmakebox[0.4em]{} [  \ottmv{Y_{{\mathrm{1}}}}  :=   \textsf{\textup{int}\relax}   ] \mathmakebox[0.3em]{} }\hspace{-0.4em}{}^\ast \hspace{0.2em}   & \ottsym{(}  \ottsym{(}   2   \ottsym{:}    \textsf{\textup{int}\relax}  \Rightarrow  \unskip ^ { \ell_{{\mathrm{1}}} }  \! \star   \ottsym{)}  \ottsym{,}  \ottsym{(}   \lambda  \ottmv{y} \!:\!   \textsf{\textup{bool}\relax}   .\,  \ottsym{(}  \ottsym{(}   \lambda  \ottmv{z} \!:\!  \ottmv{Y_{{\mathrm{2}}}}  .\,  \ottmv{z}   \ottsym{)}  \ottsym{:}   \ottmv{Y_{{\mathrm{2}}}}  \!\rightarrow\!  \ottmv{Y_{{\mathrm{2}}}} \Rightarrow  \unskip ^ { \ell_{{\mathrm{3}}} }  \! \star  \!\rightarrow\!  \star   \ottsym{)}  \, \ottsym{(}  \ottmv{y}  \ottsym{:}    \textsf{\textup{bool}\relax}  \Rightarrow  \unskip ^ { \ell_{{\mathrm{4}}} }  \! \star   \ottsym{)}  \ottsym{)} \,  \textsf{\textup{true}\relax}   \ottsym{)} \\
  \xmapsto{ \mathmakebox[0.4em]{} [  \ottmv{Y_{{\mathrm{2}}}}  :=   \textsf{\textup{bool}\relax}   ] \mathmakebox[0.3em]{} }\hspace{-0.4em}{}^\ast \hspace{0.2em}  & \ottsym{(}  \ottsym{(}   2   \ottsym{:}    \textsf{\textup{int}\relax}  \Rightarrow  \unskip ^ { \ell_{{\mathrm{1}}} }  \! \star   \ottsym{)}  \ottsym{,}  \ottsym{(}   \textsf{\textup{true}\relax}   \ottsym{:}    \textsf{\textup{bool}\relax}  \Rightarrow  \unskip ^ { \ell_{{\mathrm{3}}} }  \! \star   \ottsym{)}  \ottsym{)}
  \end{array}\]
\footnotetext{We omit trivial identity functions inserted due to type ascription
here.}
By allowing generalization of type variables appearing in not only $\ottnt{U_{{\mathrm{1}}}}$ but
also $w_{{\mathrm{1}}}$ (but not $v_{{\mathrm{1}}}$), we achieve the same semantics as let-expansion.

We show that the cast insertion is type-preserving and type safety of
the ITGL.
\begin{theorem}[name=Cast Insertion is Type-Preserving,restate=thmCastInsertion]
 \label{thm:ci_type_preservation}
 If $\Gamma  \vdash  e  \ottsym{:}  \ottnt{U}$,
 then $\Gamma  \vdash  e  \rightsquigarrow  \ottnt{f}  \ottsym{:}  \ottnt{U}$ and $\Gamma  \vdash  \ottnt{f}  \ottsym{:}  \ottnt{U}$ for some $\ottnt{f}$.
\end{theorem}
\begin{definition}[Evaluation of ITGL terms]
 We write $ \langle  \Gamma   \vdash   e  :  \ottnt{U}  \rangle   \xmapsto{ \mathmakebox[0.4em]{} S \mathmakebox[0.3em]{} }\hspace{-0.4em}{}^\ast \hspace{0.2em}    \ottnt{f} $
 if $\Gamma  \vdash  e  \rightsquigarrow  \ottnt{f'}  \ottsym{:}  \ottnt{U}$ and $\ottnt{f'} \,  \xmapsto{ \mathmakebox[0.4em]{} S \mathmakebox[0.3em]{} }\hspace{-0.4em}{}^\ast \hspace{0.2em}  \, \ottnt{f}$ for some $\ottnt{f'}$.
 We also write $ \langle  \Gamma   \vdash   e  :  \ottnt{U}  \rangle   \Uparrow  $
 if $\Gamma  \vdash  e  \rightsquigarrow  \ottnt{f}  \ottsym{:}  \ottnt{U}$ and $ \ottnt{f} \!  \Uparrow  $ for some $\ottnt{f}$.
\end{definition}
\begin{corollary}[name=Type Safety of the ITGL,restate=corTypeSafeITGL]
 If $ \emptyset   \vdash  e  \ottsym{:}  \ottnt{U}$, then:
 \begin{itemize}
  \item $ \langle   \emptyset    \vdash   e  :  \ottnt{U}  \rangle   \xmapsto{ \mathmakebox[0.4em]{} S \mathmakebox[0.3em]{} }\hspace{-0.4em}{}^\ast \hspace{0.2em}    w $ for some $S$ and $w$
        such that $ \emptyset   \vdash  w  \ottsym{:}  S  \ottsym{(}  \ottnt{U}  \ottsym{)}$;
  \item $ \langle   \emptyset    \vdash   e  :  \ottnt{U}  \rangle   \xmapsto{ \mathmakebox[0.4em]{} S \mathmakebox[0.3em]{} }\hspace{-0.4em}{}^\ast \hspace{0.2em}    \textsf{\textup{blame}\relax} \, \ell $
        for some $S$ and $\ell$; or
  \item $ \langle   \emptyset    \vdash   e  :  \ottnt{U}  \rangle   \Uparrow  $.
 \end{itemize}
\end{corollary}


\subsubsection{Precision}
We introduce precision relations for the ITGL to formalize more precise types
and more precisely annotated terms.

\paragraph{Type precision}
As we have already mentioned, intuitively, type $\ottnt{U}$ is more
precise than type $\ottnt{U'}$ if $\ottnt{U}$ is obtained by replacing some
occurrences of the dynamic type in $\ottnt{U'}$ with other types.
We generalize this notion of precision to type variables.
Type variables are similar to the dynamic type in that type variables that
appear in well-typed terms can be instantiated with any type as values of any
type can flow to the dynamic type.  \TS{We may need to refine this sentence...}
However, they differ from the dynamic type in that all occurrences of a type
variable have to be instantiated with a single static type.
For example, $ \textsf{\textup{int}\relax}   \!\rightarrow\!   \textsf{\textup{bool}\relax} $ can be considered---and \emph{is} in our
precision---more precise than $\ottmv{X}$ because we get the former by instantiating
$\ottmv{X}$ with $ \textsf{\textup{int}\relax}   \!\rightarrow\!   \textsf{\textup{bool}\relax} $, while it cannot be compared with $\ottmv{X}  \!\rightarrow\!  \ottmv{X}$
because $\ottmv{X}$ is not allowed to be instantiated in two different ways: $ \textsf{\textup{int}\relax} $ and
$ \textsf{\textup{bool}\relax} $.

\begin{figure}
  {\small
    \input{figures/def_type_precision}}
  \vspacerules
  {\small
    \input{figures/def_ITGL_term_precision}}
  \caption{Precision of the ITGL.}
  \label{fig:prec_ITGL}
\end{figure}
To ensure that all occurrences of a type variable are instantiated in the same
way, we equip type precision $ \ottnt{U}   \sqsubseteq _{ S }  \ottnt{U'} $, which means that $\ottnt{U}$
is more precise than $\ottnt{U'}$, with type substitution $S$, which gives
instantiations of type variables in $\ottnt{U'}$.
The type precision rules, given in the top of Figure~\ref{fig:prec_ITGL}, are
standard~\cite{DBLP:conf/snapl/SiekVCB15} except \rnp{P\_TyVar}: a base type
is more precise than itself \rnp{P\_IdBase}; components of
precision-related function types are also precision-related \rnp{P\_Arrow}; and
the dynamic type is the least precise type \rnp{P\_Dyn}.
The rule \rnp{P\_TyVar} allows type variables to be instantiated according to
$S$.
If we want to relate $\ottmv{X}$ to itself, we can give substitution $[  \ottmv{X}  :=  \ottmv{X}  ]$.
This type precision cooperates well with DTI in the sense that, if the
evaluation of a more precisely annotated term using type annotation
$ \textsf{\textup{int}\relax}   \!\rightarrow\!   \textsf{\textup{bool}\relax} $ does not raise blame, then that of a less precisely
annotated one using $\ottmv{X}$ does not, either, because DTI would
instantiate $\ottmv{X}$ with $ \textsf{\textup{int}\relax}   \!\rightarrow\!   \textsf{\textup{bool}\relax} $ or a less precise static type (such as $\ottmv{Y}  \!\rightarrow\!   \textsf{\textup{bool}\relax} $).
%


\paragraph{Term precision}
Rules of term precision $ e   \sqsubseteq _{ S }  e' $, which means that $e$ is
more precisely annotated than $e'$, are given in the bottom of
Figure~\ref{fig:prec_ITGL}.
All rules but \rnp{IP\_AbsIE}, \rnp{IP\_AbsEI}, and \rnp{IP\_AbsE} are just
for compatibility.
Lambda abstraction $ \lambda  \ottmv{x} .\,  e $ is more precisely annotated than $ \lambda  \ottmv{x} \!:\!  \star  .\,  e $
\rnp{IP\_AbsIE} because the static type inference gives variable $\ottmv{x}$ in
$ \lambda  \ottmv{x} .\,  e $ a static type (if any), which is more precise than the dynamic type.
Lambda abstraction $ \lambda  \ottmv{x} \!:\!  \ottnt{T}  .\,  e $ is more precisely annotated than $ \lambda  \ottmv{x} .\,  e $
\rnp{IP\_AbsEI} because the static type inference algorithm by
\citet{DBLP:conf/popl/GarciaC15} gives a principal static type to $\ottmv{x}$ in
$ \lambda  \ottmv{x} .\,  e $ and it should be less precise than $\ottnt{T}$ under some type
substitution due to principality.
Finally, $ \lambda  \ottmv{x} \!:\!  \ottnt{U}  .\,  e $ is more precisely annotated than $ \lambda  \ottmv{x} \!:\!  \ottnt{U'}  .\,  e' $ if $\ottnt{U}$
is more precise than $\ottnt{U'}$ and $e$ is more precisely annotated than
$e'$ under $S$ \rnp{IP\_AbsE}.

\subsection{The Static Gradual Guarantee}
\begin{definition}[Principal Type Inference]
 We suppose that there is a partial function $PT(\Gamma,e)$ that (1) if
 $S'  \ottsym{(}  \Gamma  \ottsym{)}  \vdash  S'  \ottsym{(}  e  \ottsym{)}  \ottsym{:}  \ottnt{U'}$ for some $S'$ and $\ottnt{U'}$,
 produces a pair $(S,\ottnt{U})$ such that
 (a) $S  \ottsym{(}  \Gamma  \ottsym{)}  \vdash  S  \ottsym{(}  e  \ottsym{)}  \ottsym{:}  \ottnt{U}$ and
 (b) for any $S_{{\mathrm{1}}}$ and $\ottnt{U_{{\mathrm{1}}}}$ such that
 $S_{{\mathrm{1}}}  \ottsym{(}  \Gamma  \ottsym{)}  \vdash  S_{{\mathrm{1}}}  \ottsym{(}  e  \ottsym{)}  \ottsym{:}  \ottnt{U_{{\mathrm{1}}}}$, there exists $S_{{\mathrm{2}}}$ such that
 $S_{{\mathrm{1}}}  \ottsym{=}   S_{{\mathrm{2}}}  \circ  S $; and (2) is undefined otherwise.
 \AI{See a comment by Reviewer C.  I don't think we should follow his/her comment, though.}
 \TS{I don't think so, either.}
\end{definition}
We can obtain a principal type inference algorithm $PT$ from
\citet{DBLP:conf/popl/GarciaC15}.
Now, we show the static gradual guarantee for the ITGL.
\begin{theorem}[name=Static Gradual Guarantee,restate=thmStaticGG]
 If $ e   \sqsubseteq _{ S_{{\mathrm{0}}} }  e' $ and $PT( \emptyset , e) = (S_{{\mathrm{1}}},\ottnt{U})$,
 then $PT( \emptyset , e') = (S'_{{\mathrm{1}}}, \ottnt{U'})$ and $ \ottnt{U}   \sqsubseteq _{ S_{{\mathrm{2}}} }  \ottnt{U'} $
 for some $S_{{\mathrm{2}}}$, $S'_{{\mathrm{1}}}$, and $\ottnt{U'}$.
\end{theorem}

\subsection{Precision of $\lambdaRTI$}
As described at the beginning of this section, we show the dynamic gradual
guarantee via reduction of ITGL term precision to $\lambdaRTI$ term precision.
We denote $\lambdaRTI$ term precision by $$ \langle  \Gamma   \vdash   \ottnt{f}  :  \ottnt{U}   \sqsubseteq _{ S }  \ottnt{U'}  :  \ottnt{f'}  \dashv  \Gamma'  \rangle ,$$
which means that $\lambdaRTI$ term $\ottnt{f}$ having type $\ottnt{U}$ under $\Gamma$ is
more precisely annotated than $\ottnt{f'}$ having type $\ottnt{U'}$ under $\Gamma'$; type
substitution $S$ plays the same role as in precision of the ITGL.

\begin{figure}
  {\small
    \input{figures/def_term_precision}}
  \caption{Term precision of $\lambdaRTI$.}
  \label{fig:def_term_precision}
\end{figure}
We show precision rules of $\lambdaRTI$ in Figure~\ref{fig:def_term_precision}.
The rules \rnp{P\_Var}, \rnp{P\_Const}, \rnp{P\_Op}, \rnp{P\_Abs}, \rnp{P\_App},
and \rnp{P\_Cast} are similar to the corresponding typing rules of $\lambdaRTI$
except that they ensure that type information on the left-hand side is
more precise than that on the right-hand side.
The rule \rnp{P\_LetP} relates two let-expressions
$ \textsf{\textup{let}\relax} \,  \ottmv{x}  =   \Lambda    \overrightarrow{ \ottmv{X_{\ottmv{i}}} }  .\,  w_{{\mathrm{1}}}   \textsf{\textup{ in }\relax}  \ottnt{f_{{\mathrm{2}}}} $ and $ \textsf{\textup{let}\relax} \,  \ottmv{x}  =   \Lambda    \overrightarrow{ \ottmv{X'_{\ottmv{j}}} }  .\,  w'_{{\mathrm{1}}}   \textsf{\textup{ in }\relax}  \ottnt{f'_{{\mathrm{2}}}} $.
This rule allows polymorphic values $w_{{\mathrm{1}}}$ and $w'_{{\mathrm{1}}}$ to be related under
$S$ augmented with a type substitution $[   \overrightarrow{ \ottmv{X'_{\ottmv{j}}} }   :=   \overrightarrow{ \ottnt{T'_{\ottmv{j}}} }   ]$, which maps
bound type variables $ \overrightarrow{ \ottmv{X'_{\ottmv{j}}} } $ on the imprecise side to more precise types
$ \overrightarrow{ \ottnt{T'_{\ottmv{j}}} } $ on the precise side.
The second premise claims that $S$ does not capture bound variables
$ \overrightarrow{ \ottmv{X_{\ottmv{i}}} } $; note that free type variables in types to which $S$ maps may
appear on only the precise side.
The rules \rnp{P\_CastL} and \rnp{P\_CastR} are given because modifying type
annotations of casts changes casts generated at run time.
The rule \rnp{P\_Blame} represents that a term on the precise side may involve
blame even if one on the imprecise side does not.

\subsection{The Dynamic Gradual Guarantee}
Now, we show the dynamic gradual guarantee.
\begin{theorem}[name=Dynamic Gradual Guarantee,restate=thmDynamicGG]
  \label{thm:dynamicGG}
 Suppose that $ e   \sqsubseteq _{ S_{{\mathrm{0}}} }  e' $.
 Let
 $(S,\ottnt{U}) = PT( \emptyset ,e)$ and
 $(S',\ottnt{U'}) = PT( \emptyset , e')$.
 \begin{enumerate}
  \item \begin{itemize}
         \item If $ \langle   \emptyset    \vdash   S  \ottsym{(}  e  \ottsym{)}  :  \ottnt{U}  \rangle   \xmapsto{ \mathmakebox[0.4em]{} S_{{\mathrm{1}}} \mathmakebox[0.3em]{} }\hspace{-0.4em}{}^\ast \hspace{0.2em}    w $,
               then $ \langle   \emptyset    \vdash   S'  \ottsym{(}  e'  \ottsym{)}  :  \ottnt{U'}  \rangle   \xmapsto{ \mathmakebox[0.4em]{} S'_{{\mathrm{1}}} \mathmakebox[0.3em]{} }\hspace{-0.4em}{}^\ast \hspace{0.2em}    w' $ and
               $ \langle   \emptyset    \vdash   w  :  S_{{\mathrm{1}}}  \ottsym{(}  \ottnt{U}  \ottsym{)}   \sqsubseteq _{ S'_{{\mathrm{0}}} }  S'_{{\mathrm{1}}}  \ottsym{(}  \ottnt{U'}  \ottsym{)}  :  w'  \dashv   \emptyset   \rangle $
               for some $w'$, $S'_{{\mathrm{1}}}$, and $S'_{{\mathrm{0}}}$.
         \item If $ \langle   \emptyset    \vdash   S  \ottsym{(}  e  \ottsym{)}  :  \ottnt{U}  \rangle   \Uparrow  $,
               then $ \langle   \emptyset    \vdash   S'  \ottsym{(}  e'  \ottsym{)}  :  \ottnt{U'}  \rangle   \Uparrow  $.
        \end{itemize}
  \item \begin{itemize}
        \item If $ \langle   \emptyset    \vdash   S'  \ottsym{(}  e'  \ottsym{)}  :  \ottnt{U'}  \rangle   \xmapsto{ \mathmakebox[0.4em]{} S'_{{\mathrm{1}}} \mathmakebox[0.3em]{} }\hspace{-0.4em}{}^\ast \hspace{0.2em}    w' $,
          then either (1)
          $ \langle   \emptyset    \vdash   S  \ottsym{(}  e  \ottsym{)}  :  \ottnt{U}  \rangle   \xmapsto{ \mathmakebox[0.4em]{} S_{{\mathrm{1}}} \mathmakebox[0.3em]{} }\hspace{-0.4em}{}^\ast \hspace{0.2em}    w $ and
          $ \langle   \emptyset    \vdash   w  :  S_{{\mathrm{1}}}  \ottsym{(}  \ottnt{U}  \ottsym{)}   \sqsubseteq _{ S'_{{\mathrm{0}}} }  S'_{{\mathrm{1}}}  \ottsym{(}  \ottnt{U'}  \ottsym{)}  :  w'  \dashv   \emptyset   \rangle $
          for some $w$, $S_{{\mathrm{1}}}$, and $S'_{{\mathrm{0}}}$;
          or (2) $ \langle   \emptyset    \vdash   S  \ottsym{(}  e  \ottsym{)}  :  \ottnt{U}  \rangle   \xmapsto{ \mathmakebox[0.4em]{} S_{{\mathrm{1}}} \mathmakebox[0.3em]{} }\hspace{-0.4em}{}^\ast \hspace{0.2em}    \textsf{\textup{blame}\relax} \, \ell $ for some $\ell$ and $S_{{\mathrm{1}}}$.
        \item If $ \langle   \emptyset    \vdash   S'  \ottsym{(}  e'  \ottsym{)}  :  \ottnt{U'}  \rangle   \xmapsto{ \mathmakebox[0.4em]{} S'_{{\mathrm{1}}} \mathmakebox[0.3em]{} }\hspace{-0.4em}{}^\ast \hspace{0.2em}    \textsf{\textup{blame}\relax} \, \ell' $,
          then $ \langle   \emptyset    \vdash   S  \ottsym{(}  e  \ottsym{)}  :  \ottnt{U}  \rangle   \xmapsto{ \mathmakebox[0.4em]{} S_{{\mathrm{1}}} \mathmakebox[0.3em]{} }\hspace{-0.4em}{}^\ast \hspace{0.2em}    \textsf{\textup{blame}\relax} \, \ell $ for some $\ell$ and $S_{{\mathrm{1}}}$.
        \item If $ \langle   \emptyset    \vdash   S'  \ottsym{(}  e'  \ottsym{)}  :  \ottnt{U'}  \rangle   \Uparrow  $,
          then either (1) $ \langle   \emptyset    \vdash   S  \ottsym{(}  e  \ottsym{)}  :  \ottnt{U}  \rangle   \Uparrow  $,
          or (2) $ \langle   \emptyset    \vdash   S  \ottsym{(}  e  \ottsym{)}  :  \ottnt{U}  \rangle   \xmapsto{ \mathmakebox[0.4em]{} S_{{\mathrm{1}}} \mathmakebox[0.3em]{} }\hspace{-0.4em}{}^\ast \hspace{0.2em}    \textsf{\textup{blame}\relax} \, \ell $ for some $\ell$ and $S_{{\mathrm{1}}}$.
       \end{itemize}
 \end{enumerate}
\end{theorem}
\iffull\else
\begin{proof}
  It is shown from auxiliary theorems that (1) the cast-inserting
  translation is precision-preserving and (2) precision-related
  $\lambdaRTI$ terms behave equivalently if the precisely annotated
  term does not raise blame.
\end{proof}
\fi

\paragraph{Connection between the gradual guarantee and the soundness--completeness of DTI}
\TS{CHECK ME!}
The dynamic gradual guarantee states that precise type annotations may find more type
errors at run time but otherwise do not change the behavior of programs.
Perhaps interestingly, the soundness and completeness of DTI also
ensure a very similar property (although it is a property of the
intermediate language, not the ITGL): they state that applying
a type substitution may find more type errors at run time
but otherwise does not\AI{To Review A: This is OK.} change the behavior of programs.
Roughly speaking, the completeness corresponds to the first item of the
dynamic gradual guarantee and the soundness to the second item.
Given the fact that $e = S  \ottsym{(}  e'  \ottsym{)}$ implies $ e   \sqsubseteq _{ S }  e' $,
the premise of Theorem~\ref{thm:dynamicGG} is weaker than
those of Theorems~\ref{thm:soundness} and \ref{thm:completeness}.
However, the conclusions of the soundness and completeness are
stronger than that of the gradual guarantee:
if compared terms evaluate to values, they are \emph{syntactically
equivalent} modulo type instantiations, while in the gradual guarantee
the evaluation results may have different occurrences of the dynamic type.
Moreover, Theorem~\ref{thm:soundness} ensures that, if the RHS
evaluates to a value, then the LHS evaluates also to a value and does
not raise blame.  

\section{Related Work} \label{sec:related_work}
\paragraph{Monotonic reference.}
\citet{DBLP:conf/esop/SiekVCTG15} study so-called \emph{monotonic references}, which are
an efficient implementation of mutable references for gradually typed languages.
The traditional
approach~\cite{DBLP:conf/sfp/HermanTF07,DBLP:journals/lisp/HermanTF10} to mutable
references in gradual typing is to create a proxy which performs run-time typechecking on reads
and writes.  The proxy-based approach works well theoretically, but it incurs significant run-time overhead even for statically
typed parts in a gradually typed program.  Monotonic references preserve a global
invariant that the type of a value in the heap is at least as precise as
the types given to references that point to the value.  As a consequence, there is no overhead
for reading and writing through \emph{references of static types}, because
no type (other than itself) is more precise than a static type.
Monotonic references are similar to DTI in the sense that monotonic references
refine types of values in the heap during evaluation whereas DTI refines type
variables to more precise types at run time.

\paragraph{Staged type inference.}  \citet{DBLP:conf/popl/ShieldsSJ98} propose a combination of dynamic
typing and staged
computation~\cite{DBLP:journals/tcs/TahaS00,DaviesPfenning01JACM} in a
statically typed language.  In their language, all code values are
given a single code type $\langle \rangle$ (without information on the
static type of the code as is typical in type systems for staged
computation~\cite{DBLP:journals/tcs/TahaS00,DaviesPfenning01JACM,Davies96LICS,DBLP:conf/popl/TahaN03,DBLP:conf/popl/KimYC06,TsukadaIgarashi10LMCS})
and they are typechecked only when it is executed by \texttt{eval}.
Such a ``typecheck before eval'' strategy is called \emph{staged type
  inference}.  They also propose an \emph{incremental} staged type
inference, where part of typechecking is performed when code fragments
are composed using the quasiquotation
mechanism~\cite{Quine81,DBLP:conf/pepm/Bawden99,DBLP:books/lib/Steele90};
if an obvious inconsistency (such as applying multiplication to a string
constant) is found, execution results in a special code constant,
which indicates that type checking has failed.  This incremental inference
is close to ours in that type variables in code fragments are unified
and instantiated at run time.  However, it is limited for code
manipulation.  We integrate a similar idea to gradual typing and cast
semantics.  Another (somewhat minor) technical difference is that our
semantics does not require full first-order unification at run time:
it is always the case that one of the two types to be unified is a
type variable and there is no need for occur check.

\paragraph{Type inference for gradually typed languages.}

\citet{HengleinRehof95FPCA} studied a type reconstruction algorithm
for a language with the dynamic type,
coercions~\cite{DBLP:journals/scp/Henglein94}, and constraint-based
polymorphism.  They did discuss the issue of uninstantiated type
variables and propose to replace them with the dynamic type before
running the program.  So, their semantics can be too permissive; it
allows successful termination even when there are no static types that
can be substituted for uninstantiated type variables without
causing run-time errors.  \AI{Correct?}

\citet{DBLP:conf/dls/SiekV08} propose a type
reconstruction algorithm for a gradually typed language and \citet{DBLP:conf/popl/GarciaC15} later propose a type inference
algorithm for a very similar language\footnote{In fact, their type
  systems are shown to be equivalent in a certain sense~\cite{DBLP:conf/popl/GarciaC15}.
  The algorithm in \citet{DBLP:conf/dls/SiekV08} replaces
  uninstantiated type variables with the dynamic type, similarly
  to \citet{DBLP:journals/scp/Henglein94}.} ITGL with a principal
type property. The key idea of Garcia and Cimini is to infer only static types for type
variables, which represent omitted type annotations.  They have also
discussed how the surface language where type annotations are inferred
can be translated to a blame calculus, where types are explicit not
only in lambda abstractions but also in casts, by showing type-directed
translation into a slight variant\footnote{Type variables, which are
  bound by type abstraction, and type parameters, which correspond to
  type variables left undecided and occur free in a program, are
  distinguished in the syntax.} of the Polymorphic Blame
Calculus~\cite{DBLP:conf/popl/AhmedFSW11,DBLP:journals/pacmpl/AhmedJSW17}.
As we have discussed, this translation raises two problems: (1) The
interpretation of type variables left undecided by type inference is
not clear---if they are interpreted as fresh base types, a program may
fail earlier at run time than a type-substitution instance of it; and
(2) the semantics based on the Polymorphic Blame Calculus does not match the
intuition that $ \textsf{\textup{let}\relax} \,  \ottmv{x}  =  w  \textsf{\textup{ in }\relax}  f $ behaves the same as $f  [  \ottmv{x}  \ottsym{:=}  w  ]$,
which we believe is desirable especially in languages where type
abstractions and applications are implicit.  Actually, Garcia and
Cimini mention that the semantics of the ITGL can be defined by first
expanding $ \textsf{\textup{let}\relax} $ and translating into a simply typed blame
calculus---although it would still have the first problem---but
differences of the two translations are not discussed.

\citet{DBLP:conf/popl/RastogiCH12} present a flow-based
type inference algorithm for ActionScript.  The type system is based
on subtyping rather than polymorphism.  Their type inference algorithm assigns
a type without type variables for all type variables, so the problem
we have tackled in this paper does not arise.

More recently, \citet{DBLP:conf/esop/XieBO18} introduce the dynamic
type into the Odersky--L{\"a}ufer type
system~\cite{DBLP:conf/popl/OderskyL96} for higher-rank polymorphism.
They also develop a bidirectional algorithmic type system, which is
used to infer static monotypes for missing type declarations at
lambda abstractions and type instantiations for polymorphic
types.  As we have already discussed, they point out that the
interpretation of type variables left undecided by type inference
affects the run-time behavior of a program and suggest to substitute
the dynamic type for such undecided type variables, which is basically the
same idea as \citet{DBLP:journals/scp/Henglein94} and
\citet{DBLP:conf/dls/SiekV08}, although some refinement based on
static and gradual type parameters~\cite{DBLP:conf/popl/GarciaC15} is
discussed.

\paragraph{Polymorhic gradual typing.}
\citet{DBLP:conf/popl/AhmedFSW11,DBLP:journals/pacmpl/AhmedJSW17}
propose the Polymorphic Blame Calculus, which is based on System F~\cite{Girard72,Reynolds74}.
They introduce the notion of type bindings to enforce parametricity.
\citet{DBLP:journals/pacmpl/IgarashiSI17} also propose
a polymorphic gradually typed language and a similar polymorphic blame
calculus.  As we have already discussed above, it has an undesirable
consequence to implement polymorphic $ \textsf{\textup{let}\relax} $ via type abstraction and type
application with dynamic enforcement of parametricity.


\paragraph{Nongeneralizable type variables in OCaml.}

In the implementation \texttt{ocaml}, a read-eval-print loop for
OCaml, when the type of a given expression (bound by a top-level
\texttt{let}) contains type variables which cannot be generalized due
to (relaxed) value
restriction~\cite{DBLP:conf/flops/Garrigue04,DBLP:journals/lisp/Wright95},
the expression is accepted with those type variables being left
uninstantiated.\footnote{OCaml FAQ:
  \url{http://ocaml.org/learn/faq.html}} (In Standard ML, such an
input requires a type annotation; otherwise it will be rejected.)  The
following session with \texttt{ocaml} (taken from the OCaml FAQ) shows an
example:
\begin{trivlist}
\item\verb+# let r = ref [];;+
\item\texttt{\textsl{val r : '\symbol{`\_}weak1 list ref = \symbol{`\{}contents = []\symbol{`\}}}}
\end{trivlist}
Here, \verb+'_weak1+ stands for a type variable which was undecided by
type inference but cannot be used polymorphically.  Such type
variables will be instantiated as the declared variable is used in a
more specific context.  For example, the following input enforces that
the contents of \texttt{r} have to be an integer list:
\begin{trivlist}
\item\verb+# r := 42 :: !r;;+
\item\texttt{\textsl{- : unit = ()}}
\end{trivlist}
Now, the type of \texttt{r} changes to \texttt{int list ref} by
substituting \texttt{int} for \verb+'_weak1+:
\begin{trivlist}
\item\verb+# r;;+
\item\textsl{\texttt{- : int list ref = \symbol{`\{}contents = [42]\symbol{`\}}}}
\end{trivlist}
This behavior is similar to our DTI in the sense
that type variables left undecided by type inference are instantiated
according to their use.  However, instantiation of type variables
is caused by (compile-time) type inference, which takes place before
evaluating every input expression.

\section{Conclusion} \label{sec:conclusion}

We have developed a new blame calculus $\lambdaRTI$ with DTI, which
infers types for undecided (at compile-time) type variables along
evaluation.  We have also extended $\lambdaRTI$ to let-polymorphism
and proposed a nonparametric semantics, which we argue is better than
the PBC-based
approach~\cite{DBLP:conf/popl/AhmedFSW11,DBLP:journals/pacmpl/AhmedJSW17}
for languages with implicit type abstraction and application.  Our
calculus can be used as an intermediate language to give semantics to
the ITGL by \citet{DBLP:conf/popl/GarciaC15}.  We have shown the type
safety of $\lambdaRTI$ and the soundness and completeness of DTI.  We
have also shown the type safety and the gradual guarantee in the ITGL.  To
our knowledge, the gradual guarantee for a gradually typed language
with let-polymorphism is shown for the first time.  Although we leave
a more formal investigation for future work, we have pointed out a
relationship between the gradual guarantee and the
soundness--completeness property of DTI.


We have implemented a prototype evaluator of DTI, which is used in an
interpreter of the ITGL but leave the study of efficiency issues on
DTI to future work.  We are interested in the impact of run-time
overhead incurred by DTI and an efficient implementation to address
it.
Note that not all type variables are subject to dynamic type inference
and some of the run-time type information can be erased
by introducing the distinction between static and gradual type
variables~\cite{DBLP:conf/popl/GarciaC15,DBLP:journals/pacmpl/IgarashiSI17}.
In particular, if a program does not use the dynamic type
$ \star $  at all, we expect that a program can be run without passing type
information or instantiating type variables dynamically.
%
%
Another direction for efficiency improvement is to integrate DTI with\AI{To Reviewer A: this is OK.} space-efficient cast
calculi~\cite{DBLP:conf/sfp/HermanTF07,DBLP:journals/lisp/HermanTF10,DBLP:conf/popl/SiekW10}.
The space-efficient calculi make efficient use of the memory space for casts
generated at run time by ``merging'' two casts into one dynamically.
It is interesting to investigate how to merge casts referring to type variables
effectively.
%

We expect DTI can be extended to other typing features such as
subtyping.  An obvious candidate to apply DTI would be the gradual
extension of higher-rank polymorphism~\cite{DBLP:conf/esop/XieBO18},
which we also leave for future work.

\begin{acks}                            
  
  We would like to thank anonymous reviewers from both PC and AEC for valuable comments and Yuu Igarashi
  for the fruitful discussions.
  This work was supported in part by the JSPS KAKENHI Grant Number
  JP17H01723 (Igarashi)
  and ERATO HASUO Metamathematics for Systems Design Project
  (No.~\grantnum{http://dx.doi.org/10.13039/501100009024}{JPMJER1603}), JST
  (Sekiyama).

\end{acks}

  \bibliography{local}


\begin{thebibliography}{53}


\ifx \showCODEN    \undefined \def \showCODEN     #1{\unskip}     \fi
\ifx \showDOI      \undefined \def \showDOI       #1{#1}\fi
\ifx \showISBNx    \undefined \def \showISBNx     #1{\unskip}     \fi
\ifx \showISBNxiii \undefined \def \showISBNxiii  #1{\unskip}     \fi
\ifx \showISSN     \undefined \def \showISSN      #1{\unskip}     \fi
\ifx \showLCCN     \undefined \def \showLCCN      #1{\unskip}     \fi
\ifx \shownote     \undefined \def \shownote      #1{#1}          \fi
\ifx \showarticletitle \undefined \def \showarticletitle #1{#1}   \fi
\ifx \showURL      \undefined \def \showURL       {\relax}        \fi
\providecommand\bibfield[2]{#2}
\providecommand\bibinfo[2]{#2}
\providecommand\natexlab[1]{#1}
\providecommand\showeprint[2][]{arXiv:#2}

\bibitem[\protect\citeauthoryear{Abadi, Cardelli, Pierce, and Plotkin}{Abadi
  et~al\mbox{.}}{1991}]%
        {DBLP:journals/toplas/AbadiCPP91}
\bibfield{author}{\bibinfo{person}{Mart{\'{\i}}n Abadi}, \bibinfo{person}{Luca
  Cardelli}, \bibinfo{person}{Benjamin~C. Pierce}, {and}
  \bibinfo{person}{Gordon~D. Plotkin}.} \bibinfo{year}{1991}\natexlab{}.
\newblock \showarticletitle{Dynamic Typing in a Statically Typed Language}.
\newblock \bibinfo{journal}{\emph{ACM Transactions on Programming Languages and
  Systems}} \bibinfo{volume}{13}, \bibinfo{number}{2} (\bibinfo{year}{1991}),
  \bibinfo{pages}{237--268}.
\newblock
\urldef\tempurl%
\url{https://doi.org/10.1145/103135.103138}
\showDOI{\tempurl}


\bibitem[\protect\citeauthoryear{Ahmed, Findler, Siek, and Wadler}{Ahmed
  et~al\mbox{.}}{2011}]%
        {DBLP:conf/popl/AhmedFSW11}
\bibfield{author}{\bibinfo{person}{Amal Ahmed}, \bibinfo{person}{Robert~Bruce
  Findler}, \bibinfo{person}{Jeremy~G. Siek}, {and} \bibinfo{person}{Philip
  Wadler}.} \bibinfo{year}{2011}\natexlab{}.
\newblock \showarticletitle{Blame for all}. In \bibinfo{booktitle}{\emph{Proc.
  of {ACM} {POPL}}}. \bibinfo{pages}{201--214}.
\newblock
\urldef\tempurl%
\url{https://doi.org/10.1145/1926385.1926409}
\showDOI{\tempurl}


\bibitem[\protect\citeauthoryear{Ahmed, Jamner, Siek, and Wadler}{Ahmed
  et~al\mbox{.}}{2017}]%
        {DBLP:journals/pacmpl/AhmedJSW17}
\bibfield{author}{\bibinfo{person}{Amal Ahmed}, \bibinfo{person}{Dustin
  Jamner}, \bibinfo{person}{Jeremy~G. Siek}, {and} \bibinfo{person}{Philip
  Wadler}.} \bibinfo{year}{2017}\natexlab{}.
\newblock \showarticletitle{Theorems for free for free: parametricity, with and
  without types}.
\newblock \bibinfo{journal}{\emph{{PACMPL}}} \bibinfo{volume}{1},
  \bibinfo{number}{{ICFP}} (\bibinfo{year}{2017}),
  \bibinfo{pages}{39:1--39:28}.
\newblock
\urldef\tempurl%
\url{https://doi.org/10.1145/3110283}
\showDOI{\tempurl}


\bibitem[\protect\citeauthoryear{{Ba{\~{n}}ados Schwerter}, Garcia, and
  Tanter}{{Ba{\~{n}}ados Schwerter} et~al\mbox{.}}{2014}]%
        {DBLP:conf/icfp/SchwerterGT14}
\bibfield{author}{\bibinfo{person}{Felipe {Ba{\~{n}}ados Schwerter}},
  \bibinfo{person}{Ronald Garcia}, {and} \bibinfo{person}{{\'{E}}ric Tanter}.}
  \bibinfo{year}{2014}\natexlab{}.
\newblock \showarticletitle{A theory of gradual effect systems}. In
  \bibinfo{booktitle}{\emph{Proc. of {ACM} {ICFP}}}. \bibinfo{pages}{283--295}.
\newblock
\urldef\tempurl%
\url{https://doi.org/10.1145/2628136.2628149}
\showDOI{\tempurl}


\bibitem[\protect\citeauthoryear{{Ba{\~{n}}ados Schwerter}, Garcia, and
  Tanter}{{Ba{\~{n}}ados Schwerter} et~al\mbox{.}}{2016}]%
        {DBLP:journals/jfp/SchwerterGT16}
\bibfield{author}{\bibinfo{person}{Felipe {Ba{\~{n}}ados Schwerter}},
  \bibinfo{person}{Ronald Garcia}, {and} \bibinfo{person}{{\'{E}}ric Tanter}.}
  \bibinfo{year}{2016}\natexlab{}.
\newblock \showarticletitle{Gradual type-and-effect systems}.
\newblock \bibinfo{journal}{\emph{Journal of Functional Programming}}
  \bibinfo{volume}{26} (\bibinfo{year}{2016}), \bibinfo{pages}{e19}.
\newblock
\urldef\tempurl%
\url{https://doi.org/10.1017/S0956796816000162}
\showDOI{\tempurl}


\bibitem[\protect\citeauthoryear{Bawden}{Bawden}{1999}]%
        {DBLP:conf/pepm/Bawden99}
\bibfield{author}{\bibinfo{person}{Alan Bawden}.}
  \bibinfo{year}{1999}\natexlab{}.
\newblock \showarticletitle{Quasiquotation in {Lisp}}. In
  \bibinfo{booktitle}{\emph{Proc. of {ACM} {PEPM}}}. \bibinfo{pages}{4--12}.
\newblock


\bibitem[\protect\citeauthoryear{Bj{\o}rner}{Bj{\o}rner}{1994}]%
        {Bjorner94ML}
\bibfield{author}{\bibinfo{person}{Nikolaj~Skallerud Bj{\o}rner}.}
  \bibinfo{year}{1994}\natexlab{}.
\newblock \showarticletitle{Minimal Typing Derivations}. In
  \bibinfo{booktitle}{\emph{Proceedings of the {ACM SIGPLAN} Workshop on {ML}
  and its Applications}}. \bibinfo{pages}{120--126}.
\newblock


\bibitem[\protect\citeauthoryear{Bracha and Griswold}{Bracha and
  Griswold}{1993}]%
        {DBLP:conf/oopsla/BrachaG93}
\bibfield{author}{\bibinfo{person}{Gilad Bracha} {and} \bibinfo{person}{David
  Griswold}.} \bibinfo{year}{1993}\natexlab{}.
\newblock \showarticletitle{Strongtalk: Typechecking Smalltalk in a Production
  Environment}. In \bibinfo{booktitle}{\emph{Proc. of {ACM} {OOPSLA}}}.
  \bibinfo{pages}{215--230}.
\newblock
\urldef\tempurl%
\url{https://doi.org/10.1145/165854.165893}
\showDOI{\tempurl}


\bibitem[\protect\citeauthoryear{Breazu-Tannen, Coquand, Gunter, and
  Scedrov}{Breazu-Tannen et~al\mbox{.}}{1991}]%
        {DBLP:journals/iandc/Breazu-TannenCGS91}
\bibfield{author}{\bibinfo{person}{Val Breazu-Tannen}, \bibinfo{person}{Thierry
  Coquand}, \bibinfo{person}{Carl~A. Gunter}, {and} \bibinfo{person}{Andre
  Scedrov}.} \bibinfo{year}{1991}\natexlab{}.
\newblock \showarticletitle{Inheritance as Implicit Coercion}.
\newblock \bibinfo{journal}{\emph{Inf. Comput.}} \bibinfo{volume}{93},
  \bibinfo{number}{1} (\bibinfo{year}{1991}), \bibinfo{pages}{172--221}.
\newblock
\urldef\tempurl%
\url{https://doi.org/10.1016/0890-5401(91)90055-7}
\showDOI{\tempurl}
\newblock
\shownote{Also in Carl A.\ Gunter and John C.\ Mitchell, editors,
  \emph{Theoretical Aspects of Object-Oriented Programming: Types, Semantics,
  and Language Design} (MIT Press, 1994.}


\bibitem[\protect\citeauthoryear{Cartwright and Fagan}{Cartwright and
  Fagan}{1991}]%
        {DBLP:conf/pldi/CartwrightF91}
\bibfield{author}{\bibinfo{person}{Robert Cartwright} {and}
  \bibinfo{person}{Mike Fagan}.} \bibinfo{year}{1991}\natexlab{}.
\newblock \showarticletitle{Soft Typing}. In \bibinfo{booktitle}{\emph{Proc. of
  {ACM} {PLDI}}}. \bibinfo{pages}{278--292}.
\newblock
\urldef\tempurl%
\url{https://doi.org/10.1145/113445.113469}
\showDOI{\tempurl}


\bibitem[\protect\citeauthoryear{Cimini and Siek}{Cimini and Siek}{2016}]%
        {DBLP:conf/popl/CiminiS16}
\bibfield{author}{\bibinfo{person}{Matteo Cimini} {and}
  \bibinfo{person}{Jeremy~G. Siek}.} \bibinfo{year}{2016}\natexlab{}.
\newblock \showarticletitle{The {Gradualizer}: A Methodology and Algorithm for
  Generating Gradual Type Systems}. In \bibinfo{booktitle}{\emph{Proc. of {ACM}
  {POPL}}}. \bibinfo{pages}{443--455}.
\newblock
\urldef\tempurl%
\url{https://doi.org/10.1145/2837614.2837632}
\showDOI{\tempurl}


\bibitem[\protect\citeauthoryear{Cimini and Siek}{Cimini and Siek}{2017}]%
        {DBLP:conf/popl/CiminiS17}
\bibfield{author}{\bibinfo{person}{Matteo Cimini} {and}
  \bibinfo{person}{Jeremy~G. Siek}.} \bibinfo{year}{2017}\natexlab{}.
\newblock \showarticletitle{Automatically generating the dynamic semantics of
  gradually typed languages}. In \bibinfo{booktitle}{\emph{Proc. of {ACM}
  {POPL}}}. \bibinfo{pages}{789--803}.
\newblock
\urldef\tempurl%
\url{http://dl.acm.org/citation.cfm?id=3009863}
\showURL{%
\tempurl}


\bibitem[\protect\citeauthoryear{Davies}{Davies}{1996}]%
        {Davies96LICS}
\bibfield{author}{\bibinfo{person}{Rowan Davies}.}
  \bibinfo{year}{1996}\natexlab{}.
\newblock \showarticletitle{A Temporal-Logic Approach to Binding-Time
  Analysis}. In \bibinfo{booktitle}{\emph{Proc. of {IEEE} LICS}}.
  \bibinfo{pages}{184--195}.
\newblock


\bibitem[\protect\citeauthoryear{Davies and Pfenning}{Davies and
  Pfenning}{2001}]%
        {DaviesPfenning01JACM}
\bibfield{author}{\bibinfo{person}{Rowan Davies} {and} \bibinfo{person}{Frank
  Pfenning}.} \bibinfo{year}{2001}\natexlab{}.
\newblock \showarticletitle{A Modal Analysis of Staged Computation}.
\newblock \bibinfo{journal}{\emph{J. ACM}} \bibinfo{volume}{48},
  \bibinfo{number}{3} (\bibinfo{year}{2001}), \bibinfo{pages}{555--604}.
\newblock


\bibitem[\protect\citeauthoryear{Findler and Felleisen}{Findler and
  Felleisen}{2002}]%
        {DBLP:conf/icfp/FindlerF02}
\bibfield{author}{\bibinfo{person}{Robert~Bruce Findler} {and}
  \bibinfo{person}{Matthias Felleisen}.} \bibinfo{year}{2002}\natexlab{}.
\newblock \showarticletitle{Contracts for higher-order functions}. In
  \bibinfo{booktitle}{\emph{Proc. of {ACM} {ICFP}}}. \bibinfo{pages}{48--59}.
\newblock
\urldef\tempurl%
\url{https://doi.org/10.1145/581478.581484}
\showDOI{\tempurl}


\bibitem[\protect\citeauthoryear{Flanagan and Felleisen}{Flanagan and
  Felleisen}{1999}]%
        {DBLP:journals/toplas/FlanaganF99}
\bibfield{author}{\bibinfo{person}{Cormac Flanagan} {and}
  \bibinfo{person}{Matthias Felleisen}.} \bibinfo{year}{1999}\natexlab{}.
\newblock \showarticletitle{Componential Set-Based Analysis}.
\newblock \bibinfo{journal}{\emph{ACM Transactions on Programming Languages and
  Systems}} \bibinfo{volume}{21}, \bibinfo{number}{2} (\bibinfo{year}{1999}),
  \bibinfo{pages}{370--416}.
\newblock
\urldef\tempurl%
\url{https://doi.org/10.1145/316686.316703}
\showDOI{\tempurl}


\bibitem[\protect\citeauthoryear{Garcia and Cimini}{Garcia and Cimini}{2015}]%
        {DBLP:conf/popl/GarciaC15}
\bibfield{author}{\bibinfo{person}{Ronald Garcia} {and} \bibinfo{person}{Matteo
  Cimini}.} \bibinfo{year}{2015}\natexlab{}.
\newblock \showarticletitle{Principal Type Schemes for Gradual Programs}. In
  \bibinfo{booktitle}{\emph{Proc. of {ACM} {POPL}}}. \bibinfo{pages}{303--315}.
\newblock
\urldef\tempurl%
\url{https://doi.org/10.1145/2676726.2676992}
\showDOI{\tempurl}


\bibitem[\protect\citeauthoryear{Garcia, Clark, and Tanter}{Garcia
  et~al\mbox{.}}{2016}]%
        {DBLP:conf/popl/GarciaCT16}
\bibfield{author}{\bibinfo{person}{Ronald Garcia}, \bibinfo{person}{Alison~M.
  Clark}, {and} \bibinfo{person}{{\'{E}}ric Tanter}.}
  \bibinfo{year}{2016}\natexlab{}.
\newblock \showarticletitle{Abstracting gradual typing}. In
  \bibinfo{booktitle}{\emph{Proc. of {ACM} {POPL}}}. \bibinfo{pages}{429--442}.
\newblock
\urldef\tempurl%
\url{https://doi.org/10.1145/2837614.2837670}
\showDOI{\tempurl}


\bibitem[\protect\citeauthoryear{Garrigue}{Garrigue}{2004}]%
        {DBLP:conf/flops/Garrigue04}
\bibfield{author}{\bibinfo{person}{Jacques Garrigue}.}
  \bibinfo{year}{2004}\natexlab{}.
\newblock \showarticletitle{Relaxing the Value Restriction}. In
  \bibinfo{booktitle}{\emph{Proc. of {FLOPS}}} \emph{(\bibinfo{series}{LNCS})},
  Vol.~\bibinfo{volume}{2998}. \bibinfo{publisher}{Springer},
  \bibinfo{pages}{196--213}.
\newblock
\urldef\tempurl%
\url{https://doi.org/10.1007/978-3-540-24754-8_15}
\showDOI{\tempurl}


\bibitem[\protect\citeauthoryear{Girard}{Girard}{1972}]%
        {Girard72}
\bibfield{author}{\bibinfo{person}{Jean-Yves Girard}.}
  \bibinfo{year}{1972}\natexlab{}.
\newblock \emph{\bibinfo{title}{Interpr\'etation fonctionnelle et \'elimination
  des coupures de l'arithm\'etique d'ordre sup\'erieur}}.
\newblock Th\`ese d'\'etat. \bibinfo{school}{Universit\'e Paris VII}.
\newblock
\newblock
\shownote{Summary in \emph{Proc. of the Second Scandinavian Logic Symposium},
  1971 (63--92).}


\bibitem[\protect\citeauthoryear{Harper and Mitchell}{Harper and
  Mitchell}{1993}]%
        {HarperMitchell93TOPLAS}
\bibfield{author}{\bibinfo{person}{Robert Harper} {and}
  \bibinfo{person}{John~C. Mitchell}.} \bibinfo{year}{1993}\natexlab{}.
\newblock \showarticletitle{On The Type structure of Standard {ML}}.
\newblock \bibinfo{journal}{\emph{ACM Transactions on Programming Languages and
  Systems}} \bibinfo{volume}{15}, \bibinfo{number}{2} (\bibinfo{year}{1993}),
  \bibinfo{pages}{211--252}.
\newblock


\bibitem[\protect\citeauthoryear{Henglein}{Henglein}{1994}]%
        {DBLP:journals/scp/Henglein94}
\bibfield{author}{\bibinfo{person}{Fritz Henglein}.}
  \bibinfo{year}{1994}\natexlab{}.
\newblock \showarticletitle{Dynamic Typing: Syntax and Proof Theory}.
\newblock \bibinfo{journal}{\emph{Sci. Comput. Program.}} \bibinfo{volume}{22},
  \bibinfo{number}{3} (\bibinfo{year}{1994}), \bibinfo{pages}{197--230}.
\newblock
\urldef\tempurl%
\url{https://doi.org/10.1016/0167-6423(94)00004-2}
\showDOI{\tempurl}


\bibitem[\protect\citeauthoryear{Henglein and Rehof}{Henglein and
  Rehof}{1995}]%
        {HengleinRehof95FPCA}
\bibfield{author}{\bibinfo{person}{Fritz Henglein} {and} \bibinfo{person}{Jakob
  Rehof}.} \bibinfo{year}{1995}\natexlab{}.
\newblock \showarticletitle{Safe polymorphic type inference for a Dynamically
  Typed Language: Translating {Scheme} to {ML}}. In
  \bibinfo{booktitle}{\emph{ACM Conference on Functional Programming and
  Computer Architecture}}. \bibinfo{pages}{192--203}.
\newblock


\bibitem[\protect\citeauthoryear{Herman, Tomb, and Flanagan}{Herman
  et~al\mbox{.}}{2007}]%
        {DBLP:conf/sfp/HermanTF07}
\bibfield{author}{\bibinfo{person}{David Herman}, \bibinfo{person}{Aaron Tomb},
  {and} \bibinfo{person}{Cormac Flanagan}.} \bibinfo{year}{2007}\natexlab{}.
\newblock \showarticletitle{Space-Efficient Gradual Typing}. In
  \bibinfo{booktitle}{\emph{Proc. of {TFP}}} \emph{(\bibinfo{series}{Trends in
  Functional Programming})}, Vol.~\bibinfo{volume}{8}.
  \bibinfo{publisher}{Intellect}, \bibinfo{pages}{1--18}.
\newblock


\bibitem[\protect\citeauthoryear{Herman, Tomb, and Flanagan}{Herman
  et~al\mbox{.}}{2010}]%
        {DBLP:journals/lisp/HermanTF10}
\bibfield{author}{\bibinfo{person}{David Herman}, \bibinfo{person}{Aaron Tomb},
  {and} \bibinfo{person}{Cormac Flanagan}.} \bibinfo{year}{2010}\natexlab{}.
\newblock \showarticletitle{Space-efficient gradual typing}.
\newblock \bibinfo{journal}{\emph{Higher-Order and Symbolic Computation}}
  \bibinfo{volume}{23}, \bibinfo{number}{2} (\bibinfo{year}{2010}),
  \bibinfo{pages}{167--189}.
\newblock
\urldef\tempurl%
\url{https://doi.org/10.1007/s10990-011-9066-z}
\showDOI{\tempurl}


\bibitem[\protect\citeauthoryear{Igarashi, Sekiyama, and Igarashi}{Igarashi
  et~al\mbox{.}}{2017}]%
        {DBLP:journals/pacmpl/IgarashiSI17}
\bibfield{author}{\bibinfo{person}{Yuu Igarashi}, \bibinfo{person}{Taro
  Sekiyama}, {and} \bibinfo{person}{Atsushi Igarashi}.}
  \bibinfo{year}{2017}\natexlab{}.
\newblock \showarticletitle{On polymorphic gradual typing}.
\newblock \bibinfo{journal}{\emph{{PACMPL}}} \bibinfo{volume}{1},
  \bibinfo{number}{{ICFP}} (\bibinfo{year}{2017}),
  \bibinfo{pages}{40:1--40:29}.
\newblock
\urldef\tempurl%
\url{https://doi.org/10.1145/3110284}
\showDOI{\tempurl}


\bibitem[\protect\citeauthoryear{Ina and Igarashi}{Ina and Igarashi}{2011}]%
        {DBLP:conf/oopsla/InaI11}
\bibfield{author}{\bibinfo{person}{Lintaro Ina} {and} \bibinfo{person}{Atsushi
  Igarashi}.} \bibinfo{year}{2011}\natexlab{}.
\newblock \showarticletitle{Gradual typing for generics}. In
  \bibinfo{booktitle}{\emph{Proc. of {ACM} {OOPSLA}}}.
  \bibinfo{pages}{609--624}.
\newblock
\urldef\tempurl%
\url{https://doi.org/10.1145/2048066.2048114}
\showDOI{\tempurl}


\bibitem[\protect\citeauthoryear{Kim, Yi, and Calcagno}{Kim
  et~al\mbox{.}}{2006}]%
        {DBLP:conf/popl/KimYC06}
\bibfield{author}{\bibinfo{person}{Ik{-}Soon Kim}, \bibinfo{person}{Kwangkeun
  Yi}, {and} \bibinfo{person}{Cristiano Calcagno}.}
  \bibinfo{year}{2006}\natexlab{}.
\newblock \showarticletitle{A polymorphic modal type system for {Lisp}-like
  multi-staged languages}. In \bibinfo{booktitle}{\emph{Proc. of {ACM}
  {POPL}}}. \bibinfo{pages}{257--268}.
\newblock
\urldef\tempurl%
\url{https://doi.org/10.1145/1111037.1111060}
\showDOI{\tempurl}


\bibitem[\protect\citeauthoryear{Leroy}{Leroy}{1993}]%
        {DBLP:conf/popl/Leroy93}
\bibfield{author}{\bibinfo{person}{Xavier Leroy}.}
  \bibinfo{year}{1993}\natexlab{}.
\newblock \showarticletitle{Polymorphism by Name for References and
  Continuations}. In \bibinfo{booktitle}{\emph{Conference Record of the
  Twentieth Annual {ACM} {SIGPLAN-SIGACT} Symposium on Principles of
  Programming Languages, Charleston, South Carolina, USA, January 1993}}.
  \bibinfo{pages}{220--231}.
\newblock
\urldef\tempurl%
\url{https://doi.org/10.1145/158511.158632}
\showDOI{\tempurl}


\bibitem[\protect\citeauthoryear{Milner}{Milner}{1978}]%
        {Milner78JCSS}
\bibfield{author}{\bibinfo{person}{Robin Milner}.}
  \bibinfo{year}{1978}\natexlab{}.
\newblock \showarticletitle{A Theory of Type Polymorphism in Programming}.
\newblock \bibinfo{journal}{\emph{J. Comput. System Sci.}}
  \bibinfo{volume}{17} (\bibinfo{year}{1978}), \bibinfo{pages}{348--375}.
\newblock


\bibitem[\protect\citeauthoryear{Odersky and L{\"{a}}ufer}{Odersky and
  L{\"{a}}ufer}{1996}]%
        {DBLP:conf/popl/OderskyL96}
\bibfield{author}{\bibinfo{person}{Martin Odersky} {and}
  \bibinfo{person}{Konstantin L{\"{a}}ufer}.} \bibinfo{year}{1996}\natexlab{}.
\newblock \showarticletitle{Putting Type Annotations to Work}. In
  \bibinfo{booktitle}{\emph{Conference Record of POPL'96: The 23rd {ACM}
  {SIGPLAN-SIGACT} Symposium on Principles of Programming Languages, Papers
  Presented at the Symposium, St. Petersburg Beach, Florida, USA, January
  21-24, 1996}}. \bibinfo{pages}{54--67}.
\newblock
\urldef\tempurl%
\url{https://doi.org/10.1145/237721.237729}
\showDOI{\tempurl}


\bibitem[\protect\citeauthoryear{Quine}{Quine}{1981}]%
        {Quine81}
\bibfield{author}{\bibinfo{person}{Willard Van~Orman Quine}.}
  \bibinfo{year}{1981}\natexlab{}.
\newblock \bibinfo{booktitle}{\emph{Mathematical Logic}
  (\bibinfo{edition}{revised edition} ed.)}.
\newblock \bibinfo{publisher}{Harvard University Press}.
\newblock


\bibitem[\protect\citeauthoryear{Rastogi, Chaudhuri, and Hosmer}{Rastogi
  et~al\mbox{.}}{2012}]%
        {DBLP:conf/popl/RastogiCH12}
\bibfield{author}{\bibinfo{person}{Aseem Rastogi}, \bibinfo{person}{Avik
  Chaudhuri}, {and} \bibinfo{person}{Basil Hosmer}.}
  \bibinfo{year}{2012}\natexlab{}.
\newblock \showarticletitle{The ins and outs of gradual type inference}. In
  \bibinfo{booktitle}{\emph{Proc. of {ACM} {POPL}}}. \bibinfo{pages}{481--494}.
\newblock
\urldef\tempurl%
\url{https://doi.org/10.1145/2103656.2103714}
\showDOI{\tempurl}


\bibitem[\protect\citeauthoryear{Reynolds}{Reynolds}{1974}]%
        {Reynolds74}
\bibfield{author}{\bibinfo{person}{John Reynolds}.}
  \bibinfo{year}{1974}\natexlab{}.
\newblock \showarticletitle{Towards a Theory of Type Structure}. In
  \bibinfo{booktitle}{\emph{Proc.\ Colloque sur la Programmation}}
  \emph{(\bibinfo{series}{LNCS})}, Vol.~\bibinfo{volume}{19}.
  \bibinfo{publisher}{Springer}, \bibinfo{pages}{408--425}.
\newblock


\bibitem[\protect\citeauthoryear{Reynolds}{Reynolds}{1983}]%
        {DBLP:conf/ifip/Reynolds83}
\bibfield{author}{\bibinfo{person}{John~C. Reynolds}.}
  \bibinfo{year}{1983}\natexlab{}.
\newblock \showarticletitle{Types, Abstraction and Parametric Polymorphism}. In
  \bibinfo{booktitle}{\emph{{IFIP} Congress}}. \bibinfo{pages}{513--523}.
\newblock


\bibitem[\protect\citeauthoryear{Sergey and Clarke}{Sergey and Clarke}{2012}]%
        {DBLP:conf/esop/SergeyC12}
\bibfield{author}{\bibinfo{person}{Ilya Sergey} {and} \bibinfo{person}{Dave
  Clarke}.} \bibinfo{year}{2012}\natexlab{}.
\newblock \showarticletitle{Gradual Ownership Types}. In
  \bibinfo{booktitle}{\emph{Proc. of {ESOP}}} \emph{(\bibinfo{series}{LNCS})},
  Vol.~\bibinfo{volume}{7211}. \bibinfo{publisher}{Springer},
  \bibinfo{pages}{579--599}.
\newblock
\urldef\tempurl%
\url{https://doi.org/10.1007/978-3-642-28869-2_29}
\showDOI{\tempurl}


\bibitem[\protect\citeauthoryear{Shields, Sheard, and {Peyton Jones}}{Shields
  et~al\mbox{.}}{1998}]%
        {DBLP:conf/popl/ShieldsSJ98}
\bibfield{author}{\bibinfo{person}{Mark Shields}, \bibinfo{person}{Tim Sheard},
  {and} \bibinfo{person}{Simon~L. {Peyton Jones}}.}
  \bibinfo{year}{1998}\natexlab{}.
\newblock \showarticletitle{Dynamic Typing as Staged Type Inference}. In
  \bibinfo{booktitle}{\emph{Proc. of {ACM} {POPL}}}. \bibinfo{pages}{289--302}.
\newblock
\urldef\tempurl%
\url{https://doi.org/10.1145/268946.268970}
\showDOI{\tempurl}


\bibitem[\protect\citeauthoryear{Siek and Taha}{Siek and Taha}{2006}]%
        {conf/Scheme/SiekT06}
\bibfield{author}{\bibinfo{person}{Jeremy~G. Siek} {and} \bibinfo{person}{Walid
  Taha}.} \bibinfo{year}{2006}\natexlab{}.
\newblock \showarticletitle{Gradual typing for functional languages}. In
  \bibinfo{booktitle}{\emph{Proc. of Workshop on Scheme and Functional
  Programming}}. \bibinfo{pages}{81--92}.
\newblock


\bibitem[\protect\citeauthoryear{Siek and Taha}{Siek and Taha}{2007}]%
        {DBLP:conf/ecoop/SiekT07}
\bibfield{author}{\bibinfo{person}{Jeremy~G. Siek} {and} \bibinfo{person}{Walid
  Taha}.} \bibinfo{year}{2007}\natexlab{}.
\newblock \showarticletitle{Gradual Typing for Objects}. In
  \bibinfo{booktitle}{\emph{Proc. of {ECOOP}}} \emph{(\bibinfo{series}{LNCS})},
  Vol.~\bibinfo{volume}{4609}. \bibinfo{publisher}{Springer},
  \bibinfo{pages}{2--27}.
\newblock
\urldef\tempurl%
\url{https://doi.org/10.1007/978-3-540-73589-2_2}
\showDOI{\tempurl}


\bibitem[\protect\citeauthoryear{Siek and Vachharajani}{Siek and
  Vachharajani}{2008}]%
        {DBLP:conf/dls/SiekV08}
\bibfield{author}{\bibinfo{person}{Jeremy~G. Siek} {and}
  \bibinfo{person}{Manish Vachharajani}.} \bibinfo{year}{2008}\natexlab{}.
\newblock \showarticletitle{Gradual typing with unification-based inference}.
  In \bibinfo{booktitle}{\emph{Proc. of {DLS}}}. \bibinfo{pages}{7}.
\newblock
\urldef\tempurl%
\url{https://doi.org/10.1145/1408681.1408688}
\showDOI{\tempurl}


\bibitem[\protect\citeauthoryear{Siek, Vitousek, Cimini, and Boyland}{Siek
  et~al\mbox{.}}{2015a}]%
        {DBLP:conf/snapl/SiekVCB15}
\bibfield{author}{\bibinfo{person}{Jeremy~G. Siek}, \bibinfo{person}{Michael~M.
  Vitousek}, \bibinfo{person}{Matteo Cimini}, {and} \bibinfo{person}{John~Tang
  Boyland}.} \bibinfo{year}{2015}\natexlab{a}.
\newblock \showarticletitle{Refined Criteria for Gradual Typing}. In
  \bibinfo{booktitle}{\emph{1st Summit on Advances in Programming Languages,
  {SNAPL}}} \emph{(\bibinfo{series}{LIPIcs})}, Vol.~\bibinfo{volume}{32}.
  \bibinfo{publisher}{Schloss Dagstuhl - Leibniz-Zentrum f{\"u}r Informatik},
  \bibinfo{pages}{274--293}.
\newblock
\urldef\tempurl%
\url{https://doi.org/10.4230/LIPIcs.SNAPL.2015.274}
\showDOI{\tempurl}


\bibitem[\protect\citeauthoryear{Siek, Vitousek, Cimini, Tobin{-}Hochstadt, and
  Garcia}{Siek et~al\mbox{.}}{2015b}]%
        {DBLP:conf/esop/SiekVCTG15}
\bibfield{author}{\bibinfo{person}{Jeremy~G. Siek}, \bibinfo{person}{Michael~M.
  Vitousek}, \bibinfo{person}{Matteo Cimini}, \bibinfo{person}{Sam
  Tobin{-}Hochstadt}, {and} \bibinfo{person}{Ronald Garcia}.}
  \bibinfo{year}{2015}\natexlab{b}.
\newblock \showarticletitle{Monotonic References for Efficient Gradual Typing}.
  In \bibinfo{booktitle}{\emph{Proc. of {ESOP}}}
  \emph{(\bibinfo{series}{LNCS})}, Vol.~\bibinfo{volume}{9032}.
  \bibinfo{publisher}{Springer}, \bibinfo{pages}{432--456}.
\newblock
\urldef\tempurl%
\url{https://doi.org/10.1007/978-3-662-46669-8_18}
\showDOI{\tempurl}


\bibitem[\protect\citeauthoryear{Siek and Wadler}{Siek and Wadler}{2010}]%
        {DBLP:conf/popl/SiekW10}
\bibfield{author}{\bibinfo{person}{Jeremy~G. Siek} {and}
  \bibinfo{person}{Philip Wadler}.} \bibinfo{year}{2010}\natexlab{}.
\newblock \showarticletitle{Threesomes, with and without blame}. In
  \bibinfo{booktitle}{\emph{Proc. of {ACM} {POPL}}}. \bibinfo{pages}{365--376}.
\newblock
\urldef\tempurl%
\url{https://doi.org/10.1145/1706299.1706342}
\showDOI{\tempurl}


\bibitem[\protect\citeauthoryear{Steele}{Steele}{1990}]%
        {DBLP:books/lib/Steele90}
\bibfield{author}{\bibinfo{person}{Guy~L. Steele, Jr.}}
  \bibinfo{year}{1990}\natexlab{}.
\newblock \bibinfo{booktitle}{\emph{Common {LISP:} the language, 2nd Edition}}.
\newblock \bibinfo{publisher}{Digital Pr.}
\newblock
\showISBNx{0131556649}
\urldef\tempurl%
\url{http://www.worldcat.org/oclc/20631879}
\showURL{%
\tempurl}


\bibitem[\protect\citeauthoryear{Taha and Nielsen}{Taha and Nielsen}{2003}]%
        {DBLP:conf/popl/TahaN03}
\bibfield{author}{\bibinfo{person}{Walid Taha} {and}
  \bibinfo{person}{Michael~Florentin Nielsen}.}
  \bibinfo{year}{2003}\natexlab{}.
\newblock \showarticletitle{Environment classifiers}. In
  \bibinfo{booktitle}{\emph{Proc. of {ACM} {POPL}}}. \bibinfo{pages}{26--37}.
\newblock
\urldef\tempurl%
\url{https://doi.org/10.1145/640128.604134}
\showDOI{\tempurl}


\bibitem[\protect\citeauthoryear{Taha and Sheard}{Taha and Sheard}{2000}]%
        {DBLP:journals/tcs/TahaS00}
\bibfield{author}{\bibinfo{person}{Walid Taha} {and} \bibinfo{person}{Tim
  Sheard}.} \bibinfo{year}{2000}\natexlab{}.
\newblock \showarticletitle{{MetaML} and multi-stage programming with explicit
  annotations}.
\newblock \bibinfo{journal}{\emph{Theoretical Computer Science}}
  \bibinfo{volume}{248}, \bibinfo{number}{1-2} (\bibinfo{year}{2000}),
  \bibinfo{pages}{211--242}.
\newblock
\urldef\tempurl%
\url{https://doi.org/10.1016/S0304-3975(00)00053-0}
\showDOI{\tempurl}


\bibitem[\protect\citeauthoryear{Thatte}{Thatte}{1990}]%
        {DBLP:conf/popl/Thatte90}
\bibfield{author}{\bibinfo{person}{Satish~R. Thatte}.}
  \bibinfo{year}{1990}\natexlab{}.
\newblock \showarticletitle{Quasi-Static Typing}. In
  \bibinfo{booktitle}{\emph{Proc. of {ACM} {POPL}}}. \bibinfo{pages}{367--381}.
\newblock
\urldef\tempurl%
\url{https://doi.org/10.1145/96709.96747}
\showDOI{\tempurl}


\bibitem[\protect\citeauthoryear{Tobin{-}Hochstadt and
  Felleisen}{Tobin{-}Hochstadt and Felleisen}{2008}]%
        {DBLP:conf/popl/Tobin-HochstadtF08}
\bibfield{author}{\bibinfo{person}{Sam Tobin{-}Hochstadt} {and}
  \bibinfo{person}{Matthias Felleisen}.} \bibinfo{year}{2008}\natexlab{}.
\newblock \showarticletitle{The design and implementation of typed {Scheme}}.
  In \bibinfo{booktitle}{\emph{Proc. of {ACM} {POPL}}}.
  \bibinfo{pages}{395--406}.
\newblock
\urldef\tempurl%
\url{https://doi.org/10.1145/1328438.1328486}
\showDOI{\tempurl}


\bibitem[\protect\citeauthoryear{Tsukada and Igarashi}{Tsukada and
  Igarashi}{2010}]%
        {TsukadaIgarashi10LMCS}
\bibfield{author}{\bibinfo{person}{Takeshi Tsukada} {and}
  \bibinfo{person}{Atsushi Igarashi}.} \bibinfo{year}{2010}\natexlab{}.
\newblock \showarticletitle{A Logical Foundation for Environment Classifiers}.
\newblock \bibinfo{journal}{\emph{Logical Methods in Computer Science}}
  \bibinfo{volume}{6}, \bibinfo{number}{4:8} (\bibinfo{year}{2010}),
  \bibinfo{pages}{1--43}.
\newblock
\urldef\tempurl%
\url{https://doi.org/10.2168/LMCS-6(4:8)2010}
\showDOI{\tempurl}


\bibitem[\protect\citeauthoryear{Wadler and Findler}{Wadler and
  Findler}{2009}]%
        {DBLP:conf/esop/WadlerF09}
\bibfield{author}{\bibinfo{person}{Philip Wadler} {and}
  \bibinfo{person}{Robert~Bruce Findler}.} \bibinfo{year}{2009}\natexlab{}.
\newblock \showarticletitle{Well-Typed Programs Can't Be Blamed}. In
  \bibinfo{booktitle}{\emph{Proc. of {ESOP}}} \emph{(\bibinfo{series}{LNCS})},
  Vol.~\bibinfo{volume}{5502}. \bibinfo{publisher}{Springer},
  \bibinfo{pages}{1--16}.
\newblock
\urldef\tempurl%
\url{https://doi.org/10.1007/978-3-642-00590-9_1}
\showDOI{\tempurl}


\bibitem[\protect\citeauthoryear{Wright}{Wright}{1995}]%
        {DBLP:journals/lisp/Wright95}
\bibfield{author}{\bibinfo{person}{Andrew~K. Wright}.}
  \bibinfo{year}{1995}\natexlab{}.
\newblock \showarticletitle{Simple Imperative Polymorphism}.
\newblock \bibinfo{journal}{\emph{Lisp and Symbolic Computation}}
  \bibinfo{volume}{8}, \bibinfo{number}{4} (\bibinfo{year}{1995}),
  \bibinfo{pages}{343--355}.
\newblock


\bibitem[\protect\citeauthoryear{Wright and Felleisen}{Wright and
  Felleisen}{1994}]%
        {WrightFelleisenIC94}
\bibfield{author}{\bibinfo{person}{Andrew~K. Wright} {and}
  \bibinfo{person}{Matthias Felleisen}.} \bibinfo{year}{1994}\natexlab{}.
\newblock \showarticletitle{A Syntactic Approach to Type Soundness}.
\newblock \bibinfo{journal}{\emph{Information and Computation}}
  \bibinfo{volume}{115}, \bibinfo{number}{1} (\bibinfo{year}{1994}),
  \bibinfo{pages}{38--94}.
\newblock


\bibitem[\protect\citeauthoryear{Xie, Bi, and d.~S.~Oliveira}{Xie
  et~al\mbox{.}}{2018}]%
        {DBLP:conf/esop/XieBO18}
\bibfield{author}{\bibinfo{person}{Ningning Xie}, \bibinfo{person}{Xuan Bi},
  {and} \bibinfo{person}{Bruno~C. d. S.~Oliveira}.}
  \bibinfo{year}{2018}\natexlab{}.
\newblock \showarticletitle{Consistent Subtyping for All}. In
  \bibinfo{booktitle}{\emph{Programming Languages and Systems - 27th European
  Symposium on Programming, {ESOP} 2018, Held as Part of the European Joint
  Conferences on Theory and Practice of Software, {ETAPS} 2018, Thessaloniki,
  Greece, April 14-20, 2018, Proceedings}}. \bibinfo{pages}{3--30}.
\newblock
\urldef\tempurl%
\url{https://doi.org/10.1007/978-3-319-89884-1\_1}
\showDOI{\tempurl}


\end{thebibliography}


  \iffull
  \newif\ifrestate \restatetrue
  \clearpage
  \appendix

  \newif\ifrestate \restatetrue

\section{Proofs}

\subsection{Type Safety}

\ifrestate
\lemTySubstConsistency*
\else
\begin{lemmaA}[Type Substitution Preserves Consistency and Typing] \leavevmode
  \begin{enumerate}
    \item If $\ottnt{U}  \sim  \ottnt{U'}$, then $S  \ottsym{(}  \ottnt{U}  \ottsym{)}  \sim  S  \ottsym{(}  \ottnt{U'}  \ottsym{)}$ for any $S$.
    \item If $\Gamma  \vdash  \ottnt{f}  \ottsym{:}  \ottnt{U}$, then $S  \ottsym{(}  \Gamma  \ottsym{)}  \vdash  S  \ottsym{(}  \ottnt{f}  \ottsym{)}  \ottsym{:}  S  \ottsym{(}  \ottnt{U}  \ottsym{)}$ for any $S$.
  \end{enumerate}
\end{lemmaA}
\fi 

\begin{proof}
  By straightforward induction on the derivation.
\end{proof}

\ifrestate
\lemGroundTypes*
\else
\begin{lemmaA}[name=Ground Types,restate=lemGroundTypes] \label{lem:ground_types}
  \leavevmode
  \begin{enumerate}
    \item If $\ottnt{U}$ is neither a type variable nor the dynamic type, then there exists a unique $\ottnt{G}$ such that $\ottnt{U}  \sim  \ottnt{G}$.
    \item $\ottnt{G}  \sim  \ottnt{G'}$ if and only if $\ottnt{G}  \ottsym{=}  \ottnt{G'}$.
  \end{enumerate}
\end{lemmaA}
\fi 

\begin{proof}
  \leavevmode
  \begin{enumerate}
    \item Straightforward by case analysis on $\ottnt{U}$.
    \item By case analysis on $\ottnt{G}  \sim  \ottnt{G'}$ and $\ottnt{G}  \ottsym{=}  \ottnt{G'}$. \qedhere
  \end{enumerate}
\end{proof}

\ifrestate
\lemCanonicalForms*
\else
\begin{lemmaA}[name=Canonical Forms,restate=lemCanonicalForms] \label{lem:canonical_forms}
  If $ \emptyset   \vdash  w  \ottsym{:}  \ottnt{U}$, then one of the following holds:
  \begin{itemize}
    \item $\ottnt{U}  \ottsym{=}  \iota$ and $w  \ottsym{=}  \ottnt{c}$ for some $\iota$ and $\ottnt{c}$;
    \item $\ottnt{U}  \ottsym{=}  \ottnt{U_{{\mathrm{1}}}}  \!\rightarrow\!  \ottnt{U_{{\mathrm{2}}}}$ and $w  \ottsym{=}   \lambda  \ottmv{x} \!:\!  \ottnt{U_{{\mathrm{1}}}}  .\,  \ottnt{f} $ for some $\ottmv{x}$, $\ottnt{f}$, $\ottnt{U_{{\mathrm{1}}}}$, and $\ottnt{U_{{\mathrm{2}}}}$;
    \item $\ottnt{U}  \ottsym{=}  \ottnt{U_{{\mathrm{1}}}}  \!\rightarrow\!  \ottnt{U_{{\mathrm{2}}}}$ and $w  \ottsym{=}  w'  \ottsym{:}   \ottnt{U'_{{\mathrm{1}}}}  \!\rightarrow\!  \ottnt{U'_{{\mathrm{2}}}} \Rightarrow  \unskip ^ { \ell }  \! \ottnt{U_{{\mathrm{1}}}}  \!\rightarrow\!  \ottnt{U_{{\mathrm{2}}}} $
      for some $w'$, $\ottnt{U_{{\mathrm{1}}}}, \ottnt{U_{{\mathrm{2}}}}, \ottnt{U'_{{\mathrm{1}}}}$, $\ottnt{U'_{{\mathrm{2}}}}$, and $\ell$; or
    \item $\ottnt{U}  \ottsym{=}  \star$ and $w  \ottsym{=}  w'  \ottsym{:}   \ottnt{G} \Rightarrow  \unskip ^ { \ell }  \! \star $
      for some $w'$, $\ottnt{G}$, and $\ell$.
  \end{itemize}
\end{lemmaA}
\fi

\begin{proof}
  By case analysis on the typing rule applied to derive $ \emptyset   \vdash  w  \ottsym{:}  \ottnt{U}$.
\end{proof}

\ifrestate
\lemProgress*
\else
\begin{lemmaA}[name=Progress,restate=lemProgress] \label{lem:progress}
  If $ \emptyset   \vdash  \ottnt{f}  \ottsym{:}  \ottnt{U}$, then one of the following holds:
  \begin{itemize}
    \item $\ottnt{f} \,  \xmapsto{ \mathmakebox[0.4em]{} S \mathmakebox[0.3em]{} }  \, \ottnt{f'}$ for some $S$ and $\ottnt{f'}$;
    \item $\ottnt{f}$ is a value; or
    \item $\ottnt{f}  \ottsym{=}  \textsf{\textup{blame}\relax} \, \ell$ for some $\ell$.
  \end{itemize}
\end{lemmaA}
\fi

\begin{proof}
  By induction on the typing derivation.

  \begin{caseanalysis}
    \case{\rnp{T\_VarP}} Cannot happen.

    \case{\rnp{T\_Const}, \rnp{T\_Abs}, \rnp{T\_Blame}} Obvious.

    \case{\rnp{T\_Op}}
    We are given $ \emptyset   \vdash  \mathit{op} \, \ottsym{(}  \ottnt{f_{{\mathrm{1}}}}  \ottsym{,}  \ottnt{f_{{\mathrm{2}}}}  \ottsym{)}  \ottsym{:}  \iota$
    for some $\ottnt{f_{{\mathrm{1}}}}$, $\ottnt{f_{{\mathrm{2}}}}$, and $\iota$ where $\ottnt{U}  \ottsym{=}  \iota$.
    By inversion, we have $ \mathit{ty} ( \mathit{op} )   \ottsym{=}  \iota_{{\mathrm{1}}}  \!\rightarrow\!  \iota_{{\mathrm{2}}}  \!\rightarrow\!  \iota$,
    $ \emptyset   \vdash  \ottnt{f_{{\mathrm{1}}}}  \ottsym{:}  \iota_{{\mathrm{1}}}$, and $ \emptyset   \vdash  \ottnt{f_{{\mathrm{2}}}}  \ottsym{:}  \iota_{{\mathrm{2}}}$
    for some $\iota_{{\mathrm{1}}}$ and $\iota_{{\mathrm{2}}}$.

    If $\ottnt{f_{{\mathrm{1}}}}$ is not a value, by case analysis on $\ottnt{f_{{\mathrm{1}}}}$ with the IH.
    \begin{caseanalysis}
      \case{$\ottnt{f_{{\mathrm{1}}}} \,  \xmapsto{ \mathmakebox[0.4em]{} S_{{\mathrm{1}}} \mathmakebox[0.3em]{} }  \, \ottnt{f'_{{\mathrm{1}}}}$}
      By case analysis on the evaluation rule applied to $\ottnt{f_{{\mathrm{1}}}}$.
      \begin{caseanalysis}
        \case{\rnp{E\_Step}}
        We are given $\ottnt{f_{{\mathrm{11}}}} \,  \xrightarrow{ \mathmakebox[0.4em]{} S_{{\mathrm{1}}} \mathmakebox[0.3em]{} }  \, \ottnt{f'_{{\mathrm{11}}}}$ where $\ottnt{f_{{\mathrm{1}}}}  \ottsym{=}  \ottnt{E}  [  \ottnt{f_{{\mathrm{11}}}}  ]$
        for some $\ottnt{E}$, $\ottnt{f_{{\mathrm{11}}}}$, and $\ottnt{f'_{{\mathrm{11}}}}$.
        Finish by \rnp{E\_Step}.

        \case{\rnp{E\_Abort}}
        We are given $\ottnt{f_{{\mathrm{1}}}}  \ottsym{=}  \ottnt{E}  [  \textsf{\textup{blame}\relax} \, \ell  ]$ for some $\ottnt{E}$ and $\ell$.
        Finish by \rnp{E\_Abort}.
      \end{caseanalysis}

      \case{$\ottnt{f_{{\mathrm{1}}}}$ is a value} Contradiction.

      \case{$\ottnt{f_{{\mathrm{1}}}}  \ottsym{=}  \textsf{\textup{blame}\relax} \, \ell$} Finish by \rnp{E\_Abort}.
    \end{caseanalysis}

    If $\ottnt{f_{{\mathrm{2}}}}$ is not a value, we finish similarly to the previous case.

    Otherwise, suppose both $\ottnt{f_{{\mathrm{1}}}}$ and $\ottnt{f_{{\mathrm{2}}}}$ are values.
    By Lemma \ref{lem:canonical_forms}, both $\ottnt{f_{{\mathrm{1}}}}$ and $\ottnt{f_{{\mathrm{2}}}}$ are constants.
    We finish by \rnp{R\_Op} and \rnp{E\_Step}.

    \case{\rnp{T\_App}}
    We are given $ \emptyset   \vdash  \ottnt{f_{{\mathrm{1}}}} \, \ottnt{f_{{\mathrm{2}}}}  \ottsym{:}  \ottnt{U}$ for some $\ottnt{f_{{\mathrm{1}}}}$ and $\ottnt{f_{{\mathrm{2}}}}$.
    By inversion, we have $ \emptyset   \vdash  \ottnt{f_{{\mathrm{1}}}}  \ottsym{:}  \ottnt{U'}  \!\rightarrow\!  \ottnt{U}$ and $ \emptyset   \vdash  \ottnt{f_{{\mathrm{2}}}}  \ottsym{:}  \ottnt{U'}$
    for some $\ottnt{U'}$.

    If $\ottnt{f_{{\mathrm{1}}}}$ is not a value, by case analysis on $\ottnt{f_{{\mathrm{1}}}}$ with the IH.
    \begin{caseanalysis}
      \case{$\ottnt{f_{{\mathrm{1}}}} \,  \xmapsto{ \mathmakebox[0.4em]{} S_{{\mathrm{1}}} \mathmakebox[0.3em]{} }  \, \ottnt{f'_{{\mathrm{1}}}}$}
      By case analysis on the evaluation rule applied to $\ottnt{f_{{\mathrm{1}}}}$.
      \begin{caseanalysis}
        \case{\rnp{E\_Step}}
        We are given $\ottnt{f_{{\mathrm{11}}}} \,  \xrightarrow{ \mathmakebox[0.4em]{} S_{{\mathrm{1}}} \mathmakebox[0.3em]{} }  \, \ottnt{f'_{{\mathrm{11}}}}$ where $\ottnt{f_{{\mathrm{1}}}}  \ottsym{=}  \ottnt{E}  [  \ottnt{f_{{\mathrm{11}}}}  ]$
        for some $\ottnt{E}$, $\ottnt{f_{{\mathrm{11}}}}$, and $\ottnt{f'_{{\mathrm{11}}}}$.
        Finish by \rnp{E\_Step}.

        \case{\rnp{E\_Abort}}
        We are given $\ottnt{f_{{\mathrm{1}}}}  \ottsym{=}  \ottnt{E}  [  \textsf{\textup{blame}\relax} \, \ell  ]$ for some $\ottnt{E}$ and $\ell$.
        Finish by \rnp{E\_Abort}.
      \end{caseanalysis}

      \case{$\ottnt{f_{{\mathrm{1}}}}$ is a value} Contradiction.

      \case{$\ottnt{f_{{\mathrm{1}}}}  \ottsym{=}  \textsf{\textup{blame}\relax} \, \ell$} Finish by \rnp{E\_Abort}.
    \end{caseanalysis}

    If $\ottnt{f_{{\mathrm{2}}}}$ is not a value, we finish similarly to the previous case.

    Otherwise, suppose both $\ottnt{f_{{\mathrm{1}}}}$ and $\ottnt{f_{{\mathrm{2}}}}$ are values.
    By case analysis on the structure of $\ottnt{f_{{\mathrm{1}}}}$
    using Lemma \ref{lem:canonical_forms}.

    \begin{caseanalysis}
      \case{$\ottnt{f_{{\mathrm{1}}}}  \ottsym{=}   \lambda  \ottmv{x} \!:\!  \ottnt{U'}  .\,  \ottnt{f'_{{\mathrm{1}}}} $ for some $\ottmv{x}$ and $\ottnt{f'_{{\mathrm{1}}}}$}
      \leavevmode\\
      We finish by \rnp{R\_Beta} and \rnp{E\_Step}.

      \case{$\ottnt{f_{{\mathrm{1}}}}  \ottsym{=}  w  \ottsym{:}   \ottnt{U_{{\mathrm{11}}}}  \!\rightarrow\!  \ottnt{U_{{\mathrm{12}}}} \Rightarrow  \unskip ^ { \ell }  \! \ottnt{U_{{\mathrm{21}}}}  \!\rightarrow\!  \ottnt{U_{{\mathrm{22}}}} $ for some $w$, $\ottnt{U_{{\mathrm{11}}}}$, $\ottnt{U_{{\mathrm{12}}}}$, $\ottnt{U_{{\mathrm{21}}}}$, and $\ottnt{U_{{\mathrm{22}}}}$}
      \leavevmode\\
      We finish by \rnp{R\_AppCast} and \rnp{E\_Step}.
    \end{caseanalysis}

    \case{\rnp{T\_Cast}}
    We are given $ \emptyset   \vdash  \ottsym{(}  \ottnt{f_{{\mathrm{1}}}}  \ottsym{:}   \ottnt{U'} \Rightarrow  \unskip ^ { \ell }  \! \ottnt{U}   \ottsym{)}  \ottsym{:}  \ottnt{U}$ for some $\ottnt{f_{{\mathrm{1}}}}$, $\ottnt{U'}$, and $\ell$.
    By inversion, we have $ \emptyset   \vdash  \ottnt{f_{{\mathrm{1}}}}  \ottsym{:}  \ottnt{U'}$ and $\ottnt{U'}  \sim  \ottnt{U}$.

    If $\ottnt{f_{{\mathrm{1}}}}$ is not a value, we finish similarly to the case for \rnp{T\_App}.
    Otherwise, we proceed by case analysis on $\ottnt{U'}  \sim  \ottnt{U}$.

    \begin{caseanalysis}
      \case{\rnp{C\_Base}} By \rnp{R\_IdBase} and \rnp{E\_Step}.

      \case{\rnp{C\_TyVar}} This contradicts Lemma \ref{lem:canonical_forms}.

      \case{\rnp{C\_DynL}}
      We are given $\star  \sim  \ottnt{U}$.
      By Lemma \ref{lem:canonical_forms}, $\ottnt{f_{{\mathrm{1}}}}  \ottsym{=}  w'  \ottsym{:}   \ottnt{G'} \Rightarrow  \unskip ^ { \ell' }  \! \star $
      for some $w'$, $\ottnt{G'}$, and $\ell'$.
      By inversion, $ \emptyset   \vdash  w'  \ottsym{:}  \ottnt{G'}$ and $\ottnt{G'}  \sim  \star$.
      By case analysis on $\ottnt{U'}$.

      \begin{caseanalysis}
        \case{$\ottnt{U}  \ottsym{=}  \star$}
        We finish by \rnp{R\_IdStar} and \rnp{E\_Step}.

        \case{$\ottnt{U}  \ottsym{=}  \ottnt{G}$ and $\ottnt{G}  \ottsym{=}  \ottnt{G'}$}
        We finish by \rnp{R\_Succeed} and \rnp{E\_Step}.

        \case{$\ottnt{U}  \ottsym{=}  \ottnt{G}$ and $\ottnt{G}  \neq  \ottnt{G'}$}
        We finish by \rnp{R\_Fail} and \rnp{E\_Step}.

        \case{$\ottnt{U}  \ottsym{=}  \ottmv{X}$ and $\ottnt{G'}  \ottsym{=}  \iota$ for some $\ottmv{X}$ and $\iota$}
        We finish by \rnp{R\_InstBase} and \rnp{E\_Step}.

        \case{$\ottnt{U}  \ottsym{=}  \ottmv{X}$ and $\ottnt{G'}  \ottsym{=}  \star  \!\rightarrow\!  \star$ for some $\ottmv{X}$}
        We finish by \rnp{R\_InstArrow} and \rnp{E\_Step}.

        \otherwise
        We finish by \rnp{R\_Expand} and \rnp{E\_Step}.
      \end{caseanalysis}

      \case{\rnp{C\_DynR}}
      We are given $\ottnt{U'}  \sim  \star$.
      By case analysis on $\ottnt{U'}$.

      \begin{caseanalysis}
        \case{$\ottnt{U'}  \ottsym{=}  \star$} We finish by \rnp{R\_IdStar} and \rnp{E\_Step}.
        \otherwise
        By Lemma \ref{lem:ground_types}.1,
        there is a unique ground type $\ottnt{G}$ such that $\ottnt{U'}  \sim  \ottnt{G}$.
        If $\ottnt{U'}  \ottsym{=}  \ottnt{G}$, then $\ottnt{f}$ is a value.
        Otherwise, we finish by \rnp{R\_Ground} and \rnp{E\_Step}.
      \end{caseanalysis}

      \case{\rnp{C\_Arrow}} $\ottnt{f}$ is a value.
    \end{caseanalysis}

    \case{\rnp{T\_LetP}}
    We are given $ \emptyset   \vdash   \textsf{\textup{let}\relax} \,  \ottmv{x}  =   \Lambda    \overrightarrow{ \ottmv{X_{\ottmv{i}}} }  .\,  w_{{\mathrm{1}}}   \textsf{\textup{ in }\relax}  \ottnt{f_{{\mathrm{2}}}}   \ottsym{:}  \ottnt{U}$
    for some $\ottmv{x}$, $ \overrightarrow{ \ottmv{X_{\ottmv{i}}} } $, $w_{{\mathrm{1}}}$, and $\ottnt{f_{{\mathrm{2}}}}$.
    We finish by \rnp{R\_LetP} and \rnp{E\_Step}.
    \qedhere
  \end{caseanalysis}
\end{proof}

\begin{lemmaA}[Weakening] \label{lem:weakening}
  If $\Gamma  \vdash  \ottnt{f}  \ottsym{:}  \ottnt{U}$ and $\Gamma$ does not contain $\ottmv{x}$,
  then $ \Gamma ,   \ottmv{x}  :  \sigma    \vdash  \ottnt{f}  \ottsym{:}  \ottnt{U}$.
\end{lemmaA}

\begin{proof}
  Straightforward by induction on the typing derivation of $\ottnt{f}$.
\end{proof}

\begin{lemmaA}[Strengthening] \label{lem:strengthening}
  If $ \Gamma ,   \ottmv{x}  :  \sigma    \vdash  \ottnt{f}  \ottsym{:}  \ottnt{U}$ and $\ottmv{x} \, \not\in \, \textit{fv} \, \ottsym{(}  \ottnt{f}  \ottsym{)}$,
  then $\Gamma  \vdash  \ottnt{f}  \ottsym{:}  \ottnt{U}$.
\end{lemmaA}

\begin{proof}
  Straightforward by induction on the typing derivation of $\ottnt{f}$.
\end{proof}

\begin{lemmaA} \label{lem:type_param_substitution_in_cosistency}
  If $\ottnt{U}  \sim  \ottnt{U'}$,
  then $S  \ottsym{(}  \ottnt{U}  \ottsym{)}  \sim  S  \ottsym{(}  \ottnt{U'}  \ottsym{)}$ for any $S$.
\end{lemmaA}

\begin{proof}
  Straightforward by induction on the consistency derivation.
\end{proof}

\begin{lemmaA} \label{lem:type_param_substitution}
  If $\Gamma  \vdash  \ottnt{f}  \ottsym{:}  \ottnt{U}$,
  then $S  \ottsym{(}  \Gamma  \ottsym{)}  \vdash  S  \ottsym{(}  \ottnt{f}  \ottsym{)}  \ottsym{:}  S  \ottsym{(}  \ottnt{U}  \ottsym{)}$ for any $S$.
\end{lemmaA}

\begin{proof}
  Straightforward by induction on the typing derivation of $\ottnt{f}$.
\end{proof}

\begin{lemmaA} \label{lem:term_var_substitution}
  If $\Gamma  \vdash  w  \ottsym{:}  \ottnt{U}$ and $ \Gamma ,   \ottmv{x}  :  \forall \,  \overrightarrow{ \ottmv{X_{\ottmv{i}}} }   \ottsym{.}  \ottnt{U}    \vdash  \ottnt{f}  \ottsym{:}  \ottnt{U'}$,
  then $\Gamma  \vdash  \ottnt{f}  [  \ottmv{x}  \ottsym{:=}   \Lambda    \overrightarrow{ \ottmv{X_{\ottmv{i}}} }  .\,  w   ]  \ottsym{:}  \ottnt{U'}$.
\end{lemmaA}

\begin{proof}
  Straightforward by induction on the typing derivation of $\ottnt{f}$.\\
  Use Lemma~\ref{lem:type_param_substitution} for the case where $\ottnt{f} = \ottmv{x}  [   \overrightarrow{ \mathbbsl{T}_{\ottmv{i}} }   ]$.
\end{proof}

\ifrestate
\lemPreservation*
\else
\begin{lemmaA}[name=Preservation,restate=lemPreservation] \label{lem:preservation}
  Suppose that $ \emptyset   \vdash  \ottnt{f}  \ottsym{:}  \ottnt{U}$.
  \begin{enumerate}
    \item If $\ottnt{f} \,  \xrightarrow{ \mathmakebox[0.4em]{} S \mathmakebox[0.3em]{} }  \, \ottnt{f'}$, then $ \emptyset   \vdash  S  \ottsym{(}  \ottnt{f'}  \ottsym{)}  \ottsym{:}  S  \ottsym{(}  \ottnt{U}  \ottsym{)}$.
    \item If $\ottnt{f} \,  \xmapsto{ \mathmakebox[0.4em]{} S \mathmakebox[0.3em]{} }  \, \ottnt{f'}$, then $ \emptyset   \vdash  \ottnt{f'}  \ottsym{:}  S  \ottsym{(}  \ottnt{U}  \ottsym{)}$.
  \end{enumerate}
\end{lemmaA}
\fi 

\begin{proof}
  \leavevmode
  \begin{enumerate}
    \item
      By case analysis on the typing rule applied to $\ottnt{f}$.
      \begin{caseanalysis}
        \case{\rnp{T\_Op}}
        We are given $ \emptyset   \vdash  \mathit{op} \, \ottsym{(}  \ottnt{f_{{\mathrm{1}}}}  \ottsym{,}  \ottnt{f_{{\mathrm{2}}}}  \ottsym{)}  \ottsym{:}  \iota$
        for some $\ottnt{f_{{\mathrm{1}}}}$, $\ottnt{f_{{\mathrm{2}}}}$, and $\iota$
        where $\ottnt{f}  \ottsym{=}  \ottnt{f_{{\mathrm{1}}}} \, \ottnt{f_{{\mathrm{2}}}}$ and $\ottnt{U}  \ottsym{=}  \iota$.
        By inversion, we have $ \mathit{ty} ( \mathit{op} )   \ottsym{=}  \iota_{{\mathrm{1}}}  \!\rightarrow\!  \iota_{{\mathrm{2}}}  \!\rightarrow\!  \iota$,
        $ \emptyset   \vdash  \ottnt{f_{{\mathrm{1}}}}  \ottsym{:}  \iota_{{\mathrm{1}}}$, and $ \emptyset   \vdash  \ottnt{f_{{\mathrm{2}}}}  \ottsym{:}  \iota_{{\mathrm{2}}}$
        for some $\iota_{{\mathrm{1}}}$ and $\iota_{{\mathrm{2}}}$.
        By case analysis on the reduction rules applicable to $\ottnt{f_{{\mathrm{1}}}} \, \ottnt{f_{{\mathrm{2}}}}$.

        \begin{caseanalysis}
          \case{\rnp{R\_Op}}
          We are given $\mathit{op} \, \ottsym{(}  w_{{\mathrm{1}}}  \ottsym{,}  w_{{\mathrm{2}}}  \ottsym{)} \,  \xrightarrow{ \mathmakebox[0.4em]{} [  ] \mathmakebox[0.3em]{} }  \,  \llbracket\mathit{op}\rrbracket ( w_{{\mathrm{1}}} ,  w_{{\mathrm{2}}} ) $
          where $S  \ottsym{=}  [  ]$, $\ottnt{f_{{\mathrm{1}}}}  \ottsym{=}  w_{{\mathrm{1}}}$, and $\ottnt{f_{{\mathrm{2}}}}  \ottsym{=}  w_{{\mathrm{2}}}$.
          $ \llbracket\mathit{op}\rrbracket ( w_{{\mathrm{1}}} ,  w_{{\mathrm{2}}} ) $ is assumed to have type $\iota$.

          \otherwise Cannot happen.
        \end{caseanalysis}

        \case{\rnp{T\_App}}
        We are given $ \emptyset   \vdash  \ottnt{f_{{\mathrm{1}}}} \, \ottnt{f_{{\mathrm{2}}}}  \ottsym{:}  \ottnt{U}$ for some $\ottnt{f_{{\mathrm{1}}}}$ and $\ottnt{f_{{\mathrm{2}}}}$
        where $\ottnt{f}  \ottsym{=}  \ottnt{f_{{\mathrm{1}}}} \, \ottnt{f_{{\mathrm{2}}}}$.
        By inversion, we have $ \emptyset   \vdash  \ottnt{f_{{\mathrm{1}}}}  \ottsym{:}  \ottnt{U'}  \!\rightarrow\!  \ottnt{U}$ and $ \emptyset   \vdash  \ottnt{f_{{\mathrm{2}}}}  \ottsym{:}  \ottnt{U'}$
        for some $\ottnt{U'}$.
        By case analysis on the reduction rules applicable to $\ottnt{f_{{\mathrm{1}}}} \, \ottnt{f_{{\mathrm{2}}}}$.

        \begin{caseanalysis}
          \case{\rnp{R\_Beta}}
          We are given $\ottsym{(}   \lambda  \ottmv{x} \!:\!  \ottnt{U'}  .\,  \ottnt{f'_{{\mathrm{1}}}}   \ottsym{)} \, \ottnt{f_{{\mathrm{2}}}} \,  \xrightarrow{ \mathmakebox[0.4em]{} [  ] \mathmakebox[0.3em]{} }  \, \ottnt{f'_{{\mathrm{1}}}}  [  \ottmv{x}  \ottsym{:=}  \ottnt{f_{{\mathrm{2}}}}  ]$
          where $S  \ottsym{=}  [  ]$, $\ottnt{f_{{\mathrm{1}}}}  \ottsym{=}   \lambda  \ottmv{x} \!:\!  \ottnt{U'}  .\,  \ottnt{f'_{{\mathrm{1}}}} $, and $\ottnt{f_{{\mathrm{2}}}}$ is a value.
          By inversion, $ \ottmv{x}  :  \ottnt{U'}   \vdash  \ottnt{f'_{{\mathrm{1}}}}  \ottsym{:}  \ottnt{U}$.
          By Lemma \ref{lem:term_var_substitution}, $ \emptyset   \vdash  \ottnt{f'_{{\mathrm{1}}}}  [  \ottmv{x}  \ottsym{:=}  \ottnt{f_{{\mathrm{2}}}}  ]  \ottsym{:}  \ottnt{U}$.
          Finally, $ \emptyset   \vdash  [  ]  \ottsym{(}  \ottnt{f'_{{\mathrm{1}}}}  [  \ottmv{x}  \ottsym{:=}  \ottnt{f_{{\mathrm{2}}}}  ]  \ottsym{)}  \ottsym{:}  [  ]  \ottsym{(}  \ottnt{U}  \ottsym{)}$.

          \case{\rnp{R\_AppCast}}
          We are given
          $\ottsym{(}  w_{{\mathrm{1}}}  \ottsym{:}   \ottnt{U_{{\mathrm{11}}}}  \!\rightarrow\!  \ottnt{U_{{\mathrm{12}}}} \Rightarrow  \unskip ^ { \ell }  \! \ottnt{U'}  \!\rightarrow\!  \ottnt{U}   \ottsym{)} \, \ottnt{f_{{\mathrm{2}}}} \,  \xrightarrow{ \mathmakebox[0.4em]{} [  ] \mathmakebox[0.3em]{} }  \, \ottsym{(}  w_{{\mathrm{1}}} \, \ottsym{(}  \ottnt{f_{{\mathrm{2}}}}  \ottsym{:}   \ottnt{U'} \Rightarrow  \unskip ^ {  \bar{ \ell }  }  \! \ottnt{U_{{\mathrm{11}}}}   \ottsym{)}  \ottsym{)}  \ottsym{:}   \ottnt{U_{{\mathrm{12}}}} \Rightarrow  \unskip ^ { \ell }  \! \ottnt{U} $
          where $S  \ottsym{=}  [  ]$, $\ottnt{f_{{\mathrm{1}}}}  \ottsym{=}  w_{{\mathrm{1}}}  \ottsym{:}   \ottnt{U_{{\mathrm{11}}}}  \!\rightarrow\!  \ottnt{U_{{\mathrm{12}}}} \Rightarrow  \unskip ^ { \ell }  \! \ottnt{U'}  \!\rightarrow\!  \ottnt{U} $,
          and $\ottnt{f_{{\mathrm{2}}}}$ is a value.
          By inversion, $ \emptyset   \vdash  w_{{\mathrm{1}}}  \ottsym{:}  \ottnt{U_{{\mathrm{11}}}}  \!\rightarrow\!  \ottnt{U_{{\mathrm{12}}}}$ and $\ottnt{U_{{\mathrm{11}}}}  \!\rightarrow\!  \ottnt{U_{{\mathrm{12}}}}  \sim  \ottnt{U'}  \!\rightarrow\!  \ottnt{U}$.
          By \rnp{C\_Arrow}, $\ottnt{U_{{\mathrm{11}}}}  \sim  \ottnt{U'}$ and $\ottnt{U_{{\mathrm{12}}}}  \sim  \ottnt{U}$.
          By \rnp{T\_Cast}, $ \emptyset   \vdash  \ottsym{(}  \ottnt{f_{{\mathrm{2}}}}  \ottsym{:}   \ottnt{U'} \Rightarrow  \unskip ^ {  \bar{ \ell }  }  \! \ottnt{U_{{\mathrm{11}}}}   \ottsym{)}  \ottsym{:}  \ottnt{U_{{\mathrm{11}}}}$.
          By \rnp{T\_App}, $ \emptyset   \vdash  w_{{\mathrm{1}}} \, \ottsym{(}  \ottnt{f_{{\mathrm{2}}}}  \ottsym{:}   \ottnt{U'} \Rightarrow  \unskip ^ {  \bar{ \ell }  }  \! \ottnt{U_{{\mathrm{11}}}}   \ottsym{)}  \ottsym{:}  \ottnt{U_{{\mathrm{12}}}}$.
          By \rnp{T\_Cast}, $ \emptyset   \vdash  \ottsym{(}  w_{{\mathrm{1}}} \, \ottsym{(}  \ottnt{f_{{\mathrm{2}}}}  \ottsym{:}   \ottnt{U'} \Rightarrow  \unskip ^ {  \bar{ \ell }  }  \! \ottnt{U_{{\mathrm{11}}}}   \ottsym{)}  \ottsym{)}  \ottsym{:}   \ottnt{U_{{\mathrm{12}}}} \Rightarrow  \unskip ^ { \ell }  \! \ottnt{U}   \ottsym{:}  \ottnt{U}$.
          Finally, $ \emptyset   \vdash  [  ]  \ottsym{(}  \ottsym{(}  w_{{\mathrm{1}}} \, \ottsym{(}  \ottnt{f_{{\mathrm{2}}}}  \ottsym{:}   \ottnt{U'} \Rightarrow  \unskip ^ {  \bar{ \ell }  }  \! \ottnt{U_{{\mathrm{11}}}}   \ottsym{)}  \ottsym{)}  \ottsym{:}   \ottnt{U_{{\mathrm{12}}}} \Rightarrow  \unskip ^ { \ell }  \! \ottnt{U}   \ottsym{)}  \ottsym{:}  [  ]  \ottsym{(}  \ottnt{U}  \ottsym{)}$.

          \otherwise Cannot happen.
        \end{caseanalysis}

        \case{\rnp{T\_Cast}}
        We are given $ \emptyset   \vdash  \ottsym{(}  \ottnt{f_{{\mathrm{1}}}}  \ottsym{:}   \ottnt{U'} \Rightarrow  \unskip ^ { \ell }  \! \ottnt{U}   \ottsym{)}  \ottsym{:}  \ottnt{U}$
        for some $\ottnt{f_{{\mathrm{1}}}}$, $\ottnt{U'}$, and $\ell$ where $\ottnt{f}  \ottsym{=}  \ottnt{f_{{\mathrm{1}}}}  \ottsym{:}   \ottnt{U'} \Rightarrow  \unskip ^ { \ell }  \! \ottnt{U} $.
        By inversion, we have $ \emptyset   \vdash  \ottnt{f_{{\mathrm{1}}}}  \ottsym{:}  \ottnt{U'}$ and $\ottnt{U'}  \sim  \ottnt{U}$.
        By case analysis on the reduction rule applicable to $\ottnt{f_{{\mathrm{1}}}}  \ottsym{:}   \ottnt{U'} \Rightarrow  \unskip ^ { \ell }  \! \ottnt{U} $.

        \begin{caseanalysis}
          \case{\rnp{R\_IdBase}, \rnp{R\_IdStar}}
          We are given $w  \ottsym{:}   \ottnt{U} \Rightarrow  \unskip ^ { \ell }  \! \ottnt{U}  \,  \xrightarrow{ \mathmakebox[0.4em]{} [  ] \mathmakebox[0.3em]{} }  \, w$ where $\ottnt{f_{{\mathrm{1}}}}  \ottsym{=}  w$ and $\ottnt{U}  \ottsym{=}  \ottnt{U'}$.
          So, $ \emptyset   \vdash  [  ]  \ottsym{(}  w  \ottsym{)}  \ottsym{:}  [  ]  \ottsym{(}  \ottnt{U}  \ottsym{)}$.

          \case{\rnp{R\_Succeed}}
          We are given $w  \ottsym{:}   \ottnt{G} \Rightarrow  \unskip ^ { \ell' }  \!  \star \Rightarrow  \unskip ^ { \ell }  \! \ottnt{G}   \,  \xrightarrow{ \mathmakebox[0.4em]{} [  ] \mathmakebox[0.3em]{} }  \, w$ where
          $\ottnt{f_{{\mathrm{1}}}}  \ottsym{=}  w  \ottsym{:}   \ottnt{G} \Rightarrow  \unskip ^ { \ell' }  \! \star $, $\ottnt{U}  \ottsym{=}  \ottnt{G}$, and $\ottnt{U'}  \ottsym{=}  \star$.
          By inversion, $ \emptyset   \vdash  \ottsym{(}  w  \ottsym{:}   \ottnt{G} \Rightarrow  \unskip ^ { \ell' }  \! \star   \ottsym{)}  \ottsym{:}  \star$.
          By inversion, $ \emptyset   \vdash  w  \ottsym{:}  \ottnt{G}$.
          So, $ \emptyset   \vdash  [  ]  \ottsym{(}  w  \ottsym{)}  \ottsym{:}  [  ]  \ottsym{(}  \ottnt{G}  \ottsym{)}$.

          \case{\rnp{R\_Fail}}
          We are given $w  \ottsym{:}   \ottnt{G_{{\mathrm{1}}}} \Rightarrow  \unskip ^ { \ell' }  \!  \star \Rightarrow  \unskip ^ { \ell }  \! \ottnt{G_{{\mathrm{2}}}}   \,  \xrightarrow{ \mathmakebox[0.4em]{} [  ] \mathmakebox[0.3em]{} }  \, \textsf{\textup{blame}\relax} \, \ell$ where
          $\ottnt{G_{{\mathrm{1}}}}  \neq  \ottnt{G_{{\mathrm{2}}}}$,
          $\ottnt{f_{{\mathrm{1}}}}  \ottsym{=}  w  \ottsym{:}   \ottnt{G_{{\mathrm{1}}}} \Rightarrow  \unskip ^ { \ell' }  \! \star $, $\ottnt{U}  \ottsym{=}  \ottnt{G_{{\mathrm{2}}}}$, and $\ottnt{U'}  \ottsym{=}  \star$.
          By \rnp{T\_Blame}, $ \emptyset   \vdash  \textsf{\textup{blame}\relax} \, \ell  \ottsym{:}  \ottnt{G_{{\mathrm{2}}}}$.
          So, $ \emptyset   \vdash  [  ]  \ottsym{(}  \textsf{\textup{blame}\relax} \, \ell  \ottsym{)}  \ottsym{:}  [  ]  \ottsym{(}  \ottnt{G_{{\mathrm{2}}}}  \ottsym{)}$.

          \case{\rnp{R\_Ground}}
          We are given $w  \ottsym{:}   \ottnt{U'} \Rightarrow  \unskip ^ { \ell }  \! \star  \,  \xrightarrow{ \mathmakebox[0.4em]{} [  ] \mathmakebox[0.3em]{} }  \, w  \ottsym{:}   \ottnt{U'} \Rightarrow  \unskip ^ { \ell }  \!  \ottnt{G} \Rightarrow  \unskip ^ { \ell }  \! \star  $ where
          $\ottnt{f_{{\mathrm{1}}}}  \ottsym{=}  w$, $\ottnt{U'}  \neq  \star$, $\ottnt{U'}  \neq  \ottnt{G}$, $\ottnt{U'}  \sim  \ottnt{G}$, and $\ottnt{U}  \ottsym{=}  \star$.
          By inversion, $ \emptyset   \vdash  w  \ottsym{:}  \ottnt{U'}$.
          By \rnp{T\_Cast}, $ \emptyset   \vdash  w  \ottsym{:}   \ottnt{U'} \Rightarrow  \unskip ^ { \ell }  \! \ottnt{G}   \ottsym{:}  \ottnt{G}$.
          By \rnp{T\_Cast}, $ \emptyset   \vdash  w  \ottsym{:}   \ottnt{U'} \Rightarrow  \unskip ^ { \ell }  \!  \ottnt{G} \Rightarrow  \unskip ^ { \ell }  \! \star    \ottsym{:}  \star$.
          So, $ \emptyset   \vdash  [  ]  \ottsym{(}  w  \ottsym{:}   \ottnt{U'} \Rightarrow  \unskip ^ { \ell }  \!  \ottnt{G} \Rightarrow  \unskip ^ { \ell }  \! \star    \ottsym{)}  \ottsym{:}  [  ]  \ottsym{(}  \star  \ottsym{)}$.

          \case{\rnp{R\_Expand}}
          We are given $w  \ottsym{:}   \star \Rightarrow  \unskip ^ { \ell }  \! \ottnt{U}  \,  \xrightarrow{ \mathmakebox[0.4em]{} [  ] \mathmakebox[0.3em]{} }  \, w  \ottsym{:}   \star \Rightarrow  \unskip ^ { \ell }  \!  \ottnt{G} \Rightarrow  \unskip ^ { \ell }  \! \ottnt{U}  $ where
          $\ottnt{f_{{\mathrm{1}}}}  \ottsym{=}  w$, $\ottnt{U}  \neq  \star$, $\ottnt{U}  \neq  \ottnt{G}$, $\ottnt{U}  \sim  \ottnt{G}$, and $\ottnt{U}  \ottsym{=}  \star$.
          By inversion, $ \emptyset   \vdash  w  \ottsym{:}  \star$.
          By \rnp{T\_Cast}, $ \emptyset   \vdash  w  \ottsym{:}   \star \Rightarrow  \unskip ^ { \ell }  \! \ottnt{G}   \ottsym{:}  \ottnt{G}$.
          By \rnp{T\_Cast}, $ \emptyset   \vdash  w  \ottsym{:}   \star \Rightarrow  \unskip ^ { \ell }  \!  \ottnt{G} \Rightarrow  \unskip ^ { \ell }  \! \ottnt{U}    \ottsym{:}  \ottnt{U}$.
          So, $ \emptyset   \vdash  [  ]  \ottsym{(}  w  \ottsym{:}   \star \Rightarrow  \unskip ^ { \ell }  \!  \ottnt{G} \Rightarrow  \unskip ^ { \ell }  \! \ottnt{U}    \ottsym{)}  \ottsym{:}  [  ]  \ottsym{(}  \star  \ottsym{)}$.

          \case{\rnp{R\_InstBase}}
          We are given $w  \ottsym{:}   \iota \Rightarrow  \unskip ^ { \ell' }  \!  \star \Rightarrow  \unskip ^ { \ell }  \! \ottmv{X}   \,  \xrightarrow{ \mathmakebox[0.4em]{} [  \ottmv{X}  :=  \iota  ] \mathmakebox[0.3em]{} }  \, w$ where
          $\ottnt{f_{{\mathrm{1}}}}  \ottsym{=}  w  \ottsym{:}   \iota \Rightarrow  \unskip ^ { \ell' }  \! \star $, $\ottnt{U}  \ottsym{=}  \ottmv{X}$, and $\ottnt{U'}  \ottsym{=}  \star$.
          By inversion, $ \emptyset   \vdash  \ottsym{(}  w  \ottsym{:}   \iota \Rightarrow  \unskip ^ { \ell' }  \! \star   \ottsym{)}  \ottsym{:}  \star$.
          By inversion, $ \emptyset   \vdash  w  \ottsym{:}  \iota$.
          By Lemma \ref{lem:type_param_substitution},
          $[  \ottmv{X}  :=  \iota  ]  \ottsym{(}   \emptyset   \ottsym{)}  \vdash  [  \ottmv{X}  :=  \iota  ]  \ottsym{(}  w  \ottsym{)}  \ottsym{:}  [  \ottmv{X}  :=  \iota  ]  \ottsym{(}  \iota  \ottsym{)}$.
          Finally, $ \emptyset   \vdash  [  \ottmv{X}  :=  \iota  ]  \ottsym{(}  w  \ottsym{)}  \ottsym{:}  [  \ottmv{X}  :=  \iota  ]  \ottsym{(}  X  \ottsym{)}$.

          \case{\rnp{R\_InstArrow}}
          \leavevmode\\
          We are given $w  \ottsym{:}   \star  \!\rightarrow\!  \star \Rightarrow  \unskip ^ { \ell' }  \!  \star \Rightarrow  \unskip ^ { \ell }  \! \ottmv{X}   \,  \xrightarrow{ \mathmakebox[0.4em]{} [  \ottmv{X}  :=  \ottmv{X_{{\mathrm{1}}}}  \!\rightarrow\!  \ottmv{X_{{\mathrm{2}}}}  ] \mathmakebox[0.3em]{} }  \, w  \ottsym{:}   \star  \!\rightarrow\!  \star \Rightarrow  \unskip ^ { \ell' }  \!  \star \Rightarrow  \unskip ^ { \ell }  \!  \star  \!\rightarrow\!  \star \Rightarrow  \unskip ^ { \ell }  \! \ottmv{X_{{\mathrm{1}}}}  \!\rightarrow\!  \ottmv{X_{{\mathrm{2}}}}   $ where
          $\ottnt{f_{{\mathrm{1}}}}  \ottsym{=}  w  \ottsym{:}   \star  \!\rightarrow\!  \star \Rightarrow  \unskip ^ { \ell' }  \! \star $, $\ottnt{U}  \ottsym{=}  \ottmv{X}$, and $\ottnt{U'}  \ottsym{=}  \star$.
          By inversion, $ \emptyset   \vdash  \ottsym{(}  w  \ottsym{:}   \star  \!\rightarrow\!  \star \Rightarrow  \unskip ^ { \ell' }  \! \star   \ottsym{)}  \ottsym{:}  \star$.
          By \rnp{T\_Cast}, $ \emptyset   \vdash  \ottsym{(}  w  \ottsym{:}   \star  \!\rightarrow\!  \star \Rightarrow  \unskip ^ { \ell' }  \!  \star \Rightarrow  \unskip ^ { \ell }  \! \star  \!\rightarrow\!  \star    \ottsym{)}  \ottsym{:}  \star  \!\rightarrow\!  \star$.
          By \rnp{T\_Cast}, $ \emptyset   \vdash  \ottsym{(}  w  \ottsym{:}   \star  \!\rightarrow\!  \star \Rightarrow  \unskip ^ { \ell' }  \!  \star \Rightarrow  \unskip ^ { \ell }  \!  \star  \!\rightarrow\!  \star \Rightarrow  \unskip ^ { \ell }  \! \ottmv{X_{{\mathrm{1}}}}  \!\rightarrow\!  \ottmv{X_{{\mathrm{2}}}}     \ottsym{)}  \ottsym{:}  \ottmv{X_{{\mathrm{1}}}}  \!\rightarrow\!  \ottmv{X_{{\mathrm{2}}}}$.
          By Lemma \ref{lem:type_param_substitution},
          $[  \ottmv{X}  :=  \ottmv{X_{{\mathrm{1}}}}  \!\rightarrow\!  \ottmv{X_{{\mathrm{2}}}}  ]  \ottsym{(}   \emptyset   \ottsym{)}  \vdash  [  \ottmv{X}  :=  \ottmv{X_{{\mathrm{1}}}}  \!\rightarrow\!  \ottmv{X_{{\mathrm{2}}}}  ]  \ottsym{(}  w  \ottsym{:}   \star  \!\rightarrow\!  \star \Rightarrow  \unskip ^ { \ell' }  \!  \star \Rightarrow  \unskip ^ { \ell }  \!  \star  \!\rightarrow\!  \star \Rightarrow  \unskip ^ { \ell }  \! \ottmv{X_{{\mathrm{1}}}}  \!\rightarrow\!  \ottmv{X_{{\mathrm{2}}}}     \ottsym{)}  \ottsym{:}  [  \ottmv{X}  :=  \ottmv{X_{{\mathrm{1}}}}  \!\rightarrow\!  \ottmv{X_{{\mathrm{2}}}}  ]  \ottsym{(}  \ottmv{X_{{\mathrm{1}}}}  \!\rightarrow\!  \ottmv{X_{{\mathrm{2}}}}  \ottsym{)}$.
          Finally, $ \emptyset   \vdash  [  \ottmv{X}  :=  \ottmv{X_{{\mathrm{1}}}}  \!\rightarrow\!  \ottmv{X_{{\mathrm{2}}}}  ]  \ottsym{(}  w  \ottsym{:}   \star  \!\rightarrow\!  \star \Rightarrow  \unskip ^ { \ell' }  \!  \star \Rightarrow  \unskip ^ { \ell }  \!  \star  \!\rightarrow\!  \star \Rightarrow  \unskip ^ { \ell }  \! \ottmv{X_{{\mathrm{1}}}}  \!\rightarrow\!  \ottmv{X_{{\mathrm{2}}}}     \ottsym{)}  \ottsym{:}  [  \ottmv{X}  :=  \ottmv{X_{{\mathrm{1}}}}  \!\rightarrow\!  \ottmv{X_{{\mathrm{2}}}}  ]  \ottsym{(}  \ottmv{X}  \ottsym{)}$.

          \otherwise Cannot happen.
        \end{caseanalysis}

        \case{\rnp{T\_LetP}}
        We are given $ \emptyset   \vdash   \textsf{\textup{let}\relax} \,  \ottmv{x}  =   \Lambda    \overrightarrow{ \ottmv{X_{\ottmv{i}}} }  .\,  w_{{\mathrm{1}}}   \textsf{\textup{ in }\relax}  \ottnt{f_{{\mathrm{2}}}}   \ottsym{:}  \ottnt{U}$
        for some $\ottmv{x}$, $ \overrightarrow{ \ottmv{X_{\ottmv{i}}} } $, $w_{{\mathrm{1}}}$, and $\ottnt{f_{{\mathrm{2}}}}$.
        By case analysis on the reduction rules applicable to $ \textsf{\textup{let}\relax} \,  \ottmv{x}  =   \Lambda    \overrightarrow{ \ottmv{X_{\ottmv{i}}} }  .\,  w_{{\mathrm{1}}}   \textsf{\textup{ in }\relax}  \ottnt{f_{{\mathrm{2}}}} $.

        \begin{caseanalysis}
          \case {\rnp{R\_LetP}}
          \leavevmode\\
          We are given $ \textsf{\textup{let}\relax} \,  \ottmv{x}  =   \Lambda    \overrightarrow{ \ottmv{X_{\ottmv{i}}} }  .\,  w_{{\mathrm{1}}}   \textsf{\textup{ in }\relax}  \ottnt{f_{{\mathrm{2}}}}  \,  \xmapsto{ \mathmakebox[0.4em]{} [  ] \mathmakebox[0.3em]{} }  \, \ottnt{f_{{\mathrm{2}}}}  [  \ottmv{x}  \ottsym{:=}   \Lambda    \overrightarrow{ \ottmv{X_{\ottmv{i}}} }  .\,  w_{{\mathrm{1}}}   ]$
          where $S  \ottsym{=}  [  ]$.
          By inversion, we have
          $ \emptyset   \vdash  w_{{\mathrm{1}}}  \ottsym{:}  \ottnt{U_{{\mathrm{1}}}}$, $  \emptyset  ,   \ottmv{x}  :  \forall \,  \overrightarrow{ \ottmv{X_{\ottmv{i}}} }   \ottsym{.}  \ottnt{U_{{\mathrm{1}}}}    \vdash  \ottnt{f_{{\mathrm{2}}}}  \ottsym{:}  \ottnt{U}$, and $ \overrightarrow{ \ottmv{X_{\ottmv{i}}} }   \cap  \textit{ftv} \, \ottsym{(}  \Gamma  \ottsym{)}  \ottsym{=}   \emptyset $.
          By Lemma \ref{lem:term_var_substitution},
          we have $ \emptyset   \vdash  \ottnt{f_{{\mathrm{2}}}}  [  \ottmv{x}  \ottsym{:=}   \Lambda    \overrightarrow{ \ottmv{X_{\ottmv{i}}} }  .\,  w_{{\mathrm{1}}}   ]  \ottsym{:}  \ottnt{U}$.
          Finally, $ \emptyset   \vdash  [  ]  \ottsym{(}  \ottnt{f_{{\mathrm{2}}}}  [  \ottmv{x}  \ottsym{:=}   \Lambda    \overrightarrow{ \ottmv{X_{\ottmv{i}}} }  .\,  w_{{\mathrm{1}}}   ]  \ottsym{)}  \ottsym{:}  [  ]  \ottsym{(}  \ottnt{U}  \ottsym{)}$.

          \otherwise Cannot happen.
        \end{caseanalysis}

        \otherwise
        Cannot happen.
      \end{caseanalysis}

    \item
      By case analysis on the evaluation rule applied to $\ottnt{f}$.

      \begin{caseanalysis}
        \case{\rnp{E\_Step}}
        We are given $\ottnt{E}  [  \ottnt{f_{{\mathrm{1}}}}  ] \,  \xmapsto{ \mathmakebox[0.4em]{} S \mathmakebox[0.3em]{} }  \, S  \ottsym{(}  \ottnt{E}  [  \ottnt{f'_{{\mathrm{1}}}}  ]  \ottsym{)}$ and $\ottnt{f_{{\mathrm{1}}}} \,  \xrightarrow{ \mathmakebox[0.4em]{} S \mathmakebox[0.3em]{} }  \, \ottnt{f'_{{\mathrm{1}}}}$
        for some $S$, $\ottnt{f_{{\mathrm{1}}}}$, and $\ottnt{f'_{{\mathrm{1}}}}$
        where $\ottnt{f}  \ottsym{=}  \ottnt{E}  [  \ottnt{f_{{\mathrm{1}}}}  ]$ and $\ottnt{f'}  \ottsym{=}  S  \ottsym{(}  \ottnt{E}  [  \ottnt{f'_{{\mathrm{1}}}}  ]  \ottsym{)}$.
        By induction on the structure of $\ottnt{E}$.

        \begin{caseanalysis}
          \case{$\ottnt{E}  \ottsym{=}  \left[ \, \right]$}
          We are given $ \emptyset   \vdash  \ottnt{f_{{\mathrm{1}}}}  \ottsym{:}  \ottnt{U}$ and $\ottnt{f_{{\mathrm{1}}}} \,  \xmapsto{ \mathmakebox[0.4em]{} S \mathmakebox[0.3em]{} }  \, S  \ottsym{(}  \ottnt{f'_{{\mathrm{1}}}}  \ottsym{)}$.
          By Lemma \ref{lem:preservation}.1, $ \emptyset   \vdash  S  \ottsym{(}  \ottnt{f'_{{\mathrm{1}}}}  \ottsym{)}  \ottsym{:}  S  \ottsym{(}  \ottnt{U}  \ottsym{)}$.

          \case{$\ottnt{E}  \ottsym{=}  \mathit{op} \, \ottsym{(}  \ottnt{E'}  \ottsym{,}  \ottnt{f''}  \ottsym{)}$ for some $\ottnt{E'}$ and $\ottnt{f''}$}\ \\
          We are given $ \emptyset   \vdash  \mathit{op} \, \ottsym{(}  \ottnt{E'}  [  \ottnt{f_{{\mathrm{1}}}}  ]  \ottsym{,}  \ottnt{f''}  \ottsym{)}  \ottsym{:}  \ottnt{U}$ and
          $\mathit{op} \, \ottsym{(}  \ottnt{E'}  [  \ottnt{f_{{\mathrm{1}}}}  ]  \ottsym{,}  \ottnt{f''}  \ottsym{)} \,  \xmapsto{ \mathmakebox[0.4em]{} S \mathmakebox[0.3em]{} }  \, S  \ottsym{(}  \mathit{op} \, \ottsym{(}  \ottnt{E'}  [  \ottnt{f'_{{\mathrm{1}}}}  ]  \ottsym{,}  \ottnt{f''}  \ottsym{)}  \ottsym{)}$.
          By inversion, we have $ \mathit{ty} ( \mathit{op} )   \ottsym{=}  \iota_{{\mathrm{1}}}  \!\rightarrow\!  \iota_{{\mathrm{2}}}  \!\rightarrow\!  \iota$,
          $ \emptyset   \vdash  \ottnt{E'}  [  \ottnt{f_{{\mathrm{1}}}}  ]  \ottsym{:}  \iota_{{\mathrm{1}}}$ and
          $ \emptyset   \vdash  \ottnt{f''}  \ottsym{:}  \iota_{{\mathrm{2}}}$ for some $\iota_{{\mathrm{1}}}$, $\iota_{{\mathrm{2}}}$, and $\iota$
          where $\ottnt{U}  \ottsym{=}  \iota$.
          By \rnp{E\_Step}, $\ottnt{E'}  [  \ottnt{f_{{\mathrm{1}}}}  ] \,  \xmapsto{ \mathmakebox[0.4em]{} S \mathmakebox[0.3em]{} }  \, S  \ottsym{(}  \ottnt{E'}  [  \ottnt{f'_{{\mathrm{1}}}}  ]  \ottsym{)}$.
          By the IH, $ \emptyset   \vdash  S  \ottsym{(}  \ottnt{E'}  [  \ottnt{f'_{{\mathrm{1}}}}  ]  \ottsym{)}  \ottsym{:}  S  \ottsym{(}  \iota_{{\mathrm{1}}}  \ottsym{)}$.
          By Lemma \ref{lem:type_param_substitution}, $ \emptyset   \vdash  S  \ottsym{(}  \ottnt{f''}  \ottsym{)}  \ottsym{:}  S  \ottsym{(}  \iota_{{\mathrm{2}}}  \ottsym{)}$.
          By definition, $S  \ottsym{(}  \iota_{{\mathrm{1}}}  \ottsym{)}  \ottsym{=}  \iota_{{\mathrm{1}}}$, $S  \ottsym{(}  \iota_{{\mathrm{2}}}  \ottsym{)}  \ottsym{=}  \iota_{{\mathrm{2}}}$,
          and $S  \ottsym{(}  \iota  \ottsym{)}  \ottsym{=}  \iota$.
          By \rnp{T\_Op}, $ \emptyset   \vdash  \mathit{op} \, \ottsym{(}  S  \ottsym{(}  \ottnt{E'}  [  \ottnt{f'_{{\mathrm{1}}}}  ]  \ottsym{)}  \ottsym{,}  S  \ottsym{(}  \ottnt{f''}  \ottsym{)}  \ottsym{)}  \ottsym{:}  \iota$.
          By definition, $ \emptyset   \vdash  S  \ottsym{(}  \mathit{op} \, \ottsym{(}  \ottnt{E'}  [  \ottnt{f'_{{\mathrm{1}}}}  ]  \ottsym{,}  \ottnt{f''}  \ottsym{)}  \ottsym{)}  \ottsym{:}  S  \ottsym{(}  \iota  \ottsym{)}$.

          \case{$\ottnt{E}  \ottsym{=}  \mathit{op} \, \ottsym{(}  \ottnt{f''}  \ottsym{,}  \ottnt{E'}  \ottsym{)}$ for some $\ottnt{E'}$ and $\ottnt{f''}$}
          Similarly to the previous case.

          \case{$\ottnt{E}  \ottsym{=}  \ottnt{E'} \, \ottnt{f''}$ for some $\ottnt{E'}$ and $\ottnt{f''}$}
          We are given $ \emptyset   \vdash  \ottnt{E'}  [  \ottnt{f_{{\mathrm{1}}}}  ] \, \ottnt{f''}  \ottsym{:}  \ottnt{U}$ and
          $\ottnt{E'}  [  \ottnt{f_{{\mathrm{1}}}}  ] \, \ottnt{f''} \,  \xmapsto{ \mathmakebox[0.4em]{} S \mathmakebox[0.3em]{} }  \, S  \ottsym{(}  \ottnt{E'}  [  \ottnt{f'_{{\mathrm{1}}}}  ] \, \ottnt{f''}  \ottsym{)}$.
          By inversion, we have $ \emptyset   \vdash  \ottnt{E'}  [  \ottnt{f_{{\mathrm{1}}}}  ]  \ottsym{:}  \ottnt{U'}  \!\rightarrow\!  \ottnt{U}$ and
          $ \emptyset   \vdash  \ottnt{f''}  \ottsym{:}  \ottnt{U'}$ for some $\ottnt{U'}$.
          By \rnp{E\_Step}, $\ottnt{E'}  [  \ottnt{f_{{\mathrm{1}}}}  ] \,  \xmapsto{ \mathmakebox[0.4em]{} S \mathmakebox[0.3em]{} }  \, S  \ottsym{(}  \ottnt{E'}  [  \ottnt{f'_{{\mathrm{1}}}}  ]  \ottsym{)}$.
          By the IH, $ \emptyset   \vdash  S  \ottsym{(}  \ottnt{E'}  [  \ottnt{f'_{{\mathrm{1}}}}  ]  \ottsym{)}  \ottsym{:}  S  \ottsym{(}  \ottnt{U'}  \!\rightarrow\!  \ottnt{U}  \ottsym{)}$.
          By Lemma \ref{lem:type_param_substitution}, $ \emptyset   \vdash  S  \ottsym{(}  \ottnt{f''}  \ottsym{)}  \ottsym{:}  S  \ottsym{(}  \ottnt{U'}  \ottsym{)}$.
          By \rnp{T\_App}, $ \emptyset   \vdash  S  \ottsym{(}  \ottnt{E'}  [  \ottnt{f'_{{\mathrm{1}}}}  ] \, \ottnt{f''}  \ottsym{)}  \ottsym{:}  S  \ottsym{(}  \ottnt{U}  \ottsym{)}$.

          \case{$\ottnt{E}  \ottsym{=}  w \, \ottnt{E'}$ for some $\ottnt{E'}$ and $w$}\ \\
          We are given $ \emptyset   \vdash  w \, \ottnt{E'}  [  \ottnt{f_{{\mathrm{1}}}}  ]  \ottsym{:}  \ottnt{U}$ and
          $w \, \ottnt{E'}  [  \ottnt{f_{{\mathrm{1}}}}  ] \,  \xmapsto{ \mathmakebox[0.4em]{} S \mathmakebox[0.3em]{} }  \, S  \ottsym{(}  w \, \ottnt{E'}  [  \ottnt{f'_{{\mathrm{1}}}}  ]  \ottsym{)}$.
          By inversion, we have $ \emptyset   \vdash  w  \ottsym{:}  \ottnt{U'}  \!\rightarrow\!  \ottnt{U}$ and
          $ \emptyset   \vdash  \ottnt{E'}  [  \ottnt{f_{{\mathrm{1}}}}  ]  \ottsym{:}  \ottnt{U'}$ for some $\ottnt{U'}$.
          By \rnp{E\_Step}, $\ottnt{E'}  [  \ottnt{f_{{\mathrm{1}}}}  ] \,  \xmapsto{ \mathmakebox[0.4em]{} S \mathmakebox[0.3em]{} }  \, S  \ottsym{(}  \ottnt{E'}  [  \ottnt{f'_{{\mathrm{1}}}}  ]  \ottsym{)}$.
          By the IH, $ \emptyset   \vdash  S  \ottsym{(}  \ottnt{E'}  [  \ottnt{f'_{{\mathrm{1}}}}  ]  \ottsym{)}  \ottsym{:}  S  \ottsym{(}  \ottnt{U'}  \ottsym{)}$.
          By Lemma \ref{lem:type_param_substitution},
          $ \emptyset   \vdash  S  \ottsym{(}  w  \ottsym{)}  \ottsym{:}  S  \ottsym{(}  \ottnt{U'}  \!\rightarrow\!  \ottnt{U}  \ottsym{)}$.
          By \rnp{T\_App}, $ \emptyset   \vdash  S  \ottsym{(}  w \, \ottnt{E'}  [  \ottnt{f'_{{\mathrm{1}}}}  ]  \ottsym{)}  \ottsym{:}  S  \ottsym{(}  \ottnt{U}  \ottsym{)}$.

          \case{$\ottnt{E}  \ottsym{=}  \ottnt{E'}  \ottsym{:}  \ottnt{U'}  \Rightarrow   \unskip ^ { \ell }  \, \ottnt{U}$ for some $\ottnt{E'}$, $\ottnt{U'}$, and $\ell$}
          \leavevmode\\
          We are given $ \emptyset   \vdash  \ottsym{(}  \ottnt{E'}  [  \ottnt{f_{{\mathrm{1}}}}  ]  \ottsym{:}   \ottnt{U'} \Rightarrow  \unskip ^ { \ell }  \! \ottnt{U}   \ottsym{)}  \ottsym{:}  \ottnt{U}$ and
          $\ottnt{E'}  [  \ottnt{f_{{\mathrm{1}}}}  ]  \ottsym{:}   \ottnt{U'} \Rightarrow  \unskip ^ { \ell }  \! \ottnt{U}  \,  \xmapsto{ \mathmakebox[0.4em]{} S \mathmakebox[0.3em]{} }  \, S  \ottsym{(}  \ottnt{E'}  [  \ottnt{f'_{{\mathrm{1}}}}  ]  \ottsym{:}   \ottnt{U'} \Rightarrow  \unskip ^ { \ell }  \! \ottnt{U}   \ottsym{)}$.
          By inversion, we have $ \emptyset   \vdash  \ottnt{E'}  [  \ottnt{f_{{\mathrm{1}}}}  ]  \ottsym{:}  \ottnt{U'}$ and $\ottnt{U'}  \sim  \ottnt{U}$.
          By \rnp{E\_Step}, $\ottnt{E'}  [  \ottnt{f_{{\mathrm{1}}}}  ] \,  \xmapsto{ \mathmakebox[0.4em]{} S \mathmakebox[0.3em]{} }  \, S  \ottsym{(}  \ottnt{E'}  [  \ottnt{f'_{{\mathrm{1}}}}  ]  \ottsym{)}$.
          By the IH, $ \emptyset   \vdash  S  \ottsym{(}  \ottnt{E'}  [  \ottnt{f'_{{\mathrm{1}}}}  ]  \ottsym{)}  \ottsym{:}  S  \ottsym{(}  \ottnt{U'}  \ottsym{)}$.
          By Lemma \ref{lem:type_param_substitution_in_cosistency}, $S  \ottsym{(}  \ottnt{U'}  \ottsym{)}  \sim  S  \ottsym{(}  \ottnt{U}  \ottsym{)}$.
          By \rnp{T\_Cast}, $ \emptyset   \vdash  S  \ottsym{(}  \ottnt{E'}  [  \ottnt{f'_{{\mathrm{1}}}}  ]  \ottsym{:}   \ottnt{U'} \Rightarrow  \unskip ^ { \ell }  \! \ottnt{U}   \ottsym{)}  \ottsym{:}  S  \ottsym{(}  \ottnt{U}  \ottsym{)}$.
        \end{caseanalysis}

        \case{\rnp{E\_Abort}}
        We are given $\ottnt{f} \,  \xmapsto{ \mathmakebox[0.4em]{} S \mathmakebox[0.3em]{} }  \, \textsf{\textup{blame}\relax} \, \ell$ for some $\ell$
        where $\ottnt{f'}  \ottsym{=}  \textsf{\textup{blame}\relax} \, \ell$.
        Finish by \rnp{T\_Blame}. \qedhere
      \end{caseanalysis}
  \end{enumerate}
\end{proof}

\ifrestate
\thmTypeSafety*
\else
\begin{theoremA}[name=Type Safety,restate=thmTypeSafety] \label{thm:type_safety}
  If $ \emptyset   \vdash  \ottnt{f}  \ottsym{:}  \ottnt{U}$,
  then one of the following holds:
  \begin{itemize}
    \item $\ottnt{f} \,  \xmapsto{ \mathmakebox[0.4em]{} S \mathmakebox[0.3em]{} }\hspace{-0.4em}{}^\ast \hspace{0.2em}  \, \ottnt{r}$ for some $S$ and $\ottnt{r}$ such that $ \emptyset   \vdash  \ottnt{r}  \ottsym{:}  S  \ottsym{(}  \ottnt{U}  \ottsym{)}$; or
    \item $ \ottnt{f} \!  \Uparrow  $.
  \end{itemize}
\end{theoremA}
\fi 

\begin{proof}
  By Lemmas \ref{lem:progress} and \ref{lem:preservation}.
\end{proof}

\subsection{Conservative Extension}

\ifrestate
\thmConservativeExtension*
\else
\begin{theoremA}[name=Conservative Extension,restate=thmConservativeExtension] \label{thm:conservative_extension}
  Suppose $\textit{ftv} \, \ottsym{(}  \ottnt{f}  \ottsym{)}  \ottsym{=}   \emptyset $ and $\ottnt{f}$ does not contain $ \nu $.

  \begin{enumerate}
    \item $\ottnt{f} \,  \xmapsto{ \mathmakebox[0.4em]{} [  ] \mathmakebox[0.3em]{} }\hspace{-0.4em}{}^\ast \hspace{0.2em}  \, \ottnt{r}$ if and only if $\ottnt{f} \, \longmapsto_{\textsf{\textup{B}\relax}\relax}^\ast \, \ottnt{r}$.
    \item $ \ottnt{f} \!  \Uparrow  $ if and only if $ \ottnt{f} \!  \Uparrow _{\textsf{\textup{B}\relax}\relax}  $.
  \end{enumerate}
\end{theoremA}
\fi

\begin{proof}
Easy.
\end{proof}

\subsection{Divergence}
In this section, we do not care about blame labels for ease of proof.
All blame labels are written with $\ell$.

\begin{definitionA}[Auxiliary Types]
  \begin{align*}
    U^X & ::= \ottmv{X} \mid U^X  \!\rightarrow\!  U^X \\
    U^\iota & ::= \iota \mid U^\iota  \!\rightarrow\!  U^\iota \\
  \end{align*}
\end{definitionA}

\begin{definitionA}[Order of Type] \leavevmode
  \begin{itemize}
    \item $\textit{ord} \, \ottsym{(}  \star  \ottsym{)} = 0$
    \item $\textit{ord} \, \ottsym{(}  \ottmv{X}  \ottsym{)} = 0$
    \item $\textit{ord} \, \ottsym{(}  \iota  \ottsym{)} = 0$
    \item $\textit{ord} \, \ottsym{(}  \ottnt{U_{{\mathrm{1}}}}  \!\rightarrow\!  \ottnt{U_{{\mathrm{2}}}}  \ottsym{)} = \textit{ord} \, \ottsym{(}  \ottnt{U_{{\mathrm{1}}}}  \ottsym{)} + 1$
  \end{itemize}
\end{definitionA}

\begin{lemmaA} \label{lem:omega_translate_to_normal_form}
  If
  \begin{itemize}
    \item $\ottnt{f_{{\mathrm{1}}}}  \ottsym{=}  \ottsym{(}   \lambda  \ottmv{x} \!:\!  X  .\,  \ottsym{(}  \ottmv{x}  \ottsym{:}   X \Rightarrow  \unskip ^ { \ell }  \!  \star \Rightarrow  \unskip ^ { \ell }  \! \star  \!\rightarrow\!  \star    \ottsym{)}  \, \ottsym{(}  \ottmv{x}  \ottsym{:}   X \Rightarrow  \unskip ^ { \ell }  \! \star   \ottsym{)}  \ottsym{)}  \ottsym{:}   X  \!\rightarrow\!  \star \Rightarrow  \unskip ^ { \ell }  \! \star  \!\rightarrow\!  \star $, and
    \item $\ottnt{f_{{\mathrm{2}}}}  \ottsym{=}  \ottsym{(}   \lambda  \ottmv{x} \!:\!  \star  .\,  \ottsym{(}  \ottmv{x}  \ottsym{:}   \star \Rightarrow  \unskip ^ { \ell }  \! \star  \!\rightarrow\!  \star   \ottsym{)}  \, \ottmv{x}  \ottsym{)}  \ottsym{:}   \star  \!\rightarrow\!  \star \Rightarrow  \unskip ^ { \ell }  \! \star $,
  \end{itemize}
  then there exist $\ottnt{f'_{{\mathrm{1}}}}$ and $\ottnt{f'_{{\mathrm{2}}}}$ such that
  \begin{itemize}
    \item $\ottnt{f'_{{\mathrm{1}}}}  \ottsym{=}  \ottsym{(}   \lambda  \ottmv{x} \!:\!  \star  .\,  \ottsym{(}  \ottmv{x}  \ottsym{:}   \star \Rightarrow  \unskip ^ { \ell }  \! \star  \!\rightarrow\!  \star   \ottsym{)}  \, \ottmv{x}  \ottsym{)}  \ottsym{:}   \star  \!\rightarrow\!  \star \Rightarrow  \unskip ^ { \ell }  \!  X_{{\mathrm{1}}}  \!\rightarrow\!  X_{{\mathrm{2}}} \Rightarrow  \unskip ^ { \ell }  \! \star  \!\rightarrow\!  \star  $,
    \item $\ottnt{f'_{{\mathrm{2}}}}  \ottsym{=}  \ottsym{(}   \lambda  \ottmv{x} \!:\!  \star  .\,  \ottsym{(}  \ottmv{x}  \ottsym{:}   \star \Rightarrow  \unskip ^ { \ell }  \! \star  \!\rightarrow\!  \star   \ottsym{)}  \, \ottmv{x}  \ottsym{)}  \ottsym{:}   \star  \!\rightarrow\!  \star \Rightarrow  \unskip ^ { \ell }  \!  X_{{\mathrm{1}}}  \!\rightarrow\!  X_{{\mathrm{2}}} \Rightarrow  \unskip ^ { \ell }  \!  \star  \!\rightarrow\!  \star \Rightarrow  \unskip ^ { \ell }  \! \star   $, and
    \item $\ottnt{f_{{\mathrm{1}}}} \, \ottnt{f_{{\mathrm{2}}}} \,  \xmapsto{ \mathmakebox[0.4em]{} [  X  :=  X_{{\mathrm{1}}}  \!\rightarrow\!  X_{{\mathrm{2}}}  ] \mathmakebox[0.3em]{} }\hspace{-0.4em}{}^\ast \hspace{0.2em}  \, \ottsym{(}  \ottnt{f'_{{\mathrm{1}}}} \, \ottnt{f'_{{\mathrm{2}}}}  \ottsym{)}  \ottsym{:}   \star \Rightarrow  \unskip ^ { \ell }  \! \star $.
  \end{itemize}
\end{lemmaA}

\begin{proof}
  \leavevmode
  \begin{enumerate}
    \item
      We are given,
      \begin{itemize}
        \item $\ottnt{f_{{\mathrm{1}}}}  \ottsym{=}  \ottsym{(}   \lambda  \ottmv{x} \!:\!  X  .\,  \ottsym{(}  \ottmv{x}  \ottsym{:}   X \Rightarrow  \unskip ^ { \ell }  \!  \star \Rightarrow  \unskip ^ { \ell }  \! \star  \!\rightarrow\!  \star    \ottsym{)}  \, \ottsym{(}  \ottmv{x}  \ottsym{:}   X \Rightarrow  \unskip ^ { \ell }  \! \star   \ottsym{)}  \ottsym{)}  \ottsym{:}   X  \!\rightarrow\!  \star \Rightarrow  \unskip ^ { \ell }  \! \star  \!\rightarrow\!  \star $ and
        \item $\ottnt{f_{{\mathrm{2}}}}  \ottsym{=}  \ottsym{(}   \lambda  \ottmv{x} \!:\!  \star  .\,  \ottsym{(}  \ottmv{x}  \ottsym{:}   \star \Rightarrow  \unskip ^ { \ell }  \! \star  \!\rightarrow\!  \star   \ottsym{)}  \, \ottmv{x}  \ottsym{)}  \ottsym{:}   \star  \!\rightarrow\!  \star \Rightarrow  \unskip ^ { \ell }  \! \star $.
      \end{itemize}
    \item
      By \rnp{R\_AppCast} and \rnp{E\_Step},
      there exist $\ottnt{f_{{\mathrm{11}}}}$ and $\ottnt{f_{{\mathrm{12}}}}$ such that
      \begin{itemize}
        \item $\ottnt{f_{{\mathrm{11}}}}  \ottsym{=}   \lambda  \ottmv{x} \!:\!  X  .\,  \ottsym{(}  \ottmv{x}  \ottsym{:}   X \Rightarrow  \unskip ^ { \ell }  \!  \star \Rightarrow  \unskip ^ { \ell }  \! \star  \!\rightarrow\!  \star    \ottsym{)}  \, \ottsym{(}  \ottmv{x}  \ottsym{:}   X \Rightarrow  \unskip ^ { \ell }  \! \star   \ottsym{)}$,
        \item $\ottnt{f_{{\mathrm{12}}}}  \ottsym{=}  \ottsym{(}   \lambda  \ottmv{x} \!:\!  \star  .\,  \ottsym{(}  \ottmv{x}  \ottsym{:}   \star \Rightarrow  \unskip ^ { \ell }  \! \star  \!\rightarrow\!  \star   \ottsym{)}  \, \ottmv{x}  \ottsym{)}  \ottsym{:}   \star  \!\rightarrow\!  \star \Rightarrow  \unskip ^ { \ell }  \!  \star \Rightarrow  \unskip ^ { \ell }  \! X  $, and
        \item $\ottnt{f_{{\mathrm{1}}}} \, \ottnt{f_{{\mathrm{2}}}} \,  \xmapsto{ \mathmakebox[0.4em]{} [  ] \mathmakebox[0.3em]{} }  \, \ottsym{(}  \ottnt{f_{{\mathrm{11}}}} \, \ottnt{f_{{\mathrm{12}}}}  \ottsym{)}  \ottsym{:}   \star \Rightarrow  \unskip ^ { \ell }  \! \star $
      \end{itemize}
    \item
      By \rnp{R\_InstArrow} and \rnp{E\_Step},
      there exist $\ottnt{f_{{\mathrm{21}}}}$ and $\ottnt{f_{{\mathrm{22}}}}$ such that
      \begin{itemize}
        \item $\ottnt{f_{{\mathrm{21}}}}  \ottsym{=}   \lambda  \ottmv{x} \!:\!  X_{{\mathrm{1}}}  \!\rightarrow\!  X_{{\mathrm{2}}}  .\,  \ottsym{(}  \ottmv{x}  \ottsym{:}   X_{{\mathrm{1}}}  \!\rightarrow\!  X_{{\mathrm{2}}} \Rightarrow  \unskip ^ { \ell }  \!  \star \Rightarrow  \unskip ^ { \ell }  \! \star  \!\rightarrow\!  \star    \ottsym{)}  \, \ottsym{(}  \ottmv{x}  \ottsym{:}   X_{{\mathrm{1}}}  \!\rightarrow\!  X_{{\mathrm{2}}} \Rightarrow  \unskip ^ { \ell }  \! \star   \ottsym{)}$,
        \item $\ottnt{f_{{\mathrm{22}}}}  \ottsym{=}  \ottsym{(}   \lambda  \ottmv{x} \!:\!  \star  .\,  \ottsym{(}  \ottmv{x}  \ottsym{:}   \star \Rightarrow  \unskip ^ { \ell }  \! \star  \!\rightarrow\!  \star   \ottsym{)}  \, \ottmv{x}  \ottsym{)}  \ottsym{:}   \star  \!\rightarrow\!  \star \Rightarrow  \unskip ^ { \ell }  \!  \star \Rightarrow  \unskip ^ { \ell }  \!  \star  \!\rightarrow\!  \star \Rightarrow  \unskip ^ { \ell }  \! X_{{\mathrm{1}}}  \!\rightarrow\!  X_{{\mathrm{2}}}   $, and
        \item $\ottsym{(}  \ottnt{f_{{\mathrm{11}}}} \, \ottnt{f_{{\mathrm{12}}}}  \ottsym{)}  \ottsym{:}   \star \Rightarrow  \unskip ^ { \ell }  \! \star  \,  \xmapsto{ \mathmakebox[0.4em]{} [  X  :=  X_{{\mathrm{1}}}  \!\rightarrow\!  X_{{\mathrm{2}}}  ] \mathmakebox[0.3em]{} }  \, \ottsym{(}  \ottnt{f_{{\mathrm{21}}}} \, \ottnt{f_{{\mathrm{22}}}}  \ottsym{)}  \ottsym{:}   \star \Rightarrow  \unskip ^ { \ell }  \! \star $
      \end{itemize}
    \item
      By \rnp{R\_Succeed} and \rnp{E\_Step},
      there exist $\ottnt{f_{{\mathrm{31}}}}$ and $\ottnt{f_{{\mathrm{32}}}}$ such that
      \begin{itemize}
        \item $\ottnt{f_{{\mathrm{31}}}}  \ottsym{=}   \lambda  \ottmv{x} \!:\!  X_{{\mathrm{1}}}  \!\rightarrow\!  X_{{\mathrm{2}}}  .\,  \ottsym{(}  \ottmv{x}  \ottsym{:}   X_{{\mathrm{1}}}  \!\rightarrow\!  X_{{\mathrm{2}}} \Rightarrow  \unskip ^ { \ell }  \!  \star \Rightarrow  \unskip ^ { \ell }  \! \star  \!\rightarrow\!  \star    \ottsym{)}  \, \ottsym{(}  \ottmv{x}  \ottsym{:}   X_{{\mathrm{1}}}  \!\rightarrow\!  X_{{\mathrm{2}}} \Rightarrow  \unskip ^ { \ell }  \! \star   \ottsym{)}$,
        \item $\ottnt{f_{{\mathrm{32}}}}  \ottsym{=}  \ottsym{(}   \lambda  \ottmv{x} \!:\!  \star  .\,  \ottsym{(}  \ottmv{x}  \ottsym{:}   \star \Rightarrow  \unskip ^ { \ell }  \! \star  \!\rightarrow\!  \star   \ottsym{)}  \, \ottmv{x}  \ottsym{)}  \ottsym{:}   \star  \!\rightarrow\!  \star \Rightarrow  \unskip ^ { \ell }  \! X_{{\mathrm{1}}}  \!\rightarrow\!  X_{{\mathrm{2}}} $, and
        \item $\ottsym{(}  \ottnt{f_{{\mathrm{21}}}} \, \ottnt{f_{{\mathrm{22}}}}  \ottsym{)}  \ottsym{:}   \star \Rightarrow  \unskip ^ { \ell }  \! \star  \,  \xmapsto{ \mathmakebox[0.4em]{} [  ] \mathmakebox[0.3em]{} }  \, \ottsym{(}  \ottnt{f_{{\mathrm{31}}}} \, \ottnt{f_{{\mathrm{32}}}}  \ottsym{)}  \ottsym{:}   \star \Rightarrow  \unskip ^ { \ell }  \! \star $.
      \end{itemize}
    \item
      By \rnp{R\_Beta} and \rnp{E\_Step},
      there exist $\ottnt{f_{{\mathrm{41}}}}$ and $\ottnt{f_{{\mathrm{42}}}}$ such that
      \begin{itemize}
        \item $\ottnt{f_{{\mathrm{41}}}}  \ottsym{=}  \ottsym{(}   \lambda  \ottmv{x} \!:\!  \star  .\,  \ottsym{(}  \ottmv{x}  \ottsym{:}   \star \Rightarrow  \unskip ^ { \ell }  \! \star  \!\rightarrow\!  \star   \ottsym{)}  \, \ottmv{x}  \ottsym{)}  \ottsym{:}   \star  \!\rightarrow\!  \star \Rightarrow  \unskip ^ { \ell }  \!  X_{{\mathrm{1}}}  \!\rightarrow\!  X_{{\mathrm{2}}} \Rightarrow  \unskip ^ { \ell }  \!  \star \Rightarrow  \unskip ^ { \ell }  \! \star  \!\rightarrow\!  \star   $,
        \item $\ottnt{f_{{\mathrm{42}}}}  \ottsym{=}  \ottsym{(}   \lambda  \ottmv{x} \!:\!  \star  .\,  \ottsym{(}  \ottmv{x}  \ottsym{:}   \star \Rightarrow  \unskip ^ { \ell }  \! \star  \!\rightarrow\!  \star   \ottsym{)}  \, \ottmv{x}  \ottsym{)}  \ottsym{:}   \star  \!\rightarrow\!  \star \Rightarrow  \unskip ^ { \ell }  \!  X_{{\mathrm{1}}}  \!\rightarrow\!  X_{{\mathrm{2}}} \Rightarrow  \unskip ^ { \ell }  \! \star  $, and
        \item $\ottsym{(}  \ottnt{f_{{\mathrm{31}}}} \, \ottnt{f_{{\mathrm{32}}}}  \ottsym{)}  \ottsym{:}   \star \Rightarrow  \unskip ^ { \ell }  \! \star  \,  \xmapsto{ \mathmakebox[0.4em]{} [  ] \mathmakebox[0.3em]{} }  \, \ottsym{(}  \ottnt{f_{{\mathrm{41}}}} \, \ottnt{f_{{\mathrm{42}}}}  \ottsym{)}  \ottsym{:}   \star \Rightarrow  \unskip ^ { \ell }  \! \star $.
      \end{itemize}
    \item
      By \rnp{R\_Ground} and \rnp{E\_Step},
      there exist $\ottnt{f_{{\mathrm{51}}}}$ and $\ottnt{f_{{\mathrm{52}}}}$ such that
      \begin{itemize}
        \item $\ottnt{f_{{\mathrm{51}}}}  \ottsym{=}  \ottsym{(}   \lambda  \ottmv{x} \!:\!  \star  .\,  \ottsym{(}  \ottmv{x}  \ottsym{:}   \star \Rightarrow  \unskip ^ { \ell }  \! \star  \!\rightarrow\!  \star   \ottsym{)}  \, \ottmv{x}  \ottsym{)}  \ottsym{:}   \star  \!\rightarrow\!  \star \Rightarrow  \unskip ^ { \ell }  \!  X_{{\mathrm{1}}}  \!\rightarrow\!  X_{{\mathrm{2}}} \Rightarrow  \unskip ^ { \ell }  \!  \star  \!\rightarrow\!  \star \Rightarrow  \unskip ^ { \ell }  \!  \star \Rightarrow  \unskip ^ { \ell }  \! \star  \!\rightarrow\!  \star    $,
        \item $\ottnt{f_{{\mathrm{52}}}}  \ottsym{=}  \ottsym{(}   \lambda  \ottmv{x} \!:\!  \star  .\,  \ottsym{(}  \ottmv{x}  \ottsym{:}   \star \Rightarrow  \unskip ^ { \ell }  \! \star  \!\rightarrow\!  \star   \ottsym{)}  \, \ottmv{x}  \ottsym{)}  \ottsym{:}   \star  \!\rightarrow\!  \star \Rightarrow  \unskip ^ { \ell }  \!  X_{{\mathrm{1}}}  \!\rightarrow\!  X_{{\mathrm{2}}} \Rightarrow  \unskip ^ { \ell }  \! \star  $, and
        \item $\ottsym{(}  \ottnt{f_{{\mathrm{41}}}} \, \ottnt{f_{{\mathrm{42}}}}  \ottsym{)}  \ottsym{:}   \star \Rightarrow  \unskip ^ { \ell }  \! \star  \,  \xmapsto{ \mathmakebox[0.4em]{} [  ] \mathmakebox[0.3em]{} }  \, \ottsym{(}  \ottnt{f_{{\mathrm{51}}}} \, \ottnt{f_{{\mathrm{52}}}}  \ottsym{)}  \ottsym{:}   \star \Rightarrow  \unskip ^ { \ell }  \! \star $.
      \end{itemize}
    \item
      By \rnp{R\_Succeed} and \rnp{E\_Step},
      there exist $\ottnt{f_{{\mathrm{61}}}}$ and $\ottnt{f_{{\mathrm{62}}}}$ such that
      \begin{itemize}
        \item $\ottnt{f_{{\mathrm{61}}}}  \ottsym{=}  \ottsym{(}   \lambda  \ottmv{x} \!:\!  \star  .\,  \ottsym{(}  \ottmv{x}  \ottsym{:}   \star \Rightarrow  \unskip ^ { \ell }  \! \star  \!\rightarrow\!  \star   \ottsym{)}  \, \ottmv{x}  \ottsym{)}  \ottsym{:}   \star  \!\rightarrow\!  \star \Rightarrow  \unskip ^ { \ell }  \!  X_{{\mathrm{1}}}  \!\rightarrow\!  X_{{\mathrm{2}}} \Rightarrow  \unskip ^ { \ell }  \! \star  \!\rightarrow\!  \star  $,
        \item $\ottnt{f_{{\mathrm{62}}}}  \ottsym{=}  \ottsym{(}   \lambda  \ottmv{x} \!:\!  \star  .\,  \ottsym{(}  \ottmv{x}  \ottsym{:}   \star \Rightarrow  \unskip ^ { \ell }  \! \star  \!\rightarrow\!  \star   \ottsym{)}  \, \ottmv{x}  \ottsym{)}  \ottsym{:}   \star  \!\rightarrow\!  \star \Rightarrow  \unskip ^ { \ell }  \!  X_{{\mathrm{1}}}  \!\rightarrow\!  X_{{\mathrm{2}}} \Rightarrow  \unskip ^ { \ell }  \! \star  $, and
        \item $\ottsym{(}  \ottnt{f_{{\mathrm{51}}}} \, \ottnt{f_{{\mathrm{52}}}}  \ottsym{)}  \ottsym{:}   \star \Rightarrow  \unskip ^ { \ell }  \! \star  \,  \xmapsto{ \mathmakebox[0.4em]{} [  ] \mathmakebox[0.3em]{} }  \, \ottsym{(}  \ottnt{f_{{\mathrm{61}}}} \, \ottnt{f_{{\mathrm{62}}}}  \ottsym{)}  \ottsym{:}   \star \Rightarrow  \unskip ^ { \ell }  \! \star $.
      \end{itemize}
    \item
      By \rnp{R\_Ground} and \rnp{E\_Step},
      there exist $\ottnt{f'_{{\mathrm{1}}}}$ and $\ottnt{f'_{{\mathrm{2}}}}$ such that
      \begin{itemize}
        \item $\ottnt{f'_{{\mathrm{1}}}}  \ottsym{=}  \ottsym{(}   \lambda  \ottmv{x} \!:\!  \star  .\,  \ottsym{(}  \ottmv{x}  \ottsym{:}   \star \Rightarrow  \unskip ^ { \ell }  \! \star  \!\rightarrow\!  \star   \ottsym{)}  \, \ottmv{x}  \ottsym{)}  \ottsym{:}   \star  \!\rightarrow\!  \star \Rightarrow  \unskip ^ { \ell }  \!  X_{{\mathrm{1}}}  \!\rightarrow\!  X_{{\mathrm{2}}} \Rightarrow  \unskip ^ { \ell }  \! \star  \!\rightarrow\!  \star  $,
        \item $\ottnt{f'_{{\mathrm{2}}}}  \ottsym{=}  \ottsym{(}   \lambda  \ottmv{x} \!:\!  \star  .\,  \ottsym{(}  \ottmv{x}  \ottsym{:}   \star \Rightarrow  \unskip ^ { \ell }  \! \star  \!\rightarrow\!  \star   \ottsym{)}  \, \ottmv{x}  \ottsym{)}  \ottsym{:}   \star  \!\rightarrow\!  \star \Rightarrow  \unskip ^ { \ell }  \!  X_{{\mathrm{1}}}  \!\rightarrow\!  X_{{\mathrm{2}}} \Rightarrow  \unskip ^ { \ell }  \!  \star  \!\rightarrow\!  \star \Rightarrow  \unskip ^ { \ell }  \! \star   $, and
        \item $\ottsym{(}  \ottnt{f_{{\mathrm{61}}}} \, \ottnt{f_{{\mathrm{62}}}}  \ottsym{)}  \ottsym{:}   \star \Rightarrow  \unskip ^ { \ell }  \! \star  \,  \xmapsto{ \mathmakebox[0.4em]{} [  ] \mathmakebox[0.3em]{} }  \, \ottsym{(}  \ottnt{f'_{{\mathrm{1}}}} \, \ottnt{f'_{{\mathrm{2}}}}  \ottsym{)}  \ottsym{:}   \star \Rightarrow  \unskip ^ { \ell }  \! \star $.
      \end{itemize}
  \end{enumerate}

  Finally, $\ottnt{f_{{\mathrm{1}}}} \, \ottnt{f_{{\mathrm{2}}}} \,  \xmapsto{ \mathmakebox[0.4em]{} [  X  :=  X_{{\mathrm{1}}}  \!\rightarrow\!  X_{{\mathrm{2}}}  ] \mathmakebox[0.3em]{} }\hspace{-0.4em}{}^\ast \hspace{0.2em}  \, \ottsym{(}  \ottnt{f'_{{\mathrm{1}}}} \, \ottnt{f'_{{\mathrm{2}}}}  \ottsym{)}  \ottsym{:}   \star \Rightarrow  \unskip ^ { \ell }  \! \star $.
\end{proof}

\begin{definitionA}
  $\Omega^X$ is a set of terms:
  $\Omega^X = \{
    \ottsym{(}  \ottnt{f_{{\mathrm{1}}}}  \ottsym{:}   \star  \!\rightarrow\!  \star \Rightarrow  \unskip ^ { \ell }  \!  U^X_{{\mathrm{1}}} \Rightarrow  \unskip ^ { \ell }  \!  \star  \!\rightarrow\!  \star \Rightarrow  \unskip ^ { \ell }  \!   \cdots  \Rightarrow  \unskip ^ { \ell }  \!  \star  \!\rightarrow\!  \star \Rightarrow  \unskip ^ { \ell }  \!  U^X_{\ottmv{m}} \Rightarrow  \unskip ^ { \ell }  \! \star  \!\rightarrow\!  \star        \ottsym{)} \, \ottsym{(}  \ottnt{f_{{\mathrm{1}}}}  \ottsym{:}   \star  \!\rightarrow\!  \star \Rightarrow  \unskip ^ { \ell }  \!  U^X_{{\mathrm{1}}} \Rightarrow  \unskip ^ { \ell }  \!  \star  \!\rightarrow\!  \star \Rightarrow  \unskip ^ { \ell }  \!   \cdots  \Rightarrow  \unskip ^ { \ell }  \!  \star  \!\rightarrow\!  \star \Rightarrow  \unskip ^ { \ell }  \!  U^X_{\ottmv{n}} \Rightarrow  \unskip ^ { \ell }  \!  \star  \!\rightarrow\!  \star \Rightarrow  \unskip ^ { \ell }  \! \star         \ottsym{)}
  \mid
    \ottnt{f_{{\mathrm{1}}}}  \ottsym{=}   \lambda  \ottmv{x} \!:\!  \star  .\,  \ottsym{(}  \ottmv{x}  \ottsym{:}   \star \Rightarrow  \unskip ^ { \ell }  \! \star  \!\rightarrow\!  \star   \ottsym{)}  \, \ottmv{x} \,\wedge\,
    \min\limits_{m,n}(\textit{ord} \, \ottsym{(}  U^X_{\ottmv{m}}  \ottsym{)}, \textit{ord} \, \ottsym{(}  U^X_{\ottmv{n}}  \ottsym{)}) > 0 \,\wedge\,
    m, n \geq 0
  \}$.
\end{definitionA}

\begin{lemmaA} \label{lem:omega_normal_form_eval}
  For any $\ottnt{E}  [  \ottnt{f}  ]$ where $\ottnt{f} \in \Omega^X$,
  there exist $\ottnt{E'}$, $\ottnt{f'}$, and $S$ such that
  $\ottnt{E}  [  \ottnt{f}  ] \,  \xmapsto{ \mathmakebox[0.4em]{} S \mathmakebox[0.3em]{} }\hspace{-0.4em}{}^\ast \hspace{0.2em}  \, \ottnt{E'}  [  \ottnt{f'}  ]$ where $\ottnt{f'} \in \Omega^X$.
\end{lemmaA}

\begin{proof}
  We are given $\ottnt{f}  \ottsym{=}  \ottsym{(}  \ottnt{f_{{\mathrm{1}}}}  \ottsym{:}   \star  \!\rightarrow\!  \star \Rightarrow  \unskip ^ { \ell }  \!  U^X_{{\mathrm{1}}} \Rightarrow  \unskip ^ { \ell }  \!  \star  \!\rightarrow\!  \star \Rightarrow  \unskip ^ { \ell }  \!   \cdots  \Rightarrow  \unskip ^ { \ell }  \!  \star  \!\rightarrow\!  \star \Rightarrow  \unskip ^ { \ell }  \!  U^X_{\ottmv{i}} \Rightarrow  \unskip ^ { \ell }  \! \star  \!\rightarrow\!  \star        \ottsym{)} \, \ottsym{(}  \ottnt{f_{{\mathrm{1}}}}  \ottsym{:}   \star  \!\rightarrow\!  \star \Rightarrow  \unskip ^ { \ell }  \!  U^X_{{\mathrm{1}}} \Rightarrow  \unskip ^ { \ell }  \!  \star  \!\rightarrow\!  \star \Rightarrow  \unskip ^ { \ell }  \!   \cdots  \Rightarrow  \unskip ^ { \ell }  \!  \star  \!\rightarrow\!  \star \Rightarrow  \unskip ^ { \ell }  \!  U^X_{\ottmv{j}} \Rightarrow  \unskip ^ { \ell }  \!  \star  \!\rightarrow\!  \star \Rightarrow  \unskip ^ { \ell }  \! \star         \ottsym{)}$ where $\ottnt{f_{{\mathrm{1}}}}  \ottsym{=}   \lambda  \ottmv{x} \!:\!  \star  .\,  \ottsym{(}  \ottmv{x}  \ottsym{:}   \star \Rightarrow  \unskip ^ { \ell }  \! \star  \!\rightarrow\!  \star   \ottsym{)}  \, \ottmv{x}$.
  By case analysis on $\ottmv{i}$.

  \begin{caseanalysis}
    \case{$i = 0$}
    We are given $\ottnt{f}  \ottsym{=}  \ottnt{f_{{\mathrm{1}}}} \, \ottsym{(}  \ottnt{f_{{\mathrm{1}}}}  \ottsym{:}   \star  \!\rightarrow\!  \star \Rightarrow  \unskip ^ { \ell }  \!  U^X_{{\mathrm{1}}} \Rightarrow  \unskip ^ { \ell }  \!  \star  \!\rightarrow\!  \star \Rightarrow  \unskip ^ { \ell }  \!   \cdots  \Rightarrow  \unskip ^ { \ell }  \!  \star  \!\rightarrow\!  \star \Rightarrow  \unskip ^ { \ell }  \!  U^X_{\ottmv{j}} \Rightarrow  \unskip ^ { \ell }  \!  \star  \!\rightarrow\!  \star \Rightarrow  \unskip ^ { \ell }  \! \star         \ottsym{)}$.

    By \rnp{R\_Beta} and \rnp{E\_Step},
    $\ottnt{E}  [  \ottnt{f}  ] \,  \xrightarrow{ \mathmakebox[0.4em]{} [  ] \mathmakebox[0.3em]{} }  \, \ottnt{E}  [  \ottnt{f'}  ]$
    where $\ottnt{f'}  \ottsym{=}  \ottsym{(}  \ottnt{f_{{\mathrm{1}}}}  \ottsym{:}   \star  \!\rightarrow\!  \star \Rightarrow  \unskip ^ { \ell }  \!  U^X_{{\mathrm{1}}} \Rightarrow  \unskip ^ { \ell }  \!  \star  \!\rightarrow\!  \star \Rightarrow  \unskip ^ { \ell }  \!   \cdots  \Rightarrow  \unskip ^ { \ell }  \!  \star  \!\rightarrow\!  \star \Rightarrow  \unskip ^ { \ell }  \!  U^X_{\ottmv{j}} \Rightarrow  \unskip ^ { \ell }  \!  \star  \!\rightarrow\!  \star \Rightarrow  \unskip ^ { \ell }  \!  \star \Rightarrow  \unskip ^ { \ell }  \! \star  \!\rightarrow\!  \star          \ottsym{)} \, \ottsym{(}  \ottnt{f_{{\mathrm{1}}}}  \ottsym{:}   \star  \!\rightarrow\!  \star \Rightarrow  \unskip ^ { \ell }  \!  U^X_{{\mathrm{1}}} \Rightarrow  \unskip ^ { \ell }  \!  \star  \!\rightarrow\!  \star \Rightarrow  \unskip ^ { \ell }  \!   \cdots  \Rightarrow  \unskip ^ { \ell }  \!  \star  \!\rightarrow\!  \star \Rightarrow  \unskip ^ { \ell }  \!  U^X_{\ottmv{j}} \Rightarrow  \unskip ^ { \ell }  \!  \star  \!\rightarrow\!  \star \Rightarrow  \unskip ^ { \ell }  \! \star         \ottsym{)}$.

    By \rnp{R\_Succeed},
    $\ottnt{E}  [  \ottnt{f'}  ] \,  \xrightarrow{ \mathmakebox[0.4em]{} [  ] \mathmakebox[0.3em]{} }  \, \ottnt{E}  [  \ottnt{f''}  ]$
    where $\ottnt{f''}  \ottsym{=}  \ottsym{(}  \ottnt{f_{{\mathrm{1}}}}  \ottsym{:}   \star  \!\rightarrow\!  \star \Rightarrow  \unskip ^ { \ell }  \!  U^X_{{\mathrm{1}}} \Rightarrow  \unskip ^ { \ell }  \!  \star  \!\rightarrow\!  \star \Rightarrow  \unskip ^ { \ell }  \!   \cdots  \Rightarrow  \unskip ^ { \ell }  \!  \star  \!\rightarrow\!  \star \Rightarrow  \unskip ^ { \ell }  \!  U^X_{\ottmv{j}} \Rightarrow  \unskip ^ { \ell }  \! \star  \!\rightarrow\!  \star        \ottsym{)} \, \ottsym{(}  \ottnt{f_{{\mathrm{1}}}}  \ottsym{:}   \star  \!\rightarrow\!  \star \Rightarrow  \unskip ^ { \ell }  \!  U^X_{{\mathrm{1}}} \Rightarrow  \unskip ^ { \ell }  \!  \star  \!\rightarrow\!  \star \Rightarrow  \unskip ^ { \ell }  \!   \cdots  \Rightarrow  \unskip ^ { \ell }  \!  \star  \!\rightarrow\!  \star \Rightarrow  \unskip ^ { \ell }  \!  U^X_{\ottmv{j}} \Rightarrow  \unskip ^ { \ell }  \!  \star  \!\rightarrow\!  \star \Rightarrow  \unskip ^ { \ell }  \! \star         \ottsym{)}$.

    Finally, $\ottnt{E}  [  \ottnt{f}  ] \,  \xmapsto{ \mathmakebox[0.4em]{} [  ] \mathmakebox[0.3em]{} }\hspace{-0.4em}{}^\ast \hspace{0.2em}  \, \ottnt{E}  [  \ottnt{f''}  ]$.
    Obviously, $\ottnt{f''} \in \Omega^X$.

    \case{$i > 0$}
    We are given $\ottnt{f}  \ottsym{=}  \ottsym{(}  \ottnt{f_{{\mathrm{1}}}}  \ottsym{:}   \star  \!\rightarrow\!  \star \Rightarrow  \unskip ^ { \ell }  \!  U^X_{{\mathrm{1}}} \Rightarrow  \unskip ^ { \ell }  \!  \star  \!\rightarrow\!  \star \Rightarrow  \unskip ^ { \ell }  \!   \cdots  \Rightarrow  \unskip ^ { \ell }  \!  \star  \!\rightarrow\!  \star \Rightarrow  \unskip ^ { \ell }  \!  U^X_{\ottmv{i}} \Rightarrow  \unskip ^ { \ell }  \! \star  \!\rightarrow\!  \star        \ottsym{)} \, \ottsym{(}  \ottnt{f_{{\mathrm{1}}}}  \ottsym{:}   \star  \!\rightarrow\!  \star \Rightarrow  \unskip ^ { \ell }  \!  U^X_{{\mathrm{1}}} \Rightarrow  \unskip ^ { \ell }  \!  \star  \!\rightarrow\!  \star \Rightarrow  \unskip ^ { \ell }  \!   \cdots  \Rightarrow  \unskip ^ { \ell }  \!  \star  \!\rightarrow\!  \star \Rightarrow  \unskip ^ { \ell }  \!  U^X_{\ottmv{j}} \Rightarrow  \unskip ^ { \ell }  \!  \star  \!\rightarrow\!  \star \Rightarrow  \unskip ^ { \ell }  \! \star         \ottsym{)}$.

    Here $\textit{ord} \, \ottsym{(}  U^X_{\ottmv{i}}  \ottsym{)} > 0$, so there exist $U^X_{\ottmv{i}\,{\mathrm{1}}}$ and $U^X_{\ottmv{i}\,{\mathrm{2}}}$
    such that $U^X_{\ottmv{i}}  \ottsym{=}  U^X_{\ottmv{i}\,{\mathrm{1}}}  \!\rightarrow\!  U^X_{\ottmv{i}\,{\mathrm{2}}}$.

    By \rnp{R\_AppCast} and \rnp{E\_Step},
    $\ottnt{E}  [  \ottnt{f}  ] \,  \xmapsto{ \mathmakebox[0.4em]{} [  ] \mathmakebox[0.3em]{} }  \, \ottnt{E}  [  \ottnt{f'}  \ottsym{:}   U^X_{\ottmv{i}\,{\mathrm{2}}} \Rightarrow  \unskip ^ { \ell }  \! \star   ]$
    where $\ottnt{f'}  \ottsym{=}  \ottsym{(}  \ottnt{f_{{\mathrm{1}}}}  \ottsym{:}   \star  \!\rightarrow\!  \star \Rightarrow  \unskip ^ { \ell }  \!  U^X_{{\mathrm{1}}} \Rightarrow  \unskip ^ { \ell }  \!  \star  \!\rightarrow\!  \star \Rightarrow  \unskip ^ { \ell }  \!   \cdots  \Rightarrow  \unskip ^ { \ell }  \!  \star  \!\rightarrow\!  \star \Rightarrow  \unskip ^ { \ell }  \! U^X_{\ottmv{i}}       \ottsym{)} \, \ottsym{(}  \ottnt{f_{{\mathrm{1}}}}  \ottsym{:}   \star  \!\rightarrow\!  \star \Rightarrow  \unskip ^ { \ell }  \!  U^X_{{\mathrm{1}}} \Rightarrow  \unskip ^ { \ell }  \!  \star  \!\rightarrow\!  \star \Rightarrow  \unskip ^ { \ell }  \!   \cdots  \Rightarrow  \unskip ^ { \ell }  \!  \star  \!\rightarrow\!  \star \Rightarrow  \unskip ^ { \ell }  \!  U^X_{\ottmv{j}} \Rightarrow  \unskip ^ { \ell }  \!  \star  \!\rightarrow\!  \star \Rightarrow  \unskip ^ { \ell }  \!  \star \Rightarrow  \unskip ^ { \ell }  \! U^X_{\ottmv{i}\,{\mathrm{1}}}          \ottsym{)}$.

    \begin{caseanalysis}
      \case{$\textit{ord} \, \ottsym{(}  U^X_{\ottmv{i}}  \ottsym{)} = 1$}
      Here, $\textit{ord} \, \ottsym{(}  U^X_{\ottmv{i}\,{\mathrm{1}}}  \ottsym{)} = 0$.  There exists $\ottmv{X}$ such that $U^X_{\ottmv{i}\,{\mathrm{1}}}  \ottsym{=}  \ottmv{X}$.
      By evaluating $\ottnt{E}  [  \ottnt{f'}  ]$ using the rules \rnp{R\_InstArrow}, \rnp{R\_Succeed}, \rnp{R\_AppCast}, and \rnp{E\_Step},
      $\ottnt{E}  [  \ottnt{f}  ] \,  \xmapsto{ \mathmakebox[0.4em]{} S \mathmakebox[0.3em]{} }\hspace{-0.4em}{}^\ast \hspace{0.2em}  \, \ottnt{E}  [  \ottnt{f''}  \ottsym{:}   U^X_{\ottmv{i}\,{\mathrm{2}}} \Rightarrow  \unskip ^ { \ell }  \!  \star \Rightarrow  \unskip ^ { \ell }  \! U^X_{\ottmv{i}\,{\mathrm{2}}}    ]$
      where
      $S  \ottsym{=}  [  \ottmv{X}  :=  \ottmv{X_{{\mathrm{1}}}}  \!\rightarrow\!  \ottmv{X_{{\mathrm{2}}}}  ]$ and
      $\ottnt{f''}  \ottsym{=}  \ottsym{(}  \ottnt{f_{{\mathrm{1}}}}  \ottsym{:}   \star  \!\rightarrow\!  \star \Rightarrow  \unskip ^ { \ell }  \!  S  \ottsym{(}  U^X_{{\mathrm{1}}}  \ottsym{)} \Rightarrow  \unskip ^ { \ell }  \!  \star  \!\rightarrow\!  \star \Rightarrow  \unskip ^ { \ell }  \!   \cdots  \Rightarrow  \unskip ^ { \ell }  \! \star  \!\rightarrow\!  \star      \ottsym{)} \, \ottsym{(}  \ottnt{f_{{\mathrm{1}}}}  \ottsym{:}   \star  \!\rightarrow\!  \star \Rightarrow  \unskip ^ { \ell }  \!  S  \ottsym{(}  U^X_{{\mathrm{1}}}  \ottsym{)} \Rightarrow  \unskip ^ { \ell }  \!  \star  \!\rightarrow\!  \star \Rightarrow  \unskip ^ { \ell }  \!   \cdots  \Rightarrow  \unskip ^ { \ell }  \!  \star  \!\rightarrow\!  \star \Rightarrow  \unskip ^ { \ell }  \!  S  \ottsym{(}  U^X_{\ottmv{j}}  \ottsym{)} \Rightarrow  \unskip ^ { \ell }  \!  \star  \!\rightarrow\!  \star \Rightarrow  \unskip ^ { \ell }  \!  \star \Rightarrow  \unskip ^ { \ell }  \!  S  \ottsym{(}  U^X_{\ottmv{i}\,{\mathrm{1}}}  \ottsym{)} \Rightarrow  \unskip ^ { \ell }  \! \star           \ottsym{)}$.
      Obviously, $\ottnt{f''} \in \Omega^X$.

      \case{$\textit{ord} \, \ottsym{(}  U^X_{\ottmv{i}}  \ottsym{)} > 1$}
      Here, $\textit{ord} \, \ottsym{(}  U^X_{\ottmv{i}\,{\mathrm{1}}}  \ottsym{)} > 0$.
      By evaluating $\ottnt{E}  [  \ottnt{f'}  ]$ using \rnp{R\_Expand}, \rnp{R\_AppCast}, and \rnp{E\_Step},
      $\ottnt{E}  [  \ottnt{f}  ] \,  \xmapsto{ \mathmakebox[0.4em]{} [  ] \mathmakebox[0.3em]{} }\hspace{-0.4em}{}^\ast \hspace{0.2em}  \, \ottnt{E}  [  \ottnt{f''}  \ottsym{:}   U^X_{\ottmv{i}\,{\mathrm{2}}} \Rightarrow  \unskip ^ { \ell }  \!  \star \Rightarrow  \unskip ^ { \ell }  \! U^X_{\ottmv{i}\,{\mathrm{2}}}    ]$
      where $\ottnt{f''}  \ottsym{=}  \ottsym{(}  \ottnt{f_{{\mathrm{1}}}}  \ottsym{:}   \star  \!\rightarrow\!  \star \Rightarrow  \unskip ^ { \ell }  \!  U^X_{{\mathrm{1}}} \Rightarrow  \unskip ^ { \ell }  \!  \star  \!\rightarrow\!  \star \Rightarrow  \unskip ^ { \ell }  \!   \cdots  \Rightarrow  \unskip ^ { \ell }  \! \star  \!\rightarrow\!  \star      \ottsym{)} \, \ottsym{(}  \ottnt{f_{{\mathrm{1}}}}  \ottsym{:}   \star  \!\rightarrow\!  \star \Rightarrow  \unskip ^ { \ell }  \!  U^X_{{\mathrm{1}}} \Rightarrow  \unskip ^ { \ell }  \!  \star  \!\rightarrow\!  \star \Rightarrow  \unskip ^ { \ell }  \!   \cdots  \Rightarrow  \unskip ^ { \ell }  \!  \star  \!\rightarrow\!  \star \Rightarrow  \unskip ^ { \ell }  \!  U^X_{\ottmv{j}} \Rightarrow  \unskip ^ { \ell }  \!  \star  \!\rightarrow\!  \star \Rightarrow  \unskip ^ { \ell }  \!  \star \Rightarrow  \unskip ^ { \ell }  \!  U^X_{\ottmv{i}\,{\mathrm{1}}} \Rightarrow  \unskip ^ { \ell }  \! \star           \ottsym{)}$.
      Obviously, $\ottnt{f''} \in \Omega^X$.
\qedhere
    \end{caseanalysis}
  \end{caseanalysis}
\end{proof}

\begin{lemmaA} \label{lem:omega_diverge}
  If
  \begin{itemize}
    \item $\ottnt{f_{{\mathrm{1}}}}  \ottsym{=}  \ottsym{(}   \lambda  \ottmv{x} \!:\!  X  .\,  \ottsym{(}  \ottmv{x}  \ottsym{:}   X \Rightarrow  \unskip ^ { \ell }  \!  \star \Rightarrow  \unskip ^ { \ell }  \! \star  \!\rightarrow\!  \star    \ottsym{)}  \, \ottsym{(}  \ottmv{x}  \ottsym{:}   X \Rightarrow  \unskip ^ { \ell }  \! \star   \ottsym{)}  \ottsym{)}  \ottsym{:}   X  \!\rightarrow\!  \star \Rightarrow  \unskip ^ { \ell }  \! \star  \!\rightarrow\!  \star $, and
    \item $\ottnt{f_{{\mathrm{2}}}}  \ottsym{=}  \ottsym{(}   \lambda  \ottmv{x} \!:\!  \star  .\,  \ottsym{(}  \ottmv{x}  \ottsym{:}   \star \Rightarrow  \unskip ^ { \ell }  \! \star  \!\rightarrow\!  \star   \ottsym{)}  \, \ottmv{x}  \ottsym{)}  \ottsym{:}   \star  \!\rightarrow\!  \star \Rightarrow  \unskip ^ { \ell }  \! \star $,
  \end{itemize}
  then $ \ottnt{f_{{\mathrm{1}}}} \, \ottnt{f_{{\mathrm{2}}}} \!  \Uparrow  $.
\end{lemmaA}

\begin{proof}
  By Lemmas \ref{lem:omega_translate_to_normal_form} and \ref{lem:omega_normal_form_eval}.
\end{proof}

\begin{lemmaA} \label{lem:omega_blame_translate_to_normal_form}
  If
  \begin{itemize}
    \item $\ottnt{f_{{\mathrm{1}}}}  \ottsym{=}  \ottsym{(}   \lambda  \ottmv{x} \!:\!  U^\iota  .\,  \ottsym{(}  \ottmv{x}  \ottsym{:}   U^\iota \Rightarrow  \unskip ^ { \ell }  \!  \star \Rightarrow  \unskip ^ { \ell }  \! \star  \!\rightarrow\!  \star    \ottsym{)}  \, \ottsym{(}  \ottmv{x}  \ottsym{:}   U^\iota \Rightarrow  \unskip ^ { \ell }  \! \star   \ottsym{)}  \ottsym{)}  \ottsym{:}   U^\iota  \!\rightarrow\!  \star \Rightarrow  \unskip ^ { \ell }  \! \star  \!\rightarrow\!  \star $,
    \item $\ottnt{f_{{\mathrm{2}}}}  \ottsym{=}  \ottsym{(}   \lambda  \ottmv{x} \!:\!  \star  .\,  \ottsym{(}  \ottmv{x}  \ottsym{:}   \star \Rightarrow  \unskip ^ { \ell }  \! \star  \!\rightarrow\!  \star   \ottsym{)}  \, \ottmv{x}  \ottsym{)}  \ottsym{:}   \star  \!\rightarrow\!  \star \Rightarrow  \unskip ^ { \ell }  \! \star $,
    \item $\textit{ord} \, \ottsym{(}  U^\iota  \ottsym{)} > 0$,
  \end{itemize}
  then
  \begin{itemize}
    \item $\ottnt{f'_{{\mathrm{1}}}}  \ottsym{=}  \ottsym{(}   \lambda  \ottmv{x} \!:\!  \star  .\,  \ottsym{(}  \ottmv{x}  \ottsym{:}   \star \Rightarrow  \unskip ^ { \ell }  \! \star  \!\rightarrow\!  \star   \ottsym{)}  \, \ottmv{x}  \ottsym{)}  \ottsym{:}   \star  \!\rightarrow\!  \star \Rightarrow  \unskip ^ { \ell }  \!  U^\iota \Rightarrow  \unskip ^ { \ell }  \! \star  \!\rightarrow\!  \star  $,
    \item $\ottnt{f'_{{\mathrm{2}}}}  \ottsym{=}  \ottsym{(}   \lambda  \ottmv{x} \!:\!  \star  .\,  \ottsym{(}  \ottmv{x}  \ottsym{:}   \star \Rightarrow  \unskip ^ { \ell }  \! \star  \!\rightarrow\!  \star   \ottsym{)}  \, \ottmv{x}  \ottsym{)}  \ottsym{:}   \star  \!\rightarrow\!  \star \Rightarrow  \unskip ^ { \ell }  \!  U^\iota \Rightarrow  \unskip ^ { \ell }  \!  \star  \!\rightarrow\!  \star \Rightarrow  \unskip ^ { \ell }  \! \star   $, and
    \item $\ottnt{f_{{\mathrm{1}}}} \, \ottnt{f_{{\mathrm{2}}}} \,  \xmapsto{ \mathmakebox[0.4em]{} [  ] \mathmakebox[0.3em]{} }\hspace{-0.4em}{}^\ast \hspace{0.2em}  \, \ottsym{(}  \ottnt{f'_{{\mathrm{1}}}} \, \ottnt{f'_{{\mathrm{2}}}}  \ottsym{)}  \ottsym{:}   \star \Rightarrow  \unskip ^ { \ell }  \! \star $.
  \end{itemize}
\end{lemmaA}

\begin{proof}
  By applying reduction rules and evaluation rules.
\end{proof}

\begin{definitionA}
  $\Omega^\iota_i$ is a set of terms:
  $\Omega^\iota_i = \{
    \ottnt{f_{{\mathrm{1}}}} \, \ottsym{(}  \ottnt{f_{{\mathrm{1}}}}  \ottsym{:}   \star  \!\rightarrow\!  \star \Rightarrow  \unskip ^ { \ell }  \! \star   \ottsym{)}
  \mid
    \ottnt{f_{{\mathrm{1}}}}  \ottsym{=}  \ottsym{(}   \lambda  \ottmv{x} \!:\!  \star  .\,  \ottsym{(}  \ottmv{x}  \ottsym{:}   \star \Rightarrow  \unskip ^ { \ell }  \! \star  \!\rightarrow\!  \star   \ottsym{)}  \, \ottmv{x}  \ottsym{)}  \ottsym{:}   \star  \!\rightarrow\!  \star \Rightarrow  \unskip ^ { \ell }  \!  U^\iota_{{\mathrm{1}}} \Rightarrow  \unskip ^ { \ell }  \!  \star  \!\rightarrow\!  \star \Rightarrow  \unskip ^ { \ell }  \!  U^\iota_{{\mathrm{2}}} \Rightarrow  \unskip ^ { \ell }  \!   \cdots  \Rightarrow  \unskip ^ { \ell }  \!  U^\iota_{\ottmv{n}} \Rightarrow  \unskip ^ { \ell }  \! \star  \!\rightarrow\!  \star       \,\wedge\,
    \min\limits_{n}(\textit{ord} \, \ottsym{(}  U^\iota_{\ottmv{n}}  \ottsym{)}) = \iota \,\wedge\,
    n \geq 0
  \}$.
\end{definitionA}

\begin{lemmaA} \label{lem:omega_blame_move_cast}
  For any $\ottnt{E}$, $w_{{\mathrm{1}}}$, $w_{{\mathrm{2}}}$, and $U^\iota_{{\mathrm{1}}}$ where $\textit{ord} \, \ottsym{(}  U^\iota_{{\mathrm{1}}}  \ottsym{)} > 1$,
  there exist $\ottnt{E'}$ and $U^\iota_{{\mathrm{11}}}$ such that
  $\ottnt{E}  [  \ottsym{(}  w_{{\mathrm{1}}}  \ottsym{:}   \star  \!\rightarrow\!  \star \Rightarrow  \unskip ^ { \ell }  \!  U^\iota_{{\mathrm{1}}} \Rightarrow  \unskip ^ { \ell }  \! \star  \!\rightarrow\!  \star    \ottsym{)} \, \ottsym{(}  w_{{\mathrm{2}}}  \ottsym{:}   \star  \!\rightarrow\!  \star \Rightarrow  \unskip ^ { \ell }  \! \star   \ottsym{)}  ] \,  \xmapsto{ \mathmakebox[0.4em]{} [  ] \mathmakebox[0.3em]{} }\hspace{-0.4em}{}^\ast \hspace{0.2em}  \, \ottnt{E'}  [  w_{{\mathrm{1}}} \, \ottsym{(}  w_{{\mathrm{2}}}  \ottsym{:}   \star  \!\rightarrow\!  \star \Rightarrow  \unskip ^ { \ell }  \!  U^\iota_{{\mathrm{11}}} \Rightarrow  \unskip ^ { \ell }  \!  \star  \!\rightarrow\!  \star \Rightarrow  \unskip ^ { \ell }  \! \star     \ottsym{)}  ]$
  and $\textit{ord} \, \ottsym{(}  U^\iota_{{\mathrm{11}}}  \ottsym{)} = \textit{ord} \, \ottsym{(}  U^\iota_{{\mathrm{1}}}  \ottsym{)} - 1$.
\end{lemmaA}

\begin{proof}
  There exist $\ottnt{U_{{\mathrm{11}}}}$ and $\ottnt{U_{{\mathrm{12}}}}$ such that $\ottnt{U_{{\mathrm{1}}}}  \ottsym{=}  \ottnt{U_{{\mathrm{11}}}}  \!\rightarrow\!  \ottnt{U_{{\mathrm{12}}}}$
  because $\textit{ord} \, \ottsym{(}  U^\iota_{{\mathrm{1}}}  \ottsym{)} > 1$.
  By definition, $\textit{ord} \, \ottsym{(}  U^\iota_{{\mathrm{11}}}  \ottsym{)} = \textit{ord} \, \ottsym{(}  U^\iota_{{\mathrm{1}}}  \ottsym{)} - 1$.
  $\ottnt{E}  [  \ottsym{(}  w_{{\mathrm{1}}}  \ottsym{:}   \star  \!\rightarrow\!  \star \Rightarrow  \unskip ^ { \ell }  \!  U^\iota_{{\mathrm{11}}}  \!\rightarrow\!  U^\iota_{{\mathrm{12}}} \Rightarrow  \unskip ^ { \ell }  \! \star  \!\rightarrow\!  \star    \ottsym{)} \, \ottsym{(}  w_{{\mathrm{2}}}  \ottsym{:}   \star  \!\rightarrow\!  \star \Rightarrow  \unskip ^ { \ell }  \! \star   \ottsym{)}  ]$ is evaluated as follows:
  \[
    \begin{array}{cl}
                 & \ottnt{E}  [  \ottsym{(}  w_{{\mathrm{1}}}  \ottsym{:}   \star  \!\rightarrow\!  \star \Rightarrow  \unskip ^ { \ell }  \!  U^\iota_{{\mathrm{11}}}  \!\rightarrow\!  U^\iota_{{\mathrm{12}}} \Rightarrow  \unskip ^ { \ell }  \! \star  \!\rightarrow\!  \star    \ottsym{)} \, \ottsym{(}  w_{{\mathrm{2}}}  \ottsym{:}   \star  \!\rightarrow\!  \star \Rightarrow  \unskip ^ { \ell }  \! \star   \ottsym{)}  ] \\
     \xmapsto{ \mathmakebox[0.4em]{} [  ] \mathmakebox[0.3em]{} }  & \ottnt{E'}  [  \ottsym{(}  w_{{\mathrm{1}}}  \ottsym{:}   \star  \!\rightarrow\!  \star \Rightarrow  \unskip ^ { \ell }  \! U^\iota_{{\mathrm{11}}}  \!\rightarrow\!  U^\iota_{{\mathrm{12}}}   \ottsym{)} \, \ottsym{(}  w_{{\mathrm{2}}}  \ottsym{:}   \star  \!\rightarrow\!  \star \Rightarrow  \unskip ^ { \ell }  \!  \star \Rightarrow  \unskip ^ { \ell }  \! U^\iota_{{\mathrm{11}}}    \ottsym{)}  ] \\
                 & \textrm{where $\ottnt{E'}  \ottsym{=}  \ottnt{E}  [  \left[ \, \right]  \ottsym{:}  U^\iota_{{\mathrm{12}}}  \Rightarrow   \unskip ^ { \ell }  \, \star  ]$} \\
     \xmapsto{ \mathmakebox[0.4em]{} [  ] \mathmakebox[0.3em]{} }  & \ottnt{E'}  [  \ottsym{(}  w_{{\mathrm{1}}}  \ottsym{:}   \star  \!\rightarrow\!  \star \Rightarrow  \unskip ^ { \ell }  \! U^\iota_{{\mathrm{11}}}  \!\rightarrow\!  U^\iota_{{\mathrm{12}}}   \ottsym{)} \, \ottsym{(}  w_{{\mathrm{2}}}  \ottsym{:}   \star  \!\rightarrow\!  \star \Rightarrow  \unskip ^ { \ell }  \!  \star \Rightarrow  \unskip ^ { \ell }  \! U^\iota_{{\mathrm{11}}}    \ottsym{)}  ] \\
     \xmapsto{ \mathmakebox[0.4em]{} [  ] \mathmakebox[0.3em]{} }  &  \ottnt{E'}  [  \ottsym{(}  w_{{\mathrm{1}}}  \ottsym{:}   \star  \!\rightarrow\!  \star \Rightarrow  \unskip ^ { \ell }  \! U^\iota_{{\mathrm{11}}}  \!\rightarrow\!  U^\iota_{{\mathrm{12}}}   \ottsym{)} \, \ottsym{(}  w_{{\mathrm{2}}}  \ottsym{:}   \star  \!\rightarrow\!  \star \Rightarrow  \unskip ^ { \ell }  \! U^\iota_{{\mathrm{11}}}   \ottsym{)}  ] \\
     \xmapsto{ \mathmakebox[0.4em]{} [  ] \mathmakebox[0.3em]{} }  &  \ottnt{E''}  [  w_{{\mathrm{1}}} \, \ottsym{(}  w_{{\mathrm{2}}}  \ottsym{:}   \star  \!\rightarrow\!  \star \Rightarrow  \unskip ^ { \ell }  \!  U^\iota_{{\mathrm{11}}} \Rightarrow  \unskip ^ { \ell }  \! \star    \ottsym{)}  ] \\
                 & \textrm{where $\ottnt{E''}  \ottsym{=}  \ottnt{E'}  [  \left[ \, \right]  \ottsym{:}  \star  \Rightarrow   \unskip ^ { \ell }  \, U^\iota_{{\mathrm{12}}}  ]$} \\
     \xmapsto{ \mathmakebox[0.4em]{} [  ] \mathmakebox[0.3em]{} }  &  \ottnt{E''}  [  w_{{\mathrm{1}}} \, \ottsym{(}  w_{{\mathrm{2}}}  \ottsym{:}   \star  \!\rightarrow\!  \star \Rightarrow  \unskip ^ { \ell }  \!  U^\iota_{{\mathrm{11}}} \Rightarrow  \unskip ^ { \ell }  \!  \star  \!\rightarrow\!  \star \Rightarrow  \unskip ^ { \ell }  \! \star     \ottsym{)}  ]
  \end{array}
  \]
\end{proof}

\begin{lemmaA} \label{lem:omega_blame_reduce_order}
  For any $\ottnt{E}$ and $\ottnt{f} \in \Omega^\iota_i$ where $i > 1$,
  there exist $\ottnt{E'}$ and $\ottnt{f'}$ such that
  $\ottnt{E}  [  \ottnt{f}  ] \,  \xmapsto{ \mathmakebox[0.4em]{} [  ] \mathmakebox[0.3em]{} }\hspace{-0.4em}{}^\ast \hspace{0.2em}  \, \ottnt{E'}  [  \ottnt{f'}  ]$ and $\ottnt{f'} \in \Omega^\iota_{i-1}$.
\end{lemmaA}

\begin{proof}
  We are given $\ottnt{f}  \ottsym{=}  \ottnt{f_{{\mathrm{1}}}} \, \ottnt{f_{{\mathrm{2}}}}$ where
  \begin{itemize}
    \item $\ottnt{f_{{\mathrm{1}}}}  \ottsym{=}  \ottsym{(}   \lambda  \ottmv{x} \!:\!  \star  .\,  \ottsym{(}  \ottmv{x}  \ottsym{:}   \star \Rightarrow  \unskip ^ { \ell }  \! \star  \!\rightarrow\!  \star   \ottsym{)}  \, \ottmv{x}  \ottsym{)}  \ottsym{:}   \star  \!\rightarrow\!  \star \Rightarrow  \unskip ^ { \ell }  \!  U^\iota_{{\mathrm{1}}} \Rightarrow  \unskip ^ { \ell }  \!  \star  \!\rightarrow\!  \star \Rightarrow  \unskip ^ { \ell }  \!  U^\iota_{{\mathrm{2}}} \Rightarrow  \unskip ^ { \ell }  \!   \cdots  \Rightarrow  \unskip ^ { \ell }  \!  U^\iota_{\ottmv{n}} \Rightarrow  \unskip ^ { \ell }  \! \star  \!\rightarrow\!  \star      $ and
    \item $\ottnt{f_{{\mathrm{2}}}}  \ottsym{=}  \ottsym{(}   \lambda  \ottmv{x} \!:\!  \star  .\,  \ottsym{(}  \ottmv{x}  \ottsym{:}   \star \Rightarrow  \unskip ^ { \ell }  \! \star  \!\rightarrow\!  \star   \ottsym{)}  \, \ottmv{x}  \ottsym{)}  \ottsym{:}   \star  \!\rightarrow\!  \star \Rightarrow  \unskip ^ { \ell }  \!  U^\iota_{{\mathrm{1}}} \Rightarrow  \unskip ^ { \ell }  \!  \star  \!\rightarrow\!  \star \Rightarrow  \unskip ^ { \ell }  \!  U^\iota_{{\mathrm{2}}} \Rightarrow  \unskip ^ { \ell }  \!   \cdots  \Rightarrow  \unskip ^ { \ell }  \!  U^\iota_{\ottmv{n}} \Rightarrow  \unskip ^ { \ell }  \! \star  \!\rightarrow\!  \star      $.
  \end{itemize}

  By applying Lemma \ref{lem:omega_blame_move_cast} $n$ times,
  there exist $\ottnt{E'}$ and $\ottnt{f'}$ such that $\ottnt{E}  [  \ottnt{f}  ] \,  \xmapsto{ \mathmakebox[0.4em]{} [  ] \mathmakebox[0.3em]{} }\hspace{-0.4em}{}^\ast \hspace{0.2em}  \, \ottnt{E'}  [  \ottnt{f'}  ]$
  where $\ottnt{f'}  \ottsym{=}  \ottnt{f'_{{\mathrm{1}}}} \, \ottnt{f'_{{\mathrm{2}}}}$, $m = 2n$,
  \begin{itemize}
    \item $\ottnt{f'_{{\mathrm{1}}}}  \ottsym{=}   \lambda  \ottmv{x} \!:\!  \star  .\,  \ottsym{(}  \ottmv{x}  \ottsym{:}   \star \Rightarrow  \unskip ^ { \ell }  \! \star  \!\rightarrow\!  \star   \ottsym{)}  \, \ottmv{x}$, and
    \item $\ottnt{f'_{{\mathrm{2}}}}  \ottsym{=}  \ottsym{(}   \lambda  \ottmv{x} \!:\!  \star  .\,  \ottsym{(}  \ottmv{x}  \ottsym{:}   \star \Rightarrow  \unskip ^ { \ell }  \! \star  \!\rightarrow\!  \star   \ottsym{)}  \, \ottmv{x}  \ottsym{)}  \ottsym{:}   \star  \!\rightarrow\!  \star \Rightarrow  \unskip ^ { \ell }  \!  U^\iota_{{\mathrm{1}}} \Rightarrow  \unskip ^ { \ell }  \!  \star  \!\rightarrow\!  \star \Rightarrow  \unskip ^ { \ell }  \!  U^\iota_{{\mathrm{2}}} \Rightarrow  \unskip ^ { \ell }  \!   \cdots  \Rightarrow  \unskip ^ { \ell }  \!  U^\iota_{\ottmv{m}} \Rightarrow  \unskip ^ { \ell }  \! \star  \!\rightarrow\!  \star      $.
  \end{itemize}
  Then, $\min\limits_{1 \leq i \leq m} \textit{ord} \, \ottsym{(}  U^\iota_{\ottmv{i}}  \ottsym{)} = \min\limits_{1 \leq i \leq n} \textit{ord} \, \ottsym{(}  U^\iota_{\ottmv{i}}  \ottsym{)} - 1$.

  By applying reduction and evaluation rules,
  there exist $\ottnt{E''}$ and $\ottnt{f''}$ such that $\ottnt{E'}  [  \ottnt{f'}  ] \,  \xmapsto{ \mathmakebox[0.4em]{} [  ] \mathmakebox[0.3em]{} }\hspace{-0.4em}{}^\ast \hspace{0.2em}  \, \ottnt{E''}  [  \ottnt{f''}  ]$
  where $\ottnt{f''}  \ottsym{=}  \ottnt{f''_{{\mathrm{1}}}} \, \ottnt{f''_{{\mathrm{2}}}}$,
  \begin{itemize}
    \item $\ottnt{f''_{{\mathrm{1}}}}  \ottsym{=}  \ottsym{(}   \lambda  \ottmv{x} \!:\!  \star  .\,  \ottsym{(}  \ottmv{x}  \ottsym{:}   \star \Rightarrow  \unskip ^ { \ell }  \! \star  \!\rightarrow\!  \star   \ottsym{)}  \, \ottmv{x}  \ottsym{)}  \ottsym{:}   \star  \!\rightarrow\!  \star \Rightarrow  \unskip ^ { \ell }  \!  U^\iota_{{\mathrm{1}}} \Rightarrow  \unskip ^ { \ell }  \!  \star  \!\rightarrow\!  \star \Rightarrow  \unskip ^ { \ell }  \!  U^\iota_{{\mathrm{2}}} \Rightarrow  \unskip ^ { \ell }  \!   \cdots  \Rightarrow  \unskip ^ { \ell }  \!  U^\iota_{\ottmv{m}} \Rightarrow  \unskip ^ { \ell }  \! \star  \!\rightarrow\!  \star      $, and
    \item $\ottnt{f''_{{\mathrm{2}}}}  \ottsym{=}  \ottsym{(}   \lambda  \ottmv{x} \!:\!  \star  .\,  \ottsym{(}  \ottmv{x}  \ottsym{:}   \star \Rightarrow  \unskip ^ { \ell }  \! \star  \!\rightarrow\!  \star   \ottsym{)}  \, \ottmv{x}  \ottsym{)}  \ottsym{:}   \star  \!\rightarrow\!  \star \Rightarrow  \unskip ^ { \ell }  \!  U^\iota_{{\mathrm{1}}} \Rightarrow  \unskip ^ { \ell }  \!  \star  \!\rightarrow\!  \star \Rightarrow  \unskip ^ { \ell }  \!  U^\iota_{{\mathrm{2}}} \Rightarrow  \unskip ^ { \ell }  \!   \cdots  \Rightarrow  \unskip ^ { \ell }  \!  U^\iota_{\ottmv{m}} \Rightarrow  \unskip ^ { \ell }  \! \star  \!\rightarrow\!  \star      $.
  \end{itemize}

  Finally, $\ottnt{f''} \in \Omega^\iota_{i-1}$.
\end{proof}

\begin{lemmaA} \label{lem:omega_blame_last_cast_order_one_blame}
  For any $\ottnt{E}$, $w_{{\mathrm{1}}}$, $w_{{\mathrm{2}}}$, and $U^\iota$ where $\textit{ord} \, \ottsym{(}  U^\iota  \ottsym{)} = 1$,
  $\ottnt{E}  [  \ottsym{(}  w_{{\mathrm{1}}}  \ottsym{:}   U^\iota \Rightarrow  \unskip ^ { \ell }  \! \star  \!\rightarrow\!  \star   \ottsym{)} \, \ottsym{(}  w_{{\mathrm{2}}}  \ottsym{:}   \star  \!\rightarrow\!  \star \Rightarrow  \unskip ^ { \ell }  \! \star   \ottsym{)}  ] \,  \xmapsto{ \mathmakebox[0.4em]{} [  ] \mathmakebox[0.3em]{} }\hspace{-0.4em}{}^\ast \hspace{0.2em}  \, \textsf{\textup{blame}\relax} \,  \bar{ \ell } $.
\end{lemmaA}

\begin{proof}
  Obviously by applying reduction and evaluation rules.
\end{proof}

\begin{lemmaA} \label{lem:omega_blame_order_one_blame}
  For any $\ottnt{E}$ and $\ottnt{f} \in \Omega^\iota_1$,
  there exists $\ell$ such that $\ottnt{E}  [  \ottnt{f}  ] \,  \xmapsto{ \mathmakebox[0.4em]{} [  ] \mathmakebox[0.3em]{} }\hspace{-0.4em}{}^\ast \hspace{0.2em}  \, \textsf{\textup{blame}\relax} \, \ell$.
\end{lemmaA}

\begin{proof}
  There exist $\ottnt{f_{{\mathrm{1}}}}$ and $\ottnt{f_{{\mathrm{2}}}}$ such that $\ottnt{f}  \ottsym{=}  \ottnt{f_{{\mathrm{1}}}} \, \ottnt{f_{{\mathrm{2}}}}$,
  \begin{itemize}
    \item $\ottnt{f_{{\mathrm{1}}}}  \ottsym{=}  \ottsym{(}   \lambda  \ottmv{x} \!:\!  \star  .\,  \ottsym{(}  \ottmv{x}  \ottsym{:}   \star \Rightarrow  \unskip ^ { \ell }  \! \star  \!\rightarrow\!  \star   \ottsym{)}  \, \ottmv{x}  \ottsym{)}  \ottsym{:}   \star  \!\rightarrow\!  \star \Rightarrow  \unskip ^ { \ell }  \!  U^\iota_{{\mathrm{1}}} \Rightarrow  \unskip ^ { \ell }  \!  \star  \!\rightarrow\!  \star \Rightarrow  \unskip ^ { \ell }  \!  U^\iota_{{\mathrm{2}}} \Rightarrow  \unskip ^ { \ell }  \!   \cdots  \Rightarrow  \unskip ^ { \ell }  \!  U^\iota_{\ottmv{n}} \Rightarrow  \unskip ^ { \ell }  \! \star  \!\rightarrow\!  \star      $, and
    \item $\ottnt{f_{{\mathrm{2}}}}  \ottsym{=}  \ottsym{(}   \lambda  \ottmv{x} \!:\!  \star  .\,  \ottsym{(}  \ottmv{x}  \ottsym{:}   \star \Rightarrow  \unskip ^ { \ell }  \! \star  \!\rightarrow\!  \star   \ottsym{)}  \, \ottmv{x}  \ottsym{)}  \ottsym{:}   \star  \!\rightarrow\!  \star \Rightarrow  \unskip ^ { \ell }  \!  U^\iota_{{\mathrm{1}}} \Rightarrow  \unskip ^ { \ell }  \!  \star  \!\rightarrow\!  \star \Rightarrow  \unskip ^ { \ell }  \!  U^\iota_{{\mathrm{2}}} \Rightarrow  \unskip ^ { \ell }  \!   \cdots  \Rightarrow  \unskip ^ { \ell }  \!  U^\iota_{\ottmv{n}} \Rightarrow  \unskip ^ { \ell }  \! \star  \!\rightarrow\!  \star      $.
  \end{itemize}
  And, there exists $i$ such that $\textit{ord} \, \ottsym{(}  U^\iota_{\ottmv{i}}  \ottsym{)} = 1$.

  By applying Lemma \ref{lem:omega_blame_move_cast} $n - i$ times,
  there exist $\ottnt{E'}$, $\ottnt{f'_{{\mathrm{1}}}}$, and $\ottnt{f'_{{\mathrm{2}}}}$ such that $\ottnt{E}  [  \ottnt{f}  ] \,  \xmapsto{ \mathmakebox[0.4em]{} [  ] \mathmakebox[0.3em]{} }\hspace{-0.4em}{}^\ast \hspace{0.2em}  \, \ottnt{E'}  [  \ottnt{f'_{{\mathrm{1}}}} \, \ottnt{f'_{{\mathrm{2}}}}  ]$,
  \begin{itemize}
    \item $\ottnt{f'_{{\mathrm{1}}}}  \ottsym{=}  \ottsym{(}   \lambda  \ottmv{x} \!:\!  \star  .\,  \ottsym{(}  \ottmv{x}  \ottsym{:}   \star \Rightarrow  \unskip ^ { \ell }  \! \star  \!\rightarrow\!  \star   \ottsym{)}  \, \ottmv{x}  \ottsym{)}  \ottsym{:}   \star  \!\rightarrow\!  \star \Rightarrow  \unskip ^ { \ell }  \!  U^\iota_{{\mathrm{1}}} \Rightarrow  \unskip ^ { \ell }  \!  \star  \!\rightarrow\!  \star \Rightarrow  \unskip ^ { \ell }  \!  U^\iota_{{\mathrm{2}}} \Rightarrow  \unskip ^ { \ell }  \!   \cdots  \Rightarrow  \unskip ^ { \ell }  \!  U^\iota_{\ottmv{i}} \Rightarrow  \unskip ^ { \ell }  \! \star  \!\rightarrow\!  \star      $, and
    \item $\ottnt{f'_{{\mathrm{2}}}}  \ottsym{=}  \ottsym{(}   \lambda  \ottmv{x} \!:\!  \star  .\,  \ottsym{(}  \ottmv{x}  \ottsym{:}   \star \Rightarrow  \unskip ^ { \ell }  \! \star  \!\rightarrow\!  \star   \ottsym{)}  \, \ottmv{x}  \ottsym{)}  \ottsym{:}   \star  \!\rightarrow\!  \star \Rightarrow  \unskip ^ { \ell }  \!  U^\iota_{{\mathrm{1}}} \Rightarrow  \unskip ^ { \ell }  \!  \star  \!\rightarrow\!  \star \Rightarrow  \unskip ^ { \ell }  \!  U^\iota_{{\mathrm{2}}} \Rightarrow  \unskip ^ { \ell }  \!   \cdots  \Rightarrow  \unskip ^ { \ell }  \!  U^\iota_{\ottmv{m}} \Rightarrow  \unskip ^ { \ell }  \! \star  \!\rightarrow\!  \star      $.
  \end{itemize}

  By Lemma \ref{lem:omega_blame_last_cast_order_one_blame},
  $\ottnt{E'}  [  \ottnt{f'_{{\mathrm{1}}}} \, \ottnt{f'_{{\mathrm{2}}}}  ] \,  \xmapsto{ \mathmakebox[0.4em]{} [  ] \mathmakebox[0.3em]{} }\hspace{-0.4em}{}^\ast \hspace{0.2em}  \, \textsf{\textup{blame}\relax} \, \ell$.
  Finally, $\ottnt{E}  [  \ottnt{f}  ] \,  \xmapsto{ \mathmakebox[0.4em]{} [  ] \mathmakebox[0.3em]{} }\hspace{-0.4em}{}^\ast \hspace{0.2em}  \, \textsf{\textup{blame}\relax} \, \ell$.
\end{proof}

\begin{lemmaA} \label{lem:omega_blame}
  For any $S$ such that $\textit{ftv} \, \ottsym{(}  S  \ottsym{(}  X  \ottsym{)}  \ottsym{)}  \ottsym{=}   \emptyset $,
  $S  \ottsym{(}  \ottnt{f_{{\mathrm{1}}}} \, \ottnt{f_{{\mathrm{2}}}}  \ottsym{)} \,  \xmapsto{ \mathmakebox[0.4em]{} [  ] \mathmakebox[0.3em]{} }\hspace{-0.4em}{}^\ast \hspace{0.2em}  \, \textsf{\textup{blame}\relax} \, \ell$ for some $\ell$ where
  \begin{itemize}
    \item $\ottnt{f_{{\mathrm{1}}}}  \ottsym{=}  \ottsym{(}   \lambda  \ottmv{x} \!:\!  X  .\,  \ottsym{(}  \ottmv{x}  \ottsym{:}   X \Rightarrow  \unskip ^ { \ell }  \!  \star \Rightarrow  \unskip ^ { \ell }  \! \star  \!\rightarrow\!  \star    \ottsym{)}  \, \ottsym{(}  \ottmv{x}  \ottsym{:}   X \Rightarrow  \unskip ^ { \ell }  \! \star   \ottsym{)}  \ottsym{)}  \ottsym{:}   X  \!\rightarrow\!  \star \Rightarrow  \unskip ^ { \ell }  \! \star  \!\rightarrow\!  \star $ and
    \item $\ottnt{f_{{\mathrm{2}}}}  \ottsym{=}  \ottsym{(}   \lambda  \ottmv{x} \!:\!  \star  .\,  \ottsym{(}  \ottmv{x}  \ottsym{:}   \star \Rightarrow  \unskip ^ { \ell }  \! \star  \!\rightarrow\!  \star   \ottsym{)}  \, \ottmv{x}  \ottsym{)}  \ottsym{:}   \star  \!\rightarrow\!  \star \Rightarrow  \unskip ^ { \ell }  \! \star $.
  \end{itemize}
\end{lemmaA}

\begin{proof}
  There exists $\ottnt{T}$ such that $S  \ottsym{(}  X  \ottsym{)}  \ottsym{=}  \ottnt{T}$.
  Here, $\textit{ftv} \, \ottsym{(}  S  \ottsym{(}  X  \ottsym{)}  \ottsym{)}  \ottsym{=}   \emptyset $.
  So, there exists $U^\iota$ such that $S  \ottsym{(}  X  \ottsym{)}  \ottsym{=}  U^\iota$,

  \begin{itemize}
    \item $\ottnt{f_{{\mathrm{1}}}}  \ottsym{=}  \ottsym{(}   \lambda  \ottmv{x} \!:\!  U^\iota  .\,  \ottsym{(}  \ottmv{x}  \ottsym{:}   U^\iota \Rightarrow  \unskip ^ { \ell }  \!  \star \Rightarrow  \unskip ^ { \ell }  \! \star  \!\rightarrow\!  \star    \ottsym{)}  \, \ottsym{(}  \ottmv{x}  \ottsym{:}   U^\iota \Rightarrow  \unskip ^ { \ell }  \! \star   \ottsym{)}  \ottsym{)}  \ottsym{:}   U^\iota  \!\rightarrow\!  \star \Rightarrow  \unskip ^ { \ell }  \! \star  \!\rightarrow\!  \star $, and
    \item $\ottnt{f_{{\mathrm{2}}}}  \ottsym{=}  \ottsym{(}   \lambda  \ottmv{x} \!:\!  \star  .\,  \ottsym{(}  \ottmv{x}  \ottsym{:}   \star \Rightarrow  \unskip ^ { \ell }  \! \star  \!\rightarrow\!  \star   \ottsym{)}  \, \ottmv{x}  \ottsym{)}  \ottsym{:}   \star  \!\rightarrow\!  \star \Rightarrow  \unskip ^ { \ell }  \! \star $.
  \end{itemize}

  By case analysis on the order of $U^\iota$.

  \begin{caseanalysis}
    \case{$\textit{ord} \, \ottsym{(}  U^\iota  \ottsym{)} = 0$}
    There exists $\iota$ such that $U^\iota  \ottsym{=}  \iota$.
    Obviously, $\ottnt{f_{{\mathrm{1}}}} \, \ottnt{f_{{\mathrm{2}}}} \,  \xmapsto{ \mathmakebox[0.4em]{} [  ] \mathmakebox[0.3em]{} }\hspace{-0.4em}{}^\ast \hspace{0.2em}  \, \textsf{\textup{blame}\relax} \, \ell$ for some $\ell$.

    \case{$\textit{ord} \, \ottsym{(}  U^\iota  \ottsym{)} > 0$}
    By Lemma \ref{lem:omega_blame_translate_to_normal_form},
    There exist
    \begin{itemize}
      \item $\ottnt{f'_{{\mathrm{1}}}}  \ottsym{=}  \ottsym{(}   \lambda  \ottmv{x} \!:\!  \star  .\,  \ottsym{(}  \ottmv{x}  \ottsym{:}   \star \Rightarrow  \unskip ^ { \ell }  \! \star  \!\rightarrow\!  \star   \ottsym{)}  \, \ottmv{x}  \ottsym{)}  \ottsym{:}   \star  \!\rightarrow\!  \star \Rightarrow  \unskip ^ { \ell }  \!  U^\iota \Rightarrow  \unskip ^ { \ell }  \! \star  \!\rightarrow\!  \star  $ and
      \item $\ottnt{f'_{{\mathrm{2}}}}  \ottsym{=}  \ottsym{(}   \lambda  \ottmv{x} \!:\!  \star  .\,  \ottsym{(}  \ottmv{x}  \ottsym{:}   \star \Rightarrow  \unskip ^ { \ell }  \! \star  \!\rightarrow\!  \star   \ottsym{)}  \, \ottmv{x}  \ottsym{)}  \ottsym{:}   \star  \!\rightarrow\!  \star \Rightarrow  \unskip ^ { \ell }  \!  U^\iota \Rightarrow  \unskip ^ { \ell }  \!  \star  \!\rightarrow\!  \star \Rightarrow  \unskip ^ { \ell }  \! \star   $,
    \end{itemize}
    such that $\ottnt{f_{{\mathrm{1}}}} \, \ottnt{f_{{\mathrm{2}}}} \,  \xmapsto{ \mathmakebox[0.4em]{} [  ] \mathmakebox[0.3em]{} }\hspace{-0.4em}{}^\ast \hspace{0.2em}  \, \ottsym{(}  \ottnt{f'_{{\mathrm{1}}}} \, \ottnt{f'_{{\mathrm{2}}}}  \ottsym{)}  \ottsym{:}   \star \Rightarrow  \unskip ^ { \ell }  \! \star $.

    Here $\ottnt{f'_{{\mathrm{1}}}} \, \ottnt{f'_{{\mathrm{2}}}} \in \Omega^\iota_n$ where $n = \textit{ord} \, \ottsym{(}  U^\iota  \ottsym{)}$.
    By Lemma \ref{lem:omega_blame_reduce_order} and
    Lemma \ref{lem:omega_blame_order_one_blame}. We finish.
    \qedhere
  \end{caseanalysis}
\end{proof}

\ifrestate
\thmDivergence*
\else
\begin{theoremA}[name=,restate=thmDivergence]
  There exists $\ottnt{f}$ such that (1) $ \emptyset   \vdash  \ottnt{f}  \ottsym{:}  \ottnt{U}$ and (2)
  $ \ottnt{f} \!  \Uparrow  $ and (3) for any $S$ such that
  $\textit{ftv} \, \ottsym{(}  S  \ottsym{(}  \ottnt{f}  \ottsym{)}  \ottsym{)}  \ottsym{=}   \emptyset $, it holds that $S  \ottsym{(}  \ottnt{f}  \ottsym{)} \, \longmapsto_{\textsf{\textup{B}\relax}\relax}^\ast \, \textsf{\textup{blame}\relax} \, \ell$ for
  some $\ell$.
\end{theoremA}
\fi 

\begin{proof}
  \leavevmode
  Let $\ottnt{f}  \ottsym{=}  \ottsym{(}   \lambda  \ottmv{x} \!:\!  X  .\,  \ottsym{(}  \ottmv{x}  \ottsym{:}   X \Rightarrow  \unskip ^ { \ell }  \!  \star \Rightarrow  \unskip ^ { \ell }  \! \star  \!\rightarrow\!  \star    \ottsym{)}  \, \ottsym{(}  \ottmv{x}  \ottsym{:}   X \Rightarrow  \unskip ^ { \ell }  \! \star   \ottsym{)}  \ottsym{)}  \ottsym{:}   X  \!\rightarrow\!  \star \Rightarrow  \unskip ^ { \ell }  \! \star  \!\rightarrow\!  \star  \, \ottsym{(}   \lambda  \ottmv{x} \!:\!  \star  .\,  \ottsym{(}  \ottmv{x}  \ottsym{:}   \star \Rightarrow  \unskip ^ { \ell }  \! \star  \!\rightarrow\!  \star   \ottsym{)}  \, \ottmv{x}  \ottsym{)}  \ottsym{:}   \star  \!\rightarrow\!  \star \Rightarrow  \unskip ^ { \ell }  \! \star $.
  By Lemma \ref{lem:omega_diverge}, $ \ottnt{f} \!  \Uparrow  $.

  Here, $\textit{ftv} \, \ottsym{(}  \ottnt{f}  \ottsym{)} = \{X\}$.
  By Lemma \ref{lem:omega_blame},
  for any $S$ where $\textit{ftv} \, \ottsym{(}  S  \ottsym{(}  X  \ottsym{)}  \ottsym{)}  \ottsym{=}   \emptyset $,
  $S  \ottsym{(}  \ottnt{f}  \ottsym{)} \,  \xmapsto{ \mathmakebox[0.4em]{} [  ] \mathmakebox[0.3em]{} }\hspace{-0.4em}{}^\ast \hspace{0.2em}  \, \textsf{\textup{blame}\relax} \, \ell$ for some $\ell$.

  By Theorem \ref{thm:conservative_extension}, $S  \ottsym{(}  \ottnt{f}  \ottsym{)} \, \longmapsto_{\textsf{\textup{B}\relax}\relax}^\ast \, \textsf{\textup{blame}\relax} \, \ell$.
\end{proof}

\subsection{Completeness of Dynamic Type Inference}

\begin{lemmaA} \label{lem:subst_value_is_value}
  If $ \emptyset   \vdash  \ottnt{f}  \ottsym{:}  \ottnt{U}$ and $S  \ottsym{(}  \ottnt{f}  \ottsym{)}$ is a value,
  then $\ottnt{f}$ is a value.
\end{lemmaA}

\begin{proof}
  By induction on the typing derivation.

  \begin{caseanalysis}
    \case{\rnp{T\_VarP}}
    Cannot happen.

    \case{\rnp{T\_Const}, \rnp{T\_Abs}}
    Obvious.

    \case{\rnp{T\_Op}}
    We are given $\ottnt{f}  \ottsym{=}  \mathit{op} \, \ottsym{(}  \ottnt{f_{{\mathrm{1}}}}  \ottsym{,}  \ottnt{f_{{\mathrm{2}}}}  \ottsym{)}$ for some $ \mathit{op} $, $\ottnt{f_{{\mathrm{1}}}}$, and $\ottnt{f_{{\mathrm{2}}}}$.
    $S  \ottsym{(}  \ottnt{f}  \ottsym{)}  \ottsym{=}  \mathit{op} \, \ottsym{(}  S  \ottsym{(}  \ottnt{f_{{\mathrm{1}}}}  \ottsym{)}  \ottsym{,}  S  \ottsym{(}  \ottnt{f_{{\mathrm{2}}}}  \ottsym{)}  \ottsym{)}$ is not a value.
    Contradiction.

    \case{\rnp{T\_App}}
    We are given $\ottnt{f}  \ottsym{=}  \ottnt{f_{{\mathrm{1}}}} \, \ottnt{f_{{\mathrm{2}}}}$ for some $\ottnt{f_{{\mathrm{1}}}}$ and $\ottnt{f_{{\mathrm{2}}}}$.
    $S  \ottsym{(}  \ottnt{f}  \ottsym{)}  \ottsym{=}  S  \ottsym{(}  \ottnt{f_{{\mathrm{1}}}}  \ottsym{)} \, S  \ottsym{(}  \ottnt{f_{{\mathrm{2}}}}  \ottsym{)}$ is not a value.
    Contradiction.

    \case{\rnp{T\_Cast}}
    We are given $\ottnt{f}  \ottsym{=}  \ottnt{f_{{\mathrm{1}}}}  \ottsym{:}   \ottnt{U'} \Rightarrow  \unskip ^ { \ell }  \! \ottnt{U} $ for some $\ottnt{f_{{\mathrm{1}}}}$, $\ottnt{U'}$ and $\ell$.
    By inversion, $ \emptyset   \vdash_{\textsf{\textup{B}\relax}\relax}  \ottnt{f_{{\mathrm{1}}}}  \ottsym{:}  \ottnt{U'}$ and $\ottnt{U'}  \sim  \ottnt{U}$.
    By case analysis on the type consistency relation.

    \begin{caseanalysis}
      \case{\rnp{C\_Base}}
      We are given $\ottnt{f}  \ottsym{=}  \ottnt{f_{{\mathrm{1}}}}  \ottsym{:}   \iota \Rightarrow  \unskip ^ { \ell }  \! \iota $ for some $\iota$.
      $S  \ottsym{(}  \ottnt{f}  \ottsym{)}  \ottsym{=}  S  \ottsym{(}  \ottnt{f_{{\mathrm{1}}}}  \ottsym{)}  \ottsym{:}   \iota \Rightarrow  \unskip ^ { \ell }  \! \iota $ is not a value.
      Contradiction.

      \case{\rnp{C\_TyVar}}
      We are given $\ottnt{f}  \ottsym{=}  \ottnt{f_{{\mathrm{1}}}}  \ottsym{:}   \ottmv{X} \Rightarrow  \unskip ^ { \ell }  \! \ottmv{X} $ for some $\ottmv{X}$.
      $S  \ottsym{(}  \ottnt{f}  \ottsym{)}  \ottsym{=}  S  \ottsym{(}  \ottnt{f_{{\mathrm{1}}}}  \ottsym{)}  \ottsym{:}   \ottnt{T} \Rightarrow  \unskip ^ { \ell }  \! \ottnt{T} $ for some $\ottnt{T}$.

      By case analysis on the structure of $\ottnt{T}$.

      \begin{caseanalysis}
        \case{$\ottnt{T}  \ottsym{=}  \iota$ for some $\iota$}
        $S  \ottsym{(}  \ottnt{f}  \ottsym{)}  \ottsym{=}  S  \ottsym{(}  \ottnt{f_{{\mathrm{1}}}}  \ottsym{)}  \ottsym{:}   \iota \Rightarrow  \unskip ^ { \ell }  \! \iota $ is not a value.
        Contradiction.

        \case{$\ottnt{T}  \ottsym{=}  \ottmv{X'}$ for some $\ottmv{X'}$}
        $S  \ottsym{(}  \ottnt{f}  \ottsym{)}  \ottsym{=}  S  \ottsym{(}  \ottnt{f_{{\mathrm{1}}}}  \ottsym{)}  \ottsym{:}   \ottmv{X'} \Rightarrow  \unskip ^ { \ell }  \! \ottmv{X'} $ is not a value.
        Contradiction.

        \case{$\ottnt{T}  \ottsym{=}  \ottnt{T_{{\mathrm{1}}}}  \!\rightarrow\!  \ottnt{T_{{\mathrm{2}}}}$ for some $\ottnt{T_{{\mathrm{1}}}}$ and $\ottnt{T_{{\mathrm{2}}}}$}
        \leavevmode\\
        We are given $S  \ottsym{(}  \ottnt{f}  \ottsym{)}  \ottsym{=}  S  \ottsym{(}  \ottnt{f_{{\mathrm{1}}}}  \ottsym{)}  \ottsym{:}   \ottnt{T_{{\mathrm{1}}}}  \!\rightarrow\!  \ottnt{T_{{\mathrm{2}}}} \Rightarrow  \unskip ^ { \ell }  \! \ottnt{T_{{\mathrm{1}}}}  \!\rightarrow\!  \ottnt{T_{{\mathrm{2}}}} $.
        By Lemma \ref{lem:canonical_forms}, $S  \ottsym{(}  \ottnt{f_{{\mathrm{1}}}}  \ottsym{)}$ is a value.
        By the IH, $\ottnt{f_{{\mathrm{1}}}}$ is a value.

        By Lemma \ref{lem:canonical_forms}, the value $\ottnt{f_{{\mathrm{1}}}}$ is not typed at a type variable.
        Contradiction.
      \end{caseanalysis}

      \case{\rnp{C\_DynL}}
      We are given $\ottnt{f}  \ottsym{=}  \ottnt{f_{{\mathrm{1}}}}  \ottsym{:}   \star \Rightarrow  \unskip ^ { \ell }  \! \ottnt{U} $.
      $S  \ottsym{(}  \ottnt{f}  \ottsym{)}  \ottsym{=}  S  \ottsym{(}  \ottnt{f_{{\mathrm{1}}}}  \ottsym{)}  \ottsym{:}   \star \Rightarrow  \unskip ^ { \ell }  \! S  \ottsym{(}  \ottnt{U}  \ottsym{)} $ is not a value.
      Contradiction.

      \case{\rnp{C\_DynR}}
      We are given $\ottnt{f}  \ottsym{=}  \ottnt{f_{{\mathrm{1}}}}  \ottsym{:}   \ottnt{U'} \Rightarrow  \unskip ^ { \ell }  \! \star $.
      By Lemma \ref{lem:canonical_forms},
      $S  \ottsym{(}  \ottnt{f_{{\mathrm{1}}}}  \ottsym{)}$ is a value and $S  \ottsym{(}  \ottnt{U'}  \ottsym{)}$ is a ground type.
      By the IH, $\ottnt{f_{{\mathrm{1}}}}$ is a value.

      By case analysis on the structure of $\ottnt{U'}$.

      \begin{caseanalysis}
        \case{$\ottnt{U'}  \ottsym{=}  \star$}
        $S  \ottsym{(}  \ottnt{f}  \ottsym{)}  \ottsym{=}  S  \ottsym{(}  \ottnt{f_{{\mathrm{1}}}}  \ottsym{)}  \ottsym{:}   \star \Rightarrow  \unskip ^ { \ell }  \! \star $ is not a value.
        Contradiction.

        \case{$\ottnt{U'}  \ottsym{=}  \iota$}
        We are given $S  \ottsym{(}  \ottnt{f}  \ottsym{)}  \ottsym{=}  S  \ottsym{(}  \ottnt{f_{{\mathrm{1}}}}  \ottsym{)}  \ottsym{:}   \iota \Rightarrow  \unskip ^ { \ell }  \! \star $.
        So, $\ottnt{f}  \ottsym{=}  \ottnt{f_{{\mathrm{1}}}}  \ottsym{:}   \iota \Rightarrow  \unskip ^ { \ell }  \! \star $ is a value.

        \case{$\ottnt{U'}  \ottsym{=}  \ottmv{X}$}
        By Lemma \ref{lem:canonical_forms}, the value $\ottnt{f_{{\mathrm{1}}}}$ is not typed at a type variable.
        Contradiction.

        \case{$\ottnt{U'}  \ottsym{=}  \star  \!\rightarrow\!  \star$}
        By definition, $\ottnt{f}  \ottsym{=}  \ottnt{f_{{\mathrm{1}}}}  \ottsym{:}   \star  \!\rightarrow\!  \star \Rightarrow  \unskip ^ { \ell }  \! \star $ is a value.

        \case{$\ottnt{U}  \ottsym{=}  \ottnt{U_{{\mathrm{1}}}}  \!\rightarrow\!  \ottnt{U_{{\mathrm{2}}}}$ where $\ottnt{U}  \neq  \star  \!\rightarrow\!  \star$}
        $S  \ottsym{(}  \ottnt{f}  \ottsym{)}  \ottsym{=}  S  \ottsym{(}  \ottnt{f_{{\mathrm{1}}}}  \ottsym{)}  \ottsym{:}   S  \ottsym{(}  \ottnt{U_{{\mathrm{1}}}}  \ottsym{)}  \!\rightarrow\!  S  \ottsym{(}  \ottnt{U_{{\mathrm{2}}}}  \ottsym{)} \Rightarrow  \unskip ^ { \ell }  \! \star $ is not a value.
        Contradiction.
      \end{caseanalysis}

      \case{\rnp{C\_Arrow}}
      We are given $\ottnt{f}  \ottsym{=}  \ottnt{f_{{\mathrm{1}}}}  \ottsym{:}   \ottnt{U'_{{\mathrm{1}}}}  \!\rightarrow\!  \ottnt{U'_{{\mathrm{2}}}} \Rightarrow  \unskip ^ { \ell }  \! \ottnt{U_{{\mathrm{1}}}}  \!\rightarrow\!  \ottnt{U_{{\mathrm{2}}}} $
      for some $\ottnt{U_{{\mathrm{1}}}}$, $\ottnt{U_{{\mathrm{2}}}}$, $\ottnt{U'_{{\mathrm{1}}}}$, and $\ottnt{U'_{{\mathrm{2}}}}$
      where $\ottnt{U}  \ottsym{=}  \ottnt{U_{{\mathrm{1}}}}  \!\rightarrow\!  \ottnt{U_{{\mathrm{2}}}}$ and $\ottnt{U'}  \ottsym{=}  \ottnt{U'_{{\mathrm{1}}}}  \!\rightarrow\!  \ottnt{U'_{{\mathrm{2}}}}$.

      By Lemma \ref{lem:canonical_forms}, $S  \ottsym{(}  \ottnt{f_{{\mathrm{1}}}}  \ottsym{)}$ is a value.
      By the IH, $\ottnt{f_{{\mathrm{1}}}}$ is a value.

      Finally, $\ottnt{f}  \ottsym{=}  \ottnt{f_{{\mathrm{1}}}}  \ottsym{:}   \ottnt{U'_{{\mathrm{1}}}}  \!\rightarrow\!  \ottnt{U'_{{\mathrm{2}}}} \Rightarrow  \unskip ^ { \ell }  \! \ottnt{U_{{\mathrm{1}}}}  \!\rightarrow\!  \ottnt{U_{{\mathrm{2}}}} $ is a value.
    \end{caseanalysis}

    \case{\rnp{T\_LetP}}
    We are given $\ottnt{f}  \ottsym{=}   \textsf{\textup{let}\relax} \,  \ottmv{x}  =   \Lambda    \overrightarrow{ \ottmv{X_{\ottmv{i}}} }  .\,  w_{{\mathrm{1}}}   \textsf{\textup{ in }\relax}  \ottnt{f_{{\mathrm{2}}}} $
    for some $\ottmv{x}$, $ \overrightarrow{ \ottmv{X_{\ottmv{i}}} } $, $w_{{\mathrm{1}}}$, and $\ottnt{f_{{\mathrm{2}}}}$.
    Obviously, $S  \ottsym{(}  \ottnt{f}  \ottsym{)}$ is not a value.
    Contradiction.

    \case{\rnp{T\_Blame}}
    We are given $\ottnt{f}  \ottsym{=}  \textsf{\textup{blame}\relax} \, \ell$ for some $\ell$.
    For any $S$,$ S  \ottsym{(}  \ottnt{f}  \ottsym{)}  \ottsym{=}  \textsf{\textup{blame}\relax} \, \ell$.  Contradiction.
\qedhere
  \end{caseanalysis}
\end{proof}

\ifrestate
\lemCompletenessValueEvalStep*
\else
\begin{lemmaA}[name=,restate=lemCompletenessValueEvalStep] \label{lem:completeness_value_eval_step}
  If
  $ \emptyset   \vdash  \ottnt{f}  \ottsym{:}  \ottnt{U}$ and
  $S  \ottsym{(}  \ottnt{f}  \ottsym{)} \,  \xmapsto{ \mathmakebox[0.4em]{} S'_{{\mathrm{1}}} \mathmakebox[0.3em]{} }  \, \ottnt{f'}$ and
  $\ottnt{f'} \,  \centernot{\xmapsto{ \mathmakebox[0.4em]{} S_{{\mathrm{0}}} \mathmakebox[0.3em]{} } }\hspace{-0.4em}{}^\ast \hspace{0.2em}  \, \textsf{\textup{blame}\relax} \, \ell$ for any $S_{{\mathrm{0}}}$ and $\ell$,
  then
  $\ottnt{f} \,  \xmapsto{ \mathmakebox[0.4em]{} S' \mathmakebox[0.3em]{} }  \, \ottnt{f''}$ and
  $S''  \ottsym{(}  \ottnt{f''}  \ottsym{)}  \ottsym{=}  \ottnt{f'}$
  for some $S'$, $S''$, and $\ottnt{f''}$.
\end{lemmaA}
\fi

\begin{proof}
  By case analysis on the evaluation rule applied to $S  \ottsym{(}  \ottnt{f}  \ottsym{)}$.

  \begin{caseanalysis}
    \case{\rnp{E\_Step}}
    \leavevmode\\
    There exist $\ottnt{E}$, $\ottnt{f_{{\mathrm{1}}}}$, and $\ottnt{f'_{{\mathrm{1}}}}$ such that
    $S  \ottsym{(}  \ottnt{f_{{\mathrm{1}}}}  \ottsym{)} \,  \xrightarrow{ \mathmakebox[0.4em]{} [  ] \mathmakebox[0.3em]{} }  \, \ottnt{f'_{{\mathrm{1}}}}$,
    $S  \ottsym{(}  \ottnt{E}  \ottsym{)}  [  S  \ottsym{(}  \ottnt{f_{{\mathrm{1}}}}  \ottsym{)}  ] \,  \xmapsto{ \mathmakebox[0.4em]{} [  ] \mathmakebox[0.3em]{} }  \, [  ]  \ottsym{(}  S  \ottsym{(}  \ottnt{E}  \ottsym{)}  [  \ottnt{f'_{{\mathrm{1}}}}  ]  \ottsym{)}$,
    $\ottnt{f}  \ottsym{=}  \ottnt{E}  [  \ottnt{f_{{\mathrm{1}}}}  ]$, and $\ottnt{f'}  \ottsym{=}  S  \ottsym{(}  \ottnt{E}  \ottsym{)}  [  \ottnt{f'_{{\mathrm{1}}}}  ]$.

    By Lemma \ref{lem:context_inversion}, $ \emptyset   \vdash  \ottnt{f_{{\mathrm{1}}}}  \ottsym{:}  \ottnt{U_{{\mathrm{1}}}}$ for some $\ottnt{U_{{\mathrm{1}}}}$.
    By case analysis on the typing derivation on $\ottnt{f_{{\mathrm{1}}}}$.

    \begin{caseanalysis}
      \case{\rnp{T\_VarP}, \rnp{T\_Const}, \rnp{T\_Abs}, \rnp{T\_Blame}}
      Cannot happen.

      \case{\rnp{T\_Op}}
      We are given $ \emptyset   \vdash  \mathit{op} \, \ottsym{(}  \ottnt{f_{{\mathrm{11}}}}  \ottsym{,}  \ottnt{f_{{\mathrm{12}}}}  \ottsym{)}  \ottsym{:}  \iota_{{\mathrm{1}}}$
      for some $ \mathit{op} $, $\ottnt{f_{{\mathrm{11}}}}$, $\ottnt{f_{{\mathrm{12}}}}$, and $\iota_{{\mathrm{1}}}$
      where $\ottnt{U_{{\mathrm{1}}}}  \ottsym{=}  \iota_{{\mathrm{1}}}$.
      By inversion, we have $ \mathit{ty} ( \mathit{op} )   \ottsym{=}  \iota_{{\mathrm{11}}}  \!\rightarrow\!  \iota_{{\mathrm{12}}}  \!\rightarrow\!  \iota_{{\mathrm{1}}}$,
      $ \emptyset   \vdash  \ottnt{f_{{\mathrm{11}}}}  \ottsym{:}  \iota_{{\mathrm{11}}}$ and $ \emptyset   \vdash  \ottnt{f_{{\mathrm{12}}}}  \ottsym{:}  \iota_{{\mathrm{12}}}$
      for some $\iota_{{\mathrm{11}}}$ and $\iota_{{\mathrm{12}}}$.
      By case analysis on the reduction rule applied to $S  \ottsym{(}  \ottnt{f_{{\mathrm{1}}}}  \ottsym{)}$.
      \begin{caseanalysis}
        \case{\rnp{R\_Op}}
        We are given $\mathit{op} \, \ottsym{(}  w'_{{\mathrm{11}}}  \ottsym{,}  w'_{{\mathrm{12}}}  \ottsym{)} \,  \xrightarrow{ \mathmakebox[0.4em]{} [  ] \mathmakebox[0.3em]{} }  \,  \llbracket\mathit{op}\rrbracket ( w'_{{\mathrm{11}}} ,  w'_{{\mathrm{12}}} ) $
        where $S  \ottsym{(}  \ottnt{f_{{\mathrm{11}}}}  \ottsym{)}  \ottsym{=}  w'_{{\mathrm{11}}}$, $S  \ottsym{(}  \ottnt{f_{{\mathrm{12}}}}  \ottsym{)}  \ottsym{=}  w'_{{\mathrm{12}}}$, and $S'  \ottsym{=}  [  ]$.
        By Lemma \ref{lem:subst_value_is_value}, $\ottnt{f_{{\mathrm{11}}}}$ and $\ottnt{f_{{\mathrm{12}}}}$ are values.
        So, there exist $w_{{\mathrm{11}}}$ and $w_{{\mathrm{12}}}$
        such that $\ottnt{f_{{\mathrm{11}}}}  \ottsym{=}  w_{{\mathrm{11}}}$ and $\ottnt{f_{{\mathrm{12}}}}  \ottsym{=}  w_{{\mathrm{12}}}$.
        By \rnp{R\_Op}, $\mathit{op} \, \ottsym{(}  w_{{\mathrm{11}}}  \ottsym{,}  w_{{\mathrm{12}}}  \ottsym{)} \,  \xrightarrow{ \mathmakebox[0.4em]{} [  ] \mathmakebox[0.3em]{} }  \,  \llbracket\mathit{op}\rrbracket ( w_{{\mathrm{11}}} ,  w_{{\mathrm{12}}} ) $.
        By \rnp{E\_Step}, $\ottnt{E}  [  \mathit{op} \, \ottsym{(}  w_{{\mathrm{11}}}  \ottsym{,}  w_{{\mathrm{12}}}  \ottsym{)}  ] \,  \xrightarrow{ \mathmakebox[0.4em]{} [  ] \mathmakebox[0.3em]{} }  \, [  ]  \ottsym{(}  \ottnt{E}  [   \llbracket\mathit{op}\rrbracket ( w_{{\mathrm{11}}} ,  w_{{\mathrm{12}}} )   ]  \ottsym{)}$.\\
        Let $S''  \ottsym{=}  S$, then $S  \ottsym{(}  \ottnt{E}  [  \mathit{op} \, \ottsym{(}  w_{{\mathrm{11}}}  \ottsym{,}  w_{{\mathrm{12}}}  \ottsym{)}  ]  \ottsym{)} \,  \xrightarrow{ \mathmakebox[0.4em]{} [  ] \mathmakebox[0.3em]{} }  \, S  \ottsym{(}  \ottnt{E}  \ottsym{)}  [   \llbracket\mathit{op}\rrbracket ( w'_{{\mathrm{11}}} ,  w'_{{\mathrm{12}}} )   ]$.

        \otherwise Cannot happen.
      \end{caseanalysis}

      \case{\rnp{T\_App}}
      We are given $ \emptyset   \vdash  \ottnt{f_{{\mathrm{11}}}} \, \ottnt{f_{{\mathrm{12}}}}  \ottsym{:}  \ottnt{U_{{\mathrm{1}}}}$ for some $\ottnt{f_{{\mathrm{11}}}}$ and $\ottnt{f_{{\mathrm{12}}}}$.
      By inversion, we have $ \emptyset   \vdash  \ottnt{f_{{\mathrm{11}}}}  \ottsym{:}  \ottnt{U'_{{\mathrm{1}}}}  \!\rightarrow\!  \ottnt{U_{{\mathrm{1}}}}$ and $ \emptyset   \vdash  \ottnt{f_{{\mathrm{12}}}}  \ottsym{:}  \ottnt{U'_{{\mathrm{1}}}}$
      for some $\ottnt{U'_{{\mathrm{1}}}}$.
      By case analysis on the reduction rule applied to $S  \ottsym{(}  \ottnt{f_{{\mathrm{1}}}}  \ottsym{)}$.
      \begin{caseanalysis}
        \case{\rnp{R\_Beta}}
        We are given $\ottsym{(}   \lambda  \ottmv{x} \!:\!  \ottnt{U'_{{\mathrm{1}}}}  .\,  \ottnt{f'_{{\mathrm{11}}}}   \ottsym{)} \, w' \,  \xrightarrow{ \mathmakebox[0.4em]{} [  ] \mathmakebox[0.3em]{} }  \, \ottnt{f'_{{\mathrm{11}}}}  [  \ottmv{x}  \ottsym{:=}  w'  ]$
        where $S  \ottsym{(}  \ottnt{f_{{\mathrm{11}}}}  \ottsym{)}  \ottsym{=}   \lambda  \ottmv{x} \!:\!  \ottnt{U'_{{\mathrm{1}}}}  .\,  \ottnt{f'_{{\mathrm{11}}}} $, $S  \ottsym{(}  \ottnt{f_{{\mathrm{12}}}}  \ottsym{)}  \ottsym{=}  w'$,
        $\ottnt{f'_{{\mathrm{1}}}}  \ottsym{=}  \ottnt{f'_{{\mathrm{11}}}}  [  \ottmv{x}  \ottsym{:=}  w'  ]$, and $S'  \ottsym{=}  [  ]$.
        By definition, $\ottnt{f_{{\mathrm{11}}}}  \ottsym{=}   \lambda  \ottmv{x} \!:\!  \ottnt{U''_{{\mathrm{1}}}}  .\,  \ottnt{f''_{{\mathrm{11}}}} $
        where $S  \ottsym{(}  \ottnt{U''_{{\mathrm{1}}}}  \ottsym{)}  \ottsym{=}  \ottnt{U'_{{\mathrm{1}}}}$ and $S  \ottsym{(}  \ottnt{f''_{{\mathrm{11}}}}  \ottsym{)}  \ottsym{=}  \ottnt{f'_{{\mathrm{11}}}}$.
        By Lemma \ref{lem:subst_value_is_value}, $\ottnt{f_{{\mathrm{12}}}}$ is a value.
        By \rnp{R\_Beta}, $\ottsym{(}   \lambda  \ottmv{x} \!:\!  \ottnt{U''_{{\mathrm{1}}}}  .\,  \ottnt{f''_{{\mathrm{11}}}}   \ottsym{)} \, \ottnt{f_{{\mathrm{12}}}} \,  \xrightarrow{ \mathmakebox[0.4em]{} [  ] \mathmakebox[0.3em]{} }  \, \ottnt{f''_{{\mathrm{11}}}}  [  \ottmv{x}  \ottsym{:=}  \ottnt{f_{{\mathrm{12}}}}  ]$.
        Also, $S  \ottsym{(}  \ottnt{f''_{{\mathrm{11}}}}  [  \ottmv{x}  \ottsym{:=}  \ottnt{f_{{\mathrm{12}}}}  ]  \ottsym{)}  \ottsym{=}  \ottnt{f'_{{\mathrm{11}}}}  [  \ottmv{x}  \ottsym{:=}  w'  ]$.
        By \rnp{E\_Step}, $\ottnt{E}  [  \ottsym{(}   \lambda  \ottmv{x} \!:\!  \ottnt{U''_{{\mathrm{1}}}}  .\,  \ottnt{f''_{{\mathrm{11}}}}   \ottsym{)} \, \ottnt{f_{{\mathrm{12}}}}  ] \,  \xmapsto{ \mathmakebox[0.4em]{} [  ] \mathmakebox[0.3em]{} }  \, [  ]  \ottsym{(}  \ottnt{E}  [  \ottnt{f''_{{\mathrm{11}}}}  [  \ottmv{x}  \ottsym{:=}  \ottnt{f_{{\mathrm{12}}}}  ]  ]  \ottsym{)}$.
        Let $S''  \ottsym{=}  S$, then $S''  \ottsym{(}  \ottnt{E}  [  \ottnt{f''_{{\mathrm{11}}}}  [  \ottmv{x}  \ottsym{:=}  \ottnt{f_{{\mathrm{12}}}}  ]  ]  \ottsym{)}  \ottsym{=}  S  \ottsym{(}  \ottnt{E}  \ottsym{)}  [  \ottnt{f'_{{\mathrm{11}}}}  [  \ottmv{x}  \ottsym{:=}  w'  ]  ]$.

        \case{\rnp{R\_AppCast}}
        \leavevmode\\
        We are given $\ottsym{(}  w'_{{\mathrm{11}}}  \ottsym{:}   \ottnt{U'_{{\mathrm{11}}}}  \!\rightarrow\!  \ottnt{U'_{{\mathrm{12}}}} \Rightarrow  \unskip ^ { \ell }  \! \ottnt{U'_{{\mathrm{13}}}}  \!\rightarrow\!  \ottnt{U'_{{\mathrm{14}}}}   \ottsym{)} \, w'_{{\mathrm{12}}} \,  \xrightarrow{ \mathmakebox[0.4em]{} [  ] \mathmakebox[0.3em]{} }  \, \ottsym{(}  w'_{{\mathrm{11}}} \, \ottsym{(}  w'_{{\mathrm{12}}}  \ottsym{:}   \ottnt{U'_{{\mathrm{13}}}} \Rightarrow  \unskip ^ {  \bar{ \ell }  }  \! \ottnt{U'_{{\mathrm{11}}}}   \ottsym{)}  \ottsym{)}  \ottsym{:}   \ottnt{U'_{{\mathrm{12}}}} \Rightarrow  \unskip ^ { \ell }  \! \ottnt{U'_{{\mathrm{14}}}} $
        where $S  \ottsym{(}  \ottnt{f_{{\mathrm{11}}}}  \ottsym{)}  \ottsym{=}  w'_{{\mathrm{11}}}  \ottsym{:}   \ottnt{U'_{{\mathrm{11}}}}  \!\rightarrow\!  \ottnt{U'_{{\mathrm{12}}}} \Rightarrow  \unskip ^ { \ell }  \! \ottnt{U'_{{\mathrm{13}}}}  \!\rightarrow\!  \ottnt{U'_{{\mathrm{14}}}} $, $S  \ottsym{(}  \ottnt{f_{{\mathrm{12}}}}  \ottsym{)}  \ottsym{=}  w'_{{\mathrm{12}}}$,
        $\ottnt{f'_{{\mathrm{1}}}}  \ottsym{=}  \ottsym{(}  w'_{{\mathrm{11}}} \, \ottsym{(}  w'_{{\mathrm{12}}}  \ottsym{:}   \ottnt{U'_{{\mathrm{13}}}} \Rightarrow  \unskip ^ {  \bar{ \ell }  }  \! \ottnt{U'_{{\mathrm{11}}}}   \ottsym{)}  \ottsym{)}  \ottsym{:}   \ottnt{U'_{{\mathrm{12}}}} \Rightarrow  \unskip ^ { \ell }  \! \ottnt{U'_{{\mathrm{14}}}} $,
        and $S'  \ottsym{=}  [  ]$.
        By definition, $\ottnt{f_{{\mathrm{11}}}}  \ottsym{=}  \ottnt{f''_{{\mathrm{11}}}}  \ottsym{:}   \ottnt{U''_{{\mathrm{11}}}}  \!\rightarrow\!  \ottnt{U''_{{\mathrm{12}}}} \Rightarrow  \unskip ^ { \ell }  \! \ottnt{U''_{{\mathrm{13}}}}  \!\rightarrow\!  \ottnt{U''_{{\mathrm{14}}}} $
        where $S  \ottsym{(}  \ottnt{f''_{{\mathrm{11}}}}  \ottsym{)}  \ottsym{=}  w'_{{\mathrm{11}}}$, $S  \ottsym{(}  \ottnt{U''_{{\mathrm{11}}}}  \ottsym{)}  \ottsym{=}  \ottnt{U'_{{\mathrm{11}}}}$, $S  \ottsym{(}  \ottnt{U''_{{\mathrm{12}}}}  \ottsym{)}  \ottsym{=}  \ottnt{U'_{{\mathrm{12}}}}$,
        $S  \ottsym{(}  \ottnt{U''_{{\mathrm{13}}}}  \ottsym{)}  \ottsym{=}  \ottnt{U'_{{\mathrm{13}}}}$, and $S  \ottsym{(}  \ottnt{U''_{{\mathrm{14}}}}  \ottsym{)}  \ottsym{=}  \ottnt{U'_{{\mathrm{14}}}}$.
        By Lemma \ref{lem:subst_value_is_value}, $\ottnt{f''_{{\mathrm{11}}}}$ and $\ottnt{f_{{\mathrm{12}}}}$ are values.
        By \rnp{R\_AppCast},
        $\ottsym{(}  \ottnt{f''_{{\mathrm{11}}}}  \ottsym{:}   \ottnt{U''_{{\mathrm{11}}}}  \!\rightarrow\!  \ottnt{U''_{{\mathrm{12}}}} \Rightarrow  \unskip ^ { \ell }  \! \ottnt{U''_{{\mathrm{13}}}}  \!\rightarrow\!  \ottnt{U''_{{\mathrm{14}}}}   \ottsym{)} \, \ottnt{f_{{\mathrm{12}}}} \,  \xrightarrow{ \mathmakebox[0.4em]{} [  ] \mathmakebox[0.3em]{} }  \, \ottsym{(}  \ottnt{f''_{{\mathrm{11}}}} \, \ottsym{(}  \ottnt{f_{{\mathrm{12}}}}  \ottsym{:}   \ottnt{U''_{{\mathrm{13}}}} \Rightarrow  \unskip ^ {  \bar{ \ell }  }  \! \ottnt{U''_{{\mathrm{11}}}}   \ottsym{)}  \ottsym{)}  \ottsym{:}   \ottnt{U''_{{\mathrm{12}}}} \Rightarrow  \unskip ^ { \ell }  \! \ottnt{U''_{{\mathrm{14}}}} $.
        Also, $S  \ottsym{(}  \ottsym{(}  \ottnt{f''_{{\mathrm{11}}}} \, \ottsym{(}  \ottnt{f_{{\mathrm{12}}}}  \ottsym{:}   \ottnt{U''_{{\mathrm{13}}}} \Rightarrow  \unskip ^ {  \bar{ \ell }  }  \! \ottnt{U''_{{\mathrm{11}}}}   \ottsym{)}  \ottsym{)}  \ottsym{:}   \ottnt{U''_{{\mathrm{12}}}} \Rightarrow  \unskip ^ { \ell }  \! \ottnt{U''_{{\mathrm{14}}}}   \ottsym{)}  \ottsym{=}  \ottsym{(}  w'_{{\mathrm{11}}} \, \ottsym{(}  w'_{{\mathrm{12}}}  \ottsym{:}   \ottnt{U''_{{\mathrm{13}}}} \Rightarrow  \unskip ^ {  \bar{ \ell }  }  \! \ottnt{U''_{{\mathrm{11}}}}   \ottsym{)}  \ottsym{)}  \ottsym{:}   \ottnt{U''_{{\mathrm{12}}}} \Rightarrow  \unskip ^ { \ell }  \! \ottnt{U''_{{\mathrm{14}}}} $.
        By \rnp{E\_Step},
        $\ottnt{E}  [  \ottsym{(}  \ottnt{f''_{{\mathrm{11}}}}  \ottsym{:}   \ottnt{U''_{{\mathrm{11}}}}  \!\rightarrow\!  \ottnt{U''_{{\mathrm{12}}}} \Rightarrow  \unskip ^ { \ell }  \! \ottnt{U''_{{\mathrm{13}}}}  \!\rightarrow\!  \ottnt{U''_{{\mathrm{14}}}}   \ottsym{)} \, \ottnt{f_{{\mathrm{12}}}}  ] \,  \xmapsto{ \mathmakebox[0.4em]{} [  ] \mathmakebox[0.3em]{} }  \, [  ]  \ottsym{(}  \ottnt{E}  [  \ottsym{(}  \ottnt{f''_{{\mathrm{11}}}} \, \ottsym{(}  \ottnt{f_{{\mathrm{12}}}}  \ottsym{:}   \ottnt{U''_{{\mathrm{13}}}} \Rightarrow  \unskip ^ {  \bar{ \ell }  }  \! \ottnt{U''_{{\mathrm{11}}}}   \ottsym{)}  \ottsym{)}  \ottsym{:}   \ottnt{U''_{{\mathrm{12}}}} \Rightarrow  \unskip ^ { \ell }  \! \ottnt{U''_{{\mathrm{14}}}}   ]  \ottsym{)}$.
        Let $S''  \ottsym{=}  S$, then
        $S''  \ottsym{(}  \ottnt{E}  [  \ottsym{(}  \ottnt{f''_{{\mathrm{11}}}} \, \ottsym{(}  \ottnt{f_{{\mathrm{12}}}}  \ottsym{:}   \ottnt{U''_{{\mathrm{13}}}} \Rightarrow  \unskip ^ {  \bar{ \ell }  }  \! \ottnt{U''_{{\mathrm{11}}}}   \ottsym{)}  \ottsym{)}  \ottsym{:}   \ottnt{U''_{{\mathrm{12}}}} \Rightarrow  \unskip ^ { \ell }  \! \ottnt{U''_{{\mathrm{14}}}}   ]  \ottsym{)}  \ottsym{=}  S  \ottsym{(}  \ottnt{E}  \ottsym{)}  [  \ottsym{(}  w'_{{\mathrm{11}}} \, \ottsym{(}  w'_{{\mathrm{12}}}  \ottsym{:}   \ottnt{U'_{{\mathrm{13}}}} \Rightarrow  \unskip ^ {  \bar{ \ell }  }  \! \ottnt{U'_{{\mathrm{11}}}}   \ottsym{)}  \ottsym{)}  \ottsym{:}   \ottnt{U'_{{\mathrm{12}}}} \Rightarrow  \unskip ^ { \ell }  \! \ottnt{U'_{{\mathrm{14}}}}   ]$.

        \otherwise
        Cannot happen.
      \end{caseanalysis}

      \case{\rnp{T\_Cast}}
      We are given $ \emptyset   \vdash  \ottsym{(}  \ottnt{f_{{\mathrm{11}}}}  \ottsym{:}   \ottnt{U'_{{\mathrm{1}}}} \Rightarrow  \unskip ^ { \ell }  \! \ottnt{U_{{\mathrm{1}}}}   \ottsym{)}  \ottsym{:}  \ottnt{U_{{\mathrm{1}}}}$
      for some $\ottnt{f_{{\mathrm{11}}}}$, $\ottnt{U'_{{\mathrm{1}}}}$, and $\ell$.
      By inversion, we have $ \emptyset   \vdash  \ottnt{f_{{\mathrm{11}}}}  \ottsym{:}  \ottnt{U'_{{\mathrm{1}}}}$ and $\ottnt{U'_{{\mathrm{1}}}}  \sim  \ottnt{U_{{\mathrm{1}}}}$.
      By case analysis on $\ottnt{U'_{{\mathrm{1}}}}  \sim  \ottnt{U_{{\mathrm{1}}}}$.

      \begin{caseanalysis}
        \case{\rnp{C\_Base}}
        We are given $\ottnt{f_{{\mathrm{1}}}}  \ottsym{=}  \ottnt{f_{{\mathrm{11}}}}  \ottsym{:}   \iota \Rightarrow  \unskip ^ { \ell }  \! \iota $ for some $\iota$.
        By definition, $S  \ottsym{(}  \ottnt{f_{{\mathrm{1}}}}  \ottsym{)}  \ottsym{=}  S  \ottsym{(}  \ottnt{f_{{\mathrm{11}}}}  \ottsym{)}  \ottsym{:}   \iota \Rightarrow  \unskip ^ { \ell }  \! \iota $.
        It must be the case that $S  \ottsym{(}  \ottnt{f_{{\mathrm{11}}}}  \ottsym{)}$ is a value and
        $S  \ottsym{(}  \ottnt{f_{{\mathrm{11}}}}  \ottsym{)}  \ottsym{:}   \iota \Rightarrow  \unskip ^ { \ell }  \! \iota  \,  \xrightarrow{ \mathmakebox[0.4em]{} [  ] \mathmakebox[0.3em]{} }  \, S  \ottsym{(}  \ottnt{f_{{\mathrm{11}}}}  \ottsym{)}$.
        By Lemma \ref{lem:subst_value_is_value}, $\ottnt{f_{{\mathrm{11}}}}$ is a value.
        By \rnp{R\_IdBase}, $\ottnt{f_{{\mathrm{11}}}}  \ottsym{:}   \iota \Rightarrow  \unskip ^ { \ell }  \! \iota  \,  \xrightarrow{ \mathmakebox[0.4em]{} [  ] \mathmakebox[0.3em]{} }  \, \ottnt{f_{{\mathrm{11}}}}$.
        By \rnp{E\_Step}, $\ottnt{E}  [  \ottnt{f_{{\mathrm{11}}}}  \ottsym{:}   \iota \Rightarrow  \unskip ^ { \ell }  \! \iota   ] \,  \xmapsto{ \mathmakebox[0.4em]{} [  ] \mathmakebox[0.3em]{} }  \, [  ]  \ottsym{(}  \ottnt{E}  [  \ottnt{f_{{\mathrm{11}}}}  ]  \ottsym{)}$.
        Let $S''  \ottsym{=}  S$, then $S''  \ottsym{(}  \ottnt{E}  [  \ottnt{f_{{\mathrm{11}}}}  ]  \ottsym{)}  \ottsym{=}  S  \ottsym{(}  \ottnt{E}  \ottsym{)}  [  S  \ottsym{(}  \ottnt{f_{{\mathrm{11}}}}  \ottsym{)}  ]$.

        \case{\rnp{C\_TyVar}}
        We are given $\ottnt{f_{{\mathrm{1}}}}  \ottsym{=}  \ottnt{f_{{\mathrm{11}}}}  \ottsym{:}   \ottmv{X} \Rightarrow  \unskip ^ { \ell }  \! \ottmv{X} $ for some $\ottmv{X}$.
        By case analysis on the structure of $S  \ottsym{(}  \ottmv{X}  \ottsym{)}$.

        It must be the case that $S  \ottsym{(}  \ottnt{f_{{\mathrm{11}}}}  \ottsym{)}$ is a value.
        By Lemma \ref{lem:subst_value_is_value}, $\ottnt{f_{{\mathrm{11}}}}$ is a value.
        By Lemma \ref{lem:canonical_forms}, a value $\ottnt{f_{{\mathrm{11}}}}$ is not typed at a type variable.
        Contradiction.

        \case{\rnp{C\_DynL}}
        We are given $\ottnt{f_{{\mathrm{1}}}}  \ottsym{=}  \ottnt{f_{{\mathrm{11}}}}  \ottsym{:}   \star \Rightarrow  \unskip ^ { \ell }  \! \ottnt{U_{{\mathrm{1}}}} $.
        By definition, $S  \ottsym{(}  \ottnt{f_{{\mathrm{1}}}}  \ottsym{)}  \ottsym{=}  S  \ottsym{(}  \ottnt{f_{{\mathrm{11}}}}  \ottsym{)}  \ottsym{:}   \star \Rightarrow  \unskip ^ { \ell }  \! S  \ottsym{(}  \ottnt{U_{{\mathrm{1}}}}  \ottsym{)} $.
        It must be the case that $S  \ottsym{(}  \ottnt{f_{{\mathrm{11}}}}  \ottsym{)}$ is a value.
        By Lemma \ref{lem:subst_value_is_value}, $\ottnt{f_{{\mathrm{11}}}}$ is a value.

        By case analysis on the structure of $\ottnt{U_{{\mathrm{1}}}}$.

        \begin{caseanalysis}
          \case{$\ottnt{U_{{\mathrm{1}}}}  \ottsym{=}  \star$}
          We are given $S  \ottsym{(}  \ottnt{f_{{\mathrm{1}}}}  \ottsym{)}  \ottsym{=}  S  \ottsym{(}  \ottnt{f_{{\mathrm{11}}}}  \ottsym{)}  \ottsym{:}   \star \Rightarrow  \unskip ^ { \ell }  \! \star $ and $S  \ottsym{(}  \ottnt{f_{{\mathrm{11}}}}  \ottsym{)}  \ottsym{:}   \star \Rightarrow  \unskip ^ { \ell }  \! \star  \,  \xrightarrow{ \mathmakebox[0.4em]{} [  ] \mathmakebox[0.3em]{} }  \, S  \ottsym{(}  \ottnt{f_{{\mathrm{11}}}}  \ottsym{)}$.
          By \rnp{R\_IdStar}, $\ottnt{f_{{\mathrm{11}}}}  \ottsym{:}   \star \Rightarrow  \unskip ^ { \ell }  \! \star  \,  \xrightarrow{ \mathmakebox[0.4em]{} [  ] \mathmakebox[0.3em]{} }  \, \ottnt{f_{{\mathrm{11}}}}$.
          By \rnp{E\_Step}, $\ottnt{E}  [  \ottnt{f_{{\mathrm{11}}}}  \ottsym{:}   \star \Rightarrow  \unskip ^ { \ell }  \! \star   ] \,  \xmapsto{ \mathmakebox[0.4em]{} [  ] \mathmakebox[0.3em]{} }  \, [  ]  \ottsym{(}  \ottnt{E}  [  \ottnt{f_{{\mathrm{11}}}}  ]  \ottsym{)}$.
          Let $S''  \ottsym{=}  S$, then $S  \ottsym{(}  \ottnt{E}  [  \ottnt{f_{{\mathrm{11}}}}  ]  \ottsym{)}  \ottsym{=}  S  \ottsym{(}  \ottnt{E}  \ottsym{)}  [  S  \ottsym{(}  \ottnt{f_{{\mathrm{11}}}}  \ottsym{)}  ]$.

          \case{$\ottnt{U_{{\mathrm{1}}}}  \ottsym{=}  \ottnt{G}$ for some $\ottnt{G}$}
          We are given $S  \ottsym{(}  \ottnt{f_{{\mathrm{1}}}}  \ottsym{)}  \ottsym{=}  S  \ottsym{(}  \ottnt{f_{{\mathrm{11}}}}  \ottsym{)}  \ottsym{:}   \star \Rightarrow  \unskip ^ { \ell }  \! \ottnt{G} $.

          By Lemma \ref{lem:canonical_forms},
          $\ottnt{f_{{\mathrm{11}}}}  \ottsym{=}  w_{{\mathrm{11}}}  \ottsym{:}   \ottnt{G'} \Rightarrow  \unskip ^ { \ell' }  \! \star $ for some $w_{{\mathrm{11}}}$, $\ottnt{G'}$, and $\ell'$.
          Here, $\ottnt{f_{{\mathrm{1}}}}  \ottsym{=}  w_{{\mathrm{11}}}  \ottsym{:}   \ottnt{G'} \Rightarrow  \unskip ^ { \ell' }  \!  \star \Rightarrow  \unskip ^ { \ell }  \! \ottnt{G}  $ and $S  \ottsym{(}  \ottnt{f_{{\mathrm{1}}}}  \ottsym{)}  \ottsym{=}  S  \ottsym{(}  w_{{\mathrm{11}}}  \ottsym{)}  \ottsym{:}   \ottnt{G'} \Rightarrow  \unskip ^ { \ell' }  \!  \star \Rightarrow  \unskip ^ { \ell }  \! \ottnt{G}  $.

          \begin{caseanalysis}
            \case{$\ottnt{G}  \ottsym{=}  \ottnt{G'}$}
            We are given $S  \ottsym{(}  w_{{\mathrm{11}}}  \ottsym{)}  \ottsym{:}   \ottnt{G'} \Rightarrow  \unskip ^ { \ell' }  \!  \star \Rightarrow  \unskip ^ { \ell }  \! \ottnt{G}   \,  \xrightarrow{ \mathmakebox[0.4em]{} [  ] \mathmakebox[0.3em]{} }  \, S  \ottsym{(}  w_{{\mathrm{11}}}  \ottsym{)}$.
            By \rnp{R\_Succeed}, $w_{{\mathrm{11}}}  \ottsym{:}   \ottnt{G'} \Rightarrow  \unskip ^ { \ell' }  \!  \star \Rightarrow  \unskip ^ { \ell }  \! \ottnt{G}   \,  \xrightarrow{ \mathmakebox[0.4em]{} [  ] \mathmakebox[0.3em]{} }  \, w_{{\mathrm{11}}}$.
            By \rnp{E\_Step}, $\ottnt{E}  [  w_{{\mathrm{11}}}  \ottsym{:}   \ottnt{G'} \Rightarrow  \unskip ^ { \ell' }  \!  \star \Rightarrow  \unskip ^ { \ell }  \! \ottnt{G}    ] \,  \xmapsto{ \mathmakebox[0.4em]{} [  ] \mathmakebox[0.3em]{} }  \, [  ]  \ottsym{(}  \ottnt{E}  [  w_{{\mathrm{11}}}  ]  \ottsym{)}$.
            Let $S''  \ottsym{=}  S$, then $S  \ottsym{(}  \ottnt{E}  [  w_{{\mathrm{11}}}  ]  \ottsym{)}  \ottsym{=}  S  \ottsym{(}  \ottnt{E}  \ottsym{)}  [  S  \ottsym{(}  w_{{\mathrm{11}}}  \ottsym{)}  ]$.

            \case{$\ottnt{G}  \neq  \ottnt{G'}$}
            We are given $S  \ottsym{(}  w_{{\mathrm{11}}}  \ottsym{)}  \ottsym{:}   \ottnt{G'} \Rightarrow  \unskip ^ { \ell' }  \!  \star \Rightarrow  \unskip ^ { \ell }  \! \ottnt{G}   \,  \xrightarrow{ \mathmakebox[0.4em]{} [  ] \mathmakebox[0.3em]{} }  \, \textsf{\textup{blame}\relax} \, \ell$.
            Contradiction.
          \end{caseanalysis}

          \case{$\ottnt{U_{{\mathrm{1}}}}  \ottsym{=}  \ottmv{X}$ for some $\ottmv{X}$}
          We are given $S  \ottsym{(}  \ottnt{f_{{\mathrm{1}}}}  \ottsym{)}  \ottsym{=}  S  \ottsym{(}  \ottnt{f_{{\mathrm{11}}}}  \ottsym{)}  \ottsym{:}   \star \Rightarrow  \unskip ^ { \ell }  \! S  \ottsym{(}  \ottmv{X}  \ottsym{)} $.

          \begin{caseanalysis}
            \case{$S  \ottsym{(}  \ottmv{X}  \ottsym{)}  \ottsym{=}  \iota$ for some $\iota$}
            By Lemma \ref{lem:canonical_forms},
            $\ottnt{f_{{\mathrm{11}}}}  \ottsym{=}  w_{{\mathrm{11}}}  \ottsym{:}   \ottnt{G} \Rightarrow  \unskip ^ { \ell' }  \! \star $ for some $w_{{\mathrm{11}}}$, $\ottnt{G}$, and $\ell'$
            and $w_{{\mathrm{11}}}$ is a value.

            \begin{caseanalysis}
              \case{$\ottnt{G}  \ottsym{=}  \iota$}
              We are given $S  \ottsym{(}  w_{{\mathrm{11}}}  \ottsym{)}  \ottsym{:}   \iota \Rightarrow  \unskip ^ { \ell' }  \!  \star \Rightarrow  \unskip ^ { \ell }  \! \iota   \,  \xrightarrow{ \mathmakebox[0.4em]{} [  ] \mathmakebox[0.3em]{} }  \, S  \ottsym{(}  w_{{\mathrm{11}}}  \ottsym{)}$.
              By \rnp{R\_InstBase}, $w_{{\mathrm{11}}}  \ottsym{:}   \iota \Rightarrow  \unskip ^ { \ell' }  \!  \star \Rightarrow  \unskip ^ { \ell }  \! \ottmv{X}   \,  \xrightarrow{ \mathmakebox[0.4em]{} S' \mathmakebox[0.3em]{} }  \, w_{{\mathrm{11}}}$
              where $S'  \ottsym{=}  [  \ottmv{X}  :=  \iota  ]$.
              By \rnp{E\_Step}, $\ottnt{E}  [  w_{{\mathrm{11}}}  \ottsym{:}   \iota \Rightarrow  \unskip ^ { \ell' }  \!  \star \Rightarrow  \unskip ^ { \ell }  \! \ottmv{X}    ] \,  \xmapsto{ \mathmakebox[0.4em]{} S' \mathmakebox[0.3em]{} }  \, S'  \ottsym{(}  \ottnt{E}  [  w_{{\mathrm{11}}}  ]  \ottsym{)}$.
              Let $S''  \ottsym{=}  S$, then $S  \ottsym{(}  \ottmv{X}  \ottsym{)}  \ottsym{=}   S''  \circ  S'   \ottsym{(}  \ottmv{X}  \ottsym{)}$,
              $S''  \ottsym{(}  S'  \ottsym{(}  \ottnt{E}  [  w_{{\mathrm{11}}}  ]  \ottsym{)}  \ottsym{)}  \ottsym{=}  S  \ottsym{(}  \ottnt{E}  \ottsym{)}  [  S  \ottsym{(}  w_{{\mathrm{11}}}  \ottsym{)}  ]$,
              and $S  \ottsym{(}  \ottnt{E}  [  w_{{\mathrm{11}}}  \ottsym{:}   \iota \Rightarrow  \unskip ^ { \ell' }  \!  \star \Rightarrow  \unskip ^ { \ell }  \! \ottmv{X}    ]  \ottsym{)}  \ottsym{=}   S''  \circ  S'   \ottsym{(}  \ottnt{E}  [  w_{{\mathrm{11}}}  \ottsym{:}   \iota \Rightarrow  \unskip ^ { \ell' }  \!  \star \Rightarrow  \unskip ^ { \ell }  \! \ottmv{X}    ]  \ottsym{)}$.

              \case{$\ottnt{G}  \ottsym{=}  \iota'$ for some $\iota'$ where $\iota  \neq  \iota'$}
              We are given $S  \ottsym{(}  w_{{\mathrm{11}}}  \ottsym{)}  \ottsym{:}   \iota' \Rightarrow  \unskip ^ { \ell' }  \!  \star \Rightarrow  \unskip ^ { \ell }  \! \iota   \,  \xrightarrow{ \mathmakebox[0.4em]{} [  ] \mathmakebox[0.3em]{} }  \, \textsf{\textup{blame}\relax} \, \ell$.
              Contradiction.

              \case{$\ottnt{G'}  \ottsym{=}  \star  \!\rightarrow\!  \star$}
              We are given $S  \ottsym{(}  w_{{\mathrm{11}}}  \ottsym{)}  \ottsym{:}   \star  \!\rightarrow\!  \star \Rightarrow  \unskip ^ { \ell' }  \!  \star \Rightarrow  \unskip ^ { \ell }  \! \iota   \,  \xrightarrow{ \mathmakebox[0.4em]{} [  ] \mathmakebox[0.3em]{} }  \, \textsf{\textup{blame}\relax} \, \ell$.
              Contradiction.
          \end{caseanalysis}

          \case{$S  \ottsym{(}  \ottmv{X}  \ottsym{)}  \ottsym{=}  \ottnt{T_{{\mathrm{11}}}}  \!\rightarrow\!  \ottnt{T_{{\mathrm{12}}}}$ for some $\ottnt{T_{{\mathrm{11}}}}$ and $\ottnt{T_{{\mathrm{12}}}}$}
          \leavevmode\\
          We are given $S  \ottsym{(}  \ottnt{f_{{\mathrm{11}}}}  \ottsym{)}  \ottsym{:}   \star \Rightarrow  \unskip ^ { \ell }  \! \ottnt{T_{{\mathrm{11}}}}  \!\rightarrow\!  \ottnt{T_{{\mathrm{12}}}}  \,  \xrightarrow{ \mathmakebox[0.4em]{} [  ] \mathmakebox[0.3em]{} }  \, S  \ottsym{(}  \ottnt{f_{{\mathrm{11}}}}  \ottsym{)}  \ottsym{:}   \star \Rightarrow  \unskip ^ { \ell }  \!  \star  \!\rightarrow\!  \star \Rightarrow  \unskip ^ { \ell }  \! \ottnt{T_{{\mathrm{11}}}}  \!\rightarrow\!  \ottnt{T_{{\mathrm{12}}}}  $.

          By Lemma \ref{lem:canonical_forms},
          $\ottnt{f_{{\mathrm{11}}}}  \ottsym{=}  w_{{\mathrm{11}}}  \ottsym{:}   \ottnt{G} \Rightarrow  \unskip ^ { \ell' }  \! \star $ for some $w_{{\mathrm{11}}}$, $\ottnt{G}$, and $\ell'$
          and $w_{{\mathrm{11}}}$ is a value.

          \begin{caseanalysis}
            \case{$\ottnt{G}  \ottsym{=}  \iota$ for some $\iota$}
            We are given
            $S  \ottsym{(}  w_{{\mathrm{11}}}  \ottsym{)}  \ottsym{:}   \iota \Rightarrow  \unskip ^ { \ell' }  \!  \star \Rightarrow  \unskip ^ { \ell }  \! \ottnt{T_{{\mathrm{11}}}}  \!\rightarrow\!  \ottnt{T_{{\mathrm{12}}}}   \,  \xrightarrow{ \mathmakebox[0.4em]{} [  ] \mathmakebox[0.3em]{} }  \, S  \ottsym{(}  w_{{\mathrm{1}}}  \ottsym{)}  \ottsym{:}   \iota \Rightarrow  \unskip ^ { \ell' }  \!  \star \Rightarrow  \unskip ^ { \ell }  \!  \star  \!\rightarrow\!  \star \Rightarrow  \unskip ^ { \ell }  \! \ottnt{T_{{\mathrm{11}}}}  \!\rightarrow\!  \ottnt{T_{{\mathrm{12}}}}   $.
            By \rnp{E\_Step},
            $S  \ottsym{(}  \ottnt{E}  \ottsym{)}  [  S  \ottsym{(}  w_{{\mathrm{11}}}  \ottsym{)}  \ottsym{:}   \iota \Rightarrow  \unskip ^ { \ell' }  \!  \star \Rightarrow  \unskip ^ { \ell }  \! \ottnt{T_{{\mathrm{11}}}}  \!\rightarrow\!  \ottnt{T_{{\mathrm{12}}}}    ] \,  \xmapsto{ \mathmakebox[0.4em]{} [  ] \mathmakebox[0.3em]{} }  \, S  \ottsym{(}  \ottnt{E}  \ottsym{)}  [  S  \ottsym{(}  w_{{\mathrm{1}}}  \ottsym{)}  \ottsym{:}   \iota \Rightarrow  \unskip ^ { \ell' }  \!  \star \Rightarrow  \unskip ^ { \ell }  \!  \star  \!\rightarrow\!  \star \Rightarrow  \unskip ^ { \ell }  \! \ottnt{T_{{\mathrm{11}}}}  \!\rightarrow\!  \ottnt{T_{{\mathrm{12}}}}     ]$.
            By \rnp{E\_Fail} and \rnp{E\_Step},
            $S  \ottsym{(}  \ottnt{E}  \ottsym{)}  [  S  \ottsym{(}  w_{{\mathrm{1}}}  \ottsym{)}  \ottsym{:}   \iota \Rightarrow  \unskip ^ { \ell' }  \!  \star \Rightarrow  \unskip ^ { \ell }  \!  \star  \!\rightarrow\!  \star \Rightarrow  \unskip ^ { \ell }  \! \ottnt{T_{{\mathrm{11}}}}  \!\rightarrow\!  \ottnt{T_{{\mathrm{12}}}}     ] \,  \xmapsto{ \mathmakebox[0.4em]{} [  ] \mathmakebox[0.3em]{} }  \, S  \ottsym{(}  \ottnt{E}  \ottsym{)}  [  \textsf{\textup{blame}\relax} \, \ell  \ottsym{:}   \star  \!\rightarrow\!  \star \Rightarrow  \unskip ^ { \ell }  \! \ottnt{T_{{\mathrm{11}}}}  \!\rightarrow\!  \ottnt{T_{{\mathrm{12}}}}   ]$.
            By \rnp{E\_Abort},
            $S  \ottsym{(}  \ottnt{E}  \ottsym{)}  [  \textsf{\textup{blame}\relax} \, \ell  \ottsym{:}   \star  \!\rightarrow\!  \star \Rightarrow  \unskip ^ { \ell }  \! \ottnt{T_{{\mathrm{11}}}}  \!\rightarrow\!  \ottnt{T_{{\mathrm{12}}}}   ] \,  \xmapsto{ \mathmakebox[0.4em]{} [  ] \mathmakebox[0.3em]{} }  \, \textsf{\textup{blame}\relax} \, \ell$.

            Contradiction.

            \case{$\ottnt{G'}  \ottsym{=}  \star  \!\rightarrow\!  \star$}
            We are given
            $S  \ottsym{(}  w_{{\mathrm{11}}}  \ottsym{)}  \ottsym{:}   \star  \!\rightarrow\!  \star \Rightarrow  \unskip ^ { \ell' }  \!  \star \Rightarrow  \unskip ^ { \ell }  \! \ottnt{T_{{\mathrm{11}}}}  \!\rightarrow\!  \ottnt{T_{{\mathrm{12}}}}   \,  \xrightarrow{ \mathmakebox[0.4em]{} [  ] \mathmakebox[0.3em]{} }  \, S  \ottsym{(}  w_{{\mathrm{11}}}  \ottsym{)}  \ottsym{:}   \star  \!\rightarrow\!  \star \Rightarrow  \unskip ^ { \ell' }  \!  \star \Rightarrow  \unskip ^ { \ell }  \!  \star  \!\rightarrow\!  \star \Rightarrow  \unskip ^ { \ell }  \! \ottnt{T_{{\mathrm{11}}}}  \!\rightarrow\!  \ottnt{T_{{\mathrm{12}}}}   $.
            By \rnp{R\_InstArrow},
            $w_{{\mathrm{11}}}  \ottsym{:}   \star  \!\rightarrow\!  \star \Rightarrow  \unskip ^ { \ell' }  \!  \star \Rightarrow  \unskip ^ { \ell }  \! \ottmv{X}   \,  \xrightarrow{ \mathmakebox[0.4em]{} S' \mathmakebox[0.3em]{} }  \, w_{{\mathrm{11}}}  \ottsym{:}   \star  \!\rightarrow\!  \star \Rightarrow  \unskip ^ { \ell' }  \!  \star \Rightarrow  \unskip ^ { \ell }  \!  \star  \!\rightarrow\!  \star \Rightarrow  \unskip ^ { \ell }  \! \ottmv{X_{{\mathrm{1}}}}  \!\rightarrow\!  \ottmv{X_{{\mathrm{2}}}}   $
            where $S'  \ottsym{=}  [  \ottmv{X}  :=  \ottmv{X_{{\mathrm{1}}}}  \!\rightarrow\!  \ottmv{X_{{\mathrm{2}}}}  ]$.
            By \rnp{E\_Step},
            $\ottnt{E}  [  w_{{\mathrm{11}}}  \ottsym{:}   \star  \!\rightarrow\!  \star \Rightarrow  \unskip ^ { \ell' }  \!  \star \Rightarrow  \unskip ^ { \ell }  \! \ottmv{X}    ] \,  \xmapsto{ \mathmakebox[0.4em]{} S' \mathmakebox[0.3em]{} }  \, S'  \ottsym{(}  \ottnt{E}  [  w_{{\mathrm{11}}}  \ottsym{:}   \star  \!\rightarrow\!  \star \Rightarrow  \unskip ^ { \ell' }  \!  \star \Rightarrow  \unskip ^ { \ell }  \!  \star  \!\rightarrow\!  \star \Rightarrow  \unskip ^ { \ell }  \! \ottmv{X_{{\mathrm{1}}}}  \!\rightarrow\!  \ottmv{X_{{\mathrm{2}}}}     ]  \ottsym{)}$.

            Since we can suppose that $\ottmv{X_{{\mathrm{1}}}}, \ottmv{X_{{\mathrm{2}}}} \, \not\in \, \textit{dom} \, \ottsym{(}  S  \ottsym{)}$,
            we can let $S''  \ottsym{=}   S  \uplus  [  \ottmv{X_{{\mathrm{1}}}}  :=  \ottnt{T_{{\mathrm{11}}}}  \ottsym{,}  \ottmv{X_{{\mathrm{2}}}}  :=  \ottnt{T_{{\mathrm{12}}}}  ] $.
            By definition $S  \ottsym{(}  \ottmv{X}  \ottsym{)}  \ottsym{=}  \ottsym{(}   S''  \circ  S'   \ottsym{)}  \ottsym{(}  \ottmv{X}  \ottsym{)}$.
            $\ottmv{X_{{\mathrm{1}}}}$ and $\ottmv{X_{{\mathrm{2}}}}$ do not appear in $\ottnt{f}$ or $S$.
            So, $\ottsym{(}   S''  \circ  S'   \ottsym{)}  \ottsym{(}  \ottnt{E}  [  w_{{\mathrm{11}}}  \ottsym{:}   \star  \!\rightarrow\!  \star \Rightarrow  \unskip ^ { \ell' }  \!  \star \Rightarrow  \unskip ^ { \ell }  \!  \star  \!\rightarrow\!  \star \Rightarrow  \unskip ^ { \ell }  \! \ottmv{X_{{\mathrm{1}}}}  \!\rightarrow\!  \ottmv{X_{{\mathrm{2}}}}     ]  \ottsym{)}  \ottsym{=}  S  \ottsym{(}  \ottnt{E}  \ottsym{)}  [  S  \ottsym{(}  w_{{\mathrm{11}}}  \ottsym{)}  \ottsym{:}   \star  \!\rightarrow\!  \star \Rightarrow  \unskip ^ { \ell' }  \!  \star \Rightarrow  \unskip ^ { \ell }  \!  \star  \!\rightarrow\!  \star \Rightarrow  \unskip ^ { \ell }  \! \ottnt{T_{{\mathrm{11}}}}  \!\rightarrow\!  \ottnt{T_{{\mathrm{12}}}}     ]$.
          \end{caseanalysis}

          \case{$S  \ottsym{(}  \ottmv{X}  \ottsym{)}  \ottsym{=}  \ottmv{X'}$ for some $\ottmv{X'}$}
          By Lemma \ref{lem:canonical_forms},
          $\ottnt{f_{{\mathrm{11}}}}  \ottsym{=}  w_{{\mathrm{11}}}  \ottsym{:}   \ottnt{G} \Rightarrow  \unskip ^ { \ell' }  \! \star $ for some $w_{{\mathrm{11}}}$, $\ottnt{G}$, and $\ell'$
          and $w_{{\mathrm{11}}}$ is a value.
          \begin{caseanalysis}
           \case{$\ottnt{G}  \ottsym{=}  \iota$}
           We are given $S  \ottsym{(}  w_{{\mathrm{11}}}  \ottsym{)}  \ottsym{:}   \iota \Rightarrow  \unskip ^ { \ell' }  \!  \star \Rightarrow  \unskip ^ { \ell }  \! \ottmv{X'}   \,  \xrightarrow{ \mathmakebox[0.4em]{} [  \ottmv{X'}  :=  \iota  ] \mathmakebox[0.3em]{} }  \, S  \ottsym{(}  w_{{\mathrm{11}}}  \ottsym{)}$.
           By \rnp{E\_Step},
           $S  \ottsym{(}  \ottnt{E}  \ottsym{)}  [  S  \ottsym{(}  w_{{\mathrm{11}}}  \ottsym{)}  \ottsym{:}   \iota \Rightarrow  \unskip ^ { \ell' }  \!  \star \Rightarrow  \unskip ^ { \ell }  \! \ottmv{X'}    ] \,  \xmapsto{ \mathmakebox[0.4em]{} [  \ottmv{X'}  :=  \iota  ] \mathmakebox[0.3em]{} }  \,  [  \ottmv{X'}  :=  \iota  ]  \circ  S   \ottsym{(}  \ottnt{E}  \ottsym{)}  [  S  \ottsym{(}  w_{{\mathrm{11}}}  \ottsym{)}  ]$.
           We have $ [  \ottmv{X'}  :=  \iota  ]  \circ  S   \ottsym{(}  \ottnt{E}  [  w_{{\mathrm{11}}}  ]  \ottsym{)}  \ottsym{=}   [  \ottmv{X'}  :=  \iota  ]  \circ  S   \ottsym{(}  \ottnt{E}  [  w_{{\mathrm{11}}}  ]  [  \ottmv{X}  \ottsym{:=}  \iota  ]  \ottsym{)}$.
           By \rnp{R\_InstBase}/\rnp{E\_Step},
           $\ottnt{E}  [  w_{{\mathrm{11}}}  \ottsym{:}   \iota \Rightarrow  \unskip ^ { \ell' }  \!  \star \Rightarrow  \unskip ^ { \ell }  \! \ottmv{X}    ] \,  \xmapsto{ \mathmakebox[0.4em]{} [  \ottmv{X}  :=  \iota  ] \mathmakebox[0.3em]{} }  \, \ottnt{E}  [  w_{{\mathrm{11}}}  ]  [  \ottmv{X}  \ottsym{:=}  \iota  ]$.
           We finish by letting $S''  \ottsym{=}   [  \ottmv{X'}  :=  \iota  ]  \circ  S $.

           \case{$\ottnt{G}  \ottsym{=}  \star  \!\rightarrow\!  \star$}
           We are given $S  \ottsym{(}  w_{{\mathrm{11}}}  \ottsym{)}  \ottsym{:}   \star  \!\rightarrow\!  \star \Rightarrow  \unskip ^ { \ell' }  \!  \star \Rightarrow  \unskip ^ { \ell }  \! \ottmv{X'}   \,  \xrightarrow{ \mathmakebox[0.4em]{} [  \ottmv{X'}  :=  \ottmv{X_{{\mathrm{1}}}}  \!\rightarrow\!  \ottmv{X_{{\mathrm{2}}}}  ] \mathmakebox[0.3em]{} }  \, S  \ottsym{(}  w_{{\mathrm{11}}}  \ottsym{)}  \ottsym{:}   \star  \!\rightarrow\!  \star \Rightarrow  \unskip ^ { \ell' }  \!  \star \Rightarrow  \unskip ^ { \ell }  \! \ottmv{X_{{\mathrm{1}}}}  \!\rightarrow\!  \ottmv{X_{{\mathrm{2}}}}  $
           for some fresh $\ottmv{X_{{\mathrm{1}}}}$ and $\ottmv{X_{{\mathrm{2}}}}$.
           By \rnp{E\_Step},
           $S  \ottsym{(}  \ottnt{E}  \ottsym{)}  [  S  \ottsym{(}  w_{{\mathrm{11}}}  \ottsym{)}  \ottsym{:}   \star  \!\rightarrow\!  \star \Rightarrow  \unskip ^ { \ell' }  \!  \star \Rightarrow  \unskip ^ { \ell }  \! \ottmv{X'}    ] \,  \xmapsto{ \mathmakebox[0.4em]{} [  \ottmv{X'}  :=  \ottmv{X_{{\mathrm{1}}}}  \!\rightarrow\!  \ottmv{X_{{\mathrm{2}}}}  ] \mathmakebox[0.3em]{} }  \,  [  \ottmv{X'}  :=  \ottmv{X_{{\mathrm{1}}}}  \!\rightarrow\!  \ottmv{X_{{\mathrm{2}}}}  ]  \circ  S   \ottsym{(}  \ottnt{E}  \ottsym{)}  [  S  \ottsym{(}  w_{{\mathrm{11}}}  \ottsym{)}  \ottsym{:}   \star  \!\rightarrow\!  \star \Rightarrow  \unskip ^ { \ell' }  \!  \star \Rightarrow  \unskip ^ { \ell }  \! \ottmv{X_{{\mathrm{1}}}}  \!\rightarrow\!  \ottmv{X_{{\mathrm{2}}}}    ]$.
           Without loss of generality, we can suppose that $\ottmv{X_{{\mathrm{1}}}}, \ottmv{X_{{\mathrm{2}}}} \, \not\in \, \textit{dom} \, \ottsym{(}  S  \ottsym{)}$.
           Thus, $ [  \ottmv{X'}  :=  \ottmv{X_{{\mathrm{1}}}}  \!\rightarrow\!  \ottmv{X_{{\mathrm{2}}}}  ]  \circ  S   \ottsym{(}  \ottnt{E}  \ottsym{)}  [  S  \ottsym{(}  w_{{\mathrm{11}}}  \ottsym{)}  \ottsym{:}   \star  \!\rightarrow\!  \star \Rightarrow  \unskip ^ { \ell' }  \!  \star \Rightarrow  \unskip ^ { \ell }  \! \ottmv{X_{{\mathrm{1}}}}  \!\rightarrow\!  \ottmv{X_{{\mathrm{2}}}}    ]  \ottsym{=}   [  \ottmv{X'}  :=  \ottmv{X_{{\mathrm{1}}}}  \!\rightarrow\!  \ottmv{X_{{\mathrm{2}}}}  ]  \circ  S   \ottsym{(}  \ottnt{E}  [  w_{{\mathrm{11}}}  \ottsym{:}   \star  \!\rightarrow\!  \star \Rightarrow  \unskip ^ { \ell' }  \!  \star \Rightarrow  \unskip ^ { \ell }  \! \ottmv{X_{{\mathrm{1}}}}  \!\rightarrow\!  \ottmv{X_{{\mathrm{2}}}}    ]  \ottsym{)} =
           \ottsym{(}   [  \ottmv{X'}  :=  \ottmv{X_{{\mathrm{1}}}}  \!\rightarrow\!  \ottmv{X_{{\mathrm{2}}}}  ]  \circ  S   \ottsym{)}  \ottsym{(}  [  \ottmv{X}  :=  \ottmv{X_{{\mathrm{1}}}}  \!\rightarrow\!  \ottmv{X_{{\mathrm{2}}}}  ]  \ottsym{(}  \ottnt{E}  [  w_{{\mathrm{11}}}  \ottsym{:}   \star  \!\rightarrow\!  \star \Rightarrow  \unskip ^ { \ell' }  \!  \star \Rightarrow  \unskip ^ { \ell }  \! \ottmv{X_{{\mathrm{1}}}}  \!\rightarrow\!  \ottmv{X_{{\mathrm{2}}}}    ]  \ottsym{)}  \ottsym{)}$.
           By \rnp{R\_InstArrow}/\rnp{E\_Step},
           $\ottnt{E}  [  w_{{\mathrm{11}}}  \ottsym{:}   \star  \!\rightarrow\!  \star \Rightarrow  \unskip ^ { \ell' }  \!  \star \Rightarrow  \unskip ^ { \ell }  \! \ottmv{X}    ] \,  \xmapsto{ \mathmakebox[0.4em]{} [  \ottmv{X}  :=  \ottmv{X_{{\mathrm{1}}}}  \!\rightarrow\!  \ottmv{X_{{\mathrm{2}}}}  ] \mathmakebox[0.3em]{} }  \, [  \ottmv{X}  :=  \ottmv{X_{{\mathrm{1}}}}  \!\rightarrow\!  \ottmv{X_{{\mathrm{2}}}}  ]  \ottsym{(}  \ottnt{E}  [  w_{{\mathrm{11}}}  \ottsym{:}   \star  \!\rightarrow\!  \star \Rightarrow  \unskip ^ { \ell' }  \!  \star \Rightarrow  \unskip ^ { \ell }  \! \ottmv{X_{{\mathrm{1}}}}  \!\rightarrow\!  \ottmv{X_{{\mathrm{2}}}}    ]  \ottsym{)}$.
           We finish by letting $S''  \ottsym{=}   [  \ottmv{X'}  :=  \ottmv{X_{{\mathrm{1}}}}  \!\rightarrow\!  \ottmv{X_{{\mathrm{2}}}}  ]  \circ  S $.
          \end{caseanalysis}

        \end{caseanalysis}

        \case{$\ottnt{U_{{\mathrm{1}}}}  \ottsym{=}  \ottnt{U_{{\mathrm{11}}}}  \!\rightarrow\!  \ottnt{U_{{\mathrm{12}}}}$ for some $\ottnt{U_{{\mathrm{11}}}}$ and $\ottnt{U_{{\mathrm{12}}}}$ where $\ottnt{U_{{\mathrm{1}}}}  \neq  \star  \!\rightarrow\!  \star$}
        We are given $S  \ottsym{(}  \ottnt{f_{{\mathrm{1}}}}  \ottsym{)}  \ottsym{=}  \ottnt{f_{{\mathrm{11}}}}  \ottsym{:}   \star \Rightarrow  \unskip ^ { \ell }  \! S  \ottsym{(}  \ottnt{U_{{\mathrm{11}}}}  \!\rightarrow\!  \ottnt{U_{{\mathrm{12}}}}  \ottsym{)} $ and
        $S  \ottsym{(}  \ottnt{f_{{\mathrm{11}}}}  \ottsym{)}  \ottsym{:}   \star \Rightarrow  \unskip ^ { \ell }  \! S  \ottsym{(}  \ottnt{U_{{\mathrm{11}}}}  \!\rightarrow\!  \ottnt{U_{{\mathrm{12}}}}  \ottsym{)}  \,  \xrightarrow{ \mathmakebox[0.4em]{} [  ] \mathmakebox[0.3em]{} }  \, S  \ottsym{(}  \ottnt{f_{{\mathrm{11}}}}  \ottsym{)}  \ottsym{:}   \star \Rightarrow  \unskip ^ { \ell }  \!  \star  \!\rightarrow\!  \star \Rightarrow  \unskip ^ { \ell }  \! S  \ottsym{(}  \ottnt{U_{{\mathrm{11}}}}  \!\rightarrow\!  \ottnt{U_{{\mathrm{12}}}}  \ottsym{)}  $.
        By \rnp{R\_Expand},
        $\ottnt{f_{{\mathrm{11}}}}  \ottsym{:}   \star \Rightarrow  \unskip ^ { \ell }  \! \ottnt{U_{{\mathrm{11}}}}  \!\rightarrow\!  \ottnt{U_{{\mathrm{12}}}}  \,  \xrightarrow{ \mathmakebox[0.4em]{} [  ] \mathmakebox[0.3em]{} }  \, \ottnt{f_{{\mathrm{11}}}}  \ottsym{:}   \star \Rightarrow  \unskip ^ { \ell }  \!  \star  \!\rightarrow\!  \star \Rightarrow  \unskip ^ { \ell }  \! \ottnt{U_{{\mathrm{11}}}}  \!\rightarrow\!  \ottnt{U_{{\mathrm{12}}}}  $.
        By \rnp{E\_Step},
        $\ottnt{E}  [  \ottnt{f_{{\mathrm{11}}}}  \ottsym{:}   \star \Rightarrow  \unskip ^ { \ell }  \! \ottnt{U_{{\mathrm{11}}}}  \!\rightarrow\!  \ottnt{U_{{\mathrm{12}}}}   ] \,  \xmapsto{ \mathmakebox[0.4em]{} [  ] \mathmakebox[0.3em]{} }  \, [  ]  \ottsym{(}  \ottnt{E}  [  \ottnt{f_{{\mathrm{11}}}}  \ottsym{:}   \star \Rightarrow  \unskip ^ { \ell }  \!  \star  \!\rightarrow\!  \star \Rightarrow  \unskip ^ { \ell }  \! \ottnt{U_{{\mathrm{11}}}}  \!\rightarrow\!  \ottnt{U_{{\mathrm{12}}}}    ]  \ottsym{)}$.
        Let $S''  \ottsym{=}  S$, then
        $S''  \ottsym{(}  \ottnt{E}  [  \ottnt{f_{{\mathrm{11}}}}  \ottsym{:}   \star \Rightarrow  \unskip ^ { \ell }  \!  \star  \!\rightarrow\!  \star \Rightarrow  \unskip ^ { \ell }  \! \ottnt{U_{{\mathrm{11}}}}  \!\rightarrow\!  \ottnt{U_{{\mathrm{12}}}}    ]  \ottsym{)}  \ottsym{=}  S  \ottsym{(}  \ottnt{E}  \ottsym{)}  [  S  \ottsym{(}  \ottnt{f_{{\mathrm{11}}}}  \ottsym{)}  \ottsym{:}   \star \Rightarrow  \unskip ^ { \ell }  \!  \star  \!\rightarrow\!  \star \Rightarrow  \unskip ^ { \ell }  \! S  \ottsym{(}  \ottnt{U_{{\mathrm{11}}}}  \!\rightarrow\!  \ottnt{U_{{\mathrm{12}}}}  \ottsym{)}    ]$.
      \end{caseanalysis}

      \case{\rnp{C\_DynR}}
      We are given $\ottnt{f_{{\mathrm{1}}}}  \ottsym{=}  \ottnt{f_{{\mathrm{11}}}}  \ottsym{:}   \ottnt{U'_{{\mathrm{1}}}} \Rightarrow  \unskip ^ { \ell }  \! \star $.
      By definition, $S  \ottsym{(}  \ottnt{f_{{\mathrm{1}}}}  \ottsym{)}  \ottsym{=}  S  \ottsym{(}  \ottnt{f_{{\mathrm{11}}}}  \ottsym{)}  \ottsym{:}   S  \ottsym{(}  \ottnt{U'_{{\mathrm{1}}}}  \ottsym{)} \Rightarrow  \unskip ^ { \ell }  \! \star $.
      It must be the case that $S  \ottsym{(}  \ottnt{f_{{\mathrm{11}}}}  \ottsym{)}$ is a value.
      By Lemma \ref{lem:subst_value_is_value}, $\ottnt{f_{{\mathrm{11}}}}$ is a value.

      \begin{caseanalysis}
        \case{$\ottnt{U'_{{\mathrm{1}}}}  \ottsym{=}  \star$}
        We are given $S  \ottsym{(}  \ottnt{f_{{\mathrm{1}}}}  \ottsym{)}  \ottsym{=}  S  \ottsym{(}  \ottnt{f_{{\mathrm{11}}}}  \ottsym{)}  \ottsym{:}   \star \Rightarrow  \unskip ^ { \ell }  \! \star $ and
        $S  \ottsym{(}  \ottnt{f_{{\mathrm{11}}}}  \ottsym{)}  \ottsym{:}   \star \Rightarrow  \unskip ^ { \ell }  \! \star  \,  \xrightarrow{ \mathmakebox[0.4em]{} [  ] \mathmakebox[0.3em]{} }  \, S  \ottsym{(}  \ottnt{f_{{\mathrm{11}}}}  \ottsym{)}$.
        By \rnp{R\_IdStar}, $\ottnt{f_{{\mathrm{11}}}}  \ottsym{:}   \star \Rightarrow  \unskip ^ { \ell }  \! \star  \,  \xrightarrow{ \mathmakebox[0.4em]{} [  ] \mathmakebox[0.3em]{} }  \, \ottnt{f_{{\mathrm{11}}}}$.
        By \rnp{E\_Step}, $\ottnt{E}  [  \ottnt{f_{{\mathrm{11}}}}  \ottsym{:}   \star \Rightarrow  \unskip ^ { \ell }  \! \star   ] \,  \xrightarrow{ \mathmakebox[0.4em]{} [  ] \mathmakebox[0.3em]{} }  \, [  ]  \ottsym{(}  \ottnt{E}  [  \ottnt{f_{{\mathrm{11}}}}  ]  \ottsym{)}$.
        Let $S''  \ottsym{=}  S$, then $S''  \ottsym{(}  \ottnt{E}  [  \ottnt{f_{{\mathrm{11}}}}  ]  \ottsym{)}  \ottsym{=}  S  \ottsym{(}  \ottnt{E}  \ottsym{)}  [  S  \ottsym{(}  \ottnt{f_{{\mathrm{11}}}}  \ottsym{)}  ]$.

        \case{$\ottnt{U'_{{\mathrm{1}}}}  \ottsym{=}  \iota$ for some $\iota$}
        We are given $S  \ottsym{(}  \ottnt{f_{{\mathrm{1}}}}  \ottsym{)}  \ottsym{=}  S  \ottsym{(}  \ottnt{f_{{\mathrm{11}}}}  \ottsym{)}  \ottsym{:}   \iota \Rightarrow  \unskip ^ { \ell }  \! \star $.
        Contradiction.

        \case{$\ottnt{U'_{{\mathrm{1}}}}  \ottsym{=}  \ottmv{X}$ for some $\ottmv{X}$}
        We are given $S  \ottsym{(}  \ottnt{f_{{\mathrm{1}}}}  \ottsym{)}  \ottsym{=}  S  \ottsym{(}  \ottnt{f_{{\mathrm{11}}}}  \ottsym{)}  \ottsym{:}   S  \ottsym{(}  \ottmv{X}  \ottsym{)} \Rightarrow  \unskip ^ { \ell }  \! \star $.
        By Lemma \ref{lem:canonical_forms}, $\ottnt{f_{{\mathrm{11}}}}$ is not typed at type variable.
        Contradiction.

        \case{$\ottnt{U'_{{\mathrm{1}}}}  \ottsym{=}  \star  \!\rightarrow\!  \star$}
        We are given $S  \ottsym{(}  \ottnt{f_{{\mathrm{1}}}}  \ottsym{)}  \ottsym{=}  S  \ottsym{(}  \ottnt{f_{{\mathrm{11}}}}  \ottsym{)}  \ottsym{:}   \star  \!\rightarrow\!  \star \Rightarrow  \unskip ^ { \ell }  \! \star $.
        Contradiction.

        \case{$\ottnt{U'_{{\mathrm{1}}}}  \ottsym{=}  \ottnt{U'_{{\mathrm{11}}}}  \!\rightarrow\!  \ottnt{U'_{{\mathrm{12}}}}$ for some $\ottnt{U'_{{\mathrm{11}}}}$ and $\ottnt{U'_{{\mathrm{12}}}}$}
        \leavevmode\\
        We are given $S  \ottsym{(}  \ottnt{f_{{\mathrm{1}}}}  \ottsym{)}  \ottsym{=}  S  \ottsym{(}  \ottnt{f_{{\mathrm{11}}}}  \ottsym{)}  \ottsym{:}   S  \ottsym{(}  \ottnt{U'_{{\mathrm{11}}}}  \!\rightarrow\!  \ottnt{U'_{{\mathrm{12}}}}  \ottsym{)} \Rightarrow  \unskip ^ { \ell }  \! \star $ and
        $S  \ottsym{(}  \ottnt{f_{{\mathrm{11}}}}  \ottsym{)}  \ottsym{:}   S  \ottsym{(}  \ottnt{U'_{{\mathrm{11}}}}  \!\rightarrow\!  \ottnt{U'_{{\mathrm{12}}}}  \ottsym{)} \Rightarrow  \unskip ^ { \ell }  \! \star  \,  \xrightarrow{ \mathmakebox[0.4em]{} [  ] \mathmakebox[0.3em]{} }  \, S  \ottsym{(}  \ottnt{f_{{\mathrm{11}}}}  \ottsym{)}  \ottsym{:}   S  \ottsym{(}  \ottnt{U'_{{\mathrm{11}}}}  \!\rightarrow\!  \ottnt{U'_{{\mathrm{12}}}}  \ottsym{)} \Rightarrow  \unskip ^ { \ell }  \!  \star  \!\rightarrow\!  \star \Rightarrow  \unskip ^ { \ell }  \! \star  $.
        By \rnp{R\_Ground}, $\ottnt{f_{{\mathrm{11}}}}  \ottsym{:}   \ottnt{U'_{{\mathrm{11}}}}  \!\rightarrow\!  \ottnt{U'_{{\mathrm{12}}}} \Rightarrow  \unskip ^ { \ell }  \! \star  \,  \xrightarrow{ \mathmakebox[0.4em]{} [  ] \mathmakebox[0.3em]{} }  \, \ottnt{f_{{\mathrm{11}}}}  \ottsym{:}   \ottnt{U'_{{\mathrm{11}}}}  \!\rightarrow\!  \ottnt{U'_{{\mathrm{12}}}} \Rightarrow  \unskip ^ { \ell }  \!  \star  \!\rightarrow\!  \star \Rightarrow  \unskip ^ { \ell }  \! \star  $.
        By \rnp{E\_Step}, $\ottnt{E}  [  \ottnt{f_{{\mathrm{11}}}}  \ottsym{:}   \ottnt{U'_{{\mathrm{11}}}}  \!\rightarrow\!  \ottnt{U'_{{\mathrm{12}}}} \Rightarrow  \unskip ^ { \ell }  \! \star   ] \,  \xrightarrow{ \mathmakebox[0.4em]{} [  ] \mathmakebox[0.3em]{} }  \, [  ]  \ottsym{(}  \ottnt{E}  [  \ottnt{f_{{\mathrm{11}}}}  \ottsym{:}   \ottnt{U'_{{\mathrm{11}}}}  \!\rightarrow\!  \ottnt{U'_{{\mathrm{12}}}} \Rightarrow  \unskip ^ { \ell }  \!  \star  \!\rightarrow\!  \star \Rightarrow  \unskip ^ { \ell }  \! \star    ]  \ottsym{)}$.
        Let $S''  \ottsym{=}  S$, then
        $S''  \ottsym{(}  \ottnt{E}  [  \ottnt{f_{{\mathrm{11}}}}  \ottsym{:}   \ottnt{U'_{{\mathrm{11}}}}  \!\rightarrow\!  \ottnt{U'_{{\mathrm{12}}}} \Rightarrow  \unskip ^ { \ell }  \!  \star  \!\rightarrow\!  \star \Rightarrow  \unskip ^ { \ell }  \! \star    ]  \ottsym{)}  \ottsym{=}  S  \ottsym{(}  \ottnt{E}  \ottsym{)}  [  S  \ottsym{(}  \ottnt{f_{{\mathrm{11}}}}  \ottsym{)}  \ottsym{:}   S  \ottsym{(}  \ottnt{U'_{{\mathrm{11}}}}  \!\rightarrow\!  \ottnt{U'_{{\mathrm{12}}}}  \ottsym{)} \Rightarrow  \unskip ^ { \ell }  \!  \star  \!\rightarrow\!  \star \Rightarrow  \unskip ^ { \ell }  \! \star    ]$.
      \end{caseanalysis}

      \case{\rnp{C\_Arrow}}
      We are given $\ottnt{f_{{\mathrm{1}}}}  \ottsym{=}  \ottnt{f_{{\mathrm{11}}}}  \ottsym{:}   \ottnt{U'_{{\mathrm{11}}}}  \!\rightarrow\!  \ottnt{U'_{{\mathrm{12}}}} \Rightarrow  \unskip ^ { \ell }  \! \ottnt{U_{{\mathrm{11}}}}  \!\rightarrow\!  \ottnt{U_{{\mathrm{12}}}} $ for some $\ottnt{U_{{\mathrm{11}}}}$, $\ottnt{U_{{\mathrm{12}}}}$,
      $\ottnt{U'_{{\mathrm{11}}}}$, $\ottnt{U'_{{\mathrm{12}}}}$.
      By definition, $S  \ottsym{(}  \ottnt{f_{{\mathrm{1}}}}  \ottsym{)}  \ottsym{=}  S  \ottsym{(}  \ottnt{f_{{\mathrm{11}}}}  \ottsym{)}  \ottsym{:}   S  \ottsym{(}  \ottnt{U'_{{\mathrm{11}}}}  \ottsym{)}  \!\rightarrow\!  S  \ottsym{(}  \ottnt{U'_{{\mathrm{12}}}}  \ottsym{)} \Rightarrow  \unskip ^ { \ell }  \! S  \ottsym{(}  \ottnt{U_{{\mathrm{11}}}}  \ottsym{)}  \!\rightarrow\!  S  \ottsym{(}  \ottnt{U_{{\mathrm{12}}}}  \ottsym{)} $.
      Contradiction. 
    \end{caseanalysis}

      \case{\rnp{T\_LetP}}
      We are given $ \emptyset   \vdash   \textsf{\textup{let}\relax} \,  \ottmv{x}  =   \Lambda    \overrightarrow{ \ottmv{X_{\ottmv{i}}} }  .\,  w_{{\mathrm{11}}}   \textsf{\textup{ in }\relax}  \ottnt{f_{{\mathrm{12}}}}   \ottsym{:}  \ottnt{U}$
      for some $\ottmv{x}$, $ \overrightarrow{ \ottmv{X_{\ottmv{i}}} } $, $w_{{\mathrm{11}}}$, and $\ottnt{f_{{\mathrm{12}}}}$.
      By inversion, we have $ \emptyset   \vdash  w_{{\mathrm{11}}}  \ottsym{:}  \ottnt{U_{{\mathrm{11}}}}$,
      $  \emptyset  ,   \ottmv{x}  :  \forall \,  \overrightarrow{ \ottmv{X_{\ottmv{i}}} }   \ottsym{.}  \ottnt{U_{{\mathrm{11}}}}    \vdash  \ottnt{f_{{\mathrm{12}}}}  \ottsym{:}  \ottnt{U_{{\mathrm{1}}}}$, and $ \overrightarrow{ \ottmv{X_{\ottmv{i}}} }   \cap  \textit{ftv} \, \ottsym{(}  \Gamma  \ottsym{)}  \ottsym{=}   \emptyset $.

      By case analysis on the reduction rule applied to $S  \ottsym{(}  \ottnt{f_{{\mathrm{1}}}}  \ottsym{)}$.
      \begin{caseanalysis}
        \case{\rnp{R\_LetP}}
        We are given $ \textsf{\textup{let}\relax} \,  \ottmv{x}  =   \Lambda    \overrightarrow{ \ottmv{X_{\ottmv{i}}} }  .\,  w'_{{\mathrm{11}}}   \textsf{\textup{ in }\relax}  \ottnt{f'_{{\mathrm{12}}}}  \,  \xrightarrow{ \mathmakebox[0.4em]{} [  ] \mathmakebox[0.3em]{} }  \, \ottnt{f'_{{\mathrm{12}}}}  [  \ottmv{x}  \ottsym{:=}   \Lambda    \overrightarrow{ \ottmv{X_{\ottmv{i}}} }  .\,  w'_{{\mathrm{11}}}   ]$ where
        $S  \ottsym{(}   \textsf{\textup{let}\relax} \,  \ottmv{x}  =   \Lambda    \overrightarrow{ \ottmv{X_{\ottmv{i}}} }  .\,  w_{{\mathrm{11}}}   \textsf{\textup{ in }\relax}  \ottnt{f_{{\mathrm{12}}}}   \ottsym{)}  \ottsym{=}   \textsf{\textup{let}\relax} \,  \ottmv{x}  =   \Lambda    \overrightarrow{ \ottmv{X_{\ottmv{i}}} }  .\,  w'_{{\mathrm{11}}}   \textsf{\textup{ in }\relax}  \ottnt{f'_{{\mathrm{12}}}} $ and
        $\ottnt{f'_{{\mathrm{1}}}}  \ottsym{=}  \ottnt{f'_{{\mathrm{12}}}}  [  \ottmv{x}  \ottsym{:=}   \Lambda    \overrightarrow{ \ottmv{X_{\ottmv{i}}} }  .\,  w'_{{\mathrm{11}}}   ]$.
        By \rnp{R\_LetP},
        $ \textsf{\textup{let}\relax} \,  \ottmv{x}  =   \Lambda    \overrightarrow{ \ottmv{X_{\ottmv{i}}} }  .\,  w_{{\mathrm{11}}}   \textsf{\textup{ in }\relax}  \ottnt{f_{{\mathrm{12}}}}  \,  \xrightarrow{ \mathmakebox[0.4em]{} [  ] \mathmakebox[0.3em]{} }  \, \ottnt{f_{{\mathrm{12}}}}  [  \ottmv{x}  \ottsym{:=}   \Lambda    \overrightarrow{ \ottmv{X_{\ottmv{i}}} }  .\,  w_{{\mathrm{11}}}   ]$.
        Also, $S  \ottsym{(}  \ottnt{f_{{\mathrm{12}}}}  [  \ottmv{x}  \ottsym{:=}   \Lambda    \overrightarrow{ \ottmv{X_{\ottmv{i}}} }  .\,  w_{{\mathrm{11}}}   ]  \ottsym{)}  \ottsym{=}  \ottnt{f'_{{\mathrm{12}}}}  [  \ottmv{x}  \ottsym{:=}   \Lambda    \overrightarrow{ \ottmv{X_{\ottmv{i}}} }  .\,  w'_{{\mathrm{11}}}   ]$.
        By \rnp{E\_Step}, $\ottnt{E}  [   \textsf{\textup{let}\relax} \,  \ottmv{x}  =   \Lambda    \overrightarrow{ \ottmv{X_{\ottmv{i}}} }  .\,  w_{{\mathrm{11}}}   \textsf{\textup{ in }\relax}  \ottnt{f_{{\mathrm{12}}}}   ] \,  \xmapsto{ \mathmakebox[0.4em]{} [  ] \mathmakebox[0.3em]{} }  \, [  ]  \ottsym{(}  \ottnt{E}  [  \ottnt{f_{{\mathrm{12}}}}  [  \ottmv{x}  \ottsym{:=}   \Lambda    \overrightarrow{ \ottmv{X_{\ottmv{i}}} }  .\,  w_{{\mathrm{11}}}   ]  ]  \ottsym{)}$.
        Let $S''  \ottsym{=}  S$, then $S''  \ottsym{(}  \ottnt{E}  [  \ottnt{f_{{\mathrm{12}}}}  [  \ottmv{x}  \ottsym{:=}   \Lambda    \overrightarrow{ \ottmv{X_{\ottmv{i}}} }  .\,  w_{{\mathrm{11}}}   ]  ]  \ottsym{)}  \ottsym{=}  S  \ottsym{(}  \ottnt{E}  \ottsym{)}  [  \ottnt{f'_{{\mathrm{12}}}}  [  \ottmv{x}  \ottsym{:=}   \Lambda    \overrightarrow{ \ottmv{X_{\ottmv{i}}} }  .\,  w'_{{\mathrm{11}}}   ]  ]$.
      \end{caseanalysis}
    \end{caseanalysis}

    %
    %

    \case{\rnp{E\_Abort}}
    Here, $\ottnt{f'}  \ottsym{=}  \textsf{\textup{blame}\relax} \, \ell$ for some $\ell$.  Contradiction.
\qedhere
  \end{caseanalysis}
\end{proof}

\begin{lemmaA} \label{lem:completeness_value}
  If $ \emptyset   \vdash  \ottnt{f}  \ottsym{:}  \ottnt{U}$ and $S  \ottsym{(}  \ottnt{f}  \ottsym{)} \,  \xmapsto{ \mathmakebox[0.4em]{} S' \mathmakebox[0.3em]{} }\hspace{-0.4em}{}^\ast \hspace{0.2em}  \, w$, 
  then $\ottnt{f} \,  \xmapsto{ \mathmakebox[0.4em]{} S'' \mathmakebox[0.3em]{} }\hspace{-0.4em}{}^\ast \hspace{0.2em}  \, w'$ and $S'''  \ottsym{(}  w'  \ottsym{)}  \ottsym{=}  w$
  for some $w'$, $S''$, and $S'''$.
\end{lemmaA}

\begin{proof}
  By mathematical induction on the length of the evaluation sequence.

  \begin{caseanalysis}
    \case{the length is 0}
    Here, $S  \ottsym{(}  \ottnt{f}  \ottsym{)}  \ottsym{=}  w$.
    By Lemma~\ref{lem:subst_value_is_value}, $\ottnt{f}$ is a value $w$.
    Then, $w \,  \xmapsto{ \mathmakebox[0.4em]{} [  ] \mathmakebox[0.3em]{} }\hspace{-0.4em}{}^\ast \hspace{0.2em}  \, w$.

    \case{the length is more than 0}
    There exist $\ottnt{f'}$, $S'_{{\mathrm{1}}}$, and $S'_{{\mathrm{2}}}$ such that $S  \ottsym{(}  \ottnt{f}  \ottsym{)} \,  \xmapsto{ \mathmakebox[0.4em]{} S'_{{\mathrm{1}}} \mathmakebox[0.3em]{} }  \, \ottnt{f'}$ and $\ottnt{f'} \,  \xmapsto{ \mathmakebox[0.4em]{} S'_{{\mathrm{2}}} \mathmakebox[0.3em]{} }\hspace{-0.4em}{}^\ast \hspace{0.2em}  \, w$
    and $S' =  S'_{{\mathrm{1}}}  \circ  S'_{{\mathrm{2}}} $.

    By Lemma~\ref{lem:completeness_value_eval_step},
    there exist $\ottnt{f''}$, $S_{{\mathrm{1}}}$, and $S'_{{\mathrm{1}}}$
    such that $\ottnt{f} \,  \xmapsto{ \mathmakebox[0.4em]{} S_{{\mathrm{1}}} \mathmakebox[0.3em]{} }  \, \ottnt{f''}$ and $S'_{{\mathrm{1}}}  \ottsym{(}  \ottnt{f''}  \ottsym{)}  \ottsym{=}  \ottnt{f'}$.

    Here, $S'_{{\mathrm{1}}}  \ottsym{(}  \ottnt{f''}  \ottsym{)} \,  \xmapsto{ \mathmakebox[0.4em]{} S'_{{\mathrm{2}}} \mathmakebox[0.3em]{} }\hspace{-0.4em}{}^\ast \hspace{0.2em}  \, w$
    and its evaluation sequence is shorter than
    that of $S  \ottsym{(}  \ottnt{f}  \ottsym{)} \,  \xmapsto{ \mathmakebox[0.4em]{} S' \mathmakebox[0.3em]{} }  \, w$.
    Also,  by Lemma~\ref{lem:preservation}, $ \emptyset   \vdash  \ottnt{f''}  \ottsym{:}  \ottnt{U}$.
    By the IH, $\ottnt{f''} \,  \xmapsto{ \mathmakebox[0.4em]{} S_{{\mathrm{2}}} \mathmakebox[0.3em]{} }\hspace{-0.4em}{}^\ast \hspace{0.2em}  \, w'$ and $S'_{{\mathrm{2}}}  \ottsym{(}  w'  \ottsym{)}  \ottsym{=}  w$ for some $S_{{\mathrm{2}}}$ and $S'_{{\mathrm{2}}}$.

    Finally, $\ottnt{f} \,  \xmapsto{ \mathmakebox[0.4em]{}  S_{{\mathrm{2}}}  \circ  S_{{\mathrm{1}}}  \mathmakebox[0.3em]{} }  \, w'$ and $S'_{{\mathrm{2}}}  \ottsym{(}  w'  \ottsym{)}  \ottsym{=}  w$.
\qedhere
  \end{caseanalysis}
\end{proof}

\ifrestate
\thmCompleteness*
\else
\begin{theoremA}[name=Completeness of Dynamic Type Inference,restate=thmCompleteness] \label{thm:completeness}
  Suppose $ \emptyset   \vdash  \ottnt{f}  \ottsym{:}  \ottnt{U}$. 
  \begin{enumerate}
    \item If $S  \ottsym{(}  \ottnt{f}  \ottsym{)} \,  \xmapsto{ \mathmakebox[0.4em]{} S' \mathmakebox[0.3em]{} }\hspace{-0.4em}{}^\ast \hspace{0.2em}  \, w$,
      then $\ottnt{f} \,  \xmapsto{ \mathmakebox[0.4em]{} S'' \mathmakebox[0.3em]{} }\hspace{-0.4em}{}^\ast \hspace{0.2em}  \, w'$ and
      $S'''  \ottsym{(}  w'  \ottsym{)}  \ottsym{=}  w$ for some $w'$, $S''$, and $S'''$.
    \item If $ S  \ottsym{(}  \ottnt{f}  \ottsym{)} \!  \Uparrow  $, then $ \ottnt{f} \!  \Uparrow  $.
  \end{enumerate}
\end{theoremA}
\fi 

\begin{proof}
  \leavevmode
  \begin{enumerate}
    \item
      By Lemma~\ref{lem:completeness_value}.

    \item Let $n$ be an arbitrary natural number and suppose $S  \ottsym{(}  \ottnt{f}  \ottsym{)}  \xmapsto{ \mathmakebox[0.4em]{} S'_{{\mathrm{1}}} \mathmakebox[0.3em]{} }  \ottnt{f'_{{\mathrm{1}}}}  \xmapsto{ \mathmakebox[0.4em]{} S'_{{\mathrm{2}}} \mathmakebox[0.3em]{} } \cdots  \xmapsto{ \mathmakebox[0.4em]{} S'_{\ottmv{n}} \mathmakebox[0.3em]{} }  \ottnt{f'_{\ottmv{n}}}$.
      By repeatedly applying Lemma \ref{lem:completeness_value_eval_step}, 
      we obtain  $\ottnt{f}  \xmapsto{ \mathmakebox[0.4em]{} S_{{\mathrm{1}}} \mathmakebox[0.3em]{} }  \ottnt{f_{{\mathrm{1}}}}  \xmapsto{ \mathmakebox[0.4em]{} S_{{\mathrm{2}}} \mathmakebox[0.3em]{} } \cdots  \xmapsto{ \mathmakebox[0.4em]{} S_{\ottmv{n}} \mathmakebox[0.3em]{} }  \ottnt{f_{\ottmv{n}}}$ and
      type substitutions $S''_{\ottmv{i}}$ such that $S''_{\ottmv{i}}  \ottsym{(}  \ottnt{f_{\ottmv{i}}}  \ottsym{)}  \ottsym{=}  \ottnt{f'_{\ottmv{i}}}$.
\qedhere
  \end{enumerate}
\end{proof}

\subsection{Soundness of Dynamic Type Inference}
\begin{lemmaA} \label{lem:context_inversion}
  If $ \emptyset   \vdash  \ottnt{E}  [  \ottnt{f_{{\mathrm{1}}}}  ]  \ottsym{:}  \ottnt{U}$,
  then $ \emptyset   \vdash  \ottnt{f_{{\mathrm{1}}}}  \ottsym{:}  \ottnt{U_{{\mathrm{1}}}}$ for some $\ottnt{U_{{\mathrm{1}}}}$.
\end{lemmaA}

\begin{proof}
  By induction on the structure of $\ottnt{E}$.
\end{proof}

\begin{lemmaA} \label{lem:value_any_subst}
  If $ \emptyset   \vdash  w  \ottsym{:}  \ottnt{U}$ and $w$ is a value,
  then $S  \ottsym{(}  w  \ottsym{)}$ is a value for any $S$.
\end{lemmaA}

\begin{proof}
  By case analysis on the structure of $w$.
\end{proof}

\begin{lemmaA} \label{lem:reduce_step_any_fresh_subst}
  If
  $\ottnt{f} \,  \xrightarrow{ \mathmakebox[0.4em]{} S \mathmakebox[0.3em]{} }  \, \ottnt{f'}$ and
  $\textit{dom} \, \ottsym{(}  S  \ottsym{)}  \cap  \textit{dom} \, \ottsym{(}  S'  \ottsym{)}  \ottsym{=}   \emptyset $ and
  $\textit{dom} \, \ottsym{(}  S'  \ottsym{)}$ are disjoint from type variables generated by
  $\ottnt{f} \,  \xrightarrow{ \mathmakebox[0.4em]{} S \mathmakebox[0.3em]{} }  \, \ottnt{f'}$,
  then
  $S'  \ottsym{(}  \ottnt{f}  \ottsym{)} \,  \xrightarrow{ \mathmakebox[0.4em]{} S \mathmakebox[0.3em]{} }  \, S'  \ottsym{(}  \ottnt{f'}  \ottsym{)}$.
\end{lemmaA}
\begin{proof}
 By case analysis on the reduction rule applied to derive $\ottnt{f} \,  \xrightarrow{ \mathmakebox[0.4em]{} S \mathmakebox[0.3em]{} }  \, \ottnt{f'}$.
 \begin{caseanalysis}
  \case{\rnp{R\_OP}, \rnp{R\_Beta}, \rnp{R\_IdBase}, \rnp{R\_IdStar}, \rnp{R\_Succeed}, \rnp{R\_Fail}, and \rnp{R\_AppCast}} Easy.
  \case{\rnp{R\_Ground}}
  We are given $w  \ottsym{:}   \ottnt{U} \Rightarrow  \unskip ^ { \ell }  \! \star  \,  \xrightarrow{ \mathmakebox[0.4em]{} [  ] \mathmakebox[0.3em]{} }  \, w  \ottsym{:}   \ottnt{U} \Rightarrow  \unskip ^ { \ell }  \!  \ottnt{G} \Rightarrow  \unskip ^ { \ell }  \! \star  $
  for some $w$, $\ottnt{U}$, $\ell$, and $\ottnt{G}$ such that
  $\ottnt{U}  \neq  \star$ and $\ottnt{U}  \neq  \ottnt{G}$ and $\ottnt{U}  \sim  \ottnt{G}$.
  By case analysis on $\ottnt{U}$.
  \begin{caseanalysis}
   \case{$\ottnt{U}  \ottsym{=}  \iota$} Contradictory since $\ottnt{U}  \neq  \ottnt{G}$ and $\ottnt{U}  \sim  \ottnt{G}$.
   \case{$\ottnt{U}  \ottsym{=}  \star$} Contradictory.
   \case{$\ottnt{U}  \ottsym{=}  \ottmv{X}$} Contradictory since $\ottnt{U}  \sim  \ottnt{G}$.
   \case{$\ottnt{U}  \ottsym{=}  \ottnt{U_{{\mathrm{1}}}}  \!\rightarrow\!  \ottnt{U_{{\mathrm{2}}}}$}
    We have $S'  \ottsym{(}  \ottnt{U}  \ottsym{)}  \ottsym{=}  S'  \ottsym{(}  \ottnt{U_{{\mathrm{1}}}}  \ottsym{)}  \!\rightarrow\!  S'  \ottsym{(}  \ottnt{U_{{\mathrm{2}}}}  \ottsym{)}$ and $\ottnt{G}  \ottsym{=}  \star  \!\rightarrow\!  \star$.
    Since $\ottnt{U}  \neq  \star  \!\rightarrow\!  \star$, we have $S'  \ottsym{(}  \ottnt{U}  \ottsym{)}  \neq  \star  \!\rightarrow\!  \star$.
    Since $S'  \ottsym{(}  w  \ottsym{)}$ is a value by Lemma~\ref{lem:value_any_subst},
    we have $S'  \ottsym{(}  w  \ottsym{:}   \ottnt{U} \Rightarrow  \unskip ^ { \ell }  \! \star   \ottsym{)} \,  \xrightarrow{ \mathmakebox[0.4em]{} [  ] \mathmakebox[0.3em]{} }  \, S'  \ottsym{(}  w  \ottsym{:}   \ottnt{U} \Rightarrow  \unskip ^ { \ell }  \!  \ottnt{G} \Rightarrow  \unskip ^ { \ell }  \! \star    \ottsym{)}$.
  \end{caseanalysis}

  \case{\rnp{R\_Expand}}
  We are given $w  \ottsym{:}   \star \Rightarrow  \unskip ^ { \ell }  \! \ottnt{U}  \,  \xrightarrow{ \mathmakebox[0.4em]{} [  ] \mathmakebox[0.3em]{} }  \, w  \ottsym{:}   \star \Rightarrow  \unskip ^ { \ell }  \!  \ottnt{G} \Rightarrow  \unskip ^ { \ell }  \! \ottnt{U}  $
  for some $w$, $\ottnt{U}$, $\ell$, and $\ottnt{G}$ such that
  $\ottnt{U}  \neq  \star$ and $\ottnt{U}  \neq  \ottnt{G}$ and $\ottnt{U}  \sim  \ottnt{G}$.
  By case analysis on $\ottnt{U}$.
  \begin{caseanalysis}
   \case{$\ottnt{U}  \ottsym{=}  \iota$} Contradictory since $\ottnt{U}  \neq  \ottnt{G}$ and $\ottnt{U}  \sim  \ottnt{G}$.
   \case{$\ottnt{U}  \ottsym{=}  \star$} Contradictory.
   \case{$\ottnt{U}  \ottsym{=}  \ottmv{X}$} Contradictory since $\ottnt{U}  \sim  \ottnt{G}$.
   \case{$\ottnt{U}  \ottsym{=}  \ottnt{U_{{\mathrm{1}}}}  \!\rightarrow\!  \ottnt{U_{{\mathrm{2}}}}$}
    We have $S'  \ottsym{(}  \ottnt{U}  \ottsym{)}  \ottsym{=}  S'  \ottsym{(}  \ottnt{U_{{\mathrm{1}}}}  \ottsym{)}  \!\rightarrow\!  S'  \ottsym{(}  \ottnt{U_{{\mathrm{2}}}}  \ottsym{)}$ and $\ottnt{G}  \ottsym{=}  \star  \!\rightarrow\!  \star$.
    Since $\ottnt{U}  \neq  \star  \!\rightarrow\!  \star$, we have $S'  \ottsym{(}  \ottnt{U}  \ottsym{)}  \neq  \star  \!\rightarrow\!  \star$.
    Since $S'  \ottsym{(}  w  \ottsym{)}$ is a value by Lemma~\ref{lem:value_any_subst},
    we have $S'  \ottsym{(}  w  \ottsym{:}   \star \Rightarrow  \unskip ^ { \ell }  \! \ottnt{U}   \ottsym{)} \,  \xrightarrow{ \mathmakebox[0.4em]{} [  ] \mathmakebox[0.3em]{} }  \, S'  \ottsym{(}  w  \ottsym{:}   \star \Rightarrow  \unskip ^ { \ell }  \!  \ottnt{G} \Rightarrow  \unskip ^ { \ell }  \! \ottnt{U}    \ottsym{)}$.
  \end{caseanalysis}

  \case{\rnp{R\_InstBase}}
  We are given $w  \ottsym{:}   \iota \Rightarrow  \unskip ^ { \ell_{{\mathrm{1}}} }  \!  \star \Rightarrow  \unskip ^ { \ell_{{\mathrm{2}}} }  \! \ottmv{X}   \,  \xrightarrow{ \mathmakebox[0.4em]{} [  \ottmv{X}  :=  \iota  ] \mathmakebox[0.3em]{} }  \, w$
  for some $w$, $\iota$, $\ottmv{X}$, $\ell_{{\mathrm{1}}}$, and $\ell_{{\mathrm{2}}}$.
  Since $ \{  \ottmv{X}  \}   \cap  \textit{dom} \, \ottsym{(}  S'  \ottsym{)}  \ottsym{=}   \emptyset $,
  we have $S'  \ottsym{(}  \ottnt{f}  \ottsym{)}  \ottsym{=}  S'  \ottsym{(}  w  \ottsym{)}  \ottsym{:}   \iota \Rightarrow  \unskip ^ { \ell_{{\mathrm{1}}} }  \!  \star \Rightarrow  \unskip ^ { \ell_{{\mathrm{2}}} }  \! \ottmv{X}  $.
  Thus, $S'  \ottsym{(}  \ottnt{f}  \ottsym{)} \,  \xrightarrow{ \mathmakebox[0.4em]{} [  \ottmv{X}  :=  \iota  ] \mathmakebox[0.3em]{} }  \, S'  \ottsym{(}  w  \ottsym{)}$.

  \case{\rnp{R\_InstArrow}}
  We are given $w  \ottsym{:}   \star  \!\rightarrow\!  \star \Rightarrow  \unskip ^ { \ell_{{\mathrm{1}}} }  \!  \star \Rightarrow  \unskip ^ { \ell_{{\mathrm{2}}} }  \! \ottmv{X}   \,  \xrightarrow{ \mathmakebox[0.4em]{} [  \ottmv{X}  :=  \ottmv{X_{{\mathrm{1}}}}  \!\rightarrow\!  \ottmv{X_{{\mathrm{2}}}}  ] \mathmakebox[0.3em]{} }  \, w  \ottsym{:}   \star  \!\rightarrow\!  \star \Rightarrow  \unskip ^ { \ell_{{\mathrm{1}}} }  \!  \star \Rightarrow  \unskip ^ { \ell_{{\mathrm{2}}} }  \!  \star  \!\rightarrow\!  \star \Rightarrow  \unskip ^ { \ell_{{\mathrm{2}}} }  \! \ottmv{X_{{\mathrm{1}}}}  \!\rightarrow\!  \ottmv{X_{{\mathrm{2}}}}   $
  for some $w$, $\ottmv{X}$, $\ell_{{\mathrm{1}}}$, $\ell_{{\mathrm{2}}}$, and fresh $\ottmv{X_{{\mathrm{1}}}}$ and $\ottmv{X_{{\mathrm{2}}}}$.
  Without loss of generality, we can suppose that $\ottmv{X_{{\mathrm{1}}}}, \ottmv{X_{{\mathrm{2}}}} \, \not\in \, \textit{dom} \, \ottsym{(}  S'  \ottsym{)}$.
  Since $ \{  \ottmv{X}  \}   \cap  \textit{dom} \, \ottsym{(}  S'  \ottsym{)}  \ottsym{=}   \emptyset $,
  we have $S'  \ottsym{(}  \ottnt{f}  \ottsym{)}  \ottsym{=}  S'  \ottsym{(}  w  \ottsym{)}  \ottsym{:}   \star  \!\rightarrow\!  \star \Rightarrow  \unskip ^ { \ell_{{\mathrm{1}}} }  \!  \star \Rightarrow  \unskip ^ { \ell_{{\mathrm{2}}} }  \! \ottmv{X}  $.
  Thus, $S'  \ottsym{(}  \ottnt{f}  \ottsym{)} \,  \xrightarrow{ \mathmakebox[0.4em]{} [  \ottmv{X}  :=  \ottmv{X_{{\mathrm{1}}}}  \!\rightarrow\!  \ottmv{X_{{\mathrm{2}}}}  ] \mathmakebox[0.3em]{} }  \, S'  \ottsym{(}  w  \ottsym{)}  \ottsym{:}   \star  \!\rightarrow\!  \star \Rightarrow  \unskip ^ { \ell_{{\mathrm{1}}} }  \!  \star \Rightarrow  \unskip ^ { \ell_{{\mathrm{2}}} }  \!  \star  \!\rightarrow\!  \star \Rightarrow  \unskip ^ { \ell_{{\mathrm{2}}} }  \! \ottmv{X_{{\mathrm{1}}}}  \!\rightarrow\!  \ottmv{X_{{\mathrm{2}}}}   $.
  Since $S'  \ottsym{(}  w  \ottsym{)}  \ottsym{:}   \star  \!\rightarrow\!  \star \Rightarrow  \unskip ^ { \ell_{{\mathrm{1}}} }  \!  \star \Rightarrow  \unskip ^ { \ell_{{\mathrm{2}}} }  \!  \star  \!\rightarrow\!  \star \Rightarrow  \unskip ^ { \ell_{{\mathrm{2}}} }  \! \ottmv{X_{{\mathrm{1}}}}  \!\rightarrow\!  \ottmv{X_{{\mathrm{2}}}}     \ottsym{=}  S'  \ottsym{(}  w  \ottsym{:}   \star  \!\rightarrow\!  \star \Rightarrow  \unskip ^ { \ell_{{\mathrm{1}}} }  \!  \star \Rightarrow  \unskip ^ { \ell_{{\mathrm{2}}} }  \!  \star  \!\rightarrow\!  \star \Rightarrow  \unskip ^ { \ell_{{\mathrm{2}}} }  \! \ottmv{X_{{\mathrm{1}}}}  \!\rightarrow\!  \ottmv{X_{{\mathrm{2}}}}     \ottsym{)}$, we finish.

  \case{\rnp{R\_LetP}}
  We are given $ \textsf{\textup{let}\relax} \,  \ottmv{x}  =   \Lambda    \overrightarrow{ \ottmv{X} }  .\,  w''   \textsf{\textup{ in }\relax}  \ottnt{f''}  \,  \xrightarrow{ \mathmakebox[0.4em]{} [  ] \mathmakebox[0.3em]{} }  \, \ottnt{f''}  [  \ottmv{x}  \ottsym{:=}   \Lambda    \overrightarrow{ \ottmv{X} }  .\,  w''   ]$.
  Type variables generated by the value substitution are supposed not to be
  captured by $S'$.
  Thus, we finish.  \qedhere
 \end{caseanalysis}
\end{proof}

\begin{lemmaA}[name=,restate=lemEvalStepAnySubst] \label{lem:eval_step_any_subst}
  If $\ottnt{f} \,  \xmapsto{ \mathmakebox[0.4em]{} S \mathmakebox[0.3em]{} }  \, \ottnt{f'}$ and
  $\textit{dom} \, \ottsym{(}  S  \ottsym{)}  \cap  \textit{dom} \, \ottsym{(}  S'  \ottsym{)}  \ottsym{=}   \emptyset $ and
  $\textit{dom} \, \ottsym{(}  S'  \ottsym{)}$ are disjoint from type variables generated by
  $\ottnt{f} \,  \xmapsto{ \mathmakebox[0.4em]{} S \mathmakebox[0.3em]{} }  \, \ottnt{f'}$ and,
  for any $\ottmv{X} \, \in \, \textit{dom} \, \ottsym{(}  S'  \ottsym{)}$, $\textit{ftv} \, \ottsym{(}  S'  \ottsym{(}  \ottmv{X}  \ottsym{)}  \ottsym{)}  \cap  \textit{dom} \, \ottsym{(}  S  \ottsym{)}  \ottsym{=}   \emptyset $,
  then
  $S'  \ottsym{(}  \ottnt{f}  \ottsym{)} \,  \xmapsto{ \mathmakebox[0.4em]{} S \mathmakebox[0.3em]{} }  \, S'  \ottsym{(}  \ottnt{f'}  \ottsym{)}$.
\end{lemmaA}
\begin{proof}
  By case analysis on the evaluation rule applied to $\ottnt{f}$.
  \begin{caseanalysis}
    \case{\rnp{E\_Step}}
    We are given $\ottnt{E}  [  \ottnt{f_{{\mathrm{1}}}}  ] \,  \xmapsto{ \mathmakebox[0.4em]{} S \mathmakebox[0.3em]{} }  \, S  \ottsym{(}  \ottnt{E}  [  \ottnt{f'_{{\mathrm{1}}}}  ]  \ottsym{)}$
    for some $\ottnt{E}$, $\ottnt{f_{{\mathrm{1}}}}$, and $\ottnt{f'_{{\mathrm{1}}}}$ such that
    $\ottnt{f_{{\mathrm{1}}}} \,  \xrightarrow{ \mathmakebox[0.4em]{} S \mathmakebox[0.3em]{} }  \, \ottnt{f'_{{\mathrm{1}}}}$.
    By Lemma~\ref{lem:reduce_step_any_fresh_subst},
    $S'  \ottsym{(}  \ottnt{f_{{\mathrm{1}}}}  \ottsym{)} \,  \xrightarrow{ \mathmakebox[0.4em]{} S \mathmakebox[0.3em]{} }  \, S'  \ottsym{(}  \ottnt{f'_{{\mathrm{1}}}}  \ottsym{)}$.
    By \rnp{E\_Step},
    $S'  \ottsym{(}  \ottnt{E}  [  \ottnt{f_{{\mathrm{1}}}}  ]  \ottsym{)} \,  \xmapsto{ \mathmakebox[0.4em]{} S \mathmakebox[0.3em]{} }  \, S  \ottsym{(}  S'  \ottsym{(}  \ottnt{E}  [  \ottnt{f'_{{\mathrm{1}}}}  ]  \ottsym{)}  \ottsym{)}$.
    By the assumption, $S  \ottsym{(}  S'  \ottsym{(}  \ottnt{E}  [  \ottnt{f'_{{\mathrm{1}}}}  ]  \ottsym{)}  \ottsym{)}  \ottsym{=}  S'  \ottsym{(}  S  \ottsym{(}  \ottnt{E}  [  \ottnt{f'_{{\mathrm{1}}}}  ]  \ottsym{)}  \ottsym{)}$.
    Thus, we finish.

   By Lemma~\ref{lem:reduce_step_any_fresh_subst}
    \case{\rnp{E\_Abort}} Obvious.  \qedhere
  \end{caseanalysis}
\end{proof}

\begin{lemmaA} \label{lem:eval_multi_step_any_fresh_subst}
  If
  $\ottnt{f} \,  \xmapsto{ \mathmakebox[0.4em]{} S \mathmakebox[0.3em]{} }\hspace{-0.4em}{}^\ast \hspace{0.2em}  \, \ottnt{f'}$ and
  $\textit{dom} \, \ottsym{(}  S  \ottsym{)}  \cap  \textit{dom} \, \ottsym{(}  S'  \ottsym{)}  \ottsym{=}   \emptyset $ and
  $\textit{dom} \, \ottsym{(}  S'  \ottsym{)}$ are disjoint from type variables generated by
  $\ottnt{f} \,  \xmapsto{ \mathmakebox[0.4em]{} S \mathmakebox[0.3em]{} }\hspace{-0.4em}{}^\ast \hspace{0.2em}  \, \ottnt{f'}$ and,
  for any $\ottmv{X} \, \in \, \textit{dom} \, \ottsym{(}  S'  \ottsym{)}$, $\textit{ftv} \, \ottsym{(}  S'  \ottsym{(}  \ottmv{X}  \ottsym{)}  \ottsym{)}  \cap  \textit{dom} \, \ottsym{(}  S  \ottsym{)}  \ottsym{=}   \emptyset $,
  then
  $S'  \ottsym{(}  \ottnt{f}  \ottsym{)} \,  \xmapsto{ \mathmakebox[0.4em]{} S \mathmakebox[0.3em]{} }\hspace{-0.4em}{}^\ast \hspace{0.2em}  \, S'  \ottsym{(}  \ottnt{f'}  \ottsym{)}$.
\end{lemmaA}
\begin{proof}
  By induction on the length of the evaluation sequence.
 \begin{caseanalysis}
  \case{the length is 0} Obvious.

  \case{the length is more than 0}
  We are given $\ottnt{f} \,  \xmapsto{ \mathmakebox[0.4em]{} S_{{\mathrm{1}}} \mathmakebox[0.3em]{} }  \, \ottnt{f''}$ and $\ottnt{f''} \,  \xmapsto{ \mathmakebox[0.4em]{} S_{{\mathrm{2}}} \mathmakebox[0.3em]{} }  \, \ottnt{f'}$
  for some $S_{{\mathrm{1}}}$, $S_{{\mathrm{2}}}$, and $\ottnt{f''}$ such that
  $S  \ottsym{=}   S_{{\mathrm{2}}}  \circ  S_{{\mathrm{1}}} $.
  Thus, $\textit{dom} \, \ottsym{(}  S_{{\mathrm{1}}}  \ottsym{)}  \cap  \textit{dom} \, \ottsym{(}  S'  \ottsym{)}  \ottsym{=}   \emptyset $ and,
  for any $\ottmv{X} \, \in \, \textit{dom} \, \ottsym{(}  S'  \ottsym{)}$, $\textit{ftv} \, \ottsym{(}  S'  \ottsym{(}  \ottmv{X}  \ottsym{)}  \ottsym{)}  \cap  \textit{dom} \, \ottsym{(}  S_{{\mathrm{1}}}  \ottsym{)}  \ottsym{=}   \emptyset $.
  Hence, by Lemma~\ref{lem:eval_step_any_subst},
  $S'  \ottsym{(}  \ottnt{f}  \ottsym{)} \,  \xmapsto{ \mathmakebox[0.4em]{} S_{{\mathrm{1}}} \mathmakebox[0.3em]{} }  \, S'  \ottsym{(}  \ottnt{f''}  \ottsym{)}$.
  By the IH, we finish.
  \qedhere
 \end{caseanalysis}
\end{proof}

\ifrestate
\lemSoundnessResult*
\else
\begin{lemmaA} \label{lem:soundness_result}
 If
 $\ottnt{f} \,  \xmapsto{ \mathmakebox[0.4em]{}  S_{{\mathrm{1}}}  \uplus  S_{{\mathrm{2}}}  \mathmakebox[0.3em]{} }\hspace{-0.4em}{}^\ast \hspace{0.2em}  \, \ottnt{r}$ and
 $\textit{dom} \, \ottsym{(}  S_{{\mathrm{1}}}  \ottsym{)}  \subseteq  \textit{ftv} \, \ottsym{(}  \ottnt{f}  \ottsym{)}$,
 then
 $S_{{\mathrm{1}}}  \ottsym{(}  \ottnt{f}  \ottsym{)} \,  \xmapsto{ \mathmakebox[0.4em]{} S'_{{\mathrm{2}}} \mathmakebox[0.3em]{} }\hspace{-0.4em}{}^\ast \hspace{0.2em}  \, \ottnt{r}$ for some $S'_{{\mathrm{2}}}$ such that
 $\textit{dom} \, \ottsym{(}  S'_{{\mathrm{2}}}  \ottsym{)}  \subseteq  \textit{dom} \, \ottsym{(}  S_{{\mathrm{2}}}  \ottsym{)}$ and $\textit{ftv} \, \ottsym{(}  S_{{\mathrm{1}}}  \ottsym{(}  \ottmv{X}  \ottsym{)}  \ottsym{)}  \cap  \textit{dom} \, \ottsym{(}  S'_{{\mathrm{2}}}  \ottsym{)}  \ottsym{=}   \emptyset $
 for any $\ottmv{X} \, \in \, \textit{dom} \, \ottsym{(}  S_{{\mathrm{1}}}  \ottsym{)}$.
\end{lemmaA}
\fi
\begin{proof}
 By mathematical induction on the length of the evaluation sequence.
 \begin{caseanalysis}
  \case{the length is 0} Obvious since $ S_{{\mathrm{1}}}  \uplus  S_{{\mathrm{2}}}   \ottsym{=}  [  ]$.

  \case{the length is more than 0}
  We are given
  \begin{itemize}
   \item $\ottnt{f} \,  \xmapsto{ \mathmakebox[0.4em]{} S'_{{\mathrm{1}}} \mathmakebox[0.3em]{} }  \, \ottnt{f_{{\mathrm{1}}}}$,
   \item $\ottnt{f_{{\mathrm{1}}}} \,  \xmapsto{ \mathmakebox[0.4em]{} S'_{{\mathrm{2}}} \mathmakebox[0.3em]{} }\hspace{-0.4em}{}^\ast \hspace{0.2em}  \, \ottnt{r}$, and
   \item $ S_{{\mathrm{1}}}  \uplus  S_{{\mathrm{2}}}   \ottsym{=}   S'_{{\mathrm{2}}}  \circ  S'_{{\mathrm{1}}} $
  \end{itemize}
  for some $\ottnt{f_{{\mathrm{1}}}}$, $S'_{{\mathrm{1}}}$, and $S'_{{\mathrm{2}}}$.
  By case analysis on the evaluation rule applied to $\ottnt{f} \,  \xmapsto{ \mathmakebox[0.4em]{} S'_{{\mathrm{1}}} \mathmakebox[0.3em]{} }  \, \ottnt{f'}$.
  \begin{caseanalysis}
   \case{\rnp{E\_Blame}} Obvious.
   \case{\rnp{E\_Step}}
   We have
   \begin{itemize}
    \item $\ottnt{f}  \ottsym{=}  \ottnt{E}  [  \ottnt{f'}  ]$,
    \item $\ottnt{f_{{\mathrm{1}}}}  \ottsym{=}  S'_{{\mathrm{1}}}  \ottsym{(}  \ottnt{E}  [  \ottnt{f'_{{\mathrm{1}}}}  ]  \ottsym{)}$, and
    \item $\ottnt{f'} \,  \xrightarrow{ \mathmakebox[0.4em]{} S'_{{\mathrm{1}}} \mathmakebox[0.3em]{} }  \, \ottnt{f'_{{\mathrm{1}}}}$
   \end{itemize}
   for some $\ottnt{E}$, $\ottnt{f'}$, and $\ottnt{f'_{{\mathrm{1}}}}$.
   By case analysis on the reduction rule applied to derive
   $\ottnt{f'} \,  \xrightarrow{ \mathmakebox[0.4em]{} S'_{{\mathrm{1}}} \mathmakebox[0.3em]{} }  \, \ottnt{f'_{{\mathrm{1}}}}$.
   \begin{caseanalysis}
    \case{\rnp{R\_InstBase}}
    We are given
    \begin{itemize}
     \item $\ottnt{f'}  \ottsym{=}  w'  \ottsym{:}   \iota \Rightarrow  \unskip ^ { \ell_{{\mathrm{1}}} }  \!  \star \Rightarrow  \unskip ^ { \ell_{{\mathrm{2}}} }  \! \ottmv{X}  $,
     \item $\ottnt{f'_{{\mathrm{1}}}}  \ottsym{=}  w'$, and
     \item $S'_{{\mathrm{1}}}  \ottsym{=}  [  \ottmv{X}  :=  \iota  ]$
    \end{itemize}
    for some $w'$, $\iota$, $\ell$, and $\ottmv{X}$.

    Since $ S_{{\mathrm{1}}}  \uplus  S_{{\mathrm{2}}}   \ottsym{=}   S'_{{\mathrm{2}}}  \circ  S'_{{\mathrm{1}}} $, $\ottmv{X} \, \in \, \textit{dom} \, \ottsym{(}  S_{{\mathrm{1}}}  \ottsym{)}$ or $\ottmv{X} \, \in \, \textit{dom} \, \ottsym{(}  S_{{\mathrm{2}}}  \ottsym{)}$.
    \begin{caseanalysis}
     \case{$\ottmv{X} \, \in \, \textit{dom} \, \ottsym{(}  S_{{\mathrm{1}}}  \ottsym{)}$}
     Since $ S_{{\mathrm{1}}}  \uplus  S_{{\mathrm{2}}}   \ottsym{=}   S'_{{\mathrm{2}}}  \circ  S'_{{\mathrm{1}}} $,
     we have $S_{{\mathrm{1}}}  \ottsym{(}  \ottmv{X}  \ottsym{)} =  S'_{{\mathrm{2}}}  \circ  S'_{{\mathrm{1}}}   \ottsym{(}  \ottmv{X}  \ottsym{)} = \iota$.
     Thus,
     \[\begin{array}{lll}
      S_{{\mathrm{1}}}  \ottsym{(}  \ottnt{f}  \ottsym{)} & = & S_{{\mathrm{1}}}  \ottsym{(}  \ottnt{E}  \ottsym{)}  [  S_{{\mathrm{1}}}  \ottsym{(}  w'  \ottsym{:}   \iota \Rightarrow  \unskip ^ { \ell_{{\mathrm{1}}} }  \!  \star \Rightarrow  \unskip ^ { \ell_{{\mathrm{2}}} }  \! \ottmv{X}    \ottsym{)}  ] \\
                & = & S_{{\mathrm{1}}}  \ottsym{(}  \ottnt{E}  \ottsym{)}  [  S_{{\mathrm{1}}}  \ottsym{(}  w'  \ottsym{)}  \ottsym{:}   \iota \Rightarrow  \unskip ^ { \ell }  \!  \star \Rightarrow  \unskip ^ { \ell_{{\mathrm{2}}} }  \! \iota    ] \\
                &  \xmapsto{ \mathmakebox[0.4em]{} [  ] \mathmakebox[0.3em]{} }  & S_{{\mathrm{1}}}  \ottsym{(}  \ottnt{E}  \ottsym{)}  [  S_{{\mathrm{1}}}  \ottsym{(}  w'  \ottsym{)}  ]
                 \qquad \text{(since $S_{{\mathrm{1}}}  \ottsym{(}  w'  \ottsym{)}$ is a value by Lemma~\ref{lem:value_any_subst})} \\
                & = & S_{{\mathrm{1}}}  \ottsym{(}  \ottnt{E}  [  \ottnt{f'_{{\mathrm{1}}}}  ]  \ottsym{)} \\
                & = & S_{{\mathrm{1}}}  \ottsym{(}  S'_{{\mathrm{1}}}  \ottsym{(}  \ottnt{E}  [  \ottnt{f'_{{\mathrm{1}}}}  ]  \ottsym{)}  \ottsym{)}
                  \qquad \text{(since $S'_{{\mathrm{1}}}  \ottsym{=}  [  \ottmv{X}  :=  \iota  ]$ and
                                      $S_{{\mathrm{1}}}  \ottsym{(}  \ottmv{X}  \ottsym{)}  \ottsym{=}  \iota$)} \\
       \end{array}\]
     Since $S'_{{\mathrm{1}}}  \ottsym{(}  \ottnt{E}  [  \ottnt{f'_{{\mathrm{1}}}}  ]  \ottsym{)} \,  \xmapsto{ \mathmakebox[0.4em]{} S'_{{\mathrm{2}}} \mathmakebox[0.3em]{} }\hspace{-0.4em}{}^\ast \hspace{0.2em}  \, \ottnt{r}$,
     we can suppose that $\ottmv{X} \, \not\in \, \textit{dom} \, \ottsym{(}  S'_{{\mathrm{2}}}  \ottsym{)}$.
     Thus, $ S'_{{\mathrm{2}}}  \circ  S'_{{\mathrm{1}}}   \ottsym{=}   S'_{{\mathrm{1}}}  \uplus  S'_{{\mathrm{2}}} $.
     Since $ S_{{\mathrm{1}}}  \uplus  S_{{\mathrm{2}}}   \ottsym{=}   S'_{{\mathrm{1}}}  \uplus  S'_{{\mathrm{2}}} $,
     there exist some $S''_{{\mathrm{1}}}$ such that
     \begin{itemize}
      \item $S'_{{\mathrm{2}}}  \ottsym{=}   S''_{{\mathrm{1}}}  \uplus  S_{{\mathrm{2}}} $ and
      \item $S_{{\mathrm{1}}}  \ottsym{=}   S'_{{\mathrm{1}}}  \uplus  S''_{{\mathrm{1}}} $.
     \end{itemize}
     Thus, $S'_{{\mathrm{1}}}  \ottsym{(}  \ottnt{E}  [  \ottnt{f'_{{\mathrm{1}}}}  ]  \ottsym{)} = \ottnt{f_{{\mathrm{1}}}} \,  \xmapsto{ \mathmakebox[0.4em]{}  S''_{{\mathrm{1}}}  \uplus  S_{{\mathrm{2}}}  \mathmakebox[0.3em]{} }\hspace{-0.4em}{}^\ast \hspace{0.2em}  \, \ottnt{r}$.
     Here,
     \[\begin{array}{lll}
      \textit{ftv} \, \ottsym{(}  \ottnt{f_{{\mathrm{1}}}}  \ottsym{)} & = & \textit{ftv} \, \ottsym{(}  S'_{{\mathrm{1}}}  \ottsym{(}  \ottnt{E}  [  \ottnt{f'_{{\mathrm{1}}}}  ]  \ottsym{)}  \ottsym{)} \\
                  & = & \textit{ftv} \, \ottsym{(}  \ottnt{E}  [  \ottnt{f'_{{\mathrm{1}}}}  ]  \ottsym{)}  \setminus   \{  \ottmv{X}  \}  \\
                  & = & \textit{ftv} \, \ottsym{(}  \ottnt{E}  [  w'  ]  \ottsym{)}  \setminus   \{  \ottmv{X}  \}  \\
                  & = & \textit{ftv} \, \ottsym{(}  \ottnt{E}  [  w'  \ottsym{:}   \iota \Rightarrow  \unskip ^ { \ell_{{\mathrm{1}}} }  \!  \star \Rightarrow  \unskip ^ { \ell_{{\mathrm{2}}} }  \! \ottmv{X}    ]  \ottsym{)}  \setminus   \{  \ottmv{X}  \}  \\
                  & = & \textit{ftv} \, \ottsym{(}  \ottnt{E}  [  \ottnt{f'}  ]  \ottsym{)}  \setminus   \{  \ottmv{X}  \}  \\
                  & \supseteq & \textit{dom} \, \ottsym{(}  S_{{\mathrm{1}}}  \ottsym{)}  \setminus   \{  \ottmv{X}  \}  \\
                  & = & \textit{dom} \, \ottsym{(}  S''_{{\mathrm{1}}}  \ottsym{)}
       \end{array}\]
     Thus, by the IH, there exists some $S''_{{\mathrm{2}}}$ such that
     $S''_{{\mathrm{1}}}  \ottsym{(}  S'_{{\mathrm{1}}}  \ottsym{(}  \ottnt{E}  [  \ottnt{f'_{{\mathrm{1}}}}  ]  \ottsym{)}  \ottsym{)} \,  \xmapsto{ \mathmakebox[0.4em]{} S''_{{\mathrm{2}}} \mathmakebox[0.3em]{} }\hspace{-0.4em}{}^\ast \hspace{0.2em}  \, \ottnt{r}$ and
     $\textit{dom} \, \ottsym{(}  S''_{{\mathrm{2}}}  \ottsym{)}  \subseteq  \textit{dom} \, \ottsym{(}  S_{{\mathrm{2}}}  \ottsym{)}$ and
     $\textit{ftv} \, \ottsym{(}  S''_{{\mathrm{1}}}  \ottsym{(}  \ottmv{X'}  \ottsym{)}  \ottsym{)}  \cap  \textit{dom} \, \ottsym{(}  S''_{{\mathrm{2}}}  \ottsym{)}  \ottsym{=}   \emptyset $
     for any $\ottmv{X'} \, \in \, \textit{dom} \, \ottsym{(}  S''_{{\mathrm{1}}}  \ottsym{)}$.
     Then,
     we have $S''_{{\mathrm{1}}}  \ottsym{(}  S'_{{\mathrm{1}}}  \ottsym{(}  \ottnt{E}  [  \ottnt{f'_{{\mathrm{1}}}}  ]  \ottsym{)}  \ottsym{)} = S_{{\mathrm{1}}}  \ottsym{(}  S'_{{\mathrm{1}}}  \ottsym{(}  \ottnt{E}  [  \ottnt{f'_{{\mathrm{1}}}}  ]  \ottsym{)}  \ottsym{)}$ and
     $S_{{\mathrm{1}}}  \ottsym{(}  \ottnt{f}  \ottsym{)} \,  \xmapsto{ \mathmakebox[0.4em]{} [  ] \mathmakebox[0.3em]{} }  \, S_{{\mathrm{1}}}  \ottsym{(}  S'_{{\mathrm{1}}}  \ottsym{(}  \ottnt{E}  [  \ottnt{f'_{{\mathrm{1}}}}  ]  \ottsym{)}  \ottsym{)}$.

     Let $\ottmv{X'} \, \in \, \textit{dom} \, \ottsym{(}  S_{{\mathrm{1}}}  \ottsym{)}$.
     If $\ottmv{X'} \, \in \, \textit{dom} \, \ottsym{(}  S'_{{\mathrm{1}}}  \ottsym{)}$, then $\textit{ftv} \, \ottsym{(}  S_{{\mathrm{1}}}  \ottsym{(}  \ottmv{X'}  \ottsym{)}  \ottsym{)}  \ottsym{=}   \emptyset $.
     If $\ottmv{X'} \, \in \, \textit{dom} \, \ottsym{(}  S''_{{\mathrm{1}}}  \ottsym{)}$, then $\textit{ftv} \, \ottsym{(}  S_{{\mathrm{1}}}  \ottsym{(}  \ottmv{X'}  \ottsym{)}  \ottsym{)}  \cap  \textit{dom} \, \ottsym{(}  S''_{{\mathrm{2}}}  \ottsym{)}  \ottsym{=}   \emptyset $ by the IH.

     \case{$\ottmv{X} \, \in \, \textit{dom} \, \ottsym{(}  S_{{\mathrm{2}}}  \ottsym{)}$}
     Since $\textit{dom} \, \ottsym{(}  S_{{\mathrm{1}}}  \ottsym{)}$ and $\textit{dom} \, \ottsym{(}  S_{{\mathrm{2}}}  \ottsym{)}$ are disjoint,
     $\ottmv{X} \, \not\in \, \textit{dom} \, \ottsym{(}  S_{{\mathrm{1}}}  \ottsym{)}$.
     Thus,
     \[\begin{array}{lll}
      S_{{\mathrm{1}}}  \ottsym{(}  \ottnt{f}  \ottsym{)} & = & S_{{\mathrm{1}}}  \ottsym{(}  \ottnt{E}  [  w'  \ottsym{:}   \iota \Rightarrow  \unskip ^ { \ell_{{\mathrm{1}}} }  \!  \star \Rightarrow  \unskip ^ { \ell_{{\mathrm{2}}} }  \! \ottmv{X}    ]  \ottsym{)} \\
                &  \xmapsto{ \mathmakebox[0.4em]{} S'_{{\mathrm{1}}} \mathmakebox[0.3em]{} }  & S'_{{\mathrm{1}}}  \ottsym{(}  S_{{\mathrm{1}}}  \ottsym{(}  \ottnt{E}  [  w'  ]  \ottsym{)}  \ottsym{)}
                 \qquad \text{($S_{{\mathrm{1}}}  \ottsym{(}  w'  \ottsym{)}$ is a value by Lemma~\ref{lem:value_any_subst})} \\
                & = & S_{{\mathrm{1}}}  \ottsym{(}  S'_{{\mathrm{1}}}  \ottsym{(}  \ottnt{E}  [  w'  ]  \ottsym{)}  \ottsym{)}
                 \qquad \text{(since $S_{{\mathrm{1}}}$ is generated by DTI)} \\
                & = & S_{{\mathrm{1}}}  \ottsym{(}  \ottnt{f_{{\mathrm{1}}}}  \ottsym{)}
       \end{array}
       \]
    Since $S'_{{\mathrm{1}}}  \ottsym{(}  \ottnt{E}  [  \ottnt{f'_{{\mathrm{1}}}}  ]  \ottsym{)} \,  \xmapsto{ \mathmakebox[0.4em]{} S'_{{\mathrm{2}}} \mathmakebox[0.3em]{} }\hspace{-0.4em}{}^\ast \hspace{0.2em}  \, \ottnt{r}$,
    we can suppose that $\ottmv{X} \, \not\in \, \textit{dom} \, \ottsym{(}  S'_{{\mathrm{2}}}  \ottsym{)}$.
    Thus, $ S'_{{\mathrm{2}}}  \circ  S'_{{\mathrm{1}}}   \ottsym{=}   S'_{{\mathrm{1}}}  \uplus  S'_{{\mathrm{2}}} $.
    Since $ S_{{\mathrm{1}}}  \uplus  S_{{\mathrm{2}}}   \ottsym{=}   S'_{{\mathrm{1}}}  \uplus  S'_{{\mathrm{2}}} $,
    there exist some $S''_{{\mathrm{2}}}$ such that
    \begin{itemize}
     \item $S'_{{\mathrm{2}}}  \ottsym{=}   S_{{\mathrm{1}}}  \uplus  S''_{{\mathrm{2}}} $ and
     \item $S_{{\mathrm{2}}}  \ottsym{=}   S'_{{\mathrm{1}}}  \uplus  S''_{{\mathrm{2}}} $.
    \end{itemize}
    Thus, $S'_{{\mathrm{1}}}  \ottsym{(}  \ottnt{E}  [  \ottnt{f'_{{\mathrm{1}}}}  ]  \ottsym{)} = \ottnt{f_{{\mathrm{1}}}} \,  \xmapsto{ \mathmakebox[0.4em]{}  S_{{\mathrm{1}}}  \uplus  S''_{{\mathrm{2}}}  \mathmakebox[0.3em]{} }\hspace{-0.4em}{}^\ast \hspace{0.2em}  \, \ottnt{r}$.
    Here,
    \[\begin{array}{lll}
     \textit{ftv} \, \ottsym{(}  \ottnt{f_{{\mathrm{1}}}}  \ottsym{)} & = & \textit{ftv} \, \ottsym{(}  S'_{{\mathrm{1}}}  \ottsym{(}  \ottnt{E}  [  \ottnt{f'_{{\mathrm{1}}}}  ]  \ottsym{)}  \ottsym{)} \\
                 & = & \textit{ftv} \, \ottsym{(}  \ottnt{E}  [  \ottnt{f'}  ]  \ottsym{)}  \setminus   \{  \ottmv{X}  \} 
                  \qquad (\text{by the discussion above}) \\
                 & \supseteq & \textit{dom} \, \ottsym{(}  S_{{\mathrm{1}}}  \ottsym{)}  \setminus   \{  \ottmv{X}  \}  \\
                 & = & \textit{dom} \, \ottsym{(}  S_{{\mathrm{1}}}  \ottsym{)}
                  \qquad (\text{since $\ottmv{X} \, \not\in \, \textit{dom} \, \ottsym{(}  S_{{\mathrm{1}}}  \ottsym{)}$})
      \end{array}\]
    By the IH, there exists some $S'''_{{\mathrm{2}}}$ such that
    $S_{{\mathrm{1}}}  \ottsym{(}  \ottnt{f_{{\mathrm{1}}}}  \ottsym{)} \,  \xmapsto{ \mathmakebox[0.4em]{} S'''_{{\mathrm{2}}} \mathmakebox[0.3em]{} }\hspace{-0.4em}{}^\ast \hspace{0.2em}  \, \ottnt{r}$ and
    $\textit{dom} \, \ottsym{(}  S'''_{{\mathrm{2}}}  \ottsym{)}  \subseteq  \textit{dom} \, \ottsym{(}  S''_{{\mathrm{2}}}  \ottsym{)}$ and,
    for any $\ottmv{X'} \, \in \, \textit{dom} \, \ottsym{(}  S_{{\mathrm{1}}}  \ottsym{)}$, $\textit{ftv} \, \ottsym{(}  S_{{\mathrm{1}}}  \ottsym{(}  \ottmv{X'}  \ottsym{)}  \ottsym{)}  \cap  \textit{dom} \, \ottsym{(}  S'''_{{\mathrm{2}}}  \ottsym{)}  \ottsym{=}   \emptyset $.

    Let $\ottmv{X'} \, \in \, \textit{dom} \, \ottsym{(}  S_{{\mathrm{1}}}  \ottsym{)}$.
    We have $\textit{ftv} \, \ottsym{(}  S_{{\mathrm{1}}}  \ottsym{(}  \ottmv{X'}  \ottsym{)}  \ottsym{)}  \cap  \textit{dom} \, \ottsym{(}  S'''_{{\mathrm{2}}}  \ottsym{)}  \ottsym{=}   \emptyset $.
    Thus, it suffices to show that $\ottmv{X} \, \not\in \, \textit{ftv} \, \ottsym{(}  S_{{\mathrm{1}}}  \ottsym{(}  \ottmv{X'}  \ottsym{)}  \ottsym{)}$.
    It is obvious because $S_{{\mathrm{1}}}$ is generated by 
    $S'_{{\mathrm{1}}}  \ottsym{(}  \ottnt{E}  [  \ottnt{f'_{{\mathrm{1}}}}  ]  \ottsym{)} \,  \xmapsto{ \mathmakebox[0.4em]{}  S_{{\mathrm{1}}}  \uplus  S''_{{\mathrm{2}}}  \mathmakebox[0.3em]{} }\hspace{-0.4em}{}^\ast \hspace{0.2em}  \, \ottnt{r}$ where $\ottmv{X} \, \not\in \, \textit{ftv} \, \ottsym{(}  S'_{{\mathrm{1}}}  \ottsym{(}  \ottnt{E}  [  \ottnt{f'_{{\mathrm{1}}}}  ]  \ottsym{)}  \ottsym{)}$.
    \end{caseanalysis}

    \case{\rnp{R\_InstArrow}}
    We are given
    \begin{itemize}
     \item $\ottnt{f'}  \ottsym{=}  w'  \ottsym{:}   \star  \!\rightarrow\!  \star \Rightarrow  \unskip ^ { \ell_{{\mathrm{1}}} }  \!  \star \Rightarrow  \unskip ^ { \ell_{{\mathrm{2}}} }  \! \ottmv{X}  $,
     \item $\ottnt{f'_{{\mathrm{1}}}}  \ottsym{=}  w'  \ottsym{:}   \star  \!\rightarrow\!  \star \Rightarrow  \unskip ^ { \ell_{{\mathrm{1}}} }  \!  \star \Rightarrow  \unskip ^ { \ell_{{\mathrm{2}}} }  \!  \star  \!\rightarrow\!  \star \Rightarrow  \unskip ^ { \ell_{{\mathrm{2}}} }  \! \ottmv{X_{{\mathrm{1}}}}  \!\rightarrow\!  \ottmv{X_{{\mathrm{2}}}}   $, and
     \item $S'_{{\mathrm{1}}}  \ottsym{=}  [  \ottmv{X}  :=  \ottmv{X_{{\mathrm{1}}}}  \!\rightarrow\!  \ottmv{X_{{\mathrm{2}}}}  ]$
    \end{itemize}
    for some $w'$, $\ottmv{X}$, $\ell_{{\mathrm{1}}}$, $\ell_{{\mathrm{2}}}$, and fresh $\ottmv{X_{{\mathrm{1}}}}$ and
    $\ottmv{X_{{\mathrm{2}}}}$.
    Since $ S_{{\mathrm{1}}}  \uplus  S_{{\mathrm{2}}}   \ottsym{=}   S'_{{\mathrm{2}}}  \circ  S'_{{\mathrm{1}}} $, $\ottmv{X} \, \in \, \textit{dom} \, \ottsym{(}  S_{{\mathrm{1}}}  \ottsym{)}$ or $\ottmv{X} \, \in \, \textit{dom} \, \ottsym{(}  S_{{\mathrm{2}}}  \ottsym{)}$.
    \begin{caseanalysis}
    \case{$\ottmv{X} \, \in \, \textit{dom} \, \ottsym{(}  S_{{\mathrm{1}}}  \ottsym{)}$}
    Since $ S_{{\mathrm{1}}}  \uplus  S_{{\mathrm{2}}}   \ottsym{=}   S'_{{\mathrm{2}}}  \circ  S'_{{\mathrm{1}}} $,
    we have $S_{{\mathrm{1}}}  \ottsym{(}  \ottmv{X}  \ottsym{)}  \ottsym{=}  S'_{{\mathrm{2}}}  \ottsym{(}  \ottmv{X_{{\mathrm{1}}}}  \ottsym{)}  \!\rightarrow\!  S'_{{\mathrm{2}}}  \ottsym{(}  \ottmv{X_{{\mathrm{2}}}}  \ottsym{)}$.
    Since $\ottmv{X_{{\mathrm{1}}}}$ and $\ottmv{X_{{\mathrm{2}}}}$ are fresh,
    we can suppose that $\ottmv{X_{{\mathrm{1}}}}, \ottmv{X_{{\mathrm{2}}}} \, \not\in \, \textit{ftv} \, \ottsym{(}  \ottnt{f}  \ottsym{)}$.
    Since $\textit{dom} \, \ottsym{(}  S_{{\mathrm{1}}}  \ottsym{)}  \subseteq  \textit{ftv} \, \ottsym{(}  \ottnt{f}  \ottsym{)}$, we have
    $\ottmv{X_{{\mathrm{1}}}}, \ottmv{X_{{\mathrm{2}}}} \, \not\in \, \textit{dom} \, \ottsym{(}  S_{{\mathrm{1}}}  \ottsym{)}$.
    Thus, $S'_{{\mathrm{2}}}  \ottsym{(}  \ottmv{X_{{\mathrm{1}}}}  \ottsym{)}  \ottsym{=}  S_{{\mathrm{2}}}  \ottsym{(}  \ottmv{X_{{\mathrm{1}}}}  \ottsym{)}$ and $S'_{{\mathrm{2}}}  \ottsym{(}  \ottmv{X_{{\mathrm{2}}}}  \ottsym{)}  \ottsym{=}  S_{{\mathrm{2}}}  \ottsym{(}  \ottmv{X_{{\mathrm{2}}}}  \ottsym{)}$
    since $ S_{{\mathrm{1}}}  \uplus  S_{{\mathrm{2}}}   \ottsym{=}   S'_{{\mathrm{2}}}  \circ  S'_{{\mathrm{1}}} $.
    Let $S''_{{\mathrm{2}}}$ be a type substitution such that
    $\textit{dom} \, \ottsym{(}  S''_{{\mathrm{2}}}  \ottsym{)}  \ottsym{=}  \textit{dom} \, \ottsym{(}  S_{{\mathrm{2}}}  \ottsym{)}  \setminus  \ottsym{(}  \textit{dom} \, \ottsym{(}  S_{{\mathrm{2}}}  \ottsym{)}  \setminus   \{  \ottmv{X_{{\mathrm{1}}}} ,  \ottmv{X_{{\mathrm{2}}}}  \}   \ottsym{)}$ and
    $S''_{{\mathrm{2}}}  \ottsym{(}  \ottmv{Y}  \ottsym{)}  \ottsym{=}  S_{{\mathrm{2}}}  \ottsym{(}  \ottmv{Y}  \ottsym{)}$ for any $\ottmv{Y} \, \in \, \textit{dom} \, \ottsym{(}  S''_{{\mathrm{2}}}  \ottsym{)}$.
    Note that $S''_{{\mathrm{2}}}  \ottsym{(}  \ottmv{X_{{\mathrm{1}}}}  \ottsym{)}  \ottsym{=}  S'_{{\mathrm{2}}}  \ottsym{(}  \ottmv{X_{{\mathrm{1}}}}  \ottsym{)}$ and $S''_{{\mathrm{2}}}  \ottsym{(}  \ottmv{X_{{\mathrm{2}}}}  \ottsym{)}  \ottsym{=}  S'_{{\mathrm{2}}}  \ottsym{(}  \ottmv{X_{{\mathrm{2}}}}  \ottsym{)}$.
    Since $S'_{{\mathrm{1}}}  \ottsym{(}  \ottnt{E}  [  \ottnt{f'_{{\mathrm{1}}}}  ]  \ottsym{)} \,  \xmapsto{ \mathmakebox[0.4em]{} S'_{{\mathrm{2}}} \mathmakebox[0.3em]{} }\hspace{-0.4em}{}^\ast \hspace{0.2em}  \, \ottnt{r}$,
    we can suppose that $\ottmv{X} \, \not\in \, \textit{dom} \, \ottsym{(}  S'_{{\mathrm{2}}}  \ottsym{)}$.
    Thus, $ S'_{{\mathrm{2}}}  \circ  S'_{{\mathrm{1}}}  =  [  \ottmv{X}  :=  S'_{{\mathrm{2}}}  \ottsym{(}  \ottmv{X_{{\mathrm{1}}}}  \ottsym{)}  \!\rightarrow\!  S'_{{\mathrm{2}}}  \ottsym{(}  \ottmv{X_{{\mathrm{2}}}}  \ottsym{)}  ]  \uplus  S'_{{\mathrm{2}}}  =  [  \ottmv{X}  :=  S''_{{\mathrm{2}}}  \ottsym{(}  \ottmv{X_{{\mathrm{1}}}}  \ottsym{)}  \!\rightarrow\!  S''_{{\mathrm{2}}}  \ottsym{(}  \ottmv{X_{{\mathrm{2}}}}  \ottsym{)}  ]  \uplus  S'_{{\mathrm{2}}} $.
    Since $ S_{{\mathrm{1}}}  \uplus  S_{{\mathrm{2}}}   \ottsym{=}   [  \ottmv{X}  :=  S''_{{\mathrm{2}}}  \ottsym{(}  \ottmv{X_{{\mathrm{1}}}}  \ottsym{)}  \!\rightarrow\!  S''_{{\mathrm{2}}}  \ottsym{(}  \ottmv{X_{{\mathrm{2}}}}  \ottsym{)}  ]  \uplus  S'_{{\mathrm{2}}} $,
    there exist some $S'''_{{\mathrm{1}}}$ and $S'''_{{\mathrm{2}}}$ such that
    \begin{itemize}
     \item $S'_{{\mathrm{2}}}  \ottsym{=}    S'''_{{\mathrm{1}}}  \uplus  S''_{{\mathrm{2}}}   \uplus  S'''_{{\mathrm{2}}} $,
     \item $S_{{\mathrm{1}}}  \ottsym{=}   [  \ottmv{X}  :=  S''_{{\mathrm{2}}}  \ottsym{(}  \ottmv{X_{{\mathrm{1}}}}  \ottsym{)}  \!\rightarrow\!  S''_{{\mathrm{2}}}  \ottsym{(}  \ottmv{X_{{\mathrm{2}}}}  \ottsym{)}  ]  \uplus  S'''_{{\mathrm{1}}} $, and
     \item $S_{{\mathrm{2}}}  \ottsym{=}   S''_{{\mathrm{2}}}  \uplus  S'''_{{\mathrm{2}}} $.
    \end{itemize}
    Here,
     \[\begin{array}{lll}
      S_{{\mathrm{1}}}  \ottsym{(}  \ottnt{f}  \ottsym{)} & = & S_{{\mathrm{1}}}  \ottsym{(}  \ottnt{E}  \ottsym{)}  [  S_{{\mathrm{1}}}  \ottsym{(}  w'  \ottsym{:}   \star  \!\rightarrow\!  \star \Rightarrow  \unskip ^ { \ell_{{\mathrm{1}}} }  \!  \star \Rightarrow  \unskip ^ { \ell_{{\mathrm{2}}} }  \! \ottmv{X}    \ottsym{)}  ] \\
                & = & S_{{\mathrm{1}}}  \ottsym{(}  \ottnt{E}  \ottsym{)}  [  S_{{\mathrm{1}}}  \ottsym{(}  w'  \ottsym{)}  \ottsym{:}   \star  \!\rightarrow\!  \star \Rightarrow  \unskip ^ { \ell }  \!  \star \Rightarrow  \unskip ^ { \ell_{{\mathrm{2}}} }  \! S'_{{\mathrm{2}}}  \ottsym{(}  \ottmv{X_{{\mathrm{1}}}}  \ottsym{)}  \!\rightarrow\!  S'_{{\mathrm{2}}}  \ottsym{(}  \ottmv{X_{{\mathrm{2}}}}  \ottsym{)}    ] \\
                &  \xmapsto{ \mathmakebox[0.4em]{} [  ] \mathmakebox[0.3em]{} }  &
                 S_{{\mathrm{1}}}  \ottsym{(}  \ottnt{E}  \ottsym{)}  [  S_{{\mathrm{1}}}  \ottsym{(}  w'  \ottsym{)}  \ottsym{:}   \star  \!\rightarrow\!  \star \Rightarrow  \unskip ^ { \ell }  \!  \star \Rightarrow  \unskip ^ { \ell_{{\mathrm{2}}} }  \!  \star  \!\rightarrow\!  \star \Rightarrow  \unskip ^ { \ell_{{\mathrm{2}}} }  \! S'_{{\mathrm{2}}}  \ottsym{(}  \ottmv{X_{{\mathrm{1}}}}  \ottsym{)}  \!\rightarrow\!  S'_{{\mathrm{2}}}  \ottsym{(}  \ottmv{X_{{\mathrm{2}}}}  \ottsym{)}     ]
                 \\ && \quad \text{($S_{{\mathrm{1}}}  \ottsym{(}  w'  \ottsym{)}$ is a value by Lemma~\ref{lem:value_any_subst})} \\
                & = &
                 S_{{\mathrm{1}}}  \ottsym{(}  \ottnt{E}  \ottsym{)}  [  S_{{\mathrm{1}}}  \ottsym{(}  w'  \ottsym{)}  \ottsym{:}   \star  \!\rightarrow\!  \star \Rightarrow  \unskip ^ { \ell }  \!  \star \Rightarrow  \unskip ^ { \ell_{{\mathrm{2}}} }  \!  \star  \!\rightarrow\!  \star \Rightarrow  \unskip ^ { \ell_{{\mathrm{2}}} }  \! S''_{{\mathrm{2}}}  \ottsym{(}  \ottmv{X_{{\mathrm{1}}}}  \ottsym{)}  \!\rightarrow\!  S''_{{\mathrm{2}}}  \ottsym{(}  \ottmv{X_{{\mathrm{2}}}}  \ottsym{)}     ] \\
                & = &
                  S_{{\mathrm{1}}}  \uplus  S''_{{\mathrm{2}}}   \ottsym{(}  \ottnt{E}  [  w'  \ottsym{:}   \star  \!\rightarrow\!  \star \Rightarrow  \unskip ^ { \ell }  \!  \star \Rightarrow  \unskip ^ { \ell_{{\mathrm{2}}} }  \!  \star  \!\rightarrow\!  \star \Rightarrow  \unskip ^ { \ell_{{\mathrm{2}}} }  \! \ottmv{X_{{\mathrm{1}}}}  \!\rightarrow\!  \ottmv{X_{{\mathrm{2}}}}     ]  \ottsym{)}
                 \\ && \quad \text{(since $\textit{dom} \, \ottsym{(}  S''_{{\mathrm{2}}}  \ottsym{)}$ are fresh type variables if any)}\\
                & = &
                   [  \ottmv{X}  :=  S''_{{\mathrm{2}}}  \ottsym{(}  \ottmv{X_{{\mathrm{1}}}}  \ottsym{)}  \!\rightarrow\!  S''_{{\mathrm{2}}}  \ottsym{(}  \ottmv{X_{{\mathrm{2}}}}  \ottsym{)}  ]  \uplus  S'''_{{\mathrm{1}}}   \uplus  S''_{{\mathrm{2}}}   \ottsym{(}  \ottnt{E}  [  w'  \ottsym{:}   \star  \!\rightarrow\!  \star \Rightarrow  \unskip ^ { \ell }  \!  \star \Rightarrow  \unskip ^ { \ell_{{\mathrm{2}}} }  \!  \star  \!\rightarrow\!  \star \Rightarrow  \unskip ^ { \ell_{{\mathrm{2}}} }  \! \ottmv{X_{{\mathrm{1}}}}  \!\rightarrow\!  \ottmv{X_{{\mathrm{2}}}}     ]  \ottsym{)} \\
                & = &
                  S'''_{{\mathrm{1}}}  \uplus  S''_{{\mathrm{2}}}   \ottsym{(}  \ottnt{E}  [  w'  \ottsym{:}   \star  \!\rightarrow\!  \star \Rightarrow  \unskip ^ { \ell }  \!  \star \Rightarrow  \unskip ^ { \ell_{{\mathrm{2}}} }  \!  \star  \!\rightarrow\!  \star \Rightarrow  \unskip ^ { \ell_{{\mathrm{2}}} }  \! \ottmv{X_{{\mathrm{1}}}}  \!\rightarrow\!  \ottmv{X_{{\mathrm{2}}}}     ]  [  \ottmv{X}  \ottsym{:=}  \ottmv{X_{{\mathrm{1}}}}  \!\rightarrow\!  \ottmv{X_{{\mathrm{2}}}}  ]  \ottsym{)} \\
                & = &
                  S'''_{{\mathrm{1}}}  \uplus  S''_{{\mathrm{2}}}   \ottsym{(}  S'_{{\mathrm{1}}}  \ottsym{(}  \ottnt{E}  [  \ottnt{f'_{{\mathrm{1}}}}  ]  \ottsym{)}  \ottsym{)}
       \end{array}\]
    We have $S'_{{\mathrm{1}}}  \ottsym{(}  \ottnt{E}  [  \ottnt{f'_{{\mathrm{1}}}}  ]  \ottsym{)} = \ottnt{f_{{\mathrm{1}}}} \,  \xmapsto{ \mathmakebox[0.4em]{}   S'''_{{\mathrm{1}}}  \uplus  S''_{{\mathrm{2}}}   \uplus  S'''_{{\mathrm{2}}}  \mathmakebox[0.3em]{} }\hspace{-0.4em}{}^\ast \hspace{0.2em}  \, \ottnt{r}$.
    Here,
    \[\begin{array}{lll}
     \textit{ftv} \, \ottsym{(}  \ottnt{f_{{\mathrm{1}}}}  \ottsym{)} & = & \textit{ftv} \, \ottsym{(}  S'_{{\mathrm{1}}}  \ottsym{(}  \ottnt{E}  [  \ottnt{f'_{{\mathrm{1}}}}  ]  \ottsym{)}  \ottsym{)} \\
                & = & \ottsym{(}  \textit{ftv} \, \ottsym{(}  \ottnt{E}  [  w'  ]  \ottsym{)}  \setminus   \{  \ottmv{X}  \}   \ottsym{)}  \cup   \{  \ottmv{X_{{\mathrm{1}}}} ,  \ottmv{X_{{\mathrm{2}}}}  \}  \\
                & \supseteq & \ottsym{(}  \textit{ftv} \, \ottsym{(}  \ottnt{E}  [  w'  ]  \ottsym{)}  \setminus   \{  \ottmv{X}  \}   \ottsym{)}  \cup  \textit{dom} \, \ottsym{(}  S''_{{\mathrm{2}}}  \ottsym{)} \\
                & = & \ottsym{(}  \ottsym{(}  \textit{ftv} \, \ottsym{(}  \ottnt{E}  [  w'  ]  \ottsym{)}  \cup   \{  \ottmv{X}  \}   \ottsym{)}  \setminus   \{  \ottmv{X}  \}   \ottsym{)}  \cup  \textit{dom} \, \ottsym{(}  S''_{{\mathrm{2}}}  \ottsym{)} \\
                & = & \ottsym{(}  \textit{ftv} \, \ottsym{(}  \ottnt{E}  [  \ottnt{f'}  ]  \ottsym{)}  \setminus   \{  \ottmv{X}  \}   \ottsym{)}  \cup  \textit{dom} \, \ottsym{(}  S''_{{\mathrm{2}}}  \ottsym{)} \\
                & \supseteq & \ottsym{(}  \textit{dom} \, \ottsym{(}  S_{{\mathrm{1}}}  \ottsym{)}  \setminus   \{  \ottmv{X}  \}   \ottsym{)}  \cup  \textit{dom} \, \ottsym{(}  S''_{{\mathrm{2}}}  \ottsym{)} \\
                & = & \textit{dom} \, \ottsym{(}  S'''_{{\mathrm{1}}}  \ottsym{)}  \cup  \textit{dom} \, \ottsym{(}  S''_{{\mathrm{2}}}  \ottsym{)}
      \end{array}\]
    Thus, by the IH,
    $ S'''_{{\mathrm{1}}}  \uplus  S''_{{\mathrm{2}}}   \ottsym{(}  \ottnt{f_{{\mathrm{1}}}}  \ottsym{)} \,  \xmapsto{ \mathmakebox[0.4em]{} S''''_{{\mathrm{2}}} \mathmakebox[0.3em]{} }\hspace{-0.4em}{}^\ast \hspace{0.2em}  \, \ottnt{r}$ for some $S''''_{{\mathrm{2}}}$ such that
    $\textit{dom} \, \ottsym{(}  S''''_{{\mathrm{2}}}  \ottsym{)}  \subseteq  \textit{dom} \, \ottsym{(}  S'''_{{\mathrm{2}}}  \ottsym{)}$ and,
    for any $\ottmv{X'} \, \in \, \textit{dom} \, \ottsym{(}   S'''_{{\mathrm{1}}}  \uplus  S''_{{\mathrm{2}}}   \ottsym{)}$,
    $\textit{ftv} \, \ottsym{(}   S'''_{{\mathrm{1}}}  \uplus  S''_{{\mathrm{2}}}   \ottsym{(}  \ottmv{X'}  \ottsym{)}  \ottsym{)}  \cap  \textit{dom} \, \ottsym{(}  S''''_{{\mathrm{2}}}  \ottsym{)}  \ottsym{=}   \emptyset $.
    We have $S_{{\mathrm{1}}}  \ottsym{(}  \ottnt{f}  \ottsym{)} \,  \xmapsto{ \mathmakebox[0.4em]{} S''''_{{\mathrm{2}}} \mathmakebox[0.3em]{} }\hspace{-0.4em}{}^\ast \hspace{0.2em}  \, \ottnt{r}$.

    Let $\ottmv{X'} \, \in \, \textit{dom} \, \ottsym{(}  S_{{\mathrm{1}}}  \ottsym{)}$.
    If $\ottmv{X'}  \ottsym{=}  \ottmv{X}$, then $\ottmv{X'} \, \not\in \, \textit{dom} \, \ottsym{(}  S'''_{{\mathrm{2}}}  \ottsym{)}$ and so
    $\ottmv{X'} \, \not\in \, \textit{dom} \, \ottsym{(}  S''''_{{\mathrm{2}}}  \ottsym{)}$.
    Otherwise, if $\ottmv{X'}  \neq  \ottmv{X}$, then $\ottmv{X'} \, \in \, \textit{dom} \, \ottsym{(}  S'''_{{\mathrm{1}}}  \ottsym{)}$.
    Since $S'''_{{\mathrm{1}}}  \ottsym{(}  \ottmv{X'}  \ottsym{)}  \ottsym{=}  S_{{\mathrm{1}}}  \ottsym{(}  \ottmv{X'}  \ottsym{)}$, we have
    $\textit{ftv} \, \ottsym{(}  S_{{\mathrm{1}}}  \ottsym{(}  \ottmv{X'}  \ottsym{)}  \ottsym{)}  \cap  \textit{dom} \, \ottsym{(}  S''''_{{\mathrm{2}}}  \ottsym{)}  \ottsym{=}   \emptyset $ by the IH.

    \case{$\ottmv{X} \, \in \, \textit{dom} \, \ottsym{(}  S_{{\mathrm{2}}}  \ottsym{)}$}
     Since $\textit{dom} \, \ottsym{(}  S_{{\mathrm{1}}}  \ottsym{)}$ and $\textit{dom} \, \ottsym{(}  S_{{\mathrm{2}}}  \ottsym{)}$ are disjoint,
     $\ottmv{X} \, \not\in \, \textit{dom} \, \ottsym{(}  S_{{\mathrm{1}}}  \ottsym{)}$.
     Furthermore,
     we can suppose that $\ottmv{X} \, \not\in \, \textit{dom} \, \ottsym{(}  S'_{{\mathrm{2}}}  \ottsym{)}$
     since $S'_{{\mathrm{1}}}  \ottsym{(}  \ottnt{E}  [  \ottnt{f'_{{\mathrm{1}}}}  ]  \ottsym{)} \,  \xmapsto{ \mathmakebox[0.4em]{} S'_{{\mathrm{2}}} \mathmakebox[0.3em]{} }\hspace{-0.4em}{}^\ast \hspace{0.2em}  \, \ottnt{r}$,
     Thus, $ S'_{{\mathrm{2}}}  \circ  S'_{{\mathrm{1}}}   \ottsym{=}   [  \ottmv{X}  :=  S'_{{\mathrm{2}}}  \ottsym{(}  \ottmv{X_{{\mathrm{1}}}}  \ottsym{)}  \!\rightarrow\!  S'_{{\mathrm{2}}}  \ottsym{(}  \ottmv{X_{{\mathrm{2}}}}  \ottsym{)}  ]  \uplus  S'_{{\mathrm{2}}} $.
     Since $ S_{{\mathrm{1}}}  \uplus  S_{{\mathrm{2}}}   \ottsym{=}   S'_{{\mathrm{2}}}  \circ  S'_{{\mathrm{1}}} $, we have
     $S_{{\mathrm{2}}}  \ottsym{(}  \ottmv{X}  \ottsym{)}  \ottsym{=}  S'_{{\mathrm{2}}}  \ottsym{(}  \ottmv{X_{{\mathrm{1}}}}  \ottsym{)}  \!\rightarrow\!  S'_{{\mathrm{2}}}  \ottsym{(}  \ottmv{X_{{\mathrm{2}}}}  \ottsym{)}$.
     Since $\ottmv{X_{{\mathrm{1}}}}$ and $\ottmv{X_{{\mathrm{2}}}}$ are fresh,
     we can suppose that $\ottmv{X_{{\mathrm{1}}}}, \ottmv{X_{{\mathrm{2}}}} \, \not\in \, \textit{ftv} \, \ottsym{(}  \ottnt{f}  \ottsym{)}$.
     Since $\textit{dom} \, \ottsym{(}  S_{{\mathrm{1}}}  \ottsym{)}  \subseteq  \textit{ftv} \, \ottsym{(}  \ottnt{f}  \ottsym{)}$, we have
     $\ottmv{X_{{\mathrm{1}}}}, \ottmv{X_{{\mathrm{2}}}} \, \not\in \, \textit{dom} \, \ottsym{(}  S_{{\mathrm{1}}}  \ottsym{)}$.
     Thus, $S'_{{\mathrm{2}}}  \ottsym{(}  \ottmv{X_{{\mathrm{1}}}}  \ottsym{)}  \ottsym{=}  S_{{\mathrm{2}}}  \ottsym{(}  \ottmv{X_{{\mathrm{1}}}}  \ottsym{)}$ and $S'_{{\mathrm{2}}}  \ottsym{(}  \ottmv{X_{{\mathrm{2}}}}  \ottsym{)}  \ottsym{=}  S_{{\mathrm{2}}}  \ottsym{(}  \ottmv{X_{{\mathrm{2}}}}  \ottsym{)}$
     since $ S_{{\mathrm{1}}}  \uplus  S_{{\mathrm{2}}}   \ottsym{=}   S'_{{\mathrm{2}}}  \circ  S'_{{\mathrm{1}}} $.
     Let $S''_{{\mathrm{2}}}$ be a type substitution such that
     $\textit{dom} \, \ottsym{(}  S''_{{\mathrm{2}}}  \ottsym{)}  \ottsym{=}  \textit{dom} \, \ottsym{(}  S_{{\mathrm{2}}}  \ottsym{)}  \setminus  \ottsym{(}  \textit{dom} \, \ottsym{(}  S_{{\mathrm{2}}}  \ottsym{)}  \setminus   \{  \ottmv{X_{{\mathrm{1}}}} ,  \ottmv{X_{{\mathrm{2}}}}  \}   \ottsym{)}$ and
     $S''_{{\mathrm{2}}}  \ottsym{(}  \ottmv{Y}  \ottsym{)}  \ottsym{=}  S_{{\mathrm{2}}}  \ottsym{(}  \ottmv{Y}  \ottsym{)}$ for any $\ottmv{Y} \, \in \, \textit{dom} \, \ottsym{(}  S''_{{\mathrm{2}}}  \ottsym{)}$.
     Note that $S''_{{\mathrm{2}}}  \ottsym{(}  \ottmv{X_{{\mathrm{1}}}}  \ottsym{)}  \ottsym{=}  S'_{{\mathrm{2}}}  \ottsym{(}  \ottmv{X_{{\mathrm{1}}}}  \ottsym{)}$ and $S''_{{\mathrm{2}}}  \ottsym{(}  \ottmv{X_{{\mathrm{2}}}}  \ottsym{)}  \ottsym{=}  S'_{{\mathrm{2}}}  \ottsym{(}  \ottmv{X_{{\mathrm{2}}}}  \ottsym{)}$.
     Thus $ S'_{{\mathrm{2}}}  \circ  S'_{{\mathrm{1}}}   \ottsym{=}   [  \ottmv{X}  :=  S''_{{\mathrm{2}}}  \ottsym{(}  \ottmv{X_{{\mathrm{1}}}}  \ottsym{)}  \!\rightarrow\!  S''_{{\mathrm{2}}}  \ottsym{(}  \ottmv{X_{{\mathrm{2}}}}  \ottsym{)}  ]  \uplus  S'_{{\mathrm{2}}} $.
     Since $ S_{{\mathrm{1}}}  \uplus  S_{{\mathrm{2}}}  =  S'_{{\mathrm{2}}}  \circ  S'_{{\mathrm{1}}}  =  [  \ottmv{X}  :=  S''_{{\mathrm{2}}}  \ottsym{(}  \ottmv{X_{{\mathrm{1}}}}  \ottsym{)}  \!\rightarrow\!  S''_{{\mathrm{2}}}  \ottsym{(}  \ottmv{X_{{\mathrm{2}}}}  \ottsym{)}  ]  \uplus  S'_{{\mathrm{2}}} $,
     there exist some $S'''_{{\mathrm{2}}}$ such that
     \begin{itemize}
      \item $S'_{{\mathrm{2}}}  \ottsym{=}    S_{{\mathrm{1}}}  \uplus  S''_{{\mathrm{2}}}   \uplus  S'''_{{\mathrm{2}}} $ and
      \item $S_{{\mathrm{2}}}  \ottsym{=}    [  \ottmv{X}  :=  S''_{{\mathrm{2}}}  \ottsym{(}  \ottmv{X_{{\mathrm{1}}}}  \ottsym{)}  \!\rightarrow\!  S''_{{\mathrm{2}}}  \ottsym{(}  \ottmv{X_{{\mathrm{2}}}}  \ottsym{)}  ]  \uplus  S''_{{\mathrm{2}}}   \uplus  S'''_{{\mathrm{2}}} $.
     \end{itemize}
     Here,
     \[\begin{array}{lll}
      S_{{\mathrm{1}}}  \ottsym{(}  \ottnt{f}  \ottsym{)} & = & S_{{\mathrm{1}}}  \ottsym{(}  \ottnt{E}  \ottsym{)}  [  S_{{\mathrm{1}}}  \ottsym{(}  w'  \ottsym{)}  \ottsym{:}   \star  \!\rightarrow\!  \star \Rightarrow  \unskip ^ { \ell_{{\mathrm{1}}} }  \!  \star \Rightarrow  \unskip ^ { \ell_{{\mathrm{2}}} }  \! \ottmv{X}    ] \\
                &  \xmapsto{ \mathmakebox[0.4em]{} S'_{{\mathrm{1}}} \mathmakebox[0.3em]{} }  & S'_{{\mathrm{1}}}  \ottsym{(}  S_{{\mathrm{1}}}  \ottsym{(}  \ottnt{E}  [  S_{{\mathrm{1}}}  \ottsym{(}  w'  \ottsym{)}  \ottsym{:}   \star  \!\rightarrow\!  \star \Rightarrow  \unskip ^ { \ell_{{\mathrm{1}}} }  \!  \star \Rightarrow  \unskip ^ { \ell_{{\mathrm{2}}} }  \!  \star  \!\rightarrow\!  \star \Rightarrow  \unskip ^ { \ell_{{\mathrm{2}}} }  \! \ottmv{X_{{\mathrm{1}}}}  \!\rightarrow\!  \ottmv{X_{{\mathrm{2}}}}     ]  \ottsym{)}  \ottsym{)} \\
                & = & S'_{{\mathrm{1}}}  \ottsym{(}  S_{{\mathrm{1}}}  \ottsym{(}  \ottnt{E}  [  \ottnt{f'_{{\mathrm{1}}}}  ]  \ottsym{)}  \ottsym{)} \\
                & = & S_{{\mathrm{1}}}  \ottsym{(}  S'_{{\mathrm{1}}}  \ottsym{(}  \ottnt{E}  [  \ottnt{f'_{{\mathrm{1}}}}  ]  \ottsym{)}  \ottsym{)} \\
                   &&
                 \text{(since $S_{{\mathrm{1}}}$ is generated by DTI and $\textit{dom} \, \ottsym{(}  S_{{\mathrm{1}}}  \ottsym{)}  \cap   \{  \ottmv{X_{{\mathrm{1}}}} ,  \ottmv{X_{{\mathrm{2}}}}  \}   \ottsym{=}   \emptyset $)} \\
                & = & S_{{\mathrm{1}}}  \ottsym{(}  \ottnt{f_{{\mathrm{1}}}}  \ottsym{)}
       \end{array}\]
     We have $S'_{{\mathrm{1}}}  \ottsym{(}  \ottnt{E}  [  \ottnt{f'_{{\mathrm{1}}}}  ]  \ottsym{)} = \ottnt{f_{{\mathrm{1}}}} \,  \xmapsto{ \mathmakebox[0.4em]{}   S_{{\mathrm{1}}}  \uplus  S''_{{\mathrm{2}}}   \uplus  S'''_{{\mathrm{2}}}  \mathmakebox[0.3em]{} }\hspace{-0.4em}{}^\ast \hspace{0.2em}  \, \ottnt{r}$.
     Here,
     \[\begin{array}{lll}
      \textit{ftv} \, \ottsym{(}  \ottnt{f_{{\mathrm{1}}}}  \ottsym{)} & = & \textit{ftv} \, \ottsym{(}  S'_{{\mathrm{1}}}  \ottsym{(}  \ottnt{E}  [  \ottnt{f'_{{\mathrm{1}}}}  ]  \ottsym{)}  \ottsym{)} \\
                  & = & \textit{ftv} \, \ottsym{(}  S'_{{\mathrm{1}}}  \ottsym{(}  \ottnt{E}  [  w'  ]  \ottsym{)}  \ottsym{)}  \cup   \{  \ottmv{X_{{\mathrm{1}}}} ,  \ottmv{X_{{\mathrm{2}}}}  \}  \\
                  & = & \ottsym{(}  \textit{ftv} \, \ottsym{(}  \ottnt{E}  [  w'  ]  \ottsym{)}  \setminus   \{  \ottmv{X}  \}   \ottsym{)}  \cup   \{  \ottmv{X_{{\mathrm{1}}}} ,  \ottmv{X_{{\mathrm{2}}}}  \}  \\
                  & = & \ottsym{(}  \textit{ftv} \, \ottsym{(}  \ottnt{E}  [  \ottnt{f'}  ]  \ottsym{)}  \setminus   \{  \ottmv{X}  \}   \ottsym{)}  \cup   \{  \ottmv{X_{{\mathrm{1}}}} ,  \ottmv{X_{{\mathrm{2}}}}  \}  \\
                  & \supseteq & \ottsym{(}  \textit{dom} \, \ottsym{(}  S_{{\mathrm{1}}}  \ottsym{)}  \setminus   \{  \ottmv{X}  \}   \ottsym{)}  \cup   \{  \ottmv{X_{{\mathrm{1}}}} ,  \ottmv{X_{{\mathrm{2}}}}  \}  \\
                  & \supseteq & \textit{dom} \, \ottsym{(}  S_{{\mathrm{1}}}  \ottsym{)}
                  \qquad (\text{since $\ottmv{X} \, \not\in \, \textit{dom} \, \ottsym{(}  S_{{\mathrm{1}}}  \ottsym{)}$})
      \end{array}\]
    By the IH, there exist some $S''''_{{\mathrm{2}}}$ such that
    $S_{{\mathrm{1}}}  \ottsym{(}  \ottnt{f_{{\mathrm{1}}}}  \ottsym{)} \,  \xmapsto{ \mathmakebox[0.4em]{} S''''_{{\mathrm{2}}} \mathmakebox[0.3em]{} }\hspace{-0.4em}{}^\ast \hspace{0.2em}  \, \ottnt{r}$ and
    $\textit{dom} \, \ottsym{(}  S''''_{{\mathrm{2}}}  \ottsym{)}  \subseteq  \textit{dom} \, \ottsym{(}   S''_{{\mathrm{2}}}  \uplus  S'''_{{\mathrm{2}}}   \ottsym{)}$ and,
    for any $\ottmv{X'} \, \in \, \textit{dom} \, \ottsym{(}  S_{{\mathrm{1}}}  \ottsym{)}$, $\textit{ftv} \, \ottsym{(}  S_{{\mathrm{1}}}  \ottsym{(}  \ottmv{X'}  \ottsym{)}  \ottsym{)}  \cap  \textit{dom} \, \ottsym{(}  S''''_{{\mathrm{2}}}  \ottsym{)}  \ottsym{=}   \emptyset $.

    Let $\ottmv{X'} \, \in \, \textit{dom} \, \ottsym{(}  S_{{\mathrm{1}}}  \ottsym{)}$.
    Since $\ottmv{X'}  \neq  \ottmv{X}$ and $\textit{ftv} \, \ottsym{(}  S_{{\mathrm{1}}}  \ottsym{(}  \ottmv{X'}  \ottsym{)}  \ottsym{)}  \cap  \textit{dom} \, \ottsym{(}  S''''_{{\mathrm{2}}}  \ottsym{)}  \ottsym{=}   \emptyset $ by the IH,
    we finish.
    \end{caseanalysis}

    \otherwise We are given $S'_{{\mathrm{1}}}  \ottsym{=}  [  ]$ and $S'_{{\mathrm{2}}}  \ottsym{=}   S_{{\mathrm{1}}}  \uplus  S_{{\mathrm{2}}} $.  We have
     $S_{{\mathrm{1}}}  \ottsym{(}  \ottnt{f'}  \ottsym{)} \,  \xrightarrow{ \mathmakebox[0.4em]{} [  ] \mathmakebox[0.3em]{} }  \, S_{{\mathrm{1}}}  \ottsym{(}  \ottnt{f'_{{\mathrm{1}}}}  \ottsym{)}$; note that type variables generated by term
     substitution is not captured by $S_{{\mathrm{1}}}$ since $\textit{dom} \, \ottsym{(}  S_{{\mathrm{1}}}  \ottsym{)}  \subseteq  \textit{ftv} \, \ottsym{(}  \ottnt{f}  \ottsym{)}$.
     We finish by the IH.
     \qedhere
   \end{caseanalysis}
  \end{caseanalysis}
 \end{caseanalysis}
\end{proof}

\begin{lemmaA} \label{lem:soundness_blame_2}
 If $\ottnt{f} \,  \xmapsto{ \mathmakebox[0.4em]{} S \mathmakebox[0.3em]{} }\hspace{-0.4em}{}^\ast \hspace{0.2em}  \, \textsf{\textup{blame}\relax} \, \ell$, then, for any $S'$ such that type variables in
 $\textit{dom} \, \ottsym{(}  S'  \ottsym{)}$ are disjoint from ones generated during the evaluation, there
 exist $S''$ and $\ell'$ such that $S'  \ottsym{(}  \ottnt{f}  \ottsym{)} \,  \xmapsto{ \mathmakebox[0.4em]{} S'' \mathmakebox[0.3em]{} }\hspace{-0.4em}{}^\ast \hspace{0.2em}  \, \textsf{\textup{blame}\relax} \, \ell'$.
\end{lemmaA}
\begin{proof}
 By induction on the length of the evaluation sequence.
 \begin{caseanalysis}
  \case{the length is 0} Obvious.
  \case{the length is more than 0}
  We are given $\ottnt{f} \,  \xmapsto{ \mathmakebox[0.4em]{} S_{{\mathrm{1}}} \mathmakebox[0.3em]{} }  \, \ottnt{f'}$ and
  $\ottnt{f'} \,  \xmapsto{ \mathmakebox[0.4em]{} S_{{\mathrm{2}}} \mathmakebox[0.3em]{} }\hspace{-0.4em}{}^\ast \hspace{0.2em}  \, \textsf{\textup{blame}\relax} \, \ell$ for some $S_{{\mathrm{1}}}$, $S_{{\mathrm{2}}}$, and $\ottnt{f'}$
  such that $S  \ottsym{=}   S_{{\mathrm{2}}}  \circ  S_{{\mathrm{1}}} $.
  By case analysis on the evaluation rule applied to $\ottnt{f}$.
  \begin{caseanalysis}
   \case{\rnp{E\_Step}}
   We are given $\ottnt{f}  \ottsym{=}  \ottnt{E}  [  \ottnt{f_{{\mathrm{1}}}}  ]$ and
   $\ottnt{f'}  \ottsym{=}  S_{{\mathrm{1}}}  \ottsym{(}  \ottnt{E}  [  \ottnt{f'_{{\mathrm{1}}}}  ]  \ottsym{)}$ and
   $\ottnt{f_{{\mathrm{1}}}} \,  \xrightarrow{ \mathmakebox[0.4em]{} S_{{\mathrm{1}}} \mathmakebox[0.3em]{} }  \, \ottnt{f'_{{\mathrm{1}}}}$
   for some $\ottnt{E}$, $\ottnt{f_{{\mathrm{1}}}}$, and $\ottnt{f'_{{\mathrm{1}}}}$.
   By case analysis on the reduction rule applied to $\ottnt{f_{{\mathrm{1}}}}$.
   \begin{caseanalysis}
    \case{\rnp{R\_InstBase}}
    We are given $w  \ottsym{:}   \iota \Rightarrow  \unskip ^ { \ell_{{\mathrm{1}}} }  \!  \star \Rightarrow  \unskip ^ { \ell_{{\mathrm{2}}} }  \! \ottmv{X}   \,  \xrightarrow{ \mathmakebox[0.4em]{} [  \ottmv{X}  :=  \iota  ] \mathmakebox[0.3em]{} }  \, w$
    for some $w$, $\iota$, $\ottmv{X}$, $\ell_{{\mathrm{1}}}$, and $\ell_{{\mathrm{2}}}$.
    By case analysis on $S'  \ottsym{(}  \ottmv{X}  \ottsym{)}$.
    \begin{caseanalysis}
     \case{$S'  \ottsym{(}  \ottmv{X}  \ottsym{)}  \ottsym{=}  \iota$}
     $S'  \ottsym{(}  w  \ottsym{:}   \iota \Rightarrow  \unskip ^ { \ell_{{\mathrm{1}}} }  \!  \star \Rightarrow  \unskip ^ { \ell_{{\mathrm{2}}} }  \! \ottmv{X}    \ottsym{)} \,  \xrightarrow{ \mathmakebox[0.4em]{} [  ] \mathmakebox[0.3em]{} }  \, S'  \ottsym{(}  w  \ottsym{)}$
     by \rnp{R\_IdBase}.
     Thus, $S'  \ottsym{(}  \ottnt{E}  [  \ottnt{f_{{\mathrm{1}}}}  ]  \ottsym{)} \,  \xmapsto{ \mathmakebox[0.4em]{} [  ] \mathmakebox[0.3em]{} }  \, S'  \ottsym{(}  \ottnt{E}  [  \ottnt{f'_{{\mathrm{1}}}}  ]  \ottsym{)}$.
     Since ${[S1 = {X :-> iota}]}$ and $S'  \ottsym{(}  \ottmv{X}  \ottsym{)}  \ottsym{=}  \iota$,
     we have $S'  \ottsym{(}  \ottnt{E}  [  \ottnt{f'_{{\mathrm{1}}}}  ]  \ottsym{)}  \ottsym{=}  S'  \ottsym{(}  S_{{\mathrm{1}}}  \ottsym{(}  \ottnt{E}  [  \ottnt{f'_{{\mathrm{1}}}}  ]  \ottsym{)}  \ottsym{)}$.
     Thus, $S'  \ottsym{(}  \ottnt{f}  \ottsym{)} \,  \xmapsto{ \mathmakebox[0.4em]{} [  ] \mathmakebox[0.3em]{} }  \, S'  \ottsym{(}  \ottnt{f'}  \ottsym{)}$.
     By the IH, we finish.

     \case{$S'  \ottsym{(}  \ottmv{X}  \ottsym{)}  \ottsym{=}  \iota'$ for some $\iota'  \neq  \iota$}
     We have $S'  \ottsym{(}  w  \ottsym{:}   \iota \Rightarrow  \unskip ^ { \ell_{{\mathrm{1}}} }  \!  \star \Rightarrow  \unskip ^ { \ell_{{\mathrm{2}}} }  \! \ottmv{X}    \ottsym{)} \,  \xrightarrow{ \mathmakebox[0.4em]{} [  ] \mathmakebox[0.3em]{} }  \, \textsf{\textup{blame}\relax} \, \ell_{{\mathrm{2}}}$.
     Thus, $S'  \ottsym{(}  \ottnt{E}  [  \ottnt{f_{{\mathrm{1}}}}  ]  \ottsym{)} \,  \xmapsto{ \mathmakebox[0.4em]{} [  ] \mathmakebox[0.3em]{} }  \, S'  \ottsym{(}  \ottnt{E}  [  \textsf{\textup{blame}\relax} \, \ell_{{\mathrm{2}}}  ]  \ottsym{)}  \xmapsto{ \mathmakebox[0.4em]{} [  ] \mathmakebox[0.3em]{} }  \textsf{\textup{blame}\relax} \, \ell_{{\mathrm{2}}}$.

     \case{$S'  \ottsym{(}  \ottmv{X}  \ottsym{)}  \ottsym{=}  \ottnt{T_{{\mathrm{1}}}}  \!\rightarrow\!  \ottnt{T_{{\mathrm{2}}}}$}
     We have $S'  \ottsym{(}  w  \ottsym{:}   \iota \Rightarrow  \unskip ^ { \ell_{{\mathrm{1}}} }  \!  \star \Rightarrow  \unskip ^ { \ell_{{\mathrm{2}}} }  \! \ottmv{X}    \ottsym{)} \,  \xrightarrow{ \mathmakebox[0.4em]{} [  ] \mathmakebox[0.3em]{} }  \, S'  \ottsym{(}  w  \ottsym{)}  \ottsym{:}   \iota \Rightarrow  \unskip ^ { \ell_{{\mathrm{1}}} }  \!  \star \Rightarrow  \unskip ^ { \ell_{{\mathrm{2}}} }  \!  \star  \!\rightarrow\!  \star \Rightarrow  \unskip ^ { \ell_{{\mathrm{2}}} }  \! \ottnt{T_{{\mathrm{1}}}}  \!\rightarrow\!  \ottnt{T_{{\mathrm{2}}}}   $.
     Thus, $S'  \ottsym{(}  \ottnt{E}  [  \ottnt{f_{{\mathrm{1}}}}  ]  \ottsym{)} \,  \xmapsto{ \mathmakebox[0.4em]{} [  ] \mathmakebox[0.3em]{} }\hspace{-0.4em}{}^\ast \hspace{0.2em}  \, S'  \ottsym{(}  \ottnt{E}  [  \textsf{\textup{blame}\relax} \, \ell  \ottsym{:}   \star  \!\rightarrow\!  \star \Rightarrow  \unskip ^ { \ell_{{\mathrm{2}}} }  \! \ottnt{T_{{\mathrm{1}}}}  \!\rightarrow\!  \ottnt{T_{{\mathrm{2}}}}   ]  \ottsym{)}  \xrightarrow{ \mathmakebox[0.4em]{} [  ] \mathmakebox[0.3em]{} }  \textsf{\textup{blame}\relax} \, \ell$.

     \case{$S'  \ottsym{(}  \ottmv{X}  \ottsym{)}  \ottsym{=}  \ottmv{X}$}
     We have $S'  \ottsym{(}  w  \ottsym{:}   \iota \Rightarrow  \unskip ^ { \ell_{{\mathrm{1}}} }  \!  \star \Rightarrow  \unskip ^ { \ell_{{\mathrm{2}}} }  \! \ottmv{X}    \ottsym{)} \,  \xrightarrow{ \mathmakebox[0.4em]{} S_{{\mathrm{1}}} \mathmakebox[0.3em]{} }  \, S'  \ottsym{(}  w  \ottsym{)}$.
     Thus, $S'  \ottsym{(}  \ottnt{E}  [  \ottnt{f_{{\mathrm{1}}}}  ]  \ottsym{)} \,  \xmapsto{ \mathmakebox[0.4em]{} S_{{\mathrm{1}}} \mathmakebox[0.3em]{} }  \,  S_{{\mathrm{1}}}  \circ  S'   \ottsym{(}  \ottnt{E}  [  w  ]  \ottsym{)}$.
     Since $ S_{{\mathrm{1}}}  \circ  S'   \ottsym{(}  \ottnt{E}  [  w  ]  \ottsym{)}  \ottsym{=}   S'  \circ  S_{{\mathrm{1}}}   \ottsym{(}  \ottnt{E}  [  w  ]  \ottsym{)}$, we finish by the IH.

     \case{$S'  \ottsym{(}  \ottmv{X}  \ottsym{)}  \ottsym{=}  \ottmv{X'}$ for some $\ottmv{X'}  \neq  \ottmv{X}$}
     We have $S'  \ottsym{(}  w  \ottsym{:}   \iota \Rightarrow  \unskip ^ { \ell_{{\mathrm{1}}} }  \!  \star \Rightarrow  \unskip ^ { \ell_{{\mathrm{2}}} }  \! \ottmv{X}    \ottsym{)} \,  \xrightarrow{ \mathmakebox[0.4em]{} [  \ottmv{X'}  :=  \iota  ] \mathmakebox[0.3em]{} }  \, S'  \ottsym{(}  w  \ottsym{)}$.
     By \rnp{E\_Step},
     $S'  \ottsym{(}  \ottnt{E}  [  \ottnt{f_{{\mathrm{1}}}}  ]  \ottsym{)} \,  \xmapsto{ \mathmakebox[0.4em]{} [  \ottmv{X'}  :=  \iota  ] \mathmakebox[0.3em]{} }  \,  [  \ottmv{X'}  :=  \iota  ]  \circ  S'   \ottsym{(}  \ottnt{E}  [  w  ]  \ottsym{)}$.
     Let $S''$ be a type substitution such that
     $S'  \ottsym{=}   S''  \uplus  [  \ottmv{X}  :=  \ottmv{X'}  ] $.
     Since $ [  \ottmv{X'}  :=  \iota  ]  \circ  S'   \ottsym{(}  \ottnt{E}  [  w  ]  \ottsym{)} =  \ottsym{(}   [  \ottmv{X'}  :=  \iota  ]  \circ  S''   \ottsym{)}  \uplus  [  \ottmv{X}  :=  \iota  ]   \ottsym{(}  \ottnt{E}  [  w  ]  \ottsym{)} =
      [  \ottmv{X'}  :=  \iota  ]  \circ  S''   \ottsym{(}  S_{{\mathrm{1}}}  \ottsym{(}  \ottnt{E}  [  w  ]  \ottsym{)}  \ottsym{)}$.
     By the IH, we finish.
    \end{caseanalysis}

    \case{\rnp{R\_InstArrow}}
    We are given $w  \ottsym{:}   \star  \!\rightarrow\!  \star \Rightarrow  \unskip ^ { \ell_{{\mathrm{1}}} }  \!  \star \Rightarrow  \unskip ^ { \ell_{{\mathrm{2}}} }  \! \ottmv{X}   \,  \xrightarrow{ \mathmakebox[0.4em]{} [  \ottmv{X}  :=  \ottmv{X_{{\mathrm{1}}}}  \!\rightarrow\!  \ottmv{X_{{\mathrm{2}}}}  ] \mathmakebox[0.3em]{} }  \, w  \ottsym{:}   \star  \!\rightarrow\!  \star \Rightarrow  \unskip ^ { \ell_{{\mathrm{1}}} }  \!  \star \Rightarrow  \unskip ^ { \ell_{{\mathrm{2}}} }  \!  \star  \!\rightarrow\!  \star \Rightarrow  \unskip ^ { \ell_{{\mathrm{2}}} }  \! \ottmv{X_{{\mathrm{1}}}}  \!\rightarrow\!  \ottmv{X_{{\mathrm{2}}}}   $
    for some $w$, $\ottmv{X}$, $\ell_{{\mathrm{1}}}$, $\ell_{{\mathrm{2}}}$, and fresh $\ottmv{X_{{\mathrm{1}}}}$ and $\ottmv{X_{{\mathrm{2}}}}$.
    By case analysis on $S'  \ottsym{(}  \ottmv{X}  \ottsym{)}$.
    \begin{caseanalysis}
     \case{$S'  \ottsym{(}  \ottmv{X}  \ottsym{)}  \ottsym{=}  \iota$} Since ${[S'(v : *->* =>l1 * =>l2 X) -->A {} blame l2]}$,
      we finish.

     \case{$S'  \ottsym{(}  \ottmv{X}  \ottsym{)}  \ottsym{=}  \ottnt{T_{{\mathrm{1}}}}  \!\rightarrow\!  \ottnt{T_{{\mathrm{2}}}}$}
     We have $S'  \ottsym{(}  w  \ottsym{:}   \star  \!\rightarrow\!  \star \Rightarrow  \unskip ^ { \ell_{{\mathrm{1}}} }  \!  \star \Rightarrow  \unskip ^ { \ell_{{\mathrm{2}}} }  \! \ottmv{X}    \ottsym{)} \,  \xrightarrow{ \mathmakebox[0.4em]{} [  ] \mathmakebox[0.3em]{} }  \, S'  \ottsym{(}  w  \ottsym{)}  \ottsym{:}   \star  \!\rightarrow\!  \star \Rightarrow  \unskip ^ { \ell_{{\mathrm{1}}} }  \!  \star \Rightarrow  \unskip ^ { \ell_{{\mathrm{2}}} }  \!  \star  \!\rightarrow\!  \star \Rightarrow  \unskip ^ { \ell_{{\mathrm{2}}} }  \! \ottnt{T_{{\mathrm{1}}}}  \!\rightarrow\!  \ottnt{T_{{\mathrm{2}}}}   $.
     Thus, $S'  \ottsym{(}  \ottnt{E}  [  \ottnt{f_{{\mathrm{1}}}}  ]  \ottsym{)} \,  \xrightarrow{ \mathmakebox[0.4em]{} [  ] \mathmakebox[0.3em]{} }  \, S'  \ottsym{(}  \ottnt{E}  [  S'  \ottsym{(}  w  \ottsym{)}  \ottsym{:}   \star  \!\rightarrow\!  \star \Rightarrow  \unskip ^ { \ell_{{\mathrm{1}}} }  \!  \star \Rightarrow  \unskip ^ { \ell_{{\mathrm{2}}} }  \!  \star  \!\rightarrow\!  \star \Rightarrow  \unskip ^ { \ell_{{\mathrm{2}}} }  \! \ottnt{T_{{\mathrm{1}}}}  \!\rightarrow\!  \ottnt{T_{{\mathrm{2}}}}     ]  \ottsym{)}$.
     By the assumption, $\ottmv{X_{{\mathrm{1}}}}, \ottmv{X_{{\mathrm{2}}}} \, \not\in \, \textit{dom} \, \ottsym{(}  S'  \ottsym{)}$.
     Since we can suppose that $\ottmv{X_{{\mathrm{1}}}}, \ottmv{X_{{\mathrm{2}}}} \, \not\in \, \textit{ftv} \, \ottsym{(}  \ottnt{E}  [  w  ]  \ottsym{)}$,
     we have $S'  \ottsym{(}  \ottnt{E}  [  S'  \ottsym{(}  w  \ottsym{)}  \ottsym{:}   \star  \!\rightarrow\!  \star \Rightarrow  \unskip ^ { \ell_{{\mathrm{1}}} }  \!  \star \Rightarrow  \unskip ^ { \ell_{{\mathrm{2}}} }  \!  \star  \!\rightarrow\!  \star \Rightarrow  \unskip ^ { \ell_{{\mathrm{2}}} }  \! \ottnt{T_{{\mathrm{1}}}}  \!\rightarrow\!  \ottnt{T_{{\mathrm{2}}}}     ]  \ottsym{)}  \ottsym{=}   S'  \uplus  [  \ottmv{X_{{\mathrm{1}}}}  :=  \ottnt{T_{{\mathrm{1}}}}  \ottsym{,}  \ottmv{X_{{\mathrm{2}}}}  :=  \ottnt{T_{{\mathrm{2}}}}  ]   \ottsym{(}  \ottnt{E}  [  w  \ottsym{:}   \star  \!\rightarrow\!  \star \Rightarrow  \unskip ^ { \ell_{{\mathrm{1}}} }  \!  \star \Rightarrow  \unskip ^ { \ell_{{\mathrm{2}}} }  \!  \star  \!\rightarrow\!  \star \Rightarrow  \unskip ^ { \ell_{{\mathrm{2}}} }  \! \ottmv{X_{{\mathrm{1}}}}  \!\rightarrow\!  \ottmv{X_{{\mathrm{2}}}}     ]  \ottsym{)}$.
     By the IH, we finish.

     \case{$S'  \ottsym{(}  \ottmv{X}  \ottsym{)}  \ottsym{=}  \ottmv{X}$}
     We have $S'  \ottsym{(}  w  \ottsym{:}   \star  \!\rightarrow\!  \star \Rightarrow  \unskip ^ { \ell_{{\mathrm{1}}} }  \!  \star \Rightarrow  \unskip ^ { \ell_{{\mathrm{2}}} }  \! \ottmv{X}    \ottsym{)} \,  \xrightarrow{ \mathmakebox[0.4em]{} [  \ottmv{X}  :=  \ottmv{X_{{\mathrm{1}}}}  \!\rightarrow\!  \ottmv{X_{{\mathrm{2}}}}  ] \mathmakebox[0.3em]{} }  \, S'  \ottsym{(}  w  \ottsym{)}  \ottsym{:}   \star  \!\rightarrow\!  \star \Rightarrow  \unskip ^ { \ell_{{\mathrm{1}}} }  \!  \star \Rightarrow  \unskip ^ { \ell_{{\mathrm{2}}} }  \!  \star  \!\rightarrow\!  \star \Rightarrow  \unskip ^ { \ell_{{\mathrm{2}}} }  \! \ottmv{X_{{\mathrm{1}}}}  \!\rightarrow\!  \ottmv{X_{{\mathrm{2}}}}   $.
     Thus, $S'  \ottsym{(}  \ottnt{E}  [  \ottnt{f_{{\mathrm{1}}}}  ]  \ottsym{)} \,  \xrightarrow{ \mathmakebox[0.4em]{} S_{{\mathrm{1}}} \mathmakebox[0.3em]{} }  \,  S_{{\mathrm{1}}}  \circ  S'   \ottsym{(}  \ottnt{E}  [  \ottnt{f'_{{\mathrm{1}}}}  ]  \ottsym{)}$.
     Let $S''$ be a type substitution such that
     $\textit{dom} \, \ottsym{(}  S''  \ottsym{)}  \ottsym{=}  \textit{dom} \, \ottsym{(}  S'  \ottsym{)}$ and
     $S''  \ottsym{(}  \ottmv{X'}  \ottsym{)}  \ottsym{=}  S_{{\mathrm{1}}}  \ottsym{(}  S'  \ottsym{(}  \ottmv{X'}  \ottsym{)}  \ottsym{)}$ for any $\ottmv{X'} \, \in \, \textit{dom} \, \ottsym{(}  S''  \ottsym{)}$.
     Since $\ottmv{X_{{\mathrm{1}}}}, \ottmv{X_{{\mathrm{2}}}} \, \not\in \, \textit{dom} \, \ottsym{(}  S'  \ottsym{)}$,
     we have $ S_{{\mathrm{1}}}  \circ  S'   \ottsym{(}  \ottnt{E}  [  \ottnt{f'_{{\mathrm{1}}}}  ]  \ottsym{)}  \ottsym{=}  S''  \ottsym{(}  S_{{\mathrm{1}}}  \ottsym{(}  \ottnt{E}  [  \ottnt{f'_{{\mathrm{1}}}}  ]  \ottsym{)}  \ottsym{)}$.
     By the IH, we finish.

     \case{$S'  \ottsym{(}  \ottmv{X}  \ottsym{)}  \ottsym{=}  \ottmv{X'}$ for some $\ottmv{X'}  \neq  \ottmv{X}$}
     We have $S'  \ottsym{(}  w  \ottsym{:}   \star  \!\rightarrow\!  \star \Rightarrow  \unskip ^ { \ell_{{\mathrm{1}}} }  \!  \star \Rightarrow  \unskip ^ { \ell_{{\mathrm{2}}} }  \! \ottmv{X}    \ottsym{)} \,  \xrightarrow{ \mathmakebox[0.4em]{} [  \ottmv{X'}  :=  \ottmv{X_{{\mathrm{1}}}}  \!\rightarrow\!  \ottmv{X_{{\mathrm{2}}}}  ] \mathmakebox[0.3em]{} }  \, S'  \ottsym{(}  w  \ottsym{)}  \ottsym{:}   \star  \!\rightarrow\!  \star \Rightarrow  \unskip ^ { \ell_{{\mathrm{1}}} }  \!  \star \Rightarrow  \unskip ^ { \ell_{{\mathrm{2}}} }  \!  \star  \!\rightarrow\!  \star \Rightarrow  \unskip ^ { \ell_{{\mathrm{2}}} }  \! \ottmv{X_{{\mathrm{1}}}}  \!\rightarrow\!  \ottmv{X_{{\mathrm{2}}}}   $.
     Thus, $S'  \ottsym{(}  \ottnt{E}  [  \ottnt{f_{{\mathrm{1}}}}  ]  \ottsym{)} \,  \xrightarrow{ \mathmakebox[0.4em]{} [  \ottmv{X'}  :=  \ottmv{X_{{\mathrm{1}}}}  \!\rightarrow\!  \ottmv{X_{{\mathrm{2}}}}  ] \mathmakebox[0.3em]{} }  \,  [  \ottmv{X'}  :=  \ottmv{X_{{\mathrm{1}}}}  \!\rightarrow\!  \ottmv{X_{{\mathrm{2}}}}  ]  \circ  S'   \ottsym{(}  \ottnt{E}  [  \ottnt{f'_{{\mathrm{1}}}}  ]  \ottsym{)}$.
     Let $S''$ be a type substitution such that
     $S'  \ottsym{=}   [  \ottmv{X}  :=  \ottmv{X'}  ]  \uplus  S'' $.
     Then, since $\ottmv{X_{{\mathrm{1}}}}, \ottmv{X_{{\mathrm{2}}}} \, \not\in \, \textit{dom} \, \ottsym{(}  S'  \ottsym{)}$, we have
     $ [  \ottmv{X'}  :=  \ottmv{X_{{\mathrm{1}}}}  \!\rightarrow\!  \ottmv{X_{{\mathrm{2}}}}  ]  \circ  S'   \ottsym{(}  \ottnt{E}  [  \ottnt{f'_{{\mathrm{1}}}}  ]  \ottsym{)}  \ottsym{=}   \ottsym{(}   [  \ottmv{X'}  :=  \ottmv{X_{{\mathrm{1}}}}  \!\rightarrow\!  \ottmv{X_{{\mathrm{2}}}}  ]  \circ  S''   \ottsym{)}  \uplus  [  \ottmv{X}  :=  \ottmv{X_{{\mathrm{1}}}}  \!\rightarrow\!  \ottmv{X_{{\mathrm{2}}}}  ]   \ottsym{(}  \ottnt{E}  [  \ottnt{f'_{{\mathrm{1}}}}  ]  \ottsym{)}
     \ottsym{(}   [  \ottmv{X'}  :=  \ottmv{X_{{\mathrm{1}}}}  \!\rightarrow\!  \ottmv{X_{{\mathrm{2}}}}  ]  \circ  S''   \ottsym{)}  \ottsym{(}  S_{{\mathrm{1}}}  \ottsym{(}  \ottnt{E}  [  \ottnt{f'_{{\mathrm{1}}}}  ]  \ottsym{)}  \ottsym{)}$.
     Thus, by the IH, we finish.
    \end{caseanalysis}

    \case{\rnp{R\_LetP}}
    Since $S'$ does not capture type variables generated by the value substitution,
    we finish by the IH.

    \otherwise Obvious by the IH.
    
   \end{caseanalysis}
   \case{\rnp{E\_Blame}} Obvious.
   \qedhere
  \end{caseanalysis}
 \end{caseanalysis}
\end{proof}



\ifrestate
\thmSoundness*
\else
\begin{theoremA}[name=Soundness of Dynamic Type Inference]
  Suppose $ \emptyset   \vdash  \ottnt{f}  \ottsym{:}  \ottnt{U}$.
  \begin{enumerate}
    \item If $\ottnt{f} \,  \xmapsto{ \mathmakebox[0.4em]{}  S_{{\mathrm{1}}}  \uplus  S_{{\mathrm{2}}}  \mathmakebox[0.3em]{} }\hspace{-0.4em}{}^\ast \hspace{0.2em}  \, \ottnt{r}$
          where $\textit{dom} \, \ottsym{(}  S_{{\mathrm{1}}}  \ottsym{)}  \subseteq  \textit{ftv} \, \ottsym{(}  \ottnt{f}  \ottsym{)}$ and
          $\textit{dom} \, \ottsym{(}  S_{{\mathrm{2}}}  \ottsym{)}  \cap  \textit{ftv} \, \ottsym{(}  \ottnt{f}  \ottsym{)}  \ottsym{=}   \emptyset $,
          then, for any $S'$ such that $\textit{ftv} \, \ottsym{(}  S'  \ottsym{(}  S_{{\mathrm{1}}}  \ottsym{(}  \ottnt{f}  \ottsym{)}  \ottsym{)}  \ottsym{)}  \ottsym{=}   \emptyset $,
          $S'  \ottsym{(}  S_{{\mathrm{1}}}  \ottsym{(}  \ottnt{f}  \ottsym{)}  \ottsym{)} \,  \xmapsto{ \mathmakebox[0.4em]{} S'_{{\mathrm{2}}} \mathmakebox[0.3em]{} }\hspace{-0.4em}{}^\ast \hspace{0.2em}  \, S'  \ottsym{(}  \ottnt{r}  \ottsym{)}$
          for some $S'_{{\mathrm{2}}}$.
    \item If $\ottnt{f} \,  \xmapsto{ \mathmakebox[0.4em]{} S \mathmakebox[0.3em]{} }\hspace{-0.4em}{}^\ast \hspace{0.2em}  \, \textsf{\textup{blame}\relax} \, \ell$,
      then, for any $S'$,
      there exist $S''$ and $\ell'$ such that $S'  \ottsym{(}  \ottnt{f}  \ottsym{)} \,  \xmapsto{ \mathmakebox[0.4em]{} S'' \mathmakebox[0.3em]{} }\hspace{-0.4em}{}^\ast \hspace{0.2em}  \, \textsf{\textup{blame}\relax} \, \ell'$.

    \item If $ \ottnt{f} \!  \Uparrow  $, then,
      for any $S$ such that $\textit{ftv} \, \ottsym{(}  S  \ottsym{(}  \ottnt{f}  \ottsym{)}  \ottsym{)}  \ottsym{=}   \emptyset $, either
      $ S  \ottsym{(}  \ottnt{f}  \ottsym{)} \!  \Uparrow  $ or
      $S  \ottsym{(}  \ottnt{f}  \ottsym{)} \,  \xmapsto{ \mathmakebox[0.4em]{} S' \mathmakebox[0.3em]{} }\hspace{-0.4em}{}^\ast \hspace{0.2em}  \, \textsf{\textup{blame}\relax} \, \ell$ for some $\ell$ and $S'$.
  \end{enumerate}
\end{theoremA}
\fi

\begin{proof}
  \leavevmode
  \begin{enumerate}
  \item \ifrestate
    We prove below an equivalent statement that, if
    $\ottnt{f} \,  \xmapsto{ \mathmakebox[0.4em]{}  S_{{\mathrm{1}}}  \uplus  S_{{\mathrm{2}}}  \mathmakebox[0.3em]{} }\hspace{-0.4em}{}^\ast \hspace{0.2em}  \, \ottnt{r}$ where $\textit{dom} \, \ottsym{(}  S_{{\mathrm{1}}}  \ottsym{)}  \subseteq  \textit{ftv} \, \ottsym{(}  \ottnt{f}  \ottsym{)}$ and
    $\textit{dom} \, \ottsym{(}  S_{{\mathrm{2}}}  \ottsym{)}  \cap  \textit{ftv} \, \ottsym{(}  \ottnt{f}  \ottsym{)}  \ottsym{=}   \emptyset $, then, for any $S'$ such that
    $\textit{ftv} \, \ottsym{(}  S'  \ottsym{(}  S_{{\mathrm{1}}}  \ottsym{(}  \ottnt{f}  \ottsym{)}  \ottsym{)}  \ottsym{)}  \ottsym{=}   \emptyset $, $S'  \ottsym{(}  S_{{\mathrm{1}}}  \ottsym{(}  \ottnt{f}  \ottsym{)}  \ottsym{)} \,  \xmapsto{ \mathmakebox[0.4em]{} S'_{{\mathrm{2}}} \mathmakebox[0.3em]{} }\hspace{-0.4em}{}^\ast \hspace{0.2em}  \, S'  \ottsym{(}  \ottnt{r}  \ottsym{)}$ for
    some $S'_{{\mathrm{2}}}$.  (From the conditions on $S_{{\mathrm{1}}}$ and $S_{{\mathrm{2}}}$, it is
    clear that $S'  \ottsym{(}  S_{{\mathrm{1}}}  \ottsym{(}  \ottnt{f}  \ottsym{)}  \ottsym{)}  \ottsym{=}  S'  \ottsym{(}  S  \ottsym{(}  \ottnt{f}  \ottsym{)}  \ottsym{)}$ where $S  \ottsym{=}   S_{{\mathrm{1}}}  \uplus  S_{{\mathrm{2}}} $.)
    \fi

      By Lemma~\ref{lem:soundness_result},
      $S_{{\mathrm{1}}}  \ottsym{(}  \ottnt{f}  \ottsym{)} \,  \xmapsto{ \mathmakebox[0.4em]{} S'_{{\mathrm{2}}} \mathmakebox[0.3em]{} }\hspace{-0.4em}{}^\ast \hspace{0.2em}  \, \ottnt{r}$ for some $S'_{{\mathrm{2}}}$
      such that $\textit{dom} \, \ottsym{(}  S'_{{\mathrm{2}}}  \ottsym{)}  \subseteq  \textit{dom} \, \ottsym{(}  S_{{\mathrm{2}}}  \ottsym{)}$ and,
      for any $\ottmv{X} \, \in \, \textit{dom} \, \ottsym{(}  S_{{\mathrm{1}}}  \ottsym{)}$, $\textit{ftv} \, \ottsym{(}  S_{{\mathrm{1}}}  \ottsym{(}  \ottmv{X}  \ottsym{)}  \ottsym{)}  \cap  \textit{dom} \, \ottsym{(}  S'_{{\mathrm{2}}}  \ottsym{)}  \ottsym{=}   \emptyset $.
      Let $S''$ be a type substitution such that
      $\textit{dom} \, \ottsym{(}  S''  \ottsym{)}  \ottsym{=}  \textit{ftv} \, \ottsym{(}  S_{{\mathrm{1}}}  \ottsym{(}  \ottnt{f}  \ottsym{)}  \ottsym{)}$ and
      $S''  \ottsym{(}  \ottmv{X}  \ottsym{)}  \ottsym{=}  S'  \ottsym{(}  \ottmv{X}  \ottsym{)}$ for any $\ottmv{X} \, \in \, \textit{dom} \, \ottsym{(}  S''  \ottsym{)}$.
      Since $\textit{dom} \, \ottsym{(}  S_{{\mathrm{2}}}  \ottsym{)}  \cap  \textit{ftv} \, \ottsym{(}  \ottnt{f}  \ottsym{)}  \ottsym{=}   \emptyset $, we have
      $\textit{dom} \, \ottsym{(}  S'_{{\mathrm{2}}}  \ottsym{)}  \cap  \textit{ftv} \, \ottsym{(}  \ottnt{f}  \ottsym{)}  \ottsym{=}   \emptyset $.
      Furthermore,
      for any $\ottmv{X} \, \in \, \textit{dom} \, \ottsym{(}  S_{{\mathrm{1}}}  \ottsym{)}$, $\textit{ftv} \, \ottsym{(}  S_{{\mathrm{1}}}  \ottsym{(}  \ottmv{X}  \ottsym{)}  \ottsym{)}  \cap  \textit{dom} \, \ottsym{(}  S'_{{\mathrm{2}}}  \ottsym{)}  \ottsym{=}   \emptyset $,
      and $\textit{dom} \, \ottsym{(}  S''  \ottsym{)}  \ottsym{=}  \textit{ftv} \, \ottsym{(}  S_{{\mathrm{1}}}  \ottsym{(}  \ottnt{f}  \ottsym{)}  \ottsym{)}$.
      Thus, $\textit{dom} \, \ottsym{(}  S''  \ottsym{)}  \cap  \textit{dom} \, \ottsym{(}  S'_{{\mathrm{2}}}  \ottsym{)}  \ottsym{=}   \emptyset $.
      Without loss of generality, we can suppose that
      $\textit{dom} \, \ottsym{(}  S''  \ottsym{)}$  are disjoint from type variables generated by
      $S_{{\mathrm{1}}}  \ottsym{(}  \ottnt{f}  \ottsym{)} \,  \xmapsto{ \mathmakebox[0.4em]{} S'_{{\mathrm{2}}} \mathmakebox[0.3em]{} }\hspace{-0.4em}{}^\ast \hspace{0.2em}  \, \ottnt{r}$.
      Since types to which $S''$ maps have no free type variables,
      we have $S''  \ottsym{(}  S_{{\mathrm{1}}}  \ottsym{(}  \ottnt{f}  \ottsym{)}  \ottsym{)} \,  \xmapsto{ \mathmakebox[0.4em]{} S'_{{\mathrm{2}}} \mathmakebox[0.3em]{} }\hspace{-0.4em}{}^\ast \hspace{0.2em}  \, S''  \ottsym{(}  \ottnt{r}  \ottsym{)}$
      by Lemma~\ref{lem:eval_multi_step_any_fresh_subst}.
      Since $S''  \ottsym{(}  S_{{\mathrm{1}}}  \ottsym{(}  \ottnt{f}  \ottsym{)}  \ottsym{)}$ has no free type variables,
      free type variables in $S''  \ottsym{(}  \ottnt{r}  \ottsym{)}$ are generated during the
      evaluation; thus, we can suppose that $S''  \ottsym{(}  \ottnt{r}  \ottsym{)}  \ottsym{=}  S'  \ottsym{(}  \ottnt{r}  \ottsym{)}$.
      Since $S''  \ottsym{(}  S_{{\mathrm{1}}}  \ottsym{(}  \ottnt{f}  \ottsym{)}  \ottsym{)}  \ottsym{=}  S'  \ottsym{(}  S_{{\mathrm{1}}}  \ottsym{(}  \ottnt{f}  \ottsym{)}  \ottsym{)}$,
      we have $S'  \ottsym{(}  S_{{\mathrm{1}}}  \ottsym{(}  \ottnt{f}  \ottsym{)}  \ottsym{)} \,  \xmapsto{ \mathmakebox[0.4em]{} S'_{{\mathrm{2}}} \mathmakebox[0.3em]{} }\hspace{-0.4em}{}^\ast \hspace{0.2em}  \, S'  \ottsym{(}  \ottnt{r}  \ottsym{)}$.

    \item
      We can suppose that type variables generated by the evaluation
      $\ottnt{f} \,  \xmapsto{ \mathmakebox[0.4em]{} S \mathmakebox[0.3em]{} }\hspace{-0.4em}{}^\ast \hspace{0.2em}  \, \textsf{\textup{blame}\relax} \, \ell$ are disjoint from $\textit{dom} \, \ottsym{(}  S'  \ottsym{)}$.
      Thus, we finish by Lemma~\ref{lem:soundness_blame_2}.

    \item  Suppose $S  \ottsym{(}  \ottnt{f}  \ottsym{)} \,  \xmapsto{ \mathmakebox[0.4em]{} S' \mathmakebox[0.3em]{} }\hspace{-0.4em}{}^\ast \hspace{0.2em}  \, w'$ for some $S'$ and $w'$.
      Then, by Theorem \ref{thm:completeness}, it must be the case that
      $\ottnt{f} \,  \xmapsto{ \mathmakebox[0.4em]{} S'' \mathmakebox[0.3em]{} }\hspace{-0.4em}{}^\ast \hspace{0.2em}  \, w$ for some $w$ and $S''$, contradicting $ \ottnt{f} \!  \Uparrow  $.
      \qedhere
  \end{enumerate}
\end{proof}

\subsection{Correctness of Cast Insertion and Type Safety of the ITGL}
\begin{definitionA}
 $\Gamma  \rightsquigarrow  \Gamma'$ is the least relation satisfying the following rules.
 \begin{center}
  $\ottdruleCIXXEmpty{}$ \qquad
  $\ottdruleCIXXExtendVar{}$
 \end{center}
\end{definitionA}

\begin{lemmaA} \label{lem:ci_val}
 If $\Gamma  \vdash  v  \rightsquigarrow  \ottnt{f}  \ottsym{:}  \ottnt{U}$, then $\ottnt{f}  \ottsym{=}  w$ for some $w$.
\end{lemmaA}
\begin{proof}
 Straightforward by case analysis on the cast insertion rule applied last to
 derive $\Gamma  \vdash  v  \rightsquigarrow  \ottnt{f}  \ottsym{:}  \ottnt{U}$.
\end{proof}

\begin{lemmaA} \label{lem:ci_typed_insert}
 If $\Gamma  \vdash  e  \ottsym{:}  \ottnt{U}$ and $\Gamma  \rightsquigarrow  \Gamma'$,
 then $\Gamma'  \vdash  e  \rightsquigarrow  \ottnt{f}  \ottsym{:}  \ottnt{U}$ for some $\ottnt{f}$.
  \end{lemmaA}
\begin{proof}
 By induction on the derivation of $\Gamma  \vdash  e  \ottsym{:}  \ottnt{U}$.
 \begin{caseanalysis}
  \case{\rnp{IT\_VarP}}
  We are given $\Gamma  \vdash  \ottmv{x}  \ottsym{:}  \ottnt{U'}  [   \overrightarrow{ \ottmv{X_{\ottmv{i}}} }   \ottsym{:=}   \overrightarrow{ \ottnt{T_{\ottmv{i}}} }   ]$
  for some $\ottmv{x}$, $\ottnt{U'}$, $ \overrightarrow{ \ottmv{X_{\ottmv{i}}} } $, and $ \overrightarrow{ \ottnt{T_{\ottmv{i}}} } $ and, by inversion,
  $\ottmv{x}  \ottsym{:}  \forall \,  \overrightarrow{ \ottmv{X_{\ottmv{i}}} }   \ottsym{.}  \ottnt{U'} \, \in \, \Gamma$.
  Since $\Gamma  \rightsquigarrow  \Gamma'$, we have
  $\ottmv{x}  \ottsym{:}  \forall \,  \overrightarrow{ \ottmv{X_{\ottmv{i}}} }  \,  \overrightarrow{ \ottmv{Y_{\ottmv{j}}} }   \ottsym{.}  \ottnt{U'} \, \in \, \Gamma'$ for some $ \overrightarrow{ \ottmv{Y_{\ottmv{j}}} } $ such that
  $ \overrightarrow{ \ottmv{Y_{\ottmv{j}}} }   \cap  \textit{ftv} \, \ottsym{(}  \ottnt{U'}  \ottsym{)}  \ottsym{=}   \emptyset $.
  Let $ \overrightarrow{ \ottmv{X'_{\ottmv{k}}} } $ be the concatenation of $ \overrightarrow{ \ottmv{X_{\ottmv{i}}} } $ and $ \overrightarrow{ \ottmv{Y_{\ottmv{j}}} } $ and
  $ \overrightarrow{ \mathbbsl{T}_{\ottmv{k}} } $ be a sequence of optional types such that:
  if $\ottmv{X'_{\ottmv{n}}} \, \in \,  \overrightarrow{ \ottmv{X'_{\ottmv{k}}} } $ appears in $\ottnt{U'}$ (thus $\ottmv{X'_{\ottmv{n}}} \, \in \,  \overrightarrow{ \ottmv{X_{\ottmv{i}}} } $),
  $ \mathbbsl{T}_{\ottmv{n}}  =  \ottnt{T_{\ottmv{j}}} $ where $\ottmv{j}$ is a number such that $\ottmv{X'_{\ottmv{n}}} = \ottmv{X_{\ottmv{j}}} \, \in \,  \overrightarrow{ \ottmv{X_{\ottmv{i}}} } $; and
  if $\ottmv{X'_{\ottmv{n}}} \, \in \,  \overrightarrow{ \ottmv{X'_{\ottmv{k}}} } $ does not appear in $\ottnt{U'}$, $ \mathbbsl{T}'_{\ottmv{n}}  =   \nu  $.
  Then, $\ottnt{U'}  [   \overrightarrow{ \ottmv{X_{\ottmv{i}}} }   \ottsym{:=}   \overrightarrow{ \ottnt{T_{\ottmv{i}}} }   ]  \ottsym{=}  \ottnt{U'}  [   \overrightarrow{ \ottmv{X'_{\ottmv{k}}} }   \ottsym{:=}   \overrightarrow{ \mathbbsl{T}_{\ottmv{k}} }   ]$.
  By \rnp{CI\_VarP}, we have
  $\Gamma  \vdash  \ottmv{x}  \rightsquigarrow  \ottmv{x}  [   \overrightarrow{ \mathbbsl{T}_{\ottmv{i}} }   ]  \ottsym{:}  \ottnt{U'}  [   \overrightarrow{ \ottmv{X_{\ottmv{i}}} }   \ottsym{:=}   \overrightarrow{ \ottnt{T_{\ottmv{i}}} }   ]$.

  \case{\rnp{IT\_Const}} By \rnp{CI\_Const}.

  \case{\rnp{IT\_Op}}
  We are given $\Gamma  \vdash  \mathit{op} \, \ottsym{(}  e_{{\mathrm{1}}}  \ottsym{,}  e_{{\mathrm{2}}}  \ottsym{)}  \ottsym{:}  \iota$ for some $ \mathit{op} $, $e_{{\mathrm{1}}}$,
  $e_{{\mathrm{2}}}$, and $\iota$ and, by inversion,
  \begin{itemize}
   \item $\Gamma  \vdash  e_{{\mathrm{1}}}  \ottsym{:}  \ottnt{U_{{\mathrm{1}}}}$,
   \item $\Gamma  \vdash  e_{{\mathrm{2}}}  \ottsym{:}  \ottnt{U_{{\mathrm{2}}}}$,
   \item $ \mathit{ty} ( \mathit{op} )   \ottsym{=}  \iota_{{\mathrm{1}}}  \!\rightarrow\!  \iota_{{\mathrm{2}}}  \!\rightarrow\!  \iota$, and
   \item $\ottnt{U_{{\mathrm{1}}}}  \sim  \iota_{{\mathrm{1}}}$ and $\ottnt{U_{{\mathrm{2}}}}  \sim  \iota_{{\mathrm{2}}}$
  \end{itemize}
  for some $\iota_{{\mathrm{1}}}$ and $\iota_{{\mathrm{2}}}$.
  By the IHs, $\Gamma'  \vdash  e_{{\mathrm{1}}}  \rightsquigarrow  \ottnt{f_{{\mathrm{1}}}}  \ottsym{:}  \ottnt{U_{{\mathrm{1}}}}$ and $\Gamma'  \vdash  e_{{\mathrm{2}}}  \rightsquigarrow  \ottnt{f_{{\mathrm{2}}}}  \ottsym{:}  \ottnt{U_{{\mathrm{2}}}}$
  for some $\ottnt{f_{{\mathrm{1}}}}$ and $\ottnt{f_{{\mathrm{2}}}}$.
  By \rnp{CI\_Op}, we have
  $\Gamma'  \vdash  \mathit{op} \, \ottsym{(}  e_{{\mathrm{1}}}  \ottsym{,}  e_{{\mathrm{2}}}  \ottsym{)}  \rightsquigarrow  \mathit{op} \, \ottsym{(}  \ottnt{f_{{\mathrm{1}}}}  \ottsym{:}   \ottnt{U_{{\mathrm{1}}}} \Rightarrow  \unskip ^ { \ell_{{\mathrm{1}}} }  \! \iota_{{\mathrm{1}}}   \ottsym{,}  \ottnt{f_{{\mathrm{2}}}}  \ottsym{:}   \ottnt{U_{{\mathrm{2}}}} \Rightarrow  \unskip ^ { \ell_{{\mathrm{2}}} }  \! \iota_{{\mathrm{2}}}   \ottsym{)}  \ottsym{:}  \iota$
  for some $\ell_{{\mathrm{1}}}$ and $\ell_{{\mathrm{2}}}$.

  \case{\rnp{IT\_AbsI}} By the IH and \rnp{CI\_AbsI}
  since $ \Gamma ,   \ottmv{x}  :  \ottnt{T}    \rightsquigarrow   \Gamma' ,   \ottmv{x}  :  \ottnt{T}  $.

  \case{\rnp{IT\_AbsE}} By the IH and \rnp{CI\_AbsE}
  since $ \Gamma ,   \ottmv{x}  :  \ottnt{U}    \rightsquigarrow   \Gamma' ,   \ottmv{x}  :  \ottnt{U}  $.

  \case{\rnp{IT\_App}} By the IHs and \rnp{CI\_App}.

  \case{\rnp{IT\_LetP}}
  We are given
  $\Gamma  \vdash   \textsf{\textup{let}\relax} \,  \ottmv{x}  =  v_{{\mathrm{1}}}  \textsf{\textup{ in }\relax}  e_{{\mathrm{2}}}   \ottsym{:}  \ottnt{U}$ for some $\ottmv{x}$, $v_{{\mathrm{1}}}$, and $e_{{\mathrm{2}}}$
  and, by inversion,
  \begin{itemize}
   \item $\Gamma  \vdash  v_{{\mathrm{1}}}  \ottsym{:}  \ottnt{U_{{\mathrm{1}}}}$,
   \item $ \Gamma ,   \ottmv{x}  :  \forall \,  \overrightarrow{ \ottmv{X_{\ottmv{i}}} }   \ottsym{.}  \ottnt{U_{{\mathrm{1}}}}    \vdash  e_{{\mathrm{2}}}  \ottsym{:}  \ottnt{U}$, and
   \item $ \overrightarrow{ \ottmv{X_{\ottmv{i}}} }   \ottsym{=}  \textit{ftv} \, \ottsym{(}  \ottnt{U_{{\mathrm{1}}}}  \ottsym{)}  \setminus  \ottsym{(}  \textit{ftv} \, \ottsym{(}  \Gamma  \ottsym{)}  \cup  \textit{ftv} \, \ottsym{(}  v_{{\mathrm{1}}}  \ottsym{)}  \ottsym{)}$
  \end{itemize}
  for some $\ottnt{U_{{\mathrm{1}}}}$.
  By the IH and Lemma~\ref{lem:ci_val}, $\Gamma'  \vdash  v_{{\mathrm{1}}}  \rightsquigarrow  w_{{\mathrm{1}}}  \ottsym{:}  \ottnt{U_{{\mathrm{1}}}}$
  for some $w_{{\mathrm{1}}}$.
  Since $\textit{ftv} \, \ottsym{(}  \Gamma'  \ottsym{)}  \ottsym{=}  \textit{ftv} \, \ottsym{(}  \Gamma  \ottsym{)}$,
  $ \overrightarrow{ \ottmv{X_{\ottmv{i}}} }   \ottsym{=}  \textit{ftv} \, \ottsym{(}  \ottnt{U_{{\mathrm{1}}}}  \ottsym{)}  \setminus  \ottsym{(}  \textit{ftv} \, \ottsym{(}  \Gamma'  \ottsym{)}  \cup  \textit{ftv} \, \ottsym{(}  v_{{\mathrm{1}}}  \ottsym{)}  \ottsym{)}$.
  Let $ \overrightarrow{ \ottmv{Y_{\ottmv{j}}} }   \ottsym{=}  \textit{ftv} \, \ottsym{(}  w_{{\mathrm{1}}}  \ottsym{)}  \setminus  \ottsym{(}  \textit{ftv} \, \ottsym{(}  \Gamma'  \ottsym{)}  \cup  \textit{ftv} \, \ottsym{(}  \ottnt{U_{{\mathrm{1}}}}  \ottsym{)}  \cup  \textit{ftv} \, \ottsym{(}  v_{{\mathrm{1}}}  \ottsym{)}  \ottsym{)}$.
  Since $\Gamma  \rightsquigarrow  \Gamma'$, we have
  $ \Gamma ,   \ottmv{x}  :  \forall \,  \overrightarrow{ \ottmv{X_{\ottmv{i}}} }   \ottsym{.}  \ottnt{U_{{\mathrm{1}}}}    \rightsquigarrow   \Gamma' ,   \ottmv{x}  :  \forall \,  \overrightarrow{ \ottmv{X_{\ottmv{i}}} }  \,  \overrightarrow{ \ottmv{Y_{\ottmv{i}}} }   \ottsym{.}  \ottnt{U_{{\mathrm{1}}}}  $.
  Thus, by the IH,
  $ \Gamma' ,   \ottmv{x}  :  \forall \,  \overrightarrow{ \ottmv{X_{\ottmv{i}}} }  \,  \overrightarrow{ \ottmv{Y_{\ottmv{i}}} }   \ottsym{.}  \ottnt{U_{{\mathrm{1}}}}    \vdash  e_{{\mathrm{2}}}  \rightsquigarrow  \ottnt{f_{{\mathrm{2}}}}  \ottsym{:}  \ottnt{U}$ for some $\ottnt{f_{{\mathrm{2}}}}$.
  By \rnp{CI\_LetP},
  $\Gamma'  \vdash   \textsf{\textup{let}\relax} \,  \ottmv{x}  =  v_{{\mathrm{1}}}  \textsf{\textup{ in }\relax}  e_{{\mathrm{2}}}   \rightsquigarrow   \textsf{\textup{let}\relax} \,  \ottmv{x}  =   \Lambda    \overrightarrow{ \ottmv{X_{\ottmv{i}}} }     \overrightarrow{ \ottmv{Y_{\ottmv{j}}} }  .\,  w_{{\mathrm{1}}}   \textsf{\textup{ in }\relax}  \ottnt{f_{{\mathrm{2}}}}   \ottsym{:}  \ottnt{U}$.
  \qedhere
 \end{caseanalysis}
\end{proof}

\begin{lemmaA} \label{lem:ci_insert_typed}
 If $\Gamma  \vdash  e  \rightsquigarrow  \ottnt{f}  \ottsym{:}  \ottnt{U}$, then $\Gamma  \vdash  \ottnt{f}  \ottsym{:}  \ottnt{U}$.
\end{lemmaA}
\begin{proof}
 By induction on the derivation of $\Gamma  \vdash  e  \rightsquigarrow  \ottnt{f}  \ottsym{:}  \ottnt{U}$.
 \begin{caseanalysis}
  \case{\rnp{CI\_VarP}} By \rnp{T\_VarP}.
  \case{\rnp{CI\_Const}} By \rnp{T\_Const}.
  \case{\rnp{CI\_Op}} By the IHs, \rnp{T\_Cast}, and \rnp{T\_Op}.
  \case{\rnp{CI\_AbsI}} By the IH and \rnp{T\_Abs}.
  \case{\rnp{CI\_AbsE}} By the IH and \rnp{T\_Abs}.
  \case{\rnp{CI\_App}} By the IHs, \rnp{T\_App}, and the fact that
  $\ottnt{U}  \triangleright  \ottnt{U_{{\mathrm{1}}}}  \!\rightarrow\!  \ottnt{U_{{\mathrm{2}}}}$ implies $\ottnt{U}  \sim  \ottnt{U_{{\mathrm{1}}}}  \!\rightarrow\!  \ottnt{U_{{\mathrm{2}}}}$.
  \case{\rnp{CI\_LetP}} By the IHs and \rnp{T\_LetP}.  \qedhere
 \end{caseanalysis}
\end{proof}

\ifrestate
\thmCastInsertion*
\else
\begin{theoremA}[Cast Insertion is Type-Preserving]\label{thm:ci_type_preservation}
 If $\Gamma  \vdash  e  \ottsym{:}  \ottnt{U}$,
 then $\Gamma  \vdash  e  \rightsquigarrow  \ottnt{f}  \ottsym{:}  \ottnt{U}$ and $\Gamma  \vdash  \ottnt{f}  \ottsym{:}  \ottnt{U}$ for some $\ottnt{f}$.
\end{theoremA}
\fi

\begin{proof}
 Since $\Gamma  \rightsquigarrow  \Gamma$, we conclude by Lemmas~\ref{lem:ci_typed_insert} and
 \ref{lem:ci_insert_typed}.
\end{proof}

\ifrestate
\corTypeSafeITGL*
\else
\begin{corollaryA}[name=Type Safety of the ITGL]
 If $ \emptyset   \vdash  e  \ottsym{:}  \ottnt{U}$, then:
 \begin{itemize}
  \item $ \langle   \emptyset    \vdash   e  :  \ottnt{U}  \rangle   \xmapsto{ \mathmakebox[0.4em]{} S \mathmakebox[0.3em]{} }\hspace{-0.4em}{}^\ast \hspace{0.2em}    w $ for some $S$ and $w$
        such that $ \emptyset   \vdash  w  \ottsym{:}  S  \ottsym{(}  \ottnt{U}  \ottsym{)}$;
  \item $ \langle   \emptyset    \vdash   e  :  \ottnt{U}  \rangle   \xmapsto{ \mathmakebox[0.4em]{} S \mathmakebox[0.3em]{} }\hspace{-0.4em}{}^\ast \hspace{0.2em}    \textsf{\textup{blame}\relax} \, \ell $
        for some $S$ and $\ell$; or
  \item $ \langle   \emptyset    \vdash   e  :  \ottnt{U}  \rangle   \Uparrow  $.
 \end{itemize}
\end{corollaryA}
\fi 

\begin{proof}
 By Theorems~\ref{thm:ci_type_preservation} and \ref{thm:type_safety}.
\end{proof}

\subsection{Gradual Guarantee}

\subsubsection{Static Gradual Guarantee in ITGL}


\begin{lemmaA} \label{lem:type_prec_consistent_ground1}
 If $ \ottnt{U_{{\mathrm{1}}}}   \sqsubseteq _{ [  ] }  \ottnt{U'_{{\mathrm{1}}}} $ and $\ottnt{U_{{\mathrm{1}}}}  \sim  \ottnt{U_{{\mathrm{2}}}}$ and $ \ottnt{U_{{\mathrm{2}}}}   \sqsubseteq _{ [  ] }  \ottnt{U'_{{\mathrm{2}}}} $,
 then $\ottnt{U'_{{\mathrm{1}}}}  \sim  \ottnt{U'_{{\mathrm{2}}}}$.
\end{lemmaA}
\begin{proof}
  By induction on $\ottnt{U_{{\mathrm{1}}}}  \sim  \ottnt{U_{{\mathrm{2}}}}$.  Base cases follow from simple case analysis
  on $ \ottnt{U_{{\mathrm{1}}}}   \sqsubseteq _{ [  ] }  \ottnt{U'_{{\mathrm{1}}}} $ and $ \ottnt{U_{{\mathrm{2}}}}   \sqsubseteq _{ [  ] }  \ottnt{U'_{{\mathrm{2}}}} $.
  Case \rn{C\_Arrow} is easy.
\end{proof}

\begin{lemmaA}\label{lem:prec-trans1}
  If $ \ottnt{U}   \sqsubseteq _{ S_{{\mathrm{1}}} }  \ottnt{U'} $ and $ \ottnt{U'}   \sqsubseteq _{ S_{{\mathrm{2}}} }  \ottnt{U''} $, then
    $ \ottnt{U}   \sqsubseteq _{  S_{{\mathrm{1}}}  \circ  S_{{\mathrm{2}}}  }  \ottnt{U''} $.
\end{lemmaA}

\begin{proof}
  By induction on the sum of the sizes of derivations of $ \ottnt{U}   \sqsubseteq _{ S_{{\mathrm{1}}} }  \ottnt{U'} $ and $ \ottnt{U'}   \sqsubseteq _{ S_{{\mathrm{2}}} }  \ottnt{U''} $.
\end{proof}

\begin{lemmaA}\label{lem:prec-trans2}
  If $ \langle  \Gamma   \vdash   e  :  \ottnt{U} 
                 \sqsubseteq _{ S_{{\mathrm{1}}} }  \ottnt{U'}  :  e'  \dashv  \Gamma'  \rangle $ and $ \langle  \Gamma'   \vdash   e'  :  \ottnt{U'} 
                 \sqsubseteq _{ S_{{\mathrm{2}}} }  \ottnt{U''}  :  e''  \dashv  \Gamma''  \rangle $, then
    $ \langle  \Gamma   \vdash   e  :  \ottnt{U} 
                 \sqsubseteq _{ \ottsym{(}   S_{{\mathrm{1}}}  \circ  S_{{\mathrm{2}}}   \ottsym{)} }  \ottnt{U''}  :  e''  \dashv  \Gamma''  \rangle $.
  \end{lemmaA}
  \begin{proof}
    By induction on $e'$.  Use Lemma~\ref{lem:prec-trans1}.
      \AI{This has to be checked carefully.}
  \end{proof}

\begin{lemmaA}\label{lem:prec-prop2}
  If $ \langle  \Gamma   \vdash   e  :  \ottnt{U} 
                 \sqsubseteq _{ S_{{\mathrm{1}}} }  S_{{\mathrm{2}}}  \ottsym{(}  \ottnt{U'}  \ottsym{)}  :  S_{{\mathrm{2}}}  \ottsym{(}  e'  \ottsym{)}  \dashv  S_{{\mathrm{2}}}  \ottsym{(}  \Gamma'  \ottsym{)}  \rangle $, then $ \langle  \Gamma   \vdash   e  :  \ottnt{U} 
                 \sqsubseteq _{  S_{{\mathrm{1}}}  \circ  S_{{\mathrm{2}}}  }  \ottnt{U'}  :  e'  \dashv  \Gamma'  \rangle $.
\end{lemmaA}

\begin{proof}
  By induction on $ \langle  \Gamma   \vdash   e  :  \ottnt{U} 
                 \sqsubseteq _{ S_{{\mathrm{1}}} }  S_{{\mathrm{2}}}  \ottsym{(}  \ottnt{U'}  \ottsym{)}  :  S_{{\mathrm{2}}}  \ottsym{(}  e'  \ottsym{)}  \dashv  S_{{\mathrm{2}}}  \ottsym{(}  \Gamma'  \ottsym{)}  \rangle $.
  \AI{This has to be checked carefully.}
\end{proof}

\begin{lemmaA}
  \label{thm:uprec-prec}
  If
  $ \Gamma   \sqsubseteq _{ S_{{\mathrm{0}}} }  \Gamma' $ and
  $ e   \sqsubseteq _{ S_{{\mathrm{0}}} }  e' $ and
  $S  \ottsym{(}  \Gamma  \ottsym{)}  \vdash  S  \ottsym{(}  e  \ottsym{)}  \ottsym{:}  \ottnt{U}$, then
  there exist $S'$ and $\ottnt{U'}$ such that
  $S'  \ottsym{(}  \Gamma'  \ottsym{)}  \vdash  S'  \ottsym{(}  e'  \ottsym{)}  \ottsym{:}  \ottnt{U'}$ and
  $ \langle  S  \ottsym{(}  \Gamma  \ottsym{)}   \vdash   S  \ottsym{(}  e  \ottsym{)}  :  \ottnt{U} 
                 \sqsubseteq _{ [  ] }  \ottnt{U'}  :  S'  \ottsym{(}  e'  \ottsym{)}  \dashv  S'  \ottsym{(}  \Gamma'  \ottsym{)}  \rangle $.   
\end{lemmaA}

\begin{proof}
  We show that there exists $\ottnt{U'}$ such that
  $\ottsym{(}   S  \circ  S_{{\mathrm{0}}}   \ottsym{)}  \ottsym{(}  \Gamma'  \ottsym{)}  \vdash  S'  \ottsym{(}  e'  \ottsym{)}  \ottsym{:}  \ottnt{U'}$ and
  $ \langle  S  \ottsym{(}  \Gamma  \ottsym{)}   \vdash   S  \ottsym{(}  e  \ottsym{)}  :  \ottnt{U} 
                 \sqsubseteq _{ [  ] }  \ottnt{U'}  :  \ottsym{(}   S  \circ  S_{{\mathrm{0}}}   \ottsym{)}  \ottsym{(}  e'  \ottsym{)}  \dashv  \ottsym{(}   S  \circ  S_{{\mathrm{0}}}   \ottsym{)}  \ottsym{(}  \Gamma'  \ottsym{)}  \rangle $
  by induction on $ e   \sqsubseteq _{ S_{{\mathrm{0}}} }  e' $.

  Most cases are straightforward; Case \rn{IP\_App} follows from Lemma~\ref{lem:type_prec_consistent_ground1}.
  \AI{Case for let should be checked.}
\end{proof}

\begin{lemmaA}
  \label{thm:uprec-PT-prec}
  If
  $S'  \ottsym{(}  \Gamma  \ottsym{)}  \vdash  S'  \ottsym{(}  e  \ottsym{)}  \ottsym{:}  \ottnt{U'}$, then
  there exist some $S$ and $S''$ and $\ottnt{U}$ such that
  $PT(\Gamma, e) = (S, \ottnt{U})$ and
  $S  \ottsym{(}  \Gamma  \ottsym{)}  \vdash  S  \ottsym{(}  e  \ottsym{)}  \ottsym{:}  \ottnt{U}$ and
  $ \langle  S'  \ottsym{(}  \Gamma  \ottsym{)}   \vdash   S'  \ottsym{(}  e  \ottsym{)}  :  \ottnt{U'} 
                 \sqsubseteq _{ S'' }  \ottnt{U}  :  S  \ottsym{(}  e  \ottsym{)}  \dashv  S  \ottsym{(}  \Gamma  \ottsym{)}  \rangle $.
\end{lemmaA}

\begin{proof}
  By completeness of $PT$, there exists $S$, $S''$ and $U$ such that
  $S  \ottsym{(}  \Gamma  \ottsym{)}  \vdash  S  \ottsym{(}  e  \ottsym{)}  \ottsym{:}  \ottnt{U}$ and $S' =  S''  \circ  S $ and $\ottnt{U'} = S''  \ottsym{(}  \ottnt{U}  \ottsym{)}$.
  Since precision is reflexive, we have
  $ \langle  S'  \ottsym{(}  \Gamma  \ottsym{)}   \vdash   S'  \ottsym{(}  e  \ottsym{)}  :  \ottnt{U'} 
                 \sqsubseteq _{ [  ] }  \ottnt{U'}  :  S'  \ottsym{(}  e  \ottsym{)}  \dashv  S'  \ottsym{(}  \Gamma  \ottsym{)}  \rangle $.
  Lemma~\ref{lem:prec-prop2} finishes the proof.
\end{proof}

\begin{lemmaA}
  \label{thm:cast-insertion-preserves-precision}
  If
  $ \langle  \Gamma   \vdash   e  :  \ottnt{U} 
                 \sqsubseteq _{ S_{{\mathrm{0}}} }  \ottnt{U'}  :  e'  \dashv  \Gamma'  \rangle $ and $\Gamma  \rightsquigarrow  \Gamma_{{\mathrm{0}}}$ and $\Gamma'  \rightsquigarrow  \Gamma'_{{\mathrm{0}}}$, then
  there exist $\ottnt{f}$ and $\ottnt{f'}$ such that
  $\Gamma_{{\mathrm{0}}}  \vdash  e  \rightsquigarrow  \ottnt{f}  \ottsym{:}  \ottnt{U}$ and
  $\Gamma'_{{\mathrm{0}}}  \vdash  e'  \rightsquigarrow  \ottnt{f'}  \ottsym{:}  \ottnt{U'}$ and
  $ \langle  \Gamma_{{\mathrm{0}}}   \vdash   \ottnt{f}  :  \ottnt{U}   \sqsubseteq _{ S_{{\mathrm{0}}} }  \ottnt{U'}  :  \ottnt{f'}  \dashv  \Gamma'_{{\mathrm{0}}}  \rangle $.
\end{lemmaA}

\begin{proof}
  By induction on $ \langle  \Gamma   \vdash   e  :  \ottnt{U} 
                 \sqsubseteq _{ S_{{\mathrm{0}}} }  \ottnt{U'}  :  e'  \dashv  \Gamma'  \rangle $.
  \AI{This has to be checked carefully.}
\end{proof}

\begin{lemmaA}
  \label{thm:ITGL-staticGG}
  If
  $ \Gamma   \sqsubseteq _{ S_{{\mathrm{0}}} }  \Gamma' $ and
  $ e   \sqsubseteq _{ S_{{\mathrm{0}}} }  e' $ and
  $PT(\Gamma, e) = (S, \ottnt{U})$, then
  there exist $S'$, $S''$, $\ottnt{U'}$, $\ottnt{f}$ and $\ottnt{f'}$
  such that
  $PT(\Gamma', e') = (S', \ottnt{U'})$ and
  $S  \ottsym{(}  \Gamma  \ottsym{)}  \vdash  S  \ottsym{(}  e  \ottsym{)}  \rightsquigarrow  \ottnt{f}  \ottsym{:}  \ottnt{U}$ and 
  $S'  \ottsym{(}  \Gamma'  \ottsym{)}  \vdash  S'  \ottsym{(}  e'  \ottsym{)}  \rightsquigarrow  \ottnt{f'}  \ottsym{:}  \ottnt{U'}$ and
  $ \langle  S  \ottsym{(}  \Gamma  \ottsym{)}   \vdash   \ottnt{f}  :  \ottnt{U}   \sqsubseteq _{ S'' }  \ottnt{U'}  :  \ottnt{f'}  \dashv  S'  \ottsym{(}  \Gamma'  \ottsym{)}  \rangle $.
\end{lemmaA}

\begin{proof}
  By soundness of $PT$, $S  \ottsym{(}  \Gamma  \ottsym{)}  \vdash  S  \ottsym{(}  e  \ottsym{)}  \ottsym{:}  \ottnt{U}$.
  By Lemma~\ref{thm:uprec-prec}, there exist $S'$, $\ottnt{U'}$ such that
  $S'  \ottsym{(}  \Gamma'  \ottsym{)}  \vdash  S'  \ottsym{(}  e'  \ottsym{)}  \ottsym{:}  \ottnt{U'}$ and
  $$ \langle  S  \ottsym{(}  \Gamma  \ottsym{)}   \vdash   S  \ottsym{(}  e  \ottsym{)}  :  \ottnt{U} 
                 \sqsubseteq _{ [  ] }  \ottnt{U'}  :  S'  \ottsym{(}  e'  \ottsym{)}  \dashv  S'  \ottsym{(}  \Gamma'  \ottsym{)}  \rangle .$$
  Since $S'  \ottsym{(}  \Gamma'  \ottsym{)}  \vdash  S'  \ottsym{(}  e'  \ottsym{)}  \ottsym{:}  \ottnt{U'}$, 
  by Lemma~\ref{thm:uprec-PT-prec}, there exist $S''$, $S'''$, and $\ottnt{U''}$
  such that
  $PT(\Gamma', e') = (S'', \ottnt{U''})$ and
  $S''  \ottsym{(}  \Gamma'  \ottsym{)}  \vdash  S''  \ottsym{(}  e'  \ottsym{)}  \ottsym{:}  \ottnt{U''}$ and
  $$ \langle  S'  \ottsym{(}  \Gamma'  \ottsym{)}   \vdash   S'  \ottsym{(}  e'  \ottsym{)}  :  \ottnt{U'} 
                 \sqsubseteq _{ S''' }  \ottnt{U''}  :  S''  \ottsym{(}  e'  \ottsym{)}  \dashv  S''  \ottsym{(}  \Gamma'  \ottsym{)}  \rangle .$$
  By Lemma~\ref{lem:prec-trans2},
  $ \langle  S  \ottsym{(}  \Gamma  \ottsym{)}   \vdash   S  \ottsym{(}  e  \ottsym{)}  :  \ottnt{U} 
                 \sqsubseteq _{ S''' }  \ottnt{U''}  :  S''  \ottsym{(}  e'  \ottsym{)}  \dashv  S''  \ottsym{(}  \Gamma'  \ottsym{)}  \rangle $.
  Finally, by Lemma~\ref{thm:cast-insertion-preserves-precision},
  there exist $\ottnt{f}$ and $\ottnt{f'}$ such that
  $S  \ottsym{(}  \Gamma  \ottsym{)}  \vdash  S  \ottsym{(}  e  \ottsym{)}  \rightsquigarrow  \ottnt{f}  \ottsym{:}  \ottnt{U}$ and
  $S''  \ottsym{(}  \Gamma'  \ottsym{)}  \vdash  S''  \ottsym{(}  e'  \ottsym{)}  \rightsquigarrow  \ottnt{f'}  \ottsym{:}  \ottnt{U'}$ and
  $ \langle  S  \ottsym{(}  \Gamma  \ottsym{)}   \vdash   \ottnt{f}  :  \ottnt{U}   \sqsubseteq _{ S''' }  \ottnt{U'}  :  \ottnt{f'}  \dashv  S''  \ottsym{(}  \Gamma'  \ottsym{)}  \rangle $. \qedhere
\end{proof}



\ifrestate
\thmStaticGG*
\else
\begin{theoremA}[name=Static Gradual Guarantee,restate=thmStaticGG]
 If $ e   \sqsubseteq _{ S_{{\mathrm{0}}} }  e' $ and $PT( \emptyset , e) = (S_{{\mathrm{1}}},\ottnt{U})$,
 then $PT( \emptyset , e') = (S'_{{\mathrm{1}}}, \ottnt{U'})$ and $ \ottnt{U}   \sqsubseteq _{ S_{{\mathrm{2}}} }  \ottnt{U'} $
 for some $S_{{\mathrm{2}}}$, $S'_{{\mathrm{1}}}$, and $\ottnt{U'}$.
\end{theoremA}
\fi
\begin{proof}
  Follows from Lemma~\ref{thm:ITGL-staticGG} as a special case.  (It is easy to see
  that $ \langle  S  \ottsym{(}  \Gamma  \ottsym{)}   \vdash   \ottnt{f}  :  \ottnt{U}   \sqsubseteq _{ S'' }  \ottnt{U'}  :  \ottnt{f'}  \dashv  S'  \ottsym{(}  \Gamma'  \ottsym{)}  \rangle $ implies
  $ \ottnt{U}   \sqsubseteq _{ S'' }  \ottnt{U'} $.)
\end{proof}

\subsubsection{Dynamic Gradual Guarantee in $\lambdaRTI$}

\begin{lemmaA} \label{lem:term_prec_to_type_prec}
  If $ \langle  \Gamma   \vdash   f  :  \ottnt{U}   \sqsubseteq _{ S }  \ottnt{U'}  :  f'  \dashv  \Gamma'  \rangle $,
  then $ \ottnt{U}   \sqsubseteq _{ S }  \ottnt{U'} $.
\end{lemmaA}
\begin{proof}
  By induction on the term precision derivation.
\end{proof}

\begin{lemmaA} \label{lem:term_prec_to_typing}
  If $ \langle  \Gamma   \vdash   f  :  \ottnt{U}   \sqsubseteq _{ S }  \ottnt{U'}  :  f'  \dashv  \Gamma'  \rangle $,
  then $\Gamma  \vdash  f  \ottsym{:}  \ottnt{U}$ and $\Gamma'  \vdash  f'  \ottsym{:}  \ottnt{U'}$.
\end{lemmaA}
\begin{proof}
  By induction on the term precision derivation.
\end{proof}

\begin{lemmaA} \label{lem:type_prec_right_var_equal}
  If $ \ottnt{U_{{\mathrm{1}}}}   \sqsubseteq _{ S }  \ottmv{X} $ and $ \ottnt{U_{{\mathrm{2}}}}   \sqsubseteq _{ S }  \ottmv{X} $,
  then $\ottnt{U_{{\mathrm{1}}}}  \ottsym{=}  \ottnt{U_{{\mathrm{2}}}}$.
\end{lemmaA}
\begin{proof}
  By case analysis on the type precision rule applied last.
\end{proof}

\begin{lemmaA} \label{lem:type_prec_consistent_ground2}
 If $ \ottnt{U_{{\mathrm{1}}}}   \sqsubseteq _{ S }  \ottnt{G_{{\mathrm{1}}}} $ and $\ottnt{U_{{\mathrm{1}}}}  \sim  \ottnt{U_{{\mathrm{2}}}}$ and $ \ottnt{U_{{\mathrm{2}}}}   \sqsubseteq _{ S }  \ottnt{G_{{\mathrm{2}}}} $,
 then $\ottnt{G_{{\mathrm{1}}}}  \ottsym{=}  \ottnt{G_{{\mathrm{2}}}}$.
\end{lemmaA}
\begin{proof}
 If $\ottnt{G_{{\mathrm{1}}}}  \ottsym{=}  \iota$, then $\ottnt{U_{{\mathrm{1}}}}  \ottsym{=}  \iota$ because the rules that can be applied
 to derive $ \ottnt{U_{{\mathrm{1}}}}   \sqsubseteq _{ S }  \iota $ are only \rnp{P\_IdBase}.
 If $\ottnt{G_{{\mathrm{2}}}}  \ottsym{=}  \star  \!\rightarrow\!  \star$, then $\ottnt{U_{{\mathrm{2}}}}  \ottsym{=}  \ottnt{U_{{\mathrm{21}}}}  \!\rightarrow\!  \ottnt{U_{{\mathrm{22}}}}$ for some $\ottnt{U_{{\mathrm{21}}}}$ and $\ottnt{U_{{\mathrm{22}}}}$
 because the rules that can be applied to derive $ \ottnt{U_{{\mathrm{2}}}}   \sqsubseteq _{ S }  \star  \!\rightarrow\!  \star $
 are only \rnp{P\_Arrow}.
 However, $\iota  \sim  \ottnt{U_{{\mathrm{21}}}}  \!\rightarrow\!  \ottnt{U_{{\mathrm{22}}}}$ is contradictory.
 We can also prove the case of $\ottnt{G_{{\mathrm{1}}}}  \ottsym{=}  \star  \!\rightarrow\!  \star$ similarly.
\end{proof}

\begin{lemmaA} \label{lem:type_prec_cast_value}
 If $ \ottnt{U_{{\mathrm{2}}}}   \sqsubseteq _{ S }  \ottnt{G} $ and $w  \ottsym{:}   \ottnt{U_{{\mathrm{1}}}} \Rightarrow  \unskip ^ { \ell }  \! \ottnt{U_{{\mathrm{2}}}} $ is a value,
 then $ \ottnt{U_{{\mathrm{1}}}}   \sqsubseteq _{ S }  \ottnt{G} $.
\end{lemmaA}
\begin{proof}
 Since $w  \ottsym{:}   \ottnt{U_{{\mathrm{1}}}} \Rightarrow  \unskip ^ { \ell }  \! \ottnt{U_{{\mathrm{2}}}} $ is a value,
 we have two cases on $\ottnt{U_{{\mathrm{2}}}}$ to be considered.
 The case $\ottnt{U_{{\mathrm{2}}}}  \ottsym{=}  \star$ leads to a contradiction because $ \ottnt{U_{{\mathrm{2}}}}   \sqsubseteq _{ S }  \ottnt{G} $.
 Otherwise, $\ottnt{U_{{\mathrm{2}}}}  \ottsym{=}  \ottnt{U_{{\mathrm{21}}}}  \!\rightarrow\!  \ottnt{U_{{\mathrm{22}}}}$ for some $\ottnt{U_{{\mathrm{21}}}}$ and $\ottnt{U_{{\mathrm{22}}}}$.
 Then, $\ottnt{U_{{\mathrm{1}}}}  \ottsym{=}  \ottnt{U_{{\mathrm{11}}}}  \!\rightarrow\!  \ottnt{U_{{\mathrm{12}}}}$ for some $\ottnt{U_{{\mathrm{11}}}}$ and $\ottnt{U_{{\mathrm{12}}}}$, and
 $\ottnt{G}  \ottsym{=}  \star  \!\rightarrow\!  \star$.
 Thus, \rnp{P\_Arrow}, $ \ottnt{U_{{\mathrm{1}}}}   \sqsubseteq _{ S }  \ottnt{G} $.
\end{proof}

\begin{lemmaA} \label{lem:prec_ground_contra}
 If $ \ottnt{U}   \sqsubseteq _{ S }  \ottnt{G'_{{\mathrm{1}}}} $ and 
 $ \langle   \emptyset    \vdash   w  :  \ottnt{U}   \sqsubseteq _{ S }  \star  :  \ottsym{(}  w'  \ottsym{:}   \ottnt{G'_{{\mathrm{2}}}} \Rightarrow  \unskip ^ { \ell' }  \! \star   \ottsym{)}  \dashv   \emptyset   \rangle $,
 then $\ottnt{G'_{{\mathrm{1}}}}  \ottsym{=}  \ottnt{G'_{{\mathrm{2}}}}$.
\end{lemmaA}
\begin{proof}
 By induction on the term precision derivation.
 There are three interesting cases.
 \begin{caseanalysis}
  \case{\rnp{P\_Cast}}
  We are given
  $ \langle   \emptyset    \vdash   \ottsym{(}  w_{{\mathrm{1}}}  \ottsym{:}   \ottnt{U_{{\mathrm{1}}}} \Rightarrow  \unskip ^ { \ell }  \! \ottnt{U}   \ottsym{)}  :  \ottnt{U}   \sqsubseteq _{ S }  \star  :  \ottsym{(}  w'  \ottsym{:}   \ottnt{G'_{{\mathrm{2}}}} \Rightarrow  \unskip ^ { \ell' }  \! \star   \ottsym{)}  \dashv   \emptyset   \rangle $ and,
  by inversion,
  $ \langle   \emptyset    \vdash   w_{{\mathrm{1}}}  :  \ottnt{U_{{\mathrm{1}}}}   \sqsubseteq _{ S }  \ottnt{G'_{{\mathrm{2}}}}  :  w'  \dashv   \emptyset   \rangle $ and $\ottnt{U_{{\mathrm{1}}}}  \sim  \ottnt{U}$.
  By Lemma~\ref{lem:term_prec_to_type_prec}, $ \ottnt{U_{{\mathrm{1}}}}   \sqsubseteq _{ S }  \ottnt{G'_{{\mathrm{2}}}} $.
  Since $ \ottnt{U}   \sqsubseteq _{ S }  \ottnt{G'_{{\mathrm{1}}}} $, we have $\ottnt{G'_{{\mathrm{1}}}}  \ottsym{=}  \ottnt{G'_{{\mathrm{2}}}}$
  by Lemma~\ref{lem:type_prec_consistent_ground2}.

  \case{\rnp{P\_CastL}}
  We are given
  $ \langle   \emptyset    \vdash   \ottsym{(}  w_{{\mathrm{1}}}  \ottsym{:}   \ottnt{U_{{\mathrm{1}}}} \Rightarrow  \unskip ^ { \ell }  \! \ottnt{U}   \ottsym{)}  :  \ottnt{U}   \sqsubseteq _{ S }  \star  :  \ottsym{(}  w'  \ottsym{:}   \ottnt{G'_{{\mathrm{2}}}} \Rightarrow  \unskip ^ { \ell' }  \! \star   \ottsym{)}  \dashv   \emptyset   \rangle $ and,
  by inversion,
  $ \langle   \emptyset    \vdash   w_{{\mathrm{1}}}  :  \ottnt{U_{{\mathrm{1}}}}   \sqsubseteq _{ S }  \star  :  \ottsym{(}  w'  \ottsym{:}   \ottnt{G'_{{\mathrm{2}}}} \Rightarrow  \unskip ^ { \ell' }  \! \star   \ottsym{)}  \dashv   \emptyset   \rangle $ and
  $\ottnt{U_{{\mathrm{1}}}}  \sim  \ottnt{U}$.
  Since $w_{{\mathrm{1}}}  \ottsym{:}   \ottnt{U_{{\mathrm{1}}}} \Rightarrow  \unskip ^ { \ell }  \! \ottnt{U} $ is a value,
  we have $ \ottnt{U_{{\mathrm{1}}}}   \sqsubseteq _{ S }  \ottnt{G'_{{\mathrm{1}}}} $ by Lemma~\ref{lem:type_prec_cast_value}.
  Thus, we finish by the IH.

  \case{\rnp{P\_CastR}}
  We are given
  $ \langle   \emptyset    \vdash   w_{{\mathrm{1}}}  :  \ottnt{U}   \sqsubseteq _{ S }  \star  :  \ottsym{(}  w'  \ottsym{:}   \ottnt{G'_{{\mathrm{2}}}} \Rightarrow  \unskip ^ { \ell' }  \! \star   \ottsym{)}  \dashv   \emptyset   \rangle $ and,
  by inversion,
  $ \langle   \emptyset    \vdash   w_{{\mathrm{1}}}  :  \ottnt{U}   \sqsubseteq _{ S }  \ottnt{G'_{{\mathrm{2}}}}  :  w'  \dashv   \emptyset   \rangle $.
  By Lemma~\ref{lem:term_prec_to_type_prec},
  $ \ottnt{U}   \sqsubseteq _{ S }  \ottnt{G'_{{\mathrm{2}}}} $.
  Since $\ottnt{U}  \sim  \ottnt{U}$ and $ \ottnt{U}   \sqsubseteq _{ S }  \ottnt{G'_{{\mathrm{1}}}} $,
  we have $\ottnt{G'_{{\mathrm{1}}}}  \ottsym{=}  \ottnt{G'_{{\mathrm{2}}}}$ by Lemma~\ref{lem:type_prec_consistent_ground2}.
  \qedhere
 \end{caseanalysis}
\end{proof}


\begin{lemmaA} \label{lem:left_subst_preserve_prec}
  For any $S$.
  \begin{enumerate}
    \item
      If $ \ottnt{U}   \sqsubseteq _{ S_{{\mathrm{0}}} }  \ottnt{U'} $,
      then $ S  \ottsym{(}  \ottnt{U}  \ottsym{)}   \sqsubseteq _{  S  \circ  S_{{\mathrm{0}}}  }  \ottnt{U'} $.
    \item
      If $ \langle  \Gamma   \vdash   f  :  \ottnt{U}   \sqsubseteq _{ S_{{\mathrm{0}}} }  \ottnt{U'}  :  f'  \dashv  \Gamma'  \rangle $,
      then $ \langle  S  \ottsym{(}  \Gamma  \ottsym{)}   \vdash   S  \ottsym{(}  f  \ottsym{)}  :  S  \ottsym{(}  \ottnt{U}  \ottsym{)}   \sqsubseteq _{  S  \circ  S_{{\mathrm{0}}}  }  \ottnt{U'}  :  f'  \dashv  \Gamma'  \rangle $.
  \end{enumerate}
\end{lemmaA}

\begin{proof}
  \leavevmode
  \begin{enumerate}
    \item By induction on the derivation of $ \ottnt{U}   \sqsubseteq _{ S_{{\mathrm{0}}} }  \ottnt{U'} $.
      \begin{caseanalysis}
        \case{\rnp{P\_TyVar}}
        We are given $ S_{{\mathrm{0}}}  \ottsym{(}  \ottmv{X}  \ottsym{)}   \sqsubseteq _{ S_{{\mathrm{0}}} }  \ottmv{X} $ for some $\ottmv{X}$.
        By definition, $ S  \circ  S_{{\mathrm{0}}}   \ottsym{(}  \ottmv{X}  \ottsym{)} = S  \ottsym{(}  S_{{\mathrm{0}}}  \ottsym{(}  \ottmv{X}  \ottsym{)}  \ottsym{)}$.
        By \rnp{P\_TyVar}, $ S  \ottsym{(}  S_{{\mathrm{0}}}  \ottsym{(}  \ottmv{X}  \ottsym{)}  \ottsym{)}   \sqsubseteq _{  S  \circ  S_{{\mathrm{0}}}  }  \ottmv{X} $.

        \otherwise Obvious.
      \end{caseanalysis}

    \item By induction on the derivation of $ \langle  \Gamma   \vdash   f  :  \ottnt{U}   \sqsubseteq _{ S_{{\mathrm{0}}} }  \ottnt{U'}  :  f'  \dashv  \Gamma'  \rangle $ using Lemma \ref{lem:left_subst_preserve_prec}.1.
      \qedhere
  \end{enumerate}
\end{proof}

\begin{lemmaA} \label{lem:prec_subst_swap}
 Suppose $\forall \ottmv{X} \, \in \, \textit{ftv} \, \ottsym{(}  \ottnt{U'}  \ottsym{)}. S  \ottsym{(}  \ottmv{X}  \ottsym{)}  \ottsym{=}  S'  \ottsym{(}  \ottmv{X}  \ottsym{)}$ and
 $\textit{ftv} \, \ottsym{(}  \ottnt{U'}  \ottsym{)}  \subseteq  \textit{dom} \, \ottsym{(}  S'  \ottsym{)}$.
 If $ \ottnt{U}   \sqsubseteq _{ S }  \ottnt{U'} $, then $ \ottnt{U}   \sqsubseteq _{ S' }  \ottnt{U'} $.
\end{lemmaA}
\begin{proof}
 By induction on the derivation of $ \ottnt{U}   \sqsubseteq _{ S }  \ottnt{U'} $.
 \begin{caseanalysis}
  \case{\rnp{P\_IdBase}} Obvious.
  \case{\rnp{P\_TyVar}}
   We are given $ S  \ottsym{(}  \ottmv{X}  \ottsym{)}   \sqsubseteq _{ S }  \ottmv{X} $.
   Since $S  \ottsym{(}  \ottmv{X}  \ottsym{)}  \ottsym{=}  S'  \ottsym{(}  \ottmv{X}  \ottsym{)}$ and $\ottmv{X} \, \in \, \textit{dom} \, \ottsym{(}  S'  \ottsym{)}$, we have
   $ S  \ottsym{(}  \ottmv{X}  \ottsym{)}   \sqsubseteq _{ S' }  \ottmv{X} $ by \rnp{P\_TyVar}.
  \case{\rnp{P\_Dyn}} Obvious.
  \case{\rnp{P\_Arrow}} By the IHs and \rnp{P\_Arrow}.
  \qedhere
 \end{caseanalysis}
\end{proof}

\begin{lemmaA} \label{lem:type_prec_ident}
 For any $S$ and $\ottnt{U}$ such that $\textit{ftv} \, \ottsym{(}  \ottnt{U}  \ottsym{)}  \subseteq  \textit{dom} \, \ottsym{(}  S  \ottsym{)}$,
 $ S  \ottsym{(}  \ottnt{U}  \ottsym{)}   \sqsubseteq _{ S }  \ottnt{U} $.
\end{lemmaA}
\begin{proof}
 By induction on $\ottnt{U}$.
 \begin{caseanalysis}
  \case{$\ottnt{U}  \ottsym{=}  \iota$ for some $\iota$}
   Since $S  \ottsym{(}  \ottnt{U}  \ottsym{)}  \ottsym{=}  \iota$, we finish by \rnp{P\_IdBase}.
  \case{$\ottnt{U}  \ottsym{=}  \ottnt{U_{{\mathrm{1}}}}  \!\rightarrow\!  \ottnt{U_{{\mathrm{2}}}}$ for some $\ottnt{U_{{\mathrm{1}}}}$ and $\ottnt{U_{{\mathrm{2}}}}$}
   We have $S  \ottsym{(}  \ottnt{U}  \ottsym{)}  \ottsym{=}  S  \ottsym{(}  \ottnt{U_{{\mathrm{1}}}}  \ottsym{)}  \!\rightarrow\!  S  \ottsym{(}  \ottnt{U_{{\mathrm{2}}}}  \ottsym{)}$.
   By the IH,
   $ S  \ottsym{(}  \ottnt{U_{{\mathrm{1}}}}  \ottsym{)}   \sqsubseteq _{ S }  \ottnt{U_{{\mathrm{1}}}} $ and
   $ S  \ottsym{(}  \ottnt{U_{{\mathrm{2}}}}  \ottsym{)}   \sqsubseteq _{ S }  \ottnt{U_{{\mathrm{2}}}} $.
   Thus, by \rnp{P\_Arrow},
   $ S  \ottsym{(}  \ottnt{U_{{\mathrm{1}}}}  \!\rightarrow\!  \ottnt{U_{{\mathrm{2}}}}  \ottsym{)}   \sqsubseteq _{ S }  \ottnt{U_{{\mathrm{1}}}}  \!\rightarrow\!  \ottnt{U_{{\mathrm{2}}}} $.
  \case{$\ottnt{U}  \ottsym{=}  \ottmv{X}$ for some $\ottmv{X}$}
   By \rnp{P\_TyVar}.
   \case{$\ottnt{U}  \ottsym{=}  \star$} By \rnp{P\_Dyn}.
   \qedhere
 \end{caseanalysis}
\end{proof}

\begin{lemmaA} \label{lem:term_prec_push_subst} \noindent
 Suppose that
 $\forall \ottmv{X} \, \in \, \textit{dom} \, \ottsym{(}  S_{{\mathrm{2}}}  \ottsym{)}. \textit{ftv} \, \ottsym{(}  S_{{\mathrm{2}}}  \ottsym{(}  \ottmv{X}  \ottsym{)}  \ottsym{)}  \subseteq  \textit{dom} \, \ottsym{(}  S_{{\mathrm{1}}}  \ottsym{)}$.
 \begin{enumerate}
  \item If $ \ottnt{U}   \sqsubseteq _{  S_{{\mathrm{1}}}  \circ  S_{{\mathrm{2}}}  }  \ottnt{U'} $,
        then $ \ottnt{U}   \sqsubseteq _{ S_{{\mathrm{1}}} }  S_{{\mathrm{2}}}  \ottsym{(}  \ottnt{U'}  \ottsym{)} $.
  \item If $ \langle  \Gamma   \vdash   \ottnt{f}  :  \ottnt{U}   \sqsubseteq _{  S_{{\mathrm{1}}}  \circ  S_{{\mathrm{2}}}  }  \ottnt{U'}  :  \ottnt{f'}  \dashv  \Gamma'  \rangle $,
        then $ \langle  \Gamma   \vdash   \ottnt{f}  :  \ottnt{U}   \sqsubseteq _{ S_{{\mathrm{1}}} }  S_{{\mathrm{2}}}  \ottsym{(}  \ottnt{U'}  \ottsym{)}  :  S_{{\mathrm{2}}}  \ottsym{(}  \ottnt{f'}  \ottsym{)}  \dashv  S_{{\mathrm{2}}}  \ottsym{(}  \Gamma'  \ottsym{)}  \rangle $.
 \end{enumerate}
\end{lemmaA}
\begin{proof}
 \leavevmode
 \begin{enumerate}
  \item By induction on the derivation of $ \ottnt{U}   \sqsubseteq _{  S_{{\mathrm{1}}}  \circ  S_{{\mathrm{2}}}  }  \ottnt{U'} $.
        The only interesting case is \rnp{P\_TyVar}.
        In that case, we are given $  S_{{\mathrm{1}}}  \circ  S_{{\mathrm{2}}}   \ottsym{(}  \ottmv{X}  \ottsym{)}   \sqsubseteq _{  S_{{\mathrm{1}}}  \circ  S_{{\mathrm{2}}}  }  \ottmv{X} $ and,
        by inversion, $\ottmv{X} \, \in \, \textit{dom} \, \ottsym{(}   S_{{\mathrm{1}}}  \circ  S_{{\mathrm{2}}}   \ottsym{)}$.
        If $\ottmv{X} \, \not\in \, \textit{dom} \, \ottsym{(}  S_{{\mathrm{2}}}  \ottsym{)}$, then $S_{{\mathrm{2}}}  \ottsym{(}  \ottmv{X}  \ottsym{)}  \ottsym{=}  \ottmv{X}$ and it suffices to show that
        $ S_{{\mathrm{1}}}  \ottsym{(}  \ottmv{X}  \ottsym{)}   \sqsubseteq _{ S_{{\mathrm{1}}} }  \ottmv{X} $, which is proved by \rnp{P\_TyVar}
        since $\ottmv{X} \, \in \, \textit{dom} \, \ottsym{(}  S_{{\mathrm{1}}}  \ottsym{)}$.
        Otherwise, if $\ottmv{X} \, \in \, \textit{dom} \, \ottsym{(}  S_{{\mathrm{2}}}  \ottsym{)}$, then
        $\textit{ftv} \, \ottsym{(}  S_{{\mathrm{2}}}  \ottsym{(}  \ottmv{X}  \ottsym{)}  \ottsym{)}  \subseteq  \textit{dom} \, \ottsym{(}  S_{{\mathrm{1}}}  \ottsym{)}$.
        Thus, we have $ S_{{\mathrm{1}}}  \ottsym{(}  S_{{\mathrm{2}}}  \ottsym{(}  \ottmv{X}  \ottsym{)}  \ottsym{)}   \sqsubseteq _{ S_{{\mathrm{1}}} }  S_{{\mathrm{2}}}  \ottsym{(}  \ottmv{X}  \ottsym{)} $
        by Lemma~\ref{lem:type_prec_ident}.

  \item By induction on the derivation of
    $ \langle  \Gamma   \vdash   \ottnt{f}  :  \ottnt{U}   \sqsubseteq _{  S_{{\mathrm{1}}}  \circ  S_{{\mathrm{2}}}  }  \ottnt{U'}  :  \ottnt{f'}  \dashv  \Gamma'  \rangle $ with the first case.
    \qedhere
 \end{enumerate}
\end{proof}

\begin{lemmaA} \label{lem:right_subst_preserve_prec}
  Suppose $\forall \ottmv{X} \in \textit{dom} \, \ottsym{(}  S_{{\mathrm{0}}}  \ottsym{)}. S_{{\mathrm{0}}}  \ottsym{(}  \ottmv{X}  \ottsym{)}  \ottsym{=}   S'_{{\mathrm{0}}}  \circ   S_{{\mathrm{0}}}  \circ  S'    \ottsym{(}  \ottmv{X}  \ottsym{)}$ and
  $\forall \ottmv{X} \, \in \, \textit{dom} \, \ottsym{(}  S'  \ottsym{)}. \textit{ftv} \, \ottsym{(}  S'  \ottsym{(}  \ottmv{X}  \ottsym{)}  \ottsym{)}  \subseteq  \textit{dom} \, \ottsym{(}  S'_{{\mathrm{0}}}  \ottsym{)}$.
  \begin{enumerate}
    \item
      If $ \ottnt{U}   \sqsubseteq _{ S_{{\mathrm{0}}} }  \ottnt{U'} $,
      then $ \ottnt{U}   \sqsubseteq _{  S'_{{\mathrm{0}}}  \circ  S_{{\mathrm{0}}}  }  S'  \ottsym{(}  \ottnt{U'}  \ottsym{)} $.
    \item
      If $ \langle  \Gamma   \vdash   f  :  \ottnt{U}   \sqsubseteq _{ S_{{\mathrm{0}}} }  \ottnt{U'}  :  f'  \dashv  \Gamma'  \rangle $,
      then $ \langle  \Gamma   \vdash   f  :  \ottnt{U}   \sqsubseteq _{  S'_{{\mathrm{0}}}  \circ  S_{{\mathrm{0}}}  }  S'  \ottsym{(}  \ottnt{U'}  \ottsym{)}  :  S'  \ottsym{(}  f'  \ottsym{)}  \dashv  S'  \ottsym{(}  \Gamma'  \ottsym{)}  \rangle $.
  \end{enumerate}
\end{lemmaA}

\begin{proof}
  \leavevmode
  \begin{enumerate}
    \item By induction on the derivation of $ \ottnt{U}   \sqsubseteq _{ S_{{\mathrm{0}}} }  \ottnt{U'} $.
      \begin{caseanalysis}
        \case{\rnp{P\_TyVar}}
        We are given $ S_{{\mathrm{0}}}  \ottsym{(}  \ottmv{X}  \ottsym{)}   \sqsubseteq _{ S_{{\mathrm{0}}} }  \ottmv{X} $ for some $\ottmv{X}$.
        By inversion, $\ottmv{X} \, \in \, \textit{dom} \, \ottsym{(}  S_{{\mathrm{0}}}  \ottsym{)}$.
        So, $  S'_{{\mathrm{0}}}  \circ  S_{{\mathrm{0}}}   \circ  S'   \ottsym{(}  \ottmv{X}  \ottsym{)} = S_{{\mathrm{0}}}  \ottsym{(}  \ottmv{X}  \ottsym{)}$.
        Thus, $ S_{{\mathrm{0}}}  \ottsym{(}  \ottmv{X}  \ottsym{)}   \sqsubseteq _{   S'_{{\mathrm{0}}}  \circ  S_{{\mathrm{0}}}   \circ  S'  }  \ottmv{X} $ by \rnp{P\_TyVar}.
        By Lemma~\ref{lem:term_prec_push_subst},
        $ S_{{\mathrm{0}}}  \ottsym{(}  \ottmv{X}  \ottsym{)}   \sqsubseteq _{  S'_{{\mathrm{0}}}  \circ  S_{{\mathrm{0}}}  }  S'  \ottsym{(}  \ottmv{X}  \ottsym{)} $.

        \otherwise Obvious.
      \end{caseanalysis}

    \item By induction on the derivation of $ \langle  \Gamma   \vdash   f  :  \ottnt{U}   \sqsubseteq _{ S_{{\mathrm{0}}} }  \ottnt{U'}  :  f'  \dashv  \Gamma'  \rangle $
      \qedhere
  \end{enumerate}
\end{proof}


\begin{lemmaA} \label{lem:term_prec_inversion1}
  If $ \langle   \emptyset    \vdash   w  :  \ottnt{U}   \sqsubseteq _{ S_{{\mathrm{0}}} }  \star  :  \ottsym{(}  w'  \ottsym{:}   \ottnt{G'} \Rightarrow  \unskip ^ { \ell' }  \! \star   \ottsym{)}  \dashv   \emptyset   \rangle $
  and $ \ottnt{U}   \sqsubseteq _{ S_{{\mathrm{0}}} }  \ottnt{G'} $,
  then $ \langle   \emptyset    \vdash   w  :  \ottnt{U}   \sqsubseteq _{ S_{{\mathrm{0}}} }  \ottnt{G'}  :  w'  \dashv   \emptyset   \rangle $,
\end{lemmaA}

\begin{proof}
 By induction on the term precision derivation.
 There are three interesting cases.
 \begin{caseanalysis}
  \case{\rnp{P\_Cast}}
  We are given $ \langle   \emptyset    \vdash   \ottsym{(}  w_{{\mathrm{1}}}  \ottsym{:}   \ottnt{U_{{\mathrm{1}}}} \Rightarrow  \unskip ^ { \ell }  \! \ottnt{U}   \ottsym{)}  :  \ottnt{U}   \sqsubseteq _{ S_{{\mathrm{0}}} }  \star  :  \ottsym{(}  w'  \ottsym{:}   \ottnt{G'} \Rightarrow  \unskip ^ { \ell' }  \! \star   \ottsym{)}  \dashv   \emptyset   \rangle $
  and, by inversion,
  $ \langle   \emptyset    \vdash   w_{{\mathrm{1}}}  :  \ottnt{U_{{\mathrm{1}}}}   \sqsubseteq _{ S_{{\mathrm{0}}} }  \ottnt{G'}  :  w'  \dashv   \emptyset   \rangle $ and
  $\ottnt{U_{{\mathrm{1}}}}  \sim  \ottnt{U}$.
  By \rnp{P\_CastL}, we finish.

  \case{\rnp{P\_CastL}}
  We are given
  $ \langle   \emptyset    \vdash   \ottsym{(}  w_{{\mathrm{1}}}  \ottsym{:}   \ottnt{U_{{\mathrm{1}}}} \Rightarrow  \unskip ^ { \ell }  \! \ottnt{U}   \ottsym{)}  :  \ottnt{U}   \sqsubseteq _{ S_{{\mathrm{0}}} }  \star  :  \ottsym{(}  w'  \ottsym{:}   \ottnt{G'} \Rightarrow  \unskip ^ { \ell' }  \! \star   \ottsym{)}  \dashv   \emptyset   \rangle $
  and, by inversion,
  $ \langle   \emptyset    \vdash   w_{{\mathrm{1}}}  :  \ottnt{U_{{\mathrm{1}}}}   \sqsubseteq _{ S_{{\mathrm{0}}} }  \star  :  \ottsym{(}  w'  \ottsym{:}   \ottnt{G'} \Rightarrow  \unskip ^ { \ell' }  \! \star   \ottsym{)}  \dashv   \emptyset   \rangle $ and
  $\ottnt{U_{{\mathrm{1}}}}  \sim  \ottnt{U}$.
  Since $w = w_{{\mathrm{1}}}  \ottsym{:}   \ottnt{U_{{\mathrm{1}}}} \Rightarrow  \unskip ^ { \ell }  \! \ottnt{U} $ is a value,
  we have $ \ottnt{U_{{\mathrm{1}}}}   \sqsubseteq _{ S_{{\mathrm{0}}} }  \ottnt{G'} $ by Lemma~\ref{lem:type_prec_cast_value}.
  By the IH, $ \langle   \emptyset    \vdash   w_{{\mathrm{1}}}  :  \ottnt{U_{{\mathrm{1}}}}   \sqsubseteq _{ S_{{\mathrm{0}}} }  \ottnt{G'}  :  w'  \dashv   \emptyset   \rangle $.
  Since $\ottnt{U_{{\mathrm{1}}}}  \sim  \ottnt{U}$ and $ \ottnt{U}   \sqsubseteq _{ S_{{\mathrm{0}}} }  \ottnt{G'} $,
  we finish by \rnp{P\_CastL}.

  \case{\rnp{P\_CastR}} By inversion.
  \qedhere
 \end{caseanalysis}
\end{proof}

\begin{lemmaA} \label{lem:term_prec_inversion3}
  If $ \langle   \emptyset    \vdash   \ottsym{(}  f  \ottsym{:}   \ottnt{U} \Rightarrow  \unskip ^ { \ell }  \! \star   \ottsym{)}  :  \star   \sqsubseteq _{ S_{{\mathrm{0}}} }  \ottnt{U'}  :  f'  \dashv   \emptyset   \rangle $
  then $ \langle   \emptyset    \vdash   f  :  \ottnt{U}   \sqsubseteq _{ S_{{\mathrm{0}}} }  \ottnt{U'}  :  f'  \dashv   \emptyset   \rangle $,
\end{lemmaA}

\begin{proof}
 By induction on the term precision derivation.
 There are three interesting cases.
 \begin{caseanalysis}
  \case{\rnp{P\_Cast}}
  We are given
  $ \langle   \emptyset    \vdash   \ottsym{(}  f  \ottsym{:}   \ottnt{U} \Rightarrow  \unskip ^ { \ell }  \! \star   \ottsym{)}  :  \star   \sqsubseteq _{ S_{{\mathrm{0}}} }  \ottnt{U'}  :  \ottsym{(}  f'_{{\mathrm{1}}}  \ottsym{:}   \ottnt{U'_{{\mathrm{1}}}} \Rightarrow  \unskip ^ { \ell' }  \! \ottnt{U'}   \ottsym{)}  \dashv   \emptyset   \rangle $
  and, by inversion,
  $ \langle   \emptyset    \vdash   f  :  \ottnt{U}   \sqsubseteq _{ S_{{\mathrm{0}}} }  \ottnt{U'_{{\mathrm{1}}}}  :  f'_{{\mathrm{1}}}  \dashv   \emptyset   \rangle $ and
  $ \star   \sqsubseteq _{ S_{{\mathrm{0}}} }  \ottnt{U'} $ and
  $\ottnt{U'_{{\mathrm{1}}}}  \sim  \ottnt{U'}$.
  Since $ \star   \sqsubseteq _{ S_{{\mathrm{0}}} }  \ottnt{U'} $, we have $\ottnt{U'}  \ottsym{=}  \star$.
  Thus, $ \ottnt{U}   \sqsubseteq _{ S_{{\mathrm{0}}} }  \ottnt{U'} $, and by \rnp{P\_CastR},
  $ \langle   \emptyset    \vdash   f  :  \ottnt{U}   \sqsubseteq _{ S_{{\mathrm{0}}} }  \ottnt{U'}  :  \ottsym{(}  f'_{{\mathrm{1}}}  \ottsym{:}   \ottnt{U'_{{\mathrm{1}}}} \Rightarrow  \unskip ^ { \ell' }  \! \ottnt{U'}   \ottsym{)}  \dashv   \emptyset   \rangle $.
  
  \case{\rnp{P\_CastL}} By inversion.
  \case{\rnp{P\_CastR}} Similar to the case of \rnp{P\_Cast}.
  \qedhere
 \end{caseanalysis}
\end{proof}

\begin{lemmaA} \label{lem:term_prec_inversion4}
  If $ \langle   \emptyset    \vdash   w  :  \iota   \sqsubseteq _{ S_{{\mathrm{0}}} }  \iota  :  w'  \dashv   \emptyset   \rangle $,
  then there exists $\ottnt{c}$ such that $w  \ottsym{=}  w' = c$.
\end{lemmaA}

\begin{proof}
  By case analysis on the term precision rule applied last.
\end{proof}


\begin{lemmaA}[Catch up to Value on the Left] \label{lem:prec_catch_up_left_value}
  If $ \langle   \emptyset    \vdash   w  :  \ottnt{U}   \sqsubseteq _{ S_{{\mathrm{0}}} }  \ottnt{U'}  :  f'  \dashv   \emptyset   \rangle $,
  then there exists $S'$, $S'_{{\mathrm{0}}}$, and $w'$
  such that
  \begin{itemize}
   \item $f' \,  \xmapsto{ \mathmakebox[0.4em]{} S' \mathmakebox[0.3em]{} }\hspace{-0.4em}{}^\ast \hspace{0.2em}  \, w'$,
   \item $ \langle   \emptyset    \vdash   w  :  \ottnt{U}   \sqsubseteq _{  S'_{{\mathrm{0}}}  \circ  S_{{\mathrm{0}}}  }  S'  \ottsym{(}  \ottnt{U'}  \ottsym{)}  :  w'  \dashv   \emptyset   \rangle $,
   \item $\forall \ottmv{X} \in \textit{dom} \, \ottsym{(}  S_{{\mathrm{0}}}  \ottsym{)}. S_{{\mathrm{0}}}  \ottsym{(}  \ottmv{X}  \ottsym{)}  \ottsym{=}   S'_{{\mathrm{0}}}  \circ   S_{{\mathrm{0}}}  \circ  S'    \ottsym{(}  \ottmv{X}  \ottsym{)}$,
   \item $\forall \ottmv{X} \in \textit{dom} \, \ottsym{(}  S_{{\mathrm{0}}}  \ottsym{)}. \textit{dom} \, \ottsym{(}  S'_{{\mathrm{0}}}  \ottsym{)}  \cap  \textit{ftv} \, \ottsym{(}  S_{{\mathrm{0}}}  \ottsym{(}  \ottmv{X}  \ottsym{)}  \ottsym{)}  \ottsym{=}   \emptyset $, and
   \item $\forall \ottmv{X} \in \textit{dom} \, \ottsym{(}  S'  \ottsym{)}. \textit{ftv} \, \ottsym{(}  S'  \ottsym{(}  \ottmv{X}  \ottsym{)}  \ottsym{)}  \subseteq  \textit{dom} \, \ottsym{(}  S'_{{\mathrm{0}}}  \ottsym{)}$.
  \end{itemize}
\end{lemmaA}

\begin{proof}
  By induction on the term precision derivation.

  \begin{caseanalysis}
    \case{\rnp{P\_Const} and \rnp{P\_Abs}}
    Obvious.

    \case{\rnp{P\_Cast}}
    We are given,
    $ \langle   \emptyset    \vdash   \ottsym{(}  w_{{\mathrm{1}}}  \ottsym{:}   \ottnt{U_{{\mathrm{1}}}} \Rightarrow  \unskip ^ { \ell }  \! \ottnt{U}   \ottsym{)}  :  \ottnt{U}   \sqsubseteq _{ S_{{\mathrm{0}}} }  \ottnt{U'}  :  \ottsym{(}  f'_{{\mathrm{1}}}  \ottsym{:}   \ottnt{U'_{{\mathrm{1}}}} \Rightarrow  \unskip ^ { \ell' }  \! \ottnt{U'}   \ottsym{)}  \dashv   \emptyset   \rangle $,
    where $w  \ottsym{=}  w_{{\mathrm{1}}}  \ottsym{:}   \ottnt{U_{{\mathrm{1}}}} \Rightarrow  \unskip ^ { \ell }  \! \ottnt{U} $ and $f'  \ottsym{=}  f'_{{\mathrm{1}}}  \ottsym{:}   \ottnt{U'_{{\mathrm{1}}}} \Rightarrow  \unskip ^ { \ell' }  \! \ottnt{U'} $
    for some $w_{{\mathrm{1}}}$, $f'_{{\mathrm{1}}}$, $\ottnt{U_{{\mathrm{1}}}}$, and $\ottnt{U'_{{\mathrm{1}}}}$, $\ell$, and $\ell'$.
    By inversion,
    \begin{itemize}
     \item $ \langle   \emptyset    \vdash   w_{{\mathrm{1}}}  :  \ottnt{U_{{\mathrm{1}}}}   \sqsubseteq _{ S_{{\mathrm{0}}} }  \ottnt{U'_{{\mathrm{1}}}}  :  f'_{{\mathrm{1}}}  \dashv   \emptyset   \rangle $,
     \item $ \ottnt{U}   \sqsubseteq _{ S_{{\mathrm{0}}} }  \ottnt{U'} $, and
     \item $\ottnt{U_{{\mathrm{1}}}}  \sim  \ottnt{U}$ and $\ottnt{U'_{{\mathrm{1}}}}  \sim  \ottnt{U'}$.
    \end{itemize}
    By the IH, there exist $S'_{{\mathrm{0}}}$, $S'_{{\mathrm{1}}}$ and $w'_{{\mathrm{1}}}$ such that
    \begin{itemize}
     \item $f'_{{\mathrm{1}}} \,  \xmapsto{ \mathmakebox[0.4em]{} S'_{{\mathrm{1}}} \mathmakebox[0.3em]{} }\hspace{-0.4em}{}^\ast \hspace{0.2em}  \, w'_{{\mathrm{1}}}$,
     \item $ \langle   \emptyset    \vdash   w_{{\mathrm{1}}}  :  \ottnt{U_{{\mathrm{1}}}}   \sqsubseteq _{  S'_{{\mathrm{0}}}  \circ  S_{{\mathrm{0}}}  }  S'_{{\mathrm{1}}}  \ottsym{(}  \ottnt{U'_{{\mathrm{1}}}}  \ottsym{)}  :  w'_{{\mathrm{1}}}  \dashv   \emptyset   \rangle $,
     \item $\forall \ottmv{X} \in \textit{dom} \, \ottsym{(}  S_{{\mathrm{0}}}  \ottsym{)}. S_{{\mathrm{0}}}  \ottsym{(}  \ottmv{X}  \ottsym{)}  \ottsym{=}   S'_{{\mathrm{0}}}  \circ   S_{{\mathrm{0}}}  \circ  S'_{{\mathrm{1}}}    \ottsym{(}  \ottmv{X}  \ottsym{)}$,
     \item $\forall \ottmv{X} \in \textit{dom} \, \ottsym{(}  S_{{\mathrm{0}}}  \ottsym{)}. \textit{dom} \, \ottsym{(}  S'_{{\mathrm{0}}}  \ottsym{)}  \cap  \textit{ftv} \, \ottsym{(}  S_{{\mathrm{0}}}  \ottsym{(}  \ottmv{X}  \ottsym{)}  \ottsym{)}  \ottsym{=}   \emptyset $, and
     \item $\forall \ottmv{X} \in \textit{dom} \, \ottsym{(}  S'_{{\mathrm{1}}}  \ottsym{)}. \textit{ftv} \, \ottsym{(}  S'_{{\mathrm{1}}}  \ottsym{(}  \ottmv{X}  \ottsym{)}  \ottsym{)}  \subseteq  \textit{dom} \, \ottsym{(}  S'_{{\mathrm{0}}}  \ottsym{)}$.
    \end{itemize}
    Thus, by \rnp{E\_Step},
    \[
     f'_{{\mathrm{1}}}  \ottsym{:}   \ottnt{U'_{{\mathrm{1}}}} \Rightarrow  \unskip ^ { \ell' }  \! \ottnt{U'}  \,  \xmapsto{ \mathmakebox[0.4em]{} S'_{{\mathrm{1}}} \mathmakebox[0.3em]{} }\hspace{-0.4em}{}^\ast \hspace{0.2em}  \, w'_{{\mathrm{1}}}  \ottsym{:}   S'_{{\mathrm{1}}}  \ottsym{(}  \ottnt{U'_{{\mathrm{1}}}}  \ottsym{)} \Rightarrow  \unskip ^ { \ell' }  \! S'_{{\mathrm{1}}}  \ottsym{(}  \ottnt{U'}  \ottsym{)} .
    \]
    By case analysis on $\ottnt{U'_{{\mathrm{1}}}}  \sim  \ottnt{U'}$.

    \begin{caseanalysis}
      \case{$\iota  \sim  \iota$ for some $\iota$}
      Since $ \ottnt{U}   \sqsubseteq _{ S_{{\mathrm{0}}} }  \ottnt{U'} $ and $\ottnt{U'}  \ottsym{=}  \iota$,
      we have $\ottnt{U}  \ottsym{=}  \iota$.
      However, this contradicts the fact that $w_{{\mathrm{1}}}  \ottsym{:}   \ottnt{U_{{\mathrm{1}}}} \Rightarrow  \unskip ^ { \ell }  \! \ottnt{U} $
      is a value.

      \case{$\ottmv{X'}  \sim  \ottmv{X'}$ for some $\ottmv{X'}$}
      By case analysis on $S'_{{\mathrm{1}}}  \ottsym{(}  \ottmv{X'}  \ottsym{)}$.
      \begin{caseanalysis}
        \case{$S'_{{\mathrm{1}}}  \ottsym{(}  \ottmv{X'}  \ottsym{)}  \ottsym{=}  \iota$}
        Similarly to the case $\iota  \sim  \iota$.

        \case{$S'_{{\mathrm{1}}}  \ottsym{(}  \ottmv{X'}  \ottsym{)}  \ottsym{=}  \ottmv{X''}$ for some $\ottmv{X''}$}
        By Lemma~\ref{lem:term_prec_to_typing}, $ \emptyset   \vdash  w'_{{\mathrm{1}}}  \ottsym{:}  \ottmv{X''}$.
        This contradicts Lemma~\ref{lem:canonical_forms}.

        \case{$S'_{{\mathrm{1}}}  \ottsym{(}  \ottmv{X'}  \ottsym{)}$ is a function type}
        We show that
        \[
          \langle   \emptyset    \vdash   w_{{\mathrm{1}}}  \ottsym{:}   \ottnt{U_{{\mathrm{1}}}} \Rightarrow  \unskip ^ { \ell }  \! \ottnt{U}   :  \ottnt{U}   \sqsubseteq _{  S'_{{\mathrm{0}}}  \circ  S_{{\mathrm{0}}}  }  S'_{{\mathrm{1}}}  \ottsym{(}  \ottmv{X'}  \ottsym{)}  :  w'_{{\mathrm{1}}}  \ottsym{:}   S'_{{\mathrm{1}}}  \ottsym{(}  \ottmv{X'}  \ottsym{)} \Rightarrow  \unskip ^ { \ell' }  \! S'_{{\mathrm{1}}}  \ottsym{(}  \ottmv{X'}  \ottsym{)}   \dashv   \emptyset   \rangle .
        \]
        Since $ \langle   \emptyset    \vdash   w_{{\mathrm{1}}}  :  \ottnt{U_{{\mathrm{1}}}}   \sqsubseteq _{ S_{{\mathrm{0}}} }  \ottmv{X'}  :  f'_{{\mathrm{1}}}  \dashv   \emptyset   \rangle $,
        we have $ \ottnt{U_{{\mathrm{1}}}}   \sqsubseteq _{ S_{{\mathrm{0}}} }  \ottmv{X'} $ by Lemma~\ref{lem:term_prec_to_type_prec}.
        Thus, since $ \ottnt{U}   \sqsubseteq _{ S_{{\mathrm{0}}} }  \ottmv{X'} $,
        we have $\ottnt{U_{{\mathrm{1}}}}  \ottsym{=}  \ottnt{U}$ by Lemma~\ref{lem:type_prec_right_var_equal}.
        Since $ \langle   \emptyset    \vdash   w_{{\mathrm{1}}}  :  \ottnt{U_{{\mathrm{1}}}}   \sqsubseteq _{  S'_{{\mathrm{0}}}  \circ  S_{{\mathrm{0}}}  }  S'_{{\mathrm{1}}}  \ottsym{(}  \ottmv{X'}  \ottsym{)}  :  w'_{{\mathrm{1}}}  \dashv   \emptyset   \rangle $,
        we have $ \ottnt{U_{{\mathrm{1}}}}   \sqsubseteq _{  S'_{{\mathrm{0}}}  \circ  S_{{\mathrm{0}}}  }  S'_{{\mathrm{1}}}  \ottsym{(}  \ottmv{X'}  \ottsym{)} $
        by Lemma \ref{lem:term_prec_to_type_prec}.
        Thus, $ \ottnt{U}   \sqsubseteq _{  S'_{{\mathrm{0}}}  \circ  S_{{\mathrm{0}}}  }  S'_{{\mathrm{1}}}  \ottsym{(}  \ottmv{X'}  \ottsym{)} $.
        Since $w'_{{\mathrm{1}}}  \ottsym{:}   S'_{{\mathrm{1}}}  \ottsym{(}  \ottmv{X'}  \ottsym{)} \Rightarrow  \unskip ^ { \ell }  \! S'_{{\mathrm{1}}}  \ottsym{(}  \ottmv{X'}  \ottsym{)} $ is a value (note that
        $S'_{{\mathrm{1}}}  \ottsym{(}  \ottmv{X'}  \ottsym{)}$ is a function type),
        $S'_{{\mathrm{1}}}  \ottsym{(}  \ottmv{X'}  \ottsym{)}  \sim  S'_{{\mathrm{1}}}  \ottsym{(}  \ottmv{X'}  \ottsym{)}$, and $ \ottnt{U}   \sqsubseteq _{  S'_{{\mathrm{0}}}  \circ  S_{{\mathrm{0}}}  }  S'_{{\mathrm{1}}}  \ottsym{(}  \ottmv{X'}  \ottsym{)} $,
        we finish by \rnp{P\_Cast}.
      \end{caseanalysis}

      \case{$\star  \sim  \star$}
      By \rnp{R\_IdStar}, $w'_{{\mathrm{1}}}  \ottsym{:}   S'_{{\mathrm{1}}}  \ottsym{(}  \star  \ottsym{)} \Rightarrow  \unskip ^ { \ell' }  \! S'_{{\mathrm{1}}}  \ottsym{(}  \star  \ottsym{)}  \,  \xmapsto{ \mathmakebox[0.4em]{} [  ] \mathmakebox[0.3em]{} }  \, w'_{{\mathrm{1}}}$.
      By definition, $ \ottnt{U}   \sqsubseteq _{  S'_{{\mathrm{0}}}  \circ  S_{{\mathrm{0}}}  }  \star $.
      Thus, by \rnp{P\_CastL}, we have
      \[
        \langle   \emptyset    \vdash   w_{{\mathrm{1}}}  \ottsym{:}   \ottnt{U_{{\mathrm{1}}}} \Rightarrow  \unskip ^ { \ell }  \! \ottnt{U}   :  \ottnt{U}   \sqsubseteq _{  S'_{{\mathrm{0}}}  \circ  S_{{\mathrm{0}}}  }  \star  :  w'_{{\mathrm{1}}}  \dashv   \emptyset   \rangle ,
      \]
      which is what we want to show.

      \case{$\star  \sim  \ottnt{U'}$ where $\ottnt{U'}  \neq  \star$}
      Here, $S'_{{\mathrm{1}}}  \ottsym{(}  \star  \ottsym{)}  \ottsym{=}  \star$.
      By Lemma~\ref{lem:canonical_forms},
      there exist $w'_{{\mathrm{11}}}$ and $\ottnt{G'}$ such that $w'_{{\mathrm{1}}}  \ottsym{=}  w'_{{\mathrm{11}}}  \ottsym{:}   \ottnt{G'} \Rightarrow  \unskip ^ { \ell'_{{\mathrm{1}}} }  \! \star $.
      By case analysis on $\ottnt{U'}$.

      \begin{caseanalysis}
        \case{$\ottnt{U'}  \ottsym{=}  \iota$ for some $\iota$}
        Here, $ \ottnt{U}   \sqsubseteq _{ S_{{\mathrm{0}}} }  \iota $, that is, $\ottnt{U}  \ottsym{=}  \iota$.
        However, this contradicts the fact that
        $w_{{\mathrm{1}}}  \ottsym{:}   \ottnt{U_{{\mathrm{1}}}} \Rightarrow  \unskip ^ { \ell }  \! \ottnt{U} $ is a value.

        \case{$\ottnt{U'}  \ottsym{=}  \star  \!\rightarrow\!  \star$}
        Here, $S'_{{\mathrm{1}}}  \ottsym{(}  \star  \!\rightarrow\!  \star  \ottsym{)}  \ottsym{=}  \star  \!\rightarrow\!  \star$.
        By Lemma \ref{lem:term_prec_to_type_prec}, $ \ottnt{U}   \sqsubseteq _{ S_{{\mathrm{0}}} }  \star  \!\rightarrow\!  \star $.
        By definition, $\ottnt{U}$ is a function type.
        Since $w_{{\mathrm{1}}}  \ottsym{:}   \ottnt{U_{{\mathrm{1}}}} \Rightarrow  \unskip ^ { \ell }  \! \ottnt{U} $ is a value,
        $\ottnt{U_{{\mathrm{1}}}}$ is also a function type.
        So, by definition, $ \ottnt{U_{{\mathrm{1}}}}   \sqsubseteq _{  S'_{{\mathrm{0}}}  \circ  S_{{\mathrm{0}}}  }  \star  \!\rightarrow\!  \star $.
        By case analysis on $\ottnt{G'}$.

        \begin{caseanalysis}
          \case{$\ottnt{G'}  \ottsym{=}  \iota$ for some $\iota$}
          We are given,
          $ \langle   \emptyset    \vdash   w_{{\mathrm{1}}}  :  \ottnt{U_{{\mathrm{1}}}}   \sqsubseteq _{  S'_{{\mathrm{0}}}  \circ  S_{{\mathrm{0}}}  }  \star  :  w'_{{\mathrm{11}}}  \ottsym{:}   \iota \Rightarrow  \unskip ^ { \ell'_{{\mathrm{1}}} }  \! \star   \dashv   \emptyset   \rangle $.
          This contradicts
          Lemma~\ref{lem:prec_ground_contra} with $ \ottnt{U_{{\mathrm{1}}}}   \sqsubseteq _{  S'_{{\mathrm{0}}}  \circ  S_{{\mathrm{0}}}  }  \star  \!\rightarrow\!  \star $.

          \case{$\ottnt{G'}  \ottsym{=}  \star  \!\rightarrow\!  \star$}
          By \rnp{R\_Succeed}, $w'_{{\mathrm{11}}}  \ottsym{:}   \star  \!\rightarrow\!  \star \Rightarrow  \unskip ^ { \ell'_{{\mathrm{1}}} }  \!  \star \Rightarrow  \unskip ^ { \ell' }  \! \star  \!\rightarrow\!  \star   \,  \xmapsto{ \mathmakebox[0.4em]{} [  ] \mathmakebox[0.3em]{} }  \, w'_{{\mathrm{11}}}$.
          By Lemma \ref{lem:term_prec_inversion1},
          $ \langle   \emptyset    \vdash   w_{{\mathrm{1}}}  :  \ottnt{U_{{\mathrm{1}}}}   \sqsubseteq _{  S'_{{\mathrm{0}}}  \circ  S_{{\mathrm{0}}}  }  \star  \!\rightarrow\!  \star  :  w'_{{\mathrm{11}}}  \dashv   \emptyset   \rangle $.
          Since $ \ottnt{U}   \sqsubseteq _{  S'_{{\mathrm{0}}}  \circ  S_{{\mathrm{0}}}  }  \star  \!\rightarrow\!  \star $, we have, by \rnp{P\_CastL}, 
          \[
            \langle   \emptyset    \vdash   w_{{\mathrm{1}}}  \ottsym{:}   \ottnt{U_{{\mathrm{1}}}} \Rightarrow  \unskip ^ { \ell }  \! \ottnt{U}   :  \ottnt{U}   \sqsubseteq _{  S'_{{\mathrm{0}}}  \circ  S_{{\mathrm{0}}}  }  \star  \!\rightarrow\!  \star  :  w''_{{\mathrm{1}}}  \dashv   \emptyset   \rangle ,
          \]
          which is what we want to show.
        \end{caseanalysis}

        \case{$\ottnt{U'}  \ottsym{=}  \ottmv{X'}$ for some $\ottmv{X'}$}
        If $S'_{{\mathrm{1}}}  \ottsym{(}  \ottmv{X'}  \ottsym{)}$ is not a type variable,
        then we can prove as the other cases.

        Here, $S'_{{\mathrm{1}}}  \ottsym{(}  \ottmv{X'}  \ottsym{)}  \ottsym{=}  \ottmv{X'}$ because $S'_{{\mathrm{1}}}$ is generated by DTI.
        By case analysis on $\ottnt{G'}$.
        \begin{caseanalysis}
          \case{$\ottnt{G'}  \ottsym{=}  \iota$ for some $\iota$}
          By Lemma \ref{lem:term_prec_to_type_prec},
          $ \ottnt{U}   \sqsubseteq _{ S_{{\mathrm{0}}} }  \ottmv{X'} $.
          By Lemma \ref{lem:canonical_forms},
          $\ottnt{U}$ is the dynamic type or a function type.
          By case analysis on $\ottnt{U}$.
          \begin{caseanalysis}
            \case{$\ottnt{U}$ is the dynamic type}
            Cannot happen since $ \ottnt{U}   \sqsubseteq _{ S_{{\mathrm{0}}} }  \ottmv{X'} $.

            \case{$\ottnt{U}$ is a function type}
            By Lemma \ref{lem:canonical_forms},
            $\ottnt{U_{{\mathrm{1}}}}$ is also a function type.
            By definition, we have $ \ottnt{U_{{\mathrm{1}}}}   \sqsubseteq _{  S'_{{\mathrm{0}}}  \circ  S_{{\mathrm{0}}}  }  \star  \!\rightarrow\!  \star $ but
            this contradicts Lemma~\ref{lem:prec_ground_contra} with
            $ \langle   \emptyset    \vdash   w_{{\mathrm{1}}}  :  \ottnt{U_{{\mathrm{1}}}}   \sqsubseteq _{  S'_{{\mathrm{0}}}  \circ  S_{{\mathrm{0}}}  }  \star  :  \ottsym{(}  w'_{{\mathrm{11}}}  \ottsym{:}   \ottnt{G'} \Rightarrow  \unskip ^ { \ell'_{{\mathrm{1}}} }  \! \star   \ottsym{)}  \dashv   \emptyset   \rangle $
            (note that $\ottnt{G'}$ is $\iota$).
          \end{caseanalysis}

          \case{$\ottnt{G'}  \ottsym{=}  \star  \!\rightarrow\!  \star$}
          We are given
          $ \langle   \emptyset    \vdash   w_{{\mathrm{1}}}  :  \ottnt{U_{{\mathrm{1}}}}   \sqsubseteq _{  S'_{{\mathrm{0}}}  \circ  S_{{\mathrm{0}}}  }  \star  :  \ottsym{(}  w'_{{\mathrm{11}}}  \ottsym{:}   \star  \!\rightarrow\!  \star \Rightarrow  \unskip ^ { \ell'_{{\mathrm{1}}} }  \! \star   \ottsym{)}  \dashv   \emptyset   \rangle $.
          %
          By Lemma \ref{lem:term_prec_to_type_prec},
          $ \ottnt{U}   \sqsubseteq _{ S_{{\mathrm{0}}} }  \ottmv{X'} $.
          By Lemma \ref{lem:canonical_forms},
          $\ottnt{U}$ is the dynamic type or a function type.
          The case $\ottnt{U}$ is the dynamic type contradicts
          $ \ottnt{U}   \sqsubseteq _{ S_{{\mathrm{0}}} }  \ottmv{X'} $.
          So, $\ottnt{U}$ is a function type and
          there exist $\ottnt{T_{{\mathrm{1}}}}$ and $\ottnt{T_{{\mathrm{2}}}}$ such that $\ottnt{U}  \ottsym{=}  \ottnt{T_{{\mathrm{1}}}}  \!\rightarrow\!  \ottnt{T_{{\mathrm{2}}}}$
          such that $S_{{\mathrm{0}}}  \ottsym{(}  \ottmv{X'}  \ottsym{)}  \ottsym{=}  \ottnt{T_{{\mathrm{1}}}}  \!\rightarrow\!  \ottnt{T_{{\mathrm{2}}}}$
          since $ \ottnt{U}   \sqsubseteq _{ S_{{\mathrm{0}}} }  \ottmv{X'} $.
          By Lemma \ref{lem:canonical_forms}, $\ottnt{U_{{\mathrm{1}}}}$ is also a function type.
          By definition, $ \ottnt{U_{{\mathrm{1}}}}   \sqsubseteq _{  S'_{{\mathrm{0}}}  \circ  S_{{\mathrm{0}}}  }  \star  \!\rightarrow\!  \star $.

          By \rnp{R\_InstArrow} and \rnp{R\_Succeed},
          $w'_{{\mathrm{11}}}  \ottsym{:}   \star  \!\rightarrow\!  \star \Rightarrow  \unskip ^ { \ell'_{{\mathrm{1}}} }  \!  \star \Rightarrow  \unskip ^ { \ell' }  \! \ottmv{X'}   \,  \xmapsto{ \mathmakebox[0.4em]{} S'_{{\mathrm{2}}} \mathmakebox[0.3em]{} }\hspace{-0.4em}{}^\ast \hspace{0.2em}  \, S'_{{\mathrm{2}}}  \ottsym{(}  w'_{{\mathrm{11}}}  \ottsym{)}  \ottsym{:}   \star  \!\rightarrow\!  \star \Rightarrow  \unskip ^ { \ell' }  \! \ottmv{X'_{{\mathrm{1}}}}  \!\rightarrow\!  \ottmv{X'_{{\mathrm{2}}}} $
          where $S'_{{\mathrm{2}}}  \ottsym{=}  [  \ottmv{X'}  :=  \ottmv{X'_{{\mathrm{1}}}}  \!\rightarrow\!  \ottmv{X'_{{\mathrm{2}}}}  ]$ and $\ottmv{X'_{{\mathrm{1}}}}$ and $\ottmv{X'_{{\mathrm{2}}}}$ are fresh.

          Let $S'_{{\mathrm{3}}}  \ottsym{=}  [  \ottmv{X'_{{\mathrm{1}}}}  :=  \ottnt{T_{{\mathrm{1}}}}  \ottsym{,}  \ottmv{X'_{{\mathrm{2}}}}  :=  \ottnt{T_{{\mathrm{2}}}}  ]$.
          It suffices to show that:
          \begin{align}
           &\forall \ottmv{X} \in \textit{dom} \, \ottsym{(}  S_{{\mathrm{0}}}  \ottsym{)}. S_{{\mathrm{0}}}  \ottsym{(}  \ottmv{X}  \ottsym{)}  \ottsym{=}   \ottsym{(}   S'_{{\mathrm{3}}}  \circ  S'_{{\mathrm{0}}}   \ottsym{)}  \circ   S_{{\mathrm{0}}}  \circ  \ottsym{(}   S'_{{\mathrm{2}}}  \circ  S'_{{\mathrm{1}}}   \ottsym{)}    \ottsym{(}  \ottmv{X}  \ottsym{)}
            \label{req:prec_catch_up_left_value:one} \\
           & \langle   \emptyset    \vdash   \ottsym{(}  w_{{\mathrm{1}}}  \ottsym{:}   \ottnt{U_{{\mathrm{1}}}} \Rightarrow  \unskip ^ { \ell }  \! \ottnt{U}   \ottsym{)}  :  \ottnt{U}   \sqsubseteq _{  \ottsym{(}   S'_{{\mathrm{3}}}  \circ  S'_{{\mathrm{0}}}   \ottsym{)}  \circ  S_{{\mathrm{0}}}  }   S'_{{\mathrm{2}}}  \circ  S'_{{\mathrm{1}}}   \ottsym{(}  \ottmv{X'}  \ottsym{)}  :  \ottsym{(}  S'_{{\mathrm{2}}}  \ottsym{(}  w'_{{\mathrm{11}}}  \ottsym{)}  \ottsym{:}   \star  \!\rightarrow\!  \star \Rightarrow  \unskip ^ { \ell' }  \! \ottmv{X'_{{\mathrm{1}}}}  \!\rightarrow\!  \ottmv{X'_{{\mathrm{2}}}}   \ottsym{)}  \dashv   \emptyset   \rangle 
            \label{req:prec_catch_up_left_value:two} \\
           &\forall \ottmv{X} \in \textit{dom} \, \ottsym{(}  S_{{\mathrm{0}}}  \ottsym{)}. \textit{dom} \, \ottsym{(}   S'_{{\mathrm{3}}}  \circ  S'_{{\mathrm{0}}}   \ottsym{)}  \cap  \textit{ftv} \, \ottsym{(}  S_{{\mathrm{0}}}  \ottsym{(}  \ottmv{X}  \ottsym{)}  \ottsym{)}  \ottsym{=}   \emptyset 
            \label{req:prec_catch_up_left_value:three}
          \end{align}
          (\ref{req:prec_catch_up_left_value:three}) is obvious.

          For (\ref{req:prec_catch_up_left_value:one}),
          let $\ottmv{X} \in \textit{dom} \, \ottsym{(}  S_{{\mathrm{0}}}  \ottsym{)}$.
          If $\ottmv{X}  \ottsym{=}  \ottmv{X'}$, then 
          \[\begin{array}{l@{ \quad = \quad }ll}
             \ottsym{(}   S'_{{\mathrm{3}}}  \circ  S'_{{\mathrm{0}}}   \ottsym{)}  \circ  S_{{\mathrm{0}}}   \circ  \ottsym{(}   S'_{{\mathrm{2}}}  \circ  S'_{{\mathrm{1}}}   \ottsym{)}   \ottsym{(}  \ottmv{X'}  \ottsym{)} &
              \ottsym{(}   S'_{{\mathrm{3}}}  \circ  S'_{{\mathrm{0}}}   \ottsym{)}  \circ  S_{{\mathrm{0}}}   \circ  S'_{{\mathrm{2}}}   \ottsym{(}  \ottmv{X'}  \ottsym{)} & \text{(since $S'_{{\mathrm{1}}}  \ottsym{(}  \ottmv{X'}  \ottsym{)}  \ottsym{=}  \ottmv{X'}$)} \\ &
             \ottsym{(}   S'_{{\mathrm{3}}}  \circ  S'_{{\mathrm{0}}}   \ottsym{)}  \circ  S_{{\mathrm{0}}}   \ottsym{(}  \ottmv{X'_{{\mathrm{1}}}}  \!\rightarrow\!  \ottmv{X'_{{\mathrm{2}}}}  \ottsym{)} \\ &
            S'_{{\mathrm{3}}}  \ottsym{(}  \ottmv{X'_{{\mathrm{1}}}}  \!\rightarrow\!  \ottmv{X'_{{\mathrm{2}}}}  \ottsym{)} & \text{(since $\ottmv{X'_{{\mathrm{1}}}}$ and $\ottmv{X'_{{\mathrm{2}}}}$ are fresh)} \\ &
            \ottnt{T_{{\mathrm{1}}}}  \!\rightarrow\!  \ottnt{T_{{\mathrm{2}}}} \\ &
            S_{{\mathrm{0}}}  \ottsym{(}  \ottmv{X'}  \ottsym{)} & \text{(since $S_{{\mathrm{0}}}  \ottsym{(}  \ottmv{X'}  \ottsym{)}  \ottsym{=}  \ottnt{T_{{\mathrm{1}}}}  \!\rightarrow\!  \ottnt{T_{{\mathrm{2}}}}$)}
            \end{array}\]
          Otherwise, suppose $\ottmv{X}  \neq  \ottmv{X'}$.
          By case analysis on whether $\ottmv{X'} \, \in \, \textit{ftv} \, \ottsym{(}  S'_{{\mathrm{1}}}  \ottsym{(}  \ottmv{X}  \ottsym{)}  \ottsym{)}$ or not.
          \begin{caseanalysis}
           \case{$\ottmv{X'} \, \in \, \textit{ftv} \, \ottsym{(}  S'_{{\mathrm{1}}}  \ottsym{(}  \ottmv{X}  \ottsym{)}  \ottsym{)}$}
           Since $\ottmv{X'} \, \in \, \textit{ftv} \, \ottsym{(}  S'_{{\mathrm{1}}}  \ottsym{(}  \ottmv{X}  \ottsym{)}  \ottsym{)}$,
           $\ottmv{X}  \neq  \ottmv{X'}$, and $S'_{{\mathrm{1}}}$ is generated by DTI,
           $\ottmv{X'}$ is generated as a fresh type variable during evaluation of
           $f'$.
           Thus, we can suppose that $\ottmv{X'}$ does not appear in $S_{{\mathrm{0}}}$.
           However, it contradicts the fact that $S_{{\mathrm{0}}}  \ottsym{(}  \ottmv{X'}  \ottsym{)}  \ottsym{=}  \ottnt{T_{{\mathrm{1}}}}  \!\rightarrow\!  \ottnt{T_{{\mathrm{2}}}}$,
           that is, $\ottmv{X'} \, \in \, \textit{dom} \, \ottsym{(}  S_{{\mathrm{0}}}  \ottsym{)}$.

           \case{$\ottmv{X'} \, \not\in \, \textit{ftv} \, \ottsym{(}  S'_{{\mathrm{1}}}  \ottsym{(}  \ottmv{X}  \ottsym{)}  \ottsym{)}$}
           \[\begin{array}{l@{ \quad = \quad }ll}
              \ottsym{(}   S'_{{\mathrm{3}}}  \circ  S'_{{\mathrm{0}}}   \ottsym{)}  \circ  S_{{\mathrm{0}}}   \circ  \ottsym{(}   S'_{{\mathrm{2}}}  \circ  S'_{{\mathrm{1}}}   \ottsym{)}   \ottsym{(}  \ottmv{X}  \ottsym{)} &
               \ottsym{(}   S'_{{\mathrm{3}}}  \circ  S'_{{\mathrm{0}}}   \ottsym{)}  \circ  S_{{\mathrm{0}}}   \circ  S'_{{\mathrm{1}}}   \ottsym{(}  \ottmv{X}  \ottsym{)} & \text{(since $\ottmv{X'} \, \not\in \, \textit{ftv} \, \ottsym{(}  S'_{{\mathrm{1}}}  \ottsym{(}  \ottmv{X}  \ottsym{)}  \ottsym{)}$)} \\
             S'_{{\mathrm{3}}}  \circ  S_{{\mathrm{0}}}   \ottsym{(}  \ottmv{X}  \ottsym{)} & \text{(since $S_{{\mathrm{0}}}  \ottsym{(}  \ottmv{X}  \ottsym{)}  \ottsym{=}   S'_{{\mathrm{0}}}  \circ   S_{{\mathrm{0}}}  \circ  S'_{{\mathrm{1}}}    \ottsym{(}  \ottmv{X}  \ottsym{)}$)} \\
            S_{{\mathrm{0}}}  \ottsym{(}  \ottmv{X}  \ottsym{)} & \text{(since $\ottmv{X'_{{\mathrm{1}}}}$ and $\ottmv{X'_{{\mathrm{2}}}}$ are fresh)}
             \end{array}\]
          \end{caseanalysis}

          For (\ref{req:prec_catch_up_left_value:two}),
          by \rnp{P\_Cast}, it suffices to show that:
          \begin{equation}
            \langle   \emptyset    \vdash   w_{{\mathrm{1}}}  :  \ottnt{U_{{\mathrm{1}}}}   \sqsubseteq _{  \ottsym{(}   S'_{{\mathrm{3}}}  \circ  S'_{{\mathrm{0}}}   \ottsym{)}  \circ  S_{{\mathrm{0}}}  }  \star  \!\rightarrow\!  \star  :  S'_{{\mathrm{2}}}  \ottsym{(}  w'_{{\mathrm{11}}}  \ottsym{)}  \dashv   \emptyset   \rangle 
            \label{req:prec_catch_up_left_value:two-one}
          \end{equation}
          \begin{equation}
            \ottnt{U}   \sqsubseteq _{  \ottsym{(}   S'_{{\mathrm{3}}}  \circ  S'_{{\mathrm{0}}}   \ottsym{)}  \circ  S_{{\mathrm{0}}}  }   S'_{{\mathrm{2}}}  \circ  S'_{{\mathrm{1}}}   \ottsym{(}  \ottmv{X'}  \ottsym{)} 
            \label{req:prec_catch_up_left_value:two-two}
          \end{equation}

          Since
          $ \langle   \emptyset    \vdash   w_{{\mathrm{1}}}  :  \ottnt{U_{{\mathrm{1}}}}   \sqsubseteq _{  S'_{{\mathrm{0}}}  \circ  S_{{\mathrm{0}}}  }  \star  :  \ottsym{(}  w'_{{\mathrm{11}}}  \ottsym{:}   \star  \!\rightarrow\!  \star \Rightarrow  \unskip ^ { \ell'_{{\mathrm{1}}} }  \! \star   \ottsym{)}  \dashv   \emptyset   \rangle $ and
          $ \ottnt{U_{{\mathrm{1}}}}   \sqsubseteq _{  S'_{{\mathrm{0}}}  \circ  S_{{\mathrm{0}}}  }  \star  \!\rightarrow\!  \star $,
          we have
          $ \langle   \emptyset    \vdash   w_{{\mathrm{1}}}  :  \ottnt{U_{{\mathrm{1}}}}   \sqsubseteq _{  S'_{{\mathrm{0}}}  \circ  S_{{\mathrm{0}}}  }  \star  \!\rightarrow\!  \star  :  w'_{{\mathrm{11}}}  \dashv   \emptyset   \rangle $
          by Lemma~\ref{lem:term_prec_inversion1}.
          Thus, by Lemma~\ref{lem:right_subst_preserve_prec},
          (\ref{req:prec_catch_up_left_value:two-one}) holds
          if we have
          \begin{equation}
           \forall \ottmv{X} \in \textit{dom} \, \ottsym{(}   S'_{{\mathrm{0}}}  \circ  S_{{\mathrm{0}}}   \ottsym{)}.  S'_{{\mathrm{0}}}  \circ  S_{{\mathrm{0}}}   \ottsym{(}  \ottmv{X}  \ottsym{)}  \ottsym{=}   S'_{{\mathrm{3}}}  \circ   \ottsym{(}   S'_{{\mathrm{0}}}  \circ  S_{{\mathrm{0}}}   \ottsym{)}  \circ  S'_{{\mathrm{2}}}    \ottsym{(}  \ottmv{X}  \ottsym{)},
            \label{req:prec_catch_up_left_value:two-one-one}
          \end{equation}
          which we will show now.
          Let $\ottmv{X} \, \in \, \textit{dom} \, \ottsym{(}   S'_{{\mathrm{0}}}  \circ  S_{{\mathrm{0}}}   \ottsym{)}$.
          If $\ottmv{X} \, \in \, \textit{dom} \, \ottsym{(}  S_{{\mathrm{0}}}  \ottsym{)}$, then, since $\textit{dom} \, \ottsym{(}  S'_{{\mathrm{0}}}  \ottsym{)}  \cap  \textit{ftv} \, \ottsym{(}  S_{{\mathrm{0}}}  \ottsym{(}  \ottmv{X}  \ottsym{)}  \ottsym{)}  \ottsym{=}   \emptyset $,
          we have $ S'_{{\mathrm{0}}}  \circ  S_{{\mathrm{0}}}   \ottsym{(}  \ottmv{X}  \ottsym{)}  \ottsym{=}  S_{{\mathrm{0}}}  \ottsym{(}  \ottmv{X}  \ottsym{)}$.
          If $\ottmv{X}  \ottsym{=}  \ottmv{X'}$,
          $  S'_{{\mathrm{3}}}  \circ  \ottsym{(}   S'_{{\mathrm{0}}}  \circ  S_{{\mathrm{0}}}   \ottsym{)}   \circ  S'_{{\mathrm{2}}}   \ottsym{(}  \ottmv{X}  \ottsym{)} =  S'_{{\mathrm{3}}}  \circ  \ottsym{(}   S'_{{\mathrm{0}}}  \circ  S_{{\mathrm{0}}}   \ottsym{)}   \ottsym{(}  \ottmv{X_{{\mathrm{1}}}}  \!\rightarrow\!  \ottmv{X_{{\mathrm{2}}}}  \ottsym{)} = \ottnt{T_{{\mathrm{1}}}}  \!\rightarrow\!  \ottnt{T_{{\mathrm{2}}}} = S_{{\mathrm{0}}}  \ottsym{(}  \ottmv{X}  \ottsym{)}$.
          Otherwise, if $\ottmv{X}  \neq  \ottmv{X'}$,
          $  S'_{{\mathrm{3}}}  \circ  \ottsym{(}   S'_{{\mathrm{0}}}  \circ  S_{{\mathrm{0}}}   \ottsym{)}   \circ  S'_{{\mathrm{2}}}   \ottsym{(}  \ottmv{X}  \ottsym{)} =  S'_{{\mathrm{3}}}  \circ  \ottsym{(}   S'_{{\mathrm{0}}}  \circ  S_{{\mathrm{0}}}   \ottsym{)}   \ottsym{(}  \ottmv{X}  \ottsym{)} = S_{{\mathrm{0}}}  \ottsym{(}  \ottmv{X}  \ottsym{)}$.
          If $\ottmv{X} \, \not\in \, \textit{dom} \, \ottsym{(}  S_{{\mathrm{0}}}  \ottsym{)}$, then $ S'_{{\mathrm{0}}}  \circ  S_{{\mathrm{0}}}   \ottsym{(}  \ottmv{X}  \ottsym{)}  \ottsym{=}  S'_{{\mathrm{0}}}  \ottsym{(}  \ottmv{X}  \ottsym{)}$.
          Since $\ottmv{X'} \, \in \, \textit{dom} \, \ottsym{(}  S_{{\mathrm{0}}}  \ottsym{)}$,
          $  S'_{{\mathrm{3}}}  \circ  \ottsym{(}   S'_{{\mathrm{0}}}  \circ  S_{{\mathrm{0}}}   \ottsym{)}   \circ  S'_{{\mathrm{2}}}   \ottsym{(}  \ottmv{X}  \ottsym{)} = S'_{{\mathrm{0}}}  \ottsym{(}  \ottmv{X}  \ottsym{)} =  S'_{{\mathrm{0}}}  \circ  S_{{\mathrm{0}}}   \ottsym{(}  \ottmv{X}  \ottsym{)}$.
          Thus, we have (\ref{req:prec_catch_up_left_value:two-one-one}).

          Finally, (\ref{req:prec_catch_up_left_value:two-two}) is shown
          by applying Lemma~\ref{lem:right_subst_preserve_prec} to
          $ \ottnt{U}   \sqsubseteq _{ S_{{\mathrm{0}}} }  \ottmv{X'} $ and (\ref{req:prec_catch_up_left_value:one}).
        \end{caseanalysis}

        \case{$\ottnt{U'}  \ottsym{=}  \ottnt{U'_{{\mathrm{11}}}}  \!\rightarrow\!  \ottnt{U'_{{\mathrm{12}}}}$ for some $\ottnt{U'_{{\mathrm{11}}}}$ and $\ottnt{U'_{{\mathrm{12}}}}$ where $\ottnt{U'}  \neq  \star  \!\rightarrow\!  \star$}
        \leavevmode\\
        By Lemma \ref{lem:term_prec_to_type_prec}, $ \ottnt{U}   \sqsubseteq _{ S_{{\mathrm{0}}} }  \ottnt{U'_{{\mathrm{11}}}}  \!\rightarrow\!  \ottnt{U'_{{\mathrm{12}}}} $.
        By definition, $\ottnt{U}$ is a function type.
        By Lemma \ref{lem:canonical_forms}, $\ottnt{U_{{\mathrm{1}}}}$ is a function type.
        By definition, $ \ottnt{U_{{\mathrm{1}}}}   \sqsubseteq _{  S'_{{\mathrm{0}}}  \circ  S_{{\mathrm{0}}}  }  \star  \!\rightarrow\!  \star $.
        By case analysis on $\ottnt{G'}$.

        \begin{caseanalysis}
          \case{$\ottnt{G'}  \ottsym{=}  \iota$ for some $\iota$}
          We are given,
          $ \langle   \emptyset    \vdash   w_{{\mathrm{1}}}  :  \ottnt{U_{{\mathrm{1}}}}   \sqsubseteq _{  S'_{{\mathrm{0}}}  \circ  S_{{\mathrm{0}}}  }  \star  :  w'_{{\mathrm{11}}}  \ottsym{:}   \iota \Rightarrow  \unskip ^ { \ell'_{{\mathrm{1}}} }  \! \star   \dashv   \emptyset   \rangle $.
          This contradicts Lemma \ref{lem:prec_ground_contra}.

          \case{$\ottnt{G'}  \ottsym{=}  \star  \!\rightarrow\!  \star$}
          We are given,
          $ \langle   \emptyset    \vdash   w_{{\mathrm{1}}}  :  \ottnt{U_{{\mathrm{1}}}}   \sqsubseteq _{  S'_{{\mathrm{0}}}  \circ  S_{{\mathrm{0}}}  }  \star  :  w'_{{\mathrm{11}}}  \ottsym{:}   \star  \!\rightarrow\!  \star \Rightarrow  \unskip ^ { \ell'_{{\mathrm{1}}} }  \! \star   \dashv   \emptyset   \rangle $.
          By definition,
          $w'_{{\mathrm{11}}}  \ottsym{:}   \star  \!\rightarrow\!  \star \Rightarrow  \unskip ^ { \ell'_{{\mathrm{1}}} }  \!  \star \Rightarrow  \unskip ^ { \ell' }  \! S'_{{\mathrm{1}}}  \ottsym{(}  \ottnt{U'_{{\mathrm{11}}}}  \!\rightarrow\!  \ottnt{U'_{{\mathrm{12}}}}  \ottsym{)}   \,  \xmapsto{ \mathmakebox[0.4em]{} [  ] \mathmakebox[0.3em]{} }\hspace{-0.4em}{}^\ast \hspace{0.2em}  \, w'_{{\mathrm{11}}}  \ottsym{:}   \star  \!\rightarrow\!  \star \Rightarrow  \unskip ^ { \ell' }  \! S'_{{\mathrm{1}}}  \ottsym{(}  \ottnt{U'_{{\mathrm{11}}}}  \!\rightarrow\!  \ottnt{U'_{{\mathrm{12}}}}  \ottsym{)} $.
          By Lemma \ref{lem:term_prec_inversion1},
          $ \langle   \emptyset    \vdash   w_{{\mathrm{1}}}  :  \ottnt{U_{{\mathrm{1}}}}   \sqsubseteq _{  S'_{{\mathrm{0}}}  \circ  S_{{\mathrm{0}}}  }  \star  \!\rightarrow\!  \star  :  w'_{{\mathrm{11}}}  \dashv   \emptyset   \rangle $
          By Lemma \ref{lem:right_subst_preserve_prec},
          $ \ottnt{U}   \sqsubseteq _{  S'_{{\mathrm{0}}}  \circ  S_{{\mathrm{0}}}  }  S'_{{\mathrm{1}}}  \ottsym{(}  \ottnt{U'_{{\mathrm{11}}}}  \!\rightarrow\!  \ottnt{U'_{{\mathrm{12}}}}  \ottsym{)} $.
          We finish by \rnp{P\_Cast}.
        \end{caseanalysis}
      \end{caseanalysis}

      \case{$\ottnt{U'_{{\mathrm{1}}}}  \sim  \star$ where $\ottnt{U'_{{\mathrm{1}}}}  \neq  \star$}
      By case analysis on $\ottnt{U'_{{\mathrm{1}}}}$.

      \begin{caseanalysis}
        \case{$S'_{{\mathrm{1}}}  \ottsym{(}  \ottnt{U'_{{\mathrm{1}}}}  \ottsym{)}$ is a ground type}
        Obvious.

        \case{$S'_{{\mathrm{1}}}  \ottsym{(}  \ottnt{U'_{{\mathrm{1}}}}  \ottsym{)}$ is not a ground type}
        There exist $\ottnt{U'_{{\mathrm{11}}}}$ and $\ottnt{U'_{{\mathrm{12}}}}$ such that $S'_{{\mathrm{1}}}  \ottsym{(}  \ottnt{U'_{{\mathrm{1}}}}  \ottsym{)}  \ottsym{=}  \ottnt{U'_{{\mathrm{11}}}}  \!\rightarrow\!  \ottnt{U'_{{\mathrm{12}}}}$.
        Since
        $ \langle   \emptyset    \vdash   w_{{\mathrm{1}}}  :  \ottnt{U_{{\mathrm{1}}}}   \sqsubseteq _{  S'_{{\mathrm{0}}}  \circ  S_{{\mathrm{0}}}  }  S'_{{\mathrm{1}}}  \ottsym{(}  \ottnt{U'_{{\mathrm{1}}}}  \ottsym{)}  :  w'_{{\mathrm{1}}}  \dashv   \emptyset   \rangle $,
        we have $ \ottnt{U_{{\mathrm{1}}}}   \sqsubseteq _{  S'_{{\mathrm{0}}}  \circ  S_{{\mathrm{0}}}  }  S'_{{\mathrm{1}}}  \ottsym{(}  \ottnt{U'_{{\mathrm{1}}}}  \ottsym{)} $
        by Lemma~\ref{lem:term_prec_to_type_prec}.
        By definition, $\ottnt{U_{{\mathrm{1}}}}$ is a function type.
        By Lemma \ref{lem:canonical_forms},
        $\ottnt{U}$ is the dynamic type or a function type.
        By case analysis on $\ottnt{U}$.

        \begin{caseanalysis}
          \case{$\ottnt{U}  \ottsym{=}  \star$}
          By Lemma \ref{lem:canonical_forms}, $\ottnt{U_{{\mathrm{1}}}}  \ottsym{=}  \star  \!\rightarrow\!  \star$.
          Here, $ \star  \!\rightarrow\!  \star   \sqsubseteq _{  S'_{{\mathrm{0}}}  \circ  S_{{\mathrm{0}}}  }  S'_{{\mathrm{1}}}  \ottsym{(}  \ottnt{U'_{{\mathrm{1}}}}  \ottsym{)} $
          and $S'_{{\mathrm{1}}}  \ottsym{(}  \ottnt{U'_{{\mathrm{1}}}}  \ottsym{)}$ is not a ground type.
          Contradiction.

          \case{$\ottnt{U}$ is a function type}
          Since $\ottnt{U'}  \ottsym{=}  \star$ and $S'_{{\mathrm{1}}}  \ottsym{(}  \ottnt{U'_{{\mathrm{1}}}}  \ottsym{)}$ is not a ground type,
          $w'_{{\mathrm{1}}}  \ottsym{:}   \ottnt{U'_{{\mathrm{11}}}}  \!\rightarrow\!  \ottnt{U'_{{\mathrm{12}}}} \Rightarrow  \unskip ^ { \ell' }  \! \star  \,  \xmapsto{ \mathmakebox[0.4em]{} [  ] \mathmakebox[0.3em]{} }  \, w'_{{\mathrm{1}}}  \ottsym{:}   \ottnt{U'_{{\mathrm{11}}}}  \!\rightarrow\!  \ottnt{U'_{{\mathrm{12}}}} \Rightarrow  \unskip ^ { \ell' }  \!  \star  \!\rightarrow\!  \star \Rightarrow  \unskip ^ { \ell' }  \! \star  $
          by \rnp{R\_Ground}.
          We finish by \rnp{P\_CastR} and \rnp{P\_Cast}.
        \end{caseanalysis}
      \end{caseanalysis}

      \case{$\ottnt{U'_{{\mathrm{11}}}}  \!\rightarrow\!  \ottnt{U'_{{\mathrm{12}}}}  \sim  \ottnt{U'_{{\mathrm{13}}}}  \!\rightarrow\!  \ottnt{U'_{{\mathrm{14}}}}$ for some $\ottnt{U'_{{\mathrm{11}}}}$, $\ottnt{U'_{{\mathrm{12}}}}$, $\ottnt{U'_{{\mathrm{13}}}}$, and $\ottnt{U'_{{\mathrm{14}}}}$}
      Since $ \ottnt{U}   \sqsubseteq _{ S_{{\mathrm{0}}} }  \ottnt{U'} $,
      $ \ottnt{U}   \sqsubseteq _{  S'_{{\mathrm{0}}}  \circ  S_{{\mathrm{0}}}  }  S'_{{\mathrm{1}}}  \ottsym{(}  \ottnt{U'}  \ottsym{)} $
      by Lemma~\ref{lem:right_subst_preserve_prec}.
      Thus, we finish by \rnp{P\_Cast}.
    \end{caseanalysis}

    \case{\rnp{P\_CastL}}
    Here, $ \langle   \emptyset    \vdash   w_{{\mathrm{1}}}  \ottsym{:}   \ottnt{U_{{\mathrm{1}}}} \Rightarrow  \unskip ^ { \ell }  \! \ottnt{U}   :  \ottnt{U}   \sqsubseteq _{ S_{{\mathrm{0}}} }  \ottnt{U'}  :  f'  \dashv   \emptyset   \rangle $,
    where $w  \ottsym{=}  w_{{\mathrm{1}}}  \ottsym{:}   \ottnt{U_{{\mathrm{1}}}} \Rightarrow  \unskip ^ { \ell }  \! \ottnt{U} $ for some $w_{{\mathrm{1}}}$ and $\ottnt{U_{{\mathrm{1}}}}$.
    By inversion,
    \begin{itemize}
     \item $ \langle   \emptyset    \vdash   w_{{\mathrm{1}}}  :  \ottnt{U_{{\mathrm{1}}}}   \sqsubseteq _{ S_{{\mathrm{0}}} }  \ottnt{U'}  :  f'  \dashv   \emptyset   \rangle $,
     \item $ \ottnt{U}   \sqsubseteq _{ S_{{\mathrm{0}}} }  \ottnt{U'} $, and
     \item $\ottnt{U_{{\mathrm{1}}}}  \sim  \ottnt{U}$.
    \end{itemize}
    By the IH, there exist $S'_{{\mathrm{0}}}$, $S'$, and $w'$ such that
    \begin{itemize}
     \item $f' \,  \xmapsto{ \mathmakebox[0.4em]{} S' \mathmakebox[0.3em]{} }\hspace{-0.4em}{}^\ast \hspace{0.2em}  \, w'$,
     \item $ \langle   \emptyset    \vdash   w_{{\mathrm{1}}}  :  \ottnt{U_{{\mathrm{1}}}}   \sqsubseteq _{  S'_{{\mathrm{0}}}  \circ  S_{{\mathrm{0}}}  }  S'  \ottsym{(}  \ottnt{U'}  \ottsym{)}  :  w'  \dashv   \emptyset   \rangle $,
     \item $\forall \ottmv{X} \in \textit{dom} \, \ottsym{(}  S_{{\mathrm{0}}}  \ottsym{)}. S_{{\mathrm{0}}}  \ottsym{(}  \ottmv{X}  \ottsym{)}  \ottsym{=}   S'_{{\mathrm{0}}}  \circ   S_{{\mathrm{0}}}  \circ  S'    \ottsym{(}  \ottmv{X}  \ottsym{)}$,
     \item $\forall \ottmv{X} \in \textit{dom} \, \ottsym{(}  S_{{\mathrm{0}}}  \ottsym{)}. \textit{dom} \, \ottsym{(}  S'_{{\mathrm{0}}}  \ottsym{)}  \cap  \textit{ftv} \, \ottsym{(}  S_{{\mathrm{0}}}  \ottsym{(}  \ottmv{X}  \ottsym{)}  \ottsym{)}  \ottsym{=}   \emptyset $, and
     \item $\forall \ottmv{X} \in \textit{dom} \, \ottsym{(}  S'  \ottsym{)}. \textit{ftv} \, \ottsym{(}  S'  \ottsym{(}  \ottmv{X}  \ottsym{)}  \ottsym{)}  \subseteq  \textit{dom} \, \ottsym{(}  S'_{{\mathrm{0}}}  \ottsym{)}$.
    \end{itemize}
    Since $ \ottnt{U}   \sqsubseteq _{ S_{{\mathrm{0}}} }  \ottnt{U'} $,
    $ \ottnt{U}   \sqsubseteq _{  S'_{{\mathrm{0}}}  \circ  S_{{\mathrm{0}}}  }  S'  \ottsym{(}  \ottnt{U'}  \ottsym{)} $
    by Lemma~\ref{lem:right_subst_preserve_prec}.
    Thus, we finish by \rnp{P\_CastL}.

    \case{\rnp{P\_CastR}}
    Similar to the case of \rnp{P\_Cast}.
    Here, $ \langle   \emptyset    \vdash   w  :  \ottnt{U}   \sqsubseteq _{ S_{{\mathrm{0}}} }  \ottnt{U'}  :  f'_{{\mathrm{1}}}  \ottsym{:}   \ottnt{U'_{{\mathrm{1}}}} \Rightarrow  \unskip ^ { \ell' }  \! \ottnt{U'}   \dashv   \emptyset   \rangle $,
    where $f'  \ottsym{=}  f'_{{\mathrm{1}}}  \ottsym{:}   \ottnt{U'_{{\mathrm{1}}}} \Rightarrow  \unskip ^ { \ell' }  \! \ottnt{U'} $ for some $f'_{{\mathrm{1}}}$ and $\ottnt{U'_{{\mathrm{1}}}}$.
    By inversion,
    \begin{itemize}
     \item $ \langle   \emptyset    \vdash   w  :  \ottnt{U}   \sqsubseteq _{ S_{{\mathrm{0}}} }  \ottnt{U'_{{\mathrm{1}}}}  :  f'_{{\mathrm{1}}}  \dashv   \emptyset   \rangle $,
     \item $ \ottnt{U}   \sqsubseteq _{ S_{{\mathrm{0}}} }  \ottnt{U'} $, and
     \item $\ottnt{U'_{{\mathrm{1}}}}  \sim  \ottnt{U'}$.
    \end{itemize}
    By the IH, there exist $S'_{{\mathrm{0}}}$, $S'_{{\mathrm{1}}}$ and $w'_{{\mathrm{1}}}$ such that
    \begin{itemize}
     \item $f'_{{\mathrm{1}}} \,  \xmapsto{ \mathmakebox[0.4em]{} S'_{{\mathrm{1}}} \mathmakebox[0.3em]{} }\hspace{-0.4em}{}^\ast \hspace{0.2em}  \, w'_{{\mathrm{1}}}$,
     \item $ \langle   \emptyset    \vdash   w  :  \ottnt{U}   \sqsubseteq _{  S'_{{\mathrm{0}}}  \circ  S_{{\mathrm{0}}}  }  S'_{{\mathrm{1}}}  \ottsym{(}  \ottnt{U'_{{\mathrm{1}}}}  \ottsym{)}  :  w'_{{\mathrm{1}}}  \dashv   \emptyset   \rangle $,
     \item $\forall \ottmv{X} \in \textit{dom} \, \ottsym{(}  S_{{\mathrm{0}}}  \ottsym{)}. S_{{\mathrm{0}}}  \ottsym{(}  \ottmv{X}  \ottsym{)}  \ottsym{=}   S'_{{\mathrm{0}}}  \circ   S_{{\mathrm{0}}}  \circ  S'_{{\mathrm{1}}}    \ottsym{(}  \ottmv{X}  \ottsym{)}$,
     \item $\forall \ottmv{X} \in \textit{dom} \, \ottsym{(}  S_{{\mathrm{0}}}  \ottsym{)}. \textit{dom} \, \ottsym{(}  S'_{{\mathrm{0}}}  \ottsym{)}  \cap  \textit{ftv} \, \ottsym{(}  S_{{\mathrm{0}}}  \ottsym{(}  \ottmv{X}  \ottsym{)}  \ottsym{)}  \ottsym{=}   \emptyset $, and
     \item $\forall \ottmv{X} \in \textit{dom} \, \ottsym{(}  S'_{{\mathrm{1}}}  \ottsym{)}. \textit{ftv} \, \ottsym{(}  S'_{{\mathrm{1}}}  \ottsym{(}  \ottmv{X}  \ottsym{)}  \ottsym{)}  \subseteq  \textit{dom} \, \ottsym{(}  S'_{{\mathrm{0}}}  \ottsym{)}$.
    \end{itemize}
    Thus, by \rnp{E\_Step},
    \[
     f'_{{\mathrm{1}}}  \ottsym{:}   \ottnt{U'_{{\mathrm{1}}}} \Rightarrow  \unskip ^ { \ell' }  \! \ottnt{U'}  \,  \xmapsto{ \mathmakebox[0.4em]{} S'_{{\mathrm{1}}} \mathmakebox[0.3em]{} }\hspace{-0.4em}{}^\ast \hspace{0.2em}  \, w'_{{\mathrm{1}}}  \ottsym{:}   S'_{{\mathrm{1}}}  \ottsym{(}  \ottnt{U'_{{\mathrm{1}}}}  \ottsym{)} \Rightarrow  \unskip ^ { \ell' }  \! S'_{{\mathrm{1}}}  \ottsym{(}  \ottnt{U'}  \ottsym{)} .
    \]
    By case analysis on $\ottnt{U'_{{\mathrm{1}}}}  \sim  \ottnt{U'}$.

    \begin{caseanalysis}
      \case{$\iota  \sim  \iota$ for some $\iota$}
      Here, $S'_{{\mathrm{1}}}  \ottsym{(}  \iota  \ottsym{)}  \ottsym{=}  \iota$.
      By definition, $w'_{{\mathrm{1}}}  \ottsym{:}   \iota \Rightarrow  \unskip ^ { \ell' }  \! \iota  \,  \xmapsto{ \mathmakebox[0.4em]{} [  ] \mathmakebox[0.3em]{} }  \, w'_{{\mathrm{1}}}$.
      Obvious.

      \case{$\ottmv{X'}  \sim  \ottmv{X'}$ for some $\ottmv{X'}$}
      By case analysis on $S'_{{\mathrm{1}}}  \ottsym{(}  \ottmv{X'}  \ottsym{)}$.
      \begin{caseanalysis}
        \case{$S'_{{\mathrm{1}}}  \ottsym{(}  \ottmv{X'}  \ottsym{)}  \ottsym{=}  \iota$}
        Similar to the case of $\iota  \sim  \iota$.

        \case{$S'_{{\mathrm{1}}}  \ottsym{(}  \ottmv{X'}  \ottsym{)}  \ottsym{=}  \ottmv{X''}$ for some $\ottmv{X''}$}
        By Lemma \ref{lem:term_prec_to_typing}, $ \emptyset   \vdash  w'_{{\mathrm{1}}}  \ottsym{:}  \ottmv{X''}$.
        This contradicts Lemma \ref{lem:canonical_forms}.

        \case{$S'_{{\mathrm{1}}}  \ottsym{(}  \ottmv{X'}  \ottsym{)}$ is a function type}
        By Lemma \ref{lem:term_prec_to_type_prec},
        $ \ottnt{U}   \sqsubseteq _{  S'_{{\mathrm{0}}}  \circ  S_{{\mathrm{0}}}  }  S'_{{\mathrm{1}}}  \ottsym{(}  \ottmv{X'}  \ottsym{)} $.
        We finish by \rnp{P\_CastR}.
      \end{caseanalysis}

      \case{$\star  \sim  \star$}
      Here, $S'_{{\mathrm{1}}}  \ottsym{(}  \star  \ottsym{)}  \ottsym{=}  \star$.
      By \rnp{R\_IdStar}, $w'_{{\mathrm{1}}}  \ottsym{:}   \star \Rightarrow  \unskip ^ { \ell' }  \! \star  \,  \xmapsto{ \mathmakebox[0.4em]{} [  ] \mathmakebox[0.3em]{} }  \, w'_{{\mathrm{1}}}$.
      Obvious.

      \case{$\star  \sim  \ottnt{U'}$ where $\ottnt{U'}  \neq  \star$}
      Here, $S'_{{\mathrm{1}}}  \ottsym{(}  \star  \ottsym{)}  \ottsym{=}  \star$.
      By Lemma \ref{lem:canonical_forms},
      there exist $w'_{{\mathrm{11}}}$ and $\ottnt{G'}$ such that $w'_{{\mathrm{1}}}  \ottsym{=}  w'_{{\mathrm{11}}}  \ottsym{:}   \ottnt{G'} \Rightarrow  \unskip ^ { \ell'_{{\mathrm{1}}} }  \! \star $.
      By case analysis on $\ottnt{U'}$.

      \begin{caseanalysis}
        \case{$\ottnt{U'}  \ottsym{=}  \iota$ for some $\iota$}
        By case analysis on $\ottnt{G'}$.
        \begin{caseanalysis}
          \case{$\ottnt{G'}  \ottsym{=}  \iota$}
          By \rnp{R\_Succeed}, $w'_{{\mathrm{11}}}  \ottsym{:}   \iota \Rightarrow  \unskip ^ { \ell'_{{\mathrm{1}}} }  \!  \star \Rightarrow  \unskip ^ { \ell' }  \! \iota   \,  \xmapsto{ \mathmakebox[0.4em]{} [  ] \mathmakebox[0.3em]{} }  \, w'_{{\mathrm{11}}}$.
          By Lemma \ref{lem:term_prec_inversion1},
          $ \langle   \emptyset    \vdash   w  :  \ottnt{U}   \sqsubseteq _{  S'_{{\mathrm{0}}}  \circ  S_{{\mathrm{0}}}  }  \iota  :  w'_{{\mathrm{11}}}  \dashv   \emptyset   \rangle $.

          \case{$\ottnt{G'}  \neq  \iota$}
          This contradicts Lemma \ref{lem:prec_ground_contra}.
        \end{caseanalysis}

        \case{$\ottnt{U'}  \ottsym{=}  \star  \!\rightarrow\!  \star$}
        By case analysis on $\ottnt{G'}$.
        \begin{caseanalysis}
          \case{$\ottnt{G'}  \ottsym{=}  \star  \!\rightarrow\!  \star$}
          By \rnp{R\_Succeed}, $w'_{{\mathrm{11}}}  \ottsym{:}   \star  \!\rightarrow\!  \star \Rightarrow  \unskip ^ { \ell'_{{\mathrm{1}}} }  \!  \star \Rightarrow  \unskip ^ { \ell' }  \! \star  \!\rightarrow\!  \star   \,  \xmapsto{ \mathmakebox[0.4em]{} [  ] \mathmakebox[0.3em]{} }  \, w'_{{\mathrm{11}}}$.
          By Lemma \ref{lem:term_prec_inversion1},
          $ \langle   \emptyset    \vdash   w  :  \ottnt{U}   \sqsubseteq _{  S'_{{\mathrm{0}}}  \circ  S_{{\mathrm{0}}}  }  \star  \!\rightarrow\!  \star  :  w'_{{\mathrm{11}}}  \dashv   \emptyset   \rangle $.

          \case{$\ottnt{G'}  \neq  \star  \!\rightarrow\!  \star$}
          This contradicts Lemma \ref{lem:prec_ground_contra}.
        \end{caseanalysis}

        \case{$\ottnt{U'}  \ottsym{=}  \ottmv{X'}$ for some $\ottmv{X'}$}
        If $S'_{{\mathrm{1}}}  \ottsym{(}  \ottmv{X'}  \ottsym{)}$ is not a type variable,
        then we can prove as the other cases.

        Here, $S'_{{\mathrm{1}}}  \ottsym{(}  \ottmv{X'}  \ottsym{)}  \ottsym{=}  \ottmv{X'}$ because $S'_{{\mathrm{1}}}$ is generated by DTI.
        By case analysis on $\ottnt{G'}$.
        \begin{caseanalysis}
          \case{$\ottnt{G'}  \ottsym{=}  \iota$ for some $\iota$}
          By case analysis on $\ottnt{U}$.

          \begin{caseanalysis}
            \case{$\ottnt{U}  \ottsym{=}  \iota$}
            By definition, $ \ottnt{U}   \sqsubseteq _{  S'_{{\mathrm{0}}}  \circ  S_{{\mathrm{0}}}  }  \iota $.
            By Lemma~\ref{lem:term_prec_inversion1},
            $ \langle   \emptyset    \vdash   w  :  \ottnt{U}   \sqsubseteq _{  S'_{{\mathrm{0}}}  \circ  S_{{\mathrm{0}}}  }  \iota  :  w'_{{\mathrm{11}}}  \dashv   \emptyset   \rangle $.

            By \rnp{R\_InstBase}, $w'_{{\mathrm{11}}}  \ottsym{:}   \iota \Rightarrow  \unskip ^ { \ell'_{{\mathrm{1}}} }  \!  \star \Rightarrow  \unskip ^ { \ell' }  \! \ottmv{X'}   \,  \xmapsto{ \mathmakebox[0.4em]{} S'_{{\mathrm{2}}} \mathmakebox[0.3em]{} }  \, S'_{{\mathrm{2}}}  \ottsym{(}  w'_{{\mathrm{11}}}  \ottsym{)}$
            where $S'_{{\mathrm{2}}}  \ottsym{=}  [  \ottmv{X'}  :=  \iota  ]$.

            We show
            \begin{equation}
             \forall \ottmv{X} \in \textit{dom} \, \ottsym{(}  S_{{\mathrm{0}}}  \ottsym{)}. S_{{\mathrm{0}}}  \ottsym{(}  \ottmv{X}  \ottsym{)}  \ottsym{=}   S'_{{\mathrm{0}}}  \circ   S_{{\mathrm{0}}}  \circ  \ottsym{(}   S'_{{\mathrm{2}}}  \circ  S'_{{\mathrm{1}}}   \ottsym{)}    \ottsym{(}  \ottmv{X}  \ottsym{)}.
              \label{req:prec_catch_up_left_value:four}
            \end{equation}
            Let $\ottmv{X} \, \in \, \textit{dom} \, \ottsym{(}  S_{{\mathrm{0}}}  \ottsym{)}$.
            If $\ottmv{X}  \ottsym{=}  \ottmv{X'}$, then, since $S'_{{\mathrm{1}}}$ is generated by DTI, $S'_{{\mathrm{1}}}  \ottsym{(}  \ottmv{X'}  \ottsym{)}  \ottsym{=}  \ottmv{X'}$.
            Thus, $  S'_{{\mathrm{0}}}  \circ  S_{{\mathrm{0}}}   \circ  \ottsym{(}   S'_{{\mathrm{2}}}  \circ  S'_{{\mathrm{1}}}   \ottsym{)}   \ottsym{(}  \ottmv{X'}  \ottsym{)}  \ottsym{=}  \iota$.
            Since $ \ottnt{U}   \sqsubseteq _{ S_{{\mathrm{0}}} }  \ottmv{X'} $ (note that $\ottnt{U'}  \ottsym{=}  \ottmv{X'}$) and
            $\ottnt{U}  \ottsym{=}  \iota$,
            we have $S_{{\mathrm{0}}}  \ottsym{(}  \ottmv{X'}  \ottsym{)}  \ottsym{=}  \iota$.
            Thus, $S_{{\mathrm{0}}}  \ottsym{(}  \ottmv{X'}  \ottsym{)}  \ottsym{=}   S'_{{\mathrm{0}}}  \circ   S_{{\mathrm{0}}}  \circ  \ottsym{(}   S'_{{\mathrm{2}}}  \circ  S'_{{\mathrm{1}}}   \ottsym{)}    \ottsym{(}  \ottmv{X'}  \ottsym{)}$.
            Otherwise, suppose that $\ottmv{X}  \neq  \ottmv{X'}$.
            By case analysis on whether $\ottmv{X'} \, \in \, \textit{ftv} \, \ottsym{(}  S'_{{\mathrm{1}}}  \ottsym{(}  \ottmv{X}  \ottsym{)}  \ottsym{)}$ or not.
            \begin{caseanalysis}
             \case{$\ottmv{X'} \, \in \, \textit{ftv} \, \ottsym{(}  S'_{{\mathrm{1}}}  \ottsym{(}  \ottmv{X}  \ottsym{)}  \ottsym{)}$}
             Since $\ottmv{X'} \, \in \, \textit{ftv} \, \ottsym{(}  S'_{{\mathrm{1}}}  \ottsym{(}  \ottmv{X}  \ottsym{)}  \ottsym{)}$,
             $\ottmv{X}  \neq  \ottmv{X'}$, and $S'_{{\mathrm{1}}}$ is generated by DTI,
             $\ottmv{X'}$ is generated as a fresh type variable during evaluation of
             $f'$.
             Thus, we can suppose that $\ottmv{X'}$ does not appear in $S_{{\mathrm{0}}}$.
             However, it contradicts the fact that $S_{{\mathrm{0}}}  \ottsym{(}  \ottmv{X'}  \ottsym{)}  \ottsym{=}  \iota$,
             that is, $\ottmv{X'} \, \in \, \textit{dom} \, \ottsym{(}  S_{{\mathrm{0}}}  \ottsym{)}$.

             \case{$\ottmv{X'} \, \not\in \, \textit{ftv} \, \ottsym{(}  S'_{{\mathrm{1}}}  \ottsym{(}  \ottmv{X}  \ottsym{)}  \ottsym{)}$}
              We have $S_{{\mathrm{0}}}  \ottsym{(}  \ottmv{X}  \ottsym{)} =   S'_{{\mathrm{0}}}  \circ  S_{{\mathrm{0}}}   \circ  S'_{{\mathrm{1}}}   \ottsym{(}  \ottmv{X}  \ottsym{)} =   S'_{{\mathrm{0}}}  \circ  S_{{\mathrm{0}}}   \circ  \ottsym{(}   S'_{{\mathrm{2}}}  \circ  S'_{{\mathrm{1}}}   \ottsym{)}   \ottsym{(}  \ottmv{X}  \ottsym{)}$.
            \end{caseanalysis}

            By Lemma \ref{lem:right_subst_preserve_prec} with
            (\ref{req:prec_catch_up_left_value:four}), we have
            $ \langle   \emptyset    \vdash   w  :  \ottnt{U}   \sqsubseteq _{  S'_{{\mathrm{0}}}  \circ  S_{{\mathrm{0}}}  }  \iota  :  S'_{{\mathrm{2}}}  \ottsym{(}  w'_{{\mathrm{11}}}  \ottsym{)}  \dashv   \emptyset   \rangle $.

            \case{$\ottnt{U}  \ottsym{=}  \ottnt{G}$ where $\ottnt{G}  \neq  \iota$}
            This contradicts Lemma \ref{lem:prec_ground_contra}.

            \case{$\ottnt{U}$ is a function type}
            By definition, $ \ottnt{U}   \sqsubseteq _{  S'_{{\mathrm{0}}}  \circ  S_{{\mathrm{0}}}  }  \star  \!\rightarrow\!  \star $.
            This contradicts Lemma \ref{lem:prec_ground_contra}.

            \case{$\ottnt{U}  \ottsym{=}  \star$}
            Cannot happen since $ \ottnt{U}   \sqsubseteq _{ S_{{\mathrm{0}}} }  \ottnt{U'} $.

            \case{$\ottnt{U}  \ottsym{=}  \ottmv{X''}$ for some $\ottmv{X''}$}
            Cannot happen by Lemma~\ref{lem:canonical_forms}.
          \end{caseanalysis}

          \case{$\ottnt{G'}  \ottsym{=}  \star  \!\rightarrow\!  \star$}
          By case analysis on $\ottnt{U}$.

          \begin{caseanalysis}
            \case{$\ottnt{U}  \ottsym{=}  \ottnt{G}$ where $\ottnt{G}  \neq  \star  \!\rightarrow\!  \star$}
            Cannot happen by Lemma \ref{lem:prec_ground_contra}.

            \case{$\ottnt{U}  \ottsym{=}  \star$}
            Cannot happen since $ \ottnt{U}   \sqsubseteq _{ S_{{\mathrm{0}}} }  \ottnt{U'} $.

            \case{$\ottnt{U}$ is a function type}
            By definition, $ \ottnt{U}   \sqsubseteq _{  S'_{{\mathrm{0}}}  \circ  S_{{\mathrm{0}}}  }  \star  \!\rightarrow\!  \star $.
            By Lemma \ref{lem:term_prec_inversion1},
            $ \langle   \emptyset    \vdash   w  :  \ottnt{U}   \sqsubseteq _{  S'_{{\mathrm{0}}}  \circ  S_{{\mathrm{0}}}  }  \star  \!\rightarrow\!  \star  :  w'_{{\mathrm{11}}}  \dashv   \emptyset   \rangle $.

            By \rnp{R\_InstArrow}, $w'_{{\mathrm{11}}}  \ottsym{:}   \star  \!\rightarrow\!  \star \Rightarrow  \unskip ^ { \ell'_{{\mathrm{1}}} }  \!  \star \Rightarrow  \unskip ^ { \ell' }  \! \ottmv{X'}   \,  \xmapsto{ \mathmakebox[0.4em]{} S'_{{\mathrm{2}}} \mathmakebox[0.3em]{} }\hspace{-0.4em}{}^\ast \hspace{0.2em}  \, S'_{{\mathrm{2}}}  \ottsym{(}  w'_{{\mathrm{11}}}  \ottsym{)}  \ottsym{:}   \star  \!\rightarrow\!  \star \Rightarrow  \unskip ^ { \ell' }  \! \ottmv{X'_{{\mathrm{1}}}}  \!\rightarrow\!  \ottmv{X'_{{\mathrm{2}}}} $
            where $S'_{{\mathrm{2}}}  \ottsym{=}  [  \ottmv{X'}  :=  \ottmv{X'_{{\mathrm{1}}}}  \!\rightarrow\!  \ottmv{X'_{{\mathrm{2}}}}  ]$ and $\ottmv{X'_{{\mathrm{1}}}}$ and $\ottmv{X'_{{\mathrm{2}}}}$ are fresh.
            Since $ \ottnt{U}   \sqsubseteq _{ S_{{\mathrm{0}}} }  \ottnt{U'} $ and $\ottnt{U'}  \ottsym{=}  \ottmv{X'}$, we have
            there exist some $\ottnt{T_{{\mathrm{1}}}}$ and $\ottnt{T_{{\mathrm{2}}}}$ such that
            $S_{{\mathrm{0}}}  \ottsym{(}  \ottmv{X'}  \ottsym{)} = \ottnt{U} = \ottnt{T_{{\mathrm{1}}}}  \!\rightarrow\!  \ottnt{U_{{\mathrm{2}}}}$.

            Let $S'_{{\mathrm{3}}}  \ottsym{=}  [  \ottmv{X'_{{\mathrm{1}}}}  :=  \ottnt{T_{{\mathrm{1}}}}  \ottsym{,}  \ottmv{X'_{{\mathrm{2}}}}  :=  \ottnt{T_{{\mathrm{2}}}}  ]$.
            It suffices to show that:
            \begin{align}
             \forall \ottmv{X} \in \textit{dom} \, \ottsym{(}  S_{{\mathrm{0}}}  \ottsym{)}. S_{{\mathrm{0}}}  \ottsym{(}  \ottmv{X}  \ottsym{)}  \ottsym{=}   \ottsym{(}   S'_{{\mathrm{3}}}  \circ  S'_{{\mathrm{0}}}   \ottsym{)}  \circ   S_{{\mathrm{0}}}  \circ  \ottsym{(}   S'_{{\mathrm{2}}}  \circ  S'_{{\mathrm{1}}}   \ottsym{)}    \ottsym{(}  \ottmv{X}  \ottsym{)}
              \label{req:prec_catch_up_left_value:five} \\
              \langle   \emptyset    \vdash   w  :  \ottnt{U}   \sqsubseteq _{  \ottsym{(}   S'_{{\mathrm{3}}}  \circ  S'_{{\mathrm{0}}}   \ottsym{)}  \circ  S_{{\mathrm{0}}}  }  \ottsym{(}   S'_{{\mathrm{2}}}  \circ  S'_{{\mathrm{1}}}   \ottsym{)}  \ottsym{(}  \ottnt{U'}  \ottsym{)}  :  \ottsym{(}   S'_{{\mathrm{2}}}  \circ  S'_{{\mathrm{1}}}   \ottsym{)}  \ottsym{(}  w'_{{\mathrm{11}}}  \ottsym{:}   \star  \!\rightarrow\!  \star \Rightarrow  \unskip ^ { \ell' }  \! \ottmv{X'_{{\mathrm{1}}}}  \!\rightarrow\!  \ottmv{X'_{{\mathrm{2}}}}   \ottsym{)}  \dashv   \emptyset   \rangle 
              \label{req:prec_catch_up_left_value:six} \\
             \forall \ottmv{X} \in \textit{dom} \, \ottsym{(}  S_{{\mathrm{0}}}  \ottsym{)}. \textit{dom} \, \ottsym{(}   S'_{{\mathrm{3}}}  \circ  S'_{{\mathrm{0}}}   \ottsym{)}  \cap  \textit{ftv} \, \ottsym{(}  S_{{\mathrm{0}}}  \ottsym{(}  \ottmv{X}  \ottsym{)}  \ottsym{)}  \ottsym{=}   \emptyset 
              \label{req:prec_catch_up_left_value:seven}
            \end{align}
            (\ref{req:prec_catch_up_left_value:seven}) is obvious.

            We first show that
            \begin{equation}
             \forall \ottmv{X} \in \textit{dom} \, \ottsym{(}   S'_{{\mathrm{0}}}  \circ  S_{{\mathrm{0}}}   \ottsym{)}.  S'_{{\mathrm{0}}}  \circ  S_{{\mathrm{0}}}   \ottsym{(}  \ottmv{X}  \ottsym{)}  \ottsym{=}   S'_{{\mathrm{3}}}  \circ    S'_{{\mathrm{0}}}  \circ  S_{{\mathrm{0}}}   \circ  S'_{{\mathrm{2}}}    \ottsym{(}  \ottmv{X}  \ottsym{)}.
              \label{req:prec_catch_up_left_value:five-one}
            \end{equation}
            Let $\ottmv{X} \, \in \, \textit{dom} \, \ottsym{(}  S_{{\mathrm{0}}}  \ottsym{)}$.
            If $\ottmv{X}  \ottsym{=}  \ottmv{X'}$, then
            $   S'_{{\mathrm{3}}}  \circ  S'_{{\mathrm{0}}}   \circ  S_{{\mathrm{0}}}   \circ  S'_{{\mathrm{2}}}   \ottsym{(}  \ottmv{X'}  \ottsym{)} =
               S'_{{\mathrm{3}}}  \circ  S'_{{\mathrm{0}}}   \circ  S_{{\mathrm{0}}}   \ottsym{(}  \ottmv{X'_{{\mathrm{1}}}}  \!\rightarrow\!  \ottmv{X'_{{\mathrm{2}}}}  \ottsym{)} =
             S'_{{\mathrm{3}}}  \ottsym{(}  \ottmv{X'_{{\mathrm{1}}}}  \!\rightarrow\!  \ottmv{X'_{{\mathrm{2}}}}  \ottsym{)} =
             \ottnt{T_{{\mathrm{1}}}}  \!\rightarrow\!  \ottnt{T_{{\mathrm{2}}}} =
             S_{{\mathrm{0}}}  \ottsym{(}  \ottmv{X'}  \ottsym{)}$.
            Since $\textit{dom} \, \ottsym{(}  S'_{{\mathrm{0}}}  \ottsym{)}  \cap  \textit{ftv} \, \ottsym{(}  S_{{\mathrm{0}}}  \ottsym{(}  \ottmv{X'}  \ottsym{)}  \ottsym{)}  \ottsym{=}   \emptyset $ by the IH,
            we have $ S'_{{\mathrm{0}}}  \circ  S_{{\mathrm{0}}}   \ottsym{(}  \ottmv{X'}  \ottsym{)} = S_{{\mathrm{0}}}  \ottsym{(}  \ottmv{X'}  \ottsym{)} =    S'_{{\mathrm{3}}}  \circ  S'_{{\mathrm{0}}}   \circ  S_{{\mathrm{0}}}   \circ  S'_{{\mathrm{2}}}   \ottsym{(}  \ottmv{X'}  \ottsym{)}$.
            Otherwise, if $\ottmv{X}  \neq  \ottmv{X'}$, then
            $   S'_{{\mathrm{3}}}  \circ  S'_{{\mathrm{0}}}   \circ  S_{{\mathrm{0}}}   \circ  S'_{{\mathrm{2}}}   \ottsym{(}  \ottmv{X}  \ottsym{)}  \ottsym{=}   S'_{{\mathrm{0}}}  \circ  S_{{\mathrm{0}}}   \ottsym{(}  \ottmv{X}  \ottsym{)}$.

            We show (\ref{req:prec_catch_up_left_value:five}).
            Let $\ottmv{X} \, \in \, \textit{dom} \, \ottsym{(}  S_{{\mathrm{0}}}  \ottsym{)}$.
            We have
            \[\begin{array}{l@{\quad = \quad}ll}
               \ottsym{(}   S'_{{\mathrm{3}}}  \circ  S'_{{\mathrm{0}}}   \ottsym{)}  \circ  S_{{\mathrm{0}}}   \circ  \ottsym{(}   S'_{{\mathrm{2}}}  \circ  S'_{{\mathrm{1}}}   \ottsym{)}   \ottsym{(}  \ottmv{X}  \ottsym{)} &
                S'_{{\mathrm{0}}}  \circ  S_{{\mathrm{0}}}   \circ  S'_{{\mathrm{1}}}   \ottsym{(}  \ottmv{X}  \ottsym{)} &
               \text{(by (\ref{req:prec_catch_up_left_value:five-one}))} \\ &
              S_{{\mathrm{0}}}  \ottsym{(}  \ottmv{X}  \ottsym{)} &
               \text{(by the IH)}
              \end{array}\]

            We show (\ref{req:prec_catch_up_left_value:six}).
            Since $ \langle   \emptyset    \vdash   w  :  \ottnt{U}   \sqsubseteq _{  S'_{{\mathrm{0}}}  \circ  S_{{\mathrm{0}}}  }  \star  \!\rightarrow\!  \star  :  w'_{{\mathrm{11}}}  \dashv   \emptyset   \rangle $,
            we have $ \langle   \emptyset    \vdash   w  :  \ottnt{U}   \sqsubseteq _{   S'_{{\mathrm{3}}}  \circ  S'_{{\mathrm{0}}}   \circ  S_{{\mathrm{0}}}  }  \star  \!\rightarrow\!  \star  :  S'_{{\mathrm{2}}}  \ottsym{(}  w'_{{\mathrm{11}}}  \ottsym{)}  \dashv   \emptyset   \rangle $
            by Lemma~\ref{lem:right_subst_preserve_prec} with
            (\ref{req:prec_catch_up_left_value:five-one}).
            Since $ \ottnt{U}   \sqsubseteq _{ S_{{\mathrm{0}}} }  \ottnt{U'} $ and $\ottnt{U'}  \ottsym{=}  \ottmv{X'}$,
            we have $ \ottnt{U}   \sqsubseteq _{  \ottsym{(}   S'_{{\mathrm{3}}}  \circ  S'_{{\mathrm{0}}}   \ottsym{)}  \circ  S_{{\mathrm{0}}}  }  \ottsym{(}   S'_{{\mathrm{2}}}  \circ  S'_{{\mathrm{1}}}   \ottsym{)}  \ottsym{(}  \ottmv{X'}  \ottsym{)} $
            by Lemma~\ref{lem:right_subst_preserve_prec} with
            (\ref{req:prec_catch_up_left_value:five}).
            Since $S'_{{\mathrm{1}}}  \ottsym{(}  w'_{{\mathrm{11}}}  \ottsym{)}  \ottsym{=}  w'_{{\mathrm{11}}}$ and
            $\ottsym{(}   S'_{{\mathrm{2}}}  \circ  S'_{{\mathrm{1}}}   \ottsym{)}  \ottsym{(}  \ottnt{U'}  \ottsym{)} = \ottsym{(}   S'_{{\mathrm{1}}}  \circ  S'_{{\mathrm{2}}}   \ottsym{)}  \ottsym{(}  \ottmv{X'}  \ottsym{)} = \ottmv{X'_{{\mathrm{1}}}}  \!\rightarrow\!  \ottmv{X'_{{\mathrm{2}}}}$,
            we have (\ref{req:prec_catch_up_left_value:six}) by \rnp{P\_CastR}.

            \case{$\ottnt{U}  \ottsym{=}  \ottmv{X''}$ for some $\ottmv{X''}$}
            Cannot happen by Lemma~\ref{lem:canonical_forms}.

          \end{caseanalysis}
        \end{caseanalysis}

        \case{$\ottnt{U'}  \ottsym{=}  \ottnt{U'_{{\mathrm{11}}}}  \!\rightarrow\!  \ottnt{U'_{{\mathrm{12}}}}$ for some $\ottnt{U'_{{\mathrm{11}}}}$ and $\ottnt{U'_{{\mathrm{12}}}}$ where $\ottnt{U'}  \neq  \star  \!\rightarrow\!  \star$}
        \leavevmode\\
        By Lemma \ref{lem:term_prec_to_type_prec}, $ \ottnt{U}   \sqsubseteq _{ S_{{\mathrm{0}}} }  \ottnt{U'_{{\mathrm{11}}}}  \!\rightarrow\!  \ottnt{U'_{{\mathrm{12}}}} $.
        By definition, $\ottnt{U}$ is a function type.
        So, by definition, $ \ottnt{U}   \sqsubseteq _{  S'_{{\mathrm{0}}}  \circ  S_{{\mathrm{0}}}  }  \star  \!\rightarrow\!  \star $.
        By Lemma \ref{lem:right_subst_preserve_prec},
        $ \ottnt{U}   \sqsubseteq _{  S'_{{\mathrm{0}}}  \circ  S_{{\mathrm{0}}}  }  S'_{{\mathrm{1}}}  \ottsym{(}  \ottnt{U'_{{\mathrm{11}}}}  \!\rightarrow\!  \ottnt{U'_{{\mathrm{12}}}}  \ottsym{)} $.

        By \rnp{R\_Expand},
        $w'_{{\mathrm{1}}}  \ottsym{:}   \star \Rightarrow  \unskip ^ { \ell' }  \! S'_{{\mathrm{1}}}  \ottsym{(}  \ottnt{U'_{{\mathrm{11}}}}  \!\rightarrow\!  \ottnt{U'_{{\mathrm{12}}}}  \ottsym{)}  \,  \xmapsto{ \mathmakebox[0.4em]{} [  ] \mathmakebox[0.3em]{} }  \, w'_{{\mathrm{1}}}  \ottsym{:}   \star \Rightarrow  \unskip ^ { \ell' }  \!  \star  \!\rightarrow\!  \star \Rightarrow  \unskip ^ { \ell' }  \! S'_{{\mathrm{1}}}  \ottsym{(}  \ottnt{U'_{{\mathrm{11}}}}  \!\rightarrow\!  \ottnt{U'_{{\mathrm{12}}}}  \ottsym{)}  $.
        We finish by \rnp{P\_CastR}.
      \end{caseanalysis}

      \case{$\ottnt{U'_{{\mathrm{1}}}}  \sim  \star$ where $\ottnt{U'_{{\mathrm{1}}}}  \neq  \star$}
      By case analysis on $S'_{{\mathrm{1}}}  \ottsym{(}  \ottnt{U'_{{\mathrm{1}}}}  \ottsym{)}$.

      \begin{caseanalysis}
        \case{$S'_{{\mathrm{1}}}  \ottsym{(}  \ottnt{U'_{{\mathrm{1}}}}  \ottsym{)}$ is a ground type}
        Obvious.

        \case{$S'_{{\mathrm{1}}}  \ottsym{(}  \ottnt{U'_{{\mathrm{1}}}}  \ottsym{)}  \ottsym{=}  \star$} Cannot happen.
        \case{$S'_{{\mathrm{1}}}  \ottsym{(}  \ottnt{U'_{{\mathrm{1}}}}  \ottsym{)}  \ottsym{=}  \ottmv{X'}$ for some $\ottmv{X'}$}
        Since $ \langle   \emptyset    \vdash   w  :  \ottnt{U}   \sqsubseteq _{  S'_{{\mathrm{0}}}  \circ  S_{{\mathrm{0}}}  }  S'_{{\mathrm{1}}}  \ottsym{(}  \ottnt{U'_{{\mathrm{1}}}}  \ottsym{)}  :  w'_{{\mathrm{1}}}  \dashv   \emptyset   \rangle $,
        we have $ \emptyset   \vdash_{\textsf{\textup{B}\relax}\relax}  w'_{{\mathrm{1}}}  \ottsym{:}  \ottmv{X'}$, which 
        contradicts Lemma~\ref{lem:canonical_forms}.

        \case{$S'_{{\mathrm{1}}}  \ottsym{(}  \ottnt{U'_{{\mathrm{1}}}}  \ottsym{)}  \ottsym{=}  \ottnt{U'_{{\mathrm{11}}}}  \!\rightarrow\!  \ottnt{U'_{{\mathrm{12}}}}$ for some $\ottnt{U'_{{\mathrm{11}}}}$ and $\ottnt{U'_{{\mathrm{12}}}}$}.
        By Lemma~\ref{lem:term_prec_to_type_prec},
        $ \ottnt{U}   \sqsubseteq _{  S'_{{\mathrm{0}}}  \circ  S_{{\mathrm{0}}}  }  S'_{{\mathrm{1}}}  \ottsym{(}  \ottnt{U'_{{\mathrm{1}}}}  \ottsym{)} $.
        By definition, $\ottnt{U}$ is a function type.
        So, $ \ottnt{U}   \sqsubseteq _{  S'_{{\mathrm{0}}}  \circ  S_{{\mathrm{0}}}  }  \star  \!\rightarrow\!  \star $.
        By \rnp{R\_Ground},
        $w'_{{\mathrm{11}}}  \ottsym{:}   \ottnt{U'_{{\mathrm{11}}}}  \!\rightarrow\!  \ottnt{U'_{{\mathrm{12}}}} \Rightarrow  \unskip ^ { \ell' }  \! \star  \,  \xmapsto{ \mathmakebox[0.4em]{} [  ] \mathmakebox[0.3em]{} }  \, w'_{{\mathrm{1}}}  \ottsym{:}   \ottnt{U'_{{\mathrm{11}}}}  \!\rightarrow\!  \ottnt{U'_{{\mathrm{12}}}} \Rightarrow  \unskip ^ { \ell' }  \!  \star  \!\rightarrow\!  \star \Rightarrow  \unskip ^ { \ell' }  \! \star  $.
        We finish by \rnp{P\_CastR}.
      \end{caseanalysis}

      \case{$\ottnt{U'_{{\mathrm{11}}}}  \!\rightarrow\!  \ottnt{U'_{{\mathrm{12}}}}  \sim  \ottnt{U'_{{\mathrm{13}}}}  \!\rightarrow\!  \ottnt{U'_{{\mathrm{14}}}}$ for some $\ottnt{U'_{{\mathrm{11}}}}$, $\ottnt{U'_{{\mathrm{12}}}}$, $\ottnt{U'_{{\mathrm{13}}}}$, and $\ottnt{U'_{{\mathrm{14}}}}$}
      Obvious.
    \end{caseanalysis}

    \otherwise
    Cannot happen.  \qedhere
  \end{caseanalysis}
\end{proof}

\begin{lemmaA} \label{lem:type_prec_right_static_type}
 If $ \ottnt{U}   \sqsubseteq _{ S }  \ottnt{T} $, then $\ottnt{U}  \ottsym{=}  S  \ottsym{(}  \ottnt{T}  \ottsym{)}$.
\end{lemmaA}
\begin{proof}
 By induction on the derivation of $ \ottnt{U}   \sqsubseteq _{ S }  \ottnt{T} $.
 \begin{caseanalysis}
  \case{\rnp{P\_IdBase}} Obvious.
  \case{\rnp{P\_TyVar}} Obvious.
  \case{\rnp{P\_Dyn}} Cannot happen.
  \case{\rnp{P\_Arrow}} By the IHs.
  \qedhere
 \end{caseanalysis}
\end{proof}

\begin{lemmaA} \label{lem:type_prec_right_subst_pullback}
 If $ \ottnt{U}   \sqsubseteq _{ S_{{\mathrm{1}}} }  S_{{\mathrm{2}}}  \ottsym{(}  \ottnt{U'}  \ottsym{)} $,
 then $ \ottnt{U}   \sqsubseteq _{  S_{{\mathrm{1}}}  \circ  S_{{\mathrm{2}}}  }  \ottnt{U'} $.
\end{lemmaA}
\begin{proof}
 By induction on the derivation of $ \ottnt{U}   \sqsubseteq _{ S_{{\mathrm{1}}} }  S_{{\mathrm{2}}}  \ottsym{(}  \ottnt{U'}  \ottsym{)} $.
 \begin{caseanalysis}
  \case{\rnp{P\_IdBase}}
   We are given $ \iota   \sqsubseteq _{ S_{{\mathrm{1}}} }  \iota $ for some $\iota$ such that
   $\ottnt{U}  \ottsym{=}  \iota$ and $S_{{\mathrm{2}}}  \ottsym{(}  \ottnt{U'}  \ottsym{)}  \ottsym{=}  \iota$.
   If $\ottnt{U'}  \ottsym{=}  \iota$, then we finish by \rnp{P\_IdBase}.
   Otherwise, $\ottnt{U'}  \ottsym{=}  \ottmv{X}$ and $S_{{\mathrm{2}}}  \ottsym{(}  \ottmv{X}  \ottsym{)}  \ottsym{=}  \iota$ for some $\ottmv{X}$.
   Since $ S_{{\mathrm{1}}}  \circ  S_{{\mathrm{2}}}   \ottsym{(}  \ottmv{X}  \ottsym{)}  \ottsym{=}  \iota$, we have
   $ \iota   \sqsubseteq _{  S_{{\mathrm{1}}}  \circ  S_{{\mathrm{2}}}  }  \ottmv{X} $ by \rnp{P\_TyVar}.

  \case{\rnp{P\_TyVar}}
   We are given $ S_{{\mathrm{1}}}  \ottsym{(}  \ottmv{X}  \ottsym{)}   \sqsubseteq _{ S_{{\mathrm{1}}} }  \ottmv{X} $ for some $\ottmv{X}$ such that
   $\ottnt{U}  \ottsym{=}  S_{{\mathrm{1}}}  \ottsym{(}  \ottmv{X}  \ottsym{)}$ and $S_{{\mathrm{2}}}  \ottsym{(}  \ottnt{U'}  \ottsym{)}  \ottsym{=}  \ottmv{X}$.
   By inversion, we have $\ottmv{X} \, \in \, \textit{dom} \, \ottsym{(}  S_{{\mathrm{1}}}  \ottsym{)}$.
   Since $S_{{\mathrm{2}}}  \ottsym{(}  \ottnt{U'}  \ottsym{)}  \ottsym{=}  \ottmv{X}$, there exist $\ottmv{X'}$ such that
   $\ottnt{U'}  \ottsym{=}  \ottmv{X'}$ and $S_{{\mathrm{2}}}  \ottsym{(}  \ottmv{X'}  \ottsym{)}  \ottsym{=}  \ottmv{X}$.
   That is, we have $ S_{{\mathrm{1}}}  \ottsym{(}  \ottmv{X}  \ottsym{)}   \sqsubseteq _{ S_{{\mathrm{1}}} }  S_{{\mathrm{2}}}  \ottsym{(}  \ottmv{X'}  \ottsym{)} $.
   We show that $ S_{{\mathrm{1}}}  \ottsym{(}  \ottmv{X}  \ottsym{)}   \sqsubseteq _{  S_{{\mathrm{1}}}  \circ  S_{{\mathrm{2}}}  }  \ottmv{X'} $.
   If $\ottmv{X'} \, \in \, \textit{dom} \, \ottsym{(}  S_{{\mathrm{2}}}  \ottsym{)}$, then we finish by \rnp{P\_TyVar}
   since $ S_{{\mathrm{1}}}  \circ  S_{{\mathrm{2}}}   \ottsym{(}  \ottmv{X'}  \ottsym{)}  \ottsym{=}  S_{{\mathrm{1}}}  \ottsym{(}  \ottmv{X}  \ottsym{)}$.
   Otherwise, suppose that $\ottmv{X'} \, \not\in \, \textit{dom} \, \ottsym{(}  S_{{\mathrm{2}}}  \ottsym{)}$.
   Since $S_{{\mathrm{2}}}  \ottsym{(}  \ottmv{X'}  \ottsym{)}  \ottsym{=}  \ottmv{X}$, $\ottmv{X}  \ottsym{=}  \ottmv{X'}$.
   Since $\ottmv{X} \, \in \, \textit{dom} \, \ottsym{(}  S_{{\mathrm{1}}}  \ottsym{)}$, we have $\ottmv{X'} \, \in \, \textit{dom} \, \ottsym{(}  S_{{\mathrm{1}}}  \ottsym{)}$.
   Since $S_{{\mathrm{1}}}  \ottsym{(}  \ottmv{X}  \ottsym{)}  \ottsym{=}   S_{{\mathrm{1}}}  \circ  S_{{\mathrm{2}}}   \ottsym{(}  \ottmv{X'}  \ottsym{)}$, we have
   $ S_{{\mathrm{1}}}  \ottsym{(}  \ottmv{X}  \ottsym{)}   \sqsubseteq _{  S_{{\mathrm{1}}}  \circ  S_{{\mathrm{2}}}  }  \ottmv{X'} $ by \rnp{P\_TyVar}.

  \case{\rnp{P\_Dyn}} By \rnp{P\_Dyn} since $\ottnt{U'}  \ottsym{=}  \star$.

  \case{\rnp{P\_Arrow}}
   We are given $ \ottnt{U_{{\mathrm{1}}}}  \!\rightarrow\!  \ottnt{U_{{\mathrm{2}}}}   \sqsubseteq _{ S_{{\mathrm{1}}} }  \ottnt{U'_{{\mathrm{1}}}}  \!\rightarrow\!  \ottnt{U'_{{\mathrm{2}}}} $ for some
   $\ottnt{U_{{\mathrm{1}}}}$, $\ottnt{U_{{\mathrm{2}}}}$, $\ottnt{U'_{{\mathrm{1}}}}$, and $\ottnt{U'_{{\mathrm{2}}}}$ such that
   $\ottnt{U}  \ottsym{=}  \ottnt{U_{{\mathrm{1}}}}  \!\rightarrow\!  \ottnt{U_{{\mathrm{2}}}}$ and $S_{{\mathrm{2}}}  \ottsym{(}  \ottnt{U'}  \ottsym{)}  \ottsym{=}  \ottnt{U'_{{\mathrm{1}}}}  \!\rightarrow\!  \ottnt{U'_{{\mathrm{2}}}}$.
   By inversion, $ \ottnt{U_{{\mathrm{1}}}}   \sqsubseteq _{ S_{{\mathrm{1}}} }  \ottnt{U'_{{\mathrm{1}}}} $ and $ \ottnt{U_{{\mathrm{2}}}}   \sqsubseteq _{ S_{{\mathrm{1}}} }  \ottnt{U'_{{\mathrm{2}}}} $.
   By case analysis on $\ottnt{U'}$.
   \begin{caseanalysis}
    \case{$\ottnt{U'}  \ottsym{=}  \ottmv{X}$ for some $\ottmv{X}$}
     Since $S_{{\mathrm{2}}}  \ottsym{(}  \ottmv{X}  \ottsym{)}  \ottsym{=}  \ottnt{U'_{{\mathrm{1}}}}  \!\rightarrow\!  \ottnt{U'_{{\mathrm{2}}}}$, there exist $\ottnt{T'_{{\mathrm{1}}}}$ and $\ottnt{T'_{{\mathrm{2}}}}$ such that
     $\ottnt{U'_{{\mathrm{1}}}}  \ottsym{=}  \ottnt{T'_{{\mathrm{1}}}}$ and $\ottnt{U'_{{\mathrm{2}}}}  \ottsym{=}  \ottnt{T'_{{\mathrm{2}}}}$.
     We have $ \ottnt{U_{{\mathrm{1}}}}  \!\rightarrow\!  \ottnt{U_{{\mathrm{2}}}}   \sqsubseteq _{ S_{{\mathrm{1}}} }  \ottnt{T'_{{\mathrm{1}}}}  \!\rightarrow\!  \ottnt{T'_{{\mathrm{2}}}} $.
     By Lemma~\ref{lem:type_prec_right_static_type},
     $\ottnt{U_{{\mathrm{1}}}}  \!\rightarrow\!  \ottnt{U_{{\mathrm{2}}}} = S_{{\mathrm{1}}}  \ottsym{(}  \ottnt{T'_{{\mathrm{1}}}}  \!\rightarrow\!  \ottnt{T'_{{\mathrm{2}}}}  \ottsym{)} = S_{{\mathrm{1}}}  \ottsym{(}  S_{{\mathrm{2}}}  \ottsym{(}  \ottmv{X}  \ottsym{)}  \ottsym{)}$.
     Thus, by \rnp{P\_TyVar}, $ \ottnt{U_{{\mathrm{1}}}}  \!\rightarrow\!  \ottnt{U_{{\mathrm{2}}}}   \sqsubseteq _{  S_{{\mathrm{1}}}  \circ  S_{{\mathrm{2}}}  }  \ottmv{X} $.
    \case{$\ottnt{U'}  \ottsym{=}  \ottnt{U''_{{\mathrm{1}}}}  \!\rightarrow\!  \ottnt{U''_{{\mathrm{2}}}}$ for some $\ottnt{U''_{{\mathrm{1}}}}$ and $\ottnt{U''_{{\mathrm{2}}}}$}
     We have $\ottnt{U'_{{\mathrm{1}}}}  \ottsym{=}  S_{{\mathrm{2}}}  \ottsym{(}  \ottnt{U''_{{\mathrm{1}}}}  \ottsym{)}$ and $\ottnt{U'_{{\mathrm{2}}}}  \ottsym{=}  S_{{\mathrm{2}}}  \ottsym{(}  \ottnt{U''_{{\mathrm{2}}}}  \ottsym{)}$.
     Since $ \ottnt{U_{{\mathrm{1}}}}   \sqsubseteq _{ S_{{\mathrm{1}}} }  S_{{\mathrm{2}}}  \ottsym{(}  \ottnt{U''_{{\mathrm{1}}}}  \ottsym{)} $ and $ \ottnt{U_{{\mathrm{2}}}}   \sqsubseteq _{ S_{{\mathrm{1}}} }  S_{{\mathrm{2}}}  \ottsym{(}  \ottnt{U''_{{\mathrm{2}}}}  \ottsym{)} $,
     we have $ \ottnt{U}   \sqsubseteq _{  S_{{\mathrm{1}}}  \circ  S_{{\mathrm{2}}}  }  \ottnt{U''_{{\mathrm{1}}}} $ and $ \ottnt{U}   \sqsubseteq _{  S_{{\mathrm{1}}}  \circ  S_{{\mathrm{2}}}  }  \ottnt{U''_{{\mathrm{2}}}} $
     by the IHs.
     By \rnp{P\_Arrow}, we have
     $ \ottnt{U}   \sqsubseteq _{  S_{{\mathrm{1}}}  \circ  S_{{\mathrm{2}}}  }  \ottnt{U''_{{\mathrm{1}}}}  \!\rightarrow\!  \ottnt{U''_{{\mathrm{2}}}} $.
     \otherwise Contradiction.
     \qedhere
   \end{caseanalysis}
 \end{caseanalysis}
\end{proof}

\begin{lemmaA} \label{lem:type_prec_subst_reflexive}
 If $ \ottnt{U}   \sqsubseteq _{ S_{{\mathrm{1}}} }  \ottnt{U'} $ and $ \ottnt{U}   \sqsubseteq _{ S_{{\mathrm{2}}} }  \ottnt{U'} $,
 then, for any $\ottmv{X} \, \in \, \textit{ftv} \, \ottsym{(}  \ottnt{U'}  \ottsym{)}$, $S_{{\mathrm{1}}}  \ottsym{(}  \ottmv{X}  \ottsym{)}  \ottsym{=}  S_{{\mathrm{2}}}  \ottsym{(}  \ottmv{X}  \ottsym{)}$.
\end{lemmaA}
\begin{proof}
 By induction on the derivation of $ \ottnt{U}   \sqsubseteq _{ S_{{\mathrm{1}}} }  \ottnt{U'} $.
 \begin{caseanalysis}
  \case{\rnp{P\_IdBase}} Obvious since $\textit{ftv} \, \ottsym{(}  \ottnt{U'}  \ottsym{)}  \ottsym{=}   \emptyset $.
  \case{\rnp{P\_TyVar}}
   We are given $ S_{{\mathrm{1}}}  \ottsym{(}  \ottmv{X}  \ottsym{)}   \sqsubseteq _{ S_{{\mathrm{1}}} }  \ottmv{X} $ for some $\ottmv{X}$ such that
   $\ottnt{U}  \ottsym{=}  S_{{\mathrm{1}}}  \ottsym{(}  \ottmv{X}  \ottsym{)}$ and $\ottnt{U'}  \ottsym{=}  \ottmv{X}$.
   Since $ S_{{\mathrm{1}}}  \ottsym{(}  \ottmv{X}  \ottsym{)}   \sqsubseteq _{ S_{{\mathrm{2}}} }  \ottmv{X} $, we have
   $S_{{\mathrm{2}}}  \ottsym{(}  \ottmv{X}  \ottsym{)}  \ottsym{=}  S_{{\mathrm{1}}}  \ottsym{(}  \ottmv{X}  \ottsym{)}$; note that the precision rules applicable
   in the case that types on the right-hand side are type variables
   are only \rnp{P\_TyVar}.

  \case{\rnp{P\_Dyn}} Obvious since $\textit{ftv} \, \ottsym{(}  \ottnt{U'}  \ottsym{)}  \ottsym{=}   \emptyset $.
  \case{\rnp{P\_Arrow}} By the IHs.
  \qedhere
 \end{caseanalysis}
\end{proof}

\begin{lemmaA} \label{lem:type_app_preserve_term_prec}
  If
  \begin{itemize}
   \item $ \langle  \Gamma   \vdash   \ottnt{f}  :  \ottnt{U}   \sqsubseteq _{  [   \overrightarrow{ \ottmv{X'_{\ottmv{j}}} }   :=   \overrightarrow{ \ottnt{T''_{\ottmv{j}}} }   ]  \uplus  S  }  \ottnt{U'}  :  \ottnt{f'}  \dashv  \Gamma'  \rangle $ and
   \item $ \ottnt{U}  [   \overrightarrow{ \ottmv{X_{\ottmv{i}}} }   \ottsym{:=}   \overrightarrow{ \ottnt{T_{\ottmv{i}}} }   ]   \sqsubseteq _{ S }  \ottnt{U'}  [   \overrightarrow{ \ottmv{X'_{\ottmv{j}}} }   \ottsym{:=}   \overrightarrow{ \ottnt{T'_{\ottmv{j}}} }   ] $ and
   \item for any $\ottmv{X} \, \in \, \textit{dom} \, \ottsym{(}  S  \ottsym{)}$, $ \overrightarrow{ \ottmv{X_{\ottmv{i}}} }   \cap  \textit{ftv} \, \ottsym{(}  S  \ottsym{(}  \ottmv{X}  \ottsym{)}  \ottsym{)}  \ottsym{=}   \emptyset $ and
         $ \overrightarrow{ \ottmv{X'_{\ottmv{j}}} }   \cap  \textit{ftv} \, \ottsym{(}  S  \ottsym{(}  \ottmv{X}  \ottsym{)}  \ottsym{)}  \ottsym{=}   \emptyset $ and
   \item $ \overrightarrow{ \ottmv{X'_{\ottmv{j}}} }   \cap  \textit{ftv} \, \ottsym{(}   \overrightarrow{ \ottnt{T'_{\ottmv{j}}} }   \ottsym{)}  \ottsym{=}   \emptyset $ and
   \item For any $\ottmv{X'_{\ottmv{k}}} \, \in \,  \overrightarrow{ \ottmv{X'_{\ottmv{j}}} }   \setminus  \textit{ftv} \, \ottsym{(}  \ottnt{U'}  \ottsym{)}$, $\ottnt{T'_{\ottmv{k}}}$ is a fresh type
         variable (we write $\ottnt{P}$ for the set of such type variables
         $\ottnt{T'_{\ottmv{k}}}$),
  \end{itemize}
  then there exist $S'$ such that
  \begin{itemize}
   \item $\textit{dom} \, \ottsym{(}  S'  \ottsym{)} = \ottnt{P}$ and
   \item for any $\ottmv{X} = \ottnt{T'_{\ottmv{k}}} \in \ottnt{P}$, $S'  \ottsym{(}  \ottmv{X}  \ottsym{)}  \ottsym{=}  \ottnt{T''_{\ottmv{k}}}  [   \overrightarrow{ \ottmv{X_{\ottmv{i}}} }   \ottsym{:=}   \overrightarrow{ \ottnt{T_{\ottmv{i}}} }   ]$ and
   \item $ \langle  \Gamma  [   \overrightarrow{ \ottmv{X_{\ottmv{i}}} }   \ottsym{:=}   \overrightarrow{ \ottnt{T_{\ottmv{i}}} }   ]   \vdash   \ottnt{f}  [   \overrightarrow{ \ottmv{X_{\ottmv{i}}} }   \ottsym{:=}   \overrightarrow{ \ottnt{T_{\ottmv{i}}} }   ]  :  \ottnt{U}  [   \overrightarrow{ \ottmv{X_{\ottmv{i}}} }   \ottsym{:=}   \overrightarrow{ \ottnt{T_{\ottmv{i}}} }   ]   \sqsubseteq _{  S  \uplus  S'  }  \ottnt{U'}  [   \overrightarrow{ \ottmv{X'_{\ottmv{j}}} }   \ottsym{:=}   \overrightarrow{ \ottnt{T'_{\ottmv{j}}} }   ]  :  \ottnt{f'}  [   \overrightarrow{ \ottmv{X'_{\ottmv{j}}} }   \ottsym{:=}   \overrightarrow{ \ottnt{T'_{\ottmv{j}}} }   ]  \dashv  \Gamma'  [   \overrightarrow{ \ottmv{X'_{\ottmv{j}}} }   \ottsym{:=}   \overrightarrow{ \ottnt{T'_{\ottmv{j}}} }   ]  \rangle $.
  \end{itemize}
\end{lemmaA}
\begin{proof}
 Since $ \langle  \Gamma   \vdash   \ottnt{f}  :  \ottnt{U}   \sqsubseteq _{  [   \overrightarrow{ \ottmv{X'_{\ottmv{j}}} }   :=   \overrightarrow{ \ottnt{T''_{\ottmv{j}}} }   ]  \uplus  S  }  \ottnt{U'}  :  \ottnt{f'}  \dashv  \Gamma'  \rangle $,
 we have
 \begin{equation*}
   \langle  \Gamma  [   \overrightarrow{ \ottmv{X_{\ottmv{i}}} }   \ottsym{:=}   \overrightarrow{ \ottnt{T_{\ottmv{i}}} }   ]   \vdash   \ottnt{f}  [   \overrightarrow{ \ottmv{X_{\ottmv{i}}} }   \ottsym{:=}   \overrightarrow{ \ottnt{T_{\ottmv{i}}} }   ]  :  \ottnt{U}  [   \overrightarrow{ \ottmv{X_{\ottmv{i}}} }   \ottsym{:=}   \overrightarrow{ \ottnt{T_{\ottmv{i}}} }   ]   \sqsubseteq _{  [   \overrightarrow{ \ottmv{X_{\ottmv{i}}} }   :=   \overrightarrow{ \ottnt{T_{\ottmv{i}}} }   ]  \circ  \ottsym{(}   [   \overrightarrow{ \ottmv{X'_{\ottmv{j}}} }   :=   \overrightarrow{ \ottnt{T''_{\ottmv{j}}} }   ]  \uplus  S   \ottsym{)}  }  \ottnt{U'}  :  \ottnt{f'}  \dashv  \Gamma'  \rangle 
 \end{equation*}
 by Lemma~\ref{lem:left_subst_preserve_prec}.
 Since, for any $\ottmv{X} \, \in \, \textit{dom} \, \ottsym{(}  S  \ottsym{)}$, $ \overrightarrow{ \ottmv{X_{\ottmv{i}}} }   \cap  \textit{ftv} \, \ottsym{(}  S  \ottsym{(}  \ottmv{X}  \ottsym{)}  \ottsym{)}  \ottsym{=}   \emptyset $,
 we have $ [   \overrightarrow{ \ottmv{X_{\ottmv{i}}} }   :=   \overrightarrow{ \ottnt{T_{\ottmv{i}}} }   ]  \circ  \ottsym{(}   [   \overrightarrow{ \ottmv{X'_{\ottmv{j}}} }   :=   \overrightarrow{ \ottnt{T''_{\ottmv{j}}} }   ]  \uplus  S   \ottsym{)}   \ottsym{=}   [   \overrightarrow{ \ottmv{X'_{\ottmv{j}}} }   :=   \overrightarrow{ \ottnt{T''_{\ottmv{j}}}  [   \overrightarrow{ \ottmv{X_{\ottmv{i}}} }   \ottsym{:=}   \overrightarrow{ \ottnt{T_{\ottmv{i}}} }   ] }   ]  \uplus  S $.
 Thus, 
 \begin{equation*}
   \langle  \Gamma  [   \overrightarrow{ \ottmv{X_{\ottmv{i}}} }   \ottsym{:=}   \overrightarrow{ \ottnt{T_{\ottmv{i}}} }   ]   \vdash   \ottnt{f}  [   \overrightarrow{ \ottmv{X_{\ottmv{i}}} }   \ottsym{:=}   \overrightarrow{ \ottnt{T_{\ottmv{i}}} }   ]  :  \ottnt{U}  [   \overrightarrow{ \ottmv{X_{\ottmv{i}}} }   \ottsym{:=}   \overrightarrow{ \ottnt{T_{\ottmv{i}}} }   ]   \sqsubseteq _{  [   \overrightarrow{ \ottmv{X'_{\ottmv{j}}} }   :=   \overrightarrow{ \ottnt{T''_{\ottmv{j}}}  [   \overrightarrow{ \ottmv{X_{\ottmv{i}}} }   \ottsym{:=}   \overrightarrow{ \ottnt{T_{\ottmv{i}}} }   ] }   ]  \uplus  S  }  \ottnt{U'}  :  \ottnt{f'}  \dashv  \Gamma'  \rangle 
 \end{equation*}
 Let $ \overrightarrow{  { \ottmv{X} }_{   1     \ottmv{j}   }'  } $ be $ \overrightarrow{ \ottmv{X'_{\ottmv{j}}} }   \cap  \textit{ftv} \, \ottsym{(}  \ottnt{U'}  \ottsym{)}$ and
 $ \overrightarrow{  { \ottmv{X} }_{   2     \ottmv{j}   }'  } $ be $ \overrightarrow{ \ottmv{X'_{\ottmv{j}}} }   \setminus  \textit{ftv} \, \ottsym{(}  \ottnt{U'}  \ottsym{)}$.
 Then,
 \begin{equation}
   \langle  \Gamma  [   \overrightarrow{ \ottmv{X_{\ottmv{i}}} }   \ottsym{:=}   \overrightarrow{ \ottnt{T_{\ottmv{i}}} }   ]   \vdash   \ottnt{f}  [   \overrightarrow{ \ottmv{X_{\ottmv{i}}} }   \ottsym{:=}   \overrightarrow{ \ottnt{T_{\ottmv{i}}} }   ]  :  \ottnt{U}  [   \overrightarrow{ \ottmv{X_{\ottmv{i}}} }   \ottsym{:=}   \overrightarrow{ \ottnt{T_{\ottmv{i}}} }   ]   \sqsubseteq _{   [   \overrightarrow{ \ottmv{X} _{   1     \ottmv{j}   }'}   :=   \overrightarrow{  { \ottnt{T} }_{   1     \ottmv{j}   }''   [   \overrightarrow{ \ottmv{X_{\ottmv{i}}} }   \ottsym{:=}   \overrightarrow{ \ottnt{T_{\ottmv{i}}} }   ] }   ]  \uplus  [   \overrightarrow{ \ottmv{X} _{   2     \ottmv{j}   }'}   :=   \overrightarrow{  { \ottnt{T} }_{   2     \ottmv{j}   }''   [   \overrightarrow{ \ottmv{X_{\ottmv{i}}} }   \ottsym{:=}   \overrightarrow{ \ottnt{T_{\ottmv{i}}} }   ] }   ]   \uplus  S  }  \ottnt{U'}  :  \ottnt{f'}  \dashv  \Gamma'  \rangle 
   \label{eq:type_app_preserve_term_prec:one}
 \end{equation}

 By Lemma~\ref{lem:term_prec_to_type_prec},
 $ \ottnt{U}  [   \overrightarrow{ \ottmv{X_{\ottmv{i}}} }   \ottsym{:=}   \overrightarrow{ \ottnt{T_{\ottmv{i}}} }   ]   \sqsubseteq _{  [   \overrightarrow{ \ottmv{X_{\ottmv{i}}} }   :=   \overrightarrow{ \ottnt{T_{\ottmv{i}}} }   ]  \circ  \ottsym{(}   [   \overrightarrow{ \ottmv{X'_{\ottmv{j}}} }   :=   \overrightarrow{ \ottnt{T''_{\ottmv{j}}} }   ]  \uplus  S   \ottsym{)}  }  \ottnt{U'} $.
 Since $ \ottnt{U}  [   \overrightarrow{ \ottmv{X_{\ottmv{i}}} }   \ottsym{:=}   \overrightarrow{ \ottnt{T_{\ottmv{i}}} }   ]   \sqsubseteq _{ S }  \ottnt{U'}  [   \overrightarrow{ \ottmv{X'_{\ottmv{j}}} }   \ottsym{:=}   \overrightarrow{ \ottnt{T'_{\ottmv{j}}} }   ] $,
 we have
 $ \ottnt{U}  [   \overrightarrow{ \ottmv{X_{\ottmv{i}}} }   \ottsym{:=}   \overrightarrow{ \ottnt{T_{\ottmv{i}}} }   ]   \sqsubseteq _{  S  \circ  [   \overrightarrow{ \ottmv{X'_{\ottmv{j}}} }   :=   \overrightarrow{ \ottnt{T'_{\ottmv{j}}} }   ]  }  \ottnt{U'} $
 by Lemma~\ref{lem:type_prec_right_subst_pullback}.
 Thus, by Lemma~\ref{lem:type_prec_subst_reflexive},
 \begin{equation*}
  \text{for any } \ottmv{X} \, \in \, \textit{ftv} \, \ottsym{(}  \ottnt{U'}  \ottsym{)},  [   \overrightarrow{ \ottmv{X_{\ottmv{i}}} }   :=   \overrightarrow{ \ottnt{T_{\ottmv{i}}} }   ]  \circ  \ottsym{(}   [   \overrightarrow{ \ottmv{X'_{\ottmv{j}}} }   :=   \overrightarrow{ \ottnt{T''_{\ottmv{j}}} }   ]  \uplus  S   \ottsym{)}   \ottsym{(}  \ottmv{X}  \ottsym{)}  \ottsym{=}   S  \circ  [   \overrightarrow{ \ottmv{X'_{\ottmv{j}}} }   :=   \overrightarrow{ \ottnt{T'_{\ottmv{j}}} }   ]   \ottsym{(}  \ottmv{X}  \ottsym{)}
 \end{equation*}
 Since $\textit{dom} \, \ottsym{(}  S  \ottsym{)}  \cap   \overrightarrow{ \ottmv{X'_{\ottmv{j}}} }   \ottsym{=}   \emptyset $,
 for any $\ottmv{k}$ such that $\ottmv{X'_{\ottmv{k}}} \, \in \,  \overrightarrow{ \ottmv{X} _{   1     \ottmv{j}   }'} $,
 \begin{equation}
   \ottnt{T''_{\ottmv{k}}}  [   \overrightarrow{ \ottmv{X_{\ottmv{i}}} }   \ottsym{:=}   \overrightarrow{ \ottnt{T_{\ottmv{i}}} }   ]  \ottsym{=}  S  \ottsym{(}  \ottnt{T'_{\ottmv{k}}}  \ottsym{)}
   \label{eq:type_app_preserve_term_prec:two}
 \end{equation}
 Here,
 \[\begin{array}{llll}
  &   [   \overrightarrow{ \ottmv{X} _{   1     \ottmv{j}   }'}   :=   \overrightarrow{  { \ottnt{T} }_{   1     \ottmv{j}   }''   [   \overrightarrow{ \ottmv{X_{\ottmv{i}}} }   \ottsym{:=}   \overrightarrow{ \ottnt{T_{\ottmv{i}}} }   ] }   ]  \uplus  [   \overrightarrow{ \ottmv{X} _{   2     \ottmv{j}   }'}   :=   \overrightarrow{  { \ottnt{T} }_{   2     \ottmv{j}   }''   [   \overrightarrow{ \ottmv{X_{\ottmv{i}}} }   \ottsym{:=}   \overrightarrow{ \ottnt{T_{\ottmv{i}}} }   ] }   ]   \uplus  S  \\ = &
      [   \overrightarrow{ \ottmv{X} _{   1     \ottmv{j}   }'}   :=   \overrightarrow{ S  \ottsym{(}   { \ottnt{T} }_{   1     \ottmv{j}   }'   \ottsym{)} }   ]  \uplus  [   \overrightarrow{ \ottmv{X} _{   2     \ottmv{j}   }'}   :=   \overrightarrow{  { \ottnt{T} }_{   2     \ottmv{j}   }''   [   \overrightarrow{ \ottmv{X_{\ottmv{i}}} }   \ottsym{:=}   \overrightarrow{ \ottnt{T_{\ottmv{i}}} }   ] }   ]   \uplus  S  & \text{(by (\ref{eq:type_app_preserve_term_prec:two}))} \\ = &
     [   \overrightarrow{ \ottmv{X} _{   2     \ottmv{j}   }'}   :=   \overrightarrow{  { \ottnt{T} }_{   2     \ottmv{j}   }''   [   \overrightarrow{ \ottmv{X_{\ottmv{i}}} }   \ottsym{:=}   \overrightarrow{ \ottnt{T_{\ottmv{i}}} }   ] }   ]  \uplus  \ottsym{(}   S  \circ  [   \overrightarrow{ \ottmv{X} _{   1     \ottmv{j}   }'}   :=   \overrightarrow{  { \ottnt{T} }_{   1     \ottmv{j}   }'  }   ]   \ottsym{)}  \\ = &
      [   \overrightarrow{ \ottmv{X} _{   2     \ottmv{j}   }'}   :=   \overrightarrow{  { \ottnt{T} }_{   2     \ottmv{j}   }''   [   \overrightarrow{ \ottmv{X_{\ottmv{i}}} }   \ottsym{:=}   \overrightarrow{ \ottnt{T_{\ottmv{i}}} }   ] }   ]  \circ  S   \circ  [   \overrightarrow{ \ottmv{X} _{   1     \ottmv{j}   }'}   :=   \overrightarrow{  { \ottnt{T} }_{   1     \ottmv{j}   }'  }   ]  \\ &
    \multicolumn{3}{r}{
     \text{(since $ \overrightarrow{ \ottmv{X'_{\ottmv{j}}} }   \cap  \textit{ftv} \, \ottsym{(}   \overrightarrow{ \ottnt{T'_{\ottmv{j}}} }   \ottsym{)}  \ottsym{=}   \emptyset $ and for any $\ottmv{X} \, \in \, \textit{dom} \, \ottsym{(}  S  \ottsym{)}$ $ \overrightarrow{ \ottmv{X'_{\ottmv{j}}} }   \cap  \textit{ftv} \, \ottsym{(}  S  \ottsym{(}  \ottmv{X}  \ottsym{)}  \ottsym{)}  \ottsym{=}   \emptyset $)}
    }
   \end{array}\]
 Thus, from (\ref{eq:type_app_preserve_term_prec:one}),
 \begin{equation*}
   \langle  \Gamma  [   \overrightarrow{ \ottmv{X_{\ottmv{i}}} }   \ottsym{:=}   \overrightarrow{ \ottnt{T_{\ottmv{i}}} }   ]   \vdash   \ottnt{f}  [   \overrightarrow{ \ottmv{X_{\ottmv{i}}} }   \ottsym{:=}   \overrightarrow{ \ottnt{T_{\ottmv{i}}} }   ]  :  \ottnt{U}  [   \overrightarrow{ \ottmv{X_{\ottmv{i}}} }   \ottsym{:=}   \overrightarrow{ \ottnt{T_{\ottmv{i}}} }   ]   \sqsubseteq _{   [   \overrightarrow{ \ottmv{X} _{   2     \ottmv{j}   }'}   :=   \overrightarrow{  { \ottnt{T} }_{   2     \ottmv{j}   }''   [   \overrightarrow{ \ottmv{X_{\ottmv{i}}} }   \ottsym{:=}   \overrightarrow{ \ottnt{T_{\ottmv{i}}} }   ] }   ]  \circ  S   \circ  [   \overrightarrow{ \ottmv{X} _{   1     \ottmv{j}   }'}   :=   \overrightarrow{  { \ottnt{T} }_{   1     \ottmv{j}   }'  }   ]  }  \ottnt{U'}  :  \ottnt{f'}  \dashv  \Gamma'  \rangle 
 \end{equation*}
 By Lemma~\ref{lem:term_prec_push_subst},
 \begin{equation*}
   \langle  \Gamma  [   \overrightarrow{ \ottmv{X_{\ottmv{i}}} }   \ottsym{:=}   \overrightarrow{ \ottnt{T_{\ottmv{i}}} }   ]   \vdash   \ottnt{f}  [   \overrightarrow{ \ottmv{X_{\ottmv{i}}} }   \ottsym{:=}   \overrightarrow{ \ottnt{T_{\ottmv{i}}} }   ]  :  \ottnt{U}  [   \overrightarrow{ \ottmv{X_{\ottmv{i}}} }   \ottsym{:=}   \overrightarrow{ \ottnt{T_{\ottmv{i}}} }   ]   \sqsubseteq _{  [   \overrightarrow{ \ottmv{X} _{   2     \ottmv{j}   }'}   :=   \overrightarrow{  { \ottnt{T} }_{   2     \ottmv{j}   }''   [   \overrightarrow{ \ottmv{X_{\ottmv{i}}} }   \ottsym{:=}   \overrightarrow{ \ottnt{T_{\ottmv{i}}} }   ] }   ]  \uplus  S  }  \ottnt{U''}  :  \ottnt{f''}  \dashv  \Gamma''  \rangle 
  \end{equation*}
  where $\ottnt{U''} = \ottnt{U'}  [   \overrightarrow{ \ottmv{X} _{   1     \ottmv{j}   }'}   \ottsym{:=}   \overrightarrow{  { \ottnt{T} }_{   1     \ottmv{j}   }'  }   ]$ and $\ottnt{f''} = \ottnt{f'}  [   \overrightarrow{ \ottmv{X} _{   1     \ottmv{j}   }'}   \ottsym{:=}   \overrightarrow{  { \ottnt{T} }_{   1     \ottmv{j}   }'  }   ]$ and $\Gamma'' = \Gamma'  [   \overrightarrow{ \ottmv{X} _{   1     \ottmv{j}   }'}   \ottsym{:=}   \overrightarrow{  { \ottnt{T} }_{   1     \ottmv{j}   }'  }   ]$.
 Since $ \overrightarrow{  { \ottnt{T} }_{   2     \ottmv{j}   }'  } $ are type variables,
 \begin{equation*}
   \langle  \Gamma  [   \overrightarrow{ \ottmv{X_{\ottmv{i}}} }   \ottsym{:=}   \overrightarrow{ \ottnt{T_{\ottmv{i}}} }   ]   \vdash   \ottnt{f}  [   \overrightarrow{ \ottmv{X_{\ottmv{i}}} }   \ottsym{:=}   \overrightarrow{ \ottnt{T_{\ottmv{i}}} }   ]  :  \ottnt{U}  [   \overrightarrow{ \ottmv{X_{\ottmv{i}}} }   \ottsym{:=}   \overrightarrow{ \ottnt{T_{\ottmv{i}}} }   ]   \sqsubseteq _{  S  \uplus  S'  }  \ottnt{U'}  [   \overrightarrow{ \ottmv{X'_{\ottmv{j}}} }   \ottsym{:=}   \overrightarrow{ \ottnt{T'_{\ottmv{j}}} }   ]  :  \ottnt{f'}  [   \overrightarrow{ \ottmv{X'_{\ottmv{j}}} }   \ottsym{:=}   \overrightarrow{ \ottnt{T'_{\ottmv{j}}} }   ]  \dashv  \Gamma'  [   \overrightarrow{ \ottmv{X'_{\ottmv{j}}} }   \ottsym{:=}   \overrightarrow{ \ottnt{T'_{\ottmv{j}}} }   ]  \rangle 
 \end{equation*}
 where $S'$ is the type substitution described in the statement.
\end{proof}

\begin{lemmaA} \label{lem:prec_subst_weak} \noindent
 Let type variables in $\textit{dom} \, \ottsym{(}  S_{{\mathrm{2}}}  \ottsym{)}$ be fresh.
 \begin{enumerate}
  \item If $ \ottnt{U}   \sqsubseteq _{ S_{{\mathrm{1}}} }  \ottnt{U'} $,
        then $ \ottnt{U}   \sqsubseteq _{  S_{{\mathrm{1}}}  \uplus  S_{{\mathrm{2}}}  }  \ottnt{U'} $.
  \item If $ \langle  \Gamma   \vdash   \ottnt{f}  :  \ottnt{U}   \sqsubseteq _{ S_{{\mathrm{1}}} }  \ottnt{U'}  :  \ottnt{f'}  \dashv  \Gamma'  \rangle $,
        then $ \langle  \Gamma   \vdash   \ottnt{f}  :  \ottnt{U}   \sqsubseteq _{  S_{{\mathrm{1}}}  \uplus  S_{{\mathrm{2}}}  }  \ottnt{U'}  :  \ottnt{f'}  \dashv  \Gamma'  \rangle $.
 \end{enumerate}
\end{lemmaA}
\begin{proof}
 By induction on the derivations of
 $ \ottnt{U}   \sqsubseteq _{ S_{{\mathrm{1}}} }  \ottnt{U'} $ and $ \langle  \Gamma   \vdash   \ottnt{f}  :  \ottnt{U}   \sqsubseteq _{ S_{{\mathrm{1}}} }  \ottnt{U'}  :  \ottnt{f'}  \dashv  \Gamma'  \rangle $.xo
\end{proof}

\begin{lemmaA} \label{lem:subst_preserve_term_prec}
  If
  \begin{itemize}
   \item $ \langle   \Gamma ,   \ottmv{x}  :  \forall \,  \overrightarrow{ \ottmv{X_{\ottmv{i}}} }   \ottsym{.}  \ottnt{U_{{\mathrm{1}}}}     \vdash   \ottnt{f}  :  \ottnt{U}   \sqsubseteq _{ S }  \ottnt{U'}  :  \ottnt{f'}  \dashv   \Gamma' ,   \ottmv{x}  :  \forall \,  \overrightarrow{ \ottmv{X'_{\ottmv{j}}} }   \ottsym{.}  \ottnt{U'_{{\mathrm{1}}}}    \rangle $ and
   \item $ \langle  \Gamma   \vdash   w  :  \ottnt{U_{{\mathrm{1}}}}   \sqsubseteq _{  [   \overrightarrow{ \ottmv{X'_{\ottmv{j}}} }   :=   \overrightarrow{ \ottnt{T''_{\ottmv{j}}} }   ]  \uplus  S  }  \ottnt{U'_{{\mathrm{1}}}}  :  w'  \dashv  \Gamma'  \rangle $ and
   \item $ \overrightarrow{ \ottmv{X_{\ottmv{i}}} } $ and $ \overrightarrow{ \ottmv{X'_{\ottmv{j}}} } $ does not occur free in the derivation of
         $ \langle   \Gamma ,   \ottmv{x}  :  \forall \,  \overrightarrow{ \ottmv{X_{\ottmv{i}}} }   \ottsym{.}  \ottnt{U_{{\mathrm{1}}}}     \vdash   \ottnt{f}  :  \ottnt{U}   \sqsubseteq _{ S }  \ottnt{U'}  :  \ottnt{f'}  \dashv   \Gamma' ,   \ottmv{x}  :  \forall \,  \overrightarrow{ \ottmv{X'_{\ottmv{j}}} }   \ottsym{.}  \ottnt{U'_{{\mathrm{1}}}}    \rangle $,
  \end{itemize}
  then there exist $S'$ such that
  \begin{itemize}
   \item $\textit{dom} \, \ottsym{(}  S'  \ottsym{)}$ is a set of fresh type variables and
   \item $ \langle  \Gamma   \vdash   \ottnt{f}  [  \ottmv{x}  \ottsym{:=}   \Lambda    \overrightarrow{ \ottmv{X_{\ottmv{i}}} }  .\,  w   ]  :  \ottnt{U}   \sqsubseteq _{  S  \uplus  S'  }  \ottnt{U'}  :  \ottnt{f'}  [  \ottmv{x}  \ottsym{:=}   \Lambda    \overrightarrow{ \ottmv{X'_{\ottmv{j}}} }  .\,  w'   ]  \dashv  \Gamma'  \rangle $.
  \end{itemize}
\end{lemmaA}

\begin{proof}
  By induction on $ \langle   \Gamma ,   \ottmv{x}  :  \forall \,  \overrightarrow{ \ottmv{X_{\ottmv{i}}} }   \ottsym{.}  \ottnt{U_{{\mathrm{1}}}}     \vdash   \ottnt{f}  :  \ottnt{U}   \sqsubseteq _{ S }  \ottnt{U'}  :  \ottnt{f'}  \dashv   \Gamma' ,   \ottmv{x}  :  \forall \,  \overrightarrow{ \ottmv{X'_{\ottmv{j}}} }   \ottsym{.}  \ottnt{U'_{{\mathrm{1}}}}    \rangle $.

  \begin{caseanalysis}
    \case{\rnp{P\_VarP}}
    We are given
    $ \langle   \Gamma ,   \ottmv{x}  :  \forall \,  \overrightarrow{ \ottmv{X_{\ottmv{i}}} }   \ottsym{.}  \ottnt{U_{{\mathrm{1}}}}     \vdash   \ottmv{x'}  [   \overrightarrow{ \mathbbsl{T}_{\ottmv{i}} }   ]  :  \ottnt{U}   \sqsubseteq _{ S }  \ottnt{U'}  :  \ottmv{x'}  [   \overrightarrow{  { \mathbbsl{T}_{\ottmv{j}} }'  }   ]  \dashv   \Gamma' ,   \ottmv{x}  :  \forall \,  \overrightarrow{ \ottmv{X'_{\ottmv{j}}} }   \ottsym{.}  \ottnt{U'_{{\mathrm{1}}}}    \rangle $
    for some $\ottmv{x'}$, $ \overrightarrow{ \mathbbsl{T}_{\ottmv{i}} } $ and $ \overrightarrow{  { \mathbbsl{T}_{\ottmv{j}} }'  } $.
    If $\ottmv{x}  \neq  \ottmv{x'}$, then obvious.
    Otherwise, we suppose that $\ottmv{x}  \ottsym{=}  \ottmv{x'}$.
    Furthermore, $ \nu $ in $ \overrightarrow{ \mathbbsl{T}_{\ottmv{i}} } $ and $ \overrightarrow{  { \mathbbsl{T}_{\ottmv{j}} }'  } $ generate fresh
    type variables.
    We write $ \overrightarrow{ \ottnt{T_{\ottmv{i}}} } $ (resp.\ $ \overrightarrow{ \ottnt{T'_{\ottmv{j}}} } $) for types which are the same as
    $ \overrightarrow{ \mathbbsl{T}_{\ottmv{i}} } $ (resp.\ $ \overrightarrow{  { \mathbbsl{T}_{\ottmv{j}} }'  } $) except that the occurrences of
    $ \nu $ are replaced with the fresh type variables.
    Note that $\ottnt{U} = \ottnt{U_{{\mathrm{1}}}}  [   \overrightarrow{ \ottmv{X_{\ottmv{i}}} }   \ottsym{:=}   \overrightarrow{ \mathbbsl{T}_{\ottmv{i}} }   ] = \ottnt{U_{{\mathrm{1}}}}  [   \overrightarrow{ \ottmv{X_{\ottmv{i}}} }   \ottsym{:=}   \overrightarrow{ \ottnt{T_{\ottmv{i}}} }   ]$ and
    $\ottnt{U'} = \ottnt{U'_{{\mathrm{1}}}}  [   \overrightarrow{ \ottmv{X'_{\ottmv{j}}} }   \ottsym{:=}   \overrightarrow{  { \mathbbsl{T}_{\ottmv{j}} }'  }   ] = \ottnt{U'_{{\mathrm{1}}}}  [   \overrightarrow{ \ottmv{X'_{\ottmv{j}}} }   \ottsym{:=}   \overrightarrow{ \ottnt{T'_{\ottmv{j}}} }   ]$.
    By inversion,
    $ \ottnt{U_{{\mathrm{1}}}}  [   \overrightarrow{ \ottmv{X_{\ottmv{i}}} }   \ottsym{:=}   \overrightarrow{ \ottnt{T_{\ottmv{i}}} }   ]   \sqsubseteq _{ S }  \ottnt{U'_{{\mathrm{1}}}}  [   \overrightarrow{ \ottmv{X'_{\ottmv{j}}} }   \ottsym{:=}   \overrightarrow{ \ottnt{T'_{\ottmv{j}}} }   ] $.
    By Lemma~\ref{lem:type_app_preserve_term_prec},
    there exist $S'$ such that
    $\textit{dom} \, \ottsym{(}  S'  \ottsym{)}$ is the same as fresh type variables generated by
    $ \nu $ in $ \overrightarrow{  { \mathbbsl{T}_{\ottmv{j}} }'  } $ and
    \[
      \langle  \Gamma  [   \overrightarrow{ \ottmv{X_{\ottmv{i}}} }   \ottsym{:=}   \overrightarrow{ \ottnt{T_{\ottmv{i}}} }   ]   \vdash   w  [   \overrightarrow{ \ottmv{X_{\ottmv{i}}} }   \ottsym{:=}   \overrightarrow{ \ottnt{T_{\ottmv{i}}} }   ]  :  \ottnt{U_{{\mathrm{1}}}}  [   \overrightarrow{ \ottmv{X_{\ottmv{i}}} }   \ottsym{:=}   \overrightarrow{ \ottnt{T_{\ottmv{i}}} }   ]   \sqsubseteq _{  S  \uplus  S'  }  \ottnt{U'_{{\mathrm{1}}}}  [   \overrightarrow{ \ottmv{X'_{\ottmv{j}}} }   \ottsym{:=}   \overrightarrow{ \ottnt{T'_{\ottmv{j}}} }   ]  :  w'  [   \overrightarrow{ \ottmv{X'_{\ottmv{j}}} }   \ottsym{:=}   \overrightarrow{ \ottnt{T'_{\ottmv{j}}} }   ]  \dashv  \Gamma'  [   \overrightarrow{ \ottmv{X'_{\ottmv{j}}} }   \ottsym{:=}   \overrightarrow{ \ottnt{T'_{\ottmv{j}}} }   ]  \rangle .
    \]
    Since $ \overrightarrow{ \ottmv{X_{\ottmv{i}}} }   \cap  \textit{ftv} \, \ottsym{(}  \Gamma  \ottsym{)}  \ottsym{=}   \emptyset $ and $ \overrightarrow{ \ottmv{X'_{\ottmv{j}}} }   \cap  \textit{ftv} \, \ottsym{(}  \Gamma'  \ottsym{)}  \ottsym{=}   \emptyset $,
    we have
    \[
      \langle  \Gamma   \vdash   w  [   \overrightarrow{ \ottmv{X_{\ottmv{i}}} }   \ottsym{:=}   \overrightarrow{ \ottnt{T_{\ottmv{i}}} }   ]  :  \ottnt{U_{{\mathrm{1}}}}  [   \overrightarrow{ \ottmv{X_{\ottmv{i}}} }   \ottsym{:=}   \overrightarrow{ \ottnt{T_{\ottmv{i}}} }   ]   \sqsubseteq _{  S  \uplus  S'  }  \ottnt{U'_{{\mathrm{1}}}}  [   \overrightarrow{ \ottmv{X'_{\ottmv{j}}} }   \ottsym{:=}   \overrightarrow{ \ottnt{T'_{\ottmv{j}}} }   ]  :  w'  [   \overrightarrow{ \ottmv{X'_{\ottmv{j}}} }   \ottsym{:=}   \overrightarrow{ \ottnt{T'_{\ottmv{j}}} }   ]  \dashv  \Gamma'  \rangle ,
    \]
    which is what we want to show.

    \case{\rnp{P\_Blame}} Obvious by Lemmas~\ref{lem:term_prec_to_typing} and
    \ref{lem:term_var_substitution}.

    \otherwise  By the IH(s) and Lemma~\ref{lem:prec_subst_weak}.
    \qedhere
  \end{caseanalysis}
\end{proof}

\begin{lemmaA} \label{lem:eval_subterm}
 If $f_{{\mathrm{1}}} \,  \xmapsto{ \mathmakebox[0.4em]{} S \mathmakebox[0.3em]{} }\hspace{-0.4em}{}^\ast \hspace{0.2em}  \, f_{{\mathrm{2}}}$ where \rnp{E\_Abort} is not applied,
 then $\ottnt{E}  [  f_{{\mathrm{1}}}  ] \,  \xmapsto{ \mathmakebox[0.4em]{} S \mathmakebox[0.3em]{} }\hspace{-0.4em}{}^\ast \hspace{0.2em}  \, S  \ottsym{(}  \ottnt{E}  \ottsym{)}  [  f_{{\mathrm{2}}}  ]$.
\end{lemmaA}
\begin{proof}
 By induction on the number of steps of $f_{{\mathrm{1}}} \,  \xmapsto{ \mathmakebox[0.4em]{} S \mathmakebox[0.3em]{} }\hspace{-0.4em}{}^\ast \hspace{0.2em}  \, f_{{\mathrm{2}}}$.
 If the number of steps is zero, then $f_{{\mathrm{1}}}  \ottsym{=}  f_{{\mathrm{2}}}$ and $S  \ottsym{=}  [  ]$.
 Thus, $\ottnt{E}  [  f_{{\mathrm{1}}}  ]  \ottsym{=}  S  \ottsym{(}  \ottnt{E}  \ottsym{)}  [  f_{{\mathrm{2}}}  ]$, and $\ottnt{E}  [  f_{{\mathrm{1}}}  ] \,  \xmapsto{ \mathmakebox[0.4em]{} [  ] \mathmakebox[0.3em]{} }\hspace{-0.4em}{}^\ast \hspace{0.2em}  \, \ottnt{E}  [  f_{{\mathrm{1}}}  ]$.
 If the number of steps is more than zero, there exist $\ottnt{E'}$, $f'_{{\mathrm{1}}}$,
 $S_{{\mathrm{1}}}$, $S_{{\mathrm{2}}}$, and $f'$ such that
 \begin{itemize}
  \item $f_{{\mathrm{1}}}  \ottsym{=}  \ottnt{E'}  [  f'_{{\mathrm{1}}}  ]$,
  \item $f'_{{\mathrm{1}}} \,  \xrightarrow{ \mathmakebox[0.4em]{} S_{{\mathrm{1}}} \mathmakebox[0.3em]{} }  \, f'$,
  \item $\ottnt{E'}  [  f'_{{\mathrm{1}}}  ] \,  \xmapsto{ \mathmakebox[0.4em]{} S_{{\mathrm{1}}} \mathmakebox[0.3em]{} }  \, S_{{\mathrm{1}}}  \ottsym{(}  \ottnt{E'}  [  f'  ]  \ottsym{)}$,
  \item $S_{{\mathrm{1}}}  \ottsym{(}  \ottnt{E'}  [  f'  ]  \ottsym{)} \,  \xmapsto{ \mathmakebox[0.4em]{} S_{{\mathrm{2}}} \mathmakebox[0.3em]{} }\hspace{-0.4em}{}^\ast \hspace{0.2em}  \, f_{{\mathrm{2}}}$, and
  \item $S  \ottsym{=}   S_{{\mathrm{2}}}  \circ  S_{{\mathrm{1}}} $.
 \end{itemize}
 By \rnp{E\_Step}, $\ottnt{E}  [  \ottnt{E'}  [  f'_{{\mathrm{1}}}  ]  ] \,  \xmapsto{ \mathmakebox[0.4em]{} S_{{\mathrm{1}}} \mathmakebox[0.3em]{} }  \, S_{{\mathrm{1}}}  \ottsym{(}  \ottnt{E}  [  \ottnt{E'}  [  f'  ]  ]  \ottsym{)}$.
 By the IH, $S_{{\mathrm{1}}}  \ottsym{(}  \ottnt{E}  [  \ottnt{E'}  [  f'  ]  ]  \ottsym{)} \,  \xmapsto{ \mathmakebox[0.4em]{} S_{{\mathrm{2}}} \mathmakebox[0.3em]{} }\hspace{-0.4em}{}^\ast \hspace{0.2em}  \, S_{{\mathrm{2}}}  \ottsym{(}  S_{{\mathrm{1}}}  \ottsym{(}  \ottnt{E}  \ottsym{)}  \ottsym{)}  [  f_{{\mathrm{2}}}  ] = S  \ottsym{(}  \ottnt{E}  \ottsym{)}  [  f_{{\mathrm{2}}}  ]$.
 Thus, $\ottnt{E}  [  \ottnt{E'}  [  f'_{{\mathrm{1}}}  ]  ] \,  \xmapsto{ \mathmakebox[0.4em]{} S \mathmakebox[0.3em]{} }\hspace{-0.4em}{}^\ast \hspace{0.2em}  \, S  \ottsym{(}  \ottnt{E}  \ottsym{)}  [  f_{{\mathrm{2}}}  ]$.
\end{proof}

\begin{lemmaA}[Simulation of Function Application] \label{lem:term_prec_simulation_app}
  If $ \langle   \emptyset    \vdash    \lambda  \ottmv{x} \!:\!  \ottnt{U_{{\mathrm{1}}}}  .\,  \ottnt{f_{{\mathrm{1}}}}   :  \ottnt{U_{{\mathrm{1}}}}  \!\rightarrow\!  \ottnt{U_{{\mathrm{2}}}}   \sqsubseteq _{ S_{{\mathrm{0}}} }  \ottnt{U'_{{\mathrm{1}}}}  \!\rightarrow\!  \ottnt{U'_{{\mathrm{2}}}}  :  w'_{{\mathrm{1}}}  \dashv   \emptyset   \rangle $
  and $ \langle   \emptyset    \vdash   w_{{\mathrm{2}}}  :  \ottnt{U_{{\mathrm{1}}}}   \sqsubseteq _{ S_{{\mathrm{0}}} }  \ottnt{U'_{{\mathrm{1}}}}  :  w'_{{\mathrm{2}}}  \dashv   \emptyset   \rangle $,
  then there exists $S'_{{\mathrm{0}}}$, $S'$, $S''$, and $f'$ such that
  \begin{itemize}
   \item $w'_{{\mathrm{1}}} \, w'_{{\mathrm{2}}} \,  \xmapsto{ \mathmakebox[0.4em]{} S' \mathmakebox[0.3em]{} }\hspace{-0.4em}{}^+ \hspace{0.2em}  \, f'$ where \rnp{E\_Abort} is not applied,
   \item $ \langle   \emptyset    \vdash   \ottnt{f_{{\mathrm{1}}}}  [  \ottmv{x}  \ottsym{:=}  w_{{\mathrm{2}}}  ]  :  \ottnt{U_{{\mathrm{2}}}}   \sqsubseteq _{  S''  \uplus  \ottsym{(}   S'_{{\mathrm{0}}}  \circ  S_{{\mathrm{0}}}   \ottsym{)}  }  S'  \ottsym{(}  \ottnt{U'_{{\mathrm{2}}}}  \ottsym{)}  :  \ottnt{f'}  \dashv   \emptyset   \rangle $,
   \item $\forall \ottmv{X} \in \textit{dom} \, \ottsym{(}  S_{{\mathrm{0}}}  \ottsym{)}. S_{{\mathrm{0}}}  \ottsym{(}  \ottmv{X}  \ottsym{)}  \ottsym{=}   S'_{{\mathrm{0}}}  \circ   S_{{\mathrm{0}}}  \circ  S'    \ottsym{(}  \ottmv{X}  \ottsym{)}$,
   \item type variables in $\textit{dom} \, \ottsym{(}  S''  \ottsym{)}$ are fresh, and
   \item $\forall \ottmv{X} \in \textit{dom} \, \ottsym{(}  S'  \ottsym{)}. \textit{ftv} \, \ottsym{(}  S'  \ottsym{(}  \ottmv{X}  \ottsym{)}  \ottsym{)}  \subseteq  \textit{dom} \, \ottsym{(}  S'_{{\mathrm{0}}}  \ottsym{)}$.
  \end{itemize}
\end{lemmaA}

\begin{proof}
  By induction on the number of casts that occur in $w'_{{\mathrm{1}}}$.
  By case analysis on the precision rule applied last to derive
  $ \langle   \emptyset    \vdash    \lambda  \ottmv{x} \!:\!  \ottnt{U_{{\mathrm{1}}}}  .\,  \ottnt{f_{{\mathrm{1}}}}   :  \ottnt{U_{{\mathrm{1}}}}  \!\rightarrow\!  \ottnt{U_{{\mathrm{2}}}}   \sqsubseteq _{ S_{{\mathrm{0}}} }  \ottnt{U'_{{\mathrm{1}}}}  \!\rightarrow\!  \ottnt{U'_{{\mathrm{2}}}}  :  w'_{{\mathrm{1}}}  \dashv   \emptyset   \rangle $.

  \begin{caseanalysis}
    \case{\rnp{P\_Abs}}
    We are given
    $ \langle   \emptyset    \vdash    \lambda  \ottmv{x} \!:\!  \ottnt{U_{{\mathrm{1}}}}  .\,  \ottnt{f_{{\mathrm{1}}}}   :  \ottnt{U_{{\mathrm{1}}}}  \!\rightarrow\!  \ottnt{U_{{\mathrm{2}}}}   \sqsubseteq _{ S_{{\mathrm{0}}} }  \ottnt{U'_{{\mathrm{1}}}}  \!\rightarrow\!  \ottnt{U'_{{\mathrm{2}}}}  :   \lambda  \ottmv{x} \!:\!  \ottnt{U'_{{\mathrm{1}}}}  .\,  \ottnt{f'_{{\mathrm{1}}}}   \dashv   \emptyset   \rangle $
    and $w'_{{\mathrm{1}}}  \ottsym{=}   \lambda  \ottmv{x} \!:\!  \ottnt{U'_{{\mathrm{1}}}}  .\,  \ottnt{f'_{{\mathrm{1}}}} $ for some $\ottnt{f'_{{\mathrm{1}}}}$.
    By inversion,
    $ \langle   \ottmv{x}  :  \ottnt{U_{{\mathrm{1}}}}    \vdash   \ottnt{f_{{\mathrm{1}}}}  :  \ottnt{U_{{\mathrm{2}}}}   \sqsubseteq _{ S_{{\mathrm{0}}} }  \ottnt{U'_{{\mathrm{2}}}}  :  \ottnt{f'_{{\mathrm{1}}}}  \dashv   \ottmv{x}  :  \ottnt{U'_{{\mathrm{1}}}}   \rangle $
    and $ \ottnt{U_{{\mathrm{1}}}}   \sqsubseteq _{ S_{{\mathrm{0}}} }  \ottnt{U'_{{\mathrm{1}}}} $.
    By \rnp{R\_Beta}, $\ottsym{(}   \lambda  \ottmv{x} \!:\!  \ottnt{U'_{{\mathrm{1}}}}  .\,  \ottnt{f'_{{\mathrm{1}}}}   \ottsym{)} \, w'_{{\mathrm{2}}} \,  \xmapsto{ \mathmakebox[0.4em]{} [  ] \mathmakebox[0.3em]{} }  \, \ottnt{f'_{{\mathrm{1}}}}  [  \ottmv{x}  \ottsym{:=}  w'_{{\mathrm{2}}}  ]$.
    By Lemma~\ref{lem:subst_preserve_term_prec},
    $ \langle   \emptyset    \vdash   \ottnt{f_{{\mathrm{1}}}}  [  \ottmv{x}  \ottsym{:=}  w_{{\mathrm{2}}}  ]  :  \ottnt{U_{{\mathrm{2}}}}   \sqsubseteq _{ S_{{\mathrm{0}}} }  \ottnt{U'_{{\mathrm{2}}}}  :  \ottnt{f'_{{\mathrm{1}}}}  [  \ottmv{x}  \ottsym{:=}  w'_{{\mathrm{2}}}  ]  \dashv   \emptyset   \rangle $.
    We finish by letting $S'_{{\mathrm{0}}}$, $S'$, and $S''$ be empty.

    \case{\rnp{P\_CastR}}
    By Lemma \ref{lem:canonical_forms},
    there exist $w'_{{\mathrm{11}}}$, $\ottnt{U'_{{\mathrm{11}}}}$, $\ottnt{U'_{{\mathrm{12}}}}$, and $\ell'$ such that
    $ \langle   \emptyset    \vdash    \lambda  \ottmv{x} \!:\!  \ottnt{U_{{\mathrm{1}}}}  .\,  \ottnt{f_{{\mathrm{1}}}}   :  \ottnt{U_{{\mathrm{1}}}}  \!\rightarrow\!  \ottnt{U_{{\mathrm{2}}}}   \sqsubseteq _{ S_{{\mathrm{0}}} }  \ottnt{U'_{{\mathrm{1}}}}  \!\rightarrow\!  \ottnt{U'_{{\mathrm{2}}}}  :  w'_{{\mathrm{11}}}  \ottsym{:}   \ottnt{U'_{{\mathrm{11}}}}  \!\rightarrow\!  \ottnt{U'_{{\mathrm{12}}}} \Rightarrow  \unskip ^ { \ell' }  \! \ottnt{U'_{{\mathrm{1}}}}  \!\rightarrow\!  \ottnt{U'_{{\mathrm{2}}}}   \dashv   \emptyset   \rangle $
    and $w'_{{\mathrm{1}}}  \ottsym{=}  w'_{{\mathrm{11}}}  \ottsym{:}   \ottnt{U'_{{\mathrm{11}}}}  \!\rightarrow\!  \ottnt{U'_{{\mathrm{12}}}} \Rightarrow  \unskip ^ { \ell' }  \! \ottnt{U'_{{\mathrm{1}}}}  \!\rightarrow\!  \ottnt{U'_{{\mathrm{2}}}} $.

    By inversion,
    $ \langle   \emptyset    \vdash    \lambda  \ottmv{x} \!:\!  \ottnt{U_{{\mathrm{1}}}}  .\,  \ottnt{f_{{\mathrm{1}}}}   :  \ottnt{U_{{\mathrm{1}}}}  \!\rightarrow\!  \ottnt{U_{{\mathrm{2}}}}   \sqsubseteq _{ S_{{\mathrm{0}}} }  \ottnt{U'_{{\mathrm{11}}}}  \!\rightarrow\!  \ottnt{U'_{{\mathrm{12}}}}  :  w'_{{\mathrm{11}}}  \dashv   \emptyset   \rangle $
    and $ \ottnt{U_{{\mathrm{1}}}}  \!\rightarrow\!  \ottnt{U_{{\mathrm{2}}}}   \sqsubseteq _{ S_{{\mathrm{0}}} }  \ottnt{U'_{{\mathrm{1}}}}  \!\rightarrow\!  \ottnt{U'_{{\mathrm{2}}}} $.
    By inversion of $ \ottnt{U_{{\mathrm{1}}}}  \!\rightarrow\!  \ottnt{U_{{\mathrm{2}}}}   \sqsubseteq _{ S_{{\mathrm{0}}} }  \ottnt{U'_{{\mathrm{1}}}}  \!\rightarrow\!  \ottnt{U'_{{\mathrm{2}}}} $, we have
    $ \ottnt{U_{{\mathrm{1}}}}   \sqsubseteq _{ S_{{\mathrm{0}}} }  \ottnt{U'_{{\mathrm{1}}}} $ and $ \ottnt{U_{{\mathrm{2}}}}   \sqsubseteq _{ S_{{\mathrm{0}}} }  \ottnt{U'_{{\mathrm{2}}}} $.
    By Lemma~\ref{lem:term_prec_to_type_prec},
    $ \ottnt{U_{{\mathrm{1}}}}  \!\rightarrow\!  \ottnt{U_{{\mathrm{2}}}}   \sqsubseteq _{ S_{{\mathrm{0}}} }  \ottnt{U'_{{\mathrm{11}}}}  \!\rightarrow\!  \ottnt{U'_{{\mathrm{12}}}} $.
    By its inversion,
    $ \ottnt{U_{{\mathrm{1}}}}   \sqsubseteq _{ S_{{\mathrm{0}}} }  \ottnt{U'_{{\mathrm{11}}}} $ and $ \ottnt{U_{{\mathrm{2}}}}   \sqsubseteq _{ S_{{\mathrm{0}}} }  \ottnt{U'_{{\mathrm{12}}}} $.

    By \rnp{R\_AppCast},
    $\ottsym{(}  w'_{{\mathrm{11}}}  \ottsym{:}   \ottnt{U'_{{\mathrm{11}}}}  \!\rightarrow\!  \ottnt{U'_{{\mathrm{12}}}} \Rightarrow  \unskip ^ { \ell' }  \! \ottnt{U'_{{\mathrm{1}}}}  \!\rightarrow\!  \ottnt{U'_{{\mathrm{2}}}}   \ottsym{)} \, w'_{{\mathrm{2}}} \,  \xmapsto{ \mathmakebox[0.4em]{} [  ] \mathmakebox[0.3em]{} }  \, \ottsym{(}  w'_{{\mathrm{11}}} \, \ottsym{(}  w'_{{\mathrm{2}}}  \ottsym{:}   \ottnt{U'_{{\mathrm{1}}}} \Rightarrow  \unskip ^ {  \bar{ \ell' }  }  \! \ottnt{U'_{{\mathrm{11}}}}   \ottsym{)}  \ottsym{)}  \ottsym{:}   \ottnt{U'_{{\mathrm{12}}}} \Rightarrow  \unskip ^ { \ell' }  \! \ottnt{U'_{{\mathrm{2}}}} $.

    By \rnp{P\_CastR},
    $ \langle   \emptyset    \vdash   w_{{\mathrm{2}}}  :  \ottnt{U_{{\mathrm{1}}}}   \sqsubseteq _{ S_{{\mathrm{0}}} }  \ottnt{U'_{{\mathrm{11}}}}  :  w'_{{\mathrm{2}}}  \ottsym{:}   \ottnt{U'_{{\mathrm{1}}}} \Rightarrow  \unskip ^ {  \bar{ \ell' }  }  \! \ottnt{U'_{{\mathrm{11}}}}   \dashv   \emptyset   \rangle $.

    By Lemma~\ref{lem:prec_catch_up_left_value},
    there exist $S'_{{\mathrm{0}}}$, $S'_{{\mathrm{2}}}$, and $w''_{{\mathrm{2}}}$ such that
    \begin{itemize}
     \item $w'_{{\mathrm{2}}}  \ottsym{:}   \ottnt{U'_{{\mathrm{1}}}} \Rightarrow  \unskip ^ {  \bar{ \ell' }  }  \! \ottnt{U'_{{\mathrm{11}}}}  \,  \xmapsto{ \mathmakebox[0.4em]{} S'_{{\mathrm{2}}} \mathmakebox[0.3em]{} }\hspace{-0.4em}{}^\ast \hspace{0.2em}  \, w''_{{\mathrm{2}}}$,
     \item $ \langle   \emptyset    \vdash   w_{{\mathrm{2}}}  :  \ottnt{U_{{\mathrm{1}}}}   \sqsubseteq _{  S'_{{\mathrm{0}}}  \circ  S_{{\mathrm{0}}}  }  S'_{{\mathrm{2}}}  \ottsym{(}  \ottnt{U'_{{\mathrm{11}}}}  \ottsym{)}  :  w''_{{\mathrm{2}}}  \dashv   \emptyset   \rangle $,
     \item $\forall \ottmv{X} \in \textit{dom} \, \ottsym{(}  S_{{\mathrm{0}}}  \ottsym{)}. S_{{\mathrm{0}}}  \ottsym{(}  \ottmv{X}  \ottsym{)}  \ottsym{=}   S'_{{\mathrm{0}}}  \circ   S_{{\mathrm{0}}}  \circ  S'_{{\mathrm{2}}}    \ottsym{(}  \ottmv{X}  \ottsym{)}$, and
     \item $\forall \ottmv{X} \in \textit{dom} \, \ottsym{(}  S'_{{\mathrm{2}}}  \ottsym{)}. \textit{ftv} \, \ottsym{(}  S'_{{\mathrm{2}}}  \ottsym{(}  \ottmv{X}  \ottsym{)}  \ottsym{)}  \subseteq  \textit{dom} \, \ottsym{(}  S'_{{\mathrm{0}}}  \ottsym{)}$.
    \end{itemize}
    So,
    $\ottsym{(}  w'_{{\mathrm{11}}} \, \ottsym{(}  w'_{{\mathrm{2}}}  \ottsym{:}   \ottnt{U'_{{\mathrm{1}}}} \Rightarrow  \unskip ^ {  \bar{ \ell' }  }  \! \ottnt{U'_{{\mathrm{11}}}}   \ottsym{)}  \ottsym{)}  \ottsym{:}   \ottnt{U'_{{\mathrm{12}}}} \Rightarrow  \unskip ^ { \ell' }  \! \ottnt{U'_{{\mathrm{2}}}}  \,  \xmapsto{ \mathmakebox[0.4em]{} S'_{{\mathrm{2}}} \mathmakebox[0.3em]{} }\hspace{-0.4em}{}^\ast \hspace{0.2em}  \, \ottsym{(}  S'_{{\mathrm{2}}}  \ottsym{(}  w'_{{\mathrm{11}}}  \ottsym{)} \, w''_{{\mathrm{2}}}  \ottsym{)}  \ottsym{:}   S'_{{\mathrm{2}}}  \ottsym{(}  \ottnt{U'_{{\mathrm{12}}}}  \ottsym{)} \Rightarrow  \unskip ^ { \ell' }  \! S'_{{\mathrm{2}}}  \ottsym{(}  \ottnt{U'_{{\mathrm{2}}}}  \ottsym{)} $
    by Lemma~\ref{lem:eval_subterm}.

    By Lemma \ref{lem:right_subst_preserve_prec},
    $ \ottnt{U_{{\mathrm{2}}}}   \sqsubseteq _{  S'_{{\mathrm{0}}}  \circ  S_{{\mathrm{0}}}  }  S'_{{\mathrm{2}}}  \ottsym{(}  \ottnt{U'_{{\mathrm{2}}}}  \ottsym{)} $ and
    $ \langle   \emptyset    \vdash    \lambda  \ottmv{x} \!:\!  \ottnt{U_{{\mathrm{1}}}}  .\,  \ottnt{f_{{\mathrm{1}}}}   :  \ottnt{U_{{\mathrm{1}}}}  \!\rightarrow\!  \ottnt{U_{{\mathrm{2}}}}   \sqsubseteq _{  S'_{{\mathrm{0}}}  \circ  S_{{\mathrm{0}}}  }  S'_{{\mathrm{2}}}  \ottsym{(}  \ottnt{U'_{{\mathrm{11}}}}  \!\rightarrow\!  \ottnt{U'_{{\mathrm{12}}}}  \ottsym{)}  :  S'_{{\mathrm{2}}}  \ottsym{(}  w'_{{\mathrm{11}}}  \ottsym{)}  \dashv   \emptyset   \rangle $.

    By the IH, there exist $S''_{{\mathrm{0}}}$, $S'_{{\mathrm{3}}}$, $S''$, and $\ottnt{f''}$ such that
    \begin{itemize}
     \item $S'_{{\mathrm{2}}}  \ottsym{(}  w'_{{\mathrm{11}}}  \ottsym{)} \, w''_{{\mathrm{2}}} \,  \xmapsto{ \mathmakebox[0.4em]{} S'_{{\mathrm{3}}} \mathmakebox[0.3em]{} }\hspace{-0.4em}{}^+ \hspace{0.2em}  \, \ottnt{f''}$ where \rnp{E\_Abort} is not applied,
     \item $ \langle   \emptyset    \vdash   \ottnt{f_{{\mathrm{1}}}}  [  \ottmv{x}  \ottsym{:=}  w_{{\mathrm{2}}}  ]  :  \ottnt{U_{{\mathrm{2}}}}   \sqsubseteq _{  S''  \uplus  \ottsym{(}   S''_{{\mathrm{0}}}  \circ   S'_{{\mathrm{0}}}  \circ  S_{{\mathrm{0}}}    \ottsym{)}  }   S'_{{\mathrm{3}}}  \circ  S'_{{\mathrm{2}}}   \ottsym{(}  \ottnt{U'_{{\mathrm{12}}}}  \ottsym{)}  :  \ottnt{f''}  \dashv   \emptyset   \rangle $,
     \item $\forall \ottmv{X} \in \textit{dom} \, \ottsym{(}   S'_{{\mathrm{0}}}  \circ  S_{{\mathrm{0}}}   \ottsym{)}.  S'_{{\mathrm{0}}}  \circ  S_{{\mathrm{0}}}   \ottsym{(}  \ottmv{X}  \ottsym{)}  \ottsym{=}   S''_{{\mathrm{0}}}  \circ    S'_{{\mathrm{0}}}  \circ  S_{{\mathrm{0}}}   \circ  S'_{{\mathrm{3}}}    \ottsym{(}  \ottmv{X}  \ottsym{)}$,
     \item type variables in $\textit{dom} \, \ottsym{(}  S''  \ottsym{)}$ are fresh, and
     \item $\forall \ottmv{X} \in \textit{dom} \, \ottsym{(}  S'_{{\mathrm{3}}}  \ottsym{)}. \textit{ftv} \, \ottsym{(}  S'_{{\mathrm{3}}}  \ottsym{(}  \ottmv{X}  \ottsym{)}  \ottsym{)}  \subseteq  \textit{dom} \, \ottsym{(}  S''_{{\mathrm{0}}}  \ottsym{)}$.
    \end{itemize}
    We have
    $\ottsym{(}  S'_{{\mathrm{2}}}  \ottsym{(}  w'_{{\mathrm{11}}}  \ottsym{)} \, w''_{{\mathrm{2}}}  \ottsym{)}  \ottsym{:}   S'_{{\mathrm{2}}}  \ottsym{(}  \ottnt{U'_{{\mathrm{12}}}}  \ottsym{)} \Rightarrow  \unskip ^ { \ell' }  \! S'_{{\mathrm{2}}}  \ottsym{(}  \ottnt{U'_{{\mathrm{2}}}}  \ottsym{)}  \,  \xmapsto{ \mathmakebox[0.4em]{} S'_{{\mathrm{3}}} \mathmakebox[0.3em]{} }\hspace{-0.4em}{}^+ \hspace{0.2em}  \, \ottnt{f''}  \ottsym{:}    S'_{{\mathrm{3}}}  \circ  S'_{{\mathrm{2}}}   \ottsym{(}  \ottnt{U'_{{\mathrm{12}}}}  \ottsym{)} \Rightarrow  \unskip ^ { \ell' }  \!  S'_{{\mathrm{3}}}  \circ  S'_{{\mathrm{2}}}   \ottsym{(}  \ottnt{U'_{{\mathrm{2}}}}  \ottsym{)} $
    by Lemma~\ref{lem:eval_subterm}.

    By Lemma~\ref{lem:right_subst_preserve_prec},
    $ \ottnt{U_{{\mathrm{2}}}}   \sqsubseteq _{   S''_{{\mathrm{0}}}  \circ  S'_{{\mathrm{0}}}   \circ  S_{{\mathrm{0}}}  }   S'_{{\mathrm{3}}}  \circ  S'_{{\mathrm{2}}}   \ottsym{(}  \ottnt{U'_{{\mathrm{2}}}}  \ottsym{)} $.
    By Lemma~\ref{lem:prec_subst_weak},
    $ \ottnt{U_{{\mathrm{2}}}}   \sqsubseteq _{  S''  \uplus  \ottsym{(}   S''_{{\mathrm{0}}}  \circ   S'_{{\mathrm{0}}}  \circ  S_{{\mathrm{0}}}    \ottsym{)}  }   S'_{{\mathrm{3}}}  \circ  S'_{{\mathrm{2}}}   \ottsym{(}  \ottnt{U'_{{\mathrm{2}}}}  \ottsym{)} $.
    Thus, by \rnp{P\_CastR},
    \[
      \langle   \emptyset    \vdash   \ottnt{f_{{\mathrm{1}}}}  [  \ottmv{x}  \ottsym{:=}  w_{{\mathrm{2}}}  ]  :  \ottnt{U_{{\mathrm{2}}}}   \sqsubseteq _{  S''  \uplus  \ottsym{(}   S''_{{\mathrm{0}}}  \circ   S'_{{\mathrm{0}}}  \circ  S_{{\mathrm{0}}}    \ottsym{)}  }   S'_{{\mathrm{3}}}  \circ  S'_{{\mathrm{2}}}   \ottsym{(}  \ottnt{U'_{{\mathrm{2}}}}  \ottsym{)}  :  \ottnt{f''}  \ottsym{:}    S'_{{\mathrm{3}}}  \circ  S'_{{\mathrm{2}}}   \ottsym{(}  \ottnt{U'_{{\mathrm{12}}}}  \ottsym{)} \Rightarrow  \unskip ^ { \ell' }  \!  S'_{{\mathrm{3}}}  \circ  S'_{{\mathrm{2}}}   \ottsym{(}  \ottnt{U'_{{\mathrm{2}}}}  \ottsym{)}   \dashv   \emptyset   \rangle .
    \]

    We show $\forall \ottmv{X} \in \textit{dom} \, \ottsym{(}  S_{{\mathrm{0}}}  \ottsym{)}. S_{{\mathrm{0}}}  \ottsym{(}  \ottmv{X}  \ottsym{)}  \ottsym{=}   S''_{{\mathrm{0}}}  \circ     S'_{{\mathrm{0}}}  \circ  S_{{\mathrm{0}}}   \circ  S'_{{\mathrm{3}}}   \circ  S'_{{\mathrm{2}}}    \ottsym{(}  \ottmv{X}  \ottsym{)}$.
    Let $\ottmv{X} \, \in \, \textit{dom} \, \ottsym{(}  S_{{\mathrm{0}}}  \ottsym{)}$.
    We have $S_{{\mathrm{0}}}  \ottsym{(}  \ottmv{X}  \ottsym{)}  \ottsym{=}   S'_{{\mathrm{0}}}  \circ   S_{{\mathrm{0}}}  \circ  S'_{{\mathrm{2}}}    \ottsym{(}  \ottmv{X}  \ottsym{)}$.
    Since $S'_{{\mathrm{2}}}$ is generated by DTI,
    free type variables in $S'_{{\mathrm{2}}}  \ottsym{(}  \ottmv{X}  \ottsym{)}$ are fresh for $S_{{\mathrm{0}}}$.
    Since $\forall \ottmv{X'}, S_{{\mathrm{0}}}  \ottsym{(}  \ottmv{X'}  \ottsym{)}  \ottsym{=}   S'_{{\mathrm{0}}}  \circ   S_{{\mathrm{0}}}  \circ  S'_{{\mathrm{2}}}    \ottsym{(}  \ottmv{X'}  \ottsym{)}$,
    it is found that $\textit{ftv} \, \ottsym{(}  S'_{{\mathrm{2}}}  \ottsym{(}  \ottmv{X}  \ottsym{)}  \ottsym{)}  \subseteq  \textit{dom} \, \ottsym{(}  S'_{{\mathrm{0}}}  \ottsym{)}$ if $\ottmv{X} \, \in \, \textit{dom} \, \ottsym{(}  S'_{{\mathrm{2}}}  \ottsym{)}$.
    If $\ottmv{X} \, \not\in \, \textit{dom} \, \ottsym{(}  S'_{{\mathrm{2}}}  \ottsym{)}$, then $S'_{{\mathrm{2}}}  \ottsym{(}  \ottmv{X}  \ottsym{)}  \ottsym{=}  \ottmv{X}$.
    Thus, $\textit{ftv} \, \ottsym{(}  S'_{{\mathrm{2}}}  \ottsym{(}  \ottmv{X}  \ottsym{)}  \ottsym{)}  \subseteq  \textit{dom} \, \ottsym{(}   S'_{{\mathrm{0}}}  \circ  S_{{\mathrm{0}}}   \ottsym{)}$, and so
    $  S'_{{\mathrm{0}}}  \circ  S_{{\mathrm{0}}}   \circ  S'_{{\mathrm{2}}}   \ottsym{(}  \ottmv{X}  \ottsym{)}  \ottsym{=}   S''_{{\mathrm{0}}}  \circ     S'_{{\mathrm{0}}}  \circ  S_{{\mathrm{0}}}   \circ  S'_{{\mathrm{3}}}   \circ  S'_{{\mathrm{2}}}    \ottsym{(}  \ottmv{X}  \ottsym{)}$.

    Finally, it is obvious that $\forall \ottmv{X} \, \in \, \textit{dom} \, \ottsym{(}   S'_{{\mathrm{3}}}  \circ  S'_{{\mathrm{2}}}   \ottsym{)}. \textit{ftv} \, \ottsym{(}   S'_{{\mathrm{3}}}  \circ  S'_{{\mathrm{2}}}   \ottsym{(}  \ottmv{X}  \ottsym{)}  \ottsym{)}  \subseteq  \textit{dom} \, \ottsym{(}   S''_{{\mathrm{0}}}  \circ  S'_{{\mathrm{0}}}   \ottsym{)}$.

    \otherwise Cannot happen.  \qedhere
  \end{caseanalysis}
\end{proof}

\begin{lemmaA}[Simulation of Unwrapping] \label{lem:term_prec_simulation_unwrap}
  If $ \langle   \emptyset    \vdash   w_{{\mathrm{1}}}  \ottsym{:}   \ottnt{U_{{\mathrm{11}}}}  \!\rightarrow\!  \ottnt{U_{{\mathrm{12}}}} \Rightarrow  \unskip ^ { \ell }  \! \ottnt{U_{{\mathrm{1}}}}  \!\rightarrow\!  \ottnt{U_{{\mathrm{2}}}}   :  \ottnt{U_{{\mathrm{1}}}}  \!\rightarrow\!  \ottnt{U_{{\mathrm{2}}}}   \sqsubseteq _{ S_{{\mathrm{0}}} }  \ottnt{U'_{{\mathrm{1}}}}  \!\rightarrow\!  \ottnt{U'_{{\mathrm{2}}}}  :  w'_{{\mathrm{1}}}  \dashv   \emptyset   \rangle $
  and $ \langle   \emptyset    \vdash   w_{{\mathrm{2}}}  :  \ottnt{U_{{\mathrm{1}}}}   \sqsubseteq _{ S_{{\mathrm{0}}} }  \ottnt{U'_{{\mathrm{1}}}}  :  w'_{{\mathrm{2}}}  \dashv   \emptyset   \rangle $,
  then there exists $S'_{{\mathrm{0}}}$, $S'$, and $f'$ such that
  \begin{itemize}
   \item $w'_{{\mathrm{1}}} \, w'_{{\mathrm{2}}} \,  \xmapsto{ \mathmakebox[0.4em]{} S' \mathmakebox[0.3em]{} }\hspace{-0.4em}{}^\ast \hspace{0.2em}  \, f'$ where \rnp{E\_Abort} is not applied,
   \item $ \langle   \emptyset    \vdash   \ottsym{(}  w_{{\mathrm{1}}} \, \ottsym{(}  w_{{\mathrm{2}}}  \ottsym{:}   \ottnt{U_{{\mathrm{1}}}} \Rightarrow  \unskip ^ {  \bar{ \ell }  }  \! \ottnt{U_{{\mathrm{11}}}}   \ottsym{)}  \ottsym{)}  \ottsym{:}   \ottnt{U_{{\mathrm{12}}}} \Rightarrow  \unskip ^ { \ell }  \! \ottnt{U_{{\mathrm{2}}}}   :  \ottnt{U_{{\mathrm{2}}}}   \sqsubseteq _{  S'_{{\mathrm{0}}}  \circ  S_{{\mathrm{0}}}  }  S'  \ottsym{(}  \ottnt{U'_{{\mathrm{2}}}}  \ottsym{)}  :  \ottnt{f'}  \dashv   \emptyset   \rangle $,
   \item $\forall \ottmv{X} \in \textit{dom} \, \ottsym{(}  S_{{\mathrm{0}}}  \ottsym{)}. S_{{\mathrm{0}}}  \ottsym{(}  \ottmv{X}  \ottsym{)}  \ottsym{=}   S'_{{\mathrm{0}}}  \circ   S_{{\mathrm{0}}}  \circ  S'    \ottsym{(}  \ottmv{X}  \ottsym{)}$, and
   \item $\forall \ottmv{X} \in \textit{dom} \, \ottsym{(}  S'  \ottsym{)}. \textit{ftv} \, \ottsym{(}  S'  \ottsym{(}  \ottmv{X}  \ottsym{)}  \ottsym{)}  \subseteq  \textit{dom} \, \ottsym{(}  S'_{{\mathrm{0}}}  \ottsym{)}$.
  \end{itemize}
\end{lemmaA}

\begin{proof}
  By \rnp{R\_AppCast},
  $\ottsym{(}  w_{{\mathrm{1}}}  \ottsym{:}   \ottnt{U_{{\mathrm{11}}}}  \!\rightarrow\!  \ottnt{U_{{\mathrm{12}}}} \Rightarrow  \unskip ^ { \ell }  \! \ottnt{U_{{\mathrm{1}}}}  \!\rightarrow\!  \ottnt{U_{{\mathrm{2}}}}   \ottsym{)} \, w_{{\mathrm{2}}} \,  \xmapsto{ \mathmakebox[0.4em]{} [  ] \mathmakebox[0.3em]{} }  \, \ottsym{(}  w_{{\mathrm{1}}} \, \ottsym{(}  w_{{\mathrm{2}}}  \ottsym{:}   \ottnt{U_{{\mathrm{1}}}} \Rightarrow  \unskip ^ {  \bar{ \ell }  }  \! \ottnt{U_{{\mathrm{11}}}}   \ottsym{)}  \ottsym{)}  \ottsym{:}   \ottnt{U_{{\mathrm{12}}}} \Rightarrow  \unskip ^ { \ell }  \! \ottnt{U_{{\mathrm{2}}}} $.
  By induction on the number of casts that occur in $w'_{{\mathrm{1}}}$.
  By case analysis on the precision rule applied last to derive
  $ \langle   \emptyset    \vdash   \ottsym{(}  w_{{\mathrm{1}}}  \ottsym{:}   \ottnt{U_{{\mathrm{11}}}}  \!\rightarrow\!  \ottnt{U_{{\mathrm{12}}}} \Rightarrow  \unskip ^ { \ell }  \! \ottnt{U_{{\mathrm{1}}}}  \!\rightarrow\!  \ottnt{U_{{\mathrm{2}}}}   \ottsym{)}  :  \ottnt{U_{{\mathrm{1}}}}  \!\rightarrow\!  \ottnt{U_{{\mathrm{2}}}}   \sqsubseteq _{ S_{{\mathrm{0}}} }  \ottnt{U'_{{\mathrm{1}}}}  \!\rightarrow\!  \ottnt{U'_{{\mathrm{2}}}}  :  w'_{{\mathrm{1}}}  \dashv   \emptyset   \rangle $.

  \begin{caseanalysis}
    \case{\rnp{P\_Cast}}
    Here, by Lemma \ref{lem:canonical_forms},
    there exist $w'_{{\mathrm{11}}}$, $\ottnt{U'_{{\mathrm{11}}}}$, $\ottnt{U'_{{\mathrm{12}}}}$, and $\ell'$ such that
    $ \langle   \emptyset    \vdash   \ottsym{(}  w_{{\mathrm{1}}}  \ottsym{:}   \ottnt{U_{{\mathrm{11}}}}  \!\rightarrow\!  \ottnt{U_{{\mathrm{12}}}} \Rightarrow  \unskip ^ { \ell }  \! \ottnt{U_{{\mathrm{1}}}}  \!\rightarrow\!  \ottnt{U_{{\mathrm{2}}}}   \ottsym{)}  :  \ottnt{U_{{\mathrm{1}}}}  \!\rightarrow\!  \ottnt{U_{{\mathrm{2}}}}   \sqsubseteq _{ S_{{\mathrm{0}}} }  \ottnt{U'_{{\mathrm{1}}}}  \!\rightarrow\!  \ottnt{U'_{{\mathrm{2}}}}  :  w'_{{\mathrm{11}}}  \ottsym{:}   \ottnt{U'_{{\mathrm{11}}}}  \!\rightarrow\!  \ottnt{U'_{{\mathrm{12}}}} \Rightarrow  \unskip ^ { \ell' }  \! \ottnt{U'_{{\mathrm{1}}}}  \!\rightarrow\!  \ottnt{U'_{{\mathrm{2}}}}   \dashv   \emptyset   \rangle $
    and $w'_{{\mathrm{1}}}  \ottsym{=}  w'_{{\mathrm{11}}}  \ottsym{:}   \ottnt{U'_{{\mathrm{11}}}}  \!\rightarrow\!  \ottnt{U'_{{\mathrm{12}}}} \Rightarrow  \unskip ^ { \ell' }  \! \ottnt{U'_{{\mathrm{1}}}}  \!\rightarrow\!  \ottnt{U'_{{\mathrm{2}}}} $.

    By inversion,
    $ \langle   \emptyset    \vdash   w_{{\mathrm{1}}}  :  \ottnt{U_{{\mathrm{11}}}}  \!\rightarrow\!  \ottnt{U_{{\mathrm{12}}}}   \sqsubseteq _{ S_{{\mathrm{0}}} }  \ottnt{U'_{{\mathrm{11}}}}  \!\rightarrow\!  \ottnt{U'_{{\mathrm{12}}}}  :  w'_{{\mathrm{11}}}  \dashv   \emptyset   \rangle $
    and $ \ottnt{U_{{\mathrm{1}}}}  \!\rightarrow\!  \ottnt{U_{{\mathrm{2}}}}   \sqsubseteq _{ S_{{\mathrm{0}}} }  \ottnt{U'_{{\mathrm{1}}}}  \!\rightarrow\!  \ottnt{U'_{{\mathrm{2}}}} $.
    By inversion of $ \ottnt{U_{{\mathrm{1}}}}  \!\rightarrow\!  \ottnt{U_{{\mathrm{2}}}}   \sqsubseteq _{ S_{{\mathrm{0}}} }  \ottnt{U'_{{\mathrm{1}}}}  \!\rightarrow\!  \ottnt{U'_{{\mathrm{2}}}} $,
    $ \ottnt{U_{{\mathrm{1}}}}   \sqsubseteq _{ S_{{\mathrm{0}}} }  \ottnt{U'_{{\mathrm{1}}}} $ and $ \ottnt{U_{{\mathrm{2}}}}   \sqsubseteq _{ S_{{\mathrm{0}}} }  \ottnt{U'_{{\mathrm{2}}}} $.
    By Lemma~\ref{lem:term_prec_to_type_prec},
    $ \ottnt{U_{{\mathrm{11}}}}  \!\rightarrow\!  \ottnt{U_{{\mathrm{12}}}}   \sqsubseteq _{ S_{{\mathrm{0}}} }  \ottnt{U'_{{\mathrm{11}}}}  \!\rightarrow\!  \ottnt{U'_{{\mathrm{12}}}} $.
    By its inversion, $ \ottnt{U_{{\mathrm{11}}}}   \sqsubseteq _{ S_{{\mathrm{0}}} }  \ottnt{U'_{{\mathrm{11}}}} $ and $ \ottnt{U_{{\mathrm{12}}}}   \sqsubseteq _{ S_{{\mathrm{0}}} }  \ottnt{U'_{{\mathrm{12}}}} $.

    By \rnp{R\_AppCast},
    $\ottsym{(}  w'_{{\mathrm{11}}}  \ottsym{:}   \ottnt{U'_{{\mathrm{11}}}}  \!\rightarrow\!  \ottnt{U'_{{\mathrm{12}}}} \Rightarrow  \unskip ^ { \ell' }  \! \ottnt{U'_{{\mathrm{1}}}}  \!\rightarrow\!  \ottnt{U'_{{\mathrm{2}}}}   \ottsym{)} \, w'_{{\mathrm{2}}} \,  \xmapsto{ \mathmakebox[0.4em]{} [  ] \mathmakebox[0.3em]{} }  \, \ottsym{(}  w'_{{\mathrm{11}}} \, \ottsym{(}  w'_{{\mathrm{2}}}  \ottsym{:}   \ottnt{U'_{{\mathrm{1}}}} \Rightarrow  \unskip ^ {  \bar{ \ell' }  }  \! \ottnt{U'_{{\mathrm{11}}}}   \ottsym{)}  \ottsym{)}  \ottsym{:}   \ottnt{U'_{{\mathrm{12}}}} \Rightarrow  \unskip ^ { \ell' }  \! \ottnt{U'_{{\mathrm{2}}}} $.

    We finish by \rnp{P\_Cast}, \rnp{P\_App}, and \rnp{P\_Cast}.

    \case{\rnp{P\_CastL}}
    Here,
    $ \langle   \emptyset    \vdash   \ottsym{(}  w_{{\mathrm{1}}}  \ottsym{:}   \ottnt{U_{{\mathrm{11}}}}  \!\rightarrow\!  \ottnt{U_{{\mathrm{12}}}} \Rightarrow  \unskip ^ { \ell }  \! \ottnt{U_{{\mathrm{1}}}}  \!\rightarrow\!  \ottnt{U_{{\mathrm{2}}}}   \ottsym{)}  :  \ottnt{U_{{\mathrm{1}}}}  \!\rightarrow\!  \ottnt{U_{{\mathrm{2}}}}   \sqsubseteq _{ S_{{\mathrm{0}}} }  \ottnt{U'_{{\mathrm{1}}}}  \!\rightarrow\!  \ottnt{U'_{{\mathrm{2}}}}  :  w'_{{\mathrm{1}}}  \dashv   \emptyset   \rangle $.

    By inversion,
    $ \langle   \emptyset    \vdash   w_{{\mathrm{1}}}  :  \ottnt{U_{{\mathrm{11}}}}  \!\rightarrow\!  \ottnt{U_{{\mathrm{12}}}}   \sqsubseteq _{ S_{{\mathrm{0}}} }  \ottnt{U'_{{\mathrm{1}}}}  \!\rightarrow\!  \ottnt{U'_{{\mathrm{2}}}}  :  w'_{{\mathrm{1}}}  \dashv   \emptyset   \rangle $
    and $ \ottnt{U_{{\mathrm{1}}}}  \!\rightarrow\!  \ottnt{U_{{\mathrm{2}}}}   \sqsubseteq _{ S_{{\mathrm{0}}} }  \ottnt{U'_{{\mathrm{1}}}}  \!\rightarrow\!  \ottnt{U'_{{\mathrm{2}}}} $.
    By inversion of $ \ottnt{U_{{\mathrm{1}}}}  \!\rightarrow\!  \ottnt{U_{{\mathrm{2}}}}   \sqsubseteq _{ S_{{\mathrm{0}}} }  \ottnt{U'_{{\mathrm{1}}}}  \!\rightarrow\!  \ottnt{U'_{{\mathrm{2}}}} $,
    $ \ottnt{U_{{\mathrm{1}}}}   \sqsubseteq _{ S_{{\mathrm{0}}} }  \ottnt{U'_{{\mathrm{1}}}} $ and $ \ottnt{U_{{\mathrm{2}}}}   \sqsubseteq _{ S_{{\mathrm{0}}} }  \ottnt{U'_{{\mathrm{2}}}} $.
    By Lemma~\ref{lem:term_prec_to_type_prec},
    $ \ottnt{U_{{\mathrm{11}}}}  \!\rightarrow\!  \ottnt{U_{{\mathrm{12}}}}   \sqsubseteq _{ S_{{\mathrm{0}}} }  \ottnt{U'_{{\mathrm{1}}}}  \!\rightarrow\!  \ottnt{U'_{{\mathrm{2}}}} $.
    By its inversion, $ \ottnt{U_{{\mathrm{11}}}}   \sqsubseteq _{ S_{{\mathrm{0}}} }  \ottnt{U'_{{\mathrm{1}}}} $ and $ \ottnt{U_{{\mathrm{12}}}}   \sqsubseteq _{ S_{{\mathrm{0}}} }  \ottnt{U'_{{\mathrm{2}}}} $.

    We finish by \rnp{P\_CastL}, \rnp{P\_App}, and \rnp{P\_CastL}.

    \case{\rnp{P\_CastR}}
    Here, by Lemma \ref{lem:canonical_forms},
    there exist $w'_{{\mathrm{11}}}$, $\ottnt{U'_{{\mathrm{11}}}}$, $\ottnt{U'_{{\mathrm{12}}}}$, and $\ell'$ such that
    $ \langle   \emptyset    \vdash   w_{{\mathrm{1}}}  \ottsym{:}   \ottnt{U_{{\mathrm{11}}}}  \!\rightarrow\!  \ottnt{U_{{\mathrm{12}}}} \Rightarrow  \unskip ^ { \ell }  \! \ottnt{U_{{\mathrm{1}}}}  \!\rightarrow\!  \ottnt{U_{{\mathrm{2}}}}   :  \ottnt{U_{{\mathrm{1}}}}  \!\rightarrow\!  \ottnt{U_{{\mathrm{2}}}}   \sqsubseteq _{ S_{{\mathrm{0}}} }  \ottnt{U'_{{\mathrm{1}}}}  \!\rightarrow\!  \ottnt{U'_{{\mathrm{2}}}}  :  w'_{{\mathrm{11}}}  \ottsym{:}   \ottnt{U'_{{\mathrm{11}}}}  \!\rightarrow\!  \ottnt{U'_{{\mathrm{12}}}} \Rightarrow  \unskip ^ { \ell' }  \! \ottnt{U'_{{\mathrm{1}}}}  \!\rightarrow\!  \ottnt{U'_{{\mathrm{2}}}}   \dashv   \emptyset   \rangle $
    and $w'_{{\mathrm{1}}}  \ottsym{=}  w'_{{\mathrm{11}}}  \ottsym{:}   \ottnt{U'_{{\mathrm{11}}}}  \!\rightarrow\!  \ottnt{U'_{{\mathrm{12}}}} \Rightarrow  \unskip ^ { \ell' }  \! \ottnt{U'_{{\mathrm{1}}}}  \!\rightarrow\!  \ottnt{U'_{{\mathrm{2}}}} $.

    By inversion,
    $ \langle   \emptyset    \vdash   w_{{\mathrm{1}}}  \ottsym{:}   \ottnt{U_{{\mathrm{11}}}}  \!\rightarrow\!  \ottnt{U_{{\mathrm{12}}}} \Rightarrow  \unskip ^ { \ell }  \! \ottnt{U_{{\mathrm{1}}}}  \!\rightarrow\!  \ottnt{U_{{\mathrm{2}}}}   :  \ottnt{U_{{\mathrm{1}}}}  \!\rightarrow\!  \ottnt{U_{{\mathrm{2}}}}   \sqsubseteq _{ S_{{\mathrm{0}}} }  \ottnt{U'_{{\mathrm{11}}}}  \!\rightarrow\!  \ottnt{U'_{{\mathrm{12}}}}  :  w'_{{\mathrm{11}}}  \dashv   \emptyset   \rangle $
    and $ \ottnt{U_{{\mathrm{1}}}}  \!\rightarrow\!  \ottnt{U_{{\mathrm{2}}}}   \sqsubseteq _{ S_{{\mathrm{0}}} }  \ottnt{U'_{{\mathrm{1}}}}  \!\rightarrow\!  \ottnt{U'_{{\mathrm{2}}}} $.
    By inversion of $ \ottnt{U_{{\mathrm{1}}}}  \!\rightarrow\!  \ottnt{U_{{\mathrm{2}}}}   \sqsubseteq _{ S_{{\mathrm{0}}} }  \ottnt{U'_{{\mathrm{1}}}}  \!\rightarrow\!  \ottnt{U'_{{\mathrm{2}}}} $,
    $ \ottnt{U_{{\mathrm{1}}}}   \sqsubseteq _{ S_{{\mathrm{0}}} }  \ottnt{U'_{{\mathrm{1}}}} $ and $ \ottnt{U_{{\mathrm{2}}}}   \sqsubseteq _{ S_{{\mathrm{0}}} }  \ottnt{U'_{{\mathrm{2}}}} $.
    By Lemma~\ref{lem:term_prec_to_type_prec},
    $ \ottnt{U_{{\mathrm{1}}}}  \!\rightarrow\!  \ottnt{U_{{\mathrm{2}}}}   \sqsubseteq _{ S_{{\mathrm{0}}} }  \ottnt{U'_{{\mathrm{11}}}}  \!\rightarrow\!  \ottnt{U'_{{\mathrm{12}}}} $.
    By its inversion, $ \ottnt{U_{{\mathrm{1}}}}   \sqsubseteq _{ S_{{\mathrm{0}}} }  \ottnt{U'_{{\mathrm{11}}}} $ and $ \ottnt{U_{{\mathrm{2}}}}   \sqsubseteq _{ S_{{\mathrm{0}}} }  \ottnt{U'_{{\mathrm{12}}}} $.

    By \rnp{R\_AppCast},
    $\ottsym{(}  w'_{{\mathrm{11}}}  \ottsym{:}   \ottnt{U'_{{\mathrm{11}}}}  \!\rightarrow\!  \ottnt{U'_{{\mathrm{12}}}} \Rightarrow  \unskip ^ { \ell' }  \! \ottnt{U'_{{\mathrm{1}}}}  \!\rightarrow\!  \ottnt{U'_{{\mathrm{2}}}}   \ottsym{)} \, w'_{{\mathrm{2}}} \,  \xmapsto{ \mathmakebox[0.4em]{} [  ] \mathmakebox[0.3em]{} }  \, \ottsym{(}  w'_{{\mathrm{11}}} \, \ottsym{(}  w'_{{\mathrm{2}}}  \ottsym{:}   \ottnt{U'_{{\mathrm{1}}}} \Rightarrow  \unskip ^ {  \bar{ \ell' }  }  \! \ottnt{U'_{{\mathrm{11}}}}   \ottsym{)}  \ottsym{)}  \ottsym{:}   \ottnt{U'_{{\mathrm{12}}}} \Rightarrow  \unskip ^ { \ell' }  \! \ottnt{U'_{{\mathrm{2}}}} $.

    By \rnp{P\_CastR},
    $ \langle   \emptyset    \vdash   w_{{\mathrm{2}}}  :  \ottnt{U_{{\mathrm{1}}}}   \sqsubseteq _{ S_{{\mathrm{0}}} }  \ottnt{U'_{{\mathrm{11}}}}  :  w'_{{\mathrm{2}}}  \ottsym{:}   \ottnt{U'_{{\mathrm{1}}}} \Rightarrow  \unskip ^ {  \bar{ \ell' }  }  \! \ottnt{U'_{{\mathrm{11}}}}   \dashv   \emptyset   \rangle $.
    By Lemma \ref{lem:prec_catch_up_left_value},
    there exist $S'_{{\mathrm{0}}}$, $S'_{{\mathrm{2}}}$, and $w''_{{\mathrm{2}}}$ such that
    \begin{itemize}
     \item $w'_{{\mathrm{2}}}  \ottsym{:}   \ottnt{U'_{{\mathrm{1}}}} \Rightarrow  \unskip ^ {  \bar{ \ell' }  }  \! \ottnt{U'_{{\mathrm{11}}}}  \,  \xmapsto{ \mathmakebox[0.4em]{} S'_{{\mathrm{2}}} \mathmakebox[0.3em]{} }\hspace{-0.4em}{}^\ast \hspace{0.2em}  \, w''_{{\mathrm{2}}}$,
     \item $ \langle   \emptyset    \vdash   w_{{\mathrm{2}}}  :  \ottnt{U_{{\mathrm{1}}}}   \sqsubseteq _{  S'_{{\mathrm{0}}}  \circ  S_{{\mathrm{0}}}  }  S'_{{\mathrm{2}}}  \ottsym{(}  \ottnt{U'_{{\mathrm{11}}}}  \ottsym{)}  :  w''_{{\mathrm{2}}}  \dashv   \emptyset   \rangle $,
     \item $\forall \ottmv{X} \in \textit{dom} \, \ottsym{(}  S_{{\mathrm{0}}}  \ottsym{)}. S_{{\mathrm{0}}}  \ottsym{(}  \ottmv{X}  \ottsym{)}  \ottsym{=}   S'_{{\mathrm{0}}}  \circ   S_{{\mathrm{0}}}  \circ  S'_{{\mathrm{2}}}    \ottsym{(}  \ottmv{X}  \ottsym{)}$, and
     \item $\forall \ottmv{X} \in \textit{dom} \, \ottsym{(}  S'_{{\mathrm{2}}}  \ottsym{)}. \textit{ftv} \, \ottsym{(}  S'_{{\mathrm{2}}}  \ottsym{(}  \ottmv{X}  \ottsym{)}  \ottsym{)}  \subseteq  \textit{dom} \, \ottsym{(}  S'_{{\mathrm{0}}}  \ottsym{)}$.
    \end{itemize}
    We have
    $\ottsym{(}  w'_{{\mathrm{11}}} \, \ottsym{(}  w'_{{\mathrm{2}}}  \ottsym{:}   \ottnt{U'_{{\mathrm{1}}}} \Rightarrow  \unskip ^ {  \bar{ \ell' }  }  \! \ottnt{U'_{{\mathrm{11}}}}   \ottsym{)}  \ottsym{)}  \ottsym{:}   \ottnt{U'_{{\mathrm{12}}}} \Rightarrow  \unskip ^ { \ell' }  \! \ottnt{U'_{{\mathrm{2}}}}  \,  \xmapsto{ \mathmakebox[0.4em]{} S'_{{\mathrm{2}}} \mathmakebox[0.3em]{} }\hspace{-0.4em}{}^\ast \hspace{0.2em}  \, \ottsym{(}  S'_{{\mathrm{2}}}  \ottsym{(}  w'_{{\mathrm{11}}}  \ottsym{)} \, w''_{{\mathrm{2}}}  \ottsym{)}  \ottsym{:}   S'_{{\mathrm{2}}}  \ottsym{(}  \ottnt{U'_{{\mathrm{12}}}}  \ottsym{)} \Rightarrow  \unskip ^ { \ell' }  \! S'_{{\mathrm{2}}}  \ottsym{(}  \ottnt{U'_{{\mathrm{2}}}}  \ottsym{)} $
    by Lemma~\ref{lem:eval_subterm}.
    By Lemma~\ref{lem:right_subst_preserve_prec},
    $ \ottnt{U_{{\mathrm{2}}}}   \sqsubseteq _{  S'_{{\mathrm{0}}}  \circ  S_{{\mathrm{0}}}  }  S'_{{\mathrm{2}}}  \ottsym{(}  \ottnt{U'_{{\mathrm{2}}}}  \ottsym{)} $ and
    $ \langle   \emptyset    \vdash   w_{{\mathrm{1}}}  :  \ottnt{U_{{\mathrm{1}}}}  \!\rightarrow\!  \ottnt{U_{{\mathrm{2}}}}   \sqsubseteq _{  S'_{{\mathrm{0}}}  \circ  S_{{\mathrm{0}}}  }  S'_{{\mathrm{2}}}  \ottsym{(}  \ottnt{U'_{{\mathrm{11}}}}  \!\rightarrow\!  \ottnt{U'_{{\mathrm{12}}}}  \ottsym{)}  :  S'_{{\mathrm{2}}}  \ottsym{(}  w'_{{\mathrm{11}}}  \ottsym{)}  \dashv   \emptyset   \rangle $.

    By the IH,
    there exist $S''_{{\mathrm{0}}}$, $S'_{{\mathrm{3}}}$, and $\ottnt{f''}$ such that
    \begin{itemize}
     \item $S'_{{\mathrm{2}}}  \ottsym{(}  w'_{{\mathrm{11}}}  \ottsym{)} \, w''_{{\mathrm{2}}} \,  \xmapsto{ \mathmakebox[0.4em]{} S'_{{\mathrm{3}}} \mathmakebox[0.3em]{} }\hspace{-0.4em}{}^\ast \hspace{0.2em}  \, \ottnt{f''}$ where \rnp{E\_Abort} is not applied,
     \item $ \langle   \emptyset    \vdash   \ottsym{(}  w_{{\mathrm{1}}} \, \ottsym{(}  w_{{\mathrm{2}}}  \ottsym{:}   \ottnt{U_{{\mathrm{1}}}} \Rightarrow  \unskip ^ {  \bar{ \ell }  }  \! \ottnt{U_{{\mathrm{11}}}}   \ottsym{)}  \ottsym{)}  \ottsym{:}   \ottnt{U_{{\mathrm{12}}}} \Rightarrow  \unskip ^ { \ell }  \! \ottnt{U_{{\mathrm{2}}}}   :  \ottnt{U_{{\mathrm{2}}}}   \sqsubseteq _{   S''_{{\mathrm{0}}}  \circ  S'_{{\mathrm{0}}}   \circ  S_{{\mathrm{0}}}  }   S'_{{\mathrm{3}}}  \circ  S'_{{\mathrm{2}}}   \ottsym{(}  \ottnt{U'_{{\mathrm{12}}}}  \ottsym{)}  :  \ottnt{f''}  \dashv   \emptyset   \rangle $,
     \item $\forall \ottmv{X} \in \textit{dom} \, \ottsym{(}   S'_{{\mathrm{0}}}  \circ  S_{{\mathrm{0}}}   \ottsym{)}.  S'_{{\mathrm{0}}}  \circ  S_{{\mathrm{0}}}   \ottsym{(}  \ottmv{X}  \ottsym{)}  \ottsym{=}   S''_{{\mathrm{0}}}  \circ    S'_{{\mathrm{0}}}  \circ  S_{{\mathrm{0}}}   \circ  S'_{{\mathrm{3}}}    \ottsym{(}  \ottmv{X}  \ottsym{)}$, and
     \item $\forall \ottmv{X} \in \textit{dom} \, \ottsym{(}  S'_{{\mathrm{3}}}  \ottsym{)}. \textit{ftv} \, \ottsym{(}  S'_{{\mathrm{3}}}  \ottsym{(}  \ottmv{X}  \ottsym{)}  \ottsym{)}  \subseteq  \textit{dom} \, \ottsym{(}  S''_{{\mathrm{0}}}  \ottsym{)}$.
    \end{itemize}
    We have
    $\ottsym{(}  S'_{{\mathrm{2}}}  \ottsym{(}  w'_{{\mathrm{11}}}  \ottsym{)} \, w''_{{\mathrm{2}}}  \ottsym{)}  \ottsym{:}   S'_{{\mathrm{2}}}  \ottsym{(}  \ottnt{U'_{{\mathrm{12}}}}  \ottsym{)} \Rightarrow  \unskip ^ { \ell' }  \! S'_{{\mathrm{2}}}  \ottsym{(}  \ottnt{U'_{{\mathrm{2}}}}  \ottsym{)}  \,  \xmapsto{ \mathmakebox[0.4em]{} S'_{{\mathrm{3}}} \mathmakebox[0.3em]{} }\hspace{-0.4em}{}^\ast \hspace{0.2em}  \, \ottnt{f''}  \ottsym{:}    S'_{{\mathrm{3}}}  \circ  S'_{{\mathrm{2}}}   \ottsym{(}  \ottnt{U'_{{\mathrm{12}}}}  \ottsym{)} \Rightarrow  \unskip ^ { \ell' }  \!  S'_{{\mathrm{3}}}  \circ  S'_{{\mathrm{2}}}   \ottsym{(}  \ottnt{U'_{{\mathrm{2}}}}  \ottsym{)} $ by Lemma~\ref{lem:eval_subterm}.

    By Lemma~\ref{lem:right_subst_preserve_prec},
    $ \ottnt{U_{{\mathrm{2}}}}   \sqsubseteq _{   S''_{{\mathrm{0}}}  \circ  S'_{{\mathrm{0}}}   \circ  S_{{\mathrm{0}}}  }   S'_{{\mathrm{3}}}  \circ  S'_{{\mathrm{2}}}   \ottsym{(}  \ottnt{U'_{{\mathrm{2}}}}  \ottsym{)} $.
    By \rnp{P\_CastR},
    \[
      \langle   \emptyset    \vdash   \ottsym{(}  w_{{\mathrm{1}}} \, \ottsym{(}  w_{{\mathrm{2}}}  \ottsym{:}   \ottnt{U_{{\mathrm{1}}}} \Rightarrow  \unskip ^ {  \bar{ \ell }  }  \! \ottnt{U_{{\mathrm{11}}}}   \ottsym{)}  \ottsym{)}  \ottsym{:}   \ottnt{U_{{\mathrm{12}}}} \Rightarrow  \unskip ^ { \ell }  \! \ottnt{U_{{\mathrm{2}}}}   :  \ottnt{U_{{\mathrm{2}}}}   \sqsubseteq _{ S'''_{{\mathrm{0}}} }  S'_{{\mathrm{32}}}  \ottsym{(}  \ottnt{U'_{{\mathrm{2}}}}  \ottsym{)}  :  \ottnt{f''}  \ottsym{:}   S'_{{\mathrm{32}}}  \ottsym{(}  \ottnt{U'_{{\mathrm{12}}}}  \ottsym{)} \Rightarrow  \unskip ^ { \ell' }  \! S'_{{\mathrm{32}}}  \ottsym{(}  \ottnt{U'_{{\mathrm{2}}}}  \ottsym{)}   \dashv   \emptyset   \rangle .
   \]
   where $S'''_{{\mathrm{0}}} =   S''_{{\mathrm{0}}}  \circ  S'_{{\mathrm{0}}}   \circ  S_{{\mathrm{0}}} $ and $S'_{{\mathrm{32}}} =  S'_{{\mathrm{3}}}  \circ  S'_{{\mathrm{2}}} $.

    We show $\forall \ottmv{X} \in \textit{dom} \, \ottsym{(}  S_{{\mathrm{0}}}  \ottsym{)}. S_{{\mathrm{0}}}  \ottsym{(}  \ottmv{X}  \ottsym{)}  \ottsym{=}   S''_{{\mathrm{0}}}  \circ     S'_{{\mathrm{0}}}  \circ  S_{{\mathrm{0}}}   \circ  S'_{{\mathrm{3}}}   \circ  S'_{{\mathrm{2}}}    \ottsym{(}  \ottmv{X}  \ottsym{)}$.
    Let $\ottmv{X} \, \in \, \textit{dom} \, \ottsym{(}  S_{{\mathrm{0}}}  \ottsym{)}$.
    We have $S_{{\mathrm{0}}}  \ottsym{(}  \ottmv{X}  \ottsym{)}  \ottsym{=}   S'_{{\mathrm{0}}}  \circ   S_{{\mathrm{0}}}  \circ  S'_{{\mathrm{2}}}    \ottsym{(}  \ottmv{X}  \ottsym{)}$.
    Since $S'_{{\mathrm{2}}}$ is generated by DTI,
    free type variables in $S'_{{\mathrm{2}}}  \ottsym{(}  \ottmv{X}  \ottsym{)}$ are fresh for $S_{{\mathrm{0}}}$.
    Since $\forall \ottmv{X'}, S_{{\mathrm{0}}}  \ottsym{(}  \ottmv{X'}  \ottsym{)}  \ottsym{=}   S'_{{\mathrm{0}}}  \circ   S_{{\mathrm{0}}}  \circ  S'_{{\mathrm{2}}}    \ottsym{(}  \ottmv{X'}  \ottsym{)}$,
    it is found that $\textit{ftv} \, \ottsym{(}  S'_{{\mathrm{2}}}  \ottsym{(}  \ottmv{X}  \ottsym{)}  \ottsym{)}  \subseteq  \textit{dom} \, \ottsym{(}  S'_{{\mathrm{0}}}  \ottsym{)}$ if $\ottmv{X} \, \in \, \textit{dom} \, \ottsym{(}  S'_{{\mathrm{2}}}  \ottsym{)}$.
    If $\ottmv{X} \, \not\in \, \textit{dom} \, \ottsym{(}  S'_{{\mathrm{2}}}  \ottsym{)}$, then $S'_{{\mathrm{2}}}  \ottsym{(}  \ottmv{X}  \ottsym{)}  \ottsym{=}  \ottmv{X}$.
    Thus, $\textit{ftv} \, \ottsym{(}  S'_{{\mathrm{2}}}  \ottsym{(}  \ottmv{X}  \ottsym{)}  \ottsym{)}  \subseteq  \textit{dom} \, \ottsym{(}   S'_{{\mathrm{0}}}  \circ  S_{{\mathrm{0}}}   \ottsym{)}$, and so
    $  S'_{{\mathrm{0}}}  \circ  S_{{\mathrm{0}}}   \circ  S'_{{\mathrm{2}}}   \ottsym{(}  \ottmv{X}  \ottsym{)}  \ottsym{=}   S''_{{\mathrm{0}}}  \circ     S'_{{\mathrm{0}}}  \circ  S_{{\mathrm{0}}}   \circ  S'_{{\mathrm{3}}}   \circ  S'_{{\mathrm{2}}}    \ottsym{(}  \ottmv{X}  \ottsym{)}$.

    Finally, it is obvious that $\forall \ottmv{X} \, \in \, \textit{dom} \, \ottsym{(}   S'_{{\mathrm{3}}}  \circ  S'_{{\mathrm{2}}}   \ottsym{)}. \textit{ftv} \, \ottsym{(}   S'_{{\mathrm{3}}}  \circ  S'_{{\mathrm{2}}}   \ottsym{(}  \ottmv{X}  \ottsym{)}  \ottsym{)}  \subseteq  \textit{dom} \, \ottsym{(}   S''_{{\mathrm{0}}}  \circ  S'_{{\mathrm{0}}}   \ottsym{)}$.
    \qedhere
  \end{caseanalysis}
\end{proof}

\begin{lemmaA}[Simulation of More Precise Programs] \label{lem:term_prec_simulation}
  If $ \langle   \emptyset    \vdash   f_{{\mathrm{1}}}  :  \ottnt{U}   \sqsubseteq _{ S_{{\mathrm{0}}} }  \ottnt{U'}  :  f'_{{\mathrm{1}}}  \dashv   \emptyset   \rangle $ and $f_{{\mathrm{1}}} \,  \xmapsto{ \mathmakebox[0.4em]{} S \mathmakebox[0.3em]{} }  \, f_{{\mathrm{2}}}$,
  then there exist $S'$, $S'_{{\mathrm{0}}}$, $S''$, and $f'_{{\mathrm{2}}}$ such that
  \begin{itemize}
   \item $f'_{{\mathrm{1}}} \,  \xmapsto{ \mathmakebox[0.4em]{} S' \mathmakebox[0.3em]{} }\hspace{-0.4em}{}^\ast \hspace{0.2em}  \, f'_{{\mathrm{2}}}$ where \rnp{E\_Abort} is not applied,
   \item $ \langle   \emptyset    \vdash   f_{{\mathrm{2}}}  :  S  \ottsym{(}  \ottnt{U}  \ottsym{)}   \sqsubseteq _{  S''  \uplus  \ottsym{(}   S'_{{\mathrm{0}}}  \circ   S  \circ  S_{{\mathrm{0}}}    \ottsym{)}  }  S'  \ottsym{(}  \ottnt{U'}  \ottsym{)}  :  f'_{{\mathrm{2}}}  \dashv   \emptyset   \rangle $,
   \item $\forall \ottmv{X} \in \textit{dom} \, \ottsym{(}   S  \circ  S_{{\mathrm{0}}}   \ottsym{)}.  S  \circ  S_{{\mathrm{0}}}   \ottsym{(}  \ottmv{X}  \ottsym{)}  \ottsym{=}   S'_{{\mathrm{0}}}  \circ    S  \circ  S_{{\mathrm{0}}}   \circ  S'    \ottsym{(}  \ottmv{X}  \ottsym{)}$,
   \item $\forall \ottmv{X} \in \textit{dom} \, \ottsym{(}  S'  \ottsym{)}. \textit{ftv} \, \ottsym{(}  S'  \ottsym{(}  \ottmv{X}  \ottsym{)}  \ottsym{)}  \subseteq  \textit{dom} \, \ottsym{(}  S'_{{\mathrm{0}}}  \ottsym{)}$, and
   \item $\textit{dom} \, \ottsym{(}  S''  \ottsym{)}$ is a set of fresh type variables.
  \end{itemize}
\end{lemmaA}

\begin{proof}
  By induction on the term precision derivation.

  \begin{caseanalysis}
    \case{\rnp{P\_Op}}
    We are given $f_{{\mathrm{1}}}  \ottsym{=}  \mathit{op} \, \ottsym{(}  f_{{\mathrm{11}}}  \ottsym{,}  f_{{\mathrm{12}}}  \ottsym{)}$ and $f'_{{\mathrm{1}}}  \ottsym{=}  \mathit{op} \, \ottsym{(}  f'_{{\mathrm{11}}}  \ottsym{,}  f'_{{\mathrm{12}}}  \ottsym{)}$
    for some $ \mathit{op} $, $f_{{\mathrm{11}}}$, $f_{{\mathrm{12}}}$, $f'_{{\mathrm{11}}}$, and $f'_{{\mathrm{12}}}$.

    By inversion,
    there exist $\iota_{{\mathrm{1}}}$, $\iota_{{\mathrm{2}}}$, and $\iota$ such that
    \begin{itemize}
     \item $ \mathit{ty} ( \mathit{op} )   \ottsym{=}  \iota_{{\mathrm{1}}}  \!\rightarrow\!  \iota_{{\mathrm{2}}}  \!\rightarrow\!  \iota$,
     \item $ \langle   \emptyset    \vdash   f_{{\mathrm{11}}}  :  \iota_{{\mathrm{1}}}   \sqsubseteq _{ S_{{\mathrm{0}}} }  \iota_{{\mathrm{1}}}  :  f'_{{\mathrm{11}}}  \dashv   \emptyset   \rangle $, and
     \item $ \langle   \emptyset    \vdash   f_{{\mathrm{12}}}  :  \iota_{{\mathrm{2}}}   \sqsubseteq _{ S_{{\mathrm{0}}} }  \iota_{{\mathrm{2}}}  :  f'_{{\mathrm{12}}}  \dashv   \emptyset   \rangle $
    \end{itemize}
    where $\ottnt{U}  \ottsym{=}  \ottnt{U'} = \iota$.
    By case analysis on the evaluation rule applied to $f_{{\mathrm{1}}}$.

    \begin{caseanalysis}
      \case{\rnp{E\_Step}}
      There exist $\ottnt{E}$ and $f_{{\mathrm{13}}}$ such that $\ottnt{E}  [  f_{{\mathrm{13}}}  ]  \ottsym{=}  \mathit{op} \, \ottsym{(}  f_{{\mathrm{11}}}  \ottsym{,}  f_{{\mathrm{12}}}  \ottsym{)}$ and
      $\ottnt{E}  [  f_{{\mathrm{13}}}  ] \,  \xmapsto{ \mathmakebox[0.4em]{} S \mathmakebox[0.3em]{} }  \, S  \ottsym{(}  \ottnt{E}  [  f'_{{\mathrm{13}}}  ]  \ottsym{)}$.
      By inversion, $f_{{\mathrm{13}}} \,  \xrightarrow{ \mathmakebox[0.4em]{} S \mathmakebox[0.3em]{} }  \, f'_{{\mathrm{13}}}$.
      By case analysis on the structure of $\ottnt{E}$.

      \begin{caseanalysis}
        \case{$\ottnt{E}  \ottsym{=}  \left[ \, \right]$}
        Here, $\mathit{op} \, \ottsym{(}  f_{{\mathrm{11}}}  \ottsym{,}  f_{{\mathrm{12}}}  \ottsym{)}  \ottsym{=}  f_{{\mathrm{13}}}$.
        By case analysis on the reduction rule applied to $f_{{\mathrm{13}}}$.

        \begin{caseanalysis}
          \case{\rnp{R\_Op}}
          $f_{{\mathrm{11}}}$ and $f_{{\mathrm{12}}}$ are values.

          By Lemma~\ref{lem:prec_catch_up_left_value},
          there exist $S'_{{\mathrm{0}}}$, $S'_{{\mathrm{11}}}$, and $w'_{{\mathrm{11}}}$ such that
          \begin{itemize}
           \item $f'_{{\mathrm{11}}} \,  \xmapsto{ \mathmakebox[0.4em]{} S'_{{\mathrm{11}}} \mathmakebox[0.3em]{} }\hspace{-0.4em}{}^\ast \hspace{0.2em}  \, w'_{{\mathrm{11}}}$,
           \item $ \langle   \emptyset    \vdash   f_{{\mathrm{11}}}  :  \iota_{{\mathrm{1}}}   \sqsubseteq _{  S'_{{\mathrm{0}}}  \circ  S_{{\mathrm{0}}}  }  S'_{{\mathrm{11}}}  \ottsym{(}  \iota_{{\mathrm{1}}}  \ottsym{)}  :  w'_{{\mathrm{11}}}  \dashv   \emptyset   \rangle $,
           \item $\forall \ottmv{X} \in \textit{dom} \, \ottsym{(}  S_{{\mathrm{0}}}  \ottsym{)}. S_{{\mathrm{0}}}  \ottsym{(}  \ottmv{X}  \ottsym{)}  \ottsym{=}   S'_{{\mathrm{0}}}  \circ   S_{{\mathrm{0}}}  \circ  S'_{{\mathrm{11}}}    \ottsym{(}  \ottmv{X}  \ottsym{)}$, and
           \item $\forall \ottmv{X} \in \textit{dom} \, \ottsym{(}  S'_{{\mathrm{11}}}  \ottsym{)}. \textit{ftv} \, \ottsym{(}  S'_{{\mathrm{11}}}  \ottsym{(}  \ottmv{X}  \ottsym{)}  \ottsym{)}  \subseteq  \textit{dom} \, \ottsym{(}  S'_{{\mathrm{0}}}  \ottsym{)}$.
          \end{itemize}
      
          So, $\mathit{op} \, \ottsym{(}  f'_{{\mathrm{11}}}  \ottsym{,}  f'_{{\mathrm{12}}}  \ottsym{)} \,  \xmapsto{ \mathmakebox[0.4em]{} S'_{{\mathrm{11}}} \mathmakebox[0.3em]{} }\hspace{-0.4em}{}^\ast \hspace{0.2em}  \, \mathit{op} \, \ottsym{(}  w'_{{\mathrm{11}}}  \ottsym{,}  S'_{{\mathrm{11}}}  \ottsym{(}  f'_{{\mathrm{12}}}  \ottsym{)}  \ottsym{)}$ by Lemma~\ref{lem:eval_subterm}.
      
          By Lemma~\ref{lem:right_subst_preserve_prec},
          $ \langle   \emptyset    \vdash   f_{{\mathrm{12}}}  :  \iota_{{\mathrm{2}}}   \sqsubseteq _{  S'_{{\mathrm{0}}}  \circ  S_{{\mathrm{0}}}  }  S'_{{\mathrm{11}}}  \ottsym{(}  \iota_{{\mathrm{2}}}  \ottsym{)}  :  S'_{{\mathrm{11}}}  \ottsym{(}  f'_{{\mathrm{12}}}  \ottsym{)}  \dashv   \emptyset   \rangle $.
          By Lemma~\ref{lem:prec_catch_up_left_value},
          there exist $S''_{{\mathrm{0}}}$, $S'_{{\mathrm{12}}}$, and $w'_{{\mathrm{12}}}$ such that
          \begin{itemize}
           \item $S'_{{\mathrm{11}}}  \ottsym{(}  f'_{{\mathrm{12}}}  \ottsym{)} \,  \xmapsto{ \mathmakebox[0.4em]{} S'_{{\mathrm{12}}} \mathmakebox[0.3em]{} }\hspace{-0.4em}{}^\ast \hspace{0.2em}  \, w'_{{\mathrm{12}}}$,
           \item $ \langle   \emptyset    \vdash   f_{{\mathrm{12}}}  :  \iota_{{\mathrm{2}}}   \sqsubseteq _{   S''_{{\mathrm{0}}}  \circ  S'_{{\mathrm{0}}}   \circ  S_{{\mathrm{0}}}  }  S'_{{\mathrm{12}}}  \ottsym{(}  \iota_{{\mathrm{2}}}  \ottsym{)}  :  w'_{{\mathrm{12}}}  \dashv   \emptyset   \rangle $,
           \item $\forall \ottmv{X} \in \textit{dom} \, \ottsym{(}   S'_{{\mathrm{0}}}  \circ  S_{{\mathrm{0}}}   \ottsym{)}.  S'_{{\mathrm{0}}}  \circ  S_{{\mathrm{0}}}   \ottsym{(}  \ottmv{X}  \ottsym{)}  \ottsym{=}   S''_{{\mathrm{0}}}  \circ    S'_{{\mathrm{0}}}  \circ  S_{{\mathrm{0}}}   \circ  S'_{{\mathrm{12}}}    \ottsym{(}  \ottmv{X}  \ottsym{)}$, and
           \item $\forall \ottmv{X} \in \textit{dom} \, \ottsym{(}  S'_{{\mathrm{12}}}  \ottsym{)}. \textit{ftv} \, \ottsym{(}  S'_{{\mathrm{12}}}  \ottsym{(}  \ottmv{X}  \ottsym{)}  \ottsym{)}  \subseteq  \textit{dom} \, \ottsym{(}  S''_{{\mathrm{0}}}  \ottsym{)}$.
          \end{itemize}
          So, $\mathit{op} \, \ottsym{(}  w'_{{\mathrm{11}}}  \ottsym{,}  S'_{{\mathrm{11}}}  \ottsym{(}  f'_{{\mathrm{12}}}  \ottsym{)}  \ottsym{)} \,  \xmapsto{ \mathmakebox[0.4em]{} S'_{{\mathrm{12}}} \mathmakebox[0.3em]{} }\hspace{-0.4em}{}^\ast \hspace{0.2em}  \, \mathit{op} \, \ottsym{(}  S'_{{\mathrm{12}}}  \ottsym{(}  w'_{{\mathrm{11}}}  \ottsym{)}  \ottsym{,}  w'_{{\mathrm{12}}}  \ottsym{)}$ by Lemma~\ref{lem:eval_subterm}.
          By Lemma \ref{lem:term_prec_inversion4},
          there exists $\ottnt{c_{{\mathrm{12}}}}$ such that $\ottnt{c_{{\mathrm{12}}}} = f_{{\mathrm{12}}}  \ottsym{=}  w'_{{\mathrm{12}}}$.
      
          By Lemma~\ref{lem:right_subst_preserve_prec},
          $ \langle   \emptyset    \vdash   f_{{\mathrm{11}}}  :  \iota_{{\mathrm{1}}}   \sqsubseteq _{   S''_{{\mathrm{0}}}  \circ  S'_{{\mathrm{0}}}   \circ  S_{{\mathrm{0}}}  }  \iota_{{\mathrm{1}}}  :  S'_{{\mathrm{12}}}  \ottsym{(}  w'_{{\mathrm{11}}}  \ottsym{)}  \dashv   \emptyset   \rangle $.
          By Lemma~\ref{lem:term_prec_inversion4},
          there exists $\ottnt{c_{{\mathrm{11}}}}$ such that $\ottnt{c_{{\mathrm{11}}}} = f_{{\mathrm{11}}}  \ottsym{=}  S'_{{\mathrm{12}}}  \ottsym{(}  w'_{{\mathrm{11}}}  \ottsym{)}$.
      
          Finally,
          \begin{itemize}
           \item $\mathit{op} \, \ottsym{(}  f'_{{\mathrm{11}}}  \ottsym{,}  f'_{{\mathrm{12}}}  \ottsym{)} \,  \xmapsto{ \mathmakebox[0.4em]{}  S'_{{\mathrm{12}}}  \circ  S'_{{\mathrm{11}}}  \mathmakebox[0.3em]{} }\hspace{-0.4em}{}^\ast \hspace{0.2em}  \,  \llbracket\mathit{op}\rrbracket ( \ottnt{c_{{\mathrm{11}}}} ,  \ottnt{c_{{\mathrm{12}}}} ) $,
           \item $ \langle   \emptyset    \vdash    \llbracket\mathit{op}\rrbracket ( \ottnt{c_{{\mathrm{11}}}} ,  \ottnt{c_{{\mathrm{12}}}} )   :  \iota   \sqsubseteq _{   S''_{{\mathrm{0}}}  \circ  S'_{{\mathrm{0}}}   \circ  S_{{\mathrm{0}}}  }  \iota  :   \llbracket\mathit{op}\rrbracket ( \ottnt{c_{{\mathrm{11}}}} ,  \ottnt{c_{{\mathrm{12}}}} )   \dashv   \emptyset   \rangle $,
           \item $\forall \ottmv{X} \in \textit{dom} \, \ottsym{(}  S_{{\mathrm{0}}}  \ottsym{)}. S_{{\mathrm{0}}}  \ottsym{(}  \ottmv{X}  \ottsym{)}  \ottsym{=}   S''_{{\mathrm{0}}}  \circ     S'_{{\mathrm{0}}}  \circ  S_{{\mathrm{0}}}   \circ  S'_{{\mathrm{12}}}   \circ  S'_{{\mathrm{11}}}    \ottsym{(}  \ottmv{X}  \ottsym{)}$, and
           \item $\forall \ottmv{X} \in \textit{dom} \, \ottsym{(}   S'_{{\mathrm{12}}}  \circ  S'_{{\mathrm{11}}}   \ottsym{)}. \textit{ftv} \, \ottsym{(}   S'_{{\mathrm{12}}}  \circ  S'_{{\mathrm{11}}}   \ottsym{(}  \ottmv{X}  \ottsym{)}  \ottsym{)}  \subseteq  \textit{dom} \, \ottsym{(}   S''_{{\mathrm{0}}}  \circ  S'_{{\mathrm{0}}}   \ottsym{)}$.
          \end{itemize}

          \otherwise Cannot happen.
        \end{caseanalysis}

        \case{$\ottnt{E}  \ottsym{=}  \mathit{op} \, \ottsym{(}  \ottnt{E'}  \ottsym{,}  f_{{\mathrm{12}}}  \ottsym{)}$ for some $\ottnt{E'}$}
        Here, $f_{{\mathrm{11}}}  \ottsym{=}  \ottnt{E'}  [  f_{{\mathrm{13}}}  ]$.
        Since $ \langle   \emptyset    \vdash   \ottnt{E'}  [  f_{{\mathrm{13}}}  ]  :  \iota_{{\mathrm{1}}}   \sqsubseteq _{ S_{{\mathrm{0}}} }  \iota_{{\mathrm{1}}}  :  f'_{{\mathrm{11}}}  \dashv   \emptyset   \rangle $
        and $\ottnt{E'}  [  f_{{\mathrm{13}}}  ] \,  \xmapsto{ \mathmakebox[0.4em]{} S \mathmakebox[0.3em]{} }  \, S  \ottsym{(}  \ottnt{E'}  [  f'_{{\mathrm{13}}}  ]  \ottsym{)}$,
        there exist $S'$, $S'_{{\mathrm{0}}}$, $S''$, and $f'_{{\mathrm{2}}}$ such that
        \begin{itemize}
         \item $f'_{{\mathrm{11}}} \,  \xmapsto{ \mathmakebox[0.4em]{} S' \mathmakebox[0.3em]{} }\hspace{-0.4em}{}^\ast \hspace{0.2em}  \, f'_{{\mathrm{2}}}$ where \rnp{E\_Abort} is not applied,
         \item $ \langle   \emptyset    \vdash   S  \ottsym{(}  \ottnt{E'}  [  f'_{{\mathrm{13}}}  ]  \ottsym{)}  :  \iota_{{\mathrm{1}}}   \sqsubseteq _{  S''  \uplus  \ottsym{(}   S'_{{\mathrm{0}}}  \circ   S  \circ  S_{{\mathrm{0}}}    \ottsym{)}  }  \iota_{{\mathrm{1}}}  :  f'_{{\mathrm{2}}}  \dashv   \emptyset   \rangle $,
         \item $\forall \ottmv{X} \in \textit{dom} \, \ottsym{(}   S  \circ  S_{{\mathrm{0}}}   \ottsym{)}.  S  \circ  S_{{\mathrm{0}}}   \ottsym{(}  \ottmv{X}  \ottsym{)}  \ottsym{=}   S'_{{\mathrm{0}}}  \circ    S  \circ  S_{{\mathrm{0}}}   \circ  S'    \ottsym{(}  \ottmv{X}  \ottsym{)}$,
         \item $\forall \ottmv{X} \in \textit{dom} \, \ottsym{(}  S'  \ottsym{)}. \textit{ftv} \, \ottsym{(}  S'  \ottsym{(}  \ottmv{X}  \ottsym{)}  \ottsym{)}  \subseteq  \textit{dom} \, \ottsym{(}  S'_{{\mathrm{0}}}  \ottsym{)}$, and
         \item $\textit{dom} \, \ottsym{(}  S''  \ottsym{)}$ is a set of fresh type variables
        \end{itemize}
        by the IH.
        By Lemma~\ref{lem:eval_subterm},
        $\mathit{op} \, \ottsym{(}  f'_{{\mathrm{11}}}  \ottsym{,}  f'_{{\mathrm{12}}}  \ottsym{)} \,  \xmapsto{ \mathmakebox[0.4em]{} S' \mathmakebox[0.3em]{} }\hspace{-0.4em}{}^\ast \hspace{0.2em}  \, \mathit{op} \, \ottsym{(}  f'_{{\mathrm{2}}}  \ottsym{,}  S  \ottsym{(}  f'_{{\mathrm{12}}}  \ottsym{)}  \ottsym{)}$.
        Since $ \langle   \emptyset    \vdash   f_{{\mathrm{12}}}  :  \iota_{{\mathrm{2}}}   \sqsubseteq _{ S_{{\mathrm{0}}} }  \iota_{{\mathrm{2}}}  :  f'_{{\mathrm{12}}}  \dashv   \emptyset   \rangle $,
        we have $ \langle   \emptyset    \vdash   S  \ottsym{(}  f_{{\mathrm{12}}}  \ottsym{)}  :  \iota_{{\mathrm{2}}}   \sqsubseteq _{  S''  \uplus  \ottsym{(}   S'_{{\mathrm{0}}}  \circ   S  \circ  S_{{\mathrm{0}}}    \ottsym{)}  }  \iota_{{\mathrm{2}}}  :  S'  \ottsym{(}  f'_{{\mathrm{12}}}  \ottsym{)}  \dashv   \emptyset   \rangle $
        by Lemmas~\ref{lem:left_subst_preserve_prec}, \ref{lem:right_subst_preserve_prec}, and \ref{lem:prec_subst_weak}.
        By \rnp{P\_Op},
        \[
          \langle   \emptyset    \vdash   S  \ottsym{(}  \mathit{op} \, \ottsym{(}  \ottnt{E'}  [  f'_{{\mathrm{13}}}  ]  \ottsym{,}  f'_{{\mathrm{12}}}  \ottsym{)}  \ottsym{)}  :  \iota   \sqsubseteq _{  S''  \uplus  \ottsym{(}   S'_{{\mathrm{0}}}  \circ   S  \circ  S_{{\mathrm{0}}}    \ottsym{)}  }  \iota  :  \mathit{op} \, \ottsym{(}  f'_{{\mathrm{2}}}  \ottsym{,}  S'  \ottsym{(}  f'_{{\mathrm{12}}}  \ottsym{)}  \ottsym{)}  \dashv   \emptyset   \rangle .
        \]
        Thus, we finish.

        \case{$\ottnt{E}  \ottsym{=}  \mathit{op} \, \ottsym{(}  f_{{\mathrm{11}}}  \ottsym{,}  \ottnt{E'}  \ottsym{)}$ for some $\ottnt{E'}$}
        Similar to the above.

        \otherwise Cannot happen.

       \case{\rnp{E\_Abort}} By \rnp{P\_Blame}.
      \end{caseanalysis}
    \end{caseanalysis}

    \case{\rnp{P\_App}}
    We are given $f_{{\mathrm{1}}}  \ottsym{=}  f_{{\mathrm{11}}} \, f_{{\mathrm{12}}}$ and $f'_{{\mathrm{1}}}  \ottsym{=}  f'_{{\mathrm{11}}} \, f'_{{\mathrm{12}}}$
    for some $f_{{\mathrm{11}}}$, $f_{{\mathrm{12}}}$, $f'_{{\mathrm{11}}}$, and $f'_{{\mathrm{12}}}$.

    By inversion,
    $ \langle   \emptyset    \vdash   f_{{\mathrm{11}}}  :  \ottnt{U_{{\mathrm{1}}}}  \!\rightarrow\!  \ottnt{U_{{\mathrm{2}}}}   \sqsubseteq _{ S_{{\mathrm{0}}} }  \ottnt{U'_{{\mathrm{1}}}}  \!\rightarrow\!  \ottnt{U'_{{\mathrm{2}}}}  :  f'_{{\mathrm{11}}}  \dashv   \emptyset   \rangle $ and
    $ \langle   \emptyset    \vdash   f_{{\mathrm{12}}}  :  \ottnt{U_{{\mathrm{2}}}}   \sqsubseteq _{ S_{{\mathrm{0}}} }  \ottnt{U'_{{\mathrm{2}}}}  :  f'_{{\mathrm{12}}}  \dashv   \emptyset   \rangle $.

    By case analysis on the reduction rule applied to $f_{{\mathrm{1}}}$.

    \begin{caseanalysis}
      \case{\rnp{E\_Step}}
      There exist $\ottnt{E}$ and $f_{{\mathrm{13}}}$ such that $\ottnt{E}  [  f_{{\mathrm{13}}}  ] \,  \xmapsto{ \mathmakebox[0.4em]{} S \mathmakebox[0.3em]{} }  \, S  \ottsym{(}  \ottnt{E}  [  f'_{{\mathrm{13}}}  ]  \ottsym{)}$.
      By inversion, $f_{{\mathrm{13}}} \,  \xrightarrow{ \mathmakebox[0.4em]{} S \mathmakebox[0.3em]{} }  \, f'_{{\mathrm{13}}}$.
      By case analysis on the structure of $\ottnt{E}$.

      \begin{caseanalysis}
        \case{$\ottnt{E}  \ottsym{=}  \left[ \, \right]$}
        Here, $f_{{\mathrm{11}}} \, f_{{\mathrm{12}}}  \ottsym{=}  f_{{\mathrm{13}}}$.
        By case analysis on the reduction rule applied to $f_{{\mathrm{13}}}$.

        \begin{caseanalysis}
         \case{\rnp{R\_Beta}}
         There exist $\ottnt{U_{{\mathrm{11}}}}$, $f_{{\mathrm{111}}}$, and $w_{{\mathrm{12}}}$ such that
         $f_{{\mathrm{11}}}  \ottsym{=}   \lambda  \ottmv{x} \!:\!  \ottnt{U_{{\mathrm{11}}}}  .\,  f_{{\mathrm{111}}} $, $f_{{\mathrm{12}}}  \ottsym{=}  w_{{\mathrm{12}}}$, and $f_{{\mathrm{2}}}  \ottsym{=}  f_{{\mathrm{111}}}  [  \ottmv{x}  \ottsym{:=}  w_{{\mathrm{12}}}  ]$.
         Here, $S  \ottsym{=}  [  ]$.
    
         By Lemma~\ref{lem:prec_catch_up_left_value},
         there exist $S'_{{\mathrm{0}}}$, $S'_{{\mathrm{11}}}$, and $w'_{{\mathrm{11}}}$ such that
         \begin{itemize}
          \item $f'_{{\mathrm{11}}} \,  \xmapsto{ \mathmakebox[0.4em]{} S'_{{\mathrm{11}}} \mathmakebox[0.3em]{} }\hspace{-0.4em}{}^\ast \hspace{0.2em}  \, w'_{{\mathrm{11}}}$,
          \item $ \langle   \emptyset    \vdash   f_{{\mathrm{11}}}  :  \ottnt{U_{{\mathrm{1}}}}  \!\rightarrow\!  \ottnt{U_{{\mathrm{2}}}}   \sqsubseteq _{  S'_{{\mathrm{0}}}  \circ  S_{{\mathrm{0}}}  }  S'_{{\mathrm{11}}}  \ottsym{(}  \ottnt{U'_{{\mathrm{1}}}}  \!\rightarrow\!  \ottnt{U'_{{\mathrm{2}}}}  \ottsym{)}  :  w'_{{\mathrm{11}}}  \dashv   \emptyset   \rangle $,
          \item $\forall \ottmv{X} \in \textit{dom} \, \ottsym{(}  S_{{\mathrm{0}}}  \ottsym{)}. S_{{\mathrm{0}}}  \ottsym{(}  \ottmv{X}  \ottsym{)}  \ottsym{=}   S'_{{\mathrm{0}}}  \circ   S_{{\mathrm{0}}}  \circ  S'_{{\mathrm{11}}}    \ottsym{(}  \ottmv{X}  \ottsym{)}$, and
          \item $\forall \ottmv{X} \in \textit{dom} \, \ottsym{(}  S'_{{\mathrm{11}}}  \ottsym{)}. \textit{ftv} \, \ottsym{(}  S'_{{\mathrm{11}}}  \ottsym{(}  \ottmv{X}  \ottsym{)}  \ottsym{)}  \subseteq  \textit{dom} \, \ottsym{(}  S'_{{\mathrm{0}}}  \ottsym{)}$.
         \end{itemize}
    
         So, $f'_{{\mathrm{11}}} \, f'_{{\mathrm{12}}} \,  \xmapsto{ \mathmakebox[0.4em]{} S'_{{\mathrm{11}}} \mathmakebox[0.3em]{} }\hspace{-0.4em}{}^\ast \hspace{0.2em}  \, w'_{{\mathrm{11}}} \, S'_{{\mathrm{11}}}  \ottsym{(}  f'_{{\mathrm{12}}}  \ottsym{)}$ by Lemma~\ref{lem:eval_subterm}.
         By Lemma~\ref{lem:right_subst_preserve_prec},
         $ \langle   \emptyset    \vdash   f_{{\mathrm{12}}}  :  \ottnt{U_{{\mathrm{2}}}}   \sqsubseteq _{  S'_{{\mathrm{0}}}  \circ  S_{{\mathrm{0}}}  }  S'_{{\mathrm{11}}}  \ottsym{(}  \ottnt{U'_{{\mathrm{2}}}}  \ottsym{)}  :  S'_{{\mathrm{11}}}  \ottsym{(}  f'_{{\mathrm{12}}}  \ottsym{)}  \dashv   \emptyset   \rangle $.
         By Lemma~\ref{lem:prec_catch_up_left_value},
         there exist $S''_{{\mathrm{0}}}$, $S'_{{\mathrm{12}}}$, and $w'_{{\mathrm{12}}}$ such that
         \begin{itemize}
          \item $S'_{{\mathrm{11}}}  \ottsym{(}  f'_{{\mathrm{12}}}  \ottsym{)} \,  \xmapsto{ \mathmakebox[0.4em]{} S'_{{\mathrm{12}}} \mathmakebox[0.3em]{} }\hspace{-0.4em}{}^\ast \hspace{0.2em}  \, w'_{{\mathrm{12}}}$,
          \item $ \langle   \emptyset    \vdash   f_{{\mathrm{12}}}  :  \ottnt{U_{{\mathrm{2}}}}   \sqsubseteq _{   S''_{{\mathrm{0}}}  \circ  S'_{{\mathrm{0}}}   \circ  S_{{\mathrm{0}}}  }  S'_{{\mathrm{12}}}  \ottsym{(}  \ottnt{U'_{{\mathrm{2}}}}  \ottsym{)}  :  w'_{{\mathrm{12}}}  \dashv   \emptyset   \rangle $,
          \item $\forall \ottmv{X} \in \textit{dom} \, \ottsym{(}   S'_{{\mathrm{0}}}  \circ  S_{{\mathrm{0}}}   \ottsym{)}.  S'_{{\mathrm{0}}}  \circ  S_{{\mathrm{0}}}   \ottsym{(}  \ottmv{X}  \ottsym{)}  \ottsym{=}   S''_{{\mathrm{0}}}  \circ    S'_{{\mathrm{0}}}  \circ  S_{{\mathrm{0}}}   \circ  S'_{{\mathrm{12}}}    \ottsym{(}  \ottmv{X}  \ottsym{)}$, and
          \item $\forall \ottmv{X} \in \textit{dom} \, \ottsym{(}  S'_{{\mathrm{12}}}  \ottsym{)}. \textit{ftv} \, \ottsym{(}  S'_{{\mathrm{12}}}  \ottsym{(}  \ottmv{X}  \ottsym{)}  \ottsym{)}  \subseteq  \textit{dom} \, \ottsym{(}  S''_{{\mathrm{0}}}  \ottsym{)}$.
         \end{itemize}
    
         So, $w'_{{\mathrm{11}}} \, S'_{{\mathrm{11}}}  \ottsym{(}  f'_{{\mathrm{12}}}  \ottsym{)} \,  \xmapsto{ \mathmakebox[0.4em]{} S'_{{\mathrm{12}}} \mathmakebox[0.3em]{} }\hspace{-0.4em}{}^\ast \hspace{0.2em}  \, S'_{{\mathrm{12}}}  \ottsym{(}  w'_{{\mathrm{11}}}  \ottsym{)} \, w'_{{\mathrm{12}}}$ by Lemma~\ref{lem:eval_subterm}.
         By Lemma~\ref{lem:right_subst_preserve_prec},
         $ \langle   \emptyset    \vdash   f_{{\mathrm{11}}}  :  \ottnt{U_{{\mathrm{1}}}}  \!\rightarrow\!  \ottnt{U_{{\mathrm{2}}}}   \sqsubseteq _{   S''_{{\mathrm{0}}}  \circ  S'_{{\mathrm{0}}}   \circ  S_{{\mathrm{0}}}  }   S'_{{\mathrm{12}}}  \circ  S'_{{\mathrm{11}}}   \ottsym{(}  \ottnt{U'_{{\mathrm{1}}}}  \!\rightarrow\!  \ottnt{U'_{{\mathrm{2}}}}  \ottsym{)}  :  S'_{{\mathrm{12}}}  \ottsym{(}  w'_{{\mathrm{11}}}  \ottsym{)}  \dashv   \emptyset   \rangle $.
         By Lemma~\ref{lem:term_prec_simulation_app},
         there exist $S'''_{{\mathrm{0}}}$, $S'$, and $f'_{{\mathrm{2}}}$ such that
         \begin{itemize}
          \item $S'_{{\mathrm{12}}}  \ottsym{(}  w'_{{\mathrm{11}}}  \ottsym{)} \, w'_{{\mathrm{12}}} \,  \xmapsto{ \mathmakebox[0.4em]{} S' \mathmakebox[0.3em]{} }\hspace{-0.4em}{}^\ast \hspace{0.2em}  \, f'_{{\mathrm{2}}}$ where \rnp{E\_Abort} is not applied,
          \item $ \langle   \emptyset    \vdash   f_{{\mathrm{2}}}  :  \ottnt{U_{{\mathrm{12}}}}   \sqsubseteq _{    S'''_{{\mathrm{0}}}  \circ  S''_{{\mathrm{0}}}   \circ  S'_{{\mathrm{0}}}   \circ  S_{{\mathrm{0}}}  }   S'  \circ   S'_{{\mathrm{12}}}  \circ  S'_{{\mathrm{11}}}    \ottsym{(}  \ottnt{U'_{{\mathrm{2}}}}  \ottsym{)}  :  f'_{{\mathrm{2}}}  \dashv   \emptyset   \rangle $,
          \item $\forall \ottmv{X} \in \textit{dom} \, \ottsym{(}    S''_{{\mathrm{0}}}  \circ  S'_{{\mathrm{0}}}   \circ  S_{{\mathrm{0}}}   \ottsym{)}.   S''_{{\mathrm{0}}}  \circ  S'_{{\mathrm{0}}}   \circ  S_{{\mathrm{0}}}   \ottsym{(}  \ottmv{X}  \ottsym{)}  \ottsym{=}   S'''_{{\mathrm{0}}}  \circ     S''_{{\mathrm{0}}}  \circ  S'_{{\mathrm{0}}}   \circ  S_{{\mathrm{0}}}   \circ  S'    \ottsym{(}  \ottmv{X}  \ottsym{)}$, and
          \item $\forall \ottmv{X} \in \textit{dom} \, \ottsym{(}  S'  \ottsym{)}. \textit{ftv} \, \ottsym{(}  S'  \ottsym{(}  \ottmv{X}  \ottsym{)}  \ottsym{)}  \subseteq  \textit{dom} \, \ottsym{(}  S'''_{{\mathrm{0}}}  \ottsym{)}$.
         \end{itemize}
    
         So, $S'_{{\mathrm{12}}}  \ottsym{(}  w'_{{\mathrm{11}}}  \ottsym{)} \, w'_{{\mathrm{12}}} \,  \xmapsto{ \mathmakebox[0.4em]{} S' \mathmakebox[0.3em]{} }  \, f'_{{\mathrm{2}}}$.
         Finally,
         \begin{itemize}
          \item $f'_{{\mathrm{11}}} \, f'_{{\mathrm{12}}} \,  \xmapsto{ \mathmakebox[0.4em]{}   S'  \circ  S'_{{\mathrm{12}}}   \circ  S'_{{\mathrm{11}}}  \mathmakebox[0.3em]{} }\hspace{-0.4em}{}^\ast \hspace{0.2em}  \, f'_{{\mathrm{2}}}$,
          \item $ \langle   \emptyset    \vdash   f_{{\mathrm{2}}}  :  \ottnt{U_{{\mathrm{12}}}}   \sqsubseteq _{    S'''_{{\mathrm{0}}}  \circ  S''_{{\mathrm{0}}}   \circ  S'_{{\mathrm{0}}}   \circ  S_{{\mathrm{0}}}  }   S'  \circ   S'_{{\mathrm{12}}}  \circ  S'_{{\mathrm{11}}}    \ottsym{(}  \ottnt{U'_{{\mathrm{2}}}}  \ottsym{)}  :  f'_{{\mathrm{2}}}  \dashv   \emptyset   \rangle $,
          \item $\forall \ottmv{X} \in \textit{dom} \, \ottsym{(}  S_{{\mathrm{0}}}  \ottsym{)}. S_{{\mathrm{0}}}  \ottsym{(}  \ottmv{X}  \ottsym{)}  \ottsym{=}   S'''_{{\mathrm{0}}}  \circ       S''_{{\mathrm{0}}}  \circ  S'_{{\mathrm{0}}}   \circ  S_{{\mathrm{0}}}   \circ  S'   \circ  S'_{{\mathrm{12}}}   \circ  S'_{{\mathrm{11}}}    \ottsym{(}  \ottmv{X}  \ottsym{)}$, and
          \item $\forall \ottmv{X} \in \textit{dom} \, \ottsym{(}    S'  \circ  S'_{{\mathrm{12}}}   \circ  S'_{{\mathrm{11}}}   \ottsym{)}. \textit{ftv} \, \ottsym{(}   S'  \circ   S'_{{\mathrm{12}}}  \circ  S'_{{\mathrm{11}}}    \ottsym{(}  \ottmv{X}  \ottsym{)}  \ottsym{)}  \subseteq  \textit{dom} \, \ottsym{(}    S'''_{{\mathrm{0}}}  \circ  S''_{{\mathrm{0}}}   \circ  S'_{{\mathrm{0}}}   \ottsym{)}$.
         \end{itemize}
    
         \case{\rnp{R\_AppCast}}
         We are given $f_{{\mathrm{11}}}  \ottsym{=}  w_{{\mathrm{11}}}  \ottsym{:}   \ottnt{U_{{\mathrm{11}}}}  \!\rightarrow\!  \ottnt{U_{{\mathrm{12}}}} \Rightarrow  \unskip ^ { \ell_{{\mathrm{1}}} }  \! \ottnt{U_{{\mathrm{13}}}}  \!\rightarrow\!  \ottnt{U_{{\mathrm{14}}}} $, $f_{{\mathrm{12}}}  \ottsym{=}  w_{{\mathrm{12}}}$,
         $f_{{\mathrm{2}}}  \ottsym{=}  \ottsym{(}  w_{{\mathrm{11}}} \, \ottsym{(}  w_{{\mathrm{12}}}  \ottsym{:}   \ottnt{U_{{\mathrm{13}}}} \Rightarrow  \unskip ^ {  \bar{ \ell_{{\mathrm{1}}} }  }  \! \ottnt{U_{{\mathrm{11}}}}   \ottsym{)}  \ottsym{)}  \ottsym{:}   \ottnt{U_{{\mathrm{12}}}} \Rightarrow  \unskip ^ { \ell_{{\mathrm{1}}} }  \! \ottnt{U_{{\mathrm{14}}}} $ for some
         $w_{{\mathrm{11}}}$, $w_{{\mathrm{12}}}$, $\ottnt{U_{{\mathrm{11}}}}$, $\ottnt{U_{{\mathrm{12}}}}$,
         $\ottnt{U_{{\mathrm{13}}}}$, $\ottnt{U_{{\mathrm{14}}}}$, and $\ell_{{\mathrm{1}}}$.
         Here, $S  \ottsym{=}  [  ]$.
    
         By Lemma~\ref{lem:prec_catch_up_left_value},
         there exist $S'_{{\mathrm{0}}}$, $S'_{{\mathrm{11}}}$, and $w'_{{\mathrm{11}}}$ such that
         \begin{itemize}
          \item $f'_{{\mathrm{11}}} \,  \xmapsto{ \mathmakebox[0.4em]{} S'_{{\mathrm{11}}} \mathmakebox[0.3em]{} }\hspace{-0.4em}{}^\ast \hspace{0.2em}  \, w'_{{\mathrm{11}}}$,
          \item $ \langle   \emptyset    \vdash   f_{{\mathrm{11}}}  :  \ottnt{U_{{\mathrm{1}}}}  \!\rightarrow\!  \ottnt{U_{{\mathrm{2}}}}   \sqsubseteq _{  S'_{{\mathrm{0}}}  \circ  S_{{\mathrm{0}}}  }  S'_{{\mathrm{11}}}  \ottsym{(}  \ottnt{U'_{{\mathrm{1}}}}  \!\rightarrow\!  \ottnt{U'_{{\mathrm{2}}}}  \ottsym{)}  :  w'_{{\mathrm{11}}}  \dashv   \emptyset   \rangle $,
          \item $\forall \ottmv{X} \in \textit{dom} \, \ottsym{(}  S_{{\mathrm{0}}}  \ottsym{)}. S_{{\mathrm{0}}}  \ottsym{(}  \ottmv{X}  \ottsym{)}  \ottsym{=}   S'_{{\mathrm{0}}}  \circ   S_{{\mathrm{0}}}  \circ  S'_{{\mathrm{11}}}    \ottsym{(}  \ottmv{X}  \ottsym{)}$, and
          \item $\forall \ottmv{X} \in \textit{dom} \, \ottsym{(}  S'_{{\mathrm{11}}}  \ottsym{)}. \textit{ftv} \, \ottsym{(}  S'_{{\mathrm{11}}}  \ottsym{(}  \ottmv{X}  \ottsym{)}  \ottsym{)}  \subseteq  \textit{dom} \, \ottsym{(}  S'_{{\mathrm{0}}}  \ottsym{)}$.
         \end{itemize}
         
         So, $f'_{{\mathrm{11}}} \, f'_{{\mathrm{12}}} \,  \xmapsto{ \mathmakebox[0.4em]{} S'_{{\mathrm{11}}} \mathmakebox[0.3em]{} }\hspace{-0.4em}{}^\ast \hspace{0.2em}  \, w'_{{\mathrm{11}}} \, S'_{{\mathrm{11}}}  \ottsym{(}  f'_{{\mathrm{12}}}  \ottsym{)}$ by Lemma~\ref{lem:eval_subterm}.
         By Lemma~\ref{lem:right_subst_preserve_prec},
         $ \langle   \emptyset    \vdash   f_{{\mathrm{12}}}  :  \ottnt{U_{{\mathrm{2}}}}   \sqsubseteq _{  S'_{{\mathrm{0}}}  \circ  S_{{\mathrm{0}}}  }  S'_{{\mathrm{11}}}  \ottsym{(}  \ottnt{U'_{{\mathrm{2}}}}  \ottsym{)}  :  S'_{{\mathrm{11}}}  \ottsym{(}  f'_{{\mathrm{12}}}  \ottsym{)}  \dashv   \emptyset   \rangle $.
         By Lemma~\ref{lem:prec_catch_up_left_value},
         there exist $S''_{{\mathrm{0}}}$, $S'_{{\mathrm{12}}}$, and $w'_{{\mathrm{12}}}$ such that
         \begin{itemize}
          \item $f'_{{\mathrm{12}}} \,  \xmapsto{ \mathmakebox[0.4em]{} S'_{{\mathrm{12}}} \mathmakebox[0.3em]{} }\hspace{-0.4em}{}^\ast \hspace{0.2em}  \, w'_{{\mathrm{12}}}$,
          \item $ \langle   \emptyset    \vdash   f_{{\mathrm{12}}}  :  \ottnt{U_{{\mathrm{2}}}}   \sqsubseteq _{   S''_{{\mathrm{0}}}  \circ  S'_{{\mathrm{0}}}   \circ  S_{{\mathrm{0}}}  }  S'_{{\mathrm{12}}}  \ottsym{(}  \ottnt{U'_{{\mathrm{2}}}}  \ottsym{)}  :  w'_{{\mathrm{12}}}  \dashv   \emptyset   \rangle $,
          \item $\forall \ottmv{X} \in \textit{dom} \, \ottsym{(}   S'_{{\mathrm{0}}}  \circ  S_{{\mathrm{0}}}   \ottsym{)}.  S'_{{\mathrm{0}}}  \circ  S_{{\mathrm{0}}}   \ottsym{(}  \ottmv{X}  \ottsym{)}  \ottsym{=}   S''_{{\mathrm{0}}}  \circ    S'_{{\mathrm{0}}}  \circ  S_{{\mathrm{0}}}   \circ  S'_{{\mathrm{12}}}    \ottsym{(}  \ottmv{X}  \ottsym{)}$, and
          \item $\forall \ottmv{X} \in \textit{dom} \, \ottsym{(}  S'_{{\mathrm{12}}}  \ottsym{)}. \textit{ftv} \, \ottsym{(}  S'_{{\mathrm{12}}}  \ottsym{(}  \ottmv{X}  \ottsym{)}  \ottsym{)}  \subseteq  \textit{dom} \, \ottsym{(}  S''_{{\mathrm{0}}}  \ottsym{)}$.
         \end{itemize}
         
         So, $w'_{{\mathrm{11}}} \, S'_{{\mathrm{11}}}  \ottsym{(}  f'_{{\mathrm{12}}}  \ottsym{)} \,  \xmapsto{ \mathmakebox[0.4em]{} S'_{{\mathrm{12}}} \mathmakebox[0.3em]{} }\hspace{-0.4em}{}^\ast \hspace{0.2em}  \, S'_{{\mathrm{12}}}  \ottsym{(}  w'_{{\mathrm{11}}}  \ottsym{)} \, w'_{{\mathrm{12}}}$ by Lemma~\ref{lem:eval_subterm}.
    
         By Lemma~\ref{lem:right_subst_preserve_prec},
         $ \langle   \emptyset    \vdash   f_{{\mathrm{11}}}  :  \ottnt{U_{{\mathrm{1}}}}  \!\rightarrow\!  \ottnt{U_{{\mathrm{2}}}}   \sqsubseteq _{   S''_{{\mathrm{0}}}  \circ  S'_{{\mathrm{0}}}   \circ  S_{{\mathrm{0}}}  }   S'_{{\mathrm{12}}}  \circ  S'_{{\mathrm{11}}}   \ottsym{(}  \ottnt{U'_{{\mathrm{1}}}}  \!\rightarrow\!  \ottnt{U'_{{\mathrm{2}}}}  \ottsym{)}  :  S'_{{\mathrm{12}}}  \ottsym{(}  w'_{{\mathrm{11}}}  \ottsym{)}  \dashv   \emptyset   \rangle $.
         By Lemma~\ref{lem:term_prec_simulation_unwrap},
         there exist $S'''_{{\mathrm{0}}}$, $S'$, and $f'_{{\mathrm{2}}}$ such that
         \begin{itemize}
          \item $S'_{{\mathrm{12}}}  \ottsym{(}  w'_{{\mathrm{11}}}  \ottsym{)} \, w'_{{\mathrm{12}}} \,  \xmapsto{ \mathmakebox[0.4em]{} S' \mathmakebox[0.3em]{} }\hspace{-0.4em}{}^\ast \hspace{0.2em}  \, f'_{{\mathrm{2}}}$ where \rnp{E\_Abort} is not applied,
          \item $ \langle   \emptyset    \vdash   f_{{\mathrm{2}}}  :  \ottnt{U}   \sqsubseteq _{    S'''_{{\mathrm{0}}}  \circ  S''_{{\mathrm{0}}}   \circ  S'_{{\mathrm{0}}}   \circ  S_{{\mathrm{0}}}  }  \ottnt{U'}  :  f'_{{\mathrm{2}}}  \dashv   \emptyset   \rangle $,
          \item $\forall \ottmv{X} \in \textit{dom} \, \ottsym{(}    S''_{{\mathrm{0}}}  \circ  S'_{{\mathrm{0}}}   \circ  S_{{\mathrm{0}}}   \ottsym{)}.   S''_{{\mathrm{0}}}  \circ  S'_{{\mathrm{0}}}   \circ  S_{{\mathrm{0}}}   \ottsym{(}  \ottmv{X}  \ottsym{)}  \ottsym{=}   S'''_{{\mathrm{0}}}  \circ     S''_{{\mathrm{0}}}  \circ  S'_{{\mathrm{0}}}   \circ  S_{{\mathrm{0}}}   \circ  S'    \ottsym{(}  \ottmv{X}  \ottsym{)}$, and
          \item $\forall \ottmv{X} \in \textit{dom} \, \ottsym{(}  S'  \ottsym{)}. \textit{ftv} \, \ottsym{(}  S'  \ottsym{(}  \ottmv{X}  \ottsym{)}  \ottsym{)}  \subseteq  \textit{dom} \, \ottsym{(}  S'''_{{\mathrm{0}}}  \ottsym{)}$.
         \end{itemize}
    
         Finally,
         \begin{itemize}
          \item $f'_{{\mathrm{11}}} \, f'_{{\mathrm{12}}} \,  \xmapsto{ \mathmakebox[0.4em]{}   S'  \circ  S'_{{\mathrm{12}}}   \circ  S'_{{\mathrm{11}}}  \mathmakebox[0.3em]{} }\hspace{-0.4em}{}^\ast \hspace{0.2em}  \, f'_{{\mathrm{2}}}$,
          \item $ \langle   \emptyset    \vdash   f_{{\mathrm{2}}}  :  \ottnt{U}   \sqsubseteq _{    S'''_{{\mathrm{0}}}  \circ  S''_{{\mathrm{0}}}   \circ  S'_{{\mathrm{0}}}   \circ  S_{{\mathrm{0}}}  }  \ottnt{U'}  :  f'_{{\mathrm{2}}}  \dashv   \emptyset   \rangle $, and
          \item $\forall \ottmv{X} \in \textit{dom} \, \ottsym{(}  S_{{\mathrm{0}}}  \ottsym{)}. S_{{\mathrm{0}}}  \ottsym{(}  \ottmv{X}  \ottsym{)}  \ottsym{=}   S'''_{{\mathrm{0}}}  \circ       S''_{{\mathrm{0}}}  \circ  S'_{{\mathrm{0}}}   \circ  S_{{\mathrm{0}}}   \circ  S'   \circ  S'_{{\mathrm{12}}}   \circ  S'_{{\mathrm{11}}}    \ottsym{(}  \ottmv{X}  \ottsym{)}$, and
          \item $\forall \ottmv{X} \in \textit{dom} \, \ottsym{(}    S'  \circ  S'_{{\mathrm{12}}}   \circ  S'_{{\mathrm{11}}}   \ottsym{)}. \textit{ftv} \, \ottsym{(}   S'  \circ   S'_{{\mathrm{12}}}  \circ  S'_{{\mathrm{11}}}    \ottsym{(}  \ottmv{X}  \ottsym{)}  \ottsym{)}  \subseteq  \textit{dom} \, \ottsym{(}    S'''_{{\mathrm{0}}}  \circ  S''_{{\mathrm{0}}}   \circ  S'_{{\mathrm{0}}}   \ottsym{)}$.
         \end{itemize}

          \otherwise Contradiction.
        \end{caseanalysis}

        \case{$\ottnt{E}  \ottsym{=}  \ottnt{E_{{\mathrm{1}}}} \, f_{{\mathrm{12}}}$ for some $\ottnt{E_{{\mathrm{1}}}}$}
        Similar to the case of \rnp{P\_Op}.

        \case{$\ottnt{E}  \ottsym{=}  w_{{\mathrm{11}}} \, \ottnt{E_{{\mathrm{1}}}}$ for some $\ottnt{E_{{\mathrm{1}}}}$ and $w_{{\mathrm{11}}}$ where $f_{{\mathrm{11}}}  \ottsym{=}  w_{{\mathrm{11}}}$}
        Similar to the case of \rnp{P\_Op}.

        \otherwise
        Contradiction.
      \end{caseanalysis}

      \case{\rnp{E\_Abort}}
      Here, $f_{{\mathrm{11}}} \, f_{{\mathrm{12}}} \,  \xmapsto{ \mathmakebox[0.4em]{} [  ] \mathmakebox[0.3em]{} }  \, \textsf{\textup{blame}\relax} \, \ell$ for some $\ell$.
      We finish by \rnp{P\_Blame}.
    \end{caseanalysis}

    \case{\rnp{P\_LetP}}
    We are given $f_{{\mathrm{1}}}  \ottsym{=}   \textsf{\textup{let}\relax} \,  \ottmv{x}  =   \Lambda    \overrightarrow{ \ottmv{X_{\ottmv{i}}} }  .\,  w_{{\mathrm{11}}}   \textsf{\textup{ in }\relax}  f_{{\mathrm{12}}} $ and
    $f'_{{\mathrm{1}}}  \ottsym{=}   \textsf{\textup{let}\relax} \,  \ottmv{x}  =   \Lambda    \overrightarrow{ \ottmv{X'_{\ottmv{j}}} }  .\,  w'_{{\mathrm{11}}}   \textsf{\textup{ in }\relax}  f'_{{\mathrm{12}}} $
    for some $ \overrightarrow{ \ottmv{X_{\ottmv{i}}} } $, $ \overrightarrow{ \ottmv{X'_{\ottmv{j}}} } $, $w_{{\mathrm{11}}}$, $w'_{{\mathrm{11}}}$,
    $f_{{\mathrm{12}}}$, and $f'_{{\mathrm{12}}}$.
    By inversion,
    \begin{itemize}
     \item $ \langle   \emptyset    \vdash   w_{{\mathrm{11}}}  :  \ottnt{U_{{\mathrm{1}}}}   \sqsubseteq _{  [   \overrightarrow{ \ottmv{X'_{\ottmv{j}}} }   :=   \overrightarrow{ \ottnt{T'_{\ottmv{j}}} }   ]  \uplus  S_{{\mathrm{0}}}  }  \ottnt{U'_{{\mathrm{1}}}}  :  w'_{{\mathrm{11}}}  \dashv   \emptyset   \rangle $,
     \item $ \langle   \ottmv{x}  :  \forall \,  \overrightarrow{ \ottmv{X_{\ottmv{i}}} }   \ottsym{.}  \ottnt{U_{{\mathrm{1}}}}    \vdash   f_{{\mathrm{12}}}  :  \ottnt{U}   \sqsubseteq _{ S_{{\mathrm{0}}} }  \ottnt{U'}  :  f'_{{\mathrm{12}}}  \dashv   \ottmv{x'}  :  \forall \,  \overrightarrow{ \ottmv{X'_{\ottmv{j}}} }   \ottsym{.}  \ottnt{U'_{{\mathrm{1}}}}   \rangle $,
     \item for any $\ottmv{X} \, \in \, \textit{dom} \, \ottsym{(}  S_{{\mathrm{0}}}  \ottsym{)}$, $ \overrightarrow{ \ottmv{X_{\ottmv{i}}} }   \cap  \textit{ftv} \, \ottsym{(}  S  \ottsym{(}  \ottmv{X}  \ottsym{)}  \ottsym{)}  \ottsym{=}   \emptyset $, and
     \item $\textit{ftv} \, \ottsym{(}  \Gamma  \ottsym{)}  \cap   \overrightarrow{ \ottmv{X_{\ottmv{i}}} }   \ottsym{=}   \emptyset $ and $\textit{ftv} \, \ottsym{(}  \Gamma'  \ottsym{)}  \cap   \overrightarrow{ \ottmv{X'_{\ottmv{j}}} }   \ottsym{=}   \emptyset $.
    \end{itemize}
    By case analysis on the evaluation rule applied to $f_{{\mathrm{1}}}$.

    \begin{caseanalysis}
      \case{\rnp{E\_Step}}
      There exists $\ottnt{E}$ and $f_{{\mathrm{13}}}$ such that $\ottnt{E}  [  f_{{\mathrm{13}}}  ] \,  \xmapsto{ \mathmakebox[0.4em]{} S \mathmakebox[0.3em]{} }  \, S  \ottsym{(}  \ottnt{E}  [  f'_{{\mathrm{13}}}  ]  \ottsym{)}$.
      By inversion, $f_{{\mathrm{13}}} \,  \xrightarrow{ \mathmakebox[0.4em]{} S \mathmakebox[0.3em]{} }  \, f'_{{\mathrm{13}}}$.
      By case analysis on the structure of $\ottnt{E}$.

      \begin{caseanalysis}
        \case{$\ottnt{E}  \ottsym{=}  \left[ \, \right]$}
        Here, $f_{{\mathrm{1}}}  \ottsym{=}  f_{{\mathrm{13}}}$.
        By case analysis on the reduction rule applied to $f_{{\mathrm{13}}}$.

        \begin{caseanalysis}
          \case{\rnp{R\_LetP}}
          We are given $ \textsf{\textup{let}\relax} \,  \ottmv{x}  =   \Lambda    \overrightarrow{ \ottmv{X_{\ottmv{i}}} }  .\,  w_{{\mathrm{11}}}   \textsf{\textup{ in }\relax}  f_{{\mathrm{12}}}  \,  \xmapsto{ \mathmakebox[0.4em]{} [  ] \mathmakebox[0.3em]{} }  \, f_{{\mathrm{12}}}  [  \ottmv{x}  \ottsym{:=}   \Lambda    \overrightarrow{ \ottmv{X_{\ottmv{i}}} }  .\,  w_{{\mathrm{11}}}   ]$.
          By \rnp{R\_LetP},
          $ \textsf{\textup{let}\relax} \,  \ottmv{x}  =   \Lambda    \overrightarrow{ \ottmv{X'_{\ottmv{j}}} }  .\,  w'_{{\mathrm{11}}}   \textsf{\textup{ in }\relax}  f'_{{\mathrm{12}}}  \,  \xmapsto{ \mathmakebox[0.4em]{} [  ] \mathmakebox[0.3em]{} }  \, f'_{{\mathrm{12}}}  [  \ottmv{x}  \ottsym{:=}   \Lambda    \overrightarrow{ \ottmv{X'_{\ottmv{j}}} }  .\,  w'_{{\mathrm{11}}}   ]$.
          Without loss of generality, we can suppose that $ \overrightarrow{ \ottmv{X_{\ottmv{i}}} } $ and $ \overrightarrow{ \ottmv{X'_{\ottmv{j}}} } $
          do not occur in the derivation of
          $ \langle   \ottmv{x}  :  \forall \,  \overrightarrow{ \ottmv{X_{\ottmv{i}}} }   \ottsym{.}  \ottnt{U_{{\mathrm{1}}}}    \vdash   f_{{\mathrm{12}}}  :  \ottnt{U}   \sqsubseteq _{ S_{{\mathrm{0}}} }  \ottnt{U'}  :  f'_{{\mathrm{12}}}  \dashv   \ottmv{x'}  :  \forall \,  \overrightarrow{ \ottmv{X'_{\ottmv{j}}} }   \ottsym{.}  \ottnt{U'_{{\mathrm{1}}}}   \rangle $.
          We finish by Lemma~\ref{lem:subst_preserve_term_prec}.

          \otherwise Cannot happen.
        \end{caseanalysis}

        \otherwise Cannot happen.
      \end{caseanalysis}

      \case{\rnp{E\_Abort}}
      Cannot happen.
    \end{caseanalysis}

    \case{\rnp{P\_Cast}}
    We are given $f_{{\mathrm{1}}}  \ottsym{=}  \ottsym{(}  f_{{\mathrm{11}}}  \ottsym{:}   \ottnt{U_{{\mathrm{1}}}} \Rightarrow  \unskip ^ { \ell }  \! \ottnt{U}   \ottsym{)}$ and $f'_{{\mathrm{1}}}  \ottsym{=}  \ottsym{(}  f'_{{\mathrm{11}}}  \ottsym{:}   \ottnt{U'_{{\mathrm{1}}}} \Rightarrow  \unskip ^ { \ell' }  \! \ottnt{U'}   \ottsym{)}$
    for some $f_{{\mathrm{11}}}$, $f'_{{\mathrm{11}}}$, $\ottnt{U_{{\mathrm{1}}}}$, $\ottnt{U'_{{\mathrm{1}}}}$, $\ell$, and $\ell'$.

    By inversion,
    $ \langle   \emptyset    \vdash   f_{{\mathrm{11}}}  :  \ottnt{U_{{\mathrm{1}}}}   \sqsubseteq _{ S_{{\mathrm{0}}} }  \ottnt{U'_{{\mathrm{1}}}}  :  f'_{{\mathrm{11}}}  \dashv   \emptyset   \rangle $,
    $ \ottnt{U}   \sqsubseteq _{ S_{{\mathrm{0}}} }  \ottnt{U'} $, $\ottnt{U_{{\mathrm{1}}}}  \sim  \ottnt{U}$, and $\ottnt{U'_{{\mathrm{1}}}}  \sim  \ottnt{U'}$.

    By case analysis on the evaluation rule applied to $f_{{\mathrm{1}}}$.

    \begin{caseanalysis}
      \case{\rnp{E\_Step}}
      There exist $\ottnt{E}$ and $f_{{\mathrm{12}}}$ such that $\ottnt{E}  [  f_{{\mathrm{12}}}  ] \,  \xmapsto{ \mathmakebox[0.4em]{} S \mathmakebox[0.3em]{} }  \, S  \ottsym{(}  \ottnt{E}  [  f'_{{\mathrm{12}}}  ]  \ottsym{)}$.
      By inversion, $f_{{\mathrm{12}}} \,  \xrightarrow{ \mathmakebox[0.4em]{} S \mathmakebox[0.3em]{} }  \, f'_{{\mathrm{12}}}$.
      By case analysis on the structure of $\ottnt{E}$.

      \begin{caseanalysis}
        \case{$\ottnt{E}  \ottsym{=}  \left[ \, \right]$}
        Here, $f_{{\mathrm{11}}}  \ottsym{:}   \ottnt{U_{{\mathrm{1}}}} \Rightarrow  \unskip ^ { \ell }  \! \ottnt{U}   \ottsym{=}  f_{{\mathrm{12}}}$.
        By case analysis on the reduction rule applied to $f_{{\mathrm{12}}}$.

        \begin{caseanalysis}
          \case{\rnp{R\_IdBase}}
          There exist $w_{{\mathrm{11}}}$ and $\iota$ such that
          $f_{{\mathrm{11}}}  \ottsym{=}  w_{{\mathrm{11}}}$, $\ottnt{U_{{\mathrm{1}}}}  \ottsym{=}  \iota$ and $\ottnt{U}  \ottsym{=}  \iota$.
          We have $w_{{\mathrm{11}}}  \ottsym{:}   \iota \Rightarrow  \unskip ^ { \ell }  \! \iota  \,  \xmapsto{ \mathmakebox[0.4em]{} [  ] \mathmakebox[0.3em]{} }  \, w_{{\mathrm{11}}}$.
          By Lemma~\ref{lem:term_prec_to_type_prec},
          $ \iota   \sqsubseteq _{ S_{{\mathrm{0}}} }  \ottnt{U'} $.
          By \rnp{P\_CastR},
          $ \langle   \emptyset    \vdash   w_{{\mathrm{11}}}  :  \iota   \sqsubseteq _{ S_{{\mathrm{0}}} }  \ottnt{U'}  :  f'_{{\mathrm{11}}}  \ottsym{:}   \ottnt{U'_{{\mathrm{1}}}} \Rightarrow  \unskip ^ { \ell' }  \! \ottnt{U'}   \dashv   \emptyset   \rangle $.

          \case{\rnp{R\_IdStar}}
          Similar to the case of \rnp{R\_IdBase}.

          \case{\rnp{R\_Succeed}}
          There exist $w_{{\mathrm{111}}}$ and $\ottnt{G}$ such that
          $f_{{\mathrm{1}}}  \ottsym{=}  \ottsym{(}  w_{{\mathrm{111}}}  \ottsym{:}   \ottnt{G} \Rightarrow  \unskip ^ { \ell_{{\mathrm{1}}} }  \!  \star \Rightarrow  \unskip ^ { \ell }  \! \ottnt{G}    \ottsym{)}$, $\ottnt{U_{{\mathrm{1}}}}  \ottsym{=}  \star$ and $\ottnt{U}  \ottsym{=}  \ottnt{G}$.
          We have $w_{{\mathrm{111}}}  \ottsym{:}   \ottnt{G} \Rightarrow  \unskip ^ { \ell_{{\mathrm{1}}} }  \!  \star \Rightarrow  \unskip ^ { \ell }  \! \ottnt{G}   \,  \xmapsto{ \mathmakebox[0.4em]{} [  ] \mathmakebox[0.3em]{} }  \, w_{{\mathrm{111}}}$.
          Here, $ \langle   \emptyset    \vdash   w_{{\mathrm{111}}}  \ottsym{:}   \ottnt{G} \Rightarrow  \unskip ^ { \ell_{{\mathrm{1}}} }  \! \star   :  \star   \sqsubseteq _{ S_{{\mathrm{0}}} }  \ottnt{U'_{{\mathrm{1}}}}  :  f'_{{\mathrm{11}}}  \dashv   \emptyset   \rangle $.
          By Lemma~\ref{lem:term_prec_inversion3},
          $ \langle   \emptyset    \vdash   w_{{\mathrm{111}}}  :  \ottnt{G}   \sqsubseteq _{ S_{{\mathrm{0}}} }  \ottnt{U'_{{\mathrm{1}}}}  :  f'_{{\mathrm{11}}}  \dashv   \emptyset   \rangle $.
          By Lemma~\ref{lem:term_prec_to_type_prec},
          $ \ottnt{G}   \sqsubseteq _{ S_{{\mathrm{0}}} }  \ottnt{U'} $.
          By \rnp{P\_CastR},
          $ \langle   \emptyset    \vdash   w_{{\mathrm{111}}}  :  \ottnt{G}   \sqsubseteq _{ S_{{\mathrm{0}}} }  \ottnt{U'}  :  f'_{{\mathrm{11}}}  \ottsym{:}   \ottnt{U'_{{\mathrm{1}}}} \Rightarrow  \unskip ^ { \ell' }  \! \ottnt{U'}   \dashv   \emptyset   \rangle $.

          \case{\rnp{R\_Fail}}
          There exist $w_{{\mathrm{111}}}$, $\ottnt{G}$, $\ottnt{G'}$ such that
          $f_{{\mathrm{1}}}  \ottsym{=}  \ottsym{(}  w_{{\mathrm{111}}}  \ottsym{:}   \ottnt{G'} \Rightarrow  \unskip ^ { \ell_{{\mathrm{1}}} }  \!  \star \Rightarrow  \unskip ^ { \ell }  \! \ottnt{G}    \ottsym{)}$, $\ottnt{U_{{\mathrm{1}}}}  \ottsym{=}  \star$ and $\ottnt{U}  \ottsym{=}  \ottnt{G}$.
          We have $w_{{\mathrm{111}}}  \ottsym{:}   \ottnt{G'} \Rightarrow  \unskip ^ { \ell_{{\mathrm{1}}} }  \!  \star \Rightarrow  \unskip ^ { \ell }  \! \ottnt{G}   \,  \xmapsto{ \mathmakebox[0.4em]{} [  ] \mathmakebox[0.3em]{} }  \, \textsf{\textup{blame}\relax} \, \ell$.
          We finish by \rnp{P\_Blame} with Lemma~\ref{lem:term_prec_to_typing}.

          \case{\rnp{R\_Ground}}
          There exist $w_{{\mathrm{11}}}$ and $\ottnt{G}$ such that
          $f_{{\mathrm{11}}}  \ottsym{=}  w_{{\mathrm{11}}}$, $\ottnt{U}  \ottsym{=}  \star$, and $\ottnt{U_{{\mathrm{1}}}}  \sim  \ottnt{G}$.
          We have $w_{{\mathrm{11}}}  \ottsym{:}   \ottnt{U_{{\mathrm{1}}}} \Rightarrow  \unskip ^ { \ell }  \! \star  \,  \xmapsto{ \mathmakebox[0.4em]{} [  ] \mathmakebox[0.3em]{} }  \, w_{{\mathrm{11}}}  \ottsym{:}   \ottnt{U_{{\mathrm{1}}}} \Rightarrow  \unskip ^ { \ell }  \!  \ottnt{G} \Rightarrow  \unskip ^ { \ell }  \! \star  $ and
          $\ottnt{U_{{\mathrm{1}}}}  \neq  \star$ and $\ottnt{U_{{\mathrm{1}}}}  \neq  \ottnt{G}$.
          Since $ \ottnt{U}   \sqsubseteq _{ S_{{\mathrm{0}}} }  \ottnt{U'} $ and $\ottnt{U}  \ottsym{=}  \star$, we have $\ottnt{U'}  \ottsym{=}  \star$.
          Since $ \ottnt{G}   \sqsubseteq _{ S_{{\mathrm{0}}} }  \star $, we have
          \[
           \langle   \emptyset    \vdash   w_{{\mathrm{11}}}  \ottsym{:}   \ottnt{U_{{\mathrm{1}}}} \Rightarrow  \unskip ^ { \ell_{{\mathrm{1}}} }  \! \ottnt{G}   :  \ottnt{G}   \sqsubseteq _{ S_{{\mathrm{0}}} }  \star  :  f'_{{\mathrm{11}}}  \ottsym{:}   \ottnt{U'_{{\mathrm{1}}}} \Rightarrow  \unskip ^ { \ell' }  \! \star   \dashv   \emptyset   \rangle 
          \]
          by \rnp{P\_Cast}.
          We finish by \rnp{P\_CastL}.

          \case{\rnp{R\_Expand}}
          There exist $w_{{\mathrm{11}}}$ and $\ottnt{G}$ such that
          $f_{{\mathrm{11}}}  \ottsym{=}  w_{{\mathrm{11}}}$, $\ottnt{U_{{\mathrm{1}}}}  \ottsym{=}  \star$, and $\ottnt{U}  \sim  \ottnt{G}$.
          We have $w_{{\mathrm{11}}}  \ottsym{:}   \star \Rightarrow  \unskip ^ { \ell }  \! \ottnt{U}  \,  \xmapsto{ \mathmakebox[0.4em]{} [  ] \mathmakebox[0.3em]{} }  \, w_{{\mathrm{11}}}  \ottsym{:}   \star \Rightarrow  \unskip ^ { \ell }  \!  \ottnt{G} \Rightarrow  \unskip ^ { \ell }  \! \ottnt{U}  $,
          $\ottnt{U}  \neq  \star$, and $\ottnt{U}  \neq  \ottnt{G}$.
          By Lemma~\ref{lem:term_prec_to_type_prec},
          $ \ottnt{U_{{\mathrm{1}}}}   \sqsubseteq _{ S_{{\mathrm{0}}} }  \ottnt{U'_{{\mathrm{1}}}} $.
          Since $\ottnt{U_{{\mathrm{1}}}}  \ottsym{=}  \star$, we have $\ottnt{U'_{{\mathrm{1}}}}  \ottsym{=}  \star$.
          Since $ \ottnt{G}   \sqsubseteq _{ S_{{\mathrm{0}}} }  \star $,
          we have
          \[
           \langle   \emptyset    \vdash   w_{{\mathrm{11}}}  \ottsym{:}   \star \Rightarrow  \unskip ^ { \ell }  \! \ottnt{G}   :  \ottnt{G}   \sqsubseteq _{ S_{{\mathrm{0}}} }  \star  :  f'_{{\mathrm{1}}}  \dashv   \emptyset   \rangle 
          \]
          by \rnp{P\_CastL}.
          We finish by \rnp{P\_Cast}.

          \case{\rnp{R\_InstBase}}
          There exist $w_{{\mathrm{111}}}$, $\iota$, and $\ottmv{X}$ such that
          $f_{{\mathrm{1}}}  \ottsym{=}  w_{{\mathrm{111}}}  \ottsym{:}   \iota \Rightarrow  \unskip ^ { \ell_{{\mathrm{1}}} }  \!  \star \Rightarrow  \unskip ^ { \ell }  \! \ottmv{X}  $, $\ottnt{U_{{\mathrm{1}}}}  \ottsym{=}  \star$, and $\ottnt{U}  \ottsym{=}  \ottmv{X}$.
          $w_{{\mathrm{111}}}  \ottsym{:}   \iota \Rightarrow  \unskip ^ { \ell_{{\mathrm{1}}} }  \!  \star \Rightarrow  \unskip ^ { \ell }  \! \ottmv{X}   \,  \xmapsto{ \mathmakebox[0.4em]{} S \mathmakebox[0.3em]{} }  \, S  \ottsym{(}  w_{{\mathrm{111}}}  \ottsym{)}$ where $S  \ottsym{=}  [  \ottmv{X}  :=  \iota  ]$.
          By Lemma \ref{lem:term_prec_inversion3},
          $ \langle   \emptyset    \vdash   w_{{\mathrm{111}}}  :  \iota   \sqsubseteq _{ S_{{\mathrm{0}}} }  \ottnt{U'_{{\mathrm{1}}}}  :  f'_{{\mathrm{11}}}  \dashv   \emptyset   \rangle $.
          By Lemma \ref{lem:left_subst_preserve_prec},
          $ \langle   \emptyset    \vdash   S  \ottsym{(}  w_{{\mathrm{111}}}  \ottsym{)}  :  \iota   \sqsubseteq _{  S  \circ  S_{{\mathrm{0}}}  }  \ottnt{U'_{{\mathrm{1}}}}  :  f'_{{\mathrm{11}}}  \dashv   \emptyset   \rangle $
          and $ S  \ottsym{(}  \ottmv{X}  \ottsym{)}   \sqsubseteq _{  S  \circ  S_{{\mathrm{0}}}  }  \ottnt{U'} $.
          So, $ \iota   \sqsubseteq _{  S  \circ  S_{{\mathrm{0}}}  }  \ottnt{U'} $.
          By \rnp{P\_CastR},
          $ \langle   \emptyset    \vdash   S  \ottsym{(}  w_{{\mathrm{111}}}  \ottsym{)}  :  \iota   \sqsubseteq _{  S  \circ  S_{{\mathrm{0}}}  }  \ottnt{U'}  :  f'_{{\mathrm{11}}}  \ottsym{:}   \ottnt{U'_{{\mathrm{1}}}} \Rightarrow  \unskip ^ { \ell' }  \! \ottnt{U'}   \dashv   \emptyset   \rangle $.

          \case{\rnp{R\_InstArrow}}
          There exist $w_{{\mathrm{111}}}$ and $\ottmv{X}$ such that
          $f_{{\mathrm{1}}}  \ottsym{=}  w_{{\mathrm{111}}}  \ottsym{:}   \star  \!\rightarrow\!  \star \Rightarrow  \unskip ^ { \ell_{{\mathrm{1}}} }  \!  \star \Rightarrow  \unskip ^ { \ell }  \! \ottmv{X}  $, $\ottnt{U_{{\mathrm{1}}}}  \ottsym{=}  \star$, and $\ottnt{U}  \ottsym{=}  \ottmv{X}$.
          We have
          $w_{{\mathrm{111}}}  \ottsym{:}   \star  \!\rightarrow\!  \star \Rightarrow  \unskip ^ { \ell_{{\mathrm{1}}} }  \!  \star \Rightarrow  \unskip ^ { \ell }  \! \ottmv{X}   \,  \xmapsto{ \mathmakebox[0.4em]{} S \mathmakebox[0.3em]{} }  \, S  \ottsym{(}  w_{{\mathrm{111}}}  \ottsym{)}  \ottsym{:}   \star  \!\rightarrow\!  \star \Rightarrow  \unskip ^ { \ell_{{\mathrm{1}}} }  \!  \star \Rightarrow  \unskip ^ { \ell }  \!  \star  \!\rightarrow\!  \star \Rightarrow  \unskip ^ { \ell }  \! \ottmv{X_{{\mathrm{1}}}}  \!\rightarrow\!  \ottmv{X_{{\mathrm{2}}}}   $
          where $S  \ottsym{=}  [  \ottmv{X}  :=  \ottmv{X_{{\mathrm{1}}}}  \!\rightarrow\!  \ottmv{X_{{\mathrm{2}}}}  ]$.
          By Lemma \ref{lem:term_prec_to_type_prec},
          $ \star   \sqsubseteq _{ S_{{\mathrm{0}}} }  \ottnt{U'_{{\mathrm{1}}}} $.
          By definition, $\ottnt{U'_{{\mathrm{1}}}}  \ottsym{=}  \star$.
          Here, $ \langle   \emptyset    \vdash   \ottsym{(}  w_{{\mathrm{111}}}  \ottsym{:}   \star  \!\rightarrow\!  \star \Rightarrow  \unskip ^ { \ell_{{\mathrm{1}}} }  \! \star   \ottsym{)}  :  \star   \sqsubseteq _{ S_{{\mathrm{0}}} }  \star  :  f'_{{\mathrm{11}}}  \dashv   \emptyset   \rangle $
          and $ \ottmv{X}   \sqsubseteq _{ S_{{\mathrm{0}}} }  \ottnt{U'} $.
          By \rnp{P\_CastL},
          $ \langle   \emptyset    \vdash   \ottsym{(}  w_{{\mathrm{111}}}  \ottsym{:}   \star  \!\rightarrow\!  \star \Rightarrow  \unskip ^ { \ell_{{\mathrm{1}}} }  \!  \star \Rightarrow  \unskip ^ { \ell }  \! \star  \!\rightarrow\!  \star    \ottsym{)}  :  \star  \!\rightarrow\!  \star   \sqsubseteq _{ S_{{\mathrm{0}}} }  \star  :  f'_{{\mathrm{11}}}  \dashv   \emptyset   \rangle $.
          By Lemma \ref{lem:left_subst_preserve_prec},
          $ \langle   \emptyset    \vdash   \ottsym{(}  S  \ottsym{(}  w_{{\mathrm{111}}}  \ottsym{)}  \ottsym{:}   \star  \!\rightarrow\!  \star \Rightarrow  \unskip ^ { \ell_{{\mathrm{1}}} }  \!  \star \Rightarrow  \unskip ^ { \ell }  \! \star  \!\rightarrow\!  \star    \ottsym{)}  :  \star  \!\rightarrow\!  \star   \sqsubseteq _{  S  \circ  S_{{\mathrm{0}}}  }  \star  :  f'_{{\mathrm{11}}}  \dashv   \emptyset   \rangle $
          and $ S  \ottsym{(}  \ottmv{X}  \ottsym{)}   \sqsubseteq _{  S  \circ  S_{{\mathrm{0}}}  }  \ottnt{U'} $.
          So, $ \ottmv{X_{{\mathrm{1}}}}  \!\rightarrow\!  \ottmv{X_{{\mathrm{2}}}}   \sqsubseteq _{  S  \circ  S_{{\mathrm{0}}}  }  \ottnt{U'} $.
          By \rnp{P\_Cast},
          $ \langle   \emptyset    \vdash   \ottsym{(}  S  \ottsym{(}  w_{{\mathrm{111}}}  \ottsym{)}  \ottsym{:}   \star  \!\rightarrow\!  \star \Rightarrow  \unskip ^ { \ell_{{\mathrm{1}}} }  \!  \star \Rightarrow  \unskip ^ { \ell }  \!  \star  \!\rightarrow\!  \star \Rightarrow  \unskip ^ { \ell }  \! \ottmv{X_{{\mathrm{1}}}}  \!\rightarrow\!  \ottmv{X_{{\mathrm{2}}}}     \ottsym{)}  :  \ottmv{X_{{\mathrm{1}}}}  \!\rightarrow\!  \ottmv{X_{{\mathrm{2}}}}   \sqsubseteq _{  S  \circ  S_{{\mathrm{0}}}  }  \ottnt{U'}  :  f'_{{\mathrm{11}}}  \ottsym{:}   \star \Rightarrow  \unskip ^ { \ell' }  \! \ottnt{U'}   \dashv   \emptyset   \rangle $.

          \otherwise
          Contradiction.
        \end{caseanalysis}

        \case{$\ottnt{E}  \ottsym{=}  \ottnt{E_{{\mathrm{1}}}}  \ottsym{:}  \ottnt{U_{{\mathrm{1}}}}  \Rightarrow   \unskip ^ { \ell }  \, \ottnt{U}$ for some $\ottnt{E_{{\mathrm{1}}}}$} Similar to the case of \rnp{P\_Op}.

        \otherwise
        Contradiction.
      \end{caseanalysis}

      \case{\rnp{E\_Abort}}
      Here, $f_{{\mathrm{11}}}  \ottsym{:}   \ottnt{U_{{\mathrm{1}}}} \Rightarrow  \unskip ^ { \ell }  \! \ottnt{U}  \,  \xmapsto{ \mathmakebox[0.4em]{} [  ] \mathmakebox[0.3em]{} }  \, \textsf{\textup{blame}\relax} \, \ell_{{\mathrm{1}}}$ for some $\ell_{{\mathrm{1}}}$.
      We finish by \rnp{P\_Blame}.
    \end{caseanalysis}

    \case{\rnp{P\_CastL}}
    We are given $f_{{\mathrm{1}}}  \ottsym{=}  f_{{\mathrm{11}}}  \ottsym{:}   \ottnt{U_{{\mathrm{1}}}} \Rightarrow  \unskip ^ { \ell }  \! \ottnt{U} $
    for some $f_{{\mathrm{11}}}$, $\ottnt{U_{{\mathrm{1}}}}$, and $\ell$.

    By inversion,
    $ \langle   \emptyset    \vdash   f_{{\mathrm{11}}}  :  \ottnt{U_{{\mathrm{1}}}}   \sqsubseteq _{ S_{{\mathrm{0}}} }  \ottnt{U'}  :  f'_{{\mathrm{1}}}  \dashv   \emptyset   \rangle $,
    $ \ottnt{U}   \sqsubseteq _{ S_{{\mathrm{0}}} }  \ottnt{U'} $, and $\ottnt{U_{{\mathrm{1}}}}  \sim  \ottnt{U}$.

    By case analysis on the evaluation rule applied to $f_{{\mathrm{1}}}$.

    \begin{caseanalysis}
      \case{\rnp{E\_Step}}
      There exist $\ottnt{E}$ and $f_{{\mathrm{12}}}$ such that $\ottnt{E}  [  f_{{\mathrm{12}}}  ] \,  \xmapsto{ \mathmakebox[0.4em]{} S \mathmakebox[0.3em]{} }  \, S  \ottsym{(}  \ottnt{E}  [  f'_{{\mathrm{12}}}  ]  \ottsym{)}$.
      By inversion, $f_{{\mathrm{12}}} \,  \xrightarrow{ \mathmakebox[0.4em]{} S \mathmakebox[0.3em]{} }  \, f'_{{\mathrm{12}}}$.
      By case analysis on the structure of $\ottnt{E}$.

      \begin{caseanalysis}
        \case{$\ottnt{E}  \ottsym{=}  \left[ \, \right]$}
        Here, $f_{{\mathrm{11}}}  \ottsym{:}   \ottnt{U_{{\mathrm{1}}}} \Rightarrow  \unskip ^ { \ell }  \! \ottnt{U}   \ottsym{=}  f_{{\mathrm{12}}}$.
        By case analysis on the reduction rule applied to $f_{{\mathrm{12}}}$.

        \begin{caseanalysis}
          \case{\rnp{R\_IdBase}}
          There exist $w_{{\mathrm{11}}}$ and $\iota$ such that
          $f_{{\mathrm{11}}}  \ottsym{=}  w_{{\mathrm{11}}}$, $\ottnt{U_{{\mathrm{1}}}}  \ottsym{=}  \iota$ and $\ottnt{U}  \ottsym{=}  \iota$.
          Since $w_{{\mathrm{11}}}  \ottsym{:}   \iota \Rightarrow  \unskip ^ { \ell }  \! \iota  \,  \xmapsto{ \mathmakebox[0.4em]{} [  ] \mathmakebox[0.3em]{} }  \, w_{{\mathrm{11}}}$, we finish.

          \case{\rnp{R\_IdStar}}
          Similar to the case \rnp{R\_IdBase}.

          \case{\rnp{R\_Succeed}}
          There exist $w_{{\mathrm{111}}}$ and $\ottnt{G}$ such that
          $f_{{\mathrm{1}}}  \ottsym{=}  w_{{\mathrm{111}}}  \ottsym{:}   \ottnt{G} \Rightarrow  \unskip ^ { \ell_{{\mathrm{1}}} }  \!  \star \Rightarrow  \unskip ^ { \ell }  \! \ottnt{G}  $, $\ottnt{U_{{\mathrm{1}}}}  \ottsym{=}  \star$, and $\ottnt{U}  \ottsym{=}  \ottnt{G}$.
          We have $w_{{\mathrm{111}}}  \ottsym{:}   \ottnt{G} \Rightarrow  \unskip ^ { \ell_{{\mathrm{1}}} }  \!  \star \Rightarrow  \unskip ^ { \ell }  \! \ottnt{G}   \,  \xmapsto{ \mathmakebox[0.4em]{} [  ] \mathmakebox[0.3em]{} }  \, w_{{\mathrm{111}}}$.
          Here, $ \langle   \emptyset    \vdash   w_{{\mathrm{111}}}  \ottsym{:}   \ottnt{G} \Rightarrow  \unskip ^ { \ell_{{\mathrm{1}}} }  \! \star   :  \star   \sqsubseteq _{ S_{{\mathrm{0}}} }  \ottnt{U'}  :  f'_{{\mathrm{1}}}  \dashv   \emptyset   \rangle $.
          By Lemma~\ref{lem:term_prec_inversion3},
          $ \langle   \emptyset    \vdash   w_{{\mathrm{111}}}  :  \ottnt{G}   \sqsubseteq _{ S_{{\mathrm{0}}} }  \ottnt{U'}  :  f'_{{\mathrm{1}}}  \dashv   \emptyset   \rangle $.

          \case{\rnp{R\_Fail}}
          There exist $w_{{\mathrm{111}}}$, $\ottnt{G}$, $\ottnt{G'}$ such that
          $f_{{\mathrm{1}}}  \ottsym{=}  w_{{\mathrm{111}}}  \ottsym{:}   \ottnt{G'} \Rightarrow  \unskip ^ { \ell_{{\mathrm{1}}} }  \!  \star \Rightarrow  \unskip ^ { \ell }  \! \ottnt{G}  $, $\ottnt{U_{{\mathrm{1}}}}  \ottsym{=}  \star$, and $\ottnt{U}  \ottsym{=}  \ottnt{G}$.
          We have $w_{{\mathrm{111}}}  \ottsym{:}   \ottnt{G'} \Rightarrow  \unskip ^ { \ell_{{\mathrm{1}}} }  \!  \star \Rightarrow  \unskip ^ { \ell }  \! \ottnt{G}   \,  \xmapsto{ \mathmakebox[0.4em]{} [  ] \mathmakebox[0.3em]{} }  \, \textsf{\textup{blame}\relax} \, \ell$.
          We finish by \rnp{P\_Blame} with Lemma~\ref{lem:term_prec_to_typing}.

          \case{\rnp{R\_Ground}}
          There exist $w_{{\mathrm{11}}}$ and $\ottnt{G}$ such that
          $f_{{\mathrm{11}}}  \ottsym{=}  w_{{\mathrm{11}}}$, $\ottnt{U}  \ottsym{=}  \star$, and $\ottnt{U_{{\mathrm{1}}}}  \sim  \ottnt{G}$.
          We have $w_{{\mathrm{11}}}  \ottsym{:}   \ottnt{U_{{\mathrm{1}}}} \Rightarrow  \unskip ^ { \ell }  \! \star  \,  \xmapsto{ \mathmakebox[0.4em]{} [  ] \mathmakebox[0.3em]{} }  \, w_{{\mathrm{11}}}  \ottsym{:}   \ottnt{U_{{\mathrm{1}}}} \Rightarrow  \unskip ^ { \ell }  \!  \ottnt{G} \Rightarrow  \unskip ^ { \ell }  \! \star  $.
          Since $ \ottnt{U}   \sqsubseteq _{ S_{{\mathrm{0}}} }  \ottnt{U'} $ and $\ottnt{U}  \ottsym{=}  \star$, we have $\ottnt{U'}  \ottsym{=}  \star$.
          We finish by \rnp{P\_CastL}.

          \case{\rnp{R\_Expand}}
          There exist $w_{{\mathrm{11}}}$ and $\ottnt{G}$ such that
          $f_{{\mathrm{11}}}  \ottsym{=}  w_{{\mathrm{11}}}$, $\ottnt{U_{{\mathrm{1}}}}  \ottsym{=}  \star$, and $\ottnt{U}  \sim  \ottnt{G}$.
          We have $w_{{\mathrm{11}}}  \ottsym{:}   \star \Rightarrow  \unskip ^ { \ell }  \! \ottnt{U}  \,  \xmapsto{ \mathmakebox[0.4em]{} [  ] \mathmakebox[0.3em]{} }  \, w_{{\mathrm{11}}}  \ottsym{:}   \star \Rightarrow  \unskip ^ { \ell }  \!  \ottnt{G} \Rightarrow  \unskip ^ { \ell }  \! \ottnt{U}  $.
          By Lemma~\ref{lem:term_prec_to_type_prec}, $ \ottnt{U_{{\mathrm{1}}}}   \sqsubseteq _{ S_{{\mathrm{0}}} }  \ottnt{U'} $.
          Since $\ottnt{U_{{\mathrm{1}}}}  \ottsym{=}  \star$, we have $\ottnt{U'}  \ottsym{=}  \star$.
          Thus, we finish by \rnp{P\_CastL}.

          \case{\rnp{R\_InstBase}}
          There exist $w_{{\mathrm{111}}}$, $\iota$, and $\ottmv{X}$ such that
          $f_{{\mathrm{1}}}  \ottsym{=}  w_{{\mathrm{111}}}  \ottsym{:}   \iota \Rightarrow  \unskip ^ { \ell_{{\mathrm{1}}} }  \!  \star \Rightarrow  \unskip ^ { \ell }  \! \ottmv{X}  $, $\ottnt{U_{{\mathrm{1}}}}  \ottsym{=}  \star$, and $\ottnt{U}  \ottsym{=}  \ottmv{X}$.
          We have $w_{{\mathrm{111}}}  \ottsym{:}   \iota \Rightarrow  \unskip ^ { \ell_{{\mathrm{1}}} }  \!  \star \Rightarrow  \unskip ^ { \ell }  \! \ottmv{X}   \,  \xmapsto{ \mathmakebox[0.4em]{} S \mathmakebox[0.3em]{} }  \, S  \ottsym{(}  w_{{\mathrm{111}}}  \ottsym{)}$ where $S  \ottsym{=}  [  \ottmv{X}  :=  \iota  ]$.
          By Lemma \ref{lem:term_prec_inversion3},
          $ \langle   \emptyset    \vdash   w_{{\mathrm{111}}}  :  \iota   \sqsubseteq _{ S_{{\mathrm{0}}} }  \ottnt{U'}  :  f'_{{\mathrm{1}}}  \dashv   \emptyset   \rangle $.
          By Lemma \ref{lem:left_subst_preserve_prec},
          $ \langle   \emptyset    \vdash   S  \ottsym{(}  w_{{\mathrm{111}}}  \ottsym{)}  :  \iota   \sqsubseteq _{  S  \circ  S_{{\mathrm{0}}}  }  \ottnt{U'}  :  f'_{{\mathrm{1}}}  \dashv   \emptyset   \rangle $.

          \case{\rnp{R\_InstArrow}}
          There exist $w_{{\mathrm{111}}}$ and $\ottmv{X}$ such that
          $f_{{\mathrm{1}}}  \ottsym{=}  w_{{\mathrm{111}}}  \ottsym{:}   \star  \!\rightarrow\!  \star \Rightarrow  \unskip ^ { \ell_{{\mathrm{1}}} }  \!  \star \Rightarrow  \unskip ^ { \ell }  \! \ottmv{X}  $, $\ottnt{U_{{\mathrm{1}}}}  \ottsym{=}  \star$, and $\ottnt{U}  \ottsym{=}  \ottmv{X}$.
          We have $w_{{\mathrm{111}}}  \ottsym{:}   \star  \!\rightarrow\!  \star \Rightarrow  \unskip ^ { \ell_{{\mathrm{1}}} }  \!  \star \Rightarrow  \unskip ^ { \ell }  \! \ottmv{X}   \,  \xmapsto{ \mathmakebox[0.4em]{} S \mathmakebox[0.3em]{} }  \, S  \ottsym{(}  w_{{\mathrm{111}}}  \ottsym{)}  \ottsym{:}   \star  \!\rightarrow\!  \star \Rightarrow  \unskip ^ { \ell_{{\mathrm{1}}} }  \!  \star \Rightarrow  \unskip ^ { \ell }  \!  \star  \!\rightarrow\!  \star \Rightarrow  \unskip ^ { \ell }  \! \ottmv{X_{{\mathrm{1}}}}  \!\rightarrow\!  \ottmv{X_{{\mathrm{2}}}}   $
          where $S  \ottsym{=}  [  \ottmv{X}  :=  \ottmv{X_{{\mathrm{1}}}}  \!\rightarrow\!  \ottmv{X_{{\mathrm{2}}}}  ]$.
          By Lemma \ref{lem:term_prec_to_type_prec},
          $ \ottnt{U_{{\mathrm{1}}}}   \sqsubseteq _{ S_{{\mathrm{0}}} }  \ottnt{U'} $.
          Since $\ottnt{U_{{\mathrm{1}}}}  \ottsym{=}  \star$, $\ottnt{U'}  \ottsym{=}  \star$.
          Here, $ \langle   \emptyset    \vdash   w_{{\mathrm{111}}}  \ottsym{:}   \star  \!\rightarrow\!  \star \Rightarrow  \unskip ^ { \ell_{{\mathrm{1}}} }  \! \star   :  \star   \sqsubseteq _{ S_{{\mathrm{0}}} }  \star  :  f'_{{\mathrm{1}}}  \dashv   \emptyset   \rangle $.
          By \rnp{P\_CastL},
          $ \langle   \emptyset    \vdash   \ottsym{(}  w_{{\mathrm{111}}}  \ottsym{:}   \star  \!\rightarrow\!  \star \Rightarrow  \unskip ^ { \ell_{{\mathrm{1}}} }  \!  \star \Rightarrow  \unskip ^ { \ell }  \! \star  \!\rightarrow\!  \star    \ottsym{)}  :  \star  \!\rightarrow\!  \star   \sqsubseteq _{ S_{{\mathrm{0}}} }  \star  :  f'_{{\mathrm{1}}}  \dashv   \emptyset   \rangle $.
          By Lemma \ref{lem:left_subst_preserve_prec},
          $ \langle   \emptyset    \vdash   \ottsym{(}  S  \ottsym{(}  w_{{\mathrm{111}}}  \ottsym{)}  \ottsym{:}   \star  \!\rightarrow\!  \star \Rightarrow  \unskip ^ { \ell_{{\mathrm{1}}} }  \!  \star \Rightarrow  \unskip ^ { \ell }  \! \star  \!\rightarrow\!  \star    \ottsym{)}  :  \star  \!\rightarrow\!  \star   \sqsubseteq _{  S  \circ  S_{{\mathrm{0}}}  }  \star  :  f'_{{\mathrm{1}}}  \dashv   \emptyset   \rangle $.
          Since $ \ottmv{X_{{\mathrm{1}}}}  \!\rightarrow\!  \ottmv{X_{{\mathrm{2}}}}   \sqsubseteq _{  S  \circ  S_{{\mathrm{0}}}  }  \star $,
          we have
          $ \langle   \emptyset    \vdash   \ottsym{(}  S  \ottsym{(}  w_{{\mathrm{111}}}  \ottsym{)}  \ottsym{:}   \star  \!\rightarrow\!  \star \Rightarrow  \unskip ^ { \ell_{{\mathrm{1}}} }  \!  \star \Rightarrow  \unskip ^ { \ell }  \!  \star  \!\rightarrow\!  \star \Rightarrow  \unskip ^ { \ell }  \! \ottmv{X_{{\mathrm{1}}}}  \!\rightarrow\!  \ottmv{X_{{\mathrm{2}}}}     \ottsym{)}  :  \ottmv{X_{{\mathrm{1}}}}  \!\rightarrow\!  \ottmv{X_{{\mathrm{2}}}}   \sqsubseteq _{  S  \circ  S_{{\mathrm{0}}}  }  \star  :  f'_{{\mathrm{1}}}  \dashv   \emptyset   \rangle $
          by \rnp{P\_CastL}.

          \otherwise
          Contradiction.
        \end{caseanalysis}

        \case{$\ottnt{E}  \ottsym{=}  \ottnt{E_{{\mathrm{1}}}}  \ottsym{:}  \ottnt{U_{{\mathrm{1}}}}  \Rightarrow   \unskip ^ { \ell }  \, \ottnt{U}$ for some $\ottnt{E_{{\mathrm{1}}}}$}
        Similar to the case of \rnp{P\_Op}.

        \otherwise
        Contradiction.
      \end{caseanalysis}

      \case{\rnp{E\_Abort}}
      Here, $f_{{\mathrm{11}}}  \ottsym{:}   \ottnt{U_{{\mathrm{1}}}} \Rightarrow  \unskip ^ { \ell }  \! \ottnt{U}  \,  \xmapsto{ \mathmakebox[0.4em]{} [  ] \mathmakebox[0.3em]{} }  \, \textsf{\textup{blame}\relax} \, \ell_{{\mathrm{1}}}$ for some $\ell_{{\mathrm{1}}}$.
      We finish by \rnp{P\_Blame}.
    \end{caseanalysis}

    \case{\rnp{P\_CastR}}
    We are given $f'_{{\mathrm{1}}}  \ottsym{=}  f'_{{\mathrm{11}}}  \ottsym{:}   \ottnt{U'_{{\mathrm{1}}}} \Rightarrow  \unskip ^ { \ell' }  \! \ottnt{U'} $
    for some $f'_{{\mathrm{11}}}$, $\ottnt{U'_{{\mathrm{1}}}}$, and $\ell'$.
    By inversion,
    $ \langle   \emptyset    \vdash   f_{{\mathrm{1}}}  :  \ottnt{U}   \sqsubseteq _{ S_{{\mathrm{0}}} }  \ottnt{U'_{{\mathrm{1}}}}  :  f'_{{\mathrm{11}}}  \dashv   \emptyset   \rangle $,
    $ \ottnt{U}   \sqsubseteq _{ S_{{\mathrm{0}}} }  \ottnt{U'} $, and $\ottnt{U'_{{\mathrm{1}}}}  \sim  \ottnt{U'}$.
    By the IH, there exist $S'_{{\mathrm{0}}}$, $S'$, $S''$,and $f'_{{\mathrm{21}}}$
    such that
    \begin{itemize}
     \item $f'_{{\mathrm{11}}} \,  \xmapsto{ \mathmakebox[0.4em]{} S' \mathmakebox[0.3em]{} }  \, f'_{{\mathrm{21}}}$,
     \item $ \langle   \emptyset    \vdash   f_{{\mathrm{2}}}  :  S  \ottsym{(}  \ottnt{U}  \ottsym{)}   \sqsubseteq _{  S''  \uplus  \ottsym{(}   S'_{{\mathrm{0}}}  \circ   S  \circ  S_{{\mathrm{0}}}    \ottsym{)}  }  S'  \ottsym{(}  \ottnt{U'_{{\mathrm{1}}}}  \ottsym{)}  :  f'_{{\mathrm{21}}}  \dashv   \emptyset   \rangle $,
     \item $\forall \ottmv{X} \in \textit{dom} \, \ottsym{(}   S  \circ  S_{{\mathrm{0}}}   \ottsym{)}.  S  \circ  S_{{\mathrm{0}}}   \ottsym{(}  \ottmv{X}  \ottsym{)}  \ottsym{=}   S'_{{\mathrm{0}}}  \circ    S  \circ  S_{{\mathrm{0}}}   \circ  S'    \ottsym{(}  \ottmv{X}  \ottsym{)}$,
     \item $\forall \ottmv{X} \in \textit{dom} \, \ottsym{(}  S'  \ottsym{)}. \textit{ftv} \, \ottsym{(}  S'  \ottsym{(}  \ottmv{X}  \ottsym{)}  \ottsym{)}  \subseteq  \textit{dom} \, \ottsym{(}  S'_{{\mathrm{0}}}  \ottsym{)}$, and
     \item $\textit{dom} \, \ottsym{(}  S''  \ottsym{)}$ is a set of fresh type variables.
    \end{itemize}
    So, $f'_{{\mathrm{11}}}  \ottsym{:}   \ottnt{U'_{{\mathrm{1}}}} \Rightarrow  \unskip ^ { \ell' }  \! \ottnt{U'}  \,  \xmapsto{ \mathmakebox[0.4em]{} S' \mathmakebox[0.3em]{} }  \, f'_{{\mathrm{21}}}  \ottsym{:}   S'  \ottsym{(}  \ottnt{U'_{{\mathrm{1}}}}  \ottsym{)} \Rightarrow  \unskip ^ { \ell' }  \! S'  \ottsym{(}  \ottnt{U'}  \ottsym{)} $.

    By Lemma~\ref{lem:left_subst_preserve_prec},
    $ S  \ottsym{(}  \ottnt{U}  \ottsym{)}   \sqsubseteq _{  S  \circ  S_{{\mathrm{0}}}  }  \ottnt{U'} $.
    By Lemma~\ref{lem:right_subst_preserve_prec},
    $ S  \ottsym{(}  \ottnt{U}  \ottsym{)}   \sqsubseteq _{   S'_{{\mathrm{0}}}  \circ  S   \circ  S_{{\mathrm{0}}}  }  S'  \ottsym{(}  \ottnt{U'}  \ottsym{)} $.
    By Lemma~\ref{lem:prec_subst_weak},
    $ S  \ottsym{(}  \ottnt{U}  \ottsym{)}   \sqsubseteq _{  S''  \uplus  \ottsym{(}   S'_{{\mathrm{0}}}  \circ   S  \circ  S_{{\mathrm{0}}}    \ottsym{)}  }  S'  \ottsym{(}  \ottnt{U'}  \ottsym{)} $.

    By \rnp{P\_CastR},
    $ \langle   \emptyset    \vdash   f_{{\mathrm{2}}}  :  S  \ottsym{(}  \ottnt{U}  \ottsym{)}   \sqsubseteq _{  S''  \uplus  \ottsym{(}   S'_{{\mathrm{0}}}  \circ   S  \circ  S_{{\mathrm{0}}}    \ottsym{)}  }  S'  \ottsym{(}  \ottnt{U'}  \ottsym{)}  :  f'_{{\mathrm{21}}}  \ottsym{:}   S'  \ottsym{(}  \ottnt{U'_{{\mathrm{1}}}}  \ottsym{)} \Rightarrow  \unskip ^ { \ell' }  \! S'  \ottsym{(}  \ottnt{U'}  \ottsym{)}   \dashv   \emptyset   \rangle $.

    \otherwise Cannot happen.  \qedhere
  \end{caseanalysis}
\end{proof}


\begin{theoremA}[Dynamic Gradual Guarantee of $\lambdaRTI$] \label{thm:gg}
  Suppose $ \langle   \emptyset    \vdash   f  :  \ottnt{U}   \sqsubseteq _{ S_{{\mathrm{0}}} }  \ottnt{U'}  :  f'  \dashv   \emptyset   \rangle $ and $ \emptyset   \vdash  f  \ottsym{:}  \ottnt{U}$.

  \begin{enumerate}
    \item $ \emptyset   \vdash  f'  \ottsym{:}  \ottnt{U'}$ and $ \ottnt{U}   \sqsubseteq _{ S_{{\mathrm{0}}} }  \ottnt{U'} $.
    \item
      \begin{itemize}
        \item If $f \,  \xmapsto{ \mathmakebox[0.4em]{} S \mathmakebox[0.3em]{} }\hspace{-0.4em}{}^\ast \hspace{0.2em}  \, w$,
          then
          $f' \,  \xmapsto{ \mathmakebox[0.4em]{} S' \mathmakebox[0.3em]{} }\hspace{-0.4em}{}^\ast \hspace{0.2em}  \, w'$ and
          $ \langle   \emptyset    \vdash   w  :  S  \ottsym{(}  \ottnt{U}  \ottsym{)}   \sqsubseteq _{  S''  \uplus  \ottsym{(}   S'_{{\mathrm{0}}}  \circ   S  \circ  S_{{\mathrm{0}}}    \ottsym{)}  }  S'  \ottsym{(}  \ottnt{U'}  \ottsym{)}  :  w'  \dashv   \emptyset   \rangle $
          for some $w'$, $S'$, $S'_{{\mathrm{0}}}$, and $S''$.
        \item If $ f \!  \Uparrow  $, then $ f' \!  \Uparrow  $.
      \end{itemize}
    \item
      \begin{itemize}
        \item If $f' \,  \xmapsto{ \mathmakebox[0.4em]{} S' \mathmakebox[0.3em]{} }\hspace{-0.4em}{}^\ast \hspace{0.2em}  \, w'$,
          then
          $f \,  \xmapsto{ \mathmakebox[0.4em]{} S \mathmakebox[0.3em]{} }\hspace{-0.4em}{}^\ast \hspace{0.2em}  \, w$ and
          $ \langle   \emptyset    \vdash   w  :  S  \ottsym{(}  \ottnt{U}  \ottsym{)}   \sqsubseteq _{  S''  \uplus  \ottsym{(}   S'_{{\mathrm{0}}}  \circ   S  \circ  S_{{\mathrm{0}}}    \ottsym{)}  }  S'  \ottsym{(}  \ottnt{U'}  \ottsym{)}  :  w'  \dashv   \emptyset   \rangle $
          for some $w$, $S$, $S'_{{\mathrm{0}}}$, and $S''$,
          or $f \,  \xmapsto{ \mathmakebox[0.4em]{} S \mathmakebox[0.3em]{} }\hspace{-0.4em}{}^\ast \hspace{0.2em}  \, \textsf{\textup{blame}\relax} \, \ell$ for some $\ell$ and $S$.
        \item If $f' \,  \xmapsto{ \mathmakebox[0.4em]{} S' \mathmakebox[0.3em]{} }\hspace{-0.4em}{}^\ast \hspace{0.2em}  \, \textsf{\textup{blame}\relax} \, \ell'$,
          then $f \,  \xmapsto{ \mathmakebox[0.4em]{} S \mathmakebox[0.3em]{} }\hspace{-0.4em}{}^\ast \hspace{0.2em}  \, \textsf{\textup{blame}\relax} \, \ell$ for some $\ell$ and $S$.
        \item If $ f' \!  \Uparrow  $, then $ f \!  \Uparrow  $,
          or $f \,  \xmapsto{ \mathmakebox[0.4em]{} S \mathmakebox[0.3em]{} }\hspace{-0.4em}{}^\ast \hspace{0.2em}  \, \textsf{\textup{blame}\relax} \, \ell$ for some $\ell$ and $S$.
      \end{itemize}
  \end{enumerate}
\end{theoremA}

\begin{proof}
  \leavevmode
  \begin{enumerate}
    \item By Lemmas \ref{lem:term_prec_to_type_prec} and \ref{lem:term_prec_to_typing}.
    \item
      \begin{itemize}
        \item
          By induction on the evaluation sequence $f \,  \xmapsto{ \mathmakebox[0.4em]{} S \mathmakebox[0.3em]{} }\hspace{-0.4em}{}^\ast \hspace{0.2em}  \, w$ with
          Lemmas~\ref{lem:prec_catch_up_left_value} and \ref{lem:term_prec_simulation}.
        \item
          Here, there is an infinite evaluation sequence from $f$,
          and it contains infinite times application of \rnp{R\_Beta}.
          By Lemma \ref{lem:term_prec_simulation_app},
          there is an infinite evaluation sequence from $f'$.
      \end{itemize}
    \item
      \begin{itemize}
        \item $\ottnt{f}$ cannot diverge as it contradicts the second item of (2).
        \item $\ottnt{f}$ cannot diverge or evaluate to a value as it contradicts
          (2).
        \item $\ottnt{f}$ cannot evaluate to a value as it contradicts
          the first item of (2).
          \qedhere
      \end{itemize}
  \end{enumerate}
\end{proof}

\ifrestate
\thmDynamicGG*
\else
\begin{theorem}[name=Dynamic Gradual Guarantee,restate=thmDynamicGG]
  \label{thm:dynamicGG}
 Suppose that $ e   \sqsubseteq _{ S_{{\mathrm{0}}} }  e' $.
 Let
 $(S,\ottnt{U}) = PT( \emptyset ,e)$ and
 $(S',\ottnt{U'}) = PT( \emptyset , e')$.
 \begin{enumerate}
  \item \begin{itemize}
         \item If $ \langle   \emptyset    \vdash   S  \ottsym{(}  e  \ottsym{)}  :  \ottnt{U}  \rangle   \xmapsto{ \mathmakebox[0.4em]{} S_{{\mathrm{1}}} \mathmakebox[0.3em]{} }\hspace{-0.4em}{}^\ast \hspace{0.2em}    w $,
               then $ \langle   \emptyset    \vdash   S'  \ottsym{(}  e'  \ottsym{)}  :  \ottnt{U'}  \rangle   \xmapsto{ \mathmakebox[0.4em]{} S'_{{\mathrm{1}}} \mathmakebox[0.3em]{} }\hspace{-0.4em}{}^\ast \hspace{0.2em}    w' $ and
               $ \langle   \emptyset    \vdash   w  :  S_{{\mathrm{1}}}  \ottsym{(}  \ottnt{U}  \ottsym{)}   \sqsubseteq _{ S'_{{\mathrm{0}}} }  S'_{{\mathrm{1}}}  \ottsym{(}  \ottnt{U'}  \ottsym{)}  :  w'  \dashv   \emptyset   \rangle $
               for some $w'$, $S'_{{\mathrm{1}}}$, and $S'_{{\mathrm{0}}}$.
         \item If $ \langle   \emptyset    \vdash   S  \ottsym{(}  e  \ottsym{)}  :  \ottnt{U}  \rangle   \Uparrow  $,
               then $ \langle   \emptyset    \vdash   S'  \ottsym{(}  e'  \ottsym{)}  :  \ottnt{U'}  \rangle   \Uparrow  $.
        \end{itemize}
  \item \begin{itemize}
        \item If $ \langle   \emptyset    \vdash   S'  \ottsym{(}  e'  \ottsym{)}  :  \ottnt{U'}  \rangle   \xmapsto{ \mathmakebox[0.4em]{} S'_{{\mathrm{1}}} \mathmakebox[0.3em]{} }\hspace{-0.4em}{}^\ast \hspace{0.2em}    w' $,
          then
          $ \langle   \emptyset    \vdash   S  \ottsym{(}  e  \ottsym{)}  :  \ottnt{U}  \rangle   \xmapsto{ \mathmakebox[0.4em]{} S_{{\mathrm{1}}} \mathmakebox[0.3em]{} }\hspace{-0.4em}{}^\ast \hspace{0.2em}    w $ and
          $ \langle   \emptyset    \vdash   w  :  S_{{\mathrm{1}}}  \ottsym{(}  \ottnt{U}  \ottsym{)}   \sqsubseteq _{ S'_{{\mathrm{0}}} }  S'_{{\mathrm{1}}}  \ottsym{(}  \ottnt{U'}  \ottsym{)}  :  w'  \dashv   \emptyset   \rangle $
          for some $w$, $S_{{\mathrm{1}}}$, and $S'_{{\mathrm{0}}}$;
          or $ \langle   \emptyset    \vdash   S  \ottsym{(}  e  \ottsym{)}  :  \ottnt{U}  \rangle   \xmapsto{ \mathmakebox[0.4em]{} S_{{\mathrm{1}}} \mathmakebox[0.3em]{} }\hspace{-0.4em}{}^\ast \hspace{0.2em}    \textsf{\textup{blame}\relax} \, \ell $ for some $\ell$ and $S_{{\mathrm{1}}}$.
        \item If $ \langle   \emptyset    \vdash   S'  \ottsym{(}  e'  \ottsym{)}  :  \ottnt{U'}  \rangle   \xmapsto{ \mathmakebox[0.4em]{} S'_{{\mathrm{1}}} \mathmakebox[0.3em]{} }\hspace{-0.4em}{}^\ast \hspace{0.2em}    \textsf{\textup{blame}\relax} \, \ell' $,
          then $ \langle   \emptyset    \vdash   S  \ottsym{(}  e  \ottsym{)}  :  \ottnt{U}  \rangle   \xmapsto{ \mathmakebox[0.4em]{} S_{{\mathrm{1}}} \mathmakebox[0.3em]{} }\hspace{-0.4em}{}^\ast \hspace{0.2em}    \textsf{\textup{blame}\relax} \, \ell $ for some $\ell$ and $S_{{\mathrm{1}}}$.
        \item If $ \langle   \emptyset    \vdash   S'  \ottsym{(}  e'  \ottsym{)}  :  \ottnt{U'}  \rangle   \Uparrow  $,
          then $ \langle   \emptyset    \vdash   S  \ottsym{(}  e  \ottsym{)}  :  \ottnt{U}  \rangle   \Uparrow  $,
          or $ \langle   \emptyset    \vdash   S  \ottsym{(}  e  \ottsym{)}  :  \ottnt{U}  \rangle   \xmapsto{ \mathmakebox[0.4em]{} S_{{\mathrm{1}}} \mathmakebox[0.3em]{} }\hspace{-0.4em}{}^\ast \hspace{0.2em}    \textsf{\textup{blame}\relax} \, \ell $ for some $\ell$ and $S_{{\mathrm{1}}}$.
       \end{itemize}
 \end{enumerate}
\end{theorem}
\fi 

\begin{proof}
  By Lemma~\ref{thm:ITGL-staticGG} and Theorems~\ref{thm:ci_type_preservation} and \ref{thm:gg}.
\end{proof}

  \fi
\end{document}